%% file: main.tex
\documentclass[11pt,a4paper]{book}

\usepackage{authordate1-4}
\usepackage{mymacros}
\usepackage{amssymb}
\usepackage{graphicx}

\begin{document}
\pagestyle{headings}

\begin{titlepage}

\title{\bf Algebras of Information  \\
Axiomatic Foundation\\
}

\author{Prof.~Dr.~J\"urg Kohlas, \\[2em]
\small Dept. of Informatics DIUF
\\
\small University of  Fribourg\\ 
\small CH -- 1700 Fribourg
(Switzerland) \\ 
\small E-mail: {\tt juerg.kohlas@unifr.ch} \\
}
\date{\it Version: \today}

\maketitle
\end{titlepage}

\pagebreak
%

\tableofcontents 



\input{chapter1.tex}

\input{chapter2.tex}

\input{chapter3.tex}

\input{chapter4.tex}

\input{chapter5.tex}

\input{chapter6.tex}

\input{chapter7.tex}

\input{chapter8.tex}

\input{chapter9.tex}

\input{chapter10.tex}

\input{chapter11.tex}

\input{chapter12.tex}



\bibliographystyle{authordate1}
\bibliography{tcslit}

\end{document}

%% file: chapter1.tex
\chapter{Introduction}

The basic idea behind information algebras \cite{kohlas03,kohlasschmid14} is that information comes in pieces, each referring to a certain question, that these pieces can be combined or aggregated and that the part relating to a given question can be extracted. This algebraic structure can be given different forms. Questions are often represented by a lattice of domains, and a popular model is based on the subset lattice of a set of variables. Pieces of information are then represented by valuations associated with these domains. This leads then to an algebraic structure called valuation algebras \cite{kohlas03}. The axiomatics of this algebraic structure was in essence proposed by \cite{shenoyshafer90}. Valuation algebras have already many important applications in Computer Science related to constraint systems, relational databases, different uncertainty formalisms like probability, belief functions, fuzzy set and possibility measures, and many more, we refer to \cite{poulykohlas11}. An important particular case of valuation algebras, both from practical as well as theoretical point of views, are \textit{idempotent} valuation algebras, also called proper information algebras: The combination of a piece of information with itself or part of itself gives nothing new. This allows to introduce an order between pieces of information reflecting information content. It relates proper information algebras also to domain theory \cite{kohlas03,kohlasschmid14}. 

The basic view of information as pieces which can be combined, which relate to questions and from which the part relating to given questions can be extracted, leads to two different but essentially equivalent algebraic structure, \textit{labeled} and \textit{domain-free} valuation algebras \cite{kohlas03,kohlasschmid14}. The original proposal of an axiomatics for valuation algebras in \cite{shenoyshafer90} was in labeled form; later \cite{shafer91} proposed the domain-free form. However, for valuation algebras, the two forms are not fully equivalent, there are labeled forms which have no domain-free form and vice vera. An important contribution of this paper is to give a new axiomatic system for proper information algebras, where there exists a full duality between these two forms.

In this text we start with a novel, reduced axiomatic form of a domain-free information  algebra. Its two basic operation are those of the combination of two pieces of information and the extraction of the part of a piece of information relating to a question. The set of questions considered is a priori without any structure. In Chaper \ref{sec:InfAlg} it is however shown that the axiomatic structure of an information algebra induces both a partial order of information, reflecting the information content (Section \ref{sec:InfOrder}, and also a partial order between questions, reflecting the granularity, the fineness or coarseness of questions (Section \ref{sec:StructOfQuest}). In fact, there is more structure among questions, namely a relation of conditional independence between questions. This relation is called a \textit{quasi-sparoid} (q-separoid), since it is a reduct of a structure called sepraoid, intorduced in \cite{dawid01} for describing conditional independence and irrelevance in many frameworks. The more general structure of a q-separoid turns out to be sufficient to permit \textit{local compuitation} similar to the possibilities in valuation algebras as described in \cite{shenoyshafer90} and \cite{kohlas03}, see Chapter \ref{sec:LocComp} and this is one of the main points for studying information algebras. A further condition for extraction operators of interest is the requirement that the extractions operators commute, that is, return the same result independent of the order in which they are applied, Section \ref{sec:CommInfAlg}. This additional condition is then related to a very special conditional independence relation, which simpliyfies local computation. Finally, in Section \ref{sec:SetAlg} a special, very important instance of an information algebra where the information elements are subsets of some universe, so-called set algebras, are examined. In this case questions are represented by partitions of the universe, extraction corresponds to saturation operations and combination is simply intersection. It is shown in Chapters \ref{sec:ExtInfAlg} and \ref{sec:AtomicAlg} that any information algebra has a representation as a set algebra, can be seen as an algebra of subsets of some universe.

In Chapter \ref{sec:LocComp} the labeled version of a domain-free information algebra is derived. From a labeld information algebra its domain-free version can be reconstructed. It is shown that these two versions are in a precise sense equivalent (duality). However, the labeled version is better adapted for computational purposes, such as local computation. The domain-free version on the other hand is better suited for structural algebraic studies.

In some cases, an information algebra may possess most informative elements, called \textit{atoms}. And in some cases these atoms determine the information algebra fully (Chapter \ref{sec:AtomicAlg}). There is also a notion of most informative elements relative to a question. Then these relative atoms may represent the possible anserws to the question and give in this way a clear explicite meaning to the questions considered.

In information processing, only "finite" pieces of information can be treated. In Chapter \ref{sec:FiniteInf}, the concept of finite elements is adapted from domain theory, see for instance \cite{daveypriestley97}. In contrast to domain theory, in information algebras there is not only an order, but in addtiion, there are the operations of combination and extraction. So, the concept of finiteness has to be examined in the context of these operations. The same holds for the weaker concept of continuity, also adapted from domain theory. Furthermore, these concepts are also examined in the framework of labeld information algebras.

Often information is uncertain, that is, it is not sure that the statement contained in it holds, is true. We may assume that a piece of infomation is true only if some assumptions are valid. Modeling this idea leads to assumption-based reasoning. If, furthermore, the likelhood of different assumtions can be measured by probabilities, we come to \textit{probabilistic assumption-based reasoning}. This approach is developped in Chapter \ref{sec:UncertainInf}. There, maps fom a probability space into an information algebra are considered. This is in fact a generalization of the theory of hints \cite{kohlasmonney95}. In this book the maps from probability space in set algebras are considered. But most of the results derived in this particular case carry over to information algebras. The theory of hints is a semantic variant of Dempster-Shafer theory \cite{shafer76}, where a more epistemic view is taken. Again, mathematicallly speaking, many concepts of this theory apply to probabilistic argumentation systems, in particulare Shafer's concept of allocations of probability and support functions. It turns out that all these concepts indeed represent infomation and form information algebras.

There are other methods to represent uncertainty, especially probability distributions on the set of unknown answers. Into this category belong Bayesian networks, which form, as is known since long, a non-idempotent information algebra (a valuation algebra, see Chapter \ref{sec:nonidempotent}). More recently, the theory of imprecise probability has been created and generated much interest. There appear several, closely related information algebra in this theory. This is discussed in Chapter \ref{sec:ProbabInf}. So this kind of probabilistic information is yet another way to represent uncertain information and it illustrates once more how widespread information algebra are.

Originally, in valuation algebras idempotency of combination is not assumed. In Chatper \ref{sec:nonidempotent} this subject is resumed. But so far, valuation algebras were studied mainly in the multivariate case. Here however, we as before do not assume any particular structure of the set questions considered, only the usual properties of extraction operators. The semigroup properties of regularity and sperativity can be extended to valuation algebras. The information order of information algebras depends on idempotency and carries not over to valuation algebras. Nevertheless in valuation algebras we may still define an information order. It is however only a preorder. But in regular and separative valuation algebras this preorder has all desirable properties of an information order. Regularity and separativity allow in particular to introduce a division operation into valuation algebras. This in turn permits to introduce the notion of \textit{conditionals}, which generalizes the corresponding concept in probability theory. It is shown that all well-known properties of conditionals in probability theory carry over to regular and to some extend also to separative information algebras.

In Chapter \ref{sec:CondIndep} finally, we take up again the notion of conditional indpendence amog valuations or pieces of information as introduced earlier but only for regular or separative valuation algebras (Chaper \ref{sec:nonidempotent}). We study the properties of this relation and examine in particular under what conditions it forms a q-separoid. In addition we study a notion of compatibility among 
pieces of information, and in particular pairwise compatibility. It is shown that pairwise compatibility is sufficinet for full compatibility, if the domains of the pieces of information form a hypertree. Finally conditional independence is related to the factorization of a piece of information.

It remains one important subject, not treated so far, and that is the relation of information algebras and valuation algebras to Shanonn's theory of information. We have seen that if an information algebras has finitely many atoms relative to each question $x \in Q$, then the set of these relative atoms can be considered as possible answers to the question. Furthermore, the infomation algebra is isomorphic to the set algebra of subsets of its atoms (see Chapter \ref{sec:AtomicAlg}). The uncertainty of a piece of information relative to a question represented by a subset of relative atoms may then be measured by Hatley's measure of the subset. Further the reduction of the uncertainty relatiuve to a piece of information with respect to the initial uncertainty can be considered as a measure of the infomation contained in the piece of informations. Obvioulsy this measure respects information order. It has many other interesting properties. But the correspondig theory has still to be worked out. Similar theories may possibly be worked out for uncertain, probaiblistic information, using the notion of entropy.

%% file: chapter2.tex
\chapter{Information algebra} \label{sec:InfAlg}

\section{Basics} \label{sec:Basics}

An information algebra is constructed based on a set $\Phi$ of elements $\phi,\psi,\ldots$ representing pieces of information and a set $Q$ of elements $x,y,\ldots$ representing questions. Pieces of information $\phi$ and $\psi$ can be aggregated or combined into new pieces $\phi \cdot \psi$. So we have an operation
\begin{eqnarray*}
\cdot : \Phi \times \Phi \rightarrow \Phi,\quad (\phi,\psi) \mapsto \phi \cdot \psi.
\end{eqnarray*}
We assume that this operation is associative and commutative, so that $(\Phi,\cdot)$ is a commutative semigroup. We further assume the existence of a unit element $1$, representing vacuous information, so that $\phi \cdot 1 = 1 \cdot \phi = \phi$ for all $\phi \in \Phi$. In addition  we assume a null element $0$ so that $\phi \cdot 0 = 0 \cdot \phi = 0$ for all $\phi \in \Phi$. This element represents contradiction, it destroys any information. So we have a commutative semigroup $(\Phi,\cdot,0,1)$ representing combination of information. We shall see below that combination is also idempotent, $\phi \cdot \phi = \phi$.

Questions will not be represented explicitly, but only implicitly by operators $\epsilon_x : \Phi \rightarrow \Phi$ for any $x \in Q$, where $\epsilon_x(\phi) $ denotes the piece of information obtained, when the information regarding question $x$ is extracted from $\phi$. So we have a family of operators $E = \{\epsilon_x:x \in Q\}$ so that $(\phi,x) \mapsto \epsilon_x(\phi)$. Any of these operators must satisfy the following conditions:
\begin{enumerate}
\item $\epsilon_x(0) = 0$,
\item $\epsilon_x(\phi) \cdot \phi = \phi$,
\item $\epsilon_x(\epsilon_x(\phi) \cdot \psi) = \epsilon_x(\phi) \cdot \epsilon_x(\psi)$.
\end{enumerate}
So, from contradiction only contradiction can be extracted. A piece of information combined with any piece of information extracted from it, gives nothing new. The last condition says if a piece of information is combined with a piece extracted for question $x$ and then the combination is extracted for $x$, we may as well first extract the information form the second piece for $x$ and then combine. This is in particular important for computation. We shall see later, that these are in fact conditions as for an existential quantor in algebraic logic (Section \ref{sec:InfOrder}). We call the operators $\epsilon_x$ extraction operators. Note that $\epsilon_x(1) = 1 \cdot \epsilon_x(1) = 1$, by item 2 above. Also, if $\epsilon_x(\phi) = 0$, then again by item 2, $\phi = \epsilon_x(\phi) \cdot \phi = 0$.

We add in most cases, but not always, another condition,
\begin{eqnarray*}
\forall \phi \in \Phi, \exists x \in Q \textrm{ such that}\ \epsilon_x(\phi) = \phi.
\end{eqnarray*}
Such an $x$ is called a support of $\phi$ and the condition is called the \textit{support axiom}. It means that the piece of information $\phi$ bears on question $x$, is information for $x$. As a consequence it follows from item 2 above that $\phi \cdot \phi = \epsilon_x(\phi) \cdot \phi = \phi$ if $x$ is a support of $\phi$, the semigroup $\Phi$ is idempotent under combination. For further reference we collect a few results on support.

\begin{lemma} \label{le:Supp1} \
\begin{enumerate}
\item For any $\phi \in \Phi$, $x$ is a support of $\epsilon_x(\phi)$,
\item If $x$ is a support of both $\phi$ and $\psi$, then it is also a support of $\phi \cdot \psi$,
\end{enumerate}
\end{lemma}

\begin{proof}
We have $\epsilon_x(\epsilon_x(\phi)) = \epsilon_x(\epsilon_x(\phi) \cdot 1) = \epsilon_x(\phi) \cdot \epsilon_x(1) = \epsilon_x(\phi) \cdot 1 = \epsilon_x(\phi)$, hence $x$ is a support of $\epsilon_x(\phi)$. Further, if $\epsilon_x(\phi) = \phi$ and $\epsilon_x(\psi) = \psi$, then $\epsilon_x(\phi \cdot \psi) = \epsilon_x(\epsilon_x(\phi) \cdot \psi) = \epsilon_x(\phi) \cdot \epsilon_x(\psi) = \phi \cdot \psi$, hence $x$ is a support of $\phi \cdot \psi$.
\end{proof}

The signature $(\Phi,\cdot,0,1;E)$ satisfying the conditions above is called a \textit{domain-free information algebra}. Domain-free, because there is another, related version called a labeled information algebra, see Section \ref{sec:LabInfAlg}. However we shall below (Section \ref{sec:StructOfQuest}) impose some additional conditions on the set $E$ of extraction operators.


\section{Information order} \label{sec:InfOrder}

Pieces of information, that is, elements of an information algebra $\Phi$,  may be ordered by information content. In fact, if $\phi \cdot \psi = \psi$, then this means that $\phi$ adds no information to $\psi$. Therefore we may say that $\phi$ has less information content than $\psi$ and write $\phi \leq \psi$. This is a partial order on $\Phi$, as can easily be verified,
\begin{enumerate}
\item \textit{Reflexivity:} $\phi \leq \phi$.
\item \textit{Antisymmetry} $\phi \leq \psi$ and $\psi \leq \phi$ implies $\phi = \psi$,
\item \textit{Transitivity:} $\phi \leq \psi$ and $\psi \leq \chi$ imply $\phi \leq \chi$.
\end{enumerate}
This order is called the information order. Here are a few simple, immediate consequences of this definition of order.
\begin{enumerate}
\item $\phi,\psi \leq \phi \cdot \psi$,
\item $\phi \leq \psi$ implies $\phi \cdot \eta \leq \psi \cdot \eta$ for all $\eta \in \Phi$,
\item $\epsilon_x(\phi) \leq \phi$ for all $x \in Q$ and $\phi \in \Phi$.
\end{enumerate}

In fact $(\Phi,\leq)$ is a join-semilattice under information order, namely
\begin{eqnarray*}
\phi \cdot \psi = \sup\{\phi,\psi\}.
\end{eqnarray*}
We have $\phi,\psi \leq \phi \cdot \psi$. Let $\chi$ be another upper bound of $\phi$ and $\psi$. Then $\phi \cdot \chi = \chi$ and $\psi \cdot \chi = \chi$ imply by idempotency that $\phi \cdot \psi \cdot \chi = \chi$, hence $\phi \cdot \psi \leq \chi$ and $\phi \cdot \psi$ is indeed the supremum of $\phi$ and $\psi$ in information order. The null element $0$ is the largest element, the unit $1$ the smallest element in information order.

Remark that the the conditions on extraction operators in the previous section may also be written as 
\begin{enumerate}
\item $\epsilon_x(0) = 0$,
\item $\epsilon_x(\phi) \leq \phi$,
\item $\epsilon_x(\epsilon_x(\phi) \cdot \psi) = \epsilon_x(\phi) \cdot \epsilon_x(\psi)$.
\end{enumerate}
In algebraic logic an operator satisfying these properties is called an \textit{existential quantifier} \footnote{Usually Boolean lattices or algebras are considered in algebraic logic, not only join-semilattices, and the converse to our information order is used.}.

Let's note that an extraction operator is monotone in the information order.

\begin{proposition} \label{prop:MonOfExtr}
Any extraction operator $\epsilon_x \in E$ preserves information order.
\end{proposition}

\begin{proof}
Assume $\phi \leq \psi$, that is $\phi \cdot \psi = \psi$. Then since $\epsilon_x(\phi) \leq \phi$ we have $\epsilon_x(\phi) \cdot \epsilon_x(\psi) = \epsilon_x(\epsilon_x(\phi) \cdot \psi) = \epsilon_x(\epsilon_x(\phi) \cdot \phi \cdot \psi) = \epsilon_x(\phi \cdot \psi) = \epsilon_x(\psi)$, so indeed $\epsilon_x(\phi) \leq \epsilon_x(\psi)$.
\end{proof}

Note that $\phi \leq \psi$ in a certain sense says that $\phi$ is implied by $\psi$; if $\psi$ is a piece of information asserted as ''true``, then $\phi$ must also be asserted as ''true`` since $\phi$ is ''part``of $\psi$. So, if $I$ is a subset of $\Phi$ such that for $\phi \in I$ and any $\psi \leq \phi$ we have also $\psi \in I$, and if furthermore $I$ is closed under combination, if $\phi,\psi \in I$, then $\phi \cdot \psi \in I$, we may say that $I$ is a consistent set of pieces of information, with all pieces it contains, it contains also all other pieces implied by them. $I$ is an \textit{ideal} in $\Phi$. If $I$ is different from $\Phi$, then it is called \textit{proper}. The down-set $\downarrow\!\phi =\ \{\psi \in \Phi:\psi \leq \phi\}$ is called \textit{principal ideal}. In some sense an ideal represents also information, and we shall see that ideals form indeed an information algebra, extending $\Phi$ (Section \ref{sec:ExtInfAlg}). In another sense, up-sets are also consistent sets of pieces of information. An up-set of $\Phi$ is a subset of $\Phi$ so that $\phi \in U$ and $\phi \leq \psi$ implies $\psi \in U$. This set is consistent in the sense that with any piece of information it contains, it contains also all other pieces which imply it. However, we should eliminate contradiction $0$ in these up-sets. So let $\Phi_0 = \Phi/\{0\}$ and $U(\Phi_0)$ be the  the family of up-sets in $\Phi_0$ and $U_p(\Phi_0$) the principal up-sets $\uparrow\!(\phi) = \{\psi \in \Phi_0:\psi \geq \phi\}$,  in it. Again we shall see (Section \ref{sec:ExtInfAlg}) that the elements both of $U(\Phi_0)$ and $U_p(\Phi_0)$ form an information algebra, even a particular one, since combination and extraction will be set operations, set intersection for combination and saturation relative to certain partitions for extraction. This means that these algebras will be so-called set algebras (see Section \ref{sec:SetAlg}). Further information algebras derived from an information algebra $\Phi$ will be presented in Section \ref{sec:AtomicAlg}.


\section{Structure of questions: Order and Independence} \label{sec:StructOfQuest}

There is also an order between questions, in the sense that some questions may be finer (or coarser) than others. This order can be defined in terms of extraction. Note that the composition of two extraction operators $\epsilon_x \circ\epsilon_y$ is, in general, no more an extraction operator. But we may have for some $x,y \in Q$ that
\begin{eqnarray*}
\epsilon_x \circ \epsilon_y = \epsilon_y \circ \epsilon_x = \epsilon_x.
\end{eqnarray*}
This condition means that if we extract first information relative to question $y$ and to question $x$ or vice versa, extract first to $x$ and then to $y$, im both cases we get the extraction relative to $x$. This means that question $y$ is finer than question $x$, can carry more information than $x$. Therefore we write $x \leq y$ in this case. This is again obviously a partial order, now between questions, comparing fineness, granularity or coarseness of questions. In Section \ref{sec:SetAlg} important concrete models of questions will be given, confirming these statements. As a consequence of this definition note that 
\begin{eqnarray*}
x \leq y \textrm{ implies}\ \epsilon_x(\phi) \leq \epsilon_y(\phi) \textrm{ for all}\ \phi \in \Phi,
\end{eqnarray*}
where on the right we have information order. In fact, $x \leq y$ means $\epsilon_x(\phi) = \epsilon_y(\epsilon_x(\phi)) \leq \epsilon_y(\phi)$ since $\epsilon_x(\phi) \leq \phi$ and extraction preserves information order.

We write $\epsilon_x \circ \epsilon_y$ also simpler as $\epsilon_x\epsilon_y$. For the sequel we assume that $(Q,\leq)$ is a join-semilattice. That is for any pair $x,y$ we assume that the supremum $\sup\{x,y\} = x \vee y$ exists in $Q$. This imposes some structure on the set $E$ of extraction operators:
\begin{enumerate}
\item For all $x,y \in Q$, an element $z \in Q$ exists such that $\epsilon_x = \epsilon_x\epsilon_z =  \epsilon_z\epsilon_x$ and $\epsilon_y = \epsilon_y\epsilon_z =  \epsilon_z\epsilon_y$ ($z$ is an upper bound of $x$ and $y$).
\item For any $u \in Q$ such that $\epsilon_x = \epsilon_x\epsilon_u =  \epsilon_u\epsilon_x$ and $\epsilon_y= \epsilon_y\epsilon_u =  \epsilon_u\epsilon_y$ we have $\epsilon_z = \epsilon_z\epsilon_u =  \epsilon_u\epsilon_z$ ($z$ is the least upper bound of $x$ and $y$).
\end{enumerate}
We write then $z = x \vee y$. The join of two questions $x$ and $y$ represents the combined question: Answers to question $x \vee y$ are also answers to questions $x$ and $y$, and it is the coarsest question with this property in $Q$. We shall see later that in important instances this is the case, so it seems not be an exaggerated assumption. In the sequel, we assume that in the information algebra $(\Phi,\cdot,0,1;E)$ the set of extraction operators induce a join-semilattice $(Q,\leq)$ in this way, that is satisfies the conditions formulated above. We call this the \textit{Join axiom}.

Here are two further results on support, this time in relation to order of questions.

\begin{lemma} \label{le:Supp2} \
\begin{enumerate}
\item If $x$ is a support of $\phi$ and $x \leq y$, then $y$ is also a support of $\phi$,
\item if $x$ is a support of $\phi$ and $y$ a support of $\psi$, then $x \vee y$ is a support of $\phi \cdot \psi$, and so is $z$, if $x,y \leq z$.
\end{enumerate}
\end{lemma}

\begin{proof}
By definition $x \leq y$ means $\epsilon_x = \epsilon_y\epsilon_x$. So, if $\epsilon_x(\phi) = \phi$, then $\epsilon_y(\phi) = \epsilon_y(\epsilon_x(\phi)) = \epsilon_x(\phi) =\phi$ and so $y$ is a support of $\phi$. According to this result, $x \vee y$ is a support both of $\phi$ and $\psi$, if $x$ is a support of $\phi$ and $y$ of $\psi$. But then by Lemma $\ref{le:Supp1}$ we conclude that $x \vee y$ is a support of $\phi \cdot \psi$. Since $x \vee y \leq z$, if $x,y \leq z$, it follows that $z$ is also a support of $\phi \cdot \psi$.
\end{proof}

In processing information the concept of conditional independence is important. Roughly it means that questions $x$ and $y$ are independent given question $z$, if the extraction for $y$ of an information given for $x$ depends only on the part of this information relative to $z$ and vice versa. Formally this means that
\begin{eqnarray*}
\epsilon_y\epsilon_x &=& \epsilon_y\epsilon_z\epsilon_x, \\
\epsilon_x\epsilon_y &=& \epsilon_x\epsilon_z\epsilon_y.
\end{eqnarray*}
Or, given information to the combined question $x \vee z$ the information extracted from it for the combined question $y \vee z$ depends again only on the part of the first information in $z$, and vice versa, hence, since $z \leq x\vee z,y \vee z$,
\begin{eqnarray*}
\epsilon_{y \vee z}\epsilon_{x \vee z} &=& \epsilon_{y \vee z}\epsilon_z\epsilon_{x \vee z} = \epsilon_{y \vee z}\epsilon_z = \epsilon_z \\
\epsilon_{x \vee z}\epsilon_{y \vee z} &=& \epsilon_{x \vee z}\epsilon_z\epsilon_{y \vee z} = \epsilon_{x \vee z}\epsilon_z. = \epsilon_z
\end{eqnarray*}
Therefore we define the relation $x \bot y \vert z$ and say $x$ and $y$ are conditionally independent given $z$, if and only if
\begin{eqnarray*}
\epsilon_{y \vee z}\epsilon_{x \vee z} &=& \epsilon_z \\
\epsilon_{x \vee z}\epsilon_{y \vee z} &=& \epsilon_z.
\end{eqnarray*}
Note that the concept of conditional independence between questions may be defined without recourse to the join axiom. But this axiom simplifies matters considerably and we shall therefore always assume it. This relation has the following basic properties.

\begin{proposition} \label{prop:QSep}
For $x,y,z,u \in Q$,
\begin{description}
\item [C1] $x \bot y \vert y$,
\item [C2] $x \bot y \vert z$ implies $y \bot x \vert z$,
\item [C3] $x \bot y \vert z$ and $u \leq y$ imply jointly $x \bot u \vert z$,
\item [C4] $x \bot y \vert z$ implies $x \vee z \bot y \vee z \vert z$.
\end{description}
\end{proposition}

\begin{proof}
We have $y = y \vee y \leq x \vee y$, hence $ \epsilon_{y \vee y}\epsilon_{x \vee y} = \epsilon_{y \vee y}\epsilon_y\epsilon_{x \vee y} = \epsilon_y$ and $ \epsilon_{x \vee y}\epsilon_{y \vee y} = \epsilon_{x \vee y}\epsilon_y\epsilon_{y \vee y} = \epsilon_y$ and this means that $x \bot y \vert y$. Item 2 is obvious from the definition of $x \bot y \vert z$. If $u \leq y$, then $z \leq u \vee z \leq y \vee z$, hence $\epsilon_{u \vee z} = \epsilon_{u \vee z}\epsilon_{y \vee z} = \epsilon_{y \vee z}\epsilon_{u \vee z}$. Now $x \bot y \vert z$ means $\epsilon_{y \vee z}\epsilon_{x \vee z} = \epsilon_z$. Hence $\epsilon_{x \vee z}\epsilon_{u \vee z} = \epsilon_{x \vee z}\epsilon_{y \vee z}\epsilon_{u \vee z} = \epsilon_z\epsilon_{u \vee z} = \epsilon_z$, so that $x \bot u \vert z$. The last item follows since $(x \vee z) \vee z = x \vee z$ and $(y \vee z) \vee z = y \vee z$.
\end{proof}

A relation $x \bot y \vert z$ satisfying Proposition \ref{prop:QSep} is called a  \textit{quasi-separoid} (q-separoid). It is a retract of the concept of a separoid, introduced in \cite{dawid01} to represent conditional independence. So in the sequel, we assume that $(Q,\leq,\bot)$ is a q-separoid, describing condition independence among questions. Here follow two important consequences of conditional independence.

\begin{theorem} \label{th:CombExtrProp}
$x \bot y \vert z$ imply for all $\phi,\psi \in \Phi$
\begin{enumerate}
\item $\epsilon_y(\epsilon_x(\phi)) = \epsilon_y(\epsilon_z(\epsilon_x(\phi)))$,
\item $\epsilon_z(\epsilon_x(\phi) \cdot \epsilon_y(\psi)) = \epsilon_z(\epsilon_x(\phi)) \cdot \epsilon_z(\epsilon_y(\psi))$.
\end{enumerate}
\end{theorem}

\begin{proof}
1.) We know that $\epsilon_x(\phi)$ has support $x$. Let $\phi$ be any element with support $x$, hence support $x \vee z$ and $\epsilon_{y \vee z}(\phi) = \epsilon_{y \vee z}(\epsilon_{x \vee z}(\phi))$. Then from $x \bot y \vert z$ we conclude that $\epsilon_{y \vee z}(\phi) =\epsilon_{y \vee z}(\epsilon_z(\phi))$. Then since $y \leq y \vee z$ we have $\epsilon_y(\phi) = \epsilon_y(\epsilon_{y \vee z}(\phi)) = \epsilon_y(\epsilon_{y \vee z}(\epsilon_z(\phi))) = \epsilon_y(\epsilon_z(\phi))$ which proves item 1.)

2.) Again, if $\psi$ has support $y$ is has also support $y \vee z$, $\epsilon_x(\phi)$ has support $x$ and $\epsilon_y(\psi)$ support $y$, so let $\phi$ and $\psi$ have support $x$ and $y$ respectively. Then $\epsilon_{y \vee z}(\phi \cdot \psi) = \epsilon_{y \vee z}(\phi) \cdot \psi$. From $x \bot y \vert z$ and the result just proved it follows further $\epsilon_{y \vee z}(\phi \cdot \psi) = \epsilon_{y \vee z}(\epsilon_z(\phi)) \cdot \psi = \epsilon_{y \vee z}(\epsilon_z(\phi) \cdot \psi)$. Note that the term within parentheses in the last term has support $y \vee z$. Therefore, this last term equals $\epsilon_z(\phi) \cdot \psi$. Then we obtain further, using $z \leq y \vee z$ or $\epsilon_z = \epsilon_z\epsilon_{y \vee z}$,
\begin{eqnarray*}
\epsilon_z(\phi \cdot \psi) = \epsilon_z(\epsilon_{y \vee z}(\phi \cdot \psi))  = \epsilon_z(\epsilon_{y \vee z}(\epsilon_z(\phi) \cdot \psi))= \epsilon_z(\epsilon_z(\phi) \cdot \psi) = \epsilon_z(\phi) \cdot \epsilon_z(\psi)
\end{eqnarray*}
and this concludes the proof.
\end{proof}

If $x \leq y$, then by items 1 and 3 of the q-separoid properties $x \bot y \vert y$ implies $x \bot x \vert y$. Now in our particular case the converse holds too.

\begin{proposition} \label{prop:BasQSep}
If $x \bot x \vert y$, then $x \leq y$.
\end{proposition}

\begin{proof}
$x \bot x \vert y$ means that $\epsilon_{x \vee y} = \epsilon_y$, such that $\epsilon_x = \epsilon_x\epsilon_{x \vee y} = \epsilon_x\epsilon_y$ and $\epsilon_x = \epsilon_{x \vee y}\epsilon_x\ = \epsilon_y\epsilon_x$, hence $x \leq y$.
\end{proof}

A separoid with the property that $x \bot x \vert y$ implies $x \leq y$ is called \textit{basic}, \cite{dawid01} and we adopt this concept for q-separoids. So, our q-eparoid is basic. In certain cases $(Q;\leq)$ may be a lattice, even a distributive one (see Section \ref{sec:SetAlg}). Then we have

\begin{proposition} \label{prop:LCond}
If $(Q,\leq)$ is a lattice and the q-separoid $(Q,\leq,\bot)$ basic, then $x \bot y\vert z$ implies $(x \vee z) \wedge (y \vee z) = z$.
\end{proposition}

\begin{proof}
This is purely a consequence of the q-separoids properties, if the q-separoid is basic. Suppose that $x \bot y \vert z$, so that also $x \vee z \bot y \vee z \vert z$ by C4. Define $w = (x \vee z) \wedge (y \vee z)$ such that $w \leq x \vee z,y \vee z$. Using C3 and C2 we deduce that $w \bot w \vert z$. Since the q-separoid is basic we conclude that $w \leq z$, Since always $z \leq w$ we conclude that $w = z$.
\end{proof}

Independent of this statement, we note that if we define the relation $x \bot_L y \vert z$ iff $(x \vee z) \wedge (y \vee z) = z$, then $x \bot_L y \vert z$ is a q-separoid, if $(Q,\leq)$ is a lattice. This is a theorem purely of q-separoid or separoid  theory, as all the other results below. 
\begin{proposition}
If $(Q,\leq)$ is a lattice, then $x \bot_L y \vert z$ is a q-separoid.
\end{proposition}

\begin{proof}
We have ($x \vee y) \wedge (y \vee y) = y$, hence C1. By the symmetry of the definition C2 holds too. If $u \leq y$, then $z \leq (x \vee z) \wedge (u \vee z) \leq (x \vee z) \wedge (y \vee z) \leq z$, so C3 follows. Finally C4 follows from $(x \vee z) \wedge (y \vee z) = z$.
\end{proof}

For basic q-separoids, Proposition \ref{prop:LCond} can be sharpened.
\begin{proposition} \label{prop:LCond_1}
If $(Q,\leq)$ is a lattice, then a q-separoid $(Q,\leq,\bot)$ is basic if and only if
\begin{eqnarray*}
x \bot y \vert z \Leftrightarrow (x \vee z) \wedge (y \vee z) = z
\end{eqnarray*}
\end{proposition}

\begin{proof}
If the condition on the right holds, then $x \bot x \vert y$ implies $x \vee y = y$, hence $x \leq y$. The other direction of the implication has been shown in Proposition \ref{prop:LCond}.
\end{proof}

A q-separoid becomes a \textit{separoid}, if two additonal conditions are satisfied,
\begin{description}
\item[C5] $x \bot y \vert z$ and $u \leq y$ imply $x \bot y \vert z \vee u$,
\item[C6] $x \bot y \vert z$ and $x \bot u \vert y \vee z$ imply $x \bot y \vee u \vert z$.
\end{description}
If $(Q,\leq)$ is a lattice, and in addition also the next condition holds, then the separoid is called a \textit{strong separoid}.
\begin{description}
\item[C7] If $z \leq y$ and $u \leq y$, then $x \bot y \vert z$ and $x \bot y \vert u$ imply $x \bot y \vert z \wedge u$.
\end{description}
It can be shown that C1 to C3 together with C5 and C6 imply C4 \cite{dawid01}.

If we meet both sides of $(x \vee z) \wedge (y \vee z) = z$ with $x$, we obtain $x \wedge (y \vee z) = x \wedge z$, which is equivalent to
\begin{eqnarray} \label{eq:modCondIndep}
x \wedge (y \vee z) \leq z.
\end{eqnarray}
This condition in turn is equivalent to $(x \vee z) \wedge (y \vee z) = z$ if the lattice $(Q,\leq)$ is \textit{modular}. So, in this case we have $x \bot_L y \vert z$ if and only if (\ref{eq:modCondIndep}) holds. 

\begin{proposition} \label{prop:ModLatt}
If $(Q,\leq)$ is a lattice, then the relation $x \bot_L y \vert z$ defines a separoid if and only if the lattice $(Q,\leq)$ is modular.
\end{proposition}

\begin{proof}
Assume $(Q,\leq)$ to be a modular lattice, that is $x \wedge (y \vee z) = x \wedge z$ if and only if $x \bot_L y \vert z$. So, for C5, if $u \leq y$ we have $x \wedge (z \vee u) \leq  x \wedge (y \vee z \vee u) = x \wedge (y \vee z) = x \wedge z \leq x \wedge (z \vee u)$, hence $x \wedge (y \vee (z \vee u)) = x \wedge (z \vee u)$. This means $x \bot_L y \vert z \vee u$, that is C5. Further $x \bot_L y \vert z$ and $x \bot_L u \vert y \vee z$ imply $x \wedge (y \vee z) = x \wedge z$ and $x \wedge (y \vee z \vee u) = x \wedge (y \vee z)$, hence $x \wedge (y \vee u \vee z) = x \wedge z$, hence $x \bot_L y \vee u \vert z$. This is C6.

On the other hand, assume the relation $x \bot_L y \vert z$ to be a separoid. Then $x \bot_L y\vert x \wedge y$ and therefore, if $z \leq x$ it follows from C5 that $x \bot_L y \vert (x \wedge y) \vee z$. This in turn means $x \wedge (y \vee (x \wedge y) \vee z) = x \wedge ((x \wedge y) \vee z)$. But $x \wedge (y \vee (x \wedge y) \vee z) = x \wedge (y \vee z)$ and $x \wedge ((x \wedge y) \vee z) = (x \wedge y) \vee z$, since $z \leq x$. So $x \wedge (y \vee z) = (x \wedge y) \vee z$ if $z \leq x$ and this is modularity.
\end{proof}

Note that in a distributive lattice $(x \vee z) \wedge (y \vee z) = (x \wedge y) \vee z$. So in this case $(x \vee z) \wedge (y \vee z) = z$ is equivalent to
\begin{eqnarray} \label{eq:CondIndepDistLatt}
x \wedge y \leq z.
\end{eqnarray}
Let's denote the relation defined by this condition by $x \bot_d y \vert z$. If the lattice is distributive, then C7 holds too. 

\begin{proposition} \label{prop:DistLatt}
If $(Q,\leq)$ is a distributive lattice, the relation $x \bot_L y \vert z$ defines a strong separoid.                                                                                                                                                            
\end{proposition}

\begin{proof}
A distributive lattice is modular so that C5 and C6 hold according to the previous proposition. It remains to prove C7.  Since the lattice $Q$ is distributive $x \bot_L y \vert z$ holds if and only if (\ref{eq:CondIndepDistLatt}). Then $x \bot_L y \vert z$ and $x \bot_L y \vert u$ imply $x \wedge y \leq z$ and $x \wedge y \leq u$, hence $x \wedge y \leq z \wedge u$. But this means $x \bot_L y \vert z \wedge u$, hence C7 is satisfied.
\end{proof}

Then the following result is due to \cite{dawid01}
\begin{proposition} \label{}
The relation $x \bot_d y \vert z$ is a strong separoid if and only if $(Q,\leq)$ is a distributive lattice.
\end{proposition}

Many of these results will be illustrated, become concrete form and are related to information algebras in the subsequent sections, especially in the next one. But all these results are important for computational aspects of information algebras, see Section \ref{sec:LocComp}.


\section{Commutative information algebras} \label{sec:CommInfAlg}

Composition $\epsilon_x\epsilon_y$ of extraction operators is, in general, no more an extraction operator. There are however  important cases where for all pairs $\epsilon_x,\epsilon_y \in E$ we have $\epsilon_x\epsilon_y = \epsilon_y\epsilon_x \in E$. Then the extraction operators are said to commute and the information algebra $(\Phi,\cdot,0,1;E)$ is called \textit{commutative}. Section \ref{sec:SetAlg} gives instances of this case, the most important being the so-called multivariate case. Note that if two extraction operators $\epsilon_x$ and $\epsilon_y$ commute, that is $\epsilon_x\epsilon_y = \epsilon_y\epsilon_x = \epsilon_z$ for some $z \in Q$, then $z \leq x,y$. If $u$ is another lower bound of $x,y$, that is $\epsilon_u = \epsilon_u\epsilon_x = \epsilon_u\epsilon_y$, then clearly $\epsilon_u\epsilon_z = \epsilon_u$, hence $u \leq z$. So $z$ is the infimum of $x$ and $y$, $z = \inf\{x,y\} = x \wedge y$. This shows that $E$ is, in this case, a commutative, idempotent semigroup under composition, $(Q,\leq)$ a meet-semilattice and $\epsilon_x\epsilon_y = \epsilon_{x \wedge y}$. An information algebra $(\Phi,\cdot,0,1;E)$, where the set $E$ of extraction operators is a commutative semigroup under composition, $(E,\circ)$, is called a \textit{commutative, domain-free information algebra}. 

For a commutative information algebra, $(Q,\leq)$ is meet-semilattice, as we have seen, but it is not necessarily closed under joins, and we do not need to require this (the Join axiom) for commutative information algebras. Then, there is no conditional independence relation in $Q$ forming a q-separoid. But again in many cases $Q$ is closed under joins, that is $(Q,\leq)$ is a \textit{lattice}. Then we may again define a conditional independence relation $x \bot y \vert z$ by
\begin{eqnarray*}
\epsilon_{y \vee z}\epsilon_{x \vee z} &=& \epsilon_{y \vee z}\epsilon_z\epsilon_{x \vee z} = \epsilon_z.
\end{eqnarray*}
Using commutativity and the fact that composition of extraction operators generates meet, we have in the commutative case equivalently
\begin{eqnarray*}
\epsilon_{(x \vee z) \wedge (y \vee z)} = \epsilon_z\epsilon_{(x \vee z) \wedge (y \vee z)}  = \epsilon_{(x \vee z) \wedge (y \vee z)}\epsilon_z.
\end{eqnarray*}
But this means that $(x \vee y) \wedge (y \vee z) \leq z$, whereas we also always have $(x \vee y) \wedge (y \vee z) \geq z$. So in the case of a commutative algebra, we obtain $x \bot y \vert z$ iff $(x \vee y) \wedge (y \vee z) = z$, that is $x \bot y \vert z = x \bot_L y \vert z$, see Section \ref{sec:StructOfQuest}. 

So we have proved that if $(Q,\leq)$ is a lattice, and the extractions operators in $E$ are commuting, then
\begin{eqnarray*}
x \bot y \vert z \Leftrightarrow (x \vee z) \wedge (y \vee z) = z
\end{eqnarray*}


The converse holds too.
\begin{proposition} \label{prop:CommQSep1}
If $(Q,\leq)$ is a lattice, then the extractor operators in $E$ commute if and only if
\begin{eqnarray*}
x \bot y \vert z \Leftrightarrow (x \vee z) \wedge (y \vee z) = z
\end{eqnarray*}
\end{proposition}

\begin{proof}
The only-if part has been proved above. Assume then that $x \bot y \vert z$ implies $(x \vee z) \wedge (y \vee z) = z$. Then we have $x \bot y \vert x \wedge y$, hence, since $\epsilon_{x \wedge y} = \epsilon_{x \wedge y}\epsilon_y$ and $\epsilon_{x \wedge y} = \epsilon_x\epsilon_{x \wedge y}$,
\begin{eqnarray*}
\epsilon_x\epsilon_y = \epsilon_x\epsilon_{x \wedge y}\epsilon_y = \epsilon_x\epsilon_{x \wedge y} = \epsilon_{x \wedge y}.
\end{eqnarray*}
Since $\epsilon_{x \wedge y} = \epsilon_{y \wedge x}$, we conclude that $\epsilon_x$ and $\epsilon_y$ commute.
\end{proof}

Now, finally it follows that if $(Q,\leq)$ is a lattice, the information algebra is necessarily \textit{commutative}. This follows, since $(Q,\leq,\bot)$ is a basic q-separoid (Proposition \ref{prop:BasQSep}) and if $(Q,\leq)$ is a lattice, then the relation $x \bot y \vert z$ is commutative, that is $x \bot y \vert z = x \bot_L y \vert z$ (Proposition \ref{prop:LCond}). Let's fix this important result in a theorem

\begin{theorem} \label{th:CommLatticInfAlg}
If in an information algebra $(\Phi,\cdot,0,1;E)$ with $E = \{\epsilon_x:x \in Q\}$ the partial order $(Q,\leq)$ is a lattice, then the information algebra is commutative.
\end{theorem}

In the next section , we present a concrete, important instance of an information algebras, including a commutative version of it.


\section{Set algebras} \label{sec:SetAlg}

So far, the set $\Phi$ of pieces of information as well as $Q$, the set of questions have been abstract sets, subject only to the conditions specified for combination and extraction. Now we construct a special type of information algebra, where pieces of information are subsets of some universe, combination is set intersection and extraction is defined by saturation operators relative to some partitions of the universe.  Such information algebras will be called \textit{set algebras}. 

Let $U$ be any set. The basic idea is to consider $U$ as a set of possible worlds and information about an unknown possible worlds is given by subsets of $U$. A piece of information given by a subset $S$ of $U$ tells us that the unknown possible world belongs to $S$. Let $\mathcal{P}(U)$ be the power set of $U$ with the usual lattice structure $(\mathcal{P}(U),\cap,\cup,\emptyset,U)$. A question $x \in Q$ will be modeled by an equivalence relation $\equiv_x$ on $U$, the idea being that for $u,u' \in U$ we have $u \equiv_x u'$ iff question $x$ has the same answer in the possible worlds $u$ and $u'$. Any equivalence relation induces a partition $P_x$ whose blocks $B _x$ are the equivalence classes of the relation $\equiv_x$, so that $u$ and $u'$ belong to the same block $B_x$ iff $u \equiv_x u'$. To an equivalence relation $\equiv_x$ or a partition $P_x$ we associate a \textit{saturation operator} $\sigma_x : \mathcal{P}(U) \rightarrow \mathcal{P}(U)$ defined by
\begin{eqnarray*}
\sigma_x(S) = \{u \in U:\exists u' \in S, \textrm{ such that}\ u \equiv_x u'\}.
\end{eqnarray*}
This is equivalent in terms of the partition $P_x$ to
\begin{eqnarray*}
\sigma_x(S) = \bigcup \{B:B \textrm{ block of}\ P_x, B \cap S \not= \emptyset\}.
\end{eqnarray*}
The following properties of saturation operators will be crucial for our purposes:
\begin{lemma} \label{le:propSatOp} \
\begin{enumerate} 
\item $\sigma_x(\emptyset) = \emptyset$,
\item $S \subseteq \sigma_x(S)$,
\item $\sigma_x(\sigma_x(S) \cap T) = \sigma_x(S) \cap \sigma_x(T)$,
\item $S \subseteq T$ implies $\sigma_x(S) \subseteq \sigma_x(T)$,
\item $S = \sigma_x(S)$ and $T = \sigma_x(T)$ imply $S \cap T =\sigma_x(S \cap T)$,
\item $\sigma_x(S \cup T) = \sigma_x(S)\cup \sigma_x(T)$.
\end{enumerate}
\end{lemma}

\begin{proof}
Items 1, 2,4  and 6 are obvious from the definition of saturation operators.

For 5. observe that $S = \sigma_x(S)$ iff $S$ is a union of whole blocks of partition $P_x$, and that for two blocks $B_1$ and $B_2$ of $P_x$  either $B_1 \cap B_2 = \emptyset$ of $B_1 = B_2$.

For 3. observe that $\sigma_x(S) \cap T \subseteq \sigma_x(S) \cap \sigma_x(T)$, so $\sigma_x(\sigma_x(S) \cap T)  \subseteq \sigma_x(\sigma_x(S) \cap \sigma_x(T)) = \sigma_x(S) \cap \sigma_x(T)$, by 2. and 5. For the reverse inclusion, we have $\sigma_x(S) \cap \sigma_x(T) = \bigcup \{B \in P_x:B \cap S \not= \emptyset \not= B \cap T\}$, where $B \in P_x$ means that $B$ is a block of $P_x$. Obviously, for each such $B$ we have $B \cap \sigma_x(S) = B$, so that $B \cap \sigma_x(S) \cap T \not= \emptyset$ and $B$ participates in the union of all $B' \in P_x$ forming $\sigma_x(\sigma_x(S) \cap T)$. Therefore $\sigma_x(S) \cap \sigma_x(T) \subseteq \sigma_x(\sigma_x(S) \cap T)$.
\end{proof}

As stated above, we consider subsets of $U$, elements of $\mathcal{P}(U)$, as pieces of information about possible worlds. Then, in a natural way, combination of two pieces of information $S,T \in \mathcal{P}(U)$ is given by \textit{set intersection}, $S \cap T$. Note then that in information order (Section \ref{sec:InfOrder}) we have $S \leq T$ if $T \subseteq S$. $T$ limits the unknown possible world more than $S$ does. So information order is the inverse of the usual order in the lattice of subsets given by inclusion. Further the universe $U$ is the unit of combination and the emptyset $\emptyset$ the null element. Given these considerations, we see that by items 1 to 3 of Lemma \ref{le:propSatOp} a saturation operator is an \textit{extraction operator}. Note also that $\sigma_x(\emptyset) = \emptyset$ and that $\sigma_x(S) = \emptyset$ implies $S = \emptyset$.

We study now the order between questions in $Q$ following the discussion in Section \ref{sec:StructOfQuest}. For this purpose we need to study compositions of saturation operator $\sigma_x\sigma_y$ or also corresponding combination of the corresponding relations $\equiv_x\equiv_y$, defined as
\begin{eqnarray*}
u \equiv_x\equiv_y u' = \{(u,u'):\exists u'' \textrm{ such that}\ u \equiv_x u'' \equiv_y u'\}.
\end{eqnarray*}
Note that $\equiv_x\equiv_y$ is, in general, no more an equivalence relation, no more than $\sigma_x\sigma_y$ is a saturation operator. As in Section \ref{sec:StructOfQuest} lets define $x \leq y$ iff $\sigma_x\sigma_y = \sigma_x$ or $\equiv_x\equiv_y\ =\ \equiv_y\equiv_x$. We know from Section \ref{sec:StructOfQuest} that this is a partial order. Now, $\equiv_x\equiv_y\ =\ \equiv_x$ means that $u \equiv_x\equiv_y u'$ iff $u \equiv_x u'$ and from this we conclude that $u \equiv_y u'$ implies $u \equiv_x u'$, that is $P_x \leq P_y$ in partition order.

So a question $y$ is finer than a question $x$, $y \geq x$, if two possible worlds $u$ and $u'$ which have the same answer to question $y$, also have the same answer to question $x$. Or, in yet another view, the set of blocks $B_x$ of a partition $P_x$ represents all possible answers to question $x$. Then $y \leq x$ or $P_y \leq P_x$ means that any possible answer to question $y$ determines also a possible answer to question $x$.  That is, any block of the finer partition $P_y$ is contained in a block of the coarser one \footnote{In the literature this usually is defined as the inverse order $P_x \leq P_y$ \cite{graetzer78}.}. This shows that this way of introducing order between questions makes sense.

A subset $S$ of $U$ is called $x$-saturated if $\sigma_x(S) = S$. The intersection $S \cap T$ of two $x$-saturated sets is still $x$-saturated (see item 5 of Lemma \ref{le:propSatOp}). Note that if $S$ is $x$-saturated and $x \leq y$, then $S$ is also $y$-saturated, since $u \equiv_y u'$ implies $u \equiv_x u'$. So, if a subset $S$ is $x$-saturated and a subset $T$ is $y$-saturated and $x,y \leq z$, then both $S$ and $T$ are $z$-saturated. Further, for any subset $S$, $\sigma_x(S)$ is $x$-saturated.

We need not necessarily consider all possible partitions $P$ of the universe $U$ as questions of interest. As seen above, this set is ordered by the order induced by the saturation operators $\sigma_x$. Let $(Q,\leq)$ be the partial order introduced above and $P_Q = \{P_x:x \in Q)$. We have seen that $x \leq y$ iff $P_x \leq P_y$ in the order defined above. Now, assume that $(Q,\leq)$ is a join-semilattice, $x \vee y$ exists in $(Q,\leq)$ for any pair of elements $x$ and $y$ from $Q$. Then $P_{x \vee y}$ is also the join of $P_x$ and $P_y$ in $(P_Q,\leq)$, written as $P_x \vee P_y$. However, this is, in general, not the join of $P_x$ and $P_y$ in the lattice of partitions ($Part(U),\leq)$, which we denote by $P_x \vee_P P_y$ to distinguish it from the former join. This latter join is the partition whose blocks are exactly the non-empty intersection $B_x \cap B_y$ of blocks $B_x$ from $P_x$ and blocks $B_y$ from $P_y$ \footnote{Again in the inverse order as usually used in the literature our join becomes the meet.}. Obviously, we have $P_x \vee_P P_y \leq P_x \vee P_y$, since the latter join is an upper bound of $P_x$ and $P_y$ in partition order.

We may now define what we mean by a set algebra. Consider a set $U$ (of possible worlds) and a subset  $\Phi \subseteq \mathcal{P}(U)$, that is a family of subsets of $U$ and $Q$ a family of questions represented by equivalence relations $\equiv_x$ in $U$ or, equivalently, by partition $P_x$ of $U$. Let further $\Sigma_Q$ be the set of saturation operators $\sigma_x$ for $x \in Q$. We assume that any $S \in \Phi$ is $x$-saturated for some $x \in Q$ and that $\Phi$ is closed under intersection. This means that if $S$ and $T$ are elements of $\Phi$ which are $x$ and $y$ saturated respectively, there is a $z \in Q$ so that $x,y \leq z$ and $S \cap T$ is $z$ saturated. In other words, we assume that $(Q,\leq)$ is \textit{upwards directed}. By the discussion above, $\Phi$ is closed under combination, if $(Q,\leq)$ is a join-semilattice. This condition is also satisfied, if the top partition of $U$ whose blocks are single elements $\{u\}$ belongs to the family of partition $P_x$ for $x \in Q$.  And $\Phi$ is also closed under all saturation operators $\sigma_x$ for $x \in Q$, since $\sigma_x(S)$ is $x$-saturated. The signature $(\Phi,\cap,\emptyset,U;\Sigma)$ is then an information algebra, called a \textit{set algebra}. We remark that the Support axiom is satisfied in a set algebra by definition.

As explained in Section \ref{sec:StructOfQuest} we can also introduce a relation of conditional independence between questions. Here we assume the Join axiom, so that $(Q,\leq)$ is a join-semilattice. What does it mean in the present case, where questions are represented by equivalence relation $\equiv_x$ or partitions $P_x$? We have $x \bot y \vert z$ iff 
\begin{eqnarray} \label{eq:SatCondIndep}
\sigma_{x \vee z}\sigma_{y \vee z} = \sigma_{y \vee z}\sigma_{x \vee z} = \sigma_z.
\end{eqnarray}
In terms of composition of equivalence relations this is equivalent 
\begin{eqnarray} \label{eq:CondIndep}
\equiv_{x \vee z}\equiv_{y \vee z}\ =\ \equiv_{x \vee z}\equiv_{y \vee z}\ = \ \equiv_z.
\end{eqnarray}
Note that $u \equiv_{x \vee z}\equiv_{y \vee z} u'$ always implies $u \equiv_z u'$, since $u \equiv_{x \vee z} u''$ and $u'' \equiv_{y \vee z} u'$ imply $u \equiv_z u'$. So $x \bot y\vert z$ holds, if the converse of this implication holds also. Therefore this can be expressed as stated in the following proposition

\begin{proposition} \label{prop:PartCondIndep}
For the join-semilattice $(P_Q,\leq)$ induced by partitions as above, we have $x \bot y\vert z$ if and only if
\begin{eqnarray} \label{eq:CondIndepPart}
u \equiv_z u' \Rightarrow \exists w \in U \textrm{ such that}\ u \equiv_{x \vee z} w \equiv_{y \vee z} u'.
\end{eqnarray} 
for any pair $u,u' \in U$.
\end{proposition} 

Since $x \vee_P z \leq x \vee z$ and $y \vee_P z \leq y \vee z$ we have that $u \equiv_{x \vee z} u'$ implies $u \equiv_{x \vee_P z} u'$ and $u \equiv_{y \vee z} u'$ implies $u \equiv_{y \vee_P z} u'$ so that $u \equiv_z u'$ implies that there is an element $w$ so that $u \equiv_{x \vee_P z} w \equiv_{y \vee_P z} u'$. This is the usual definition of conditional independence between partitions in the lattice of partitions $(Part(U),\leq)$ \cite{shafershenoymellouli97,kohlasmonney95}. So, if $B_x$, $B_y$ and $B_z$ are blocks of partitions $P_x$, $P_y$ and $P_z$ respectively, then $B_{x \vee_P z} = B_x \cap B_z$ and  $B_{y \vee_P z} = B_y \cap B_z$ are blocks of partitions $P_x \vee_P P_z$ and $P_y \vee_P P_z$ respectively. Then $P_x$ and $P_y$ are conditionally independent given $P_z$ if and only if $B_x \cap B_z \not= \emptyset$ and $B_x \cap B_z \not= \emptyset$ implies $(B_x \cap B_z) \cap B_y \cap B_z) = B_x \cap B_y \cap B_z \not= \emptyset$. Then we write $P_x \bot P_y \vert P_z$. In summary, we have in a set algebra $x \bot y \vert z$ if and only if $P_x \bot P_y \vert P_z$.

As stated above the product $\equiv_x\equiv_y$ of equivalence relations is, in general, no more an equivalence relation. There is a notable exception \cite{kohlasschmid21}:

\begin{lemma} \label{le:CompEqRel}
Given equivalence relation $\equiv_x$ and $\equiv_y$ for $x,y \in Q$, their relational product $\equiv_x\equiv_y$ is an equivalence relation if and only if the equivalence relations commute, that is $\equiv_x\equiv_y\ =\ \equiv_y\ \equiv_x$.
\end{lemma}

\begin{proof}
Assume $\equiv_x\equiv_y\ = \ \equiv_y\equiv_x$. Since $u \equiv_x u \equiv_y u$ for all $u \in U$, $\equiv_x\equiv_y$ is reflexive. Now $u \equiv_x\equiv_y u'$ iff $u \equiv_y\equiv_x u'$, hence $u' \equiv_x\equiv_y u$. This is symmetry.  It remains to establish transitivity. Assume $u \equiv_x\equiv_y w$ and $w \equiv_x\equiv_y u'$. Then there are elements $s,t \in U$ so that $u \equiv_x s \equiv_y w \equiv_x t \equiv_y u'$, so that $s \equiv_y\equiv_x t$. But then we have also $s \equiv_x\equiv_y t$, that is, there is an element $w'$ such that $u \equiv_x s \equiv_x w' \equiv_y t \equiv_y u'$, hence $u \equiv_x w' \equiv_y u'$ and so $u \equiv_x\equiv_y u'$. This is transitivity.

Conversely assume $\equiv_x\equiv_y$ to be an equivalence relation. Then the relation is symmetric, that is $u \equiv_x\equiv_y u'$ iff $u' \equiv_x\equiv_y u$ for all pairs $u,u' \in U$. But the latter implies $u \equiv_y\equiv_x u'$ so that indeed $\equiv_x\equiv_y\ = \ \equiv_y\equiv_x$.
\end{proof}

Then, obviously, the corresponding saturation operator $\sigma_x$ and $\sigma_y$ commute too under composition and their composition equals $\sigma_{x \wedge y}$ (see Section \ref{sec:CommInfAlg}),
\begin{eqnarray*}
\sigma_x\sigma_y = \sigma_y\sigma_x = \sigma_{x \wedge y}.
\end{eqnarray*}
Now, then $\sigma_{x \wedge y}$ belongs to a partition $P_{x \wedge y}$ and it turns out that this partition is, if the saturation operators commute, the infimum or meet among partitions in the lattice $(Part(U),\leq)$ in the order defined above,
\begin{eqnarray*}
P_{x \wedge y}  = P_x \wedge P_y.
\end{eqnarray*}
This partition $P_{x \wedge y}$ can be characterized as follows: If $B_x$, $B_y$ and $B_{x \wedge y}$ are respectively blocks of $P_x$, $P_y$ and $P_{x \wedge y}$ such that $B_x,B_y \subseteq B_{x \wedge y}$, then $B_x \cap B_y \not = 0$. Such partitions are called commuting (or type I partitions, \cite{graetzer78}). If all saturation operators in $\Sigma = \{\sigma_x:x \in Q\}$ commute pairwise, then the set algebra $(\Phi,\cap,\emptyset,U;\Sigma)$ is called a \textit{commutative set algebra} and it is a commutative information algebra,

The most important case of a commutative set algebra is given by the multivariate model.  Here the universe $U$ is the Cartesian product of domains $U_j$,
\begin{eqnarray*}
U = \prod_{j \in J} U_j.
\end{eqnarray*}
In practical cases $J$ will be countable or even finite. The elements of $U$ are tuples $t : j \in J \mapsto t_j \in U_j$. These tuples are the possible worlds. Define for any tuple $t$ its restriction to a subset $s$ of $J$ by $t \vert s$. Based on this define an equivalence relation in $U$ by
\begin{eqnarray*}
t \equiv_s t' \textrm{ iff}\ t \vert s = t' \vert s.
\end{eqnarray*}
Any such relation defines a partition $P_s$ of $U$ and then an associated saturation operator $\sigma_s$ for any subset $S$ of $U$
\begin{eqnarray*}
\sigma_s(S) = \{t' \in U:\exists t \in S \textrm{ such that}\ t \equiv_s t'\}.
\end{eqnarray*}
This is the so-called \textit{cylindrical exentsion} of $S$ and $s$-saturated sets are also called \textit{cylindrical sets}. Note that $\equiv_s\equiv_r\ = \{(t,t'):\exists t'' \textrm{ such that}\ t \equiv_s t'' \equiv_r t'\}$. Then $\sigma_s = \sigma_r\sigma_s$ holds iff $s \subseteq r$, so $s \leq r$ is simply set inclusion and if $Q$ is the power set of $J$, then $(Q,\leq)$ is a distributive lattice with meet as set intersection and join as set union. Clearly the relations $\equiv_s$ commute for all subsets $s$ and $r$ of $J$. In this case, or if $(Q,\leq)$ is an sublattice of the power set of $J$, this is called a \textit{multivariate model}. For this model we have $s \bot r \vert u$ iff $s \cap r \subseteq u$, see Section \ref{sec:StructOfQuest} and this relation defines a strong separoid (Proposition \ref{prop:DistLatt}).

%% file: chapter3.tex
\chapter{Labeled Information Algebras} \label{sec:LabInfAlg}


\section{Derivation of a labeled information algebras} \label{subsec:LabInfAlg}

In this section another view on an information algebra is presented, stressing more the aspect of questions and  information relative to questions. We derive this alternative form of the algebra from a domain-free information algebra $(\Phi,\cdot,0,1;E)$ with $E =\{\epsilon_x:x \in Q\}$ and $(Q,\leq)$ the join-semilattice derived from $E$. We stated above, that if $x$ is a support of  an element $\phi \in \Phi$, $\phi = \epsilon_x(\phi)$, then it is a piece of information directly bearing on question $x$. Let us therefore collect pairs $(\phi,x)$, where $\phi$ has support $x$ of such pieces of information relating to $x$ and denote the set of these pairs by $\Psi_x$. Define the the set 
\begin{eqnarray*}
\Psi = \bigcup_{x \in Q} \Psi_x
\end{eqnarray*}
of all pairs for all questions. Its elements are called labeled pieces of information. Recall that the null and unit elements $0$ and $1$ have all $x \in Q$ as support.  In $\Psi$ we define the operations of combination and of transport based on the combination and extraction in $\Phi$ and a further operation called labeling.
\begin{enumerate}
\item \textit{Combination:} $(\phi,x) \cdot (\psi,y) = (\phi \cdot \psi,x \vee y)$,
\item \textit{Transport:} $t_y(\phi,x) = (\epsilon_y(\phi),y)$,
\item \textit{Labeling:} $d(\phi,x) = x$.
\end{enumerate}

From these definitions we derive immediately the following basic properties of labeled pieces of information.
\begin{enumerate}
\item \textit{Semigroup} $(\Psi,\cdot)$ is a commutative semigroup,
\item \textit{q-Separoid:} $(Q,\leq,\bot)$ is a q-separoid.
\item \textit{Labeling:} $d((\phi,x) \cdot (\psi,y)) = d(\phi,x) \vee d(\psi,y)$, $d(t_y(\phi,x)) = y$,
\item \textit{Null and Unit:} $(\phi,x) \cdot (0,x) = (0,x)$, $(\phi,x) \cdot (1,x) = (\phi,x)$, $t_y(0,x) = (0,y)$ and $t_y(1,x) = (1,y)$,
\item \textit{Idempotency:} $t_y(\phi,x) \cdot (\phi,x) = (\phi,x \vee y)$,
\item \textit{Combination:} $t_x((\phi,x) \cdot (\psi,y)) = (\phi,x) \cdot t_x(\psi,y)$,
\item \textit{Identity:} $t_x(\phi,x) = (\phi,x)$.
\end{enumerate}
In addition, we have also that $x \bot y \vert z$ implies
\begin{eqnarray} \label{eq:DEfLabCondIndep}
t_{y \vee z}t_{x \vee z} &=& t_{y \vee z}t_zt_{x \vee z}, \nonumber \\
t_{x \vee z}t_{y \vee z} &=& t_{x \vee z}t_zt_{y \vee z}.
\end{eqnarray}
This algebraic system will be called the labeled information algebra derived from the domain-free algebra $(\Phi,\cdot,0,1;E)$. 

We may also define a labeled information algebra independent of a domain-free one. Let, as in the domain-free case, $Q$ be an index set of questions. At this point the set $Q$ has no internal structure whatsoever. The idea is the each piece of information $\psi$ from a set $\Psi$ refers to a question $x \in Q$, which will be its label. As in the domain-free case, we assume that elements of $\Psi$ may be combined or aggregated and that a piece of information $\psi$ refering to some $x \in Q$ may be transport to some other $y \in Q$, or that the part of information refering to $y$ may be extracted from $\psi$. The transported piece of information will then refer to $y$ or be labeled by $y$ And there will be the labeling operation which extracts from each piece of information its label, the question it refers to. So, in summary, we assume the existence of the following operations:
\begin{enumerate}
\item \textit{Combination:} $\cdot :  \Psi \times \Psi \rightarrow ; (\phi,\psi) \mapsto \phi \cdot \psi$,
\item \textit{Transport:} $t : \Psi  \times Q \rightarrow ; (\psi,x) \mapsto t_x(\psi)$,
\item \textit{Labeling:} $d :  \Psi \rightarrow Q; \psi \mapsto d(\psi)$.
\end{enumerate}

Concerning combination, we assume as in the domain-free case that $(\Psi,\cdot)$ is a commutative semigroup. We may also consider all elements of $\Psi$ having a fixed label $x$. Let's denote this set by $\Psi_x = \{\psi \in \Psi:d(\psi) = x\}$. Combination of two pieces of information referring to the same question $x$ should result in a piece of information again referring to $x$. Therefore, $(\Psi_x,\cdot)$ is a sub-semigroups of $(\Psi,\cdot)$. So, if $d(\phi) = d(\psi) = x$, then $d(\phi \cdot \psi) = x$. And as in the domain-free case there must be elements representing vacuous information and contradiction. But now, to keep to the picture of elements each referring to some question $x$, we must assume the existence of unit and null elements $1_x$ and $0_x$ with respect to every semigroup $(\Psi_x,\cdot)$. The transport of vacuous information can not generate information and the transport of contradiction can not eliminate contradiction, so we must have $t_y(1_x) = 1_y$ and $t_y(0_x) = 0_y$.

Now, consider the combination of two elements $\phi$ and $\psi$ referring to two different labels or questions $x$ and $y$. What should be the label of the combination $\phi \cdot \psi$? In the domain-free case we assume that the extraction operators induce a join-semilattice $(Q,\leq)$. And if two elements have support $x$ and $y$, then they also have support $x \vee y$. We have seen that supports in the domain-free case correspond to labels in the labeled view. So it seems to make sense to translate this idea into the labeled view. That is, we impose some requirements on the family $T_Q = \{t_x:x \in Q\}$ of transport operations. As in the domain-free case, we may can consider a question $x$ to be coarser than a question $y$, if $t_x = t_xt_y$. Note however that $t_x = t_yt_x$ makes no sense because the application of the two sides results in different labels. However, if $t_x = t_xt_y$ and $t_y = t_yt_x$, we assume that $x = y$. We call this the symmetry condition. So, we define $x \leq y$ iff $t_x = t_xt_y$. Because of the last condition imposed, the relation is antisymmetric, it is reflexive and also transitive, since $t_x = t_xt_y$ and $t_y = t_yt_z$ imply $t_x = t_xt_z$. So $(Q,\leq)$ becomes a partially ordered set \footnote{Without the symmetry condition, we would have a preorder. Most of what follows, especially conditional independence, would also hold under this weaker condition.}. 

In addition, we may force the existence of a join in this order just as in the domain-free case by requiring the following condition on $T_Q$:
\begin{enumerate}
\item For any pair $x,y \in Q$, there exists a $z \in Q$ such that $t_x =t_xt_z$ and $t_y = t_yt_z$.
\item If for an $u \in Q$, if we have $t_x =t_xt_u$ and $t_y = t_yt_u$, then $t_z = t_zt_u$.
\end{enumerate}
Then the element $z$ is the join of $x$ and $y$, we write $z = x \vee y$. Since the join $x \vee y$ represents the combined question of $x$ and $y$, it makes sense to require that $d(\phi \cdot \psi) = d(\phi) \vee d(\psi)$. This is also valid in the labeled algebra derived from domain-free one.

So, in summary, a labeled information algebra corresponds to a signature $(\Psi,\cdot,d;T)$, where $T = \{t_x:x \in Q\}$, $\Psi_x =\{\psi \in \Psi:d(\psi) = x\}$, subject to the following axioms:
\begin{enumerate}
\item \textit{Semigroup:} $(\Psi,\cdot)$ is a commutative semigroup.
\item \textit{Transport:} 
\begin{enumerate}
\item For all pairs $x,y \in Q$ exists a $z = x \vee y \in Q$ such that $t_x = t_xt_z$ and $t_y = t_yt_z$,
\item for all $u \in Q$, $t_x = t_xt_u$ and $t_y = t_yt_u$ imply $t_z = t_zt_u$,
\item for all pairs $x,y \in Q$, $t_x = t_xt_y$ and $t_y = t_yt_x$ jointly imply $x = y$.
\end{enumerate}
\item \textit{Labeling:} $d(\phi \cdot \psi) = d(\phi) \vee d(\psi)$, $d(t_x(\psi)) = x$.
\item \textit{Unit and Null:} For all $x \in Q$ the semigroups $(\Psi_x,\cdot)$ have a unit element $1_x$ and a null element $0_x$ and for all $x,y \in Q$, $t_y(0_x) = 0_y$ and, if $y \leq x$, then $t_y(1_x) = 1_y$.
\item \textit{Idempotency:} For all $\psi \in \Psi$ and for all $y \in Q$, $t_y(\psi) \cdot \psi = t_{y \vee d(\psi)}(\psi)$.
\item \textit{Combination:} For all $\phi,\psi \in \Psi$ and $x \in Q$, if $d(\phi) = x$, then $t_x(\phi \cdot \psi) = \phi \cdot t_x(\psi)$.
\item \textit{Identity:} For all $x \in Q$ if $d(\psi) = x$, then $t_x(\psi) = \psi$.
\end{enumerate}

Here are a few elementary consequences for further reference for labeled information algebras, derived from the axioms.

\begin{lemma} \label{le:ElPropLabInfAlg}
\begin{enumerate}
\item If $d(\phi) = x \leq y$, then $t_y(\phi) = \phi \cdot 1_y$, 
\item $d(\phi) = x$ and $d(\psi) = y$ imply $\phi \cdot \psi =t_{x \vee y}(\phi) \cdot t_{x \vee y}(\psi)$.
\item if $d(\phi) = x$, then $t_y(\phi) = t_y(t_{x \vee y}(\phi))$,
\item if $d(\phi) = x \leq y$, then $t_x(t_y(\phi)) = \phi$,
\item if $d(\psi) = x \leq y$, then for all $z \in Q$, $t_z(\psi) = t_z(t_y(\psi))$,
\item if $z \geq d(\phi),d(\psi)$, then $t_z(\phi \cdot \psi) = t_z(\phi) \cdot t_z(\psi)$.
\item if $d(\phi) = x$, then $\phi \cdot 0_y = 0_{x \vee y}$.
\end{enumerate}
\end{lemma}

\begin{proof}
We use the axioms above in the proof without explicit reference to them. So, for 1.) we have
\begin{eqnarray*}
t_y(\phi) = t_y(\phi) \cdot 1_y = t_y(\phi \cdot 1_y) = \phi \cdot 1_y.
\end{eqnarray*}
since $d(\phi \cdot 1_y) = x \vee y = y$. In particular, we have $t_y(1_x) = t_y(1_x) \cdot 1_y = 1_y$ (by idempotency) if $x \leq y$. Further by item 1 just proved,
\begin{eqnarray*}
\phi \cdot \psi = t_{x \vee y}(\phi \cdot \psi) = t_{x \vee y}(\phi \cdot (\psi \cdot 1_{x \vee y})) = t_{x \vee y}(\phi) \cdot t_{x \vee y}(\psi).
\end{eqnarray*}
Note that this implies in particular $1_x \cdot 1_y = t_{x \vee y}(1_x) \cdot t_{x \vee y}(1_y) = 1_{x \vee y} \cdot 1_{x \vee y} = 1_{x \vee y}$.
Then, further, if $d(\phi) = x$,
\begin{eqnarray*}
t_y(t_{x \vee y}(\phi)) = t_y(\phi \cdot 1_{x \vee y}) = t_y(\phi \cdot 1_x \cdot 1_y) = t_y(\phi \cdot 1_y) = t_y(\phi) \cdot 1_y = t_y(\phi).
\end{eqnarray*}
This is 3.). Still using 1.) we have, assuming $d(\phi) = x \leq y$,
\begin{eqnarray*}
t_x(t_y(\phi)) = t_x(\phi \cdot 1_y) = \phi \cdot t_x(1_y) = \phi \cdot 1_x = \phi,
\end{eqnarray*}
hence item 4.). For 5.) assume first $y \leq z$. Then $t_z(t_y(\psi)) = t_z(1_y \cdot \psi) = 1_z \cdot 1_y \cdot \psi = 1_z \cdot \psi = t_z(\psi)$. Then using this result, $t_z = t_zt_{x \vee z} = t_zt_{y \vee z}$ and $t_{x \vee z} = t_{x \vee z}t_{y \vee z}$, since $x \leq y \leq y \vee z$,
\begin{eqnarray*}
\lefteqn{t_z(t_y(\psi)) =  t_z(t_{y \vee z}(t_y(\psi))) = t_z(t_{y \vee z}(\psi))  } \\
&&= t_z(t_{x \vee z}(t_{y \vee z}(\psi)) = t_z(t_{x \vee z}(\psi)) = t_z(\psi).
\end{eqnarray*}
So, we have 5. Next we have, if $z \geq d(\phi),d(\psi)$,
\begin{eqnarray*}
t_z(\phi \cdot \psi) = \phi \cdot \psi \cdot 1_z = (\phi \cdot 1_z) \cdot (\psi \cdot 1_z) = t_z(\phi) \cdot t_z(\psi).
\end{eqnarray*}
and thus 6.) holds. Finally if $d(\phi) = x$, using  2.),
\begin{eqnarray*}
\phi \cdot 0_y = t_{x \vee y}(\phi) \cdot t_{x \vee y}(0_y) = t_{x \vee y}(\phi) \cdot 0_{x \vee y} = 0_{x \vee y},
\end{eqnarray*}
and this is 7.) and concludes the proof.
\end{proof} 

As a corollary we add the following important properties of unit and null elements, properties we shall use often without reference to this lemma.

\begin{lemma} \label{le:UnitNull}
For all $x,y \in Q$
\begin{enumerate}
\item $1_x \cdot 1_y = 1_{x \vee y}$,
\item $t_y(1_x) = 1_y$,
\item if $d(\psi) = x$, then $t_y(\psi) = 0_y$ implies $\psi = 0_x$,
\item $0_x \cdot 0_y = 0_{x \vee y}$.
\end{enumerate}
\end{lemma}

\begin{proof}
Item 1 is proved in the previous lemma, and item 2 for $x \leq y$. In the general case we have $t_y(1_x) = t_y(1_x) \cdot 1_y = t_y(1_x \cdot 1_y) = t_y(1_{x \vee y}) = 1_y$. Then, if $y \geq x = d(\psi)$, if $t_y(\psi) = 0$ we have $\psi = t_x(\psi) = t_x(t_y(\psi)) = t_x(0_y) = 0_x$. If $y \leq x$, then $\psi = t_y(\psi) \cdot \psi = t_x(t_y(\psi)) \cdot \psi = t_x(0_y) \cdot \psi = 0_x \cdot \psi = 0_x$. In the general case, $0_y = t_y(\psi) = t_y(t_{x \vee y}(\psi))$, hence $t_{x \vee y}(\psi) = 0_{x \vee y}$, since $y \leq x \vee y$. But then $\psi = t_x(\psi) = t_x(t_{x \vee y}(\psi)) = t_x(0_{x \vee y}) = 0_x$. The last item is a direct consequence of 6.) of the previous lemma.
\end{proof}

We remark, that we may introduce in $Q$ a relation $x \bot y \vert z$ of conditional independence, just as in the domain-free case, by the conditions (\ref{eq:DEfLabCondIndep}) on the transport operation. So we may define $x \bot y \vert z$ if
\begin{eqnarray*}
t_{x \vee z}t_{y \vee z} &=& t_{x \vee z}t_zt_{y \vee z} = t_{x \vee z}t_z, \\
t_{y \vee z}t_{x \vee z} &=& t_{y \vee z}t_zt_{x \vee z} = t_{y \vee z}t_z.
\end{eqnarray*}
Again, this relation $x \bot y \vert z$ defines a q-separoid. C1, C2 and C4 are obvious. For C3 note that, since $u \leq y$ implies $z \leq u \vee z \leq y \vee z$, that by item 5 of Lemma \ref{le:ElPropLabInfAlg}, using $x \bot y \vert z$, $t_{x \vee z}t_{u \vee z} = t_{x \vee z}t_{y \vee z}t_{u \vee z} = t_{x \vee z}t_zt_{u \vee z} = t_{x \vee z}t_z$ and $t_{u \vee z}t_{x \vee z} = t_{u \vee z}t_{y \vee z}t_{x \vee z} = t_{u \vee z}t_{y \vee z}t_z = t_{u \vee z}t_z$, that is $x \bot u \vert z$, hence C3.

We remark that there is the equivalent of Theorem \ref{th:CombExtrProp} in Section \ref{sec:StructOfQuest}.

\begin{theorem} \label{the:CombExtrProplab}
If $x \bot y \vert z$, then
\begin{enumerate}
\item if $d(\phi) = x$, $t_y(\phi) = t_y(t_z(\phi))$,
\item if $d(\phi) = x$ and $d(\psi) = y$, then $t_z(\phi \cdot \psi) = t_z(\phi) \cdot t_z(\psi)$.
\end{enumerate}
\end{theorem}

\begin{proof}
Since $x \bot y \vert z$ we have $t_{y \vee t}(t_z(\phi)) = t_{y \vee z}(\phi)$ and since $y \leq y \vee z$ further $t_y(\phi) = t_y(t_{y \vee z}(\phi)) = t_y(t_{y \vee z}(t_z(\phi))) = t_y(t_z(\phi))$. And $x \bot y \vee z$ implies $x \bot y \vee z \vert z$, so that by item 1.) $t_{y \vee z}(\phi \cdot \psi) = t_{y \vee z}(\phi \cdot \psi) \cdot 1_{y \vee z} = t_{y \vee z}(\phi \cdot (\psi \cdot 1_{y \vee z})) = t_{y \vee z}(\phi) \cdot (\psi \cdot 1_{y \vee z}) = t_{y \vee z}(t_z(\phi)) \cdot (\psi \cdot 1_{y \vee z}) = t_{y \vee z}(t_z\phi)) \cdot t_{y \vee z}(\psi) = t_{y \vee z}(t_z(\phi \cdot \psi))$. From this we obtain
\begin{eqnarray*}
t_z(\phi \cdot \psi) = t_z(t_{y \vee z}(\phi \cdot \psi)) = t_z(t_{y \vee z}(t_z((\phi) \cdot \psi)))= t_z(t_z(\phi) \cdot \psi) = t_z(\phi) \cdot t_z(\psi).
\end{eqnarray*}
This completes the proof.
\end{proof}

In summary, we may then characterize labeled information algebras also in the following way:

\begin{enumerate}
\item \textit{Semigroup:} $(\Psi,\cdot)$ is a commutative semigroup.
\item \textit{Q-separoid:} $(Q,\leq,\bot)$ is a q-aeparoid.
\item \textit{Labeling:} $d(\phi \cdot \psi) = d(\phi) \vee d(\psi)$, $d(t_x(\psi)) = x$.
\item \textit{Unit and Null:} For all $x \in Q$ the semigroups $(\Psi_x,\cdot)$ have a unit element $1_x$ and a null element $0_x$ and for all $x,y \in Q$, $t_y(0_x) = 0_y$ and for $x \leq y$, $t_x(1_y) = 1_x$.
\item \textit{Idempotency:} For all $\psi \in \Psi$ and for all $y \in Q$, $t_y(\psi) \cdot \psi = t_{y \vee d(\psi)}(\psi)$.
\item \textit{Combination:} For all $\phi,\psi \in \Psi$ and $x \in Q$, $t_x(t_x(\phi) \cdot \psi) = t_x(\phi) \cdot t_x(\psi)$.
\item \textit{Independence:} If $x \bot y \vert z$, then for all $\psi \in \Psi$, $t_y(t_x(\psi)) = t_y(t_z(t_x(\psi)))$.
\item \textit{Identity:} For all $x \in Q$ if $d(\psi) = x$, then $t_x(\psi) = \psi$.
\end{enumerate}

In Section \ref{subsec:Dual} we shall show that just as a labeled algebra may be obtained from a domain-free one, conversely, from a labeled information algebra, a domain-free algebra may be derived. But before, we examine the case of a commutative algebra.


\section{Commutative labeled information algebras} \label{subsec:CommLabInfAlg}

As in the domain-free case, we obtain commutative labeled information algebras from an information algebra, if we assume that $(Q,\leq)$ is a lattice and $x \bot_L y \vert z$ iff $(x \vee z) \wedge (y \vee z) = z$. Then we have in particular $x \bot_L y \vert x \wedge y$. Therefore by the Combination axiom, if $d(\phi) = x$, then $t_x(\phi \cdot \psi) = \phi \cdot t_x(\psi)$. Now, if $d(\psi) = y$, then $t_x(\psi) = t_x(t_{x \wedge y}(\psi))$, hence 
\begin{eqnarray*}
t_x(\phi \cdot \psi) = \phi \cdot t_x(t_{x \wedge y}(\psi)) = \phi \cdot t_{x \wedge y}(\psi).
\end{eqnarray*}
This will be the new form of the Combination axiom in the commutative case. In addition, it turns out, that we need not to consider the transport operation $t_x$ in its general forma but only in the limited form of a \textit{projection}, that is,
\begin{eqnarray*}
\textrm{for}\ x \leq d(\psi), \pi_x(\psi) =: t_x(\psi).
\end{eqnarray*}
Note then that, if $x \leq y \leq d(\psi) = z$, since then $z \bot_L y \vert y$ implies $z \bot_L x \vert y$, we have
\begin{eqnarray*}
\pi_x(\psi) = \pi_x(\pi_y(\psi)).
\end{eqnarray*}
This property of stepwise projection will be another axiom for commutative labeled information algebras. Since in this view, we do no more dispose of general transport operations, but only of the partial operation of projection, we can not derive an order in $Q$, but have to assume a priori that $(Q,\leq)$ is a lattice. In summary, we require for the signature $( \Psi,\cdot,\Pi)$, where $\Pi = \{\pi_x:x \in Q\}$, and $\pi : \Psi \times Q \rightarrow \Psi$ is defined for $x \leq d(\psi)$, $(\psi,x) \mapsto \pi_x(\psi)$, the following axioms, where as before $ \Psi_x = \{\psi \in \Psi:d(\psi) = x\}$,
\begin{enumerate}
\item \textit{Semigroup:} $( \Psi,\cdot)$ is a commutative semigroup.
\item \textit{Lattice:} $(Q,\leq)$ is a lattice.
\item \textit{Labeling:} $d(\phi \cdot \psi) = d(\phi) \vee d(\psi)$, $d(\pi_y(\psi)) = y$ if $y \leq d(\psi)$.
\item \textit{Unit and Null:} For all $x \in Q$, the semigroups $(\Psi_x,\cdot)$ have a unit element $1_x$ and a null element $0_x$, and for all $y \leq x \in Q$, if $d(\psi) = x$, $\pi_y(\psi) = 0_y$ if and only if $\psi = 0_x$ , $\pi_y(1_x) = 1_y$ and $1_x \cdot 1_y = 1_{x \vee y}$.
\item \textit{Projection:} If $x \leq y \leq z = d(\psi)$, then $\pi_x(\psi) = \pi_x(\pi_y(\psi))$.
\item \textit{Combination:} If $d(\phi) = x$ and $d(\psi) = y$, then $\pi_x(\phi \cdot \psi) = \phi \cdot \pi_{x \wedge y}(\psi)$.
\item \textit{Idempotency:} If $x \leq d(\psi)$, then $\pi_x(\psi) \cdot \psi = \psi$.
\item \textit{Identity:} If $x = d(\psi)$, then $\pi_x(\psi) = \psi$.
\end{enumerate}

Then $( \Psi,\cdot,\Pi)$ is called a commutative labeled information algebra. Note that projection operators can not commute because of the Labeling axiom. But we shall show in the next section, that nonetheless there is a commutativity in a more general sense. This is an extension the axioms proposed in \cite{shenoyshafer90} for valuation algebras for the multivariate case. However in valuation algebra idempotency is not required, and the existence of null and unit elements are not necessarily assumed. Also the condition that $\pi_y(1_x) = 1_y$, called \textit{stability}, may not hold, even if the existence of unit elements are assumed, for instance in Bayesian networks. We refer to \cite{kohlas03} and Section \ref{sec:nonidempotent} for details about these issues. There are also various alternative axiomatic systems for valuation algebras, especially in the multivariate case, \cite{kohlas03}. In our case stability is essential, as we shall see. There is a strengthening of the Combination axiom in a special case.

\begin{lemma} \label{le:StCombAx}
If $(Q,\leq)$ is a distributive lattice, $d(\phi) = x$, $d(\psi) = y$ and $x \leq z \leq x \vee y$, then
\begin{eqnarray*}
\pi_z(\phi \cdot \psi) = \phi \cdot \pi_{y \wedge z}.
\end{eqnarray*}
\end{lemma}

\begin{proof}
We have by the Labeling axiom $\phi \cdot \psi = \phi \cdot \psi \cdot 1_{x \vee y} = \phi \cdot \psi \cdot 1_z \cdot 1_{x \vee y} = \phi \cdot \psi \cdot 1_z$. Therefore we obtain using the Combination axiom and by distributivity, $x \vee (y \wedge z) = (x \vee y) \wedge (x \vee z) = z$,
\begin{eqnarray*}
\pi_z(\phi \cdot \psi)= \pi_z((\phi \cdot 1_z) \cdot \psi) = (\phi \cdot 1_z) \cdot \pi_{y \wedge z}(\psi) = \phi \cdot \pi_{y \wedge z}(\psi).
\end{eqnarray*}
This concludes the proof.
\end{proof}

We remark that a general commutative domain-free information algebra as defined in Section \ref{sec:CommInfAlg}, has no associated labeled algebra as derived in the previous section. The reason is that for a commutative domain-free information algebra $(Q,\leq)$ is not necessarily a lattice, and then the Labeling axiom can not be valid.

We show now that from a commutative labeled information algebra a labeled information algebra can be reconstructed. This is achieved by recovering the transport operation and it is is done in two steps. First, in addition to the projection operation $\pi_y(\psi)$ defined for labels $y \leq d(\psi)$, we introduce an operation of vacuous extension
$e_yx$ defined for labels $y \geq d(\psi)$,
\begin{eqnarray*}
e_y(\psi) =\psi \cdot 1_y, \textrm{ if}\ y \geq d(\psi).
\end{eqnarray*}
We have, if $d(\psi) = x$, $\pi_x(e_y(\psi)) = \pi_x(\psi \cdot 1_y) = \psi \cdot \pi_x(1_y) = \psi \cdot 1_x = \psi$, hence the extension is indeed vacuous, does not add any information. Note also that here stability is essential. We have also, if $d(\psi) = x \leq y \leq z$ that $e_z(\psi) = \psi \cdot 1_z = \psi \cdot 1_y \cdot 1_z = e_z(e_y(\psi))$. Vacuous extension as projection can be done stepwise. We remark further that if $x \leq y$, then $e_y(0_x) = 0_y$ by the Null axiom, since $\pi_x(e_y(0_x)) = 0_x$. Also, if $d(\psi) = x$, then $\psi \cdot 0_y = (\psi  \cdot 1_{x \vee y}) \cdot (0_y \cdot 1_{x \vee y}) = e_{x \vee y}(\psi) \cdot e_{x \vee y}(0_y) = e_{x \vee y}(\psi) \cdot 0_{x \vee y} = 0_{x \vee y}$.

Then we define the transport operation for any label $x$ as
\begin{eqnarray*}
t_x(\psi) = \pi_x(e_{x \vee y}(\psi)), \textrm{ if}\ d(\psi) = y.
\end{eqnarray*}
Obviously we have $t_y(\psi) = \pi_y(\psi)$, if $y \leq d(\psi)$ and $t_y(\psi) = e_y(\psi)$, if $y \geq d(\psi)$.  Note also that if $d(\psi) = x$ and $x \vee y \leq z$, then 
\begin{eqnarray*}
t_y(\psi) = \pi_y(e_z(\psi)).
\end{eqnarray*}
In fact, $t_y(\psi) = \pi_y(e_{x \vee y}(\psi)) = \pi_x(\pi_{x \vee y}(e_z(e_{x \vee y}(\psi)))) = \pi_x(e_z(\psi))$.

We now have to show that this transport operation satisfies the axioms stipulated for a labeled information algebra, see Section \ref{subsec:LabInfAlg}. Since $d(t_x(\psi)) = d(\pi_x(e_{x \vee y}(\psi)) = x$, we have the Labeling axiom. The Null and Unit axiom follows from $t_y(0_x) = \pi_y(e_{x \vee y}(0_x)) = \pi_y(0_{x \vee y}) = 0_y$ and $t_x(1_y) = \pi_x(1_y) = 1_x$. 
Further, if $d(\psi) = x$, then $t_{x \vee y}\psi) = \psi \cdot 1_{x \vee y} = \psi \cdot 1_{x \vee y} \cdot \pi_y(\psi \cdot 1_{x \vee y}) = \psi \cdot 1_{x \vee y} \cdot t_y(\psi) = \psi \cdot t_y(\psi)$, so Idempotency is valid. Next, assume $d(\phi) = x$ and $d(\psi) = y$ so that $t_x(\phi \cdot \psi) = \pi_x(\phi \cdot \psi \cdot 1_{x \vee y}) = \phi \cdot \pi_x(e_{x \vee y}(\psi)) =  \phi \cdot t_x(\psi)$. This is the Combination axiom. In order to verify the Independence axiom recall that $x \bot_L y \vert z$ if and only if $x \vee z \bot_L y \vee z \vert z$. Assume then $d(\psi) = x \vee z$ so that 
\begin{eqnarray*}
t_{y \vee z}(\psi) &=& \pi_{y \vee z}(\psi \cdot 1_{x \vee y \vee z}) = \pi_{y \vee z}(\psi \cdot 1_{y \vee z}) = \pi_{(x \vee z) \wedge (y \vee z)}(\psi) \cdot 1_{y \vee z} \\
&=& \pi_z(\psi) \cdot 1_{y \vee z} = t_{y \vee z}(t_z(\psi)).
\end{eqnarray*}
 Finally, Identiy is obvious. So, we have reconstructed the labeled algebra from the commutative labeled information algebra.


\section{Duality} \label{subsec:Dual}

As we have seen, from a domain-free information algebra, we may derive a labeled one. This goes also the other way round, which means that the two versions of information algebra are equivalent. So, let $(\Psi,\cdot,T)$ with $T = \{t_x:x \in Q\}$ be a labeled information algebra. Two elements $\phi$ and $\psi$, whatever their labels are, encode the same information if 
\begin{eqnarray*}
t_z(\phi) = t_z(\psi) \textrm{ for all}\ z \in Q.
\end{eqnarray*}
We write then $\phi \equiv_\sigma \psi$. This is clearly an equivalence relation in $\Psi$. If $x$ and $y$ are the labels of $\phi$ and $\psi$ respectively, then $\phi \equiv_\sigma \psi$ imply for $z = x \vee y$,
\begin{eqnarray*}
t_{x \vee y}(\phi) = t_{y \vee x}(\psi)
\end{eqnarray*}
and also
\begin{eqnarray*}
t_y(\phi) = \psi, \quad \phi = t_x(\psi).
\end{eqnarray*}
The former condition $t_{x \vee y}(\phi) = t_{y \vee x}(\psi)$ is in fact equivalent to $t_z(\phi) = t_z(\psi)$ for all $z \in Q$. In fact, $t_{x \vee y}(\phi) = t_{x \vee y}(\psi)$ implies $t_{x \vee y \vee z}(\phi) = t_{x \vee y\vee z}(\psi)$. Then we have also $t_z(\phi) = t_z(t_{x \vee y \vee z}(\phi))$ and similarly $t_z(\psi) = t_z(t_{x \vee y \vee z}(\psi))$, and therefore $t_z(\phi) = t_z(\psi)$.

Now, the relation $\equiv_\sigma$ is not only an equivalence relation, but also a congruence relative to combination and transport. This means that for any pair $\phi$ and $\psi$ in $\Psi$ and $y \in Q$, $\phi \equiv_\sigma \psi$ implies $t_y(\phi) \equiv_\sigma t_y(\psi)$ and $\phi \cdot \chi \equiv_\sigma \psi \cdot \chi$ for any other element $\chi \in \Psi$.

\begin{proposition} \label{prop:SigmaCongr}
The relation $\equiv_\sigma$ is a congruence in the labeled information algebra $\Psi$. 
\end{proposition}

\begin{proof}
Assume $\phi \equiv_\sigma \psi$ and let $d(\phi) = x$, $d(\psi) = y$. Consider any element $\chi$ with $d(\chi) = z$. Then we have, by Lemma \ref{le:ElPropLabInfAlg}, since $x \vee y \vee z \geq d(\phi) \vee d(\chi) = x \vee z$,
\begin{eqnarray*}
t_{x \vee y \vee z}(\phi \cdot \chi) = t_{x \vee y \vee z}(\phi) \cdot t_{x \vee y \vee z}(\chi)
\end{eqnarray*}
and in the same way we obtain
\begin{eqnarray*}
t_{x \vee y \vee z}(\psi \cdot \chi) = t_{x \vee y \vee z}(\psi) \cdot t_{x \vee y \vee z}(\chi).
\end{eqnarray*}
Then $\phi \equiv_\sigma \psi$ implies $t_{x \vee y \vee z}(\phi) = t_{x \vee y \vee z}(\psi)$, and so $t_{x \vee y \vee z}(\phi \cdot \chi) = t_{x \vee y \vee z}(\psi \cdot \chi)$ and this means that $\phi \cdot \chi \equiv_\sigma \psi \cdot \chi$. And $\phi \equiv_\sigma \psi$ implies also $t_y(\phi) = t_y(\psi)$, hence $t_y(\phi) \equiv_\sigma t_y(\psi)$.
\end{proof}

Based on this result, we consider equivalence classes $[\phi]$ of the congruence $\equiv_\sigma$ and define the operations of combination and extraction in the set $\Psi/\sigma$ of these classes,
\begin{enumerate}
\item \textit{Combination:} $[\phi] \cdot [\psi] = [\phi \cdot \psi]$,
\item \textit{Extraction:} $\epsilon_x([\phi]) = [t_x(\phi)]$.
\end{enumerate}
These operations are well defined because $\equiv_\sigma$ is a congruence relative to combination and transport in $\Psi$. It is obvious that $(\Psi/\sigma,\cdot)$ is a commutative semigroup, the class $[0_x]$ is the null element and the class $[1_x]$, for any $x \in Q$, is the unit of combination in $\Psi/\sigma$. In addition, if $d(\phi) = x$, then $\epsilon_x([\phi]) = [\phi]$, so that in particular $\epsilon_y([1_x]) =[1_x]$ and $\epsilon_x([\phi]) = [0_x]$ if and only if $[\phi] = [0_x]$. This shows also that the support axiom is satisfied

The following proposition shows that the operator $\epsilon_x$ is an existential quantiffier with respect to $/\sigma$.

\begin{proposition} \label{prop:QuotExQuant}
Let $(\Psi,\cdot,T)$ be a labeled information algebra. Then in $\Psi/\sigma$ the following holds for all $x \in Q$:
\begin{enumerate}
\item $\epsilon_x([0_x]) = [0_x]$,
\item $\epsilon_x([\phi]) \cdot [\phi] = [\phi]$,
\item $\epsilon_x(\epsilon_x([\phi] \cdot [\psi]) = \epsilon_x([\phi]) \cdot \epsilon_x([\psi]$.
\end{enumerate}
\end{proposition}

\begin{proof}
The first item has been stated above. For the second one, we have $\epsilon_x([\phi]) \cdot [\phi] = [t_x(\phi) \cdot \phi] = [t_{x \vee y}(\phi)]$ if $d(\phi) = y$. Since $t_{x \vee y}(\phi) \equiv_\sigma \phi$, this equals $[\phi]$. The third item follows from the definition of combination and extraction and the Combination axiom for labeled algebras, $\epsilon_x(\epsilon_x([\phi] \cdot [\psi]) = [t_x(t_x(\phi) \cdot \psi)] = [t_x(\phi) \cdot t_x(\psi)] = \epsilon_x([\phi]) \cdot \epsilon_x([\psi]$.
\end{proof}

Al this together shows that $\Psi/\sigma$ is a domain-free information algebra.

\begin{theorem}
If $(\Psi,\cdot,T)$ is a labeled information algebra, then $(\Psi/\sigma,\cdot,[0_x].[1_x];E)$ with $E = \{\epsilon_x:x \in Q\}$ is a domain-free information algebra.
\end{theorem}

So, from a domain-free information algebra $\mathbf{D} = (\phi,\cdot,0,1;E)$ we may obtain a labeled information algebra $\mathbf{L}\mathbf{D} = (\Psi,\cdot,T)$, where $\Psi$ is the set of pairs $(\phi,x)$ such that $\epsilon_x(\phi) =x$, and vice versa from a labeled algebra $\mathbf{L} = (\Psi,\cdot,T)$, we derive a domain-free one $\mathbf{D}\mathbf{L} = (\Psi/\sigma,\cdot,[0],[1];E)$. Now, in this way from a derived labeled algebra $\mathbf{L}\mathbf{D}$ we may retrieve again a domain-free one $\mathbf{D}\mathbf{L}\mathbf{D}$, and similarly, from a derived domain-free algebra $\mathbf{D}\mathbf{L}$ we may retrieve again a labeled one $\mathbf{L}\mathbf{D}\mathbf{L}$. It may be conjectured that the algebras $\mathbf{D}$ and $\mathbf{D}\mathbf{L}\mathbf{D}$ as well as $\mathbf{L}$ and $\mathbf{L}\mathbf{D}\mathbf{L}$ are in some sense the same. This will be discussed in the next section.

%% file: chapter4.tex

\chapter{Some algebraic notions} \label{sec:AlgNotions}

We define in this section the concepts of homomorphism, embedding and isomophisms between two information algebras, as well as the concept of a subalgebra of an information algebra.

Let $(\Phi_1,\cdot,0,1;E_1)$ and $(\Phi_2,\cdot,0,1;E_2)$ be two domain-free information, where $E_i = \{\epsilon^i_x:x \in Q\}$ for $i = 1,2$ are the sets of extraction operators in the two algebras, based on \textit{identical} sets of questions. We do not index the combination operations and the null and unity elements, it will always be clear form the context, which algebra is concerned. 

\begin{definition} \textbf{Homomorphism (domain-free):}
A map $f : \Phi_1 \rightarrow \Phi_2$ is called a domain-free homomorphism, if
\begin{enumerate}
\item $f(\phi \cdot \psi) = f(\phi) \cdot f(\psi)$ for all pairs $\phi,\psi \in \Phi_1$,
\item $f(0) = 0$ and $f(1) = 1$,
\item $f(\epsilon^1_x(\phi)) = \epsilon^2_x(f(\phi))$ for all $\phi \in \Phi_1$ and $x \in Q$.
\end{enumerate}
\end{definition}
If the map $f$ is injective, the homomorphism is called an \textit{embedding}, and if $f$ is bijective, it is called an \textit{isomophim} and the two algebras are called \textit{isomorphic}. Note that the inverse $f^{-1}$ of an isomorphism $f: \Phi \rightarrow \Psi$ is itself an isomorphism $f^{-1} : \Psi \rightarrow \Phi$. This is so, since
\begin{eqnarray*}
f^{-1}(\psi_1 \cdot \psi_2) &=& f^{-1}(f(\phi_1) \cdot f(\phi_2)) = f^{-1}(f (\phi_1 \cdot \phi_2)) = \phi_1\cdot \phi_2 = f^{-1}(\psi_1) \cdot f^{-1}(\psi_2), \\
f^{-1}(0) &=& f^{-1}(f(0)) = 0, f^{-1}(1) = f^{-1}(f(1)) = 1, \\
f^{-1}(\epsilon^2_x(\psi)) &=& f^{-1}(\epsilon^2_x(f(\phi))) = f^{-1}(f(\epsilon^1_x(\phi)))  = \epsilon^1_x(\phi) =
\epsilon^1_x(f^{-1}(\psi)).
\end{eqnarray*}
We do not extend the definition of these concepts to information algebras with different sets of questions.

\begin{definition} \textbf{Subalgebra (domain-free)}:
If $(\Phi,\cdot,0,1,;E)$ is a domain-free information algebra with $E = \{\epsilon_x:x \in Q\}$ and $\Phi'$ a subset of $\Phi$, $Q'$ a subset of $Q$ and $E' = \{\epsilon'_x = \epsilon_x\vert \Phi':x \in Q'\}$, where $\epsilon_x\vert \Phi'$ is the restriction of $\epsilon_x$ to $\Phi'$, such that
\begin{enumerate}
\item $\Phi'$ is closed under combination, $(\Phi',\cdot)$ is a sub-semigroup of $(\Phi,\cdot)$, and $0,1 \in \Phi'$,
\item $\Phi'$ is closed under extraction for $x \in Q'$, that is $\phi \in \Phi'$ and $x \in Q'$ imply $\epsilon'_x(\phi) \in \Phi'$ for all $\epsilon'_x \in E'$
\end{enumerate}
\end{definition}

A subalgebra is still a domain-free information algebra. An example of a subalgebras is given for any $x \in Q$ by the set $\epsilon_x(\Phi) = \{\phi \in \Phi:\epsilon_x(\phi) = \phi\}$ and $Q' = \{y \in Q:y \leq x\}$. Note that $Q'$ is still a q-separoid under the restriction of the relation $x \bot y\vert z$ to $Q'$ if $(Q'.\leq)$ is still a join-semilattice.

The image of $\Phi_1$ under a homomorphism, $(f(\Phi_1),\cdot,0,1;f(E_1))$, where $f(E_1)$ is the set of restrictions of $\epsilon^2_x$ to $f(\Phi_1)$, is a subalgebra of $\Phi_2$ with $Q' = Q$. A homomorphism $f$ preserves order between pieces of information, since $\phi \cdot \psi = \psi$ implies $f(\phi) \cdot f(\psi) = f(\psi)$. It preserves also order between questions in the following sense:  Let $x \leq_1 y$ if $\epsilon^1_x = \epsilon^1_x\epsilon^1_y =  \epsilon^1_y\epsilon^1_x$. Then we have by item 3 of a homomorphism $f$ that $\epsilon^2_x = \epsilon^2_x\epsilon^2_y =  \epsilon^2_y\epsilon^2_x$ \textit{as restricted to} the image of $\Phi_1$, $f(\Phi_1)$. Define $x \leq_2 y$ if $\epsilon^2_x = \epsilon^2_x\epsilon^2_y =  \epsilon^2_y\epsilon^2_x$ \textit{as restricted to} the image of $\Phi_1$, $f(\Phi_1)$, then $x_1 \leq_1 y$ implies $x \leq_2 y$. If $f$ is an isomorphism, $\leq_2$ is the order induced in $\Phi_2$, and then we have $x \leq_1 y$ iff $x \leq_2 y$.

A similar situation we have regarding conditional independence. If $x \bot_1 y \vert z$ is the q-separoid induced by $E_1$ and $x \bot_2 y \vert z$ is the q-separoid induced by $f(E_2)$, then, as for order, $x \bot_1 y \vert z$ implies $x \bot_2 y \vert z$. Note that $x \bot_2 y \vert z$ is not necessarily the same as the conditional independence relation induced by $E_2$.  If $f$ is an isomorphism, then $x \bot_1 y \vert z$ iff $x \bot_2 \vert z$ where the latter is the q-separoid induced by $E_2$, the two relations are identical. Furthermore, for a commutative algebra $\Phi_1$, we have, still by item 3 of a homomorphism $f$, that the subalgebra $f(\Phi_1)$ is also commutative, and if $f$ is an isomorphism, then $\Phi_2$ is commutative too.

For the case of labeled information algebra, we have similar definitions. Let $(\Psi_1,\cdot,T_1)$ and $(_2,\cdot,T_2)$ be two labeled information, based on the \textit{identical} sets of questions, that is $T _1= \{t^1_x:x \in Q\}$ and $T '= \{t^2_x:x \in Q\}$. Again, we do not index the combination operations and the null and unity elements, it will always be clear form the context, which algebra is concerned. 

\begin{definition} \textbf{Homomorphism (labeled):}
A map $f : \Psi_1 \rightarrow  \Psi_2$ is called a labeled homomorphism, if
\begin{enumerate}
\item $f(\phi \cdot \psi) = f(\phi) \cdot f(\psi)$ for all pairs $\phi,\psi \in  \Psi_1$,
\item $f(0_x) = 0_x$ and $f(1_x) = 1_x$ for all $x \in Q$,
\item $f(t^1_x(\phi)) = t^2_x(f(\phi))$ for all $\psi \in  \Psi_1$ and $x \in Q$.
\end{enumerate}
\end{definition}
If the map $f$ is injective, the homomorphism is called an \textit{embedding}, and if $f$ is bijective, it is called an \textit{isomophims} and the two algebras are called \textit{isomorphic}. The concept of a labeled subalgebra is also similar to the one of a domain-free algebra

\begin{definition} \textbf{Subalgebra (labeled)}:
If $(\Psi,\cdot,T)$ is a labeled information algebra, and $\Psi'$ a subset of $\Psi$, $Q'$ a subset of $Q$ and $T' = \{t'_x = t_x\vert  \Psi':x \in Q'\}$, where $t_x\vert  \Psi'$ is the restriction of $t_x$ to $\Psi'$, such that
\begin{enumerate}
\item $Q'$ is closed under joins, $(Q',\leq)$ is a sub-join-semilattice of $(Q,\leq)$,
\item $Q'$ is closed under combination, $( \Psi',\cdot)$ is a sub-semigroup of $( \Psi,\cdot)$, $0_x,1_x \in '$ for all $x \in Q'$,
\item $ \Psi'$ is closed under projection for $x \in Q'$, that is $\psi \in '$ and $x \in Q'$ imply $t'_x(\psi) \in '$ for all $t'_x \in T'$.
\end{enumerate}
\end{definition}

For labeled homomorphisms and subalgebras, similar results hold as for domain-free ones. We do not enter into details.

We examine now the relations between domain-free information algebras $\mathbf{D}$ and $\mathbf{DLD}$ and as well between labeled algebras $\mathbf{L}$ and $\mathbf{LDL}$ (see Section \ref{subsec:Dual}). In the first case we define the following map $f : \mathbf{D} \rightarrow \mathbf{DLD}$,
\begin{eqnarray*}
f(\phi) = [(\phi,x)], \textrm{ if}\ \epsilon_x(\phi) = \phi.
\end{eqnarray*}
Here $[(\phi,x)]$ denotes the equivalence class of the relation $\equiv_\sigma$ in the labeled information algebra $\mathbf{LD}$, see in Section \ref{subsec:Dual}. The map $f$ is well-defined, does not depend on $x$. This is because if $x$ and $y$ are supports of $\phi$, then $t_{x \vee y}(\phi,x) = (\epsilon_{x \vee y}(\phi),x \vee y) = t_{x \vee y}(\phi,y)$ and so $[(\phi,x)] = [(\phi,y)]$.

Similarly, we define a map $g : \mathbf{L} \rightarrow \mathbf{LDL}$,
\begin{eqnarray*}
g(\phi) = ([\phi],x) \textrm{ if}\ d(\phi) = x.
\end{eqnarray*}
Again, $[\phi]$ denotes the equivalence classes of the relation $\equiv_\sigma$, this time in the labeled information algebra $\mathbf{L}$. 

We claim that $f$ and $g$ are domain-free and labeled isomorphisms respectively.

\begin{theorem} \label{th:Duality}
If $\mathbf{D}$ is a domain-free information algebra and $\mathbf{L}$ a labeled information algebra, then $f : \mathbf{D} \rightarrow \mathbf{DLD}$ and $g : \mathbf{L} \rightarrow \mathbf{LDL}$ are domain-free and labeled isomorphisms respectively and correspondingly $\mathbf{D}$ and $\mathbf{DLD}$ as well as $\mathbf{L}$ and $\mathbf{LDL}$ are isomorphic domain-free and labeled information algebras respectively.
\end{theorem}

\begin{proof}
We start with the domain-free case. Consider two elements $\phi$ and $\psi$ from the domain-free algebra $\mathbf{D}$ with support $x$ and $y$ respectively. Then by the definition of $f$ and combination in the different algebras concerned,
\begin{eqnarray*}
\lefteqn{f(\phi \cdot \psi) = [(\phi \cdot \psi,x \vee y)] = [(\phi,x) \cdot (\psi,y)] } \\
&&= [(\phi,x)] \cdot [(\psi,y)] = f(\phi) \cdot f(\psi).
\end{eqnarray*}
Further, $f(0) = [(0,x)]$ and $f(1) = [(1,x)]$ are clearly the null and unit elements of $\mathbf{DLD}$. Next, assume that $y$ is a support of $\phi$. If we denote extraction both in $\mathbf{D}$ and $\mathbf{DLD}$ by $\epsilon_x$, then we have, again by the definition of $f$ and extraction in the different algebras,
\begin{eqnarray*}
f(\epsilon_x(\phi)) = [(\epsilon_x(\phi),x)] = [t_x(\phi,y)] = \epsilon_x([(\phi,y]) = \epsilon_x(f(\phi)).
\end{eqnarray*}
This shows that $f$ is a domain-free homomorphism. Now, if $[(\phi,x)]$ is an element of $\mathbf{DLD}$, then $x$ is a support of $\phi$ and $f$ maps $\phi$ to 
$[(\phi,x)]$, so the map $f$ is surjective. Finally, if $[(\phi,x)] = [(\psi,y)]$, then $t_{x \vee y}(\phi,x) = t_{x \vee y}(\psi)$ and  $x$ and $y$ are supports of $\phi$ and $\psi$ respectively and so $x \vee y$ is a support of both. Therefore, we have $(\phi,x \vee y) = (\epsilon_x(\phi),x \vee y) = (\epsilon_{x \vee y}(\phi),x \vee y) = t_{x \vee y}(\phi,x) = t_{x \vee y}(\psi,y) = (\epsilon_{x \vee y}(\psi),x \vee y) = (\epsilon_y(\psi),x \vee y) = (\psi,x \vee y)$, hence $\phi = \psi$. The map $f$ is injective, hence bijective and therefore an isomorphism.

For the labeled case, we proceed similarly. Consider elements $\phi$ and $\psi$ from $\mathbf{L}$ witth $d(\phi) = x$ and $d(\psi) = y$. Then
\begin{eqnarray*}
\lefteqn{g(\phi \cdot \psi) = ([\phi \cdot \psi],x \vee y) = ([\phi] \cdot [\psi],x \vee y) } \\
&&= ([\phi],x) \cdot ([\psi],y) = g(\phi) \cdot g(\psi).
\end{eqnarray*}
Since $f(0_x) = ([0_x],x)$ and $f(1_x) = ([1_x],x)$, null and unit element are preserved by $f$. Assume further $d(\phi) = x$. Then, if we denote transport both in $\mathbf{L}$ and $\mathbf{LDL}$ by $t_x$,
\begin{eqnarray*}
g(t_x(\phi)) = ([t_x(\phi)],x) = (\epsilon_x([\phi]),x) = t_x([\phi],y) = t_x(g(\phi)).
\end{eqnarray*}
So $f$ is a homomorphism. Any element $([\phi],x)$ in $\mathbf{LDL}$ is the image $f(\phi)$ of some element $\phi$ from $\mathbf{L}$. So $f$ is surjective. If $([\phi],x) = ([\psi],y)$, then $x = y$ and $[\phi] = [\psi]$. By definition of the map $g$, we have $d(\phi) = x = y$ and $d(\psi) = y = x$. But this implies $\phi = \psi$. The map $f$ is injective, hence bijective and therefore a labeled isomorphism. 
\end{proof}

According to this theorem, labeled and domain-free information algebras are dual in a technical sense given by the theorem. We may freely pass from labeled to domain-free algebras and back. The two kinds of algebras are the two sides of the same coin. 

As an application let us consider order in $\mathbf{D}$, $\mathbf{LD}$ and $\mathbf{DLD}$. We have $x \leq_{\mathbf{D}} y$ if $\epsilon_x = \epsilon_x\epsilon_y = \epsilon_y\epsilon_x$. Then, in $\mathbf{LD}$ we have $t_x = t_xt_y$, since $t_x(t_y(\phi,z)) = t_x(\epsilon_y(\phi),y) = (\epsilon_x(\phi),x)) = t_x(\phi,z)$. But this means $x \leq_\mathbf{LD}$ and so $x \leq_\mathbf{D} y$ implies $x \leq_\mathbf{LD} y$. Further $t_x = t_xt_y$ implies $\epsilon_x = \epsilon_x\epsilon_y = \epsilon_y\epsilon_x$ in $\mathbf{DLD}$, since 
\begin{eqnarray*}
t_x(\phi) \equiv_\sigma t_y(t_x(\phi)) \equiv_\sigma t_x(t_y(\phi))
\end{eqnarray*}
if $x \leq_{\mathbf{DLD}} y$. But $\mathbf{D}$ is isomorphic to $\mathbf{DLD}$ so that $x \leq_\mathbf{DLD} y$ implies $x \leq_\mathbf{D} y$, so that the three order relations $\leq_\mathbf{D}$, $\leq_\mathbf{LD}$ and $\leq_\mathbf{DLD}$ are all identical. In the same way we conclude that $\leq_\mathbf{L}$, $\leq_\mathbf{DL}$ and $\leq_\mathbf{LDL}$ are all identical. The same holds for the conditional independence relations $x \bot y \vert z$ in $Q$, induced by the different domain-free and labeled information algebras.

%% file: chapter5.tex
\chapter{Extensions} \label{sec:ExtInfAlg}


\section{Ideal extension} \label{subsec:IdealExt}

In this and the next sections we construct new information algebras derived from a domain-free information algebra $\Phi$, in particular also set algebras in the technical sense defined in Section \ref{sec:SetAlg}. The main result is that information algebras may be embedded into different algebras of sets, that is algebras whose elements are subsets of some universe. But these algebras of sets may, but need not necessarily, be set algebras in the sense of \ref{sec:SetAlg}. This will be the case for the construction presented in  the present subsection.

Consider a domain-free information algebra $(\Phi,\cdot,0,1;E)$ with $E = \{\epsilon_x:x \in Q\}$. Note that information order $\psi \leq \phi$ can also be interpreted as $\psi$ is implied by $\phi$. If $\phi$ can be assured, then so can $\psi$. Now, instead of looking at a particular piece information we consider consistent and complete subsets $I \subseteq \Phi$ of pieces of information. This means that for any element $\phi \in I$, all elements implied by it, that is all $\psi \leq \phi$ belong to $I$, and if $\phi$ and $\psi$ belong to $I$, then $\phi \cdot \psi$ belongs to $I$ to. This says that $I$ is an \textit{ideal} of the join-semilattice $(\Phi,\leq)$, or more formally
\begin{enumerate}
\item $\psi \leq \phi$ and $\phi \in I$ imply $\psi \in I$,
\item $\phi,\psi \in I$ imply $\phi \cdot \psi = \phi \vee \psi \in I$.
\end{enumerate}
The down-set $\downarrow\!\phi$ of all elements less informative than or implied by $\phi$, $\downarrow\!\phi =\{\psi \in \Phi:\psi \leq \phi\}$, forms an ideal, a \textit{principal ideal}. The unit belongs to all ideals and if $\phi$ is in an ideal, then so is $\epsilon_x(\phi)$ for all $x \in Q$. The null element belongs only to the improper ideal $\Phi$. All other ideals, different from $\Phi$, are called \textit{proper ideals}.

An ideal can also be seen as a piece of information. In fact, we may extend the operations of combination and extraction from the algebra $\Phi$ to its set of ideals $I_\Phi$:
\begin{enumerate}
\item \textit{Combination:}
\begin{eqnarray*}
I_1 \cdot I_2 =: \{\phi \in \Phi:\exists \phi_1 \in I_1,\phi_2 \in I_2 \textrm{ such that}\ \phi \leq \phi_1 \cdot \Phi_2\},
\end{eqnarray*}
\item \textit{Extraction:}
\begin{eqnarray*}
\epsilon_x(I) =: \{\phi \in \Phi:\exists \psi \in I \textrm{ such that}\ \phi \leq \epsilon_x(\psi)\}.
\end{eqnarray*}
\end{enumerate}
It can easily be verified that both $I_1 \cdot I_2$ as well as $\epsilon_x(I)$ are ideals, so these operations are well defined. Note that $\Phi$ and $\{1\}$ are the null and unit elements of combination.

It turns out that $I_\Phi$ with these operations is a domain-free information algebra. In order to show this, we need some preparation. First, the intersection of any family of ideals is still an ideal, the family of ideals of an information algebra $\Phi$ forms a $\cap$-system \cite{daveypriestley97}. Therefore, the ideal generated by a family $X$ of elements of $\Phi$, that is, the minimal ideal containing $X$, can be obtained as
\begin{eqnarray*}
I(X) = \bigcap \{I:I \in I_\Phi,X \subseteq I\}.
\end{eqnarray*}
Alternatively, we have also 
\begin{eqnarray*}
I(X) = \{\phi \in \Phi:\exists \phi_1,\ldots,\phi_n \in X,n \geq 1, \textrm{ such that}\ \phi \leq \phi_1 \vee \ldots \vee \phi_n\},
\end{eqnarray*}
since the right hand side is an ideal containing $X$. In particular, we see that $I_1 \cdot I_2 = I(I_1 \cup I_2)$. If $X$ is a finite set, then

\begin{eqnarray*}
I(X) = \downarrow\!\bigvee X,
\end{eqnarray*}
These are well-know results, see for instance \cite{kohlas03}. 

From lattice theory, \cite{daveypriestley97}, we know that a $\cap$-system with a top element ($\Phi$ in our case) forms a complete lattice under set inclusion, infimum is intersection and supremum is given by
\begin{eqnarray*}
\bigvee Y = \bigcap \{I:I \in I_\Phi,\bigcup_{J \in Y} J \subseteq I\}.
\end{eqnarray*}
In particular, we have $I_1 \cdot I_1 = I_1 \vee I_2$, set inclusion is also information order. Now we show that $I_\Phi$ with the operations of combination and extraction as defined above forms an information  algebra.

\begin{theorem} \label{th:IdExt}
Let $I_\Phi$ be the set of ideals of a domain-free information algebra $\Phi$, then $(I_\Phi,\cap)$ is a commutative semigroup with $\{1\}$ as unit and $\Phi$ as null element, and for all $x \in Q$, the operators $\epsilon_x$ are existential quantifers with respect to $I_\Phi$.
\end{theorem}

\begin{proof}
Since $(I_\Phi,\subseteq)$ is a lattice it is a commutative semigroup under combination or join, $\{1\}$ is the smallest and $\Phi$ the greatest ideal, hence the unit and null. It is obvious that $\epsilon_x(\Phi) = \Phi$ and $\epsilon_x(I) \subseteq I$ hence $\epsilon_x(I) \cdot I = I$. Further we must show that $\epsilon_x(\epsilon_x(I_1) \cdot I_2) = \epsilon_x(I_1) \cdot \epsilon_x(I_2)$. Consider an element $\phi \in \epsilon_x(\epsilon_x(I_1) \cdot I_2)$. Then there is an element $\phi' \in \epsilon_x(I_1) \cdot I_2$ such that $\phi \leq \epsilon_x(\phi')$. Further, there are elements $\phi'_1 \in \epsilon_x(I_1)$ and $\phi_2 \in I_2$ so that that $\phi' \leq \phi'_1 \cdot \phi_2$. And there is an element $\phi_1 \in I_1$ so that $\phi'_1 \leq \epsilon_x(\phi_1)$. So finally we have $\phi \leq \epsilon_x(\epsilon_x(\phi_1) \cdot \phi_2) = \epsilon_x(\phi_1) \cdot \epsilon_x(\phi_2)$, since $\epsilon_x$ is an existential quantifier in $\Phi$. But this shows that $\phi \in \epsilon_x(I_1) \cdot \epsilon_x(I_2)$. Conversely, assume $\phi \in \epsilon_x(I_1) \cdot \epsilon_x(I_2)$. Then there are elements $\phi_1 \in I_1$ and $\phi_2 \in I_2$ such that $\phi \leq \epsilon_x(\phi_1) \cdot \epsilon_x(\phi_2) = \epsilon_x(\epsilon_x(\phi_1) \cdot \phi_2)$. But this means that $\phi \in \epsilon_x(\epsilon_x(I_1) \cdot I_2)$, and this proves the required identity. So the operators $\epsilon_x$ are existential quantifiers relative to $(I_\Phi),\cap)$.
\end{proof}

This indicates that ideals $I_\Phi$ form a domain-free information algebra, a kind of reduct of it. We call this algebra $I_\Phi$ the \textit{ideal completion} of $\Phi$. The Support condition (see Section \ref{sec:InfAlg}) however is a different story. If the join lattice $(Q,\leq)$ does not have a greatest element, then an ideal may have no support in $Q$. 

Now we show that the ideal algebra $I_\Phi$ is an extension of the information algebra $\Phi$, or, in other words, $\phi$ is embedded in $I_\Phi$.

\begin{theorem}
Let $(\Phi,\cdot,0,1;E)$ be a domain-free information algebra. Then the map $f : \Phi \rightarrow I_\Phi$ defined by $h(\phi) = \downarrow\!\phi$ is an embedding.
\end{theorem}

\begin{proof}
We have obviously by the definition of combination and extraction among ideals that $\downarrow\!(\phi \cdot \psi) = \downarrow\!\phi\ \cdot \downarrow\!\psi$ and $\epsilon_x(\downarrow\!\phi) = \epsilon_x(\downarrow\!\phi)$. Further $\downarrow\!0 = \Phi$ and $\downarrow\!1 = \{1\}$. So $h$ is a homomorphism. And $\downarrow\!\phi = \downarrow\!\psi$ implies that $\phi = \psi$, hence $f$ is an embedding.
\end{proof}

The embedding $f(\Phi)$ in $I_\Phi$ is a domain-free information algebra, in particular satisfying the support axiom, even if $I_\Phi$ does not. Often we identify $I_\Phi$ with $\Phi$, so that, in this view, $I_\Phi$ is an extension of $\Phi$. 

We remark further that the order in $Q$ induced by the information algebra $\Phi$ is the same as the one induced by the ideal completion. This is a consequence of the following lemma-

\begin{lemma}
If $\epsilon_x(\phi) = \epsilon_x(\epsilon_y(\phi)) = \epsilon_y(\epsilon_x(\phi))$ for all $\phi \in \Phi$, then $\epsilon_x(I) =\epsilon_x(\epsilon_y(I)) = \epsilon_y(\epsilon_x(I))$ for all ideals $I \in I_\Phi$.
\end{lemma}

\begin{proof}
Note that $\epsilon_x(I)  \supseteq \epsilon_x(\epsilon_y(I)),\epsilon_y(\epsilon_x(I))$. Assume now that $\epsilon_x = \epsilon_x\epsilon_y = \epsilon_y\epsilon_x$. Consider an element $\phi \in \epsilon_x(I)$ for any ideal $I$. Then there is an element $\phi' \in I$ such that $\phi \leq \epsilon_x(\phi') =  \epsilon_x(\epsilon_y(\phi')) = \epsilon_y(\epsilon_x(\phi'))$. But this means that $\phi \in \epsilon_x(\epsilon_y(I)),\epsilon_y(\epsilon_x(I))$ and this establishes the equality claimed in the lemma.
\end{proof}

As a further consequence we conclude that the Join axiom is satisfied in $I_\Phi$ if it is in $\Phi$.

There is another view of ideal extension, more in the spirit of a logical calculus. In fact, the operator $I(X)$ is a consequence operator on the sets of elements of the domain-free information algebra $\Phi$ so that the ideal algebra can also be seen as a logical calculus \cite{kohlas03,daveypriestley97}. To conclude this section, lets mention that $I_\Phi$ may contain maximal ideals, that is ideals different from $\Phi$ but contained in no other ideal. Such maximal ideals are atoms of the information algebra $I_\Phi$, see Section \ref{sec:AtomicAlg}. We return to this issue in Section \ref{subsec:BooleInfAlg}. Note also that $I_\Phi$, although being an algebra of sets is not a set algebra in the formal sense of Section \ref{sec:SetAlg}, since combination is not set intersection and extraction is not saturation in the set theoretical sense.


\section{Up-set algebra extension}

Instead of considering sets of pieces of information which contain together with an element all its implied elements, that is, ideals, we may, alternatively, also consider sets which contain with an element all other element which imply this element. This are \textit{up-sets} in the partially ordered set $(\Phi,\leq)$ associated with a domain-free information algebra $(\Phi,\cdot,0,1,E)$. More precisely, a subset $U$ of $\Phi$ is called an up-set if $\phi \in U$ and $\phi \leq \psi$ jointly imply $\psi \in U$. It seems however not reasonable to consider the null element as implying an element $\phi$. Therefore we consider up-sets $U$ in $\Phi_0 = \Phi/\{0\}$. Let $U(\Phi_0)$ denote the family of these up-sets. The up-sets $\uparrow\!\phi = \{\psi \in \Phi_0:\phi \leq \psi\}$ are called principal up-sets in $\Phi_0$. The family of principal up-sets in $\Phi_0$ is denoted by $U_p(\Phi_0)$. Now, we consider the families $U^+(\Phi_0) = U(\Phi_0) \cup \{\emptyset\}$ and $U^+_p(\Phi_0) = U_p(\Phi_0) \cup \{\emptyset\}$ and construct set algebras of these subsets of the universe $\Phi$.

Consider in $\Phi_0$ the equivalence relations $\phi \equiv_x \psi$ defined for any $x \in  Q$ by $\epsilon_x(\phi) = \epsilon_x(\psi)$, for $\phi,\psi \not= 0$. They induce corresponding partitions $P_x$ in $\Phi_0$. Let $P_Q = \{P_x:x \in Q\}$ the family of these partitions associated with questions $x \in Q$. Based on $\Phi_0$ we can construct a set algebra. Consider $U^+(\Phi_0)$ and $U^+_p(\Phi_0)$. Both of these families of subsets of $\Phi_0$ are clearly closed under set intersection. This will be combination in the set algebra we construct. Then information order is the inverse of inclusion, therefore $\Phi_0$ and $\emptyset$ are the smallest and largest element in these families, the unit and null of combination. So $(U^+(\phi_0),\cap,\emptyset,\Phi_0)$ and  $(U^+_p(\phi_0),\cap,\emptyset,\Phi_0)$ are both commutative semigroups with null and unit elements.

As usual, we denote the saturation operator associated with partition $P_x$ or equivalence relation $\equiv_x$ in $\Phi_0$ by $\sigma_x$ as an abbreviation for $\sigma_{P_x}$ and let $\Sigma_Q$ denote the set of all saturation operators $\sigma_x$ for $x \in Q$.

\begin{proposition} \label{prop:SatofUpSets}
$U^+(\Phi_0)$ and $U^+_p(\Phi_0)$ are both closed under the application of the operations $\sigma_x \in \Sigma_Q$.
\end{proposition}

\begin{proof}
For an up-set $U$ in $U(\Phi_0)$ we have
\begin{eqnarray*}
U = \bigcup_{\phi \in U} \uparrow\!\phi
\end{eqnarray*}
and for any saturation operator $\sigma_x$
\begin{eqnarray*}
\sigma_x(U) = \sigma_x(\bigcup_{\phi \in U} \uparrow\!\phi) =\ \bigcup_{\phi \in U} \sigma_x(\uparrow\!\phi).
\end{eqnarray*}
We show that $\sigma_x(\uparrow\!\phi)$ belongs to $U_p(\Phi_0)$, hence to $U(\Phi_0)$, for all $\phi$ in $\Phi_0$. We have $\phi \equiv_x \epsilon_x(\phi)$, since $\epsilon_x(\phi) = \epsilon_x(\epsilon_x(\phi))$.  Let $\psi \geq \epsilon_x(\phi)$ and consider $\chi = \phi \vee \epsilon_x(\psi) \in\ \uparrow\!\phi$. Then we obtain $\epsilon_x(\chi) = \epsilon_x(\phi \vee \epsilon_x(\psi)) = \epsilon_x(\phi) \vee \epsilon_x(\psi)$. But $\psi \geq \epsilon_x(\phi)$ implies $\epsilon_x(\psi) \geq \epsilon_x(\phi)$, so we get $\epsilon_x(\chi) = \epsilon_x(\psi)$, that is $\chi \equiv_x \psi$ and thus $\psi \in \sigma_x(\uparrow\!\phi)$. Conversely, if $\psi \in \sigma_x(\uparrow\!\phi)$, then for some element $\chi \geq \phi$, we have $\psi \equiv_x \chi$ and so $\psi \geq \epsilon_x(\psi) = \epsilon_x(\chi) \geq \epsilon_x(\phi)$, hence $\psi \in \uparrow\!\epsilon_x(\phi)$. Summing up, we see that

\begin{eqnarray} \label{eq:HomoExtr}
\sigma_x(\uparrow\!\phi) =\ \uparrow\! \epsilon_x(\phi)
\end{eqnarray}
so indeed $\sigma_x(\uparrow\!\phi) \in U_p(\Phi_0)$. From this result we obtain
\begin{eqnarray*}
\sigma_x(U) = \bigcup_{\phi \in U} \uparrow\!\epsilon_x(\phi).
\end{eqnarray*}
This is an up-set and this concludes the proof.
\end{proof}

This proposition shows that $(U^+(\Phi_0),\cap,\emptyset,\Phi_0;\Sigma_Q)$ and  $(U^+_p(\Phi_0),\cap,\emptyset,\Phi_0;\Sigma_Q)$ are both set algebras, the latter a subalgebra of the former. 

What are the connections between the information algebra $\Phi$ and the set algebras  $U^+(\Phi_0)$ and $U^+_p(\Phi_0)$? Consider first the map $f : \Phi \rightarrow U^+_p(\Phi_0)$ defined by $f(\phi) =\ \uparrow\!\phi$, if $\phi \not= 0$ and $f(0) = \emptyset$. This map preserves combination and extraction as the following proposition shows.

\begin{proposition} \label{prop:InfAlgIsom}
If $\Phi$ is an domain-free information algebra, then the map $f(\phi) =\ \uparrow\!\phi$ and $f(0) = \emptyset$ between $\Phi$ and $U^+(\Phi_0)$, defined above satisfies the following:
\begin{enumerate}
\item $f(\phi \cdot \psi) = f(\phi) \cap f(\psi)$,
\item $f(0) = \emptyset$, $f(1) = \Phi_0$,
\item $f(\epsilon_x(\phi)) = \sigma_x(f(\phi))$.
\end{enumerate}
\end{proposition}

\begin{proof}
Since $\chi \geq \phi \cdot \psi = \phi \vee \psi$ if and only if $\chi \geq \phi$ and $\chi \geq \psi$, we have $\uparrow\!(\phi \cdot \psi)\ =\ \uparrow\!\phi\ \cap \uparrow\!\psi$. This is item 1. Item 2 is obvious and item 3 is proved in proposition \ref{prop:SatofUpSets}, see (\ref{eq:HomoExtr}) for $\phi \not= 0$. We have $\sigma_x(\emptyset) = \emptyset$, hence for $\phi = 0$,  $f(\epsilon_x(0)) =  f(0) = \sigma_x(f(0))$.
\end{proof} 

The map is obviously also bijective on $U^+_p(\Phi_0)$. So it is an embedding of $\Phi$ in a set algebra $U^+_p(\Phi_0)$. However, what is the relation between the set of questions $Q$ and the family of partitions $P_Q$? We have the following result.

\begin{proposition} \label{prop:IsoQSep}
For any domain-free information algebra $\Phi$, $x \leq y$ in $Q$ if and only if $P_x \leq P_y$. Furthermore $x \bot y \vert z$ if and only if $P_x \bot P_y \vert P_x$ with the usual conditional independence relation between partitions.
\end{proposition}

\begin{proof}
Suppose $x \leq y$, that is $\epsilon_x = \epsilon_x\epsilon_y$, and assume $\phi \equiv_y \psi$, hence $\epsilon_x(\phi) = \epsilon_x(\epsilon_y(\phi)) = \epsilon_x(\epsilon_y(\psi)) = \epsilon_x(\psi)$ and so $\phi \equiv_y \psi$ implies $\phi \equiv_x \psi$. This shows that for the associated partitions we have $P_x \leq P_y$. 

On the other hand assume $P_x \leq P_y$. Then $\sigma_y(\uparrow\!\phi) \subseteq \sigma_x(\uparrow\!\phi)$ for all $\phi \in \Phi$, hence $\uparrow\!\epsilon_y(\phi) \subseteq\ \uparrow\!\epsilon_x(\phi)$ and therefore $\epsilon_x(\phi) \leq \epsilon_y(\phi)$ for all $\phi \in \Phi$. If we apply $\epsilon_x$ to this inequality, we obtain $\epsilon_x(\epsilon_x(\phi)) = \epsilon_x(\phi) \leq \epsilon_x(\epsilon_y(\phi)) \leq \epsilon_x(\phi)$, so $\epsilon_x = \epsilon_x\epsilon_y$ and $\epsilon_y(\epsilon_x(\phi)) \geq \epsilon_x(\epsilon_x(\phi)) = \epsilon_x(\phi)$, but also $\epsilon_x(\phi) \geq \epsilon_y(\epsilon_x(\phi))$, hence $\epsilon_x = \epsilon_y\epsilon_x$. It follows therefore finally $x \leq y$.

Next assume $x \bot y \vert z$. We want to show that then $P_x \vee P_z \bot P_y \vee P_z \vert P_z$, which implies $P_x \bot P_y \vert P_z$, since the partitions form a q-separoid. Consider therefore blocks $[\epsilon_{x \vee z}(\phi_1)]$, $[\epsilon_{y \vee z}(\phi_2)]$ and $[\epsilon_z(\phi)]$ of these partitions so that
\begin{eqnarray*}
[\epsilon_{x \vee z}(\phi_1)] \cap [\epsilon_z(\phi)] \not= \emptyset, \quad [\epsilon_{y \vee z}(\phi_2)] \cap [\epsilon_z(\phi)] \not= \emptyset.
\end{eqnarray*}
Then there are  elements $\phi'_1$ and $\phi'_2$ so that $\phi'_1 \equiv_{x \vee z} \phi_1$, $\phi'_1 \equiv_z \phi$ and  $\phi'_2 \equiv_{y \vee z} \phi_2$, $\phi'_2 \equiv_z \phi$. Thus $\phi'_1 \equiv_z \phi'_2$. Define $\phi' = \epsilon_{x \vee z}(\phi'_1) \cdot \epsilon_{y \vee z}(\phi'_2)$. Now $x \bot y \vert z$ implies $x \vee z \bot y \vee z \vert z$, and using this, it follows that
\begin{eqnarray*}
\lefteqn{\epsilon_{x \vee z}(\phi') = \epsilon_{x \vee z}(\phi'_1) \cdot \epsilon_{x \vee z}(\epsilon_{y \vee z}(\phi'_2)) } \\
&&=  \epsilon_{x \vee z}(\phi'_1) \cdot \epsilon_{x \vee z}(\epsilon_z(\epsilon_{y \vee z}(\phi'_2))) = \epsilon_{x \vee z}(\phi'_1) \cdot \epsilon_{x \vee z}(\epsilon_z(\phi'_2))= \\
&&= \epsilon_{x \vee z}(\phi'_1) \cdot \epsilon_{x \vee z}(\epsilon_z(\phi'_1) )= \epsilon_{x \vee z}(\phi'_1).
\end{eqnarray*}
Similarly we obtain $\epsilon_{y \vee z}(\phi') = \epsilon_{y \vee z}(\phi'_2)$. Furthermore, $\epsilon_z(\phi') = \epsilon_z(\epsilon_{x \vee z}(\phi')) = \epsilon_z(\epsilon_{x \vee z}(\phi'_1)) = \epsilon_z(\phi'_1)$. So we see that $\phi' \equiv_{x \vee z} \phi_1$, $\phi' \equiv_{y \vee z} \phi_2$ and $\phi' \equiv_z \phi$ or $[\epsilon_{x \vee z}(\phi_1)] \cap [\epsilon_{y \vee z}(\phi_2)] \cap [\epsilon_z(\phi)] \not= \emptyset$. But this means $P_x \vee P_z \bot P_y \vee P_z \vert P_z$, which implies $P_x \bot P_y \vert P_z$. So, indeed $x \bot y \vert z$ implies $P_x \bot P_y \vert P_z$.

Conversely, $P_x \bot P_y \vert P_z$ implies $\sigma_x\sigma_y = \sigma_x\sigma_z\sigma_y$ and $\sigma_y\sigma_x = \sigma_y\sigma_z\sigma_x$. But by Proposition \ref{prop:InfAlgIsom}, if we restrict the operators $\sigma_x$ to $U_p(\Phi_0)$, this entails $\epsilon_x\epsilon_y = \epsilon_x\epsilon_z\epsilon_y$ and $\epsilon_y\epsilon_x = \epsilon_y\epsilon_z\epsilon_x$, hence $x \bot y \vert z$. This concludes the proof.
\end{proof}

According to this proposition, the map $x \mapsto P_x$ and its inverse preserve order. This implies that joins map to joins.  As a warning let's stress that although $(P_Q,\leq)$ is a join-semilattice, it is not a sub-join-semilattice of $(Part(U_0),\leq)$ in general, hence the q-separoid is, in general, not a sub-q-separoid of the separoid $(Part(U_0),\leq,\bot)$ of all partitions.


%% file: chapter6.tex

\chapter{Atoms} \label{sec:AtomicAlg} \label{sec:AtomicAlg}


\section{Atomistic algebras} \label{subsec:AtomicAlg}

A domain-free information algebra $(\Phi,\cdot,0,1;E)$ with $E = \{\epsilon_x:x \in Q\}$, may have maximal elements, different from $0$. Such elements will be called atoms, since, as we shall show, in a certain sense, information algebras may be built up from atoms. We start with the definition of an atom.

\begin{definition} \textbf{Atom:}
An element $\alpha \in \Phi$ from a domain-free information algebra $(\Phi,\cdot,0,1;E)$ is called an atom, if
\begin{enumerate}
\item $\alpha \not= 0$,
\item $\phi \geq \alpha$ implies either $\phi = \alpha$ or $\phi = 0$.
\end{enumerate}
\end{definition}
In general, in order theory, atoms are defined as minimal elements. But in our information order, maximal elements are more interesting. The following will justify this view. There are a few alternative, equivalent definitions.

\begin{lemma} \label{le:ElPropAtoms}
The following are equivalent statements:
\begin{enumerate}
\item $\alpha$ is an atom,
\item $\alpha \cdot \phi = \alpha$ or $= 0$ for all $\phi \in \Phi$,
\item $\phi \leq \alpha$ or $\phi \cdot \alpha = 0$ for all $\phi \in \Phi$.
\end{enumerate}
\end{lemma}

\begin{proof}
If $\alpha$ is an atom, then from $\alpha \leq \alpha \cdot \phi$ we conclude that either $\alpha \cdot \phi = \alpha$ or $= 0$. In the first case we have $\phi \leq \alpha$. So we have $(1) \Rightarrow (2) \Rightarrow (3)$.  Assume (3) and consider an element $\phi$ so that $\alpha \leq \phi$, then either $\alpha = \phi$ or $\phi \cdot \alpha = \phi = 0$, hence $\alpha$ is an atom. This concludes the proof.
\end{proof}
 Another simple result is that atoms are contradictory among themselves. That is, it $\alpha$ and $\beta$ are atoms, then either $\alpha = \beta$ or $\alpha \cdot \beta = 0$. In fact, from (2) of Lemma \ref{le:ElPropAtoms}, we obtain either $\alpha \cdot \beta = 0$ or $\alpha \cdot \beta = \alpha$ and $\alpha \cdot \beta = \beta$, so that in this case $\alpha = \beta$.

Atoms, if they exist, represent the most precise pieces of information in the algebra. Let $At(\Phi)$ the set of all atoms of $\Phi$. Note that $At(\Phi)$ may be empty. Further, for any element $\phi$ of $\Phi$, define 
\begin{eqnarray*}
At(\phi) = \{\alpha \in At(\Phi):\phi \leq \alpha\},
\end{eqnarray*}
the set of all atoms implying $\phi$. We say also $\alpha$ is contained in $\phi$ if $\alpha \in At(\phi)$.. This motivates the following definition.

\begin{definition} \textbf{Atomistic information algebras:}
Let $(\Phi,\cdot,0,1;E)$ be a domain-free information algebra.
\begin{enumerate}
\item If for all $\phi \in \Phi$, $\phi \not= 0$, the set $At(\phi)$ is not empty, then the algebra $\Phi$ is called \textbf{atomic}.
\item If $\Phi$ is atomic and for all $\phi \in \Phi$,
\begin{eqnarray*}
\phi = \bigwedge At(\phi),
\end{eqnarray*}
the algebra $\Phi$ is called \textbf{atomistic}.
\item If $\Phi$ is atomistic and for any subset $A$ of $At(\phi)$, the infimum $\bigwedge A$ exists in $\Phi$, the algebra $\Phi$ is called \textbf{atomistic closed}.
\item If $\Phi$ is atomistic closed and for any subset $A$ of $At(\phi)$,
\begin{eqnarray*}
\bigwedge A = \bigwedge At(\bigwedge A),
\end{eqnarray*}
the algebra is called \textbf{completely atomistic}.
\end{enumerate}
\end{definition}

For illustration let's briefly consider set algebras of subsets of some universe $U$, see Section \ref{sec:SetAlg}. A set algebra needs not necessarily be atomic, but if it is, it is necessarily atomistic. To verify this, recall first that information order corresponds to set inclusion. So atoms are the smallest, non-empty subsets of $U$ in $\Phi$. Then we claim that $At(\Phi)$ forms a partition of $U$, the atoms are the blocks of some partition $P_x$ for $x \in Q$. As stated above, atoms are disjoint, since combination is set-intersection. Further we have $\bigcup At(\Phi) = U$. If this were not the case, there would be a non-empty set $S$ so that $U = S \cup (\bigcup At(\Phi))$ and $S$ would be $x$-saturated, hence belonging to $\Phi$. But since we assume the set algebra to be atomic, there must be an atom in $At(S)$ and not in $At(\Phi)$. But this is a contradiction. Then the partition $P_x$ must be the finest partition in $P_Q$ since any block of any other partition $P_y$, $y \in Q$ must contain an atom, that is a block of $P_x$. So $P_y \leq P_x$. But then any set $S \in \Phi$ must be $x$-saturated, that is a union of blocks of $P_x$, where union is the meet in information order. Therefore the set algebra $\Phi$ is atomistic, if and only if $P_Q$ contains a finest partition $P_x$. It is completely atomistic if all sets which are $x$-saturated are in $\Phi$. In particular, if $\Phi$ is the power set of $U$, the algebra is completely atomistic.

The upset algebra $U(\Phi_0)$ introduced in the previous section is atomic if $\Phi$ is so. Then the atoms of $U(\Phi_0)$ are the principal ideals $\uparrow\! \alpha = \{\alpha\}$ associated with atoms $\alpha \in At(\Phi)$. In fact we have $\{\alpha\} \cap U = \{\alpha\}$, if $\alpha \in U$ or $= \emptyset$ otherwise, for any upset $U$. And if $U \not= \emptyset$. then if $\phi \in U$ for every atom $\alpha \in At(\phi\}$ we have $\alpha \in U$, hence $\{\alpha\} \subseteq U$. However, $U(\Phi_0)$ is not atomistic, even if $\Phi$ is atomistic.

We introduce a further small example, namely \textbf{String algebras}. Consider a finite alphabet $\Sigma$, the set $\Sigma^*$ of finite strings over $\Sigma$, including the empty string $\epsilon$, and the set $\Sigma^\omega$ of infinite strings over $\Sigma$. Let $\Sigma^{**} = \Sigma^* \cup \Sigma^\omega \cup \{0\}$, where $0$ is a symbol not contained in $\Sigma$. For two strings $\phi,\psi \in \Sigma^{**}$ define $\phi \leq \psi$ if $\phi$ is a prefix of $\psi$ and for all $\phi$,  $\phi \leq 0$. The empty string $\epsilon$ is a prefix of any string $\phi$, $\epsilon \leq \phi$.  We define combination among strings by
\begin{eqnarray*}
\phi \cdot \psi =  \left\{\begin{array}{ll} \psi & \textrm{if}\ \phi \leq \psi, \\ \phi & \textrm{if}\ \psi \leq \phi, \\ 0 & \textrm{otherwise}. \end{array} \right.
\end{eqnarray*}
Then $(\Sigma^{**},\cdot,0,\epsilon)$ is clearly a commutative semigoup, with $\epsilon$ the unit and $0$ the null element. For extraction, we define operators $\epsilon_n$ for $n \in \mathbb{N} \cup \{\infty\}$, where $\epsilon_n(\phi)$ is the prefix of length $n$ of $\phi$, if the length of $\phi$ is at least $n$, or $\epsilon_n(\phi) = \phi$ otherwise. It is easy to verify that $\epsilon_n$ is an existential quantifier for all $n$. The order induced by these operators on $\mathbb{N} \cup \{\infty\}$ is just the natural order of integers, and under this order $\mathbb{N} \cup \{\infty\}$ is a lattice.  Therefore the Join axiom is valid. Finally the Support axiom is obviously satisfied. Therefore, $(\Sigma^{**},\cdot,0,\epsilon;E)$ with $E = \{\epsilon_n:n \in \mathbb{N} \cup \{\infty\}\}$ is a domain-free information algebra. Since the operators $\epsilon_n$ commute, it is even a commutative information algebra. The infinite strings in $\Sigma^\omega$ are the atoms of this algebra. For a finite string $\phi$, $At(\phi)$ is the set of all infinite strings which have $\phi$ as prefix. Also $\phi$ is the infimum of $At(\phi)$, so the algebra is atomistic. It is even atomistic closed  since for any set $A$ of infinite strings either $\epsilon$ is the infimum, if the strings have no common prefix or else the infimum is the longest common prefix of the elements of $A$. But it is clearly not completely atomistic.

If subalgebras of an information algebra $\Phi$ have atoms, they may not be the same as those of $\Phi$. An important case are the subalgebras $\epsilon_x(\Phi)$ of all elements with support $x$, see Section \ref{sec:AlgNotions}. Here we have the following result.

\begin{proposition} 
If $\Phi$ is an atomic information algebra, then the subalgebra $\epsilon_x(\Phi)$ is atomic and its atoms are $At(\epsilon_x(\Phi)) = \epsilon_x(At(\phi))$.
\end{proposition}

\begin{proof}
We show first that the elements $\epsilon_x(\alpha)$, where $\alpha \in At(\Phi)$ are atoms in the algebra $\epsilon_x(\Phi)$. We have $\alpha \not= 0$, hence $\epsilon_x(\alpha) \not= 0$. Assume $\phi = \epsilon_x(\phi) \geq \epsilon_x(\alpha)$ for some $\phi \in \epsilon_x(\Phi)$. Then $\epsilon_x(\alpha \cdot \epsilon_x(\phi)) = \epsilon_x(\alpha) \cdot \epsilon_x(\phi) = \epsilon_x(\phi)$. Since $\alpha$ is an atom we have either $\alpha \cdot \epsilon_x(\phi) = \alpha$ or $\alpha \cdot \epsilon_x(\phi) = 0$. In the first case we conclude $\epsilon_x(\phi) = \epsilon_x(\alpha)$ and in the second case $\epsilon_x(\phi) = 0$. So $\epsilon_x(\alpha)$ is an atom in $\epsilon_x(\Phi)$.

Next we show that for any element $\phi = \epsilon_x(\phi) \not= 0$ in $\epsilon_x(\Phi)$ there is an atom $\epsilon_x(\alpha)$ of $\epsilon_x(\Phi)$ such that $\phi \leq \epsilon_x(\alpha)$. Since $\Phi$ is atomic, there is an atom $\alpha$ such that $\epsilon_x(\phi) \leq \alpha$ and therefore $\epsilon_x(\phi) \leq \epsilon_x(\alpha)$. This shows that $\epsilon_x(\Phi)$ is atomic. 

Let $\beta$ be a local atom in $\epsilon_x(\Phi)$. Since $\Phi$ is atomic, there is an atom $\alpha \in At(\beta)$, hence $\beta \leq \alpha$ and so $\beta = \epsilon_x(\beta) \leq \epsilon_x(\alpha)$. But since $\beta$ is a local atom relative to $x$, we conclude that $\beta = \epsilon_x(\alpha)$. Thus shows that $At(\epsilon_x(\Phi)) = \epsilon_x(At(\Phi))$.
\end{proof}

The elements $\epsilon_x(\alpha)$ for $\alpha \in At(\Phi)$ are called relative atoms or local atoms relative to $x$. Let $At_x(\Phi) = \epsilon_x(At(\Phi))$. Local atoms inherit the results of Lemma \ref{le:ElPropAtoms} with respect to $\epsilon_x(\Phi)$. In addition, we have the following result.

\begin{lemma} \label{le:PropRelAtoms}
\begin{enumerate}
\item If $\beta$ is a local atom relative to $x$ and $y \leq x$, then $\epsilon_y(\beta)$ is a local atoms relative to $y$,
\item if $\alpha$ and $\beta$ are local atoms relative to $x$ and $y$, then either $\alpha \cdot \beta = 0$ or else $\epsilon_x(\alpha \cdot \beta) = \alpha$ and $\epsilon_y(\alpha \cdot \beta) = \beta$.
\end{enumerate}
\end{lemma}

\begin{proof}
Item 1 holds since $\epsilon_y(\Phi)$ is a subalgebra of $\epsilon_x(\Phi)$ and $\beta$ is an atom in $\epsilon_x(\Phi)$. For the second item assume that $\alpha \cdot \beta \not= 0$. Then $\epsilon_x(\alpha \cdot \beta) = \alpha \cdot \epsilon_x(\beta) \not= 0$. Since $\alpha \leq \alpha \cdot \epsilon_x(\beta)$ we conclude that $\alpha = \alpha \cdot \epsilon_x(\beta) = \epsilon_x(\alpha \cdot \beta)$. The identity $\beta = \epsilon_y(\alpha \cdot \beta)$ follows in the same way.
\end{proof}

If we consider again a set algebra $\Phi$, then we see that for any $x \in Q$ the blocks of $P_x$ associated with the saturation oiperators $\sigma_x$ are local atoms relative to $x$. This shows that relative atoms may exist even without the existence of atoms. Further if $A$ and $B$ are local atoms relative to $x$ and $y$, that is blocks of partitions $P_x$ and $P_y$ respectively, then $A \cdot B = A \cap B$. This is, in general, no more a block of $P_{x \vee y}$. This shows that the combination of local atoms does not give, in general, a relative atom. An exception occurs in set algebras if $P_{x \vee y} = P_x \vee P_y$, that is if $P_{x \vee y}$ is the usual join of partitions in $Part(U)$.  As we have remarked in Section \ref{sec:SetAlg} this is, in general, not the case. 

If $\Phi$ is atomistic, the so is the subalgebra $\epsilon_x(\Phi)$. This follows since for $\phi \in \epsilon_x(\Phi)$, we have
\begin{eqnarray*}
At(\phi) = \bigcup_{\beta \in At_x(\phi)} At(\beta).
\end{eqnarray*}
In fact, if $\alpha \in At(\phi)$, then $\phi = \epsilon_x(\phi) \leq \alpha$, thus $\phi \leq \epsilon_x(\alpha) = \beta \in At_x(\Phi)$ and $\beta \leq \alpha$ so that $\alpha \in At(\beta)$. And if $\alpha \in At(\beta)$ for some $\beta \in At_x(\phi)$, then $\phi = \epsilon_x(\phi) \leq \beta \leq \alpha$, so that $\alpha \in At(\phi)$. As the example of set algebras shows, the subalgebras $\epsilon_x(\Phi)$ may be atomic or atomistic, without $\Phi$ being so. We call $\Phi$ locally atomic or locally atomistic, if all the subalgebras $\epsilon_x(\Phi)$ are atomic or atomistic respectively for all $x \in Q$.


\section{Set algebras of atoms} \label{subsec:SetAlgAtoms}

In this section we are going to consider set algebras of subsets of atoms of an atomic information algebra $(\Phi,\cdot,0,1;E)$, where as always $E = \{\epsilon_x:x \in Q\}$. So, the universe of the set algebra is $At(\Phi)$ and the elements of the algebra are subsets of it. We consider a family of partitions $At_x$ of $At(\Phi)$ defined by the equivalence relation $\alpha \equiv \beta$ if $\epsilon_x(\alpha) = \epsilon_y(\beta)$. So, the blocks of the partition $At_x$ are the sets $At(\epsilon_x(\alpha))$ of the atoms contained in the relative atoms $\epsilon_x(\alpha)$. We denote the corresponding saturation operators by $\sigma_x$ and define $\Sigma_Q = \{\sigma_x:x \in Q\}$. Note that $x \leq y$ means $\epsilon_x = \epsilon_y\epsilon_x$ and this implies that $\alpha \equiv_y \beta \Rightarrow \alpha \equiv_x \beta$ and therefore $P_x \leq P_y$. Let further $\mathcal{S}_Q$ be the family of subsets of $At(\Phi)$ saturated with respect to a $x \in Q$. According to Section \ref{sec:SetAlg}, $(\mathcal{S}_Q,\cap,\emptyset,At(\Phi);\Sigma_Q)$ is a set algebra. 

Consider the map $f : \Phi \rightarrow At(\Phi)$ defined by $f(\phi) = At(\phi)$. It turns out that this map is an information algebra homomorphism, if $\Phi$ is atomic.

\begin{theorem} \label{th:AtomiHomom}
Let $(\Phi,\cdot,0,1;E)$ with $E = \{\epsilon_x:x \in Q\}$ be an atomic information algebra. Then for all $\phi,\psi \in \Phi$ and for all $x \in Q$,
\begin{enumerate}
\item $At(\phi \cdot \psi) = At(\phi) \cap At(\psi)$,
\item $At(0) = \emptyset$, $At(1) = At(\Phi)$,
\item $At(\epsilon_x(\phi)) = \sigma_x(At(\phi))$.
\end{enumerate}
\end{theorem}

\begin{proof}
Since the algebra is atomic we have $At(\phi) \not= \emptyset$ if $\phi \not= 0$. Assume $\phi \cdot \psi \not= 0$ and let $\alpha \in At(\phi \cdot \psi)$. Then $\phi,\psi \leq \phi \cdot \psi \leq \alpha$, hence $\alpha \in At(\phi) \cap At(\psi)$. Conversely, if $\alpha \in At(\phi) \cap At(\psi)$, then $\phi,\psi \leq \alpha$ and therefore $\phi \cdot \psi \leq \alpha$, thus $\alpha \in At(\phi \cdot \psi)$. This proves item 1. Item 2 is  obvious. 

For item 3 assume first that $\alpha \in \sigma_x(At(\phi))$. Then there is a $\beta \in At(\phi)$ such that $\epsilon_x(\alpha) = \epsilon_x(\beta)$. $\beta \in At(\phi)$ implies $\phi \leq \beta$, so $\epsilon_x(\phi) \leq \epsilon_x(\beta) = \epsilon_x(\alpha) \leq \alpha$ and thus $\alpha \in At(\epsilon_x(\phi))$. 

Conversely, consider an atom $\alpha \in At(\epsilon_x(\phi))$. We claim that $\epsilon_x(\alpha) \cdot \phi \not= 0$. Indeed, otherwise we would have $\epsilon_x(\alpha \cdot \epsilon_x(\phi)) = \epsilon_x(\alpha) \cdot \epsilon_x(\phi) =\epsilon_x(\epsilon_x(\alpha) \cdot \phi) = 0$ implying $\alpha \cdot \epsilon_x(\phi) = 0$ which contradicts $\alpha \in At(\epsilon_x(\phi))$. So there exists an atom $\beta \in At(\epsilon_x(\alpha) \cdot \phi)$ and thus $\phi \leq \epsilon_x(\alpha) \cdot \phi \leq \beta$ and thus $\beta \in At(\phi)$. Further, $\epsilon_x(\epsilon_x(\alpha) \cdot \phi) = \epsilon_x(\alpha) \cdot \epsilon_x(\phi) \leq \epsilon_x(\beta)$, hence $\epsilon_x(\alpha) \cdot \epsilon_x(\beta) \cdot \epsilon_x(\phi) = \epsilon_x(\beta)$. This implies $\epsilon_x(\alpha) \cdot \epsilon_x(\beta) \not= 0$. Since $\epsilon_x(\alpha) \cdot \epsilon_x(\beta) = \epsilon_x(\alpha \cdot \epsilon_x(\beta))$ we conclude that $\alpha \cdot \epsilon_x(\beta) \not= 0$, hence $\epsilon_x(\beta) \leq \alpha$, since $\alpha$ is an atom. We infer $\epsilon_x(\beta) \leq \epsilon_x(\alpha)$. Proceed in the same way from $\epsilon_x(\alpha) \cdot \epsilon_x(\beta) = \epsilon_x(\epsilon_x(\alpha) \cdot \beta)$ in order to obtain $\epsilon_x(\alpha) \leq \epsilon_x(\beta)$ so that finally $\epsilon_x(\alpha) = \epsilon_x(\beta)$. But this means that $\alpha \in \sigma_x(At(\phi))$ and so $At(\epsilon_x(\phi)) = \sigma_x(At(\phi))$ as claimed.
\end{proof}

If $\phi$ has support $x$, then $At(\phi) = At(\epsilon_x(\phi)) = \sigma_x(At(\phi))$, so that $At(\phi)$ is $x$-saturated. That is, the map $f : \phi \mapsto At(\Phi)$ maps into the set algebra $\mathcal{S}_Q$ or the atomic information algebra $\Phi$ is homomorphic to the set algebra $\mathcal{S}_Q$. If the information algebra $\Phi$ is atomistic, then $At(\phi) = At(\psi)$ implies $\phi = \bigwedge At(\phi) = \bigwedge At(\psi) = \psi$ so that the map $f$ is injective, hence an embedding.

\begin{corollary} \label{cor:AtomisticEmbed}
If $(\Phi,\cdot,0,1;E)$ is an atomistic information algebra, then the map $f : \phi \mapsto At(\phi)$ is an embedding of $\Phi$ in the set algebra $\mathcal{S}_Q$.
\end{corollary}

This means that an atomistic information algebra $\Phi$ is isomorphic to a set algebra, a sub-set algebra of $\mathcal{S}_Q$. This can also be interpreted as follows: Atoms are maximally informative pieces of information. Therefore $At(\Phi)$ can be considered as a set of possible worlds and the piece of information $\phi \in \Phi$ defines by $At(\phi)$ their set of possible worlds which remain possible, if $\phi$ is assumed. Obviously in this view, combination of two pieces of information $\phi$ and $\psi$ corresponds to the intersection of theirs sets of atoms $At(\phi) \cap At(\psi)$. Further, question $x$ has the same answer for the atoms (possible worlds) $\alpha$ and $\beta$, if $\epsilon_x(\alpha) = \epsilon_x(\beta)$. The blocks $At(\epsilon_x(\alpha))$ of the corresponding partition $P_x$ represent then the possible answers to question $x$. This concurs with the view of the local atoms $\epsilon_x(\alpha)$ as possible answers to question $x$. Then clearly, saturation with respect to partition $x$ means extraction of information relative to $x$ from $At(\phi)$. Corollary \ref{cor:AtomisticEmbed} tells us that this view of a set algebra of atoms is a really equivalent picture of the information algebra $\Phi$ in the atomistic case.

If the information algebra is completely atomistic, then the map $f : \phi \mapsto At(\phi)$ is surjective, hence bijective on $\mathcal{S}_Q$. So, $\phi$ is isomorphic (as an information algebra) to the set algebra $\mathcal{S}_Q$. Then a much stronger result holds.

\begin{theorem} \label{th:CoplAtomisticAlg}
Let $(\Phi,\cdot,0,1;E)$ be a completely atomistic information algebra, then $(\Phi,\leq)$ is a complete Boolean lattice and map $f : \phi \mapsto At(\phi)$ preserves arbitrary joins and meets (in the information order) as well as complements.
\end{theorem}

\begin{proof}
Let $X$ be any subset of $\Phi$ and define
\begin{eqnarray*}
A_X = \bigcap_{\psi \in X} At(\psi).
\end{eqnarray*}
Assume $A_X \not= \emptyset$. Since the algebra is completely atomistic, there exists a $\phi \in \Phi$ such that $A_X = At(\phi)$ and $\phi = \bigwedge A_X$. For any $\alpha \in A_X$ and $\psi \in X$, we have $\psi \leq \alpha$ therefore $\psi \leq \bigwedge A_X$ which shows that $\bigwedge A_X$ is an upper bound of $X$. Let $\chi$ be any other upper bound of $X$. Then $At(\chi) \subseteq At(\psi)$ for all $\psi \in X$, hence $\alpha \in At(\chi)$ implies $\alpha \in A_X$ and therefore $\chi = \bigwedge At(\chi) \geq \bigwedge A_X$. It follows that $\bigwedge A_X$ is the supremum of $X$, that is $\bigvee X = \bigwedge A_X$. Consequently, since $\bigvee X = \phi$ and $At(\phi) = A_X$,
\begin{eqnarray*}
At(\bigvee X) = \bigcap_{\psi \in X} At((\psi).
\end{eqnarray*}
If $A_X = \emptyset$, then $\bigvee X = 0$ and $At(0) = \emptyset$. So join (in the information order) is preservedunder the map $f$.

Consider $\phi \in \Phi$ and define $At^c(\phi) = At(\Phi)/At(\phi)$. Since $\Phi$ is completely atomistic, $\psi = \bigwedge At^c(\phi)$ exists and belongs to $\Phi$. Moreover, $At(\psi) = At^c(\phi)$. We know that $At(\phi \cdot \psi) = At(\phi) \cap At(\psi)$, $At(1) = At(\Phi)$ and $At(0) = \emptyset$. Then
\begin{eqnarray*}
\phi \vee \psi &=& \bigwedge At(\phi \cdot \psi) = \bigwedge (At(\phi) \cap At(\psi)) \\
&=& \bigwedge (At(\phi) \cap At^c(\phi)) = \bigwedge \emptyset = 0.
\end{eqnarray*}
Further, we show that $\phi \wedge \psi$ exists in $\Phi$ and $At(\phi \wedge \psi) = At(\phi) \cup At(\psi)$. For this purpose, put $A = At(\phi) \cup At(\psi)$. Then there exists a $\chi$ such that $A = At(\chi)$ and $\chi = \bigwedge A$. Since $At(\phi),At(\psi) \subseteq At(\phi) \cup At(\psi)$, we have $\chi \leq \phi$ and $\psi \leq \chi$. Let $\xi$ be another lower bound of $\phi$ and $\psi$. Then $At(\xi) \supseteq A = At(\chi)$, hence $\xi \leq \chi$ by atomisticity and therefore $\chi = \phi \wedge \psi$. Certainly $A \subseteq At(\phi \wedge \psi)$ so that $At(\phi \wedge \psi) = A$. If $\psi = \bigwedge At^c(\phi)$, then 
\begin{eqnarray*}
\phi \wedge \psi = \bigwedge (At(\phi) \cup At^c(\phi)) = At(\Phi) = 1.
\end{eqnarray*}
So $\psi$ is the complement of $\phi$, $\psi = \phi^c$.

The map $\phi \mapsto At(\phi)$ thus preserves arbitrary joins and complements and consequently also arbitrary meets, completing the proof.
\end{proof}

If in a domain-free information algebra, $(\Phi,\leq)$ is a Boolean lattice in information order, then the information algebra is called Boolean. Such Boolean information algebras will be discussed in the next section. But before we examine the example of String algebras.

We have seen that the information algebra of strings $(\Sigma^{**},\cdot,0,\epsilon;E)$ is atomistic (see Section \ref{subsec:AtomicAlg}). By the results above, the algebra is embedded into the set algebra of its atoms, by the map $s \mapsto At(s)$. The combination $s \cdot r$ is mapped into the intersection of the sets of infinite strings with both $s$ and $r$ as prefixes, which is empty, if not either $s$ is a prefix of $r$ or $r$ a prefix of $s$. In the first case $At(s) \cap At(r) = At(s)$, in the second case $At(s) \cap At(r) = At(r)$. For any $n$, the saturation operator $\sigma_n$ maps a set $S$ of infinite strings into the set of all infinite strings which have a common prefix of some length $n$ with some string of $S$. Compare this with the representation of the same algebra by up-sets $\uparrow\! s$ of all strings (including finite ones) which have $s$ as a prefix.


\section{Representing Boolean information algebras} \label{subsec:BooleInfAlg}

Let $(\Phi,\cdot,0,1;E)$ with $E = \{\epsilon_x:x \in Q\}$ be a domain-free information algebra. If $(\Phi,\leq)$, is a Boolean lattice in information order $\leq$, then the information algebra is called \textit{Boolean}. Recall that this means that it is a distributive lattice and for all element $\phi$ there is a complement $\phi^c$ \cite{daveypriestley97}. We show first that for any $x \in Q$, extraction distributes over meet.

\begin{proposition} \label{porp:ExtrDidsMeet}
If $\Phi$ is a Boolean information algebra, then for all pair of elements $\phi$ and $\psi$ in $\Phi$,
\begin{eqnarray*}
\epsilon_x(\phi \wedge \psi) = \epsilon_x(\phi) \wedge \epsilon_x(\psi).
\end{eqnarray*}
\end{proposition}

\begin{proof}
Put $\eta = \phi \wedge \psi$. Then $\eta \leq \phi,\psi$ implies $\epsilon_x(\eta) \leq \epsilon_x(\phi),\epsilon_x(\psi)$. Hence $\epsilon_x(\eta)$ is a lower bound of $\epsilon_x(\phi)$ and $\epsilon_x(\psi)$. Let $\chi$ be any other lower bound of $\epsilon_x(\phi)$ and $\epsilon_x(\psi)$. Recall that  in a Boolean algebra or lattice we have $\psi \leq \phi$ if and only if $\phi \cdot \psi^c = \phi \vee \psi^c = 0$. So, we have $\epsilon_x(\phi) \cdot \chi^c = 0$ and $\epsilon_x(\psi) \cdot \chi^c = 0$. It follows that 
\begin{eqnarray*}
0 = \epsilon_x(0) = \epsilon_x(\epsilon_x(\phi) \cdot \psi^c) = \epsilon_x(\phi) \cdot \epsilon_x(\phi^c) = \epsilon_x(\phi \cdot \epsilon_x(\chi^c).
\end{eqnarray*}
This implies $\phi \cdot \epsilon_x(\chi^c) = 0$ and in the same way we obtain $\psi \cdot \epsilon_x(\chi^c) = 0$. Using distributivity we get
\begin{eqnarray*}
0 = (\phi \cdot \epsilon_x(\chi^c)) \wedge (\psi \cdot \epsilon_x(\chi^c)) = (\phi \wedge \psi) \cdot \epsilon_x(\chi^c) = \eta \cdot \epsilon_x(\chi^c).
\end{eqnarray*}
It follows that 
\begin{eqnarray*}
0 = \epsilon_x(0) = \epsilon_x(\eta \cdot \epsilon_x(\chi^c)) = \epsilon_x(\eta) \cdot \epsilon_x(\chi^c) = \epsilon_x(\epsilon_x(\eta) \cdot \chi^c),
\end{eqnarray*}
hence $\epsilon_x(\eta) \cdot \chi^c = 0$ But this implies $\chi \leq \epsilon_x(\eta)$ and $\epsilon_x(\eta)$ is thus the greatest lower bound of $\epsilon_x(\phi)$ and $\epsilon_x(\psi)$, that is $\epsilon_x(\phi \wedge \psi) = \epsilon_x(\phi) \wedge \epsilon_x(\psi)$. 
\end{proof}

Now, we consider first the case of a finite Boolean information algebra. If $\Phi$ is a finite Boolean lattice there are surely atoms so that $At(\Phi)$ is not empty. Again recall that atoms in our case are maximal elements, not minimal ones as usually meant by atoms in a Boolean lattice. Now, $\Phi$ is atomistic, which is a well-know result in Boolean lattice theory, \cite{daveypriestley97}. But since we use the inverse order, we give the (simple) proof here.

\begin{proposition} \label{prop:FinBooleAtom}
If $\Phi$ is a Boolean information algebra, then for every $\phi \in \Phi$, $\phi \not= 0$,
\begin{eqnarray*}
\phi = \bigwedge At(\phi).
\end{eqnarray*}
\end{proposition}

\begin{proof}
Obviously, $\phi \leq \bigwedge At(\phi)$. Let $\chi$ be any other lower bound of $At(\Phi)$. We claim that $\chi \leq \phi$. Otherwise we would have $\phi \cdot \chi^c \not= 0$. Then there would be an atom $\alpha \in At(\phi \cdot \chi^c)$ such that $\alpha \geq \phi \cdot \chi^c \geq \phi,\chi^c$, hence $\alpha \in At(\phi)$ and therefore $\alpha \geq \chi$. But this would imply $\alpha \geq \chi \cdot \chi^c = 0$ which is a contradiction. So $\chi \leq \phi$, hence $\phi$ is the infimum of $At(\phi)$.
\end{proof}

 Now, since $\Phi$ is a finite Boolean lattice, $\bigwedge A$ exists for all subsets $A$ of $At(\Phi)$ and if $\phi = \bigwedge A$, then $\phi \leq \alpha$ for all $\alpha \in A$ and $At(\phi) \supseteq A$ so that $\phi = \bigwedge At(\phi) \leq \bigwedge A = \phi$. Therefore, $\Phi$ is completely atomistic. So we conclude, using Theorem \ref{th:CoplAtomisticAlg} that the finite Boolean information algebra $\Phi$ is isomorphic, both as an information algebra as well as a Boolean lattice to the powerset algebra of $At(\Phi)$, namely $(2^{At(\phi)},\cap,\emptyset,At(\Phi),\Sigma_Q)$, where $\Sigma_Q$ is the set of the saturation operators $\sigma_x$ for $x \in Q$ related to the partitions $At_x$ defined by the equivalence relation  $\alpha \equiv_x \beta$ iff $\epsilon_x(\alpha) = \epsilon_x(\beta)$. The isomorphism is given by the map $f : \phi \mapsto At(\phi)$.
 
We turn to the general case of a Booelan information algebra $(\Phi,\cdot,0,1;E)$. In general it is no more atomic or even less atomistic, but its ideal $I_\Phi$ completion is. This case is essentially based on Stone's representation theory for Boolean lattices, see \cite{daveypriestley97}. The key concept in this theory is the one of a maximal ideal.

\begin{definition} \textbf{Maximal ideal:}
A proper ideal of $\Phi$ is called maximal, if $J \in I_\Phi$ and $I \leq J$ implies $I = J$ or $J = \Phi$.
\end{definition}

Obviously, maximal ideals are atoms in the information algebra $I_\Phi$, associated with $\Phi$, see Section \ref{subsec:IdealExt}. It is well-known that in a Boolean lattice maximal ideals are also prime ideals, that is, if $\phi \wedge \psi \in I$ if $I$ is a maximal ideal, then either $\phi \in I$ or $\psi \in I$, see \cite{daveypriestley97} for this and also the following summary of well-known results. If $I$ is a maximal ideal, then for all $\phi \in \Phi$, either $\phi$ or $\phi^c \in I$, and, if $\phi \not= \psi$, then there is a maximal ideal which contains exactly one of the two elements. Finally, for any proper ideal $J$ of $\Phi$, there is a maximal ideal $I$ so that $J \subseteq I$, or $J \leq I$ in information order. This means that the information algebra $I_\Phi$ is atomic.

Let $I_P(\Phi)$ be the set of all maximal ideals of $\Phi$. Since the information algebra $I_\Phi$ is atomic, the map $f : I_\Phi \rightarrow 2^{I_P(\Phi)}$ defined by $J \mapsto At(J) = \{I \in I_P(\Phi):J \leq I\}$ satisfies, according to Theorem \ref{th:AtomiHomom}
\begin{eqnarray*}
J_1 \cdot J_2 &\mapsto& At(J_1) \cap At(J_2), \\
\epsilon_x(J) &\mapsto& \sigma_x(At(J)),
\end{eqnarray*}
where here $\sigma_x$ is the saturation operator associated with the partition $P_x$ induced by the equivalence relation $J_1\equiv_x J_2$ iff $\epsilon_x(J_1) = \epsilon_x(J_2)$.

In particular, the restriction of this map to the subalgebra of principal ideals $\downarrow\! \phi$ is still an information algebra homomorphism. Note that
\begin{eqnarray*}
At(\downarrow\!\phi) = \{I \in I_P(\Phi):\ \downarrow\!\phi \subseteq I\} = \{I \in I_P(\Phi):\phi \in I\} =: X_\phi.
\end{eqnarray*}
The map $\downarrow\!\phi \mapsto X_\phi$ may be extended with the map $\phi \mapsto \downarrow\!\phi$ to a map $\Phi \rightarrow 2^{I_P(\Phi)}$ so that $\phi \mapsto X_\phi$. This map still satisfies the homomorphism conditions
\begin{eqnarray*}
\phi \cdot \psi = \phi \vee \psi &\mapsto& X_\phi \cap X_\psi, \\
\epsilon_x(\phi) &\mapsto& \sigma_x(X_\phi).
\end{eqnarray*}
In addition, we have obviously
\begin{eqnarray*}
0 \mapsto \emptyset, \quad 1 \mapsto I_P(\Phi),
\end{eqnarray*}
thus completing the homomorphism conditions. Furthermore, the map is one-to-one, because $\phi \not= \psi$ implies that there is a maximal ideal $I$ which contains one, but not the other of the two elements, so that $X_\phi \not= X_\psi$. So, this map is an \textit{embedding} of $\Phi$ in $I_\Phi$. We remark that maximal ideals have an information-theoretic interpretation as complete, consistent theories. A maximal ideal is consistent, since it is an ideal: It contains with any piece of information all pieces implied by it and with any two pieces also its combination. It is complete in the sense that it contains any piece of information $\phi$ or its negation (complement) and if $\phi \wedge \psi$ belongs to the ideal then either $\phi$ or $\psi$ belong to it too. So a Boolean information algebra $\Phi$ an be represented by the set algebra of all the consistent, complete theories, each element $\phi$ of $\Phi$ is uniquely represented by the consistent, complete theories it is contained in.

According to Stone's representation theory there is much more, \cite{daveypriestley97}. First, the map $\phi \mapsto X_\phi$ is a Boolean algebra homomorphism:
\begin{eqnarray*}
\phi \vee \psi &\mapsto& X_\phi \cap X_\psi, \\
\phi \wedge \psi &\mapsto& X_\phi \cup X_\psi, \\
\phi^c &\mapsto& X^c_\phi.
\end{eqnarray*}
Note that the fact that join (meet) maps to intersection (union) is due to the fact, that our information order in $I_\Phi$ is the converse of the usual order among sets, inclusion corresponds to more information \footnote{Usually, $X_\phi$ is defined as the set $\{I \in I_\Phi:\phi \not\in I\}$. In the spirit of the idea of information, it is more natural to define $X_\phi$ as above, namely a complete, consistent theories compatible with $\phi$.}. Further, the image $X_\phi$ of $\phi$ under this map can be characterized topologically as \textit{clopen } sets (simultaneously closed and open sets) in a topological space, the \textit{Stone space} or \textit{Boolean space}. In $I_P(\Phi)$, this topology is defined by
\begin{eqnarray*}
\mathcal{B} = \{X_\phi:\phi \in \Phi\}
\end{eqnarray*}
as an open base. Then the family of open sets is given by
\begin{eqnarray*}
\mathcal{T} = \{U \subseteq I_P(\Phi):U \textrm{ is a union of members of}\ \mathcal{B}\}.
\end{eqnarray*}
The topological space $(I_P(\Phi),\mathcal{T})$ is called the \textit{dual} or \textit{prime ideal space} of $\Phi$. The sets $X_\phi$ are open and then, since $X_{\phi^c} = X^c_\phi$ is also open, $X_\phi$ is also closed, hence clopen. In fact, the sets in $\mathcal{B}$ are precisely all the clopen subsets of $I_P(\Phi)$. They form a Boolean lattice. The topological space $I_P(\Phi)$ is \textit{compact}. Further, for any pair $I,J \in I_P(\Phi)$ there exists a clopen subset $X_\phi$ of $I_P(\Phi)$ such that $I \in X_\phi$ and $J \not\in X_\phi$. This means that the topological space $I_P(\Phi)$ is \textit{totally disconnected}. A compact, totally disconnected topological space is called a \textit{Boolean space}. We refer to \cite{daveypriestley97} for details. The Stone representation theorem asserts that the map $\phi \mapsto X_\phi$ is a Boolean algebra isomorphism between the Boolean algebra $\Phi$ and the field of clopen subsets $\mathcal{B}$ of the Boolean space $(I_P(\Phi),\mathcal{T})$.

This leads then to an extension of this representation theorem to a representation theorem for Boolean information algebras. This summarizes the discussion above.

\begin{theorem} \label{thBooleReprTh}
If $(\Phi,\cdot,0,1;E)$ is a Boolean information algebra with $E = \{\epsilon_x:x \in Q\}$, then it is is isomorphic, both as an information algebra as well as a Boolean lattice, to the set algebra $(\mathcal{B},\cap,\emptyset,I_P(\Phi);\Sigma)$, where $\Sigma = \{\sigma_x:x \in Q\}$ is the set of saturation operators associatied with the partitions $P_x$ of $I_P(\Phi)$ associated with the equivalence relation $I \equiv_x J$ iff $\epsilon_x(I) = \epsilon_:x(J)$, restricted to set of $\mathcal{B}$. This isomorphism is established by the map $\phi \mapsto X_\phi$.
\end{theorem}

This result can be extended to a full fletched duality theory between Boolean information algebras and topological Boolean spaces with a family of partitions. This will not be pursued here, we refer to \cite{Jontarsky51}. It can also be extended to information algebras, where $(\Phi,\leq)$ is a distributive lattice in the information order. This is based on \textit{Priestely spaces}, \cite{daveypriestley97}.  For the case of commutative information algebras we refer to \cite{kohlasschmid16}. Further, we remark that due to duality there is also a similar representation theory of labeled information algebras. This too will not be worked out here.


%% file: chapter7.tex

\chapter{Local computation} \label{sec:LocComp}


\section{Conditional independence structures} \label{subsec:CondIndepStruct}

In this section, we introduce a number of conditional independence structures related to a domain-free information algebra $(\Phi,\cdot,0,1;E)$ where as usual $E = \{\epsilon_x:x \in Q\}$ is a family of extraction operators (i.e. existential quantifiers, see Section  \ref{sec:InfAlg}) and $(Q,\leq)$ is a join semilattice under the order $x \leq y$ if and only if $\epsilon_x = \epsilon_x\epsilon_y = \epsilon_y\epsilon_x$ (see Section \ref{sec:StructOfQuest}). These structures will then serve to propose efficient computational methods, extending well-known local computation schema in Bayesian networks for instance or more generally in valuation algebras based on multivariate models, as discussed in \cite{lauritzenspiegelhalter88,shenoyshafer90,kohlas03} to cite only a few references. These approaches depend in the multivariate case on a conditional independence structure among variables called \textit{join trees, junction trees} or also \textit{hypertrees}. These concepts can also be modelled by graphical structures describing dependence relations among variables and there is a large body of literature on this subject. However, these concepts can not be transferred simply as such to our more general model of information algebras with $Q$ being only a join-semilattice, and, in general nothing more. We need concepts adapted to the present structure. 

Assume that $(Q,\leq,\bot)$ is a q-separoid, and recall that any information algebra induces such a structure (Section \ref{sec:StructOfQuest}). We start by extending the conditional independence relation $x \bot y \vert z$ in $Q$ to a more general relation describing conditional independence of a set of questions or domains $\{x_1,\ldots,x_n\}$ from $Q$ given a $z \in Q$ for $n \geq 2$. If $J$ is a finite subset of elements of $Q$ let
\begin{eqnarray*}
x_J = \vee_{j \in J} x_j.
\end{eqnarray*}
Then we can define the concept of conditional independence for any finite subset of elements of $Q$.

\begin{definition} \textbf{Conditional independence of a set of questions:}
Let $(Q,\leq,\bot)$ be a q-separoid. If $\{x_1,\ldots,x_n\}$ is a finite set of elements from $Q$, $n \geq 2$ and $z \in Q$, then the elements in the set $\{x_1,\ldots,x_n\}$ are called (mutually) conditionally indpendent given $z$, if for any pair of disjoint subsets $J$ and $K$ of $\{x_1,\ldots,x_n\}$
\begin{eqnarray*}
x_J \bot x_K \vert z.
\end{eqnarray*}
Then we write $\bot \{x_1,\ldots,x_n\} \vert z$. 
\end{definition}
Recall that  given a domain-free information algebra $(\Phi,\cdot,0,1;E)$ with $E = \{\epsilon_x:x \in Q\}$, we have $x_J \bot x_K \vert z$ if and only if (see Section \ref{sec:StructOfQuest})
\begin{eqnarray*}
\epsilon_{\vee_{j \in J} x_j \vee z}\epsilon_{\vee_{k \in K} x_k \vee z} &=& \epsilon_z, \\
\epsilon_{\vee_{k \in K} x_k \vee z}\epsilon_{\vee_{j \in J} x_j \vee z} &=& \epsilon_z.
\end{eqnarray*}

By convention, for all $x \in Q$, we define $\bot \{x\} \vert z$ and $\bot \emptyset \vert z$. Note first that due to condition C3 of a q-separoid, we may assume $J \cup K = \{1,\ldots,n\}$ in the definition above. Here are a few further elementary results on this relation

\begin{proposition} \label{prop:ElResMultCondIndep}
Assume $\bot \{x_1,\ldots,x_n\} \vert z$. Then,
\begin{enumerate}
\item if $\sigma$ is any permutation of $1,\ldots,n$, then $\bot \{x_{\sigma(1)},\ldots,x_{\sigma(n)}\} \vert z$.
\item If $J \subseteq \{1,\ldots,n\}$, then $\bot \{x_j:j \in J\} \vert z$,
\item if $y \leq x_1$, then $\bot\{y,x_2,\ldots,x_n\} \vert z$,
\item $\bot \{x_1 \vee x_2,x_3,\ldots,x_n\} \vert z$,
\item  $\bot \{x_1 \vee z,x_2,\ldots,x_n\} \vert z$.
\end{enumerate}
\end{proposition}

These statements are all obvious from the definition of the relation and the q-separoid properties of $x \bot y \vert z$. In case $(Q,\leq)$ is a lattice, $\bot_L \{x_1,\ldots,x_n\} \vert z$ implies $x_1 \bot_L x_2 \vert z$, $x_2 \bot_L x_3 \vert z$, etc. which means $(x_1 \vee z) \wedge (x_2 \vee z) = z$, $(x_2 \vee z) \wedge (x_3 \vee z) = z$, etc. and this implies
\begin{eqnarray*}
(x_1 \vee z) \wedge (x_2 \vee z) \wedge \cdots \wedge (x_n \vee z) = z.
\end{eqnarray*}
If, in addition, the lattice is also distributive, then
\begin{eqnarray*}
(\vee_{j \in J} x_j \vee z) \wedge (\vee_{k \in K} x_k \vee z) = \vee_{j \in J,k \in K} (x_j \wedge x_k) \vee z = z,
\end{eqnarray*}
hence $x_j \wedge x_k \leq z$ for all $j \not= k$. Therefore, in this case $\bot_L \{x_1,\ldots.x_n\} \vert z$ if and only if $x_j \bot_L x_k \vert z$ or $x_j \wedge x_k \leq z$ for all pairs of distinct $j$ and $k$, $j,k = 1,\ldots,n$.

Theorem \ref{th:CombExtrProp} in Section \ref{sec:StructOfQuest} generalizes to $\bot \{x_1,\ldots,x_n\} \vert z$ (and the same is true for Theorem \ref{the:CombExtrProplab} in Section \ref{sec:LabInfAlg} in the labeled case) and this is a fundamental result for local computation.

\begin{theorem}
Assume $\bot \{x_1,\ldots,x_n\} \vert z$ and let $\phi = \phi_1 \cdots \phi_n$ where $x_i$ is a support for $\phi_i$ for $i = 1,\ldots.n$. Then
\begin{eqnarray*}
\epsilon_z(\phi) = \epsilon_z(\phi_1) \cdots \epsilon_z(\phi_n).
\end{eqnarray*}
\end{theorem}

\begin{proof}
The proof goes by induction. The claim holds for $n = 2$ (Theorem \ref{th:CombExtrProp}). Assume it holds for $n-1$. Then  $\phi = \phi_1 \cdot \psi_{n-1}$, where $\psi_{n-1} = \phi_2 \cdot \ldots \cdot \phi_n$. Then we have $\epsilon_z(\phi) = \epsilon_z(\phi_1 \cdot \psi_{n-1}) = \epsilon_z(\phi) \cdot \epsilon_z(\psi_{n-1})$ and by the assumption of induction $\epsilon_z(\psi_{n-1}) = \epsilon_z(\phi_2) \cdot \ldots \cdot \epsilon_z(\phi_n)$. Therefore we obtain indeed $\epsilon_z(\phi) = \epsilon_z(\phi_1) \cdots \epsilon_z(\phi_n)$.
\end{proof}

We introduce a further important conditional independence structure. Let $T = (V,E)$ be a tree with a finite set of vertices $V$ and edges $E \subseteq V ^2$, where $V^2$ is the family of two-elements subsets of $V$. Let further $\lambda : V \rightarrow Q$ be a labeling of the vertices of $T$ with elements of $Q$. Then the pair $(T,\lambda)$ is called a labeled tree. By $ne(v)$ we denote the set of neighbors of a vertex $v \in V$, that is $ne(v) = \{w \in V:\{v,w\} \in E\}$. For any subset of nodes $U$ of $V$ we define
\begin{eqnarray*}
\lambda(U) = \vee_{v \in U} \lambda(v).
\end{eqnarray*}
When a node $v$ is eliminated from the tree $T$ together with all edges $\{v,w\}$ incident to $v$, then a family of subtrees $\{T_{v,u} = (V_{v,w},E_{v,u}):u \in ne(v)\}$ are created, where $V_{v,u}$ is the set of vertices of the subtree containing the node $u \in ne(v)$ and $E_{v,u}$ the set of edges of $T$ linking vertices of $V_{v,u}$, that is $E_{v,u} = \{\{w,w'\} \in E:w,w' \in V_{v,u}\}$. This allows now to define the concept of a Markov tree.

\begin{definition} \textbf{Markov tree:}
A labeled tree $(T,\lambda)$ with $T = (V,E)$ and $\lambda : V \rightarrow Q$, is called a Markov tree, if for all vertices $v \in V$
\begin{eqnarray} \label{eq:MarkovProp}
\bot \{\lambda(V_{v,u}):u \in ne(V)\} \vert \lambda(v).
\end{eqnarray}
\end{definition}

Markov trees and derived concepts have been early identified as important independence structures for efficient computation with belief functions using Demster's rule \cite{SSM87,kohlasmonney95,shenoyshafer90}. In the first two of these references qualitative Markov trees for partitions are discussed, whereas in the last one a derived structure, join trees, are used in a multivariate setting. In the multivariate setting join or junction trees and hypertrees are widely discussed for local computation purposes and related to various graphical models for describing conditional independence. Below we shall discuss how these conditional independence structures are related to our concept of Markov trees. Also the concept is generalized and adapted from the probabilistic concept of Markov random fields. We prove two fundamental propositions about Markov trees whose proofs are adapted from \cite{kohlasmonney95}.

\begin{theorem} \label{th:NeighCondIndep}
If $(T,\lambda)$ is a Markov tree, then for any node $v \in V$ and all nodes $u \in ne(v)$,
\begin{eqnarray} \label{eq:NeighCondIndep}
\lambda(v) \bot \lambda(V_{v,u}) \vert \lambda(u).
\end{eqnarray}
\end{theorem}

\begin{proof}
For a node $w \in ne(v)$, the Markov condition (\ref{eq:MarkovProp}) reads
\begin{eqnarray*}
\bot \{\lambda(V_{w,u}):u \in ne(w)\} \vert \lambda(w).
\end{eqnarray*}
Then
\begin{eqnarray} \label{eq:DerMarkov}
\lambda(V_{w,v}) \bot \bigvee_{u \in ne(w)/\{v\}} \lambda(V_{w,u}) \vert \lambda(w).
\end{eqnarray}
Note that 
\begin{eqnarray}
\lambda(V_{v,w}) = \bigvee_{u \in ne(w)/\{v\}} \lambda(V_{w,u}) \vee \lambda(w).
\end{eqnarray}
Hence form property C4 of a q-separoid we obtain
\begin{eqnarray*}
\lambda(V_{w,v}) \bot \lambda(V_{v,w}) \vert \lambda(w).
\end{eqnarray*}
Finally, since $\lambda(v) \leq \lambda(V_{w,v})$, we conclude using (\ref{eq:NeighCondIndep}) using C3.
\end{proof}

\begin{theorem}
If $(T,\lambda)$ is a Markov tree, then any subtree $(T_{v,u},\lambda)$ is also a Markov tree.
\end{theorem}

\begin{proof}
Assume $T' = (V',E')$ to be a subtree of $T = (V,E)$ and $\lambda'$ the restriction  of $\lambda$ to $V '$. Consider a node $v \in V'$ and let $ne'(v)$ be the set of neighbours of $v$ in $T'$. Also consider subtrees $T'_{v,w} = (V'_{v.w},E'_{v,w})$ obtained after removing node $v$ and the edges incident to it in $T'$. Then $ne'(v) \subseteq ne(v)$ and $V'_{v,w} \subseteq V_{v,w}$ so that $\lambda(V'_{v,w}) \leq \lambda(V_{v,w})$ for all $w \in ne'(v)$ Therefore, from Proposition \ref{prop:ElResMultCondIndep} we conclude that
\begin{eqnarray*}
\bot \{\lambda'(V'_{v,w}):w \in ne'(v)\} \vert \lambda'(v)
\end{eqnarray*}
for all $v \in V'$. This shows that $(T',\lambda')$ is a Markov tree.
\end{proof}

From Markov trees two important derived structures  can be obtained. In a tree $T$ we may select any node $v$ and then number the $n$ nodes $v_i$ for $i = 1,\ldots,n = \vert V \vert$ such that $i < j$ if $v_i$ is on the (unique) path from $v_j$ to $v = v_n$.  Assume such a numbering $v_i$ of nodes in $V$ and define $x_i = \lambda(v_i)$. The set of nodes $\{v_{i+1},\ldots,v_n\}$ together with the all edges in $E$ linking these nodes determine a subtree of $T$. Indeed, there is a path in $T$ from $v_j$, $j > i$ to $v_n$ and it can not pass through any node $h \leq i$. So the subgraph determined by the nodes $\{v_{i+1},\ldots,v_n\}$ is connected, hence a tree. There is exactly one node $v_j \in ne(v_i)$ so that $j > i$. Denote this index $j$ by $b(i)$. Then, by Theorem \ref{th:NeighCondIndep} we have for $i = 1,\ldots,n-2$,
\begin{eqnarray} \label{eq:Hypertree}
x_i \bot \vee_{j = i+1}^n x_j \vert x_{b(i)}.
\end{eqnarray}
This result is defining a hypertree according to the following definition.

\begin{definition} \label{def:Hypertree}  \textbf{Hypertree:}
Let $(Q,\leq,\bot)$ be a q-separoid. An $n$-element subset $S$ of $Q$ is called a hypertree if there is a numbering of its elements $S = \{x_1,\ldots,x_n\}$ such that for all $i = 1,\ldots,n-1$ there are elements $x_{b(i)}$ with $b(i) > i$ such that (\ref{eq:Hypertree}) holds.
\end{definition}
In the literature, a hypergraph is usually defined as a set of subsets of some set of nodes, in other words as a set of elements of a lattice of subsets of a set. In a generalization of this view we take a hypergraph to be a set of elements of some join-semilattice $(Q,\leq)$. The concept of a hypertree given in Definition \ref{def:Hypertree} is then a transcription of the usual definition of a hypertree in the context of subset lattices. Hypertrees in the classical sense are studied for instance in relational database theory, where they are also called acyclic hypergraphs and shown to have desirable properties \cite{beeri81,beeri83,maier83}. In particular, hypertrees are interesting with respect to computational complexity \cite{gottlobleonescarcello99,gottlob99comparison,gottlobleonescarcello01}. These papers treat all hypertrees in the multivariate framework, wheres we take up this issue in the following sections in our more general case of hypertrees in q-separoids.

So, any Markovtree determines a hypertree, even several different hypertrees, according to the numbering of nodes selected. The sequence $x_1,\ldots,x_n$ defining the hypertree is also called a hypertree construction sequences \cite{shafer96}. Any hypertree construction sequence $x_1,\ldots,x_n$ defines a tree $T = (V,E)$ with nodes $V = \{1,\ldots,n\}$ and edges $E = \{\{i,b(i)\},i=1,\ldots,n\}$. In fact, $T$ is connected: if $i$ and $j$ are two nodes, then the node sequence $i,b(i),b(b(i)),\ldots$ and $j,b(j),b(b(j)),\ldots$ define both paths from $i$ and $j$ to $n$ respectively.  And since the number of edges is one less the number of nodes, the graph must be a tree. 

However, the labeling $i \mapsto x_i$ in this tree does not, in general, give a Markov tree. To see this consider a construction sequence $\{x_1,x_2,x_3,x_4\}$ such that $x_1 \bot x_2 \vee x_3 \vee x_4 \vert x_4$ and $x_2 \bot x_3 \vee x_4 \vert x_4$. Then $S = \{x_1,x_2,x_3,x_4\}$ is a hypertree. The construction sequence defines the tree $T = (\{1,2,3,4\},\{\{1,4\},\{2,4\},\{3.4\}\})$. In order for this tree to be a Markov tree we should have $\bot \{x_1,x_x,x_3\} \vert x_4$ and for this to be valid, for instance $x_1 \vee x_2 \bot x_3 \vert x_4$ must hold. But this is not guaranteed by the construction sequence. However, we shall see that if $(Q,\leq)$ is a distributive lattice, then in a q-separoid $(Q,\leq,\bot_L)$ any hypertree defines by the tree obtained from its construction sequence indeed a Markov tree.

Let $(T,\lambda)$ again be a Markov tree and consider two nodes $u$ and $v$. Let $w$ be any node on the (unique) path between $u$ and $v$, different from $u$ and $v$. Let $u'$ and $v'$ be the neighbors of $w$ on the path from $u$ to $w$ and $v$ to $w$ respectively. Then from the Markov property (\ref{eq:MarkovProp}) it follows that
\begin{eqnarray*}
\lambda(V_{w,u'}) \bot \lambda(V_{w,v'}) \vert \lambda(w)
\end{eqnarray*}
and therefore $\lambda(u) \bot \lambda(v) \vert \lambda(w)$. This holds for any node on the path between $u$ and $v$ (including $u$ and $v$ themselves). This is the defining property of another concept.

\begin{definition} \label{def:JoinTree} \textbf{Join tree}
Let $(Q,\leq,\bot)$ be a q-separoid and $(T,\lambda)$ a labeled tree with $T = (V,E)$. If for any pair of nodes $u$ and $v$ and for any node $w$ on the path between $u$ and $v$
\begin{eqnarray} \label{eq:JoinTreeProp}
\lambda(u) \bot \lambda(v) \vert \lambda(w),
\end{eqnarray}
then $(T,\lambda)$ is called a join tree.
\end{definition}

Join trees have been considered in relational database theory  \cite{beeri83,maier83} and, under varying names, also in local computation theory \cite{lauritzenspiegelhalter88,cowelletal99,shenoyshafer90}, but only in the multivariate setting. In this case the cocept of a join tree is also connected with a diversity of graphical modeling tools for representing conditional independence. In the case of a multivariate model, or more generally a commutative information algebra, we have $\lambda(u) \bot_L \lambda(v) \vert \lambda(w)$ if and only if $\lambda(u) \wedge \lambda(v) \leq \lambda(w)$. This is the well-know running intersection property of join trees. In our general case however, $(Q,\leq)$ is not necessarily a lattice, hence meet may not exist, but Definition \ref{def:JoinTree} above catches the essence of the concept of a join tree.

Again, any Markov tree is a join tree, but also again, the converse does not hold. Consider the same tree $T = (\{1,2,3,4\},\{\{1,4\},\{2,4\},\{3.4\}\})$ as above, and assume $x_1 \bot x_2 \vert x_4$, $x_1 \bot x_3 \vert x_4$ and $x_2 \bot x_3 \vert x_4$. Then $T$ labeled with $x_1$ to $x_4$ is a join tree. But the pairwise conditional independence relations are not sufficient to imply $\bot \{x_1.x_2,x_3\} \vert x_4$, except if in the q-separoid $(Q,\leq,\bot_L)$ the lattice $(Q,\leq)$ is \textit{distributive}. In fact, if $(Q,\leq)$ is a distributive lattice, then then the three concepts of a Markov tree, a hypertree and a join tree turn out to be equivalent, a fact that is well-known in the framework of multivariate models. 

Before we prove this result, we show that a hypertree in the q-separoid $(Q,\leq,\bot_L)$ induces always a join tree. It is an open question whether this is true for any q-separoid.

\begin{theorem} \label{th:HyptreeJoinTree}
Let $(Q,\leq)$ be a lattice and $(Q,\leq,\bot_L)$ the associated q-separoid. If $S \subseteq Q$ is a hypertree with construction sequence $x_1,\ldots,x_n$ then the labeled tree $(T,\lambda)$ with $T = (V,E)$, $V = \{1,\ldots,n\}$, $E = \{\{i,b(i)\}:i = 1,\ldots,n-1\}$ and $\lambda(i) = x_i$ is a join tree.
\end{theorem}

\begin{proof}
Consider two nodes $i$ and $j$ and the path between $i$ and $j$. Note that by definition of $T$ $b(i),b(b(i)),\ldots$ is as sequence of neighboring node, starting with the neighbor of $i$, on the path from $i$ to $n$. The same holds for $b(j),b(b(j)),\ldots$, starting with a neighbor of $j$. The two paths from $i$ to $n$ and from $j$ to $n$ meet in a node $h \leq n$, where $i,j \leq h$, $i=h$ or $j=h$ not excluded. We have either $i < j$ or $j < i$ Assume $i < j$. Then there is in the sequence of nodes $b(i),b(b(i)),\ldots$ a first node $i_1$ so that $i_1 \geq j$, $i_1 = j$ not excluded. Further, there is in the sequence $b(j),b(b(j)),\ldots$, a first node $j_1$ so that $j_1 \geq i_1$. If $j_1 \not= h$, then there is a next node $i_2$ in the first sequence so that $i_2 \geq j_1$, then again a node $j_2 \geq i_2$, etc until $h$ is reached.

Now, by the hypertree condition (\ref{eq:Hypertree}) we have, since $i < j$,
\begin{eqnarray*}
x_i \wedge x_j \leq x_i \wedge (\vee_{k = i+1}^n x_k)  \leq (x_i \vee x_{b(i)}) \wedge (\vee_{k = i+1}^n x_k) = x_{b(i)}.
\end{eqnarray*}
If we iterate this argument with $x_i \wedge x_j \leq x_{b(i)} \wedge x_j \leq x_{b(b(i))}$ until $i_1$ is reached, then we can conclude that $x_i \wedge x_j \leq x_k$ for any node $k$ on the path from $i$ to $i_1$. Then using the same argument on
\begin{eqnarray*}
x_{i_1} \wedge x_j \leq x_j \wedge (\vee_{k=j+1}^n x_k) \leq x_{b(j)},
\end{eqnarray*}
and iterating this up to $j_1$, we obtain $x_i \wedge x_j \leq x_k$ for any node on the path from $j$ to $j_1$, Alternating this reasoning between the two paths from $i$ and $j$ to $n$, node $h$ is finally reached and then we $x_i \wedge x_j \leq x_k$ for all nodes on the path from $i$ to $j$. If $j < i$ the same procedure applies. So we have proved that $(T,\lambda)$ is a join tree.
\end{proof}

Now, we can prove the equivalence of the concepts of Markov trees, hyper trees and join trees with respect to a q-seproid $(Q,\leq,\bot_L)$ if $(Q,\leq)$ is a distributive lattice.

\begin{theorem} \label{th:TotEquiv}
Let $(Q,\leq)$ be a distributive lattice and $(Q,\leq,\bot_L)$ the associated q-separoid. If the labeled tree $(T,\lambda)$ with $T = (V,E)$ is a join tree, then
\begin{enumerate}
\item the set $\lambda(V)$ is a hypertree,
\item the labeled tree $(T,\lambda)$ is a Markov tree.
\end{enumerate}
\end{theorem}

\begin{proof}
We need to find a hypertree construction sequence. For this purpose select any node $v \in V$ and let the number of nodes $\vert V \vert = n$. Then there is a numbering of nodes $i : V \rightarrow \{1,\ldots,n\}$ such that $i(v) = n$ and $i(u) < i(w)$ if node $w$ is on the path between nodes $u$ and $v$. Define $x_{i(u)} = \lambda(u)$. We claim that $x_1,x_2,\ldots,x_n$ is a hypertree construction sequence and hence $\lambda(V)$ a hypertree. In order to prove this we identify the nodes with their number in the numbering above and define $b(i) = j$, if $i < j$ and $\{i,j\} \in E$. Note that $b(i)$ is uniquely determined, since there is only one path from $i$ to $n$. Now, by distributivity,
\begin{eqnarray*}
x_i \wedge (\vee_{j=i+1}^n x_j) = \vee_{j=i+1}^n (x_i \wedge x_j).
\end{eqnarray*}
For $i < j$, the path from $i$ to $j$ passes through $b(i)$, so that by the join tree property $x_i \wedge x_j \leq x_{b(i)}$ for all $j=i+1,\ldots,n$. Therefore,
\begin{eqnarray*}
x_i \wedge (\vee_{j=i+1}^n x_j) \leq x_{b(i)}.
\end{eqnarray*}
On the other hand, since $i+1 \leq b(i) \leq n$, we have also
\begin{eqnarray*}
x_i \wedge (\vee_{j=i+1}^n x_j) \geq x_i \wedge x_{b(i)}, 
\end{eqnarray*}
hence 
\begin{eqnarray*}
x_i \wedge (\vee_{j=i+1}^n x_j) = x_i \wedge x_{b(i)} \leq x_{b(i)}, 
\end{eqnarray*}
In a distributive lattice this is equivalent to $x_i \bot_L \vee_{j=i+1}^n x_j \vert x_{b(i)}$. This means that $x_1,x_2,\ldots,x_n$ is indeed a hypertree construction sequence.

Since $(Q,\leq)$ is a distributive lattice, the Markov property (\ref{eq:MarkovProp}) holds if and only if $\lambda(V_{v,w}) \bot_L \lambda(V_{v,u}) \vert \lambda(v)$ for all pairs $u,w$ of distinct neighbours of $v$, as noted above. We claim that these pairwise conditional independence properties hold in a join tree. In fact, by distributivity,
\begin{eqnarray*}
\lefteqn{ \lambda(v) \leq (\lambda(V_{v,w}) \vee \lambda(v)) \wedge (\lambda(V_{v,u}) \vee \lambda(v)) } \\
&&=\left( \bigvee_{w' \in V_{v,w}} \lambda(w') \vee \lambda(v)  \right) \wedge \left( \bigvee_{u' \in V_{v,u}} \lambda(u') \vee \lambda(v)  \right) \\
&&=\left( \bigvee_{w' \in V_{v,w},u' \in V_{v,u}} (\lambda(w') \wedge \lambda(u'))  \right) \\
&&\vee \left( \bigvee_{w' \in V_{v,w}} (\lambda(w') \wedge \lambda(v))  \right) \vee \left( \bigvee_{u' \in V_{v,u}} (\lambda(u') \vee \lambda(v))  \right) \vee \lambda(v) \\
&&\leq \lambda(v),
\end{eqnarray*}
by the join tree property (\ref{eq:JoinTreeProp}) since $v$ is on all paths from nodes $w' \in V_{v,w}$ to nodes $u' \in V_{v,u}$. Therefore we have
\begin{eqnarray*}
\left( \bigvee_{w' \in V_{v,w}} \lambda(w') \vee \lambda(v)  \right) \wedge \left( \bigvee_{u' \in V_{v,u}} \lambda(u') \vee \lambda(v)  \right) = \lambda(v)
\end{eqnarray*}
and this is $\lambda(V_{v,w}) \bot_L \lambda(V_{v,u}) \vert \lambda(v)$. So, $(T,\lambda)$ is a Markov tree.
\end{proof}

In summary, a Markov tree induces a hypertree and is also a join tree. The converse does not hold in general, but for q-separoid $(Q,\leq,\bot:L)$, where $(Q,\leq)$ is a distributive lattice, join trees are Markov tree and hypertrees induce Markov trees. This is true in particular for multivariate models.

.


\section{Markov tree propagation} \label{subsec:MarkovPropag}

A basic computational problem regarding information algebras consists in determining the extraction of information relative to one or several different question from a number of pieces of information. More precisely,, consider a domain-free information algebra $(\Phi,\cdot,0,1;E)$ with $E = \{\epsilon_x:x \in Q\}$. Suppose a family of pieces of information $\phi_1,\ldots,\phi_n$ from $\Phi$ are given and let $\phi = \phi_1 \cdot \ldots \cdot \phi_n$ be the combined information. Then the projection problem consists in computing
\begin{eqnarray*}
\epsilon_x(\phi) = \epsilon_x(\phi_1 \cdot \ldots\cdot \phi_n)
\end{eqnarray*}
for a question $x \in Q$, or for several different questions $x_1,\ldots,x_m$. A corresponding labeled version of the projection problem can also be formulated, and in fact, in computational studies, labeled versions are usually considered. In our general discussion here however, we stick to the domain-free version. 

Any piece of information $\phi_i$ for $i=1,\ldots,n$ has some support $x_i$, if we assume the Support Axiom as we shall do in this section. It is conceivable that the complexity of the basic operations of combination and extraction depend on the support of the pieces of information involved, or the label of them in the labeled view. In a set algebra for instance a piece of information with support $x$ can be seen as a subset of blocks of partition $P_x$ and the coarser the partition, the less space is needed for storing and the less operations are to executed for combination (intersection) or extraction (saturation). The same observation applies to other examples, see Sections \ref{sec:UncertainInf} and \ref{sec:ProbabInf} for instance. So we may assume that a complexity measure $c(x)$ is monotone in the order of $Q$, that is $x \leq y$ implies $c(x) \leq c(y)$. In view of this  the naive solution of the projection problem, where one piece of information after the other is combined becomes problematic, since, if the factors $\phi_i$ have supports $x_i$, the successive combinations have supports $x_1 \vee x_2$, $x_1 \vee x_2 \vee x_3$ up to $x_1 \vee \ldots \vee x_n$ and the operation of combination and extraction become more and more expensive. The solution to this problem consists in so-called \textit{local computation} schemes, where combination and extraction is performed, if possible, only on the supports $x_i$ of the factors of the projection problem. Such a scheme has first been proposed in \cite{lauritzenspiegelhalter88} for probabilistic networks and then extended by \cite{shenoyshafer90} for more general formalisms, especially belief functions. These approaches were however all in the framework of multivariate models. Here we show that local computation schemes are also possible in our more general frame.

The key for this is provided by Markov trees. Consider a Markov tree $(T,\lambda)$ with $T =(V,E)$ such that for any $\phi_i$ of the projection problem with support $x_i$ there is a node $v \in V$ with $\lambda(v) = x_i$. Without loss of generality we may assume that for the projection problem we have
\begin{eqnarray*}
\phi = \prod_{v \in V} \phi_v.
\end{eqnarray*}
In fact, if there are nodes $v$ such that there is no $i$ such that $\lambda(v) = x_i$, then let $\phi_v = 1$ and if there are nodes such that $\lambda(v) = x_i$ for several factor $\phi_i$, then combine them. Further, we assume that in the projection problem $x = x_n$ (or $x \leq x_n$), that is, we want to extract the information of the combination relative to the label of one of the nodes of the Markov tree. If this is originally not the case, then we extend the Markov tree to cover $x$. This and related issues will discussed below at the end of the section.

So we consider now the projection problem
\begin{eqnarray} \label{eq:MarkTreeFact}
\epsilon_{x_n}(\phi) = \epsilon_{x_n}(\prod_{v \in V} \phi_v)
\end{eqnarray}
where $(T,\lambda)$ is a Markov tree and $\phi_v$ has support $\lambda(v)$ for all $v \in V$. Then we call $\prod_{v \in V}$ a \textit{Markov tree factorization}.  The corresponding projection problem has a local computation solution as the following theorem shows. 

\begin{theorem} \label{th:MarkovTreeProj}
Let $(T,\lambda)$ be a Markov tree with $T = (V,E)$ and $\phi$ given by Markov tree factorization (\ref{eq:MarkTreeFact}) according to this Markov tree. Then, for any node $v \in V$
\begin{eqnarray} \label{eq:TreeRecur}
\epsilon_{\lambda(v)}(\phi) = \phi_v \cdot \prod_{u \in ne(v)} \epsilon_{\lambda(v)}(\epsilon_{\lambda(u)}(\phi_{v,u})),
\end{eqnarray}
where
\begin{eqnarray} \label{eq:SubtreeFact}
\phi_{v,u} = \prod_{w \in V_{v,u}} \phi_w
\end{eqnarray}
and $V_{v,u}$ is the node set of the subtree $T_{v,u}$ rooted in the neighbor node $u$ of $v$ obtained by eliminating node $v$ from $T$.
\end{theorem}

\begin{proof}
Note that $\lambda(V_{v,u})$ is a support of $\phi_{v,u}$ as defined in (\ref{eq:SubtreeFact}) and by Theorem \ref{th:NeighCondIndep} $\lambda(v) \bot \lambda(V_{v,u}) \vert \lambda(u)$. Therefore we have
\begin{eqnarray*}
\epsilon_{\lambda(v)}(\phi_{v,u}) = \epsilon_{\lambda(v)}(\epsilon_{\lambda(u)}(\phi_{v,u}))
\end{eqnarray*}
Further,
\begin{eqnarray*}
\phi = \phi_v \cdot \prod_{u \in ne(v)} \phi_{v,u}.
\end{eqnarray*}
By property C1 of a q-separoid $\lambda(v) \bot \vee_{u \in ne(v)} \lambda(V_{v,u}) \bot \lambda(v)$, and therefore
\begin{eqnarray*}
\epsilon_{\lambda(v)}(\phi) = \epsilon_{\lambda(v)}(\phi_v) \cdot \epsilon_{\lambda(v)}(\prod_{u \in ne(v)} \phi_{v,u}).
\end{eqnarray*}
From the Markov property (\ref{eq:MarkovProp}) it follows that 
\begin{eqnarray*}
\epsilon_{\lambda(v)}(\prod_{u \in  ne(v)} \phi_{v,u}) = \prod_{u \in ne(v)} \epsilon_{\lambda(v)}(\phi_{v,u}).
\end{eqnarray*}
Finally, $\lambda(v)$ is a support of $\phi_v$, such that, if we combine the last identity with the former one, we obtain
\begin{eqnarray*}
\epsilon_{\lambda(v)}(\phi) = \phi_v \cdot \prod_{u \in ne(v)} \epsilon_{\lambda(v)}(\epsilon_{\lambda(u)}(\phi_{v,u})),
\end{eqnarray*}
which concludes the proof.
\end{proof}

Formula (\ref{eq:TreeRecur}) defines a tree recursion on the tree $T$ since the subtrees $T_{v,u}$ are again Markov trees. The operations occurring in this formula are a combination on label $\lambda(v)$ and and extractions on labels $\lambda(u)$. In this sense Theorem \ref{th:MarkovTreeProj} establishes a local computation scheme.

Once the projection of $\phi$ to the root $v$ has been computed, the projection of $\phi$ to any other node of the Markov tree can be obtained, provided the intermediate results $\epsilon_{\lambda(u)}(\epsilon_{\lambda(w)}(\phi_{u,w}))$ have been cached during the recursion. Indeed we have for $u \in ne(v)$,
\begin{eqnarray*}
\epsilon_{\lambda(u)}(\phi_{u,v}) = \phi_v \cdot \prod_{w \in ne(v),w \not= u} \epsilon_{\lambda(v)}(\epsilon_{\lambda(w)}(\phi_{v,w})).
\end{eqnarray*}
Then, using this, and the cached intermediate results of the recursion, according to Theorem \ref{th:MarkovTreeProj} we obtain with node $u$ as the new root
\begin{eqnarray*}
\epsilon_{\lambda(u)}(\phi) = \phi_u \cdot \prod_{w \in ne(u)} \epsilon_{\lambda(u)}(\epsilon_{\lambda(w)}(\phi_{u,w})).
\end{eqnarray*}
In this way we can work backwards the tree until the projections of $\phi$ has been obtained for all nodes. In the following section, an equivalent, but more systematic non-recursive computational scheme will be proposed.

In the case of a commutative information algebra, we note that (\ref{eq:TreeRecur}) simplifies slightly to 
\begin{eqnarray*}
\epsilon_{\lambda(v)}(\phi) = \phi_v \cdot \prod_{u \in ne(v)} \epsilon_{\lambda(v) \wedge \lambda(u)}(\phi_{v,u}).
\end{eqnarray*}

The question arises whether there is a Markov tree for any projection problem, and how to find it. The second question is, to the best of our knowledge, an open question. In the multivariate case there is a huge body of literature on methods to find a good join tree. It is not possible at this place to survey it. But the approaches in the multivariate can not easily be transported to our present more general case because they depend in some way or other to a successive elimination of variables and on graphical methods. But we want to make a few observations. Supports $x_i$ of factors $\phi_i$ of a projection problem are not unique. For instance any $x'_i \geq x_i$ is also a support of $\phi_i$. So if $x_1,\ldots,x_n$ may not define a Markov tree, may be some larger $x'_1,\ldots,x'_n$ do. In fact trivially, the one node tree $\{v\}$ with label $x = x_1 \vee \ldots x_n$ is a Markov tree for the projection problem, albeit of course not a very usefull one. It may also be that some $x''_i \leq x_i$ is still a support of $\phi_i$ and such smaller domains $x''_1,\ldots,x''_n$ may define a Markov tree. This would then be a desirable situation, since it reduces the complexity of computation. So there may be a multitude of Markov tree factorizations for a giver problem and the questions is how to find a good or even best one.


\section{Computation in a hypertree} \label{subsec:CompHypTree}

Local computation schemes are also available relative to a hypertree. We reconsider the projection problem
\begin{eqnarray*}
\phi = \phi_1 \cdot \dots \cdot \phi_n,
\end{eqnarray*}
where the $\phi_i$ have supports $x_i$ for $i=1,\ldots,n$. We suppose now that $x_1,\ldots,x_n$ is a hypertree construction sequence and we want to compute
\begin{eqnarray*}
\epsilon_{x_n}(\phi) = \epsilon_{x_n}(\phi_1 \cdot \dots \cdot \phi_n).
\end{eqnarray*}
In order to construct a local computation scheme, let's try to eliminate the factors $\phi_1,\phi_2,\ldots$ one after the other. To eliminate $\phi_1$ means to extract the information for $x_2 \vee \ldots \vee x_n$ from the the combination $\phi$. So define, more generally
\begin{eqnarray*}
y_i = x_{i+1} \vee \ldots \vee x_n
\end{eqnarray*}
for $i=1,\ldots,n-1$. Let's start to compute $\epsilon_{y_1}(\phi)$, that is
\begin{eqnarray*}
\epsilon_{y_1}(\phi) = \epsilon_{y_1}(\phi_1 \cdot \phi_2 \cdot \ldots \cdot \phi_n) = \epsilon_{y_1}(\phi_1) \cdot \phi_2 \cdot \ldots \cdot \phi_n,
\end{eqnarray*}
since $\phi_2 \cdot \ldots \cdot \phi_n$ has support $y_1$. The hypertree condition $x_1 \bot y_1 \vert x_{b(1)}$, see (\ref{eq:Hypertree}) implies $\epsilon_{y_1}(\phi_1) = \epsilon_{y_1}(\epsilon_{x_{b(1)}}(\phi_1))$ and therefore 
\begin{eqnarray*}
\epsilon_{y_1}(\phi) = \epsilon_{y_1}(\epsilon_{x_{b(1)}}(\phi_1)) \cdot \phi_2 \cdot \ldots \cdot \phi_n.
\end{eqnarray*}
Since $x_{b(1)} \leq y_1$, we conclude that
\begin{eqnarray*}
\epsilon_{y_1}(\phi)  
= \epsilon_{x_{b(1)}}(\phi_1) \cdot \phi_2 \cdot \ldots \cdot \phi_n.
\end{eqnarray*}
Define $\psi_i^1 =: \phi_i$ and then $\psi_{b(1)}^2 =: \psi_{b(1)}^1 \cdot \epsilon_{x_{b(1)}}(\psi_1^1)$ and $\psi_i^2 =: \psi_i^1$ for $i=2,\ldots,n$, $i \not= b(1)$. Note that all $\psi_i^2$ have still support $x_i$ for all $i$ from $2$ to $n$. So, we obtain a new factorization after elimination of $\phi_1$,
\begin{eqnarray*} 
\epsilon_{y_1}(\phi) = \psi_2^2 \cdot \ldots \cdot \psi_n^2.
\end{eqnarray*}

We may now proceed in exactly the same way to eliminate $\psi_2^2,\psi_3^3 \ldots$ etc. By induction lets assume
\begin{eqnarray} \label{eq:RecFact}
\epsilon_{y_{i-1}}(\phi) = \psi_i^i \cdot \ldots \cdot \psi_n^i.
\end{eqnarray}
and each $\psi_j^i$ has support $x_j$. Since $y_i \leq y_{i-1}$ we have $\epsilon_{y_i} = \epsilon_{y_i}\epsilon_{y_{i-1}}$. Now we eliminate $\psi_i^i$ from (\ref{eq:RecFact}) in the same was as we did above and obtain
\begin{eqnarray*}
\epsilon_{y_i}(\phi)  
&=& \epsilon_{y_i}(\psi_i^i \cdot \psi^i_{i+1} \ldots \cdot \psi_n^i) \\
&=& \epsilon_{y_i}(\epsilon_{x_{b(i)}}(\psi_i^i)) \cdot \psi_{i+1}^i \cdot \ldots \cdot \psi_n^i \\
&=&= \epsilon_{x_{b(i)}}(\psi_i^i )\cdot  \psi_{i+1}^i \cdot \ldots \cdot \psi_n^i.
\end{eqnarray*}
Define
\begin{eqnarray} \label{eq:IndStep}
\psi_{b(i)}^{i+1} =: \psi_{b(i)}^i \cdot \epsilon_{x_{b(i)}}(\psi_i^i).
\end{eqnarray}
and $\psi_j^{i+1} =: \psi_j^i$ for $j=i+1,\ldots,n$, $j \not= b(i)$. Then we obtain the new factorzation
\begin{eqnarray*}
\epsilon_{y_i}(\phi) = \psi_{i+1}^{i+1} \cdot \ldots \cdot \psi_n^{i+1},
\end{eqnarray*}
where the fastors $\psi_j^{i+1}$ again we still have the old supports $x_j$ for $j=i+1,\dots,n$. This concludes the induction step. At the end, for $i = n-1$, we obtain
\begin{eqnarray*}
\epsilon_{x_n}(\phi) = \psi_n^n.
\end{eqnarray*}
This solves the projection problem on the hypertree $\{x_1,\ldots,x_n\}$ similar to the Markov tree propagation. And it does so by local computation: in any step (\ref{eq:IndStep}) we extract on domain $x_i$ and combine on domain $x_{b(i)}$ and this for $i=1$ up to $i=n-1$.

We may in a second step also compute $\epsilon_{x_i}(\phi)$ for $i = n-1,\ldots,1$. This is formulated in the following theorem.

\begin{theorem}
Let $x_1,\ldots,x_n$ be a hypertree construction sequence and $\psi_i^i$ for $i = n.\ldots,1$ be as defined during the algorithm as described above. Then, for $i=n-1,\ldots,1$ 
\begin{eqnarray} \label{eq:BackIt}
\epsilon_{x_i}(\phi) = \epsilon_{x_i}(\epsilon_{x_{b(i)}}(\phi)) \cdot \psi_i^i.
\end{eqnarray}
\end{theorem}

\begin{proof}
As before, define $y_i = x_{i+1} \vee \ldots \vee x_n$. Since $x_{b(i)} \leq y_i$ 
\begin{eqnarray*}
\epsilon_{x_i}(\epsilon_{x_{b(i)}}(\phi)) \cdot \psi_i^i = \epsilon_{x_i}(\epsilon_{x_{b(i)}}(\epsilon_{y_i}(\phi))) \cdot \psi_i^i .
\end{eqnarray*}
Since $x_1,\ldots,x_n$ is a hypertree construction sequence, we have $x_i \bot y_i \vert x_{b(i)}$, hence, using (\ref{eq:RecFact}) 
\begin{eqnarray*}
\epsilon_{x_i}(\epsilon_{x_{b(i)}}(\phi)) \cdot \psi_i^i = \epsilon_{x_i}(\epsilon_{y_i}(\phi)) \cdot \psi_i^i 
= \epsilon_{x_i}(\psi_{i+1}^{i+1} \cdot \ldots \cdot \psi_n^{i+1}) \cdot \psi_i^i.
\end{eqnarray*}
Using (\ref{eq:IndStep}) we obtain further
\begin{eqnarray*}
\epsilon_{x_i}(\epsilon_{x_{b(i)}}(\phi)) \cdot \psi_i^i = \epsilon_{x_i}(\psi_i^i \cdot \psi_{i+1}^i \cdot \ldots \cdot \psi_n^i \cdot \epsilon_{x_{b(i)}}(\psi_i^i)).
\end{eqnarray*}
By idempotency we have $\psi_i^i \cdot \epsilon_{x_{b (i)}}(\psi_i^i) = \psi_i^i$. Therefore it follows that 
\begin{eqnarray*}
\psi_i^i \cdot \psi_{i+1}^i \cdot \ldots \cdot \psi_n^i \cdot \epsilon_{x_{b(i)}}(\psi_i^i) = \psi_i^i \cdot \psi_{i+1}^i \cdot \ldots \cdot \psi_n^i.
\end{eqnarray*}
From this we obtain finally
\begin{eqnarray*}
\epsilon_{x_i}(\epsilon_{x_{b(i)}}(\phi)) \cdot \psi_i^i = \epsilon_{x_i}(\psi_i^i \cdot \psi_{i+1}^i \cdot \ldots \cdot \psi_n^i) = \epsilon _{x_i}(\epsilon_{y_{i-1}}(\phi)) = \epsilon_{x_i}(\phi),
\end{eqnarray*}
since $x_i \leq y_{i-1}$. This concludes the proof.
\end{proof}

According to this theorem, once $\epsilon_{x_n}(\phi)$ has been computed by the scheme above, the other extractions $\epsilon_{x_i}(\phi)$ for $i=n-1,\ldots,1$ can be computed in this inverse order of the construction sequence. At step $i$ the extraction $\epsilon_{x_j}(\phi)$ is known for all $j \geq i$ and then by (\ref{eq:BackIt}) $\epsilon_{x_{i-1}}(\phi)$ can be computed since $b(i-1) \geq i$. The problem of how to find a hypertree construction sequence for a given projection problem is similar to the one for Markov tree. It is an open question in our general framework. For the multivariate case all reduces to find a join tree, and for this good methods are known, see the end of Section \ref{subsec:MarkovPropag}.


%% file: chapter8.tex

\chapter{Finite information} \label{sec:FiniteInf}


\section{Compact information algebras} \label{subsec:CompInfAlg}

In information processing only ``finite'' pieces of information can be handled. ``Infinite'' pieces of information can however possibly be approximated by ``finite'' ones. This aspect of finiteness will be addressed in this section, although it must be stressed that not all aspects of it will be treated. For instance no questions of computability and related issues will be considered. On the other hand, many aspects of finiteness as discussed here are also considered in \textit{domain theory}, in fact much of this section is motivated by domain theory. However, the one critical issue not addressed in domain theory is the one of information extraction. Also domain theory places almost exclusively emphasis on order and approximation, whereas combination is neglected. So, although the subject is similar to domain theory, it is treated here with a somewhat different focus. 

Consider a domain-free information algebra $(\Phi,\cdot,0,1;E)$ with $E = \{\epsilon_x:x \in Q\}$. In the set $\Phi$ of pieces of information we want to single out a subset of elements to be considered as finite. An important role for this task play directed sets in the ordered set $(\Phi,\leq)$, where $\leq$ denotes the information order, see Section \ref{sec:InfOrder}. A subset $D$ of $\Phi$ is called \textit{directed} if it is not empty, and if with any two elements $\phi_1$ and $\phi_2$ in $D$, there is an element $\phi \in D$ which dominates both, $\phi_1,\phi_2 \leq \phi$. Directed subsets $D$ of $\Phi$ are used to define convergence. The limit of the directed set $D$ is its supremum $\bigsqcup D$, where the symbol $\bigsqcup$ indicates the supremum of a directed set. So, a directed subset $D$ of $\Phi$ is said to converge in $\Phi$ if $\bigsqcup D$ exists in $\Phi$.

Now let's single out a subset $\Phi_f$ of $\Phi$ of elements which we consider as finite elements. We require first that finite elements are closed under combination, and we consider that the neutral element $1$ and the null elements are finite,  and that all directed subsets $D$ of $\Phi_f$ converge, that is $\bigsqcup D$ exists and is an element of $\Phi$. But we want more: Any element $\phi$ of $\Phi$ should be approximated by the finite elements it dominates, that is
\begin{eqnarray*}
\phi = \bigsqcup \{\psi \in \Phi_f:\psi \leq \phi\}.
\end{eqnarray*}
This means that the finite elements $\Phi_f$ are \textit{dense} in $\Phi$. We require an even stronger property, namely that any element $\phi$ of $\Phi$ with support $x$ should be approximated by finite elements with the same support,
\begin{eqnarray*}
\phi = \bigsqcup \{\psi \in \Phi_f:\psi \leq \phi,\epsilon_x(\psi) = \psi\},
\end{eqnarray*}
if $\phi$ has support $x$. So the finite elements of $\Phi$ must be dense in the subalgebra $\epsilon_x(\Phi)$. This we call \textit{Local Density}.

But this does not yet characterize finiteness sufficiently. One thing which follows from density is that if $\phi$ is finite, then it belongs itself to the directed set of finite elements approximating it. This is certainly an important property of finiteness, but again, we need more. We may possibly approximate an element $\phi$ by a directed set $D$ of finite elements which is smaller than the set of all finite elements dominated by $\phi$, $\phi = \bigsqcup D$. Then, if $\psi$ is a finite element such that $\psi \leq \bigsqcup D$, there must be an element $\phi \in D$ such that $\psi \leq \phi$. This we call \textit{compactness}. As we shall see below this is closely related to the compactness property in order theory, \cite{daveypriestley97}. So, in summary, we require the set $\Phi_f$ of finite elements in $\Phi$ to satisfy the following properties:
\begin{enumerate}
\item \textit{Combination:} If $\psi_1,\psi_2 \in \Phi_f$, then $\psi_1 \cdot \psi_2 \in \Phi_f$, and $0,1 \in \Phi_f$,
\item \textit{Convergence:} If $D \subseteq \Phi_f$ is a directed set, then $\bigsqcup D$ exists and belongs to $\Phi$.
\item \textit{Local Density:} For all $\phi \in \Phi$.
\begin{eqnarray*}
\epsilon_x(\phi) = \bigsqcup \{\psi \in \Phi_f:\psi \leq \phi,\epsilon_x(\psi) = \psi\}.
\end{eqnarray*}
\item \textit{Compactness:} If $D \subseteq \Phi_f$ is a directed set and $\psi \in \Phi_f$ such that $\psi \leq \bigsqcup D$, then there is a $\phi \in D$ such that $\psi \leq \phi$.
\end{enumerate}

A system $(\Phi,\Phi_f,\cdot,0.1;E)$, where $(\Phi,\cdot,0,1;E)$ is a domain-free information algebra and $\Phi_f$ a subset of $\Phi$ satisfying the four conditions above, is called a \textit{compact information algebra}. Note that local density implies density, if the Support axiom is valid, since any element $\phi \in \Phi$ has then a support $x$ so that by Local density and Convergence,
\begin{eqnarray*}
\phi = \epsilon_x(\phi) = \bigsqcup \{\psi \in \Phi_f:\psi \leq \phi,\epsilon_x(\psi) = \psi\} \leq
\bigsqcup \{\psi \in \Phi_f:\psi \leq \phi\} \leq \phi.
\end{eqnarray*}
The converse however does not hold in general, density does not imply local density. Note that any finite information algebra $\Phi$ is trivially compact with $\Phi_f = \Phi$. Here follow for illustration two simple examples, string algebras and set algebras of convex sets. Further examples can be found in Sections \ref{sec:UncertainInf} and \ref{sec:ProbabInf}.

In a string algebra (see Section \ref{subsec:AtomicAlg}), the finite elements are finite strings. A directed set in this algebra is a monotone sequence of (finite) strings, where each string is prefix of a following one. The supremum of such a directed set $D$ of strings is then the shortest string, finite or infinite, such that all elements of $D$ are prefix of it. This shows that the Convergence, the Local density and the Compactness property are all valid.

Convex sets in a linear space like $\mathbb{R}^n$ are set algebras in a multivariate setting. Indeed intersection of convex sets are convex and cylindrification of convex sets yields convex sets. The finite elements here are convex polyhedra. The approximation of a convex set by convex polyhedra is from the outside, by polyhedra containing the convex set. 

Here follows a main result about compact information algebras.

\begin{theorem} \label{th:AlgLatt}
Let $(\Phi,\Phi_f,\cdot,0.1;E)$ be a compact information algebra. Then
\begin{enumerate}
\item $(\Phi,\leq)$ is a complete lattice under information order,
\item An element $\psi \in \Phi$, $\psi \not= 0$, belongs to $\Phi_f$ if and only if for every directed subset $D$ of $\Phi$, $\psi \leq \bigsqcup D$ implies there is a $\phi \in D$ such that $\psi \leq \phi$,
\item An element $\psi \in \Phi$ belongs to $\Phi_f$ if and only if for all subsets $X$ of $\Phi$, $\psi \leq \bigvee X$ implies there is a finite subsets $Y$ of $X$ such that $\psi \leq \bigvee Y$.
\end{enumerate}
\end{theorem}

\begin{proof}
The proof follows the one given in \cite{kohlas03}. Let $X$ be any non-empty subset of $\Phi$. Define $Y$ to be the set of finite elements smaller than all elements of $X$, $Y = \{\psi \in \Phi_f:\psi \leq \phi\}$. This set is not empty, because $1$ is a finite element. Then $Y$ is a directed set, since if $\psi_1$ and $\psi_2$ belong to $Y$, then $\psi_1,\psi_2 \leq \psi_1 \cdot \psi_2 \in Y$ by the Combination property. By the Convergence property the supremum $\bigsqcup Y$ exists and it is a lower bound of $X$. We claim that $\bigsqcup Y$ is the infimum of $X$. In fact, assume $\chi$ to be a lower bound of $X$. Then by the Density property $\chi = \bigsqcup \{\psi \in \Phi_f:\psi \leq \chi\}$ and $\{\psi \in \Phi_f:\psi \leq \chi\}$ is a subset of $Y$. Therefore we have $\chi \leq \bigsqcup Y$ so that indeed $\bigsqcup Y = \bigwedge X$.

Since $(\Phi,\leq)$ has a top element $0$, the set $Y$ of all elements greater than those of $X$ is not empty, and thus it has an infimum $\bigwedge Y$, which is an upper bound of $X$. But then this infimum must be the supremum of $X$, which shows that $(\Phi,\leq)$ is indeed a complete lattice. This is a standard result of lattice theory, see \cite{daveypriestley97}.

To prove 2.) assume first that $\psi \in \Phi_f$ and that $D \subseteq \Phi$ is a directed set such that $\psi \leq \bigsqcup D$. Define
\begin{eqnarray*}
Y = \{\chi \in \Phi_f:\exists \phi \in D \textrm{ such that}\ \chi \leq \phi\}.
\end{eqnarray*}
Since $D$ is directed so is $Y$. Let now $\eta$ be an element of $D$. Then the set $\{\chi \in \Phi_f:\chi \leq \eta\}$ is contained in $Y$, hence $\eta = \bigsqcup \{\chi \in \Phi_f:\chi \leq \eta\} \leq \bigsqcup Y$ which shows that $\bigsqcup Y$ is an upper bound of $D$. Therefore we conclude that $\psi \leq \bigsqcup D \leq \bigsqcup Y$. By the Compactness property there must then be an element $\chi \in Y$ such that $\psi \leq \chi$ and by the definition of $Y$ there is a $\phi \in D$ such that $\chi \leq \phi$, hence $\psi \leq \phi$. 

For the converse assume that for all directed subsets $D$ of $\Phi$ if $\psi \leq \bigsqcup D$, $\psi \in \Phi_f$, then there is a $\phi \in D$ such that $\psi \leq \phi$. Consider then the directed set $\{\chi \in \Phi_f:\chi \leq \psi\}$. Since $\psi = \bigsqcup \{\chi \in \Phi_f:\chi \leq \psi\}$,, hence $\psi \leq \bigsqcup \{\chi \in \Phi_f:\chi \leq \psi\}$, there must be a $\phi \in  \{\chi \in \Phi_f:\chi \leq \psi\}$ such that $\psi \leq \phi$. But on the other hand $\phi \in \{\chi \in \Phi_f:\chi \leq \psi\}$ implies $\phi \leq \psi$ so that $\phi = \psi$ and thus $\psi$ belongs to $\Phi_f$.

The third assertion follows from from the previous one by the following observation: Let $X$ be any subset of $\Phi$ and define
\begin{eqnarray*}
Z = \{\vee Y:Y \subseteq X,Y \textrm{ finite}\}.
\end{eqnarray*}
We claim that $Z$ is directed and $\bigvee X = \bigsqcup Z$. Indeed, $1$ belongs to $Z$, since $1 = \vee \emptyset$. If $Y_1$ and $Y_2$ are finite subsets of $X$ then $Y_1 \cup Y_2$ is finite too, is a subset of $X$ and $\vee (Y_1 \cup Y_2) \in Z$ is an upper bound of $\vee Y_1$ and $\vee Y_2$ in $Z$. So, $Z$ is directed. Clearly we have $\bigsqcup Z \leq \bigvee X$, since for all elements $\vee Y$ of $Z$, $\vee Y \leq \bigvee X$. But on the other hand, $X$ is contained in $Z$, since $\phi = \vee \{\phi\}$ for all $\phi \in X$. Hence we obtain $\bigvee X \leq \bigsqcup Z$, hence $\bigvee X = \bigsqcup Z$.

Assume then that $\psi \in \Phi_f$ and $\psi \leq \bigvee X = \bigsqcup Z$. By item 2 just proved, there is a $\phi \in Z$ such that $\psi \leq \phi = \vee Y$ for some finite subset $Y$ of $X$. Conversely, assume $\psi \leq \bigvee X = \bigsqcup Z$ and that $Y$ is a finite subset of $X$ such that $\psi \leq \vee Y$. Since $Z$ is directed and $\vee Y \in Z$, it follows by item 2 that $\psi \in \Phi_f$ and this concludes the proof.
\end{proof}

As an application, the following result shows that the extraction operators $\epsilon_x$ are continuous maps of a compact information algebra into itself, see Section \ref{subsec:Mappings} for more about continuous maps.

\begin{theorem} \label{th:ContOfExtr}
If $(\Phi,\Phi_f,\cdot,0,1;E)$ with $E = \{\epsilon_x:x \in Q\}$ is a compact information algebra, and $D$ a directed subset of $\Phi$, then
\begin{eqnarray*}
\epsilon_x(\bigsqcup D) = \bigsqcup_{\phi \in D} \epsilon_x(\phi).
\end{eqnarray*}
\end{theorem}

\begin{proof}
If $\phi \in D$, then $\phi \leq \bigsqcup D$, so that $\epsilon_x(\phi) \leq \epsilon_x(\bigsqcup D)$, hence $\epsilon_x(\bigsqcup D)$ is an upper bound for the extractions $\epsilon_x(\phi)$ for $\phi \in D$, $\bigsqcup_{\phi \in D} \epsilon_x(\phi) \leq \epsilon_x(\bigsqcup D)$. By Density we have
\begin{eqnarray*}
\epsilon_x(\bigsqcup D) = \bigsqcup \{\psi \in \Phi_f:\psi = \epsilon_x(\psi) \leq \epsilon_x(\bigsqcup D)\}.
\end{eqnarray*}
Now, $\psi = \epsilon_x(\psi) \leq \epsilon_x(\bigsqcup D) \leq \bigsqcup D$ implies, using Theorem \ref{th:AlgLatt}, that there is a $\phi \in D$ such that $\psi \leq \phi$. Then we obtain $\psi = \epsilon_x(\psi) \leq \epsilon_x(\phi)$,  hence $\epsilon_x(\bigsqcup D) \leq \bigsqcup_{\phi \in D} \epsilon_x(\phi)$. Therefore we conclude that $\epsilon_x(\bigsqcup D) = \bigsqcup_{\phi \in D} \epsilon_x(\phi)$.
\end{proof}

Note that by Theorem \ref{th:AlgLatt}, finite elements are determined by information order alone. Elements which satisfy item 2 of this theorem are called finite in order theory, \cite{daveypriestley97}. So our concept of finiteness corresponds to the one of order theory. Elements, satisfying item 3 of the theorem are called compact in order theory, and our finite elements are therefore also compact elements in this sense. It is well-known that finiteness and compactness coincide in complete lattices, \cite{daveypriestley97}. Finite elements in the order-theoretic sense are also closed under combination. This follows since if $\psi_1$ and $\psi_2$ are finite and $D$ is a directed set, such that $\psi_1,\psi_2 \leq \bigsqcup D$, then $\psi_1 \cdot \psi_2 \leq \bigsqcup D$ and there exist elements $\phi_1,\phi_2 \in D$ such that $\psi_1 \leq \phi_1$ and $\psi_2 \leq \phi_2$. Since $D$ is directed, there is an element $\phi \in X$ such that $\phi_1,\phi_2 \leq \phi$, hence $\psi_1 \cdot \psi_2 \leq \phi_1 \cdot \phi_2 \leq \phi$, so that indeed $\psi_1 \cdot \psi_2$ are finite according to the order-theoretic sense. A complete lattice satisfying density is called \textit{algebraic}. So in a compact information algebra, $(\Phi,\leq)$ is an \textit{algebraic lattice}.

A few words on the finiteness of the null element are in order. Assume that the combination of finite elements may result in the null element, as for example the combination of two strings without a common prefix in the string algebra or the intersection of two convex polyhedra yielding the empty set. By the argument above, in this case the null element must be finite. Also, a directed set $D$ containing two incompatible elements must also contain $0$ and $\bigsqcup D = 0$ in this case. This concurs with the property of finite elements that if $\phi = \bigsqcup D$ is a finite element, then $\phi$ must belong to $D$. 

It turns out that elements in $\Phi$ with support $x$ are finite if and only if if they are finite in the subalgebra $\epsilon_x(\Phi)$.

\begin{proposition} \label{prop:ExtrFinite}
If $(\Phi,\Phi_f,\cdot,0,1;E)$ is a compact information algebra, then for all $x \in Q$ an element $\psi$ with support $x$ is finite in $\Phi$, $\psi \in \Phi_f$,  if and only if it is finite in $\epsilon_x(\Phi)$, that is $\psi \in (\epsilon_x(\Phi))_f$.
\end{proposition}

\begin{proof}
Consider first a finite element $\psi \in \Phi_f$ with support $x$ and a directed set $D$ in $\epsilon_x(\Phi)$ such that $\psi \leq \bigsqcup D$. We have then
\begin{eqnarray*}
\psi = \epsilon_x(\psi) \leq \bigsqcup D = \bigsqcup_{\phi \in D} \epsilon_x(\phi).
\end{eqnarray*}
Obviouly, the set $D$ is also directed in $\Phi$. Therefore, $\psi \leq \bigsqcup D$ implies that there is a $\phi = \epsilon_x(\phi) \in D$ such that $\psi \leq \phi$. By Theorem \ref{th:AlgLatt} this proves then that $\psi$ is also finite in the subalgebra $\epsilon_x(\Phi)$, that is $\psi \in (\epsilon_x(\Phi))_f$.

Conversely, assume $\psi = \epsilon_x(\psi)$ to be finite in $\epsilon_x(\Phi)$, that is $\psi \in (\epsilon_x(\Phi))_f$. Consider a directed set $D$ in $\Phi$ such that $\psi \leq \bigsqcup D$. Then by continuity of extraction, Theorem \ref{th:ContOfExtr}, it follows
\begin{eqnarray*}
\psi = \epsilon_x(\psi) \leq \epsilon_x(\bigsqcup D) = \bigsqcup_{\phi \in D} \epsilon_x(\phi).
\end{eqnarray*}
The set $\{\epsilon_x(\phi):\phi \in D\}$ is directed in $\epsilon_x(\Phi)$. So, since $\psi$ is finite in $\epsilon_x(\Phi)$, there is an element $\phi \in D$ such that $\psi \leq \epsilon_x(\phi) \leq \phi$ (Theorem \ref{th:AlgLatt}). But this implies also that $\psi$ is finite in $\Phi$, $\psi \in \Phi_f$.
\end{proof}

Compact information algebras may be obtained from any domain-free information algebra by ideal completion. Recall that $\Phi$ is embedded in its ideal completion $I_\Phi$ by the map $\phi \mapsto \downarrow\!\phi$ (see Section \ref{subsec:IdealExt}) so that $\Phi$ may be considered as a subalgebra of $I_\Phi$. In this sense, the elements of $\Phi$ or rather its images $\downarrow\!\phi$ are the finite elements of $I_\Phi$.

\begin{theorem} \label{th:AlgIdealCompl}
If $(\Phi,\cdot,0,1;E)$ with $E = \{\epsilon_x:x \in Q\}$ is a domain-free information algebra, then its ideal completion $(I_\Phi,\cdot,\{1\},\Phi;E)$ is a compact information algebra with $\Phi$ as its finite elements. 
\end{theorem}

\begin{proof}
The ideal completion $I_\Phi$ of an information algebra $\Phi$ is itself an information algebra, although one where the support axiom is not  necessarily valid. It remains to show that the principal ideals $\downarrow\!\phi$ for $\phi \in \Phi$ are its finite elements. We know that the combination of two principal ideals $\downarrow\!\phi$ and $\downarrow\!\psi$ is the principal ideal $\downarrow\!(\phi \cdot \psi)$.

To simplify notation we identify the image of $\Phi$ by the embedding $\phi \mapsto \downarrow\!\phi$ with $\Phi$. We have seen that $I_\Phi$ is a complete lattice under inclusion, that is under information order. In particular we have $\bigvee X = I(X)$ for any subset $X$ of $\Phi$. So Convergence holds.  

Further, we have for an ideal $I$ in $I_\Phi$
\begin{eqnarray*}
\epsilon_x(I) = \{\psi \in \Phi:\psi \leq \epsilon_x(\phi) \textrm{ for some}\ \phi \in I\}.
\end{eqnarray*}
We need to show that $\epsilon_x(I) = \bigvee X = I(X)$ for the set $X = \{\phi \in \Phi:\phi = \epsilon_x(\phi) \leq I\}$. Suppose first $\psi \in I(X)$ such that
\begin{eqnarray*}
\psi \leq \phi_1 \cdot \ldots \phi_n = \epsilon_x(\phi_1) \cdot \ldots \epsilon_x(\phi_n) \leq I.
\end{eqnarray*}
So we have $\psi \leq \epsilon_x(\phi)$ for some $\phi \in I$, hence $\psi \in \epsilon_x(I)$ and $I(X) \subseteq \epsilon_x(I)$. Conversely assume $\psi \in \epsilon_x(I)$, that is $\psi \leq \epsilon_x(\phi)$ for some $\phi \in I$. But then we have $\epsilon_x(\phi) \in I$. From this we conclude that $\psi \in I(X)$, since $\epsilon_x(\phi)$ has support $x$. This shows $\epsilon_x(I) = I(X) = \bigvee X$, hence local density.

To show Compactness, consider a directed subset $D$ of $\Phi$, and an element $\psi \in \Phi$ so that $\psi \leq \bigsqcup D$ in $I_\Phi$. Denote $\bigsqcup D = I(D)$ by $I$, $I = I(D)$. Then $\psi \in I$, hence
\begin{eqnarray*}
\psi \leq \phi_1 \cdot \ldots \cdot \phi_n \textrm{ for some}\ \phi_1,\ldots\phi_n \in D.
\end{eqnarray*}
Since $D$ is directed, there is some element $\phi \in D$ such that $\phi_1,\ldots,\phi_n \leq \phi$, thus $\psi \leq  \phi_1 \cdot \ldots \cdot \phi_n \leq \phi$. This is Compactness.
\end{proof}

Recall that in general the support axiom is not satisfied in $I_\Phi$, unless for example $(Q,\leq)$ has a top element. If the Support axiom does not hold in $I_\Phi$, global density does not necessarily follow from local density. 

Above we said that the finite elements of a compact information are fully determined by the information order. Conversely, the compact algebra is fully determined by its finite elements, as the following theorem shows.

\begin{theorem} \label{th:IdCompFiniteEl}
Let $(\Phi,\Phi_f,\cdot,0,1;E)$ be a compact information algebra with finite elements $\Phi_f$. Then the ideal completion $I_{\Phi_f}$ of the finite elements is a compact information algebra isomorphic to $\Phi$.
\end{theorem}

\begin{proof}
If $\Phi_f$ is closed under all extractions, then $\Phi_f$ is a subalgebra of $\Phi$, and then it follows form Theorem \ref{th:IdExt} that $I_{\Phi_f}$ is an information algebra (possibly not satisfying the support axiom). But even if $\Phi_f$ is not closed under extractions, its ideal extension is still an information algebra as we shall prove first. Note that $(\Phi_f,\leq)$ is partially ordered under information order restricted to $\Phi_f$.  So ideals in $\Phi_f$ are well defined. We define first combination among ideals of $\Phi_f$ as before by
\begin{eqnarray*}
I_1 \cdot I_2 = \{\phi \in \Phi_f:\exists \phi_1 \in I_1,\phi_2 \in I_2 \textrm{ such that}\ \phi \leq \phi_1 \cdot \phi_2\}.
\end{eqnarray*}
The ideals of $\Phi_f$ form still a $\cap$-system, hence a complete lattice with combination as join. Note that inclusion of ideals corresponds to information order. So $I_{\phi_f}$ is a commutative semigroup. with $\{1\}$ as unit and $\Phi_f$ as null element. 

Next, for any $x \in Q$ we define an extraction operator
\begin{eqnarray*}
\epsilon_x(I) = \{\phi \in \Phi_f:\exists \psi \in I \textrm{ such that}\ \phi \leq \epsilon_x(\psi)\}.
 \end{eqnarray*}
Clearly, $\epsilon_x(I)$ is still an ideal in $\Phi_f$. Now, we show that the operators $\epsilon_x$ for all $x \in Q$ are existential quantifiers. Obviously $\epsilon_x(\Phi_f) = \phi_f$ and if $I_1 \subseteq I_2$, then $\epsilon_x(I_1) \subseteq \epsilon_x(I_2)$. It remain to show that $\epsilon_x(\epsilon_x(I_1) \cdot I_2) = \epsilon_x(I_1) \cdot \epsilon_x(I_2)$. But this can be shown exactly as in the proof of Theorem \ref{th:IdExt}. This shows that the ideals of $\Phi_f$ form a domain-free information algebra, although possibly without satisfying the support axiom (but see remark below, after the proof). As in the previous Theorem \ref{th:AlgIdealCompl} it can be shown that this algebra is compact. 

Let $A_\phi = \{\psi \in \Phi_f:\psi \leq \phi\}$ for every $\phi \in \Phi$. This is an ideal in $\Phi_f$. We consider the map $\phi \mapsto A_\phi$, which maps $\Phi$ to $I_{\Phi_f}$. We show that this is an information algebra isomorphism. First, the map is onto $I_{\Phi_f}$: Consider any ideal of $\Phi_f$. Then the supremum of $I$ exists in $\Phi$, since the algebra $\Phi$ is compact. Let $\phi = \bigsqcup I$ and consider any element in $\psi \in \Phi_f$ such that $\psi \leq \phi$. Then, by compactness, there is an element $\chi \in I$ dominating $\psi$. This implies $\psi \in I$, hence $A_\phi \subseteq I$, and this shows that $I = A_\phi$, since by Density $\phi = \bigsqcup A_\phi \leq \bigsqcup I = \phi$. The map is also injective, since $A_\Phi = A_\psi$ implies, again by Density that $\phi = \psi$. Therefore, the map is bijective. 

We show further that it is a homomorphism. For two elements $\phi$ and $\psi$ from $\Phi$, clearly $A_{\phi \cdot \psi}$ contains $A_\phi$ and $A_\psi$, hence also $A_\phi \cdot A_\psi = I(A_\phi \cup A_\psi) \subseteq A_{\phi \cdot \psi}$. On the other hand, if $I$ is an ideal in $\Phi_f$ which contains $A_\phi$ and $A_\psi$, then, since the map is surjective, there is an element $\chi$ in $\Phi$ such that $I = A_\chi$, hence $\phi,\psi \leq \chi$ and $\phi \cdot \psi \leq \chi$. Therefore, if $\eta \in A_{\phi \cdot \psi}$, that is $\eta \leq \phi \cdot \psi \leq \chi$, we conclude that $\eta \in I$, hence $A_{\phi \cdot \psi} \subseteq I$. So we have $A_{\phi \cdot \psi} \subseteq A_\phi,A_\psi$. But this implies $A_{\phi \cdot \psi} = A_\phi \cdot A_\psi$, hence $A_{\phi \cdot \psi} = A_\phi \cdot A_\psi$..  Further, $A_1 = \{1\}$ and $A_0 = \Phi_f$. So, unit and null are preserved too. 

Finally, for any $x \in Q$, we have by definition
\begin{eqnarray*}
\epsilon_x(A_\phi) = \{\psi \in \Phi_f:\exists\chi \in A_\phi \textrm{ such that}\ \psi \leq \epsilon_x(\chi)\}.
\end{eqnarray*}
Since $\epsilon_x(\chi) \leq \epsilon_x(\phi)$, it follows that $\epsilon_x(A_\phi) \subseteq A_{\epsilon_x(\phi)}$. Consider then conversely an element $\psi \in A_{\epsilon_x(\phi)}$, that is  $\psi \leq \epsilon_x(\phi)$ and $\psi \in \Phi_f$. From $\phi = \bigsqcup A_\phi$ and from Theorem \ref{th:ContOfExtr} we have
\begin{eqnarray*}
\epsilon_x(\phi) = \bigsqcup A_{\epsilon_x(\phi)} = \epsilon_x(\bigsqcup A_\phi) = \bigsqcup_{\chi \in A_\phi} \epsilon_x(\chi).
\end{eqnarray*}
The set $\{\epsilon_x(\chi):\chi \in A_\phi\}$ is directed. By Compactness there is then an element $\eta \in A_\phi$ such that $\psi \leq \epsilon_x(\eta)$. But this means that $\psi \in \epsilon_x(A_\phi)$. So we conclude that $A_{\epsilon_x(\phi)} = \epsilon_x(A_\phi)$. The map $\phi \mapsto A_\phi$ is therefore a bijective information algebra homomorphism, hence the information algebras $\Phi$ and $I_{\Phi_f}$ are isomorphic, This concludes the proof.
\end{proof}

This is a representation theorem for compact information algebras, asserting that the algebra is fully determined by its finite elements. We remark that from the isomorphism between $\Phi$ and $I_{\Phi_f}$ it follows that the support axiom holds also in the ideal completion $I_{\Phi_f}$, since $\phi = \epsilon_x(\phi)$ implies $A_\phi = A_{\epsilon_x(\phi)} = \epsilon_x(A_\phi)$.

To conclude this section, we remark that if $\Phi$ and $\Psi$ are isomorphic information algebras and $\Phi$ is compact, then so is $\Psi$. More precisely, we have the following result.

\begin{proposition} \label{prop:IsomCompactAlg}
If $(\Phi,\Phi_f,\cdot,0,1;E_1)$ is a compact information algebra, $(\Psi,\cdot,0,1;E_2)$ an information algebra and $\Phi$ and $\Psi$ are isomorphic under the map $f : \Phi \rightarrow \Psi$, then $\Psi$ is compact too with finite elements $\Psi_f = f(\Phi_f)$.
\end{proposition}

\begin{proof}
We verify that $\Psi_f$ satisfies the defining properties of finite elements in $\Psi$, that is Combination, Convergence, Local Density and  Compactness.

Consider two elements $\psi_1,\psi_2 \in \Psi_f$. Then $\psi_1 = f(\phi_1)$ and $\psi_2 = f(\phi_2)$ and $\phi_1,\phi_2 \in \Phi_f$. It follows that $\psi_1 \cdot \psi_2 = f(\phi_1) \cdot f(\phi_2) = f(\phi_1 \cdot \phi_2) \in \Psi_f$, since $\phi_1 \cdot \phi_2 \in \Phi_f$. So \textit{Combination} is valid in $\Psi_f$.

Next let $D \subseteq \Psi_f$ be a directed subset of $\Psi_f$. Recall that the inverse map $f^{-1}$ is also an isomorphism (see Section \ref{sec:AlgNotions}). Consider the  subset $f^{-1}(D)$ of $\Phi_f$. It is directed in $\Phi_f$, since for $\phi_1,\phi_2 \in D$, we have $\phi_1 = f^{-1}(\psi_1)$ and  $\phi_2 = f^{-1}(\psi_2)$ with $\psi_1,\psi_2 \in D$. Then there is a $\psi \in D$ such that $\psi_1,\psi_2 \leq \psi$ and therefore $\phi_1,\phi_2 \leq f^{-1}(\psi) \in f^{-1}(D)$. Now, the supremum $\phi = \bigsqcup f^{-1}(D)$ exists in $\Phi$. But then $f(\phi)$ is the supremum of $D$ in $\Psi$. Obviously $f(\phi)$ is an upper bound of $D$ and if $\psi$ is another upper bound of $D$, then $f^{-1}(\psi)$ is an upper bound of $f^{-1}(D)$, hence $f^{-1}(\psi) \geq \phi$ and therefore $\psi \geq f(\phi)$. This proves \textit{Convergence} for $\Psi_f$.

Next let $\psi$ be any element with support $x$ in $\Psi$ and consider the set $\{\psi' \in \Psi_f: \epsilon^2_x(\psi') = \psi' \leq \psi\}$.
Apply the map $f^{-1}$ to this set to obtain the set $\{\phi' \in \Phi_f: \epsilon^1_x(\phi') = \phi' \leq \phi\}$, where $\phi' = f^{-1}(\psi')$ and $\phi = f^{-1}(\psi)$. This works since $\epsilon^1_x(f^{-1}(\psi')) = f^{-1}(\epsilon^2_x(\psi')) = f^{-1}(\psi') = \phi'$, support is preserved by $f$ and $f^{-1}$. So we have also $\phi = \epsilon^1_x(\phi)$ and therefore by Local Density in $\Phi$,
\begin{eqnarray*}
\epsilon^1_x(\phi) ) = \bigsqcup \{\phi' \in \Phi_f: \epsilon^1_x(\phi') = \phi' \leq \phi\}
\end{eqnarray*}
It follows by applyng the map $f$ that
\begin{eqnarray*}
\epsilon^2_x(\psi) ) = \bigsqcup \{\psi' \in \Psi_f: \epsilon^2_x(\psi') = \psi' \leq \psi\}
\end{eqnarray*}
and this shows that \textit{Local Density} is valid in $\Phi$.

Finally let $D \subseteq \Psi_f$ again be a directed set in $\Psi_f$, and $\psi \in \Psi_f$ such that $\psi \leq \bigsqcup D$. Then $f^{-1}(\psi) \leq \bigsqcup f^{-1}(D)$ and the set $f^{-1}(D)$ is directed in $\Phi_f$ as seen above. Then there is a $\phi' \in f^{-1}(D)$ such that $\phi' \leq f^{-1}(\psi)$ and therefore $f(\phi') \leq \psi$ and $f(\phi') \in D$. This is \textit{Compactness} in $\Psi$. 

So the set $\Psi_f$ represents indeed the finite elements in $\Psi$ and the information algebra $\Psi$ is compact. This concludes the proof.
\end{proof}


\section{Continuous information algebras}

The notion of approximation can be somewhat weakened. This leads to a generalisation of the concept of compact information algebras. The present section is partially based on \cite{guanli10}. The basic notion in this section is the way-below relation in an ordered set.

\begin{definition} \textbf{Way-Below.} 
Let $(\Phi;\leq)$ be a partially ordered set. For $\phi,\psi \in \Phi$ we write $\psi \ll \phi$ and say $\psi$ is way-below $\phi$, if for every directed set $D \subseteq \Phi$, for which the supremum exists, $\phi \leq \bigsqcup D$ implies that there is an element $\chi \in D$ such that $\psi \leq \chi$.
\end{definition}

Note that $\phi$ is a finite element if and only if $\phi \ll \phi$. The following lemma lists some well-known elementary results on the way-below relation, see for instance \cite{gierz03}.

\begin{lemma} \label{le:Way-Below}
Let $(\Phi;\leq)$ be a partially ordered set. Then the following holds for $\phi,\psi \in \Phi$
\begin{enumerate}
\item $\psi \ll \phi$ implies $\psi \leq \phi$,
\item $\psi \ll \phi$ and $\phi \leq \chi$ imply $\psi \ll \chi$,
\item $\chi \leq \psi$ and $\psi \ll \phi$ imply $\chi \ll \phi$.
\item $\chi \ll \psi$ and $\psi \ll \phi$ imply $\chi \ll \phi$.
\end{enumerate}
 \end{lemma}
 
We are of course interested in the way-below relation in case that $(\Phi,\cdot,0,1;E)$ with $E = \{\epsilon_x:x \in Q\}$ is a domain-free information algebra, that is, $(\Phi,\leq)$ is a semilattice under information order. Then the way-below relation has some additional properties.

\begin{lemma} \label{le:Way-Below-InfAlg}
 Let $(\Phi,\cdot,0,1;E)$ be a domain-free information algebra. Then
 \begin{enumerate}
 \item $1 \ll \phi$ for all $\phi \in \Phi$.
 \item $\psi_{1},\psi_{2} \ll \phi$ implies $\psi_{1} \vee \psi_{2} = \psi_1 \cdot \psi_2 \ll \phi$ for all $\psi_{1},\psi_{2} \in \Phi$.
 \item The set $\{\psi \in \Phi:\psi \ll \phi\}$ is an ideal for all $\phi \in \Phi$.
 \item $\psi \ll \phi$ if and only if for all $X \subseteq \Phi$ such that $\bigvee X$ exists and $\phi \leq \bigvee X$, there is a finite subset $F$ of $X$ such that $\psi \leq \bigvee F$. 
 \end{enumerate}
 \end{lemma}

\begin{proof}
(1) Let $D \subseteq \Phi$ be a directed set, and $\phi \leq \bigsqcup D$. Since $D$ is non-empty, there is a $\psi \in D$ and $1 \leq \psi$, hence $1 \ll \phi$.

(2) Assume $\psi_{1},\psi_{2} \ll \phi$. Consider any directed set $D \subseteq \Phi$ such that $\phi \leq \bigsqcup D$. Then there exist elements $\chi_{1},\chi_{2} \in D$ so that $\psi_{1} \leq \chi_{1}$ and $\psi_{2} \leq \chi_{2}$. Since $D$ is directed, there is also an element $\chi \in D$ so that $\chi_{1},\chi_{2} \leq \chi$. But then, $\psi_{1} \vee \psi_{2} \leq \chi_{1} \vee \chi_{2} \leq \chi$. This shows that $\psi_{1} \vee \psi_{2} \ll \phi$.

(3) Assume $\psi \ll \phi$ and $\chi \leq \psi$. Then by Lemma \ref{le:Way-Below} (3) $\chi \ll \phi$. Further let $\psi_{1} \ll \phi$ and $\psi_{2} \ll \phi$. By (2) just proved, $\psi_{1} \vee \psi_{2} \ll \phi$. Hence $\{\psi \in \Phi:\psi \ll \phi\}$ is an ideal.

(4) Suppose first that $\psi \ll \phi$. Let $X$ be a subset of $\Phi$ such that $\bigvee X$ exists and $\phi \leq \bigvee X$. Let $Y$ be the set of all joins of finite subsets of $X$. Then $X \subseteq Y$ and $\bigvee X$ is an upper bound for $Y$. Let $\chi$ be another upper bound of $Y$. Then $\chi$ is an upper bound of $X$, hence $\bigvee X \leq \chi$. So $\bigvee X$ is the supremum of $Y$, $\bigvee X = \bigvee Y$. Furthermore $Y$ is a directed set. So there is an element $\eta = \vee F \in Y$  for some finite subset $F$ of $X$, such that $\psi \leq \eta = \vee F$.

Conversely, consider elements $\psi,\phi \in \Phi$ such that condition 4 of the lemma holds. Let $D$ be a directed subset of $\Phi$ such that $\bigsqcup D$ exists and $\phi \leq \bigsqcup D$. There is then by assumption a finite subset $F$ of $D$ such that $\psi \leq \vee F$. Since $D$ is directed, there is a $\chi \in D$ such that $\vee F \leq \chi$, hence $\psi \leq \chi$. So $\psi \ll \phi$. 
\end{proof} 

With the aid of the way-below relation, algebraic information algebras can be alternatively characterized.

\begin{theorem} \label{th:AltCharCompInfAlg}
If $(\Phi,\cdot,0,1;E)$ is a domain-free information algebra, then the following conditions are equivalent:
\begin{enumerate}
\item $(\Phi,\Phi_f,\cdot,0,1;E)$ is a compact information algebra with finite elements $\Phi_{f}$.
\item $(\Phi;\leq)$ is an algebraic lattice with finite elements $\Phi_{f}$ and $\forall x \in D$, $\forall \phi \in \Phi$
\begin{eqnarray} \label{eq:Way-Below-Dens}
\epsilon_{x}(\phi) = \bigsqcup \{\psi \in \Phi_{f}:\psi = \epsilon_{x}(\psi) \ll \phi\}.
\end{eqnarray}
\end{enumerate}
\end{theorem}

\begin{proof}
(1) $\Rightarrow$ (2): By Theorem \ref{th:AlgLatt},  $(\Phi;\leq)$ is an algebraic lattice, that is a complete lattice with finite elements $\Phi_{f}$. Then condition (\ref{eq:Way-Below-Dens}) follows from Local density and Lemma \ref{le:Way-Below} in the following way,
\begin{eqnarray}
\epsilon_{x}(\phi) &=&\bigsqcup \{\psi \in \Phi_{f}:\psi = \epsilon_{x}(\psi) \leq \phi\}
\nonumber \\
&=&\bigsqcup \{\psi:\psi \ll \psi = \epsilon_{x}(\psi) \leq \phi\}
\nonumber \\
&=&\bigsqcup \{\psi:\psi \ll \psi = \epsilon_{x}(\psi) \ll \phi\}
\nonumber \\
&=&\bigsqcup \{\psi \in \Phi_{f}:\psi = \epsilon_{x}(\psi) \ll \phi\}.
\nonumber
\end{eqnarray}

(2) $\Rightarrow$ (1): We verify the definition of a compact information algebra in Section \ref{subsec:CompInfAlg}. We have seen that in an algebraic lattice, the finite elements are closed under join, hence Combination is valid. Convergence holds, since $(\Phi;\leq)$ is a complete lattice, Density follows from (\ref{eq:Way-Below-Dens}) since $\psi \ll \phi$ implies $\psi \leq \phi$ and Compactness follows from the lattice-theoretic finiteness.
\end{proof}

Another important property of finite elements in a compact information algebra is given by the following theorem:

\begin{theorem} \label{th:SepFinEl}
If $(\Phi,\Phi_f,\cdot,0,1;E)$ is a compact domain-free information algebra, then $\psi \ll \phi$ implies that here is an element $\chi \in \Phi_{f}$ so that $\psi \leq \chi \leq \phi$.
\end{theorem}

\begin{proof}
The set $A_{\phi} = \{\chi \in \Phi_{f}:\chi \leq \phi\}$ is directed and $\phi = \bigsqcup A_{\phi}$, hence $\phi \leq \bigsqcup A_{\phi}$. Then $\psi \ll \phi$ implies the existence of an element $\chi \in A_{\phi}$ so that $\psi \leq \chi$. But $\chi \leq \phi$. So $\psi \leq \chi \leq \phi$ and $\chi \in _{f}$.
\end{proof}

A set of elements having the property that $\psi \ll \phi$ implies the existence of a $\chi \in S$ such that $\psi \leq \chi \leq \phi$ is called \textit{separating}. So the set of finite elements in a compact information algebra is separating.

We now introduce continuous information algebras and show that they are a generaliszation of compact ones. We remark for the following that both the sets $\{\psi \in B:\psi \ll \phi\}$ and $\{\psi \in B:\psi = \epsilon_{x}(\psi) \ll \phi\}$ are directed. Note also that $\psi \ll \phi$ does not imply $\epsilon_{x}(\psi) \ll \epsilon_{x}(\phi)$.

\begin{definition} \label{def:ContInfAlg} \textbf{Continuous Information Algebras.}
A domain-free information algebra $(\Phi,\cdot,0,1;E)$ is called continuous with basis $B \subseteq \Phi$ if $B$ is closed under join (combination), contains the unit $1$ and the null element $0$, and $B$ satisfies the following conditions:
\begin{enumerate}
\item Convergence: If $D \subseteq B$ is directed, then $\bigsqcup D$ exists in $\Phi$.
\item Local $B$-Densitiy: For all  $\phi \in \Phi$ and for all $x \in Q$,
\begin{eqnarray}
\epsilon_{x}(\phi) = \bigsqcup \{\psi \in B:\psi = \epsilon_{x}(\psi) \ll \epsilon_{x}(\phi)\}.
\nonumber
\end{eqnarray}
\end{enumerate}
\end{definition}

Note that in a compact information algebra $(\Phi,\cdot,0,1;E)$ the finite elements $\Phi_{f}$ form a basis. So, an algebraic information algebra is also continuous with basis $\Phi_{f}$. We shall present below an example of a continuous information algebra which is not compact. So continuous information algebras present a genuine generalization of compact information algebras. The approximation by finite elements is replaced by an approximation of some more general elements in a basis $B$. 

Local $B$-density implies $B$-density if the Support axiom holds. In fact let $\phi \in \Phi$, then there is a $x \in Q$ so that $\phi = \epsilon_{x}(\phi)$. Then by the strong $B$-density:
\begin{eqnarray}
\phi &=& \epsilon_{x}(\phi) = \bigsqcup \{\psi \in B:\psi = \epsilon_{x}(\psi) \ll \phi\}
\nonumber \\
&\leq&\bigsqcup \{\psi \in B:\psi  \ll \phi\} \leq \phi.
\nonumber 
\end{eqnarray}
This is $B$-density. 

Just as in an compact information algebra $(\Phi,\Phi_f,\cdot,0,1;E)$, the partial order $(\Phi,\leq)$ determines an algebraic lattice, it follows that in a continuous information algebra $(\Phi,\cdot,0,1;E)$ the partial order $(\Phi,\leq)$ is a \textit{continuous lattice}, namely a complete lattice such that for all $\phi \in \Phi$
\begin{eqnarray} \label{eq:Way-BelowDens}
\phi = \bigsqcup\{\psi \in \Phi:\psi \ll \phi\}.
\end{eqnarray}

The following theorem states the situation more precisely.

\begin{theorem} \label{th:ContLatt}
If $(\Phi,\cdot,0,1;E)$ is a domain-free information algebra, then the following are equivalent:
\begin{enumerate}
\item $(\Phi,\cdot,0,1;E)$ is a continuous information algebra.
\item $(\Phi,\leq)$ is a continuous lattice, and $\forall x \in Q$, $\forall \phi \in \Phi$.
\begin{eqnarray} \label{eq:FullStrDens}
\epsilon_{x}(\phi) = \bigsqcup \{\psi \in \Phi:\psi = \epsilon_{x}(\psi) \ll \epsilon_{x}(\phi)\}.
\end{eqnarray}
\end{enumerate}
\end{theorem}

\begin{proof}
Assume first $(\Phi,\cdot,0,1;E)$ to be a continuous information algebra with basis $B$. We show first that $(\Phi;\leq)$ is a complete lattice. Consider a non-empty subset $X$ of $\Phi$. Define $Y$ to be the set of all elements in $B$, which are way-below all elements in $X$,
\begin{eqnarray}
Y = \{\psi \in B:\psi \ll \phi \textrm{ for all}\ \phi \in X\}.
\nonumber
\end{eqnarray}
Since $1\in Y$, the set is non-empty, and with $\psi_{1},\psi_{2} \in Y$ also $\psi_{1} \vee \psi_{2} \in Y$ (Lemma \ref{le:Way-Below-InfAlg}). So the subset $Y$ of $B$ is directed. Therefore $\bigsqcup Y$ exists and is a lower bound of $X$. Assume $\psi$ to be another lower bound of $X$. Then $A_{\psi} = \{\eta \in B:\eta \ll \psi\} \subseteq Y$, since $\eta \ll \psi \leq \phi$ implies $\eta \ll \phi$. From this we conclude that $\psi = \bigsqcup A_{\psi} \leq \bigsqcup Y$, hence $\bigsqcup Y$ is the infimum of $X$. Since $(\Phi;\leq)$ has a top element $0$ it follows from standard results of lattice theory, that $(\Phi;\leq)$ is a complete lattice. Further, using $B$-density, we obtain for all $\phi \in \Phi$,
\begin{eqnarray}
\phi = \bigsqcup \{\psi \in B:\psi \ll \phi\} \leq \bigsqcup \{\psi \in :\psi \ll \phi\} \leq \phi.
\nonumber
\end{eqnarray}
So $(\Phi;\leq)$ is indeed a continuous lattice. Further, again by Local density,
\begin{eqnarray}
\epsilon_{x}(\phi) &=& \bigsqcup \{\psi \in B:\psi = \epsilon_{x}(\psi) \ll \epsilon_{x}(\phi)\}
\nonumber \\
&\leq& \bigsqcup \{\psi \in \Phi:\psi = \epsilon_{x}(\psi) \ll \epsilon_{x}(\phi)\} \leq \epsilon_{x}(\phi),
\nonumber
\end{eqnarray}
so (\ref{eq:FullStrDens}) holds.

If $(\Phi;\leq)$, on the other hand, is a complete lattice, then convergence holds with $\Phi$ as a basis. And (\ref{eq:FullStrDens}) is exactly $B$-density with respect to the basis $\Phi$. Hence $(\Phi,\cdot,0,1;E)$ is a continuous information algebra.
\end{proof}

Here follows an example of a continuous information algebra.

\begin{example} \textbf{Continuous Valuation Algebra:}
This example is from \cite{guanli10}. Let $ = [0,1]$ be the real interval between $0$ and $1$ and $D = \{0,1\}$. Join is defined as maximum, the number $0$ is the unit and the number $1$ the null element of the algebra. Information extraction is defined as follows:
\begin{eqnarray}
\epsilon_{1}(\phi) &=& \phi,
\nonumber \\
\epsilon_{0}(\phi) &=& \left\{ \begin{array}{lll} \phi & \textrm{if} & \phi \in [0,1/2], \\ 1/2 & \textrm{if} & \phi \in (1/2,1].
\end{array} \right.
\nonumber
\end{eqnarray}
We leave it to reader to verify the axioms of an iinformation algebra.

Any non-empty subset $X$ of $[0,1]$ is in this example directed and $\sup X$ exists always. The relation $\psi \ll \phi$ holds if either $0 < \psi < \phi$ or in particular if $\psi = \phi = 0$. As a basis we take $B = [0,1]$. Then it can be verified that $\epsilon_{x}(\phi) = \bigvee \{\psi \in B:\psi = \epsilon_{x}(\psi) \ll \phi\}$ holds both for $x = 0$ and $x = 1$. So it is a continuous information algebra. But it is not compact: The only element satisfying $\phi \ll \phi$ is $\phi = 0$.
\end{example}

We have seen above that a compact information algebra is continuous. But the converse does not hold as the example above shows. Here follows a necessary and sufficient condition for a continuous information algebra to be compact.

\begin{theorem} \label{th:ContBeComp}
A continuous information algebra $(\Phi,\cdot,0,1;E)$ is compact, if and only if the set $\{\phi \in \Phi:\phi \ll \phi\}$ is a basis for $(\Phi,\cdot,0,1;E)$.
\end{theorem}

\begin{proof}
We know already that if $(\Phi,\Phi_f,\cdot,0,1;E)$ is compact, then it is continuous, with basis $B = \Phi_{f} = \{\phi \in \Phi:\phi \ll \phi\}$. 

So, assume that $(\Phi,\cdot,0,1;E)$ is continuous with basis $B = \{\phi \in \Phi:\phi \ll \phi\}$. The lattice $(\Phi;\leq)$ is complete, hence it is a dcpo. Local density is derived as follows:
\begin{eqnarray}
\epsilon_{x}(\phi) &=&\bigsqcup \{\psi \in B: \psi = \epsilon_{x}(\psi) \ll \epsilon_{x}(\phi)\} 
\nonumber \\
&=&\bigsqcup \{\psi \in B: \psi = \epsilon_{x}(\psi) \ll \psi \leq \epsilon_{x}(\phi)\} 
\nonumber \\
&=&\bigsqcup \{\psi \in B: \psi = \epsilon_{x}(\psi) \leq \epsilon_{x}(\phi)\}
\nonumber \\
\end{eqnarray}
So, the algebra is compact with the set $\{\phi \in \Phi:\phi \ll \phi\}$ as finite elements.
\end{proof}

The following Theorem gives another necessary and sufficient condition for an information algebra to be continuous.

\begin{theorem} \label{th:ComExtrJoin2}
An domain-free information algebra $(\Phi,\cdot,0,1;E)$ with $E = \{\epsilon_x:x \in Q\}$ is continuous if and only if, 
\begin{enumerate}
\item $(\Phi;\leq)$ is a continuous lattice,
\item for all $x \in Q$ and any directed set $D \subset \Phi$,
\begin{eqnarray} \label{eq:ComExtrJoin2}
\epsilon_{x}(\bigsqcup D) = \bigsqcup_{\phi \in D} \epsilon_{x}(\phi).
\end{eqnarray}
\end{enumerate}
\end{theorem}

\begin{proof}
Assume $(\Phi;\leq)$ to be a continuous lattice, so that density holds (\ref{eq:Way-BelowDens}), and that (\ref{eq:ComExtrJoin2}) holds too. Then $(\Phi;\leq)$ is a complete lattice. Consider a $\phi \in \Phi$. Then by density $\epsilon_{x}(\phi) = \bigsqcup \{\psi \in \Phi:\psi \ll \epsilon_{x}(\phi)\}$, and $\{\psi \in \Phi:\psi \ll \epsilon_{x}(\phi)\}$ is a directed set. From this we deduce, using (\ref{eq:ComExtrJoin2})
\begin{eqnarray}
\epsilon_{x}(\phi) &=& \epsilon_{x}(\epsilon_{x}(\phi)) = \epsilon_{x}(\bigsqcup\{\psi \in :\psi \ll \epsilon_{x}(\phi)\})
\nonumber \\
&=&\bigsqcup \{\epsilon_{x}(\psi):\psi \ll \epsilon_{x}(\phi)\}.
\nonumber
\end{eqnarray}
Let $\eta = \epsilon_{x}(\psi)$ so that $\eta = \epsilon_{x}(\eta) \leq \psi \ll \epsilon_{x}(\phi)$. From this it follows that $\eta \ll \epsilon_{x}(\phi)$ and therefore,
\begin{eqnarray}
\epsilon_{x}(\phi) &=&\bigsqcup \{\eta:\eta = \epsilon_{x}(\eta) = \epsilon_{x}(\psi), \psi \ll \epsilon_{x}(\phi)\}
\nonumber \\
&\leq&\bigsqcup \{\eta:\eta = \epsilon_{x}(\eta) \ll \epsilon_{x}(\phi)\} \leq \epsilon_{x}(\phi).
\nonumber
\end{eqnarray}
Hence we have $\epsilon_{x}(\phi) = \bigsqcup \{\eta:\eta = \epsilon_{x}(\eta) \ll \epsilon_{x}(\phi)\}$ and by Theorem \ref{th:ContLatt} $(\Phi,\cdot,0,1;E)$ is a continuous information algebra.

Conversely, assume $(\Phi,\cdot,0,1;E)$ to be a continuous information algebra with basis $B$. Then $(\Phi;\leq)$ is a continuous, hence complete lattice (Theorem \ref{th:ContLatt}). Consider a directed set $D \subseteq \Phi$ and $x \in Q$. For $\phi \in D$ we have $\phi \leq \bigsqcup D$, hence $\epsilon_{x}(\phi) \leq \epsilon_{x}(\bigsqcup D)$ and therefore $\bigsqcup_{\phi \in D} \epsilon_{x}(\phi) \leq \epsilon_{x}(\bigsqcup D)$. By local $B$-density,
\begin{eqnarray}
\epsilon_{x}(\bigsqcup D) = \bigsqcup \{\psi \in B:\psi = \epsilon_{x}(\psi) \ll \epsilon_{x}(\bigsqcup D)\}.
\nonumber
\end{eqnarray}
Now, $\psi = \epsilon_{x}(\psi) \ll \epsilon_{x}(\bigsqcup D) \leq \bigsqcup D$ implies that there is a $\phi \in D$ so that $\psi \leq \phi$ and thus also $\psi = \epsilon_{x}(\psi) \leq \epsilon_{x}(\phi)$. From this we conclude that $\epsilon_{x}(\bigsqcup D) \leq \bigsqcup_{\phi \in D} \epsilon_{x}(\phi)$ and thus $\epsilon_{x}(\bigsqcup D) = \bigsqcup_{\phi \in D} \epsilon_{x}(\phi)$. Hence (\ref{eq:ComExtrJoin2}) is valid.
\end{proof}

Similar to finite elements, for any elements $\psi$ and $\phi$ with support $x$, we have $\psi \ll \phi$ in the partial order $(\Phi;\leq)$ if and only if $\psi \ll \phi$ in $(\epsilon_x(\Phi),\leq)$ if $\Phi$ is a continuous information algebra.

\begin{proposition}
Let $(\Phi,\cdot,0,1;E)$ be a continuous information algebra. Then for all $x \in Q$ and elements $\psi,\phi$ with support $x$, $\psi \ll \phi$ in $(\Phi,\leq)$ if and only if $\psi \ll \phi$ in $(\epsilon_x(\Phi);\leq)$. 
\end{proposition}

\begin{proof}
Consider first elements $\psi,\phi \in \Phi$ with support $x$ such that $\psi \ll \phi$ in $(\Phi;\leq)$. Let $D$ be a directed set in $(\epsilon_x(\Phi;\leq)$ such that $\psi = \epsilon_x(\psi) \leq \bigsqcup D$. Note that $D$ is also directed in $(\Phi,\leq)$ and therefore $\psi \ll \phi$ implies that there is a $\chi = \epsilon_x(\chi) \in D$ such that $\phi \leq \chi$. But this shows that $\psi \ll \phi$ in $(\epsilon_x(\Phi);\leq)$.

Conversely consider elements $\psi,\phi \in \Phi$ with support $x$ such that $\psi \ll \phi$ in $(\epsilon_x(\Phi);\leq)$. This time let $D$ be a directed set in $(\Phi,\leq)$ and such that $\psi \leq \bigsqcup D$. By Theorem \ref{th:ComExtrJoin2} we have then
\begin{eqnarray*}
\psi = \epsilon_x(\psi) \leq \epsilon_x(\bigsqcup D) = \bigsqcup_{\chi \in D} \epsilon_x(\chi).
\end{eqnarray*}
The set $\{\epsilon_x(\chi):\chi \in D\}$ is directed in $(\epsilon_x(\Phi);\leq)$ and therefore $\psi \ll \phi$ in $(\epsilon_x(\Phi);\leq)$ implies that there is an element $\chi \in D$ such that $\phi = \epsilon_x(\phi) \leq \epsilon_x(\chi) \leq \chi$ and this shows that $\psi \ll \phi$ also in $(\Phi,\leq)$.
\end{proof}

In the following section, we consider maps between information algebras and show that they form themselves information algebras. Further in Section \ref{subsec:CartClos}, we look at compact and continuous information algebras from a categorical point of view


\section{Algebra of mappings} \label{subsec:Mappings}

There are many ways to construct new information algebras from old ones. For instance, maps from any set into a generalised information algebra algebra form again an information under point-wise combination and extraction, see Section \ref{subsec:RanMaps} for more on this subject. In this section however, we consider order-preserving maps, between domain-free information algebras and show that these structures form themselves information algebras, This will be the base to show in the following section that information algebras form Cartesian closed categories. 

Consider two domain-free information algebras $(\Phi_1,\cdot,0,1;E_1)$ and $(\Phi_2,\cdot,0,1;E_2)$ with $E_1 = \{\epsilon^1_x:x \in Q_1\}$ and $E_2 =\{ \epsilon^2_x:x \in Q_2\}$.  A map $f : \Phi_1 \rightarrow \Phi_2$ is order-preserving, if $\phi \leq \psi $ in $\Phi_1$ implies $f(\phi) \leq f(\psi)$ in $\Phi_2$, a more informative piece of information is mapped to a more informative piece of information. For the maps to be considered, we may for semantic reasons require  a little bit more: For innstance he null element in $\Phi_1$ and only the null element should map to the null element in $\Phi_2$, the map $f$ can neither eliminate nor create contradiction. Or vacuous information should map to vacuous information. This leads us to the following definition:

\begin{definition} \label{Information Map:}
If $(\Phi_1,\cdot,0,1;E_1)$ and $(\Phi_2,\cdot,0,1;E_2)$ with $E_1 = \{\epsilon^1_x:x \in Q_1\}$ and $E_2 =\{ \epsilon^2_x:x \in Q_2\}$ are two (not necessarily distinct) information algebras, then an order-preserving map $f : \Psi_1 \rightarrow \Psi_2$ is called an information map. If $f(\phi) = 0$ if and only if $\phi = 0$, it is called a proper information map. If furthermore $f(1) = 1$, the information map is called strict.
\end{definition}

In this definition, as well as in the sequel it should be clear that the symbols $0$ and $1$ denote unit and null elements both in $\Phi_1$ and $\Phi_2$ according to the context, we do not differentiate between them by notation. The same holds for combination and relational symbols (like information order), it will always be clear from the context, whether the operation or relation is in $\Phi_1$ or $\Phi_2$. In the sequel we essentially consider general information maps and do not specially consider proper or strict maps.

Denote the set of all information maps between $\Phi_1$ and $\Phi_2$ by $[\Phi_1 \rightarrow \Phi_2]$. We define the  following operations for information maps $f,g \in [\Phi_1 \rightarrow \Phi_2]$ and extraction operators $\epsilon^1_x \in E_1$ and $\epsilon^2_y \in E_2$:
\begin{enumerate}
\item \textit{Combination:} $f \cdot g$ defined by $(f \cdot g)(\phi) = f(\phi) \cdot g(\phi)$ for all $\phi \in \Phi_1$,
\item \textit{Extraction:} $(\epsilon^1_x,\epsilon^2_y)(f)$ defined by $(\epsilon^1_x,\epsilon^2_y)(f)(\phi) = \epsilon^2_y(f(\epsilon^1_x(\phi))$ for all $\phi \in \Phi_1$.
\end{enumerate}
It is obvious that $f \cdot g$ and $(\epsilon^1_x,\epsilon^2_y)(f)$ belong to $[\Phi_1 \rightarrow \Phi_2]$, so $[\Phi_1 \rightarrow \Phi_2]$ is closed both under combination as well as extraction. Note that the map $0$ defined by $0(\phi) = 0$  for all $\phi \in \Phi$ and the map $1$ defined by $1(\phi) = 1$ for all $\phi \not= 0$ in $\Phi$, $1(0) =0$, are the null and unit elements of combination.

We show that these operations define a domain-free information algebra of information maps.

\begin{theorem}
If $(\Phi_1,\cdot,0,1;E_1)$ and $(\Phi_2,\cdot,0,1;E_2)$ are two domain-free information algebras, then $([\Phi_1 \rightarrow \Phi_2],\cdot,0,1;E_1 \times E_2)$ is a domain-free information algebra, albeit not satisfying necessarily the support axiom..
\end{theorem}

\begin{proof}
Obviously, the combination operation between information maps is associative and commutative, and has unit map $1$ and as null map $0$, so $([\Phi_1 \rightarrow \Phi_2],\ cdot,0,1)$ is a commutative semigroup with unit and null element. 

We show that the operators $(\epsilon^1_x,\epsilon^2_y)$ are existential quantifiers, (see Section \ref{sec:InfAlg}). First, $(\epsilon^1_x,\epsilon^2_y)(0)(\phi) = \epsilon^2_y(0(\epsilon^1_x(\phi))) = \epsilon^2_y(0) = 0$. So any extraction of the null map yields the null map. Secondly, 
\begin{eqnarray*}
((\epsilon^1_x,\epsilon^2_y)(f) \cdot f)(\phi) = (\epsilon^1_x,\epsilon^2_y)(f)(\phi) \cdot f(\phi) = \epsilon^2_y(f(\epsilon^1_x(\phi)) \cdot f(\phi) = f(\phi),
\end{eqnarray*}
since $f(\epsilon^1(\phi)) \leq f(\phi)$. hence $\epsilon^2_y(f(\epsilon^1_x(\phi)) \leq \epsilon^2_y(f(\phi)) \leq f(\phi)$. So we obtain $(\epsilon^1_x,\epsilon^2_y)(f) \cdot f = f$. Finally,
\begin{eqnarray*}
\lefteqn{ (\epsilon^{1}_x,\epsilon^{2}_y)((\epsilon^1_x,\epsilon^2_y)(f) \cdot g))(\phi) = \epsilon^2_y((((\epsilon^1_x,\epsilon^2_y)(f) \cdot g))(\epsilon^1_x(\phi))  } \\
&=& \epsilon^2_y((\epsilon^1_x,\epsilon^2_y)(f(\epsilon^1_x(\phi))) \cdot g(\epsilon^1_x(\phi)) \\
&=& \epsilon^2_y((\epsilon^2_y(f(\epsilon^1_x(\epsilon^1_x(\phi)))) \cdot g(\epsilon^1_x(\phi))) \\
&=& \epsilon^2_y(\epsilon^2_y(f(\epsilon^1_x(\phi))) \cdot g(\epsilon^1_x(\phi))) \\
&=& \epsilon^2_y(f(\epsilon^1_x(\phi)) \cdot \epsilon^2_y(g(\epsilon^1_x)) \\
&=& (\epsilon^1_x,\epsilon^2_y)(f)(\phi) \cdot (\epsilon^1_x,\epsilon^2_y)(g)(\phi) \\
&=& ((\epsilon^1_x,\epsilon^2_y)(f) \cdot (\epsilon^1_x,\epsilon^2_y)(g))(\phi). 
\end{eqnarray*}
So we have $(\epsilon^1_x,\epsilon^2_y)((\epsilon^1_x,\epsilon^2_y)(f) \cdot g)) = (\epsilon^1_x,\epsilon^2_y)(f) \cdot (\epsilon^1_x,\epsilon^2_y)(g)$. This tells us that $(\epsilon^1_x,\epsilon^2_y)$ is an existential quantifier relative to information maps. This concludes the proof.
\end{proof}

Considering proper and strict information maps, we see that both sets are closed under combination and extraction. But the null map is neither proper nor strict and the unit map is not proper. So proper and strict maps do not form a subalgebra of the information algebra $[\Phi_1 \rightarrow \Phi_2]$ of information maps.

Note that we may as usual derive an order between questions in $Q_1 \times Q_2$, and also a conditional independence relation. In fact, it is easy to see that $(x',y') \leq (x,y)$ if and only if $x' \leq x$ and $y' \leq y$, since $(\epsilon^1_{x'},\epsilon^2_{y'})(\epsilon^1_x,\epsilon^2_y) = (\epsilon^1_{x'}\epsilon^1_x,\epsilon^2_{y'}\epsilon_y')$. Similarly, we have $(x,y) \bot (x',y') \vert (x'',y'')$ if and only if $x \bot x' \vert x''$ and $y \bot y' \vert y''$. Information order in $[\Phi_1 \rightarrow \Phi_2]$ is as usual defined by $g \leq f$ if $g \cdot f = f$ and this holds clearly if and only if $g(\phi) \leq f(\phi)$ for all $\phi \in \Phi_1$.

The Support axiom is in the algebra $[\Phi_1 \rightarrow \Phi_2]$ in general not satisfied, even if it is so in $\Phi_1$ and $\Phi_2$. We recall that this axiom is important especially for the derivation of labeled algebras from domain-free ones, that is, for duality. In this case in addition, we must require that the order among questions defines a lattice. All this is in the present framework of less interest, so we do not require the Support axiom for the algebra $[\Phi_1 \rightarrow \Phi_2]$. 

It is also obvious that if both the information algebras $\Phi_1$ and $\Phi_2$ are commutative, so is the algebra $[\Phi_1 \rightarrow \Phi_2]$.

Next we consider continuous (and compact) information algebras. In this context, we need the concept of \textit{continuous maps}. 

\begin{definition}
If $(\Phi_1,\cdot,0,1;E_1)$ and $(\Phi_2,\cdot,0,1;E_2)$ are two continuous domain-free information algebras with bases $B_1$ and $B_2$ respectively, then a map $f : \Phi_1 \rightarrow \Phi_2$ is called continuous, if for all $\phi \in \Phi_1$,
\begin{eqnarray*}
f(\phi) = \bigsqcup \{f(\psi):\psi \in B_1,\psi \ll \phi\}.
\end{eqnarray*}
\end{definition}

Continuous maps are order preserving, that is information maps. Let $[\Phi_1 \rightarrow \Phi_2]_c$ denote the set of continuous information maps between $\Phi_1$ and $\Phi_2$.  Continuity of maps is a purely order-theoretic concept and there are several equivalent definitions \cite{daveypriestley97}. In particular, continuous maps are maps which preserve limits, as the following lemma shows.

\begin{lemma} \label{le:EquiContMap}
The following are equivalent:
\begin{enumerate}
\item $f(\phi) = \bigsqcup \{f(\psi):\psi \in B_1,\psi \ll \phi\}$ for all $\phi \in \Phi_1$,
\item $\{\psi \in B_2:\psi \ll f(\phi)\} \subseteq \{\psi \in \Psi_2:\psi \leq f(\chi), \chi \ll \phi \textrm{ for some}\ \chi \in B_1\}$ for all $\phi \in \Phi_1$,
\item if $D \subseteq \Phi_1$ is directed, then 
\begin{eqnarray*}
f(\bigsqcup D) = \bigsqcup_{\phi \in D} f(\phi).
\end{eqnarray*}
\end{enumerate}
\end{lemma}

\begin{proof}
$(1) \Rightarrow (2):$ Consider an element $\psi \in B_2$ such that $\psi \ll f(\phi)$. Then we have by (1)
\begin{eqnarray*}
\psi \ll f(\phi) = \bigsqcup \{f(\chi): \chi \in B_1,\chi \ll \phi\}.
\end{eqnarray*}
The set $ \{f(\chi): \chi \in B_1,\chi \ll \phi|\}$ is directed in $\Phi_2$. Therefore, there is an element $\chi \in B_1$ with $\chi \ll \phi$, and such that $\psi \leq f(\chi)$. So (2) holds.

$(2) \Rightarrow (3):$ Consider a directed subset $D$ of $\Phi_1$ and define $\phi = \bigsqcup D$. If $\psi \in B_2$ such that $\psi \ll f(\phi)$, then there exists by (2) an element $\chi \in B_1$ such that $\chi \ll \phi$ and $\psi \leq f(\chi)$. There is then further an element $\eta \in D$ such that $\chi \leq \eta$. Hence we conclude that $\psi \leq f(\chi) \leq f(\eta) \leq \bigsqcup f(D)$. So, by continuity in $\Phi_2$, we have
\begin{eqnarray*}
f(\bigsqcup D) = \bigsqcup \{\psi \in B_2:\psi \ll f(\bigsqcup D)\} \leq \bigsqcup f(D).
\end{eqnarray*}
Obviously, $f(\bigsqcup D) \geq \bigsqcup f(D)$, so that $f(\bigsqcup D) = \bigsqcup f(D)$, hence (3) holds.

$(3) \Rightarrow (1):$ By continuity in $\Phi_1$, we have $\phi = \bigsqcup \{\psi \in B_1:\psi \ll \phi\}$, the set $\{\psi \in B_1:\psi \ll \phi\}$ is directed, and therefore, (1) follows from (3),
\end{proof}

As a corollary, it follows from Theorem \ref{th:ComExtrJoin2} that the extraction operators $\epsilon_x \in E$ of a continuous information algebra $(\Phi,\cdot,0.1;E)$ are continuous maps, hence belongs to $[\Phi \rightarrow \Phi]_c$. We proceed to show that combination and extraction operators of continuous maps produce continuous maps. This implies then, that $([\Phi_1 \rightarrow \Phi_2]_c$ is again an information algebra, a subalgebra of $([\Phi_1 \rightarrow \Phi_2]$. In fact, we shall prove further that it is a continuous information algebra.

\begin{theorem} \label{th:ContOfCombExtr}
If $(\Phi_1,\cdot,0,1;E_1)$ and $(\Phi_2,\cdot,0,1;E_2)$ are two continuous domain-free informartion algebras, $f,g \in [\Phi_1 \rightarrow \Phi_2]_c$,  $(\epsilon^1_x,\epsilon^2_y) \in E_1 \times E_2$, then $f \cdot g, (\epsilon^1_x,\epsilon^2_y)(f) \in  [\Phi_1 \rightarrow \Phi_2]_c$.
\end{theorem}

\begin{proof}
The proof is straightforward using item 3 of Lemma \ref{le:EquiContMap} and continuity of extractor operators in $\Phi_1$ and $\Phi_2$. So, let $D$ be a directed subset of $\Phi_1$, then 
\begin{eqnarray*}
\lefteqn{(f \cdot g)(\bigsqcup D)} \\
&&= f(\bigsqcup D) \vee g(\bigsqcup D) = (\bigsqcup_{\phi \in D} f(\phi)) \vee (\bigsqcup_{\phi \in D} g(\phi)) \\
&&= \bigsqcup_{\phi \in D} (f(\phi) \vee g(\phi)) = \bigsqcup_{\phi \in D} (f \cdot g)(\phi).
\end{eqnarray*}
This shows that $f \cdot g$ is continuous.

In a similar way, since both $\epsilon^1_x(D)$ and $f(\epsilon^1_y(D))$ are directed sets, 
\begin{eqnarray*}
\lefteqn{(\epsilon^1_x,\epsilon^2_y)(f)(\bigsqcup D)} \\
&&= \epsilon^2_y(f(\epsilon^1_x(\bigsqcup D)) = \epsilon^2_y(f(\bigsqcup_{\phi \in D} \epsilon^1_x(\phi)) \\
&&= \bigsqcup_{\phi \in D} \epsilon^2_y(f(\epsilon^1_x(\phi)) = \bigsqcup_{\phi \in D} (\epsilon^1_x,\epsilon^2_y)(f)(\phi).
\end{eqnarray*}
This shows that $((\epsilon^1_x,\epsilon^2_y)(f)$ is a continuous map.
\end{proof}

We remark that information order in $[\Phi \rightarrow \Psi]_c$, as in $[\Phi \rightarrow \Psi]$ is pointwise. It is well-known from order theory that $([\Phi_1 \rightarrow \Phi_2]_c;\leq)$ is a continuous lattice. Then we can use Theorem \ref{th:ComExtrJoin2} to show thaf $([\Phi_1 \rightarrow \Phi_2]_c,\cdot,0,1;E_1 \times E_2)$ is a continuous information algebra. This has been shown in \cite{Guan14}. We want here to be a bit more explicit, based on \cite{scott71}. 

\begin{proposition} \label{prop:ContMapsComplLatt}
If $\Phi_1$ and $\Phi_2$ are continuous information algebras, then $([\Phi_1 \rightarrow \Phi_2]_c,\leq)$ is a complete lattice under information order.
\end{proposition}

\begin{proof}
By Theorem \ref{th:ContOfCombExtr} combination, that is join in information order, of continuous maps yields a continuous map. Hence $([\Phi_1 \rightarrow \Phi_2]_c,\leq)$ is closed under join. Also the unit function $1$ belongs to $[\Phi_1 \rightarrow \Phi_2]_c$ as the least element. Let $G$ be a directed subset of $[\Phi_1 \rightarrow \Phi_2]_c$. Then $\{g(\phi):g \in G\}$ is a directed set in $\Phi_2$ for everyl $\phi \in \Phi_1$. Define
\begin{eqnarray*}
f(\phi) = \bigsqcup_{g \in G} g(\phi).
\end{eqnarray*}
The supremum on the right hand side exists, since $\Phi_2$ is a continuous lattice. Let $D$ be a directed set in $\Phi_1$. Then we have, since all $g \in G$ are continuous maps,
\begin{eqnarray*}
f(\bigsqcup D) = \bigsqcup_{g \in G} g(\bigsqcup D) = \bigsqcup_{g \in G} \bigsqcup_{\phi \in D} g(\phi)
= \bigsqcup_{\phi \in D}  \bigsqcup_{g \in G} g(\phi) = \bigsqcup_{\phi \in D} f(\phi).
\end{eqnarray*}
This shows that $f$ is a continuous map. The map $f$ is an upper bound of $G$ and it must be the supremum of $G$, $f = \bigsqcup G$, since for any other upper bound $h$ of $G$ in $[\Phi_1 \rightarrow \Phi_2]_c$ we have $h \geq f$.

So, $[\Phi_1 \rightarrow \Phi_2]_c$ contains the supremum of every directed subset and is bounded. By standard methods of lattice theory, it follows that it must be a complete lattice, see for instance \cite{daveypriestley97}.
\end{proof}

The proof of the continuity of $([\Phi_1 \rightarrow \Phi_2]_c,\leq)$ in \cite{scott71,gierz03} uses topological arguments, which are not easily translatable into purely lattice-based arguments. Therefore we renounce to give the proof here.

It remains to show that $([\Phi_1 \rightarrow \Phi_2]_c,\leq)$ is a continuous information algebra. This follows from  the next Theorem and Theorem \ref{th:ComExtrJoin2}.

\begin{theorem} \label{th:ContInfAlg}
For all $x \in Q_1$, $y \in Q_2$ and all any directed set $G \subseteq [\Phi_1 \rightarrow \Phi_2]_c$,
\begin{eqnarray*}
(\epsilon^1_x,\epsilon^2_y)(\bigsqcup G) = \bigsqcup_{g \in G} (\epsilon^1_x,\epsilon^2_y)(g).
\end{eqnarray*}
\end{theorem}

\begin{proof}
Since $([\Phi_1 \rightarrow \Phi_2]_c,\leq)$ is a continuous lattice, $\bigsqcup G$ is a continuous map. Since the order in $[\Phi_1 \rightarrow \Phi_2]_c$ is pointwise, we have $(\bigvee_{i \in I} f_i)(\phi) = \bigvee_{i \in I} f_i(\phi)$ for any family of continuous map $f_i$, $i \in I$ and any element $\phi \in \Phi_1$. Since $\epsilon^1_x$ and $\epsilon^2_y$ are continuous maps we obtain therefore
\begin{eqnarray*}
(\epsilon^1_x,\epsilon^2_y)(\bigsqcup G)(\phi) &=& \epsilon^2_y(\left(\bigsqcup G\right)(\epsilon^1_x(\phi))) 
= \epsilon^2_y( \bigsqcup_{g \in G} g(\epsilon^1_x(\phi))) \\
&=& \bigsqcup_{g \in G} \epsilon^2_y(g(\epsilon^1_x(\phi))) 
= \bigsqcup_{g \in G} (\epsilon^1_x,\epsilon^2_y)(g(\phi)) \\
&=& \left(\bigsqcup_{g \in G} (\epsilon^1_x,\epsilon^2_y)(g)\right)(\phi).
\end{eqnarray*}
This shows that $(\epsilon^1_x,\epsilon^2_y)(\bigsqcup G) = \bigsqcup_{g \in G} (\epsilon^1_x,\epsilon^2_y)(g).$.
\end{proof} 

So we see that indeed $([\Phi_1 \rightarrow \Phi_2]_c,\cdot,0,1,;E_1 \times E_2)$ is a continuous information algebra.

In case that the information algebras $\Phi_1$ and $\Phi_2$ are compact, we may conjecture that the information algebra $[\Phi_1 \rightarrow \Phi_2]$ is compact too. To show this, we first identify the finite elements, following \cite{kohlas03}. Let $Y$ be a finite subset of $\Phi_{1,f}$. A mapping $s: Y \rightarrow \Phi_{2,f}$, where $\Phi_{2,f}$ is the set of finite elements of $\Phi_2$, is called a \textit{simple} map. Let $S$ be the set of simple maps. For any $s \in S$ let $Y(s)$ be the domain of $s$. A simple function will be extended to the whole of $\Phi_1$ by defining
\begin{eqnarray*}
\hat{s}(\phi) = \vee \{s(\psi):\psi \in Y(s),\psi \leq \phi\}.
\end{eqnarray*}
We set here $\hat{s}(\phi) = 1$, if the set on the righthand side is empty. Note that $\hat{s}$ is a map from $\Phi$ into $\Phi_{2,f}$. Let $\hat{S}$ be the set all such maps, $\hat{S} = \{\hat{s}:s \in S\}$. Note that the unit map and the null map belong to $\hat{S}$. 

Obviously the maps $\hat{s}$ preserve order. In fact we show that they are continuos.

\begin{proposition} \label{prop:SimpleCont}
Any map $\hat{s} \in \hat{S}$ is continuous.
\end{proposition} 

\begin{proof}
Let $D \subseteq \Phi_1$ be a directed set. Since $\hat{s}$ preserves order we have $\hat{s}(\phi) \leq \hat{s}(\bigsqcup D)$ for all $\phi \in D$, hence $\bigsqcup_{\phi \in D} \hat{s}(\phi) \leq \hat{s}(\bigsqcup D)$. We claim that the inverse inequality holds too.

In fact, consider an element $\psi \in Y(s)$ such that $\psi \leq \bigsqcup D$. Recall that $\psi$ is finite, therefore by compactness there is an element $\phi \in D$ such that $\psi \leq \phi$. It follows that $s(\psi) \leq \hat{s}(\psi) \leq \hat{s}(\phi)$ and so
\begin{eqnarray*}
\hat{s}(\bigsqcup D) = \vee \{s(\psi):\psi \in Y(s),\psi \leq \bigsqcup D\} \leq \bigsqcup_{\phi \in D} \hat{s}(\phi).
\end{eqnarray*}
This shows that $\hat{s}(\bigsqcup D) = \bigsqcup_{\phi \in D} \hat{s}(\phi)$ and thus $\hat{s}$ is continuous, see Lemma \ref{le:EquiContMap}.
\end{proof}

So, $\hat{S}$ is a subset of $[\Phi_1 \rightarrow \Phi_2]_c$. We show now that this set represents the finite elements of the information algebra $[\Phi_1 \rightarrow \Phi_2]_c$ by verifying the conditions of Combination, Convergence, Local Density and Compactness, see Section \ref{subsec:CompInfAlg}. First of all, we claim that the simple function $s$ defined by $Y(s) = Y(s_1) \cup Y(s_2)$ and $s(\psi) = \hat{s}_1(\psi) \vee \hat{s}_2(\psi)$ defines the combination of $\hat{s}_1$ and $\hat{s}_2$, that is $\hat{s} = \hat{s}_1 \cdot \hat{s}_2$. By a simple computation using transitivity of join we obtain for a $\phi \in \Phi_1$,
\begin{eqnarray*}
\hat{s}(\phi) &=& \vee \{\hat{s}_1(\psi) \vee \hat{s}_2(\psi):\psi \in Y(s_1) \cup Y(s_2),\psi \leq \phi\} \\
&=& \vee\{(\vee \{s_1(\psi_1):\psi_1 \in Y(s_1),\psi_1 \leq \psi\}) \vee \\
&& (\vee \{s_2(\psi_2):\psi_2 \in Y(s_1),\psi_2 \leq \psi\}):\psi \in Y(s_1) \cup Y(s_2),\psi \leq \phi\} \\
&=& (\vee \{s_1(\psi_1):\psi_1 \in Y(s_1),\psi_1 \leq \psi,\psi \in Y(s_1) \cup Y(s_2),\psi \leq \phi\}) \vee \\
&& (\vee \{s_2(\psi_2):\psi_2 \in Y(s_2),\psi_2 \leq \psi,\psi \in Y(s_1) \cup Y(s_2),\psi \leq \phi\}) \\
&=& (\vee \{s_1(\psi_1):\psi_1 \in Y(s_1),\psi_1 \leq \phi\}) \vee (\vee \{s_2(\psi_2):\psi_2 \in Y(s_2),\psi_2 \leq \phi\}) \\
&=& \hat{s}_1 \vee \hat{s}_2.
\end{eqnarray*}
Thus, $\hat{S}$ is closed under joins in information order, which corresponds to combination. The unit and null maps are obviously simple. So, the condition of \textit{Combination} is satisfied.

\textit{Convergence} follows since $([\Phi_1 \rightarrow \Phi_2]_c,\leq)$ is a complete lattice, see also the first part of the proof of Theorem \ref{prop:ContMapsComplLatt}.

For \textit{Density}, we have to show that
\begin{eqnarray*}
(\epsilon^1_x,\epsilon^2_y)(f) = \bigsqcup \{\hat{s}:\hat{s} \leq f\}.
\end{eqnarray*}
By the definition of extraction in the information algebra $[\Phi_1 \rightarrow \Phi_2]_c$, the assumption that $\Phi_1$ is a compact information algebra and the continuity of $f$, 
\begin{eqnarray*}
(\epsilon^1_x,\epsilon^2_y)(f)(\phi) &=& \epsilon^2_y(f(\epsilon^1_x(\phi)) 
= \epsilon^2_y(f(\bigsqcup \{\psi:\psi \in \Phi_{1,f},\psi = \epsilon^1_x(\psi) \leq \epsilon^1_x(\phi))\} ) \\
&=& \epsilon^2_y(\bigsqcup \{f(\psi):\psi \in \Phi_{1,f},\psi = \epsilon^1_x(\psi) \leq \phi)\} )
\end{eqnarray*}
for any $\phi = \epsilon^1_x(\phi) \in \Phi_1$. The set $\{\psi \in \Phi_{1,f},\psi = \epsilon^1_x(\psi) \leq \phi)\}$ is directed, so by the continuity of the extraction operator $\epsilon^2_y$, see Theorem \ref{th:ContOfExtr},
\begin{eqnarray*}
(\epsilon^1_x,\epsilon^2_y)(f)(\phi) = \bigsqcup \{\epsilon^2_y(f(\psi)):\psi \in \Phi_{1,f},\psi = \epsilon^1_x(\psi) \leq \phi)\}.
\end{eqnarray*}
We claim that
\begin{eqnarray} \label{eqKeyRel}
\epsilon^2_y(f(\psi)) = \bigsqcup \{\hat{s}(\psi):s \in S,\hat{s} = (\epsilon^1_x,\epsilon^2_y)(\hat{s}) \leq f\}
\end{eqnarray}
for any $\psi \in \Phi_1$. If this holds, then by the continuity of $\hat{s}$ we obtain for an element $\phi = \epsilon^1_x(\phi) \in \Phi_1$,
\begin{eqnarray*}
(\epsilon^1_x,\epsilon^2_y)(f)(\phi) &=& \bigsqcup \{\bigsqcup \{ \hat{s}(\psi):s \in S,\hat{s} = (\epsilon^1_x,\epsilon^2_y)(\hat{s}) \leq f\} \\
&& : \psi \in \Phi_{1,f},\psi = \epsilon^1_x(\psi) \leq \phi)\} \\
&=& \bigsqcup \{\bigsqcup \{ \hat{s}(\psi):\psi \in \Phi_{1,f},\psi = \epsilon^1_x(\psi) \leq \phi)\} \\
&& : s \in S,\hat{s} = (\epsilon^1_x,\epsilon^2_y)(\hat{s}) \leq f\}  \\
&=& \bigsqcup \{ \hat{s}(\phi):s \in S,\hat{s} = (\epsilon^1_x,\epsilon^2_y)(\hat{s}) \leq f\}
\end{eqnarray*}
and this means then $(\epsilon^1_x,\epsilon^2_y)(f) = \bigsqcup \{\hat{s}:\hat{s} \leq f\}$, that is \textit{Density} in $[\Phi_1 \rightarrow \Phi_2]_c$.

In order to prove (\ref{eqKeyRel}) consider a finite element $\psi = \epsilon^1_x(\psi) \in \Phi_{1,f}$. Then, by density in $\Phi_{2,f}$,
\begin{eqnarray*}
\epsilon^2_y(f(\psi)) = \bigsqcup \{\beta:\beta \in \Phi_{2,f},\beta = \epsilon^2_y(\beta) \leq \epsilon^2_y(f(\psi))\}.
\end{eqnarray*}
As always, we may on the left replace $\epsilon^2_y(f(\psi))$ simply by $f(\psi)$. Fix an element $\beta \in \Phi_{2,f}$ such that $\beta = \epsilon^2_y(\beta) \leq f(\psi)$ and define a simple map $s$ with $Y(s) = \{\psi\}$ and $s(\psi) = \beta$. It follows that
\begin{eqnarray*}
\hat{s}(\phi) = \left\{ \begin{array}{ll} \beta & \textrm{if}\ \psi \leq \phi, \\ 1 & \textrm{otherwise}. \end{array} \right..
\end{eqnarray*}
Obviously we have $\hat{s}(\phi) \leq f(\phi)$ for all $\phi \in \Phi_1$, hence $\hat{s} \leq f$. 

Next we show that these maps $\hat{s}$ have support $(x,y)$, that is $\hat{s} = (\epsilon^1_x,\epsilon^2_y)(\hat{s})$ or $\hat{s}(\phi) = (\epsilon^1_x,\epsilon^2_y)(\hat{s}(\phi))$ for all $\phi \in \Phi_1$. Assume first that $\psi \leq \phi$. Then we have
\begin{eqnarray*}
\epsilon^2_y(\hat{s}(\epsilon^1_x(\phi)) = \epsilon^2_y(\beta) = \beta = \hat{s}(\phi),
\end{eqnarray*}
since $\psi = \epsilon^1_x(\psi) \leq \phi$ if and only if $\psi = \epsilon^1_x(\psi) \leq \epsilon^1_x(\phi))$. Otherwise we see that $\epsilon^2_y(\hat{s}(\epsilon^1_x(\phi)) = 1 = \hat{s}(\phi)$. So, we conclude indeed that $\hat{s} = (\epsilon^1_x,\epsilon^2_y)(\hat{s})$.

Note that these maps $\hat{s}$ are of particular form, so a fortiori we see that for elements $\psi \in \Phi_{1,f}$, $\psi = \epsilon^1_x(\psi) \leq \phi$,
\begin{eqnarray*}
\epsilon^2_y(f(\psi)) \leq \bigsqcup \{\hat{s}(\psi):s \in S,\hat{s} = (\epsilon^1_x,\epsilon^2_y)(\hat{s}) \leq f\}.
\end{eqnarray*}
The right hand side is obviously smaller than $f(\psi)$. This shows that (\ref{eqKeyRel}) holds and thus \textit{Density} is valid.

It remains to verify \textit{Compactness} of $\hat{S}$. Let $G \subseteq [\Phi_1 \rightarrow \Phi_2]_c$ be a directed set of continuous maps and $\hat{s} \leq \bigsqcup G$. Then, if $\psi \in Y(s)$, we have $\hat{s}(\psi) \in \Phi_{2,f}$, and $\hat{s}(\psi) \leq (\bigsqcup G)(\psi) = \bigsqcup_{g \in G} g(\psi)$. The set of elements $g(\psi)$ for $g \in G$ is directed in $\Phi_{2,f}$. By compactness of the information algebra $\Phi_2$, there is a $g_\psi \in G$ such that $\hat{s}(\psi) \leq g_\psi(\psi)$. But $Y(s)$ is a finite set, therefore there is a $g \in G$ so that $g_\psi \leq g$ for $\psi \in Y(s)$.

Then we have $s(\psi) \leq \hat{s}(\psi) \leq g_\psi(\psi) \leq g(\psi)$ for any $\psi \in Y(s)$. But for any $\phi \in \Phi_1$, $\hat{s}(\phi)$ is the join of finitely many $s(\psi)$, and therefore we conclude that $\hat{s}(\phi) \leq g(\phi)$, hence $\hat{s} \leq g$ for some $g \in G$. This is compactness. At the same time we have shown that $\hat{s} \ll \hat{s}$ in the continuous information algebra $[\Phi_1 \rightarrow \Phi_2]_c$ if both $\Phi_1$ and $\Phi_2$ are compact. This confirms that the algebra of continuous maps is indeed compact in this case, see Theorem \ref{th:ContBeComp}.

Let's state this result in a Theorem.

\begin{theorem} \label{th:CompAlgMaps}
If $(\Phi_1,\cdot,0,1;E_1)$ and $(\Phi_2,\cdot,0,1;E_2)$ are compact information algebras, then the information algebra $([\Phi_ \rightarrow \Phi_2]_c,\cdot,0,1;E_1 \cdot E_2)$ is compact too. Its finite elements are the maps $\hat{S}$ defined on the base of simple functions $S$.
\end{theorem}

Based on these results about information maps between information algebras, different Cartesian closed categories of information algebras will be defined in the next section.


\section{Cartesian closed categories of information algebras} \label{subsec:CartClos}

We consider the categories of idempotent, domain-free valuation algebras \textbf{IA}, and of compact and continuous valuation algebra \textbf{COMPIA} and \textbf{CONTIA} and we are going to show that these categories are all Cartesian closed. We we do not require in the sequel, that the information algebras $(\Phi,\cdot,0,1;E)$ with $E = \{\epsilon_x:x \in Q)$ satisfy the Support axiom. More precisely we consider the following categories.

\begin{enumerate}
\item The category $\mathbf{IA}$ has as objects domain-free information algebras and as morphisms \textit{information maps} $\Phi \rightarrow \Psi$.
\item The category of \textit{continuous} valuation algebras $\mathbf{CONTIA}$ has as objects continuous information algebras and as morphisms \textit{continuous maps} $\Phi \rightarrow \Psi$.
\item The category of \textit{algebraic} valuation algebras $\mathbf{COMPIA}$ has as objects compact information algebras and as morphisms \textit{continuous maps} $\Phi \rightarrow \Psi$.
\end{enumerate}
 
The category $\mathbf{COMPIA}$ is a subcategory of $\mathbf{CONTIA}$, which itself is a subcategory of $\mathbf{IA}$. We are going to show that all these categories are \textit{Cartesian closed}. To remind: A category $\mathbf{C}$ is Cartesian closed, if it satisfies the following three conditions:
\begin{enumerate}
\item  The category $\mathbf{C}$ has a \textit{terminal object}: There is an object $T \in \mathbf{C}$ such that there is exactly one morphism from any object to $T$.
\item The category $\mathbf{C}$ has \textit{finite products}: For any pair of objects $A,B \in \mathbf{C}$, there is an object $A \times B$ and morphisms $p_{A};A \times B \rightarrow A$ and $p_{B}:A \times B \rightarrow B$, such for any object $C$ and for any pair of morphisms $f_{1}:C \rightarrow A$ and $f_{2}:C \rightarrow B$ there is a morphism $f:C \rightarrow A \times B$ so that $p_{A} \circ f = f_{1}$ and $p_{B} \circ f = f_{2}$.
\item The category $\mathbf{C}$ has \textit{exponentials}: For any pair of objects $B,C \in \mathbf{C}$, there is an object $C^{B}$ and a morphism $eval:C^{B} \times B \rightarrow C$ such that for for every morphism $f:A \times B \rightarrow C$ there is a unique morphism $\lambda f:A \rightarrow C^{B}$ so that $eval \circ (\lambda f,id_{B}) = f$. 
\end{enumerate}

We are going to show that these elements exist for our three categories \textbf{IA}, \textbf{CONTIA} and \textbf{COMPIA}. The terminal object in all three cases is simply the valuation algebra $(\{0\},\cdot,0,0;\{id\})$. The finite product is the Cartesian product of valuation algebras.

\begin{theorem} \label{th:CatProd}
The Cartesian product $(\Phi _1\times \Phi_2,\cdot,(0,0),(1,1);E_1 \times E_2)$ of two (continuous, compact) information algebras $(\Phi_1,\cdot,0,1;E_1)$ and $(\Phi_2,\cdot,0,1;E_2)$ with $E_1 = \{\epsilon^1_x:x \in Q_1\}$ and $E_1 = \{\epsilon^2_y:y \in Q_2\}$ under component-wise combination and also component-wise information extraction, is the categorial direct product of the two valuation algebras in \textbf{IA} (\textbf{CONTIA}, \textbf{COMPIA}, respectively).
\end{theorem}

\begin{proof}
We verify first that $(\Phi _1\times \Phi_2,\cdot,(0,0),(1,1);E_1 \times E_2)$ is an information algebra. Combination in $\Phi_1 \times \Phi_2$ is defined component-wise and it is obvious that $(\Phi_1 \times \Phi_2;\cdot)$ is then an idempotent commutative semigroup with null element $(0,0)$ and unit $(1,1)$.  

For any pair $(\epsilon^1_x,\epsilon^2_y)$ in $E_1 \times E_2$, an operator
\begin{eqnarray*}
(\epsilon^1_x,\epsilon^2_y)(\phi_1,\phi_2) = (\epsilon^1_x(\phi_1),\epsilon^2_y(\phi_2))
\end{eqnarray*}
is defined. It is straightforward to verify that this operator is an existential quantifier in $\Phi_1 \times \Phi_2$, which is therefore an information algebra.

We define the projections $p_i$ by $p_i(\phi_1,\phi_2) = \phi_i$ for $i=1,2$. These projections are clearly information maps. Consider then an idempotent information algebra $(\Phi,\cdot,0,1;E)$ and two information maps $f_i :\Phi \rightarrow \Phi_i$, for $i=1,2$. Define $f : \Phi \rightarrow  \Phi_1 \times \Phi_2$ by $f(\phi) = (f_1(\phi),f_2(\phi))$. Again, $f$ is an information map. Then, $f_i = p_i \circ f$ for $i=1,2$. Thus, the product algebra $\Phi_1 \times \Phi_2$ is the direct product of then information algebras $\Phi_1$ and $\Phi_2$ in \textbf{IA}. 

Next, we show that the Cartesian product of two continuous valuation algebras is continuous. Let then $B_1$ and $B_2$ be bases in $\Phi_1$ and $\Phi_2$ respectively. Obviously $B_1 \times B_2$ is closed under join and contains the unit element $(1,1)$ as well as the null element $(0,0)$. We claim that $B_1 \times B_2$ is a basis of $\Phi_1 \times \Phi_2$. Let $D \subseteq B_1 \times B_2$ be a directed set and define $D_{1} = \{\phi_1 \in B_1:\exists \phi_2 \in B_2 \textrm{ so that}\ (\phi_1,\phi_2) \in D\}$. $D_{2}$ is defined similarly as the set of elements in $B_2$ obtained from $D$. Both $D_{1}$ and $D_{2}$ are clearly directed. Then $(\bigsqcup D_{1},\bigsqcup D_{2})$ is an upper bound of $D$, and it is obviously its supremum. So $\bigsqcup D = (\bigsqcup D_{1},\bigsqcup D_{2})$ exists in $\Phi_1 \times \Phi_2$. This is the convergence property.

We have $(\phi_1',\phi_2') \ll (\phi_1,\phi_2)$ if and only if $\phi_1' \ll \phi_1$ and $\phi_2'\ll \phi_2$, the $\ll$-relation taken in $\Phi_1 \times \Phi_2$, $\Phi_1$ and $\Phi_2$ respectively. Consider $(\phi_1,\phi_2) \in \Phi_1 \times \Phi_2$. Then
\begin{eqnarray*}
\lefteqn{\bigsqcup \{(\phi_1',\phi_2') \in B_1 \times B_2: (\phi_1',\phi_2') \ll (\phi_1,\phi_2)\}} \\
&&=(\bigsqcup \{\phi_1' \in B_1:\phi_1' \ll \phi_1\},\bigsqcup \{\phi_2' \in B_2:\phi_2' \ll \psi_2\}) \\
&&=(\phi_1,\phi_2).
\end{eqnarray*}
This shows that $\Phi_1 \times \Phi_2$ is a continuous lattice.

If $(\epsilon^1_x,\epsilon^2_y) \in E_1 \times E_2$, then we obtain in the same way
\begin{eqnarray*}
\lefteqn{\bigsqcup \{(\phi_1',\phi_2') \in B_1 \times B_2:  
(\phi_1',\phi_2')  = (\epsilon^1_x,\epsilon^2_y)(\phi_1',\phi_2') \ll (\epsilon^1_x,\epsilon^2_y)(\phi_1,\phi_2)\} }  \\
&&= (\bigsqcup \{\phi_1' \in B_1:\phi_1' = \epsilon_1(\phi_1') \ll \epsilon_1(\phi_1)\},
\bigsqcup \{\phi_2' \in B_2:\phi_2' = \epsilon_2(\phi_2') \ll \epsilon_2(\phi_2)\} ) \\
&&=(\epsilon_1(\phi_1),\epsilon_2(\phi_2)) = (\epsilon^1_x,\epsilon^2_y)(\phi_1,\phi_2). 
\end{eqnarray*}
So strong density holds too. This proves that $\Phi_1 \times \Phi_2$ is a continuous information algebra, see Theorem \ref{th:ComExtrJoin2}.

The projections $p_1$ and $p_2$  are obviously continuous maps. Let then $(\Phi,\cdot,0,1;E)$ be a continuous information algebra and $f_{1}$ and $f_{2}$ be continuous maps $f_{1}:\Phi \rightarrow \Phi_1$ and $f_{2}:\Phi \rightarrow \Phi_2$. Then we define $f = (f_{1},f_{2})$ as a map from $\Phi$ to $\Phi_1 \times \Phi_2$. It is continuous, since its components $f_{1}$ and $f_{2}$ are so. Then clearly $p_1 \circ f = f_{1}$ and $p_2 \circ f = f_{2}$. It follows that $\Phi_1 \times \Phi_2$ is the direct product in \textbf{CONTIA}.

If $\Phi_1$ and $\Phi_2$ are compact information algebras, then $(\Phi_1 \times \Phi_2,\cdot,0.,1;E_1 \times E_2)$ is a compact information algebra, and its finite elements are given by the Cartesian product of the finite elements of each factor since $(\phi_1',\phi_2') \ll (\phi_1,\phi_2)$ exactly if $\phi_1' \ll \phi_1$ and $\phi_2' \ll \phi_2$.  So, $\Phi_1 \times \Phi_2$ is the direct product in \textbf{COMPIA}. This completes the proof.
\end{proof}

Next we show that the information algebras of monotone or continuous maps are the exponentials of the respective category of idempotent, continuous or compact information algebras.

\begin{theorem} \label{th:ContExpo}
If $(\Phi_1,\cdot,0,1;E_1)$ and $(\Phi_2,\cdot,0,1;E_2)$ are two objects of the category \textbf{IA}, then the information algebra $([\Phi_1 \rightarrow \Phi_2],\cdot,0,1;E_1 \times E_2)$ is an exponential of \textbf{IA}. If $(\Phi_1,\cdot,0,1;E_1)$ and $(\Phi_2,\cdot,0,1;E_2)$ are two objects of the categories \textbf{CONTIA} or \textbf{COMPIA}, then the information algebra $([\Phi_1 \rightarrow \Phi_2],\cdot,0,1;E_1 \times E_2)$ is an exponential of the respective categories.
\end{theorem}

\begin{proof}
We treat only the case of continuous information algebras, the other cases follow in the same way. We know from Theorem \ref{th:ContOfCombExtr} that $([\Phi_1 \rightarrow \Phi_2],\cdot,0,1;E_1 \times E_2)$ is a continuous information algebra. We define the morphism $eval:[\Phi_1 \rightarrow \Phi_2]_{c} \times \Phi_1 \rightarrow \Phi_2$  for $f \in [\Phi_1 \rightarrow \Phi_2]_{c}$ and $\phi \in \Phi_1$ by
\begin{eqnarray}
eval(f,\phi) = f(\phi).
\nonumber
\end{eqnarray}
The map $eval$ is continuous.

Consider another continuous valuation algebra $(\Phi,\cdot,0,1;E)$ and let $f:\Phi \times \Phi_1 \rightarrow \Phi_2$ be a continuous map. Then we define a map $\lambda f:\Phi \rightarrow [\Phi_1 \rightarrow \Phi_2]_{c}$ for $\chi \in \Phi$ and $\phi \in \Phi_1$ by
\begin{eqnarray}
\lambda f(\chi)(\phi) = f(\chi,\phi).
\nonumber
\end{eqnarray}
The map $\lambda f$ is continuous if $f$ is so. In fact, let $D$ be a directed set in $\Phi$. Then we have for $\phi \in \Phi_1$, 
\begin{eqnarray}
\lambda f(\bigsqcup D)(\phi) = f(\bigsqcup D,\phi) = f(\bigsqcup_{\chi \in D}(\chi,\phi)) = \bigsqcup_{\chi \in D} f(\chi,\phi) = \bigsqcup_{\chi \in D} \lambda f(\chi)(\phi).
\nonumber
\end{eqnarray}
Thus we see that $\lambda f(\bigsqcup D) = \bigsqcup_{\chi \in D} \lambda f(\chi)$.

Now finally for $(\chi,\phi) \in \Phi \times \Phi_1$, we obtain that $eval \circ (\lambda f,id_{\Phi_1})(\chi,\phi) = eval(\lambda f(\chi),\phi) = \lambda f(\chi)(\phi) = f(\chi,\phi)$. So indeed $eval \circ (\lambda f,id_{\Phi_1}) = f$. 

The cases of ordinary and of compact information algebras are treated in exactly the same way. 
\end{proof}

This shows that the categories $\mathbf{IA}$, $\mathbf{COMPIA}$ and $\mathbf{CONTIA}$ are all Cartesian closed.


\section{Lattice-valued information algebras}

As an illustration, we introduce in the section a further example of a class of information algebras, among which we have both compact and continuous information algebras. Consider an set $U$ as an universe (of possible worlds), and, as with set algebras, we assume that questions $x \in Q$ are represented  by equivalence relations $\equiv_x$, so that question $x$ has the same answer in two possible worlds $u$ and $v$, if $u \equiv_x v$ (see Section \ref{sec:SetAlg}). As there, we have $x \leq y$ if $u \equiv_y v$ implies $u \equiv_x v$ for all pairs $\{u,v\}$. We assume for simplicity's sake that all equivalence classes $[u]_x$ (or blocks $B_x$ of the associated partitions $P_x$) have finite cardinality. 

Consider now a bounded, distributive lattice $(\Lambda,\wedge,\vee,0,1)$ with $0$ as least and $1$ as greatest element. Recall that in $\Lambda$ an order $\alpha \leq \beta$ is defined either by $\alpha \wedge \beta = \alpha$ or equivalently by $\alpha \vee \beta = \beta$. Then we define $\Lambda$-valuations $\phi$ on $U$ as maps $\phi : U \rightarrow \Lambda$. Let $\Phi$ be the set of all $\Lambda$-valuations on $U$. Then we define in $\Lambda$ the following operations of combination and extractions:\begin{enumerate}
\item \textit{Combination:} For all $\phi,\psi \in \Phi$, $\phi \cdot \psi$ is defined by $(\phi \cdot \psi)(u) = \phi(u) \wedge \psi(u)$ for all $u \in U$,
\item \textit{Extraction:} For all $\phi \in \Phi$ and $x \in Q$, $\epsilon_x(\phi)$ is defined by $\epsilon_x(\phi)(u) = \vee_{v \equiv_x u} \phi(v)$. for all $u \in U$
\end{enumerate}

It is clear that $(\Phi,\cdot)$ is a commutative semigroup with the valuations $1(u) = 1$ and $0(u) = 0$ for all $u \in U$ as unit and null elements. 

Further, we have $\epsilon_x(0) = 0$ since $(\epsilon_x(0))(u) = \vee_{v \equiv_x u} 0(u) = 0$ for all $u \in U$. Also, by distributivity of the lattice $\Lambda$,
\begin{eqnarray*}
(\epsilon_x(\phi) \cdot \phi)(u) = (\vee_{v \equiv_x u} \phi(v)) \wedge \phi(u) = \vee_{v \equiv_x u} (\phi(v) \wedge \phi(u)) = \phi(u),
\end{eqnarray*}
so that $\epsilon_x(\phi) \cdot \phi = \phi$. And, then we have
\begin{eqnarray*}
\lefteqn{(\epsilon_x(\epsilon_x(\phi) \cdot \psi))(u) = \vee_{v \equiv_x u} (\epsilon_x(\phi) \cdot \psi))(v) } \\
&&= \vee_{v \equiv_x u} ((\vee_{w \equiv_x v} \phi(w)) \wedge \psi(v)) \\
&&= (\vee_{v \equiv_x u} (\vee_{w \equiv_x v} \phi(w))) \wedge (\vee_{v \equiv_x u} \psi(v)) \\
&&= (\vee_{w \equiv_x u} \phi(w)) \wedge (\vee_{v \equiv_x u} \psi(v)) \\
&&= (\epsilon_x(\phi) \cdot \epsilon_x(\psi))(u)
\end{eqnarray*}
since $w \equiv_x v \equiv_x u$ if and only if $w \equiv_x u$. So we have $\epsilon_x(\epsilon_x(\phi) \cdot \psi) = \epsilon_x(\phi) \cdot \epsilon_x(\psi)$ and the operators $\epsilon_x$ are existential quantifiers.

A valuation $\phi$ which takes constant values on any equivalence class $[u]_x$, that is $\phi(u) = \phi(v)$ has support $x$, $\epsilon_x(\phi) = \phi$. Note further, that if $x$ is a support of $\phi$ and $y \geq x$, then $y$ is also a support of $\phi$, since $u \equiv_y v$ implies $u \equiv_x v$. All this together shows that $(\Phi,\cdot,0,1;E)$ with $E = \{\epsilon_x:x \in Q\}$ is a domain-free information algebra, called a lattice valued information algebra.

Concerning the information order in $\Phi$ we remark that $\phi \leq \psi$ if and only if $\phi(u) \geq \psi(u)$ for all $u \in U$. This inversion of the information order with respect to the order in $\Lambda$ is underlined by the fact that combination $\phi \cdot \psi$ is join (supremum) in information order, but defined by meet (infimum) in $\Lambda$. The assumption that $\Lambda$ is a distributive lattice implies in fact that $\Phi$ is also a distributive lattice in information order. Indeed we have
\begin{eqnarray*}
\phi \cdot \psi = \phi \vee \psi \textrm{ if and only if}\ (\phi \vee \psi)(u) = \phi(u) \wedge \psi(u), \\
\end{eqnarray*}
And similarly, meet in $\Phi$ is defined by
\begin{eqnarray*}
(\phi \wedge \psi)(u)  = \phi(u) \vee \psi(u). 
\end{eqnarray*}
It can easily be verified that this valuation $\phi \wedge \psi$ is indeed the infimum in information order. The unit valuation $1$ and the null valuation $0$ are the smallest and the greatest elements in information order. So $(\Phi,\leq)$ is a bounded lattice. Distributivity follows from the definitions of join and meet in $\Phi$ and the distributivity of $\Lambda$. In addition, extraction distributes over meet.

\begin{proposition} \label{prop:DistrExtrMee}
For all valuations $\phi,\psi \in \Phi$ and for all $x \in Q$,
\begin{eqnarray*}
\epsilon_x(\phi \wedge \psi) = \epsilon_x(\phi) \wedge \epsilon_x(\psi).
\end{eqnarray*}
\end{proposition}

\begin{proof}
The proof is straightforward: For any $u \in U$, we have by definition and associativity of join
\begin{eqnarray*} 
(\epsilon_x(\phi \wedge \psi))(u) &=& \vee_{v \equiv_x u} (\phi(v) \vee \psi(v))
= (\vee_{v \equiv_x u} \phi(v)) \vee (\vee_{v \equiv_x u} \psi(v))  \\
&=& (\epsilon_x(\phi) \wedge \epsilon_x(\psi))(u).
\end{eqnarray*}
This proves the identity $\epsilon_x(\phi \wedge \psi) = \epsilon_x(\phi) \wedge \epsilon_x(\psi)$.
\end{proof}

We refer to the end of Section \ref{subsec:BooleInfAlg} for a note on the representation theory based on Priestley spaces of such an information algebra where $(\Phi,\leq)$ is a distributive lattice.

If we take for $\Lambda$ the Boolean lattice $\{0,1\}$ with $0 \leq 1$, then we see that the corresponding $\{0,1\}$-valuations  on $U$ are set-indicator functions relative to the the subsets of $U$. And the information algebra of these $\{0,1\}$-valuations corresponds to a set algebra (see Section \ref{sec:SetAlg}). 

Are there compact or continuous lattice-valued information algebras? The answer is yes, see \cite{guanlikohlas15}. In fact, it is sufficient and necessary that the underlying lattice $\Lambda$ has the same property. Note however that $\phi \leq \psi$ if and only if $\phi(u) \geq \psi(u)$, that is information order in $\Phi$ inverses order in $\Lambda$. There fore we must rather consider the lattice $(\Lambda;\leq^\vartheta)$ with $u \leq v^\vartheta$ iff $v \leq u$ in the original order. Then meet and join interchange, $\wedge^\vartheta = \vee$ and $\vee^\vartheta = \wedge$.

\begin{theorem} \label{th:ContLattValIA}
A lattice-valued information algebra $(\Phi,\cdot,0,1;E)$ with $E = \{\epsilon_x:x \in Q\}$ based on a lattice $\Lambda$ is continuous (compact) if and only if the lattice $(\Lambda,\leq^\vartheta)$ is continuous (compact).
\end{theorem}

\begin{proof}
Using Theorem \ref{th:ContLatt}, the proof is straightforward, since the relevant properties of $\Lambda$ carry over to $\Phi$. We verify first that the lattice $\Phi$ is complete if and only if the lattice $\Lambda$ is so. Consider any subset $X$ of $\Phi$ and associate with it the subsets $X_u = \{\phi(u):\phi \in X\}$ of $\Lambda$ for $u \in U$. Let $\psi(w) = \bigwedge^\vartheta X_u$, if $\Lambda$ is a complete lattice. Then $\psi$ is a lower bound of $X$. If $\chi$ is another lower bound of $X$, then $\chi(u)$ is a lower bound of $X_u$, hence $\chi(u) \leq^\vartheta \psi(u)$ and therefore $\chi \leq \psi$. So $\psi = \bigwedge X$. Conversely, if $\Phi$ is a complete lattice and $X$ any subset of $\Lambda$, consider the subset $X' =\{\psi \in \Phi: \psi(u) = \lambda, \forall u \in U, \lambda \in X\}$ of constant maps in $\Phi$. Then, by assumption, the meet of $X'$ exists in $\Phi$. Let $\phi = \wedge X'$. As before it follows that $\phi(u)$ is the least upper bound of $X$, hence the meet $\bigwedge^\vartheta X = \phi(u)$ exists. In both cases it follows from the existence of arbitrary meets the existence of arbitrary join since the lattices are bounded \cite{daveypriestley97}. Therefore $\Phi$ is a complete lattice if and only if $\Lambda$ is a complete lattice.

Next we show in the same way that $\psi \ll \phi$ if and only if $\psi(u) \ll^\vartheta \phi(u)$ for all $u \in U$. Consider a directed subset $D$ of $\Phi$ and the associated subsets $D_u = \{\psi(u):\psi \in D\}$ for $u \in U$. Obviously all $D_u$ are directed in $\Lambda$ (under the order $\leq^\vartheta$). If $\phi \leq \bigsqcup D$, then $\phi(u) \leq^\vartheta \bigsqcup^\vartheta D_u$ for all $u \in U$. And if $\psi \in \Phi$ such that $\psi(u) \ll^\vartheta \phi(u)$, then there is an element $\chi(u) \in  D_u$ such that $\psi(u) \leq^\vartheta \chi(u)$. But then $\psi \leq \chi \in D$ and $\psi \ll \phi$. Conversely, assume $\psi \ll \phi$ and consider a directed subset $D$ of $\Lambda$. Suppose $\phi(u) \leq^\vartheta \bigsqcup ^\vartheta D$. Then define $D' = \{\chi:\chi(u) = \lambda, \forall u \in U,\lambda \in X\}$. This set is directed in $\Phi$. Then we have $\phi \leq \bigsqcup D'$, hence there is a $\chi \in D'$ such that $\psi \leq \chi$, hence $\psi(u)^\vartheta \leq \chi(u) \in D$, so that $\psi(u) \ll^\vartheta \phi(u)$ for all $u \in U$. This proves the claim at the beginning of the paragraph.

Finally, recall that $\phi$ has support $x$ if and only if $\phi(u)$ is constant on the equivalence class $[u]_x$ of the equivalence relation $u \equiv_x v$. This implies that in $\Phi$ local density holds if and only if $density$ is valid in $\Lambda$. Indeed, note that $\psi = \epsilon_x(\psi) \ll \epsilon_x(\phi)$ implies $\psi(v) = \psi(u) \ll^\vartheta \phi(u) = \phi(v)$ for all $v \equiv_x u$. By density in $(\Lambda,\leq^\vartheta)$ we have 
\begin{eqnarray*}
\phi(u) = \bigsqcup \{\lambda \in \Lambda:\lambda = \psi(u) \ll^\vartheta \phi(u)\}
\end{eqnarray*}
This implies $\phi = \epsilon_x(\phi) = \bigsqcup \{\psi \in \Phi:\psi = \epsilon_x(\psi) \ll \phi\}$, that is, local density in $\Phi$. Conversely consider the set $\{\eta \in \Lambda:\eta \ll^\vartheta \lambda\}$ and define constant $\Lambda$-valuations $D =\{\psi \in \Phi:\psi(u) ) = \eta, \forall \eta \ll^\vartheta \lambda\}$ and $\phi(u) = \lambda$. Any $x \in Q$ is a support of any $\psi \in D$ and for $\phi$ and $\psi \ll \phi$. So by local density $\phi = \epsilon_x(\phi) = \bigsqcup D$, hence $\lambda = \phi(u) = \bigsqcup \{\eta \in \Lambda:\eta \ll^\vartheta \lambda\}$ and density holds in $\Lambda$. This concludes the proof for the case of continuous lattices $\Phi$ and $\Lambda$.

The case of compact lattices follows from the continuous one, since $\phi \ll \phi$ if and only if $\phi(u) \ll^\vartheta \phi(u)$ for all $u \in U$.
\end{proof}

So lattice-valued information algebras provide a large family of information algebras, including compact and continuous ones.


\section{Duality for compact and continuous algebras}

In this section we examine duality between domain-free and labeled compact and continuous information algebras. For this purpose we need first to establish what we mean by a compact or continuous labeled information algebra. This can be done by looking at the labeled algebras derived from compact and continuous domain-free algebras. 

We first remark, that if $(\Phi,\cdot,0,1,;E)$ is a compact or continuous domain-free information algebra with set $E = \{\epsilon_x:x \in Q\}$ of extraction operator, we may always add the trivial extraction operator $id$, the identity map of $\Phi$ to $E$. Let $E' = E \cup \{id\}$ and consider $(\Phi,\cdot,0,1,;E')$. Adjoin an element $\top$ to $Q$ corresponding to 
$id$, $id = \epsilon_\top$. Since $\epsilon_x \circ id = id \circ \epsilon_x = \epsilon_x$, we have $x \leq \top$ for all $x \in Q$. Note that $(\Phi,\cdot,0,1,;E')$ is still compact or continuous. This is so, because, thanks to the support axiom, local density implies density, which is local density on $\top$. So, we assume throughout this section that $id$ belongs to $E$ in a domain-free information algebra or that $(Q,\leq)$ has a top element $\top$. We remark that under this assumption, the support axiom is automatically (and trivially) satisfied, since $\top$ is a support for any element $\phi$ of $\Phi$. Further, we recall that any $x \in Q$ is at least a support of elements $0$ and $1$. In this section we always suppose the support axiom to be valid.

Consider a first compact domain-free generalized information algebra $(\Phi,\cdot,0,1,;E)$. We form the dual labeled algebra $(\Psi,\cdot;T)$, where $\Psi$ is the set of pairs $(\phi,x)$ with $\phi \in \Phi$ and $\epsilon_{x}(\phi) = \phi$, see Section \ref{subsec:LabInfAlg}. In particular, let $\Psi_{x}$ be the set of all pairs $(\phi,x)$ for a fixed $x$, so that
\begin{eqnarray}
\Psi = \bigcup_{x \in D} \Psi_{x} .
\nonumber
\end{eqnarray}
Note that idempotency allows, as in the domain-free case, to define a partial order in $\Psi$. In fact, define $(\phi,x) \leq (\psi,y)$ if and only if $(\phi,x) \cdot (\psi,y) = (\phi \cdot \psi,x \vee y) = (\psi,y)$. This implies $\phi \cdot \psi = \psi$ or $\phi \leq \psi$ in $(\Phi,\leq)$ and $x \leq y$ in $(D;\leq)$. Further $T$ is the set of all transport operators $t_x$ for $x \in Q$.

As a preparation, we prove two simple, but useful results about the labeled algebra $(\Psi,\cdot;T)$.

\begin{lemma} \label{le:SupInDualLLabAlg}
Let $(\Phi,\cdot,0,1,;E)$ be a domain-free information algebra and $(\Psi,\cdot;T)$ its dual labeled version. If the supremum of a subset $X$ of $\Psi$ exists in $\Psi$, then
\begin{eqnarray} \label{eq:SupOfPairs}
\bigvee X = (\bigvee_{(\phi,x) \in X} \phi,\bigvee_{(\phi,x) \in X} x).
\end{eqnarray}
\end{lemma}

\begin{proof}
Assume $\bigvee X = (\chi,y)$. Then $(\phi,x) \leq (\chi,y)$ for all $(\phi,x) \in X$, hence $\phi \leq \chi$ and $x \leq y$. Consider other upper bounds $\chi'$ and $y'$ for the elements $\phi$ and $x$, $(\phi,x) \in X$. Then $(\phi,x) \leq (\chi',y')$, hence $(\chi,y) \leq (\chi',y')$. But this implies $\chi \leq \chi'$ and $y \leq y'$ and so indeed $\chi = \bigvee_{(\phi,x) \in X} \phi$ and $y = \bigvee_{(\phi,x) \in X} x$. This is (\ref{eq:SupOfPairs}).
\end{proof}

\begin{lemma} \label{le:SupInDomain}
Let $(\Phi,\cdot,0,1,;E)$ be a domain-free information algebra and $(\Psi,\cdot;T)$ its dual labeled version. Let $X$ be a subset of $\Phi$ such that $\epsilon_{x}(X) = X$, that is, all elements of $X$ have support $X$. If the supremum of $X$ exists in $\Phi$, then $(\bigvee X,x) \in \Psi$ and 
\begin{eqnarray}
 \bigvee_{\psi \in X} (\psi,x) = (\bigvee X,x).
\nonumber
\end{eqnarray}
 \end{lemma}

\begin{proof}
We need only to show that $\bigvee X$ has support $x$. Define $\phi = \bigvee X$. Then, for all $\psi \in X$ we have $\psi = \epsilon_{x}(\psi) \leq \phi$, hence $\psi = \epsilon_{x}(\psi) \leq \epsilon_{x}(\phi)$. So, $\epsilon_{x}(\phi)$ is an upper bound of $X$, therefore $\phi \leq \epsilon_{x}(\phi)$, hence $\phi = \epsilon_{x}(\phi)$. 
\end{proof}

We have further the following result as a corollary of this lemma.

\begin{proposition} \label{prop:ComplLatt}
If $(\Phi,\cdot,0,1,;E)$ is an information algebra such that $(\Phi,\leq)$ is a complete lattice and $(\Psi,\cdot;T)$ its dual labeled information algebra, then $(\Psi_x,\leq)$ is a complete lattice for any $x \in Q$.
\end{proposition}

\begin{proof}
By Lemma  \ref{le:SupInDomain} any subset $X$ of $\Psi_x$ has a supremum if $(\Phi,\leq)$ is a complete lattice. The existence of an infimum of $X$ follows in the same way as in the proof of this lemma, and $\bigwedge X = (\bigwedge_{\phi,x) \in X} \phi,x)$.
\end{proof}

We remark, that if $(\Phi,\leq)$ is a complete lattie, this does not imply that $(\Psi,\leq)$ is also a complete lattice. The next theorem shows how finite elements in $(\Psi_{x};\leq)$ relate to finite elements in $(\Phi;\leq)$.

\begin{theorem} \label{th:DualFiniteEl}
Let $(\Phi,\cdot,0,1,;E)$ be a domain-free compact information algebra with finite elements $\Phi_{f}$ and $(\Psi,\cdot;T)$ its dual labeled version. Then $(\phi,x) \in \Phi$ is finite in $(\Psi_{x};\leq)$ if and only if $\phi$ is finite in $(\Phi;\leq)$, that is, $\phi \in \Phi_{f}$.
\end{theorem}

\begin{proof}
Consider an element $(\phi,x)$ of $\Psi$ with $\phi \in \Phi_{f}$. Let $X$ be a directed subset of $\Psi_{x}$ such that $(\phi,x) \leq \bigvee X$. By Proposition \ref{prop:ComplLatt} this supremum exists. Define $X' = \{\psi \in \Phi:(\psi,x) \in X\}$. Clearly, $X'$ is directed too and since $\bigvee X = (\bigvee X',x)$ (Lemma \ref{le:SupInDualLLabAlg}) the supremum of $X'$ exists in $\Phi$ and $\phi \leq \bigvee X'$. Since $\phi$ is finite in $(\Phi;\leq)$ there is a $\psi \in X'$ such that $\phi \leq \psi$, hence $(\phi,x) \leq (\psi,x) \in X$. This shows that $(\phi,x)$ is finite in $(\Psi_{x};\leq)$. 

Conversely, assume that $(\phi,x)$ is finite in $(\Psi_{x};\leq)$. Let $X$ be a directed subset of $\Phi$, whose supremum exists  in $\Phi$ since $(\Phi,\leq)$ is a complete lattice, and such that $\phi \leq \bigsqcup X$. Then we have $\phi = \epsilon_{x}(\phi) \leq \epsilon_{x}(\bigsqcup X) = \bigsqcup \epsilon_{x}(X)$ (Theorem \ref{th:ContOfExtr}). Define $X' = \{(\epsilon_{x}(\psi),x):\psi \in X\}$. It is a directed set in $(\Psi_{x};\leq)$ and we have $(\phi,x) \leq (\bigsqcup \epsilon_{x}(X),x) = \bigsqcup X'$ (Lemma \ref{le:SupInDomain}). Since $(\phi,x)$ is assumed to be finite in $(\Psi_{x};\leq)$ there is an element $(\epsilon_{x}(\psi),x) \in X'$ such that $(\phi,x) \leq (\epsilon_{x}(\psi),x)$. This implies $\phi \leq \psi$ for an element $\psi \in X$. This shows that $\phi$ is  finite in $(\Phi;\leq)$. 
\end{proof}

According to this theorem, finite elements in $(\Phi;\leq)$ correspond to finite elements in $(\Psi_{x};\leq)$ for domains $x$ which are supports of the finite elements in $(\Phi;\leq)$. Note that finite elements in $(\Psi_{x};\leq)$ are not necessarily finite in $(\Psi;\leq)$ and  that the finite elements in $(\Phi;\leq)$ do not induce finite elements in $(\Psi;\leq)$, as one might have expected. So, if we denote the finite elements in $(\Psi_{x};\leq)$ by $\Psi_{x,f}$, and
\begin{eqnarray}
\Psi_{f} = \bigcup_{x \in D} \Psi_{x,f},
\nonumber
\end{eqnarray}
then $\Psi_{f}$ does not represent the finite elements of $(\Psi;\leq)$ but the union of the locally finite ones. Note that if $(\Phi,\cdot,0,1,;E)$ is a \textit{compact} information algebra, then $\Psi_{f}$ is closed under combination. In fact, if $(\phi,x) \in \Psi_{x,f}$ and $(\psi,y) \in \Psi_{y,f}$, then by Theorem \ref{th:DualFiniteEl} $\phi$ and $\psi$ are finite elements in $(\Phi;\leq)$ and so is its combination $\phi \cdot \psi$. This combination has $x \vee y$ as a support and again by the same theorem, therefore $(\phi,x) \cdot (\psi,y) = (\phi \cdot \psi,x \vee y)$ are finite in $\Psi_{x \vee y,f}$. However, transport of finite elements keeps them not necessarily finite, except if the finite elements of $(\Phi;\leq)$ are closed under extraction. Nevertheless, for $x \leq y$, the element $t_{y}(\phi,x) = (\phi,x) \cdot (1,y)$ remains finite, if $(\phi,x)$ is finite. This is true because $(1,y)$ is a finite element.

Next we show that strong density of the compact algebra $(\Psi,D;\leq,\bot,\cdot,\epsilon)$ induces local density within the domains $\Psi_{x}$ of the dual labeled algebra. That is, the finite elements in $(\Psi_{x};\leq)$ are dense in $\Psi_{x}$ and approximate thus the elements of $\Phi_{x}$.

\begin{theorem} \label{th:DualLocDensity}
Let $(\Phi,\cdot,0,1,;E)$ be a domain-free compact information algebra and $(\Psi,\cdot;T)$ its dual labeled version. Then, for all $(\phi,x) \in \Psi$,
\begin{eqnarray} \label{eq:DualLocalDensity}
(\phi,x) = \bigsqcup\{(\psi,x) \in \Phi_{x,f}:(\psi,x) \leq (\phi,x)\}.
\end{eqnarray}
\end{theorem}

\begin{proof}
By strong density in the algebra $(\Phi,\cdot,0,1,;E)$ we have
\begin{eqnarray}
(\phi,x) &=&(\bigsqcup \{\psi \in \Phi_{f}:\psi = \epsilon_{x}(\psi) \leq \phi\},x)
\nonumber \\
&=&\bigsqcup \{(\psi,x) \in \Psi_{x,f}:(\psi,x) \leq (\phi,x)\}.
\nonumber
\end{eqnarray}
This equality holds by Lemma \ref{le:SupInDomain}.
\end{proof}

So, the dual, labeled version of a compact information algebra is a labeled algebra, where local density according to (\ref{eq:DualLocalDensity}) holds. We take this below as the model to define labeled compact information algebras. Note that order in a labeled information algebra $(\Psi,\cdot;T)$ is defined again by $\phi \leq \psi$ if $\phi \cdot \psi = \psi$. This induces also a partial order in $(\Psi_{x};\leq)$ between the elements $\Psi_{x} = \{\phi \in \Psi:d(\phi) = x\}$ in domain $x$. The following lemma states a few elementary properties of this labeled order.

\begin{lemma} \label{leLabOrder}
Let $(\Psi,\cdot;T)$ be an idempotent labeled information algebra. Then
\begin{enumerate}
\item $x \leq d(\phi)$ implies $t_{x}(\phi) \leq \phi$,
\item $x \geq d(\phi)$ implies $t_{x}(\phi) \geq \phi$,
\item $\phi \leq \psi$ implies $t_{x}(\phi) \leq t_{x}(\psi)$ for any $x \in D$,
\item $\phi,\psi \leq \phi \cdot \psi$,
\item $\phi \leq \psi$ implies $\phi \cdot \chi \leq \psi \cdot \chi$ for any $\chi \in \Phi$.
\end{enumerate}
\end{lemma}

\begin{proof}
1.) follows from the Idempotency Axiom of a labeled information algebra, $t_{x}(\phi) \cdot \phi = \phi$ since $x \vee d(\phi) = d(\phi)$.

2.) follows from $t_{x}(\phi) = \phi \cdot 1_{x}$, hence by idempotency, $t_{x}(\phi) \cdot \phi = \phi \cdot 1_{x} \cdot \phi = \phi \cdot 1_{x} = t_{x}(\phi)$.

3.) Let $d(\phi) = y$ and $d(\psi) = z$ and note that by the Combination axiom $t_x(\phi) \cdot t_x(\psi) = t_x(t_x(\phi) \cdot \psi)$ . Assume first $x \leq y$. Then, since $\phi \cdot \psi= \psi$, we have by item 1, $t_x(\phi) \cdot t_x(\psi) = t_x(t_x(\phi) \cdot \phi \cdot \psi) = t_x(\phi \cdot \psi) = t_x(\psi)$, hence $t_x(\phi) \leq t_x(\psi)$. Next assume $x \geq y$. Then $t_x(\phi) \cdot t_x(\psi) = t_x(\phi \cdot 1_x \cdot \psi) = t_x(1_x \cdot \psi) = 1_x \cdot t_x(\psi) = t_x(\psi)$. Hence again $t_x(\phi) \leq t_x(\psi)$. In the general case, for $x \vee y \vee z$ we conclude, using the first case above, that $t_{x \vee y \vee z}(\phi) \leq t_{x \vee y \vee z}(\psi)$. Since $x \vee y \vee z \geq x$, using the second case above, we obtain $t_x(t_{x \vee y \vee z}(\phi)) \leq t_x(t_{x \vee y \vee z}(\phi))$. But we have (see Lemma \ref{le:ElPropLabInfAlg})) $t_x(\phi) = t_x(t_{x \vee y \vee z}(\phi))$ and $t_x(\psi) = t_x(t_{x \vee y \vee z}(\psi))$, so that $t_x(\phi) \leq t_x(\psi)$

4.) follows from idempotency, $\phi \cdot (\phi \cdot \psi) = \phi \cdot \psi$ and $\psi \cdot (\phi \cdot \psi) = \phi \cdot \psi$.

5.) If $\phi \leq \psi$, we have by idempotency $(\phi \cdot \chi) \cdot (\psi \cdot \chi) = (\phi \cdot \psi) \cdot \chi = \psi \cdot \chi$.
\end{proof}

The lemma shows in particular, that the combination and the transport operations preserve order.

What is the labeled version of a continuous labeled information algebra? To examine this question, we consider the labeled version $(\Psi,\cdot;T)$ with $T = \{t_x:x \in Q\}$ of a continuous information algebra $(\Phi,\cdot,0,1;E)$ with $E = \{\epsilon_x:x \in Q\}$. We recall again that $\Psi$ consists of all pairs $(\phi,x)$, where $\phi \in \Phi$ and $\phi = \epsilon_{x}(\phi)$. 

Assume that $B$ is a basis of the continuous information algebra $(\Phi,\cdot,0,1;E)$. Define $B_{x} = \{(\phi,x):\phi \in B,\epsilon_{x}(\phi) = \phi\}$. We claim that this is a basis in $\Psi_{x}$. In fact, if $(\phi,x),(\psi,x) \in B_{x}$, then $(\phi,x) \cdot (\psi,x) = (\phi \cdot \psi,x) \in B_{x}$ since $B$ is closed under combination or  join. So $B_{x}$ is closed under combination. Further also $(0,x)$ and $(1,x) $ belong to $B$. Consider any directed subset $X$ of $B_{x}$. By Lemma \ref{le:SupInDomain} we have $\bigsqcup X = (\bigsqcup_{(\phi,x) \in X}\ \phi,x) \in \Phi_{x}$. This is the convergence property in $\Psi_{x}$.

Define 
\begin{eqnarray}
\bar{B} = \bigcup_{x \in D} B_{x}.
\nonumber 
\end{eqnarray}
Then, $\bar{B}$ is still closed under combination. In fact, let $(\phi,x) \in B_{x}$ and $(\psi,y) \in B_{y}$, then $\phi,\psi \in B$ and $x$ is a support of $\phi$, $y$ a support of $\psi$. But then $x \vee y$ is a support of $\phi \cdot \psi$. So, since $(\phi,x) \cdot (\psi,y) = (\phi \cdot \psi,x \vee y)$ and $\phi \cdot \psi \in B$, we see that $(\phi,x) \cdot (\psi,y) \in B_{x \vee y}$.

We claim also that a \textit{density} property holds in $\Psi_{x}$. Denote the way-below relation in $(\Psi_{x};\leq)$ by $\ll_{x}$. We prove first the following lemma.

\begin{lemma} \label{le:LabWayBelow}
Let $(\Phi,\cdot,0,1;E)$ be a continuous domain-free information algebra and let $\phi,\psi \in \Phi$ and $\epsilon_{x}(\phi) = \phi$, $\epsilon_{x}(\psi) = \psi$. Then $\psi \ll \phi$, if and only if $(\psi,x) \ll_{x} (\phi,x)$.
\end{lemma}

\begin{proof}
Assume $\psi \ll \phi$ and $\epsilon_{x}(\phi) = \phi$, $\epsilon_{x}(\psi) = \psi$. Consider a directed set $D \subseteq \Psi_{x}$. Then $D' = \{\phi:(\phi,x) \in D\}$ is directed too. Recall that $(\Psi_x,\leq)$ is a complete lattice (Proposition  \ref{prop:ComplLatt}). Now, $(\phi,x) \leq \bigsqcup D$ implies $\phi \leq \bigsqcup D'$. Then there is a $\chi \in D'$ such that $\psi \leq \chi$. Note that $\epsilon_{x}(\chi) = \chi$. Hence we see that $(\psi,x) \leq (\chi,x) \in D$. So indeed $(\phi,x) \ll_{x} (\psi,x)$.

Conversely, assume $(\psi,x) \ll_{x} (\phi,x)$. Consider a directed set $D \subseteq \Phi$ such that $\phi \leq \bigsqcup D$. In a continuous information algebra we have $\epsilon_{x}(\bigsqcup D) = \bigsqcup_{\phi \in D} \epsilon_{x}(\phi)$ (Theorem \ref{th:ComExtrJoin2}). Then $\phi = \epsilon_{x}(\phi) \leq \epsilon_{x}(\bigsqcup D) = \bigsqcup_{\chi \in D} \epsilon_{x}(\chi)$. Therefore $(\phi,x) \leq (\bigsqcup_{\chi \in D} \epsilon_{x}(\chi),x) = \bigsqcup_{\chi \in D} (\epsilon_{x}(\chi),x)$ (Lemma \ref{le:SupInDomain}). Since the set $\{(\epsilon_{x}(\chi),x):\chi \in D\}$ is directed, there must then be a $\chi \in D$ such that $(\psi,x) \leq (\epsilon_{x}(\chi),x)$. Then $\psi = \epsilon_{x}(\psi) \leq \epsilon_{x}(\chi) \leq \chi \in D$. This proves that $\psi \ll \phi$.
\end{proof}

This allows us to derive density, using Lemma  \ref{le:SupInDomain} and Lemma \ref{le:LabWayBelow} in $(\Phi_x,\leq)$,
\begin{eqnarray}
&&\bigsqcup \{(\psi,x) \in B_{x}:(\psi,x) \ll_{x} (\phi,x)\}
\nonumber \\
&=&(\bigsqcup \{\psi:\psi \in B,\psi = \epsilon_{x}(\psi) \ll \phi = \epsilon_{x}(\phi)\},x)
\nonumber \\
&=&(\phi,x).
\nonumber
\end{eqnarray}
This is the density property claimed above. 

Finally, assume $(\psi,x) \ll_{x} (\phi,x)$. By Lemma \ref{le:LabWayBelow} we have $\psi \ll \phi$ and $x$ is a support of both $\psi$ and $\phi$. If $x \leq y$, then $y$ is also a support of both elements. Therefore, again by Lemma \ref{le:LabWayBelow}, we have that $t_{y}(\psi,x) = (\psi,y) \ll_{y} (\phi,y) = t_{y}(\phi,x)$. Conversely, assume that $x$ is a support of $\psi$ and $\phi$ and $x \leq y$. Then, if $(\psi,y) \ll_{y} (\phi,y)$, Lemma \ref{le:LabWayBelow} implies that $\psi \ll \phi$, hence $(\psi,x) \ll_{x} (\phi,x)$. This is an important \textit{compatibility} relation between the way-below relations in different domains $\Psi_{x}$ and $\Psi_{y}$

We summarise these results in the following theorem.

\begin{theorem} \label{th:PropLabVers}
Let $(\Phi,\cdot,0,1;E)$ be a continuous domain-free information algebra with basis $B$ and $(\Psi,\cdot;T)$ the associated dual labeled information algebra. Then the following properties hold:
\begin{enumerate}
\item $B_{x}$ is a basis in $(\Psi_{x};\leq)$, that is $B_{x}$ is closed under combination and contains $(0,x)$ and $(1,x)$. Any directed subset of $B_{x}$ has a supremum in $\Psi_{x}$.
\item $(\phi,x) = \bigsqcup \{(\psi,x) \in B_{x}:(\psi,x) \ll_{x} (\phi,x)\}$, for all $(\phi,x) \in \Psi_{x}$.
\item If $x \leq y$, then $(\psi,x) \ll_{x} (\phi,x)$ if and only if $t_{y}(\psi,x) \ll_{y} t_{y}(\phi,x)$.
\end{enumerate}
\end{theorem}

This theorem serves as a base to define the concept of a labeled continuous information below. But first, we discuss the case of a compact information algebra.

\begin{definition} \label{def:LaCompAlg}
A labeled information algebra $(\Psi,\cdot;T)$ with $T = \{t_x:x \in Q\}$ is called compact, if $(Q;\leq)$ has a greatest element $\top$, and
\begin{enumerate}
\item for all domains $x \in Q$ and elements $\phi$ with $d(\phi) = x$,
\begin{eqnarray}
\phi = \bigsqcup \{\psi \in \Psi_{x,f}:\psi \leq \phi\},
\end{eqnarray}
where $\Psi_{x,f}$ denotes the set of the finite elements of $(\Psi_{x};\leq)$. 
\item If $\psi \in \Psi_{x,f}$ and $y \geq x$, then $t_{y}(\phi) \in \Psi_{y,f}$.
\end{enumerate}
\end{definition}

Let
\begin{eqnarray}
\Psi_{f} = \bigcup_{x \in D} \Psi_{x,f}
\nonumber
\end{eqnarray}
be the set of all locally finite elements. Again, we emphasise that this is not the set of the finite elements of $(\Phi;\leq)$. 

Note that $(\Psi_x,\leq)$ is for any $x \in Q$ a complete lattice. This follows as in the proof of Theorem \ref{th:AlgLatt}. Justification of this definition of compact labeled information algebras will be that the associated dual domain-free information $\Psi/\sigma$ is again compact. Before we show this, we give some useful results. The first one shows that the projection operators $t_x$ are continuous.

\begin{lemma} \label{le:ProjOfSup}
Let $(\Psi,\cdot;T)$  by a labeled compact information algebra, $D$ a directed subset of $(\Psi_y,\leq)$ and $x \leq y$. Then
\begin{eqnarray} \label{eq:ProjOfSu}
t_{x}(\bigsqcup D) = \bigsqcup t_{x}(D).
\end{eqnarray}
\end{lemma}

\begin{proof}
Since $(\Psi_y,\leq)$ is a complete lattice, the supremum of $D$ extists in $\Psi$. Assume first $\phi \in \Psi_y$ such that $\phi \leq \bigsqcup D$, hence $t_{x}(\phi) \leq t_{x}(\bigsqcup D)$. So, $t_{x}(\bigsqcup D)$ is an upper bound of the elements $t_{x}(\phi)$ for $\phi \in D$.

On the other hand, by density in the compact labeled algebra,
\begin{eqnarray}
t_{x}(\bigsqcup D) &=& \bigsqcup \{\psi \in \Phi_{x,f}:\psi \leq t_{x}(\bigsqcup D)\} 
\nonumber \\
&=&\bigsqcup \{\psi \in \Phi_{x,f}:t_{y}(\psi) \leq \bigsqcup  D\}.
\end{eqnarray}
Since $t_{y}(\psi)$ is finite in $\Psi_y$, if $\psi$ is so in domain $\Psi_x$ with $x \leq y$, there is an element $\phi \in D$ such that $t_{y}(\psi) \leq \phi$ if $t_{y}(\psi) \leq \bigsqcup D$. But then it follows that $\psi \leq t_{x}(\phi) \in t_{x}(D)$ and therefore $t_{x}(\bigsqcup D)$ is the least upper bound of $t_{x}(D)$. The set $t_x(D)$ is clearly directed. So, indeed $t_{x}(\bigsqcup D) = \bigsqcup t_{x}(D)$.
\end{proof}

This lemma implies that $\Psi_{f}$ is closed under combination. In fact, consider $\phi \in \Psi_{x,f}$ and $\psi \in \Psi_{y,f}$, and a directed set $D$ in $\Psi_{x \bigvee y}$ such that $\phi \cdot \psi \leq \bigsqcup D$. Then $\phi \leq t_{x}(\bigsqcup D) = \bigsqcup t_{x}(D)$ by Lemma \ref{le:ProjOfSup} and similarly $\psi \leq t_{y}(\bigsqcup D) = \bigsqcup t_{y}(D)$. Both sets $t_{x}(D)$ and $t_{y}(D)$ are directed, and therefore there are elements $t_{x}(\phi') \in t_{x}(D)$ such that $\phi \leq t_{x}(\phi')$ and $t_{y}(\psi') \in t_{y}(D)$ such that $\psi \leq t_{y}(\psi')$. Both $\phi',\psi'$ belong to $D$ and so there is also an element $\chi$ in $D$ such that $\phi',\psi' \leq \chi$. Hence, we conclude that $\phi \cdot \psi \leq \phi' \cdot \psi' \leq \chi \in D$. This proves that $\phi \cdot \psi \in \Psi_{x \bigvee y,f}$, hence $\phi \cdot \psi$ belongs to $\Psi_{f}$. But $\Psi_{f}$ is not necessarily closed under transport.

As a preparation for the examination of the dual domain-free algebra associated with a labeled compact information algebra $(\Psi,\cdot;T)$ we prove the following lemma. Recall that the congruence $\equiv_\sigma$ is defined in Section \ref{subsec:Dual} by $\phi \equiv_\sigma \psi$ if $t_z(\phi) = t_z(\psi)$ for all $z \in Q$.

\begin{lemma} \label{le:SupOfEquivClasses}
Let $(\Psi,\cdot;T)$  be a labeled information algebra, $X$ a subset of $(\Psi,\leq)$ such that its supremum exists in $\Psi$. Then in $\Psi/\sigma$,
\begin{eqnarray}
[\bigvee X]_{\sigma} = \bigvee [X]_{\sigma},
\end{eqnarray}
where $[D]_{\sigma} = \{[\phi]_{\sigma}:\phi \in D\}$. 
\end{lemma}

\begin{proof}
Define $\psi = \bigvee X$ such that $[\psi]_{\sigma} = [\bigvee X]_{\sigma}$ and assume that $d(\psi) = x$. Then, for all $\phi \in X$ we have $\phi \leq \psi$ and $d(\phi) \leq x$. Therefore, for all $\phi \in X$ we have $[\phi]_{\sigma} \leq [\psi]_{\sigma}$ and so $[\psi]_{\sigma}$ is an upper bound of $[X]_{\sigma}$.

Assume $[\chi]_{\sigma}$ to be another upper bound of $[X]_{\sigma}$ and $d(\chi) = y$. For any $\phi$ in $X$ we have $[\chi]_{\sigma} = [\phi]_{\sigma} \cdot [\chi]_{\sigma} = [\phi \cdot \chi]_{\sigma} = [t_{x \vee y}(\phi) \cdot t_{x \vee y}(\chi)]_{\sigma}$. This implies $t_{x \vee y}(\phi) \leq t_{x \vee y}(\chi)$. Since for $\phi \in X$ we have $d(\phi) \leq x$, it follows that $\phi \leq t_{x}(\phi) = t_{x}(t_{x \vee y}(\phi)) \leq t_{x}(t_{x \vee y}(\chi))$. But then $\psi = \bigvee X \leq t_{x}(t_{x \vee y}(\chi))$. It follows that $t_{x \vee y}(\psi) \leq t_{x \vee y}(t_{x}(t_{x \vee y}(\chi))) \leq t_{x \vee y}(t_{x \vee y}(\chi)) = t_{x \vee y}(\chi)$. From this we conclude that $[\psi]_{\sigma} \leq [\chi]_{\sigma}$, such that $[\psi]_{\sigma}$ is the supremum of $[X]_{\sigma}$.
\end{proof}

Now we show that the domain-free information algebra $(\Phi/\sigma,\cdot,[0]_\sigma,[1]_\sigma;E)$, with $\epsilon_x \in E$ defined by $\epsilon_x([\phi]_\sigma) = [t_x(\phi)]_\sigma$, associated with a labeled compact information algebra $(\Psi,\cdot;T)$ is indeed again compact. This justifies the definition of a labeled compact information algebra above. 

\begin{theorem} \label{th:DomFreeCompFromLab}
Let $(\Psi,\cdot;T)$  by a labeled compact information algebra. Then the domain-free information lagebra $(\Psi/\sigma,\cdot,[0]_\sigma,[1]_\sigma;E)$ is a compact information algebra and its finite elements are the elements $[\psi]_{\sigma}$ for $\psi \in \Psi_{f}$.
\end{theorem}

\begin{proof}
We know already that $(\Psi/\sigma,\cdot,[0]_\sigma,[1]_\sigma;E)$ is a domain-free information algebra (see Section \ref{subsec:Dual}). We prove that $(\Psi/\sigma,\leq)$ is an algebraic lattice with finite elements elements $[\psi]_{\sigma}$ for $\psi \in \Psi_{f}$ and that local density holds in the algebra $\Psi/\sigma$. Then from Theorem \ref{th:AltCharCompInfAlg} it follows that $(\Psi/\sigma,\cdot,[0]_\sigma,[1]_\sigma;E)$ is a compact information algebra. 

To show that $(\Psi/\sigma,\leq)$ is complete consider first a subset $X$ of $\Psi/\sigma$. Since in a compact labeled algebra, we assume that $(Q,\leq)$ has a greatest element $\top$, we may take for any $[\psi]_\sigma \in X$ a representant $\psi$ with $d(\psi) = \top$. Let then $X' = \{\psi \in \Psi_\top:[\psi]_\sigma \in X\}$ so that $X = [X']_\sigma$. By Lemma \ref{le:SupOfEquivClasses} we have $[\bigvee X'_\sigma] = \bigvee [X'] = \bigvee X$. So all sets in $\Psi/\sigma$ have a supremum. Since $\Psi/\sigma$ has a least element $[1_x]_\sigma$, it follows by standard results of order theory that $(\Psi/\sigma,\leq)$ is a  complete lattice, \cite{daveypriestley97}.

We show next that the elements $[\psi]_{\sigma}$ for $\psi \in \Psi_{f}$ are exactly the finite elements in $(\Psi/\sigma;\leq)$. So, assume first that $[\psi]_{\sigma}$ is finite in $(\Phi/\sigma;\leq)$. By the Support Axiom, $[\psi]_{\sigma}$ has a support $x$, hence we may select a representant $\psi$ of the class $[\psi]_{\sigma}$ with label $d(\psi) = x$. Consider then a directed set $D$ in $\Psi_{x}$ such that $\psi \leq \bigsqcup D$. Using Lemma \ref{le:SupOfEquivClasses}, we conclude that  $[\psi]_{\sigma} \leq [\bigsqcup D]_{\sigma} = \bigvee [D]_{\sigma}$. Further, the set $[D]_{\sigma}$ is directed in $(\Psi/\sigma;\leq)$. Since $[\psi]_{\sigma}$ is finite in $(\Psi/\sigma;\leq)$ there is an element $[\phi]_{\sigma}$ in $[D]_{\sigma}$ such that $[\psi]_{\sigma} \leq [\phi]_{\sigma}$. But then we may select $\phi \in D$ such that $\psi \leq \phi$. This shows that $\psi$ is finite in $(\Psi_{x};\leq)$.

Conversely, assume that $\psi$ is finite in $(\Psi_{y};\leq)$. Consider a directed set $D$ in $(\Psi/\sigma;\leq)$ such that $[\psi]_\sigma \leq \bigsqcup D$. Since in a compact labeled information algebra $(Q;\leq)$ has a greatest element $\top$, the supremum $\bigsqcup D$ has support $\top$. Let $ = [\eta]_\sigma = \bigsqcup D$. Note that any class $[\phi]_\sigma$ has a representant in $\top$. Define $D' = \{\phi \in \Psi_\top:[\phi]_\sigma \in D\}$. The set $D'$ is directed in $(\Psi_\top;\leq)$ and $\bigsqcup D'$ exists in $\Psi_\top$ and $D = [D']_\sigma$. Take further a representant $\eta$ of the class $[\eta]_\sigma$ in $\Psi_\top$. Then we have $\phi \leq \eta$ for all $\phi \in D'$. Further, by Lemma \ref{le:SupOfEquivClasses} we have $[\eta]_\sigma = [\bigsqcup D']_\sigma$. Since $\eta \in \Psi_\top$ we conclude that $\eta = \bigsqcup D'$. We have therefore $t_\top(\psi) \leq \eta$. Because $t_\top(\psi)$ is finite in $\Psi_\top$ if $\psi \in \Psi_f$, there is a $\phi \in D'$ such that $\psi \leq t_\top(\psi) \leq \phi$. It follows that $[\psi]_\sigma \leq [\phi]_\sigma \in D$, which shows that $[\psi]_\sigma$ is finite in $(\Psi/\sigma;\leq)$.

It remains to show local density. For this purpose consider an element $[\phi]_{\sigma} = \epsilon_{x}([\phi]_{\sigma})$ in $\Psi/\sigma$. We take a representant of $[\phi]_{\sigma}$ with label $d(\phi) = x$. By the local density in the labeled algebra $(\Psi,\cdot;T)$ we have
\begin{eqnarray}
[\phi]_{\sigma} = [\bigsqcup \{\psi \in \Psi_{x,f}: \psi \leq \phi\}]_{\sigma}.
\nonumber
\end{eqnarray}
From Lemma \ref{le:SupOfEquivClasses} and the first partof the proof of this theorem it follows then that 
\begin{eqnarray}
[\phi]_{\sigma} = \bigsqcup \{[\psi]_{\sigma}: [\psi]_{\sigma} \textrm{ finite in}\ (\Psi/\sigma;\leq),[\psi]_{\sigma} = \epsilon_{x}([\psi]_{\sigma}) \leq [\phi]_{\sigma}\}]_{\sigma}.
\nonumber
\end{eqnarray}
This is local density in the domain-free information algebra $\Psi/\sigma$ and this concludes the proof that this algebra is compact.
\end{proof}

In summary, a domain-free compact information algebra $\mathbf{D}$ transforms into an associated dual labeled compact information algebra $\mathbf{LD}$. Conversely, a labeled compact information algebra $\mathbf{L}$  has an associated dual domain-free compact information algebra $\mathbf{DL}$. Then the labeled compact algebra $\mathbf{LD}$ transforms back into the domain-free compact algebra $\mathbf{DLD}$. Similarly, the domain-free compact algebra $\mathbf{DL}$ transforms back into the labeled compact algebra $\mathbf{LDL}$. All this holds under the assumption that $(Q;\leq)$ has a greatest element $\top$, what we assume by definition. We have seen in Section \ref{subsec:Dual} that $\mathbf{D}$ and $\mathbf{DLD}$ are isomorphic under the map $\psi \mapsto [(\psi,x)]_{\sigma}$. Similarly, the labeled algebra $\mathbf{L}$ is isomorphic to the algebra $\mathbf{LDL}$ under the map $\phi \mapsto ([\phi]_{\sigma},x)$. We show that in the case of compact algebras these maps are continuous.

\begin{theorem} \label{th:ContIsom}
Let $(\Phi,\cdot,0,1;E)$ and $(\Psi,\cdot;T)$ be compact domain-free and compact labeled generalised information algebras respectively. Then, if $D$ is a directed subset of $(\Phi,\leq)$ whose supremum has support $x$,
\begin{eqnarray} \label{eq:ContDomainFreeIsom}
[(\bigsqcup D,x)]_{\sigma} = \bigsqcup_{\phi \in D} [(\phi,x)]_{\sigma}.
\end{eqnarray}
Further, if $D$ is a directed subset of $\Psi$ whose supremum exists in $\Psi$ and has label $x$, then
\begin{eqnarray} \label{eq:ContLabeledIsom}
([\bigsqcup D]_{\sigma},x) = \bigsqcup_{\psi \in D} ([\psi]_{\sigma},x).
\end{eqnarray}
\end{theorem}

\begin{proof}
We start with (\ref{eq:ContDomainFreeIsom}). By Theorem \ref{th:ContOfExtr} we have $\bigsqcup D = \epsilon_{x}(\bigsqcup D) = \bigsqcup \epsilon_{x}(D)$. So, using Lemma \ref{le:SupInDomain}
\begin{eqnarray}
[(\bigsqcup D,x)]_{\sigma} 
= [\bigsqcup_{\phi \in D} (\epsilon_{x}(\phi),x)]_{\sigma}.
\nonumber
\end{eqnarray}
From this it follows, using Lemma \ref{le:SupOfEquivClasses},
\begin{eqnarray}
[(\bigsqcup D,x)]_{\sigma}  = \bigsqcup_{\phi \in D} [(\epsilon_{x}(\phi),x)]_{\sigma} = \bigsqcup_{\phi \in D} \epsilon_{x}([(\phi,x)]_{\sigma}).
\nonumber
\end{eqnarray}
But all elements $[(\phi,x)]_{\sigma}$ have support $x$, therefore we conclude
\begin{eqnarray}
[(\bigsqcup D,x)]_{\sigma} =  \bigsqcup_{\phi \in D} [(\phi,x)]_{\sigma}.
\nonumber
\end{eqnarray}   
This is (\ref{eq:ContDomainFreeIsom}).

In order to prove (\ref{eq:ContLabeledIsom}) we note that for $\psi \in D$, we have $\psi \leq \bigsqcup D$ and $d(\psi) \leq x$. This implies $t_{x}(\psi) \equiv_{\sigma} \psi$, hence $\epsilon_{x}([\psi]_{\sigma}) = [t_{x}(\psi)]_{\sigma} = [\psi]_{\sigma}$. So, $x$ is a support for all $[\psi]_{\sigma}$ such that $\psi \in D$. Define $D' = \{t_{x}(\psi):\psi \in D\}$. Then, by Lemma \ref{le:ProjOfSup}, $\bigsqcup D = \bigvee D' = \bigsqcup_{\psi \in D} t_{x}(\psi)$. Therefore, we obtain, (Lemma \ref{le:SupOfEquivClasses}),
\begin{eqnarray}
[\bigsqcup D]_{\sigma} = [\bigsqcup D']_{\sigma} = [\bigsqcup_{\psi \in D} t_{x}(\psi)]_{\sigma}
= \bigsqcup_{\psi \in D} [t_{x}(\psi)]_{\sigma} =  \bigsqcup_{\psi \in D} [\psi]_{\sigma}
\end{eqnarray}
So, from Lemma \ref{le:SupInDomain} we obtain
\begin{eqnarray}
([\bigsqcup D]_{\sigma},x) = (\bigsqcup_{\psi \in D} [\psi]_{\sigma},x) = \bigsqcup_{\psi \in D} ([\psi]_{\sigma},x).
\nonumber
\end{eqnarray}
This is (\ref{eq:ContLabeledIsom}).
\end{proof}

As remarked above, this theorem shows that $\mathbf{D} \cong \mathbf{DLD}$ and $\mathbf{L} \cong \mathbf{LDL}$ under \textit{continuous} isomorphisms, if $\mathbf{D}$ and $\mathbf{L}$ are compact domain-free or labeled information  algebras respectively. 

Next we turn the duality theory of continuous information algebras. We propose the following definition.

\begin{definition} \textbf{Labeled Continuous Information Algebra:} A labeled information algebra $(\Psi,\cdot;T)$ with $T = \{t_x:x \in Q\}$ is called continuous, if $(Q,\leq)$ has a greatest element $\top$ and if there is for all $x \in Q$ a set $B_{x} \subseteq \Psi_{x}$ (the basis in $x$), closed under combination and contains $0_x$ and $1_{x}$, satisfying the following conditions for all $x \in D$:
\begin{enumerate}
\item Convergence: If $D \subseteq B_{x}$ is directed, then $\bigsqcup D \in \Psi_{x}$.
\item Density: For all $\phi \in \Psi_{x}$, $\phi = \bigsqcup \{\psi \in B_{x}:\psi \ll_{x} \phi\}$.
\item Compatibility: If $d(\phi) = d(\psi) = x \leq y$, then $\psi \ll_{x} \phi$ if and only if $t_{y}(\psi) \ll_{y} t_{y}(\phi)$.
\end{enumerate}
\end{definition}

According to this definition and Theorem \ref{th:PropLabVers}, the dual labeled information algebra $(\Psi,\cdot;T)$ associated with a continuous domain-free information algebra $(\Phi,\cdot,0,1;E)$ is itself continuous. We remark that, as in Theorem \ref{th:ContLatt}, it follows that $(\Psi_{x};\leq)$ is a \textit{continuous lattice} for every $x \in D$. 

To establish duality for continuous information algebras, let's start with a labeled continuous information algebra $(\Psi,\cdot;T)$ and consider its associated dual domain-free information algebra $(\Psi/\sigma,\cdot,[0_x]_\sigma,[1_x]_\sigma;E)$. Is this algebra continuous too? An affirmative answer is given by Theorem \ref{th:LabCont-Cont} below. In order to prove this theorem we need two auxiliary results, which have some interest by themselves.

\begin{lemma} \label{le:WayBelowLabAbs}
Let $(\Psi,\cdot;T)$ be a labeled information algebra. Then $\epsilon_{x}([\psi]_{\sigma}) = [\psi]_{\sigma} \ll  [\phi]_{\sigma} = \epsilon_{x}([\phi]_{\sigma})$ in $\Psi/\sigma$ implies $\psi \ll_{x} \phi$ for the representants $\psi$ and $\phi$ of $[\psi]_{\sigma}$ and $[\phi]_{\sigma}$ with $d(\psi) = d(\phi) = x$. Further, if $(\Psi,\cdot;T)$ is a labeled continuous information algebra, then, if $d(\psi) = d(\phi) =x$,  $\psi \ll_{x} \phi$ implies $[\psi]_{\sigma} \ll  [\phi]_{\sigma}$.
\end{lemma}

\begin{proof}
Consider for the first part of the theorem a directed subset $D$ of $\Psi_{x}$, $\phi,\psi \in \Psi_{x}$ representants of the classes $[\phi]_{\sigma}$ and $[\psi]_{\sigma}$ with label $x$ respectively and $\phi \leq \bigsqcup D$. Then $[\phi]_{\sigma} \leq \bigsqcup [D]_{\sigma}$ with $[D]_{\sigma} = \{[\chi]_{\sigma}: \chi \in D\}$  (Lemma \ref{le:SupOfEquivClasses}). The set $[D]_{\sigma}$ is directed, therefore $[\psi]_{\sigma} \ll [\phi]_{\sigma}$ implies that there is a $\eta \in D$ such that $[\psi]_{\sigma} \leq [\eta]_{\sigma}$, hence $\psi \leq \eta$. This proves that $\psi \ll_{x} \phi$. 

For the second part, assume first $\psi \ll_{\top} \phi$ and consider a directed set $D$ in $\Psi/\sigma$ such that $[\phi]_{\sigma} \leq \bigsqcup D$. We may take as representants of the classes $[\eta]_{\sigma}$ in the set $D$ their representants  in $\Psi_{\top}$. Let then $D' = \{\eta \in \Psi_{\top}: [\eta]_{\sigma} \in D\}$. $D'$ is still directed. Now, if $[\phi]_{\sigma} \leq \bigsqcup D$ and $\phi$ is again a representant of $[\phi]_{\sigma}$ in $\Psi_{\top}$, then also $\phi \leq \bigsqcup D'$. Since $\psi \ll_{\top} \phi$, there is an element $\eta \in D'$ such that $\psi \leq \eta$. But then $[\eta]_{\sigma} \in D$ and $[\psi]_{\sigma} \leq [\eta]_{\sigma}$. This shows that $[\psi]_{\sigma} \ll [\phi]_{\sigma}$. Now, if $d(\psi) = d(\phi) = x$ and $\psi \ll_{x} \phi$, then by the compatibility property $t_{\top}(\psi) \ll_{\top} t_{\top}(\phi)$, and $[\psi]_{\sigma} = [t_{\top}(\psi)]_{\sigma}$ and $[\phi]_{\sigma} = [t_{\top}(\phi)]_{\sigma}$, hence $[\psi]_{\sigma} \ll [\phi]_{\sigma}$ as just proved.
\end{proof}

The next lemma is similar as Lemma \ref{le:ProjOfSup} for labeled compact algebras.

\begin{lemma} \label{le:CommProjJoin}
Let $(\Psi,\cdot;T)$ be a labeled continuous information algebra. If $D \subseteq \Psi_{y}$ directed, then for all $x \leq y \in Q$, 
\begin{eqnarray} \label{eq:CommProjJoin}
t_{x}(\bigsqcup D) = \bigsqcup t_{x}(D).
\end{eqnarray} 
\end{lemma}

\begin{proof}
Note that $\bigsqcup D$ exists in $\Psi_{y}$, since $(\Psi_{y};\leq)$ is a complete lattice. Consider a $\psi \in D$ so that $\psi \leq \bigsqcup D$, then $t_{x}(\psi) \leq t_{x}(\bigsqcup D)$, thus $\bigsqcup t_{x}(D) \leq t_{x}(\bigsqcup D)$.

Conversely by density in $\Psi_{x}$ we have
\begin{eqnarray}
t_{x}(\bigsqcup D) = \bigsqcup \{\psi \in \Psi_{x}:\psi \ll_{x} t_{x}(\bigsqcup D)\}.
\nonumber
\end{eqnarray}
By the compatibility condition, $\psi \ll_{x} t_{x}(\bigsqcup D)$ implies $t_{y}(\psi) \ll_{y} t_{y}(t_{x}(\bigsqcup D)) \leq \bigsqcup D$. By the definition of the way-below relation $\ll_{y}$ this means that there is a $\chi \in D$ such that $t_{y}(\psi) \leq \chi$. But then it follows that $\psi = t_{x}(t_{y}(\psi)) \leq t_{x}(\chi) \in t_{x}(D)$, hence $t_{x}(\bigsqcup D) \leq \bigsqcup t_{x}(D)$ and therefore $t_{x}(\bigsqcup D) = \bigsqcup t_{x}(D)$. 
\end{proof} 

Now we are in a position to prove the following theorem.

\begin{theorem} \label{th:LabCont-Cont}
Let $(\Psi,\cdot;T)$ be a labeled continuous information algebra, then the associated dual domain-free information algebra $(\Psi/\sigma,\cdot,[0_x]_\sigma,[1_x]_\sigma|;E)$ is continuous. 
\end{theorem}

\begin{proof}
We first show that $(\Psi/\sigma;\leq)$ is a complete lattice. To this end consider any non-empty subset $X \subseteq \Psi/\sigma$. For any element $[\psi]_{\sigma}$ of $X$ we may take the representant $\psi$ in the top domain $\Psi_{\top}$, $d(\psi) = \top$. Let then $X' = \{\psi \in \Psi_{\top}:[\psi]_{\sigma} \in X\}$. But $(\Psi_{\top};\leq)$ is a complete lattice, hence $\bigvee X'$ exists in $\Psi_{\top}$. By Lemma \ref{le:SupOfEquivClasses}, we have $[\bigvee X']_{\sigma} = \bigvee X$, and so $X$ has a supremum in $\Psi/\sigma$. Since $(\Psi/\sigma;\leq)$ has a smallest element $[1_{\top}]_{\sigma}$, by standard results of lattice theory $(\Psi/\sigma;\leq)$ is a complete lattice.

Next consider any class $[\phi]_{\sigma} \in \Psi/\sigma$. The set $\{[\psi]_{\sigma}:[\psi]_{\sigma} \ll [\phi]_{\sigma} \}$ is directed. Consider the representants of the classes of this set in $\Psi_{\top}$: $\{\psi \in \Psi_{\top}:[\psi]_{\sigma} \ll [\phi]_{\sigma} \}$ and also $\phi \in \Psi_{\top}$. Then, by Lemma \ref{le:SupOfEquivClasses}, Lemma \ref{le:WayBelowLabAbs} and density in the labeled algebra,
\begin{eqnarray*}
\lefteqn{\bigsqcup \{[\psi]_{\sigma}:[\psi]_{\sigma} \ll [\phi]_{\sigma} \}
= [\bigsqcup \{\psi \in \Psi_{\top}:[\psi]_{\sigma} \ll [\phi]_{\sigma}\}]_{\sigma} }\\
&&=[\bigsqcup \{\psi \in \Psi_{\top}:\psi \ll_{\top} \phi\}]_{\sigma}
= [\phi]_{\sigma}.
\end{eqnarray*}
This shows that density hold. Therefore, $(\Psi/\sigma;\leq)$ is a continuous lattice.

By Theorem \ref{th:ComExtrJoin2} it is now sufficient to prove (\ref{eq:ComExtrJoin2}). So, consider a directed set $D \subseteq \Psi/\sigma$. For any $[\psi]_{\sigma} \in D$ we may select the representant $\psi$ in $\Psi_{\top}$. Define $D' = \{\psi \in \Psi_{\top}:[\psi]_{\sigma} \in D\}$. This set is still directed in $\Psi_{\top}$. Now, using repeatedly Lemma \ref{le:SupOfEquivClasses} and Lemma \ref{le:CommProjJoin} 
\begin{eqnarray*}
\lefteqn{\epsilon_{x}(\bigsqcup D)
=\epsilon_{x}(\bigsqcup \{[\phi]_{\sigma}:\phi \in D'\}) = \epsilon_{x}([\bigsqcup D']_{\sigma}) } \\
&& = [t_{x}(\bigsqcup D')]_{\sigma} = [\bigsqcup t_{x}(D')]_{\sigma} = \bigsqcup \{[t_{x}(\phi)]_{\sigma}:[\phi]_{\sigma} \in D\} \\
&&=  \bigsqcup \{\epsilon_{x}([\phi]_{\sigma}):[\phi]_{\sigma} \in D\} = \bigvee \epsilon_{x}(D).
\end{eqnarray*}
This proves that $(\Psi/\sigma,D)$ is a domain-free continuous information algebra.
\end{proof}

To conclude, we remark that Theorem \ref{th:ContIsom} is also valid in the case of continuous dual information algebras: the maps $\psi \mapsto [(\psi,x)]_\sigma$ and $\phi \mapsto ([\phi]_\sigma,x)$ are continuous.

This gives us the full duality between labeled and domain-free continuous information algebras. However, the definition of a continuous labeled information algebra makes also sense without the assumption of a top element in $Q$. It remains so far an open question, whether a labeled continuous information algebra $(\Psi,D)$ can be extended to a labeled continuous information algebra with a top domain. The problem is the extension of the compatibility condition to the new top domain.

%% file: chapter9.tex

\chapter{Uncertain information} \label{sec:UncertainInf}


\section{Simple random maps} \label{subsec:SimpleRanMaps}

In practice it can not be excluded that \textit{contradictory information} is asserted. Then at least one of these assertions must be wrong. This immediately leads to the idea that information may be uncertain, at least in the sense that its assertion may be wrong. For instance, if the source of an information is a witness, an expert or a sensor, there is always the possibility that the witness  lies, the expert errs or that the sensor is faulty. More generally, the truth of a piece of information may depend on certain assumptions whose validity is uncertain. Turned the other way round: Assuming the validity of certain assumptions out of a set of possible assumptions, certain pieces of information may be asserted. The uncertainty of the information stems in this view from the uncertainty about which assumption is valid. Also different assumptions may have different likelihood or probabilities to be valid. Viewed from this angle, \textit{uncertain information} is represented by a map from a probability space into an information algebra.

Given such a map, for any piece of information in the information algebra, or more generally each consistent system of information in its ideal completion, the assumptions \textit{supporting} the information considered can be determined: These are all the assumptions whose validity entails the information. The probability of the assumptions supporting a piece of information measures the degree of support for it. Here enters the question of the \textit{measurability} of the support. To overcome the restrictions imposed by measurability considerations, allocations of probability in the probability algebra associated with the probability space of assumptions can be considered \cite{kappos69,shafer73}. 
 
Maps representing uncertain information inherit the structure of an information algebra from their range. Uncertain information thus still is in this sense \textit{information}. In many cases, \textit{finite} uncertain information is in a natural way to be defined, which turns these algebras of uncertain information into compact information algebras.

This concept of uncertain information has its roots in the \textit{theory of hints} \cite{kohlasmonney95} which in turn is based on Dempster's multivalued mappings \cite{dempster67a}. However, whereas Dempster derives probability bounds from these multivalued mappings, the semantics of the theory of hints is in the spirit of assumption-based reasoning as sketched above. Seen from the point of view of information algebra, hints are mappings into a subset-algebra. The theory can also be given a logical flavour. It may for instance be combined with propositional logic \cite{haennikohlas00,kohlas03}. Since this approach combines logic for deduction of arguments with probability to evaluate likelihood or reliabiility of arguments, we speak also of \textit{probabilistic argumentation systems}. A more abstract presentation of this point of view is given in \cite{kohlas03b}. 

Dempster's approach to multivalued mappings was given by Shafer a more epistemological flavor \cite{shafer76}. The primary object in this view is the \textit{belief function} which corresponds formally to our degree of support and leads to an allocation of probability as hinted above \cite{shafer73}. Therefore, in the spirit of Shafer, we study allocations of belief and show that they too lead to information algebras (Section \ref{subsec:AllocProb}). In particular, we study how these allocations of probabilities relate to the mappings representing uncertain information. 

We start with simple random maps. Consider a domain-free information algebra $(\Phi,\cdot,0,1;E)$ with $E = \{\epsilon_x:x \in Q\}$.  We do however not necessarily assume the support axiom. Let $\Omega$ be a set whose elements represent different possible assumptions. In applications, $\Omega$ often will be a finite set. But we drop this requirement for the sake of generality. In order to introduce probability, we assume $(\Omega,\mathcal{A},P)$ to be a \textit{probability space} with $\mathcal{A}$ a $\sigma$-algebra of subsets of $\Omega$ and $P$ a probability measure on $\mathcal{A}$. Uncertain information will be represented by a map $\Delta$ from $\Omega$ to $\Phi$. The idea is that $\Delta(\omega) \in \Phi$ represents the piece of information valid, provided assumption $\omega \in \Omega$ is valid. In order to simplify, and for considerations of measurability, which will be dropped later, we restrict in a first step the maps to be considered. Let $\mathcal{B} = \{B_{1},\ldots,B_{n}\}$ be any finite partition of $\Omega$, whose blocks $B_{i}$ belong all to $\mathcal{A}$. A mapping $\Delta : \Omega \rightarrow \Phi$, such that $\Delta(\omega)$ is constant for all $\omega$ of a block $B_{i}$,
\begin{eqnarray}
\Delta(\omega) = \psi_{i}, \textrm{ for all}\ \omega \in B_{i},
\nonumber
\end{eqnarray}
is called a \textit{simple random variable} in $\Phi$. 

Denote the family of all simple random variables by $\mathcal{R}_{s}$. These maps inherit the operations of the information algebra:
\begin{enumerate}
\item \textit{Combination:} Let $\Delta_{1}$ and $\Delta_{2}$ be simple random variables  in $(\Phi,\cdot0,1;E)$. Then $\Delta_{1} \cdot\Delta_{2}$ is defined pointwise by
\begin{eqnarray}
(\Delta_{1} \cdot \Delta_{2})(\omega) = \Delta_{1}(\omega) \cdot \Delta_{2}(\omega),
\nonumber
\end{eqnarray}
where on the right combination is in $\Phi$.
\item \textit{Extraction:} Let $\Delta$ be a simple random variable in $(\Phi,\cdot0,1;E)$. Then define $\epsilon_x(\Delta)$ for $x \in D$ by 
\begin{eqnarray}
\epsilon_x(\Delta)(\omega) = \epsilon_x(\Delta(\omega)),
\nonumber
\end{eqnarray}
where on the right extraction takes place in $\Phi$.
\end{enumerate}
We have to verify that the maps so defined are still simple random variables. Let $\mathcal{B}_{1}$ and $\mathcal{B}_{2}$ be the finite partitions of $\Omega$ associated with $\Delta_{1}$ and $\Delta_{2}$ respectively. Then $\mathcal{B} = \mathcal{B}_{1} \vee \mathcal{B}_{2}$ is defined as  the partition of $\Omega$ whose blocks are the pairwise intersections of blocks from $\mathcal{B}_{1}$ and $\mathcal{B}_{2}$ (as always in this text, see Section \ref{sec:SetAlg}). Clearly, the map $\Delta_{1} \cdot \Delta_{2}$ is constant on each block of $\mathcal{B}$, hence a simple random variable. If further $\Delta$ is defined relative to a partition $\mathcal{B}$ of $\Omega$, then $\epsilon_x(\Delta)$ is also constant on the blocks of $\mathcal{B}$, hence also a simple random variable. Obviously, $(\mathcal{R}_{s},\cdot,0,1;E)$ (where by abuse of notation $E$ is here the set of extraction operators $\epsilon_x$ on $\mathcal{R}_s$) becomes a domain-free information algebra  with these operations. The null element is the simple random variable $0$ defined by $0(\omega) = 0$, the unit element the simple random variable $1$ defined by $1(\omega) = 1$ for all $\omega \in \Omega$. Furthermore, for every $\phi \in \Phi$ the map $D_{\phi}(\omega) = \phi$, for all $\omega \in \Omega$, is a simple random variable. By the mapping $\phi \mapsto D_{\phi}$ the information algebra $(\Phi,\cdot0,1;E)$ is embedded in the information algebra $(\mathcal{R}_{s},\cdot,0,1;E)$.

Note that the partial order in $\mathcal{R}_{s}$ is also defined point-wise such that $\Delta_{1} \leq \Delta_{2}$ in $\mathcal{R}_{s}$ if and only if, $\Delta_{1}(\omega) \leq \Delta_{2}(\omega)$ for all $\omega \in \Omega$. 

There are two important special classes of simple random variables: If for a random variable $\Delta$ defined relative to a partition $\mathcal{B} = \{B_{1},\ldots,B_{n}\}$ it holds that $\phi_{i} \not= \phi_{j}$ for $i \not= j$, the variable is called \textit{canonical}. It is a simple matter to transform any random variable $\Delta$ into an associated canonical one: Take the union of all blocks $B_{i} \in \mathcal{B}$ with identical values $\phi_{i}$.  This yields a new partition $\mathcal{B}'$ of $\Omega$. Define $\Delta'(\omega) = \Delta(\omega)$. Then $\Delta'$ is the \textit{canonical version} of $\Delta$ and we write $\Delta' = \Delta^{\rightarrow}$. We may consider the set of \textit{canonical random variables}, $\mathcal{R}_{s,c}$, and define between elements of this set combination and extraction as follows:
\begin{eqnarray}
\Delta_{1} \cdot_{c} \Delta_{2} &= &(\Delta_{1} \cdot \Delta_{2})^{\rightarrow},
\nonumber \\
\epsilon_{x,c}(\Delta) &= &(\epsilon_x(\Delta))^{\rightarrow}.
\nonumber
\end{eqnarray}
Then $(\mathcal{R}_{s,c},\cdot,0,1;E)$ is still an information algebra under these modified operations. We remark also that $(\Delta_{1} \cdot \Delta_{2})^{\rightarrow} = (\Delta_{1}^{\rightarrow} \cdot \Delta_{2}^{\rightarrow})^{\rightarrow}$ and $(\epsilon_x(\Delta))^{\rightarrow} = (\epsilon_x(\Delta^{\rightarrow}))^{\rightarrow}$. In fact, $(\mathcal{R}_{s,c},\cdot,0,1;E)$ is the quotient algebra of $(\mathcal{R}_{s},\cdot,0,1;E)$ relative to the congruence $\Delta_{1} \equiv \Delta_{2}$, if $\Delta_{1}^{\rightarrow} = \Delta_{2}^{\rightarrow}$. 

Secondly, if $\Delta(\omega) = 0$ with probability zero, then $\Delta$ is called \textit{normalised}. We can associate a normalised simple random variables $\Delta^{\downarrow}$ with any simple random variable $\Delta$ provided $\Delta(\omega) \not=  0$ occurs with a positive probability. In fact, let $\Omega^{\downarrow} = \{\omega \in \Omega: \Delta(\omega) \not=  0\}$. This is a measurable set with probability $P(\Omega^{\downarrow}) = 1 - P\{\omega \in \Omega: \Delta(\omega) =  0\} > 0$. We consider then  the new probability space $(\Omega, \mathcal{A},P')$, where $P'$ is the conditional probability measure on $\mathcal{A}$ defined by 
\begin{eqnarray} \label{eq:NormProb}
P'(A) 
= \frac{P(A \cap \Omega^{\downarrow})}{P( \Omega^{\downarrow})},
\end{eqnarray}
if $A \cap \Omega^{\downarrow} \not= \emptyset$ and $P'(A) = 0$, otherwise. On this new probability space define $\Delta^{\downarrow}(\omega) = \Delta(\omega)$. Clearly, it holds that $(\Delta^{\rightarrow})^{\downarrow} = (\Delta^{\downarrow})^{\rightarrow}$. 

The idea behind normalisation becomes clear, when we consider \textit{combination} of random variables: Each of two (normalised) random variables $\Delta_{0}$ and $\Delta_{2}$ represents some (uncertain) information with the following interpretation: One of the $\omega \in \Omega$ must be the (unknown) correct assumption. However, if $\omega$ happens to be the correct assumption, then under the first random variable $\Delta_{1}(\omega)$ can be asserted, and under the second variable $\Delta_{2}(\omega)$. Thus, together, still under the assumption $\omega$, $\Delta_{1}(\omega) \cdot \Delta_{2}(\omega)$ can be asserted. However, it is possible that $\Delta_{1}(\omega) \cdot \Delta_{2}(\omega) = 0$, even if both $\Delta_{1}$ and $\Delta_{2}$ are \textit{normalised}. But the element $0$ represents a \textit{contradiction}. Thus in view of the information given by the variables $\Delta_{0}$ and $\Delta_{2}$, the assumption $\omega$ can not be valid, since it leads to a contradiction;  it can (and must) be excluded. This amounts to normalise the random variable $\Delta_{1} \cdot \Delta_{2}$, by excluding all $\omega \in \Omega$ for which the combination results in a contradiction, and then to condition (i.e. normalise) the probability on non-contradictory assumptions. We refer to \cite{kohlasmonney95,haennikohlas00} for a discussion and further justification of these issues.

Two partitions $\mathcal{B}_{1}$ and $\mathcal{B}_{2}$ of $\Omega$ are called \textit{independent}, if $B_{1,i} \cap B_{2,j} \not= \emptyset$ for all blocks $B_{1,i} \in \mathcal{B}_{1}$ and $B_{2,j} \in \mathcal{B}_{2}$. If furthermore $P(B_{1,i} \cap B_{2,j}) = P(B_{1,i}) \cdot P(B_{2,j})$ for all these pairs of blocks, then the two partitions $\mathcal{B}_{1}$ and $\mathcal{B}_{2}$ are called \textit{stochastically independent}. In addition, if $\Delta_{1}$ and $\Delta_{2}$ are two simple random variables defined on these two partitions respectively, then these random variables are called \textit{stochastically independent} too. Note that if $\Delta_{1}$ and $\Delta_{2}$ are stochastically independent, then their canonical versions $\Delta_{1}^{\rightarrow}$ and $\Delta_{2}^{\rightarrow}$ are also stochastically independent. 

We now turn to the study of the \textit{probability distribution} of simple random variables. The starting point is the following question: Given a simple random variable $\Delta$ in an information algebra $(\Phi,\cdot,0,1;E)$, and an element $\phi \in \Phi$, under what assumptions can the information represented by $\phi$ be asserted to hold? And how likely is it, that these assumptions are valid? 

If $\omega \in \Omega$ is an assumption such that $\Delta(\omega) \geq \phi$, then $\phi$ is part of $\Delta(\omega)$, or in other words, $\Delta(\omega)$ implies $\phi$. In this case we may say that $\omega$ is an assumption \textit{supporting} $\phi$, in view of the information conveyed by  $\Delta$. Therefore we define for every $\phi \in \Phi$ the set 
\begin{eqnarray}
qs_{\Delta}(\phi) = \{\omega \in \Omega: \phi \leq \Delta(\omega)\}
\nonumber
\end{eqnarray}
of assumptions supporting $\phi$. However, if $\Delta(\omega) = 0$, then $\omega$ is supporting every $\phi \in \Phi$, since $\phi \leq 0$. The null element $0$ represents the \textit{contradiction}, which implies everything. In a consistent theory, contradictions must be excluded. Thus, we conclude that assumptions such that $\Delta(\omega) = 0$ are not really possible assumptions and must be excluded. Let
\begin{eqnarray}
qs_{\Delta}(0) = \{\omega \in \Omega: \Delta(\omega) = 0\}.
\nonumber
\end{eqnarray}
We assume that $qs_{\Delta}(0)$ is not equal to $\Omega$; otherwise $\Delta$ is representing fully contradictory ``information''. In other words, we assume that proper information is never fully contradictory. If we eliminate the contradictory assumptions from $qs(\phi)$, we obtain the \textit{support set}
\begin{eqnarray}
s_{\Delta}(\phi) = \{\omega \in \Omega: \phi \leq \Delta(\omega) \not= 0\}
= qs_{\Delta}(\phi) - qs_{\Delta}(0).
\nonumber
\end{eqnarray}
of $\phi$, which is the set of assumptions properly supporting $\phi$ and the mapping $s_{\Delta}: \Phi \rightarrow \mathbb{P}(\Omega)$ is called the \textit{allocation of support} induced by $\Delta$. The set $qs(\phi)$ is called the \textit{quasi-support set} to underline that it contains contradictory assumptions. This set has little interest from a semantic point of view, but it is useful for technical and especially for computational purposes. These concepts capture the essence of probabilistic assumption-based reasoning in information algebras as discussed in more detail in \cite{kohlasmonney95,haennikohlas00,kohlas03} in a less general setting.

Here are the basic properties of allocations of support:

\begin{theorem} \label{th:AllolcSupp}
If $\Delta$ is a simple random variable on an information algebra $(\Phi,\cdot,0,1;E)$, then the following holds for the associated allocations of support $qs_{\Delta}$ and $s_{\Delta}$:
\begin{enumerate}
\item $qs_{\Delta}(1) = \Omega$, $s(0) = \emptyset$.
\item If $\Delta$ is normalised, then $qs_{\Delta} = s_{\Delta}$ and $qs_{\Delta}(0) = \emptyset$.
\item For any pair $\phi,\psi \in \Phi$, 
\begin{eqnarray}
qs_{\Delta}(\phi \cdot \psi) &= &qs_{\Delta}(\phi) \cap qs_{\Delta}(\psi),
\nonumber \\
s_{\Delta}(\phi \cdot \psi) &= &s_{\Delta}(\phi) \cap s_{\Delta}(\psi).
\nonumber
\end{eqnarray}
\end{enumerate}
\end{theorem}

\begin{proof}
(1) and (2) follow immediately from the definition of the allocation of support. (3) follows since $\phi \cdot \psi \leq \Delta(\omega)$ if and only if $\phi \leq \Delta(\omega)$ and $\psi \leq \Delta(\omega)$.
\end{proof}

Knowing assumptions supporting a hypothesis $\psi$ is already interesting and important. It is the part logic can provide. On top of this, it is important to know how \textit{likely} it is that a supporting assumption is valid. This is the part added by probability. If we know or may assume that the information is consistent, then we should condition the original probability measure $P$ in $\Omega$ on the event $qs_{\Delta}^{c}(0)$. This leads then to the probability space $(qs_{\Delta}^{c}(0),\mathcal{A} \cap qs_{\Delta}^{c}(0),P')$, where $P'(A) = P(A) / P(qs_{\Delta}^{c}(0))$. The likelihood of supporting assumptions for $\phi \in \Phi$ can then be measured by 
\begin{eqnarray}
sp_{\Delta}(\phi) = P'(s_{\Delta}(\phi)).
\nonumber
\end{eqnarray}
The value $sp_{\Delta}(\phi)$ is called the \textit{degree of support} of $\phi$ associated with the random variable $\Delta$. The function $sp: \Phi \rightarrow [0,1]$ is called the \textit{support function} of $\Delta$. It corresponds to the concept of a \textit{distribution function} of ordinary random variables. 

It is for technical reasons convenient to define the \textit{degree of quasi-support} 
\begin{eqnarray}
qsp_{\Delta}(\phi) = P(qs_{\Delta}(\phi)).
\nonumber
\end{eqnarray}
Then, the degree of support can also be expressed in terms of degrees of quasi-support
\begin{eqnarray}
sp_{\Delta}(\phi) = \frac{qsp_{\Delta}(\phi) - qsp(0)}{1 - qsp_{\Delta}(0) }.
\nonumber
\end{eqnarray}
This is the form which is usually used in applications \cite{haennikohlas00}. 

In another consideration, we can also ask for assumptions $\omega \in \Omega$, under which $\Delta$ shows $\phi$ to be possible, that is, not excluded, although not necessarily supported. If $\Delta(\omega)$ is such that combined with $\phi$ it leads to a contradiction, i.e. if $\Delta(\omega)  \cdot \phi = 0$, then under $\omega$ the information $\phi$ is excluded by a consistency consideration as above. So we define the set 
\begin{eqnarray}
p_{\Delta}(\phi) = \{\omega \in \Omega: \Delta(\omega)  \cdot \phi \not= 0\}.
\nonumber
\end{eqnarray}
This is the set of assumptions under which $\phi$ is not excluded, hence can be considered as possible. Therefore we call it the \textit{possibility set} of $\phi$. Note that $p_{\Delta}(\phi) \subseteq qs^{c}_{\Delta}(0)$. We can then define the \textit{degree of possibility}, also sometimes called \textit{degree of plausibility} (e.g. in \cite{shafer76}), by
\begin{eqnarray}
pl_{\Delta}(\phi) = P'(p_{\Delta}(\phi)).
\nonumber
\end{eqnarray}
If $\omega \in qs^{c}_{\Delta}(0) - p_{\Delta}(\phi)$, then, under this assumption, $\phi$ is impossible, that is contradictory with $\Delta(\omega)$. So the set $qs^{c}_{\Delta}(0) - p_{\Delta}(\phi)$ contains arguments \textit{against} $\phi$ and
\begin{eqnarray}
do_{\Delta}(\phi) = P'(qs^{c}_{\Delta}(0) - p_{\Delta}(\phi)) = 1 - pl_{\Delta}(\phi).
\nonumber
\end{eqnarray}
can be called the \textit{degree of doubt} in $\phi$. Note that $s_{\Delta}(\phi) \subseteq p_{\Delta}(\phi)$ since $\phi \leq \Delta(\omega) \not= 0$ implies $\phi \cdot \Delta(\omega) = \Delta(\omega) \not= 0$. Hence, we see that for all $\phi \in \Phi$ we have that $sp_{\Delta}(\phi) \leq pl_{\Delta}(\phi)$. These consideration put simple random variables in the realm of the so-called Dempster-Shafer theory \cite{dempster67b,shafer76}, although the latter is based on simple sets (or set algebras) and not on general information algebras.. 

To underline this further, consider for a simple random variable $\Delta$ with possible values $\phi_{1},\ldots,\phi_{n}$ the probabilities
\begin{eqnarray}
m(\phi_{i}) = \sum_{j:\phi_{j} = \phi_{i}} P(B_{j}).
\nonumber
\end{eqnarray}
Note that $m(\phi_{i}) = P(B_{i})$, if the random variable $\Delta$ is canonical. Remark also that 
\begin{eqnarray}
\sum_{i=1}^{n} m(\phi_{i}) = 1.
\nonumber
\end{eqnarray}
Such a finite collection of probabilities $m(\phi_{i})$ summing up to one for $i = 1,\ldots,n$ is called a \textit{basic probability assignment (bpa)} in $\Phi$. Since $qs_{\Delta}(\phi) = \cup_{\phi \leq \phi_{i}} B_{i}$ and $p_{\Delta}(\phi) = \cup_{\phi \cdot \phi_{i} \not= 0} B_{i}$, we see that
\begin{eqnarray}
qs_{\Delta}(\phi) = \sum_{\phi \leq \phi_{i}} m(\phi_{i}), \quad
pl_{\Delta}(\phi) = \sum_{\phi \cdot \phi_{i} \not= 0} m(\phi_{i}).
\nonumber
\end{eqnarray}
So, the bpa of a simple random variable determines its degrees of support and plausibilities. In \cite{shafer76}, support function are called \textit{belief functions}. Furthermore, if $\Delta_{1}$ and $\Delta_{2}$ are two \textit{stochastically independent} simple random variables with possible values $\phi_{1,1},\ldots,\phi_{1,n}$ and $\phi_{2,1},\ldots,\phi_{2,m}$, then the possible values of the combined random variable $\Delta = \Delta_{1} \cdot \Delta_{2}$ are $\phi_{k}$, where each $\phi_{k}$ is equal to a combination $\phi_{1,i} \cdot \phi_{2,j}$. Therefore, the bpa of the combined variable $\Delta$ is
\begin{eqnarray}
m(\phi_{k}) = \sum_{\phi_{1,i} \cdot \phi_{2,j} = \phi_{k}} m_{1}(\phi_{1,i}) \cdot m_{1}(\phi_{2,j}).
\nonumber
\end{eqnarray}
If only \textit{normalised} random variables are considered, then the combined variable $\Delta$ is to be normalised. Then, if 
\begin{eqnarray}
m(0) = \sum_{\phi_{1,i} \cdot \phi_{2,j} = 0} m_{1}(\phi_{1,i}) \cdot m_{0}(\phi_{2,j}) < 1,
\nonumber
\end{eqnarray}
we obtain the normalised bpa of $\Delta^{\downarrow}$ as
\begin{eqnarray} \label{eq:DempstRule}
m^{\downarrow}(\phi_{k}) = \frac{\sum_{\phi_{1,i} \cdot \phi_{2,j} = \phi_{k}} m_{1}(\phi_{1,i}) \cdot m_{1}(\phi_{2,j})}{1 - m(0)}
\end{eqnarray}
So, the bpa are also sufficient to compute the bpa of the combination of stochastically independent pieces of  uncertain information. This has been proposed in a setting of set algebras in \cite{dempster67a} and the formula (\ref{eq:DempstRule}) is therefore also called \textit{Dempster's rule}. \cite{shafer76} took up Dempster's theory and proposed ``A Mathematical Theory of Evidence'' where bpa and Dempster's rule play an import role. In both theories the concept of a bpa is central. Although Dempster's and Shafer's interpretation of the theory are not quite the same, one speaks often of the \textit{Dempster-Shafer Theory}. At least the underlying mathematics in both views are identical. We shall argue in this chapter that our present theory is a natural generalisation of Dempster-Shafer theory which was confined essentially to finite subset algebras and simple random variables (in our terminology). However bpa can no more play the same basic role relative to general information algebras and general random maps as in classical Dempster-Shafer theory, since bpa works only of simple random variables, but not for more general uncertain information. Also, the full flavour of the duality relation between support and plausibility as described in Dempster-Sahfer theory is deployed only in the case of Boolean information algebras (Section \ref{subsec:BooleanCase}).


\section{Random maps} \label{subsec:RanMaps}

When we want to go beyond simple random mappings, there are several ways to do this. The most radical one is to consider any mapping $\Gamma: \Omega \rightarrow \Psi$ from a probability space $(\Omega,\mathcal{A},P)$  into an  information algebra $(\Phi,\cdot,0,1;E)$ with $E = \{\epsilon_x:x \in Q\}$ or may be even its ideal completion. Let's call such maps \textit{random mappings}. As before, in the case of simple random variables, we may define the operations of combination and extraction between random mappings point-wise in $(\Phi,\cdot,0,1;E)$:
\begin{enumerate}
\item \textit{Combination:} Let $\Gamma_{1}$ and $\Gamma_{2}$ be two random mappings into $\Phi$, then $\Gamma_{1} \cdot \Gamma_{2}$ is the random mapping defined by
\begin{eqnarray}
(\Gamma_{1} \cdot \Gamma_{2})(\omega) = \Gamma_{1}(\omega) \cdot \Gamma_{2}(\omega).
\end{eqnarray}
\item\textit{Extraction:} Let $\Gamma$ be a random mapping into $\Phi$ and $x \in Q$, then $\epsilon_x(\Gamma)$ is the random mapping defined by
\begin{eqnarray}
\epsilon_x(\Gamma)(\omega) = \epsilon_x(\Gamma(\omega)).
\end{eqnarray}
\end{enumerate}
For a fixed probability space $(\Omega,\mathcal{A},P)$, let $\mathcal{R}_{\Phi}$ denote the set of all random mappings into $\Phi$. With the two operations defined above, $(\mathcal{R}_{\Phi},\cdot,0,1;E\})$, where here $E$ is the set of extraction operators of random maps, becomes a domain-free information algebra (excluding the Support Axiom). The mapping $1(\omega) = 1$ for all $\omega \in \Omega$ is the neutral element of combination; the map $0(\omega) = 0$ the null element. It is obvious that $\Gamma' \leq \Gamma$ if and only if $\Gamma'(\omega) \leq \Gamma(\omega)$ for all $\omega \in \Omega$. 

Consider the ideal completion $I_{\mathcal{R}_{\Phi}}$ of the information algebra of random mappings. The elements $\Gamma$ of $I_{\mathcal{R}_{\Phi}}$ are ideals of random maps $\Delta : \Omega \rightarrow \Phi$. The sets $\{\Delta(\omega):\Delta \in \Gamma\}$ are then ideals in $\Phi$ for all 
$\omega \in \Omega$. In fact, if $\Delta_1,\Delta_2 \in \Gamma$, $\Delta_1(\omega) \cdot \Delta_2(\omega) = (\Delta_1 \cdot \Delta_2)(\omega)$ and $\Delta_1 \cdot \Delta_2 \in \Gamma$, so that the set $\{\Delta(\omega):\Delta \in \Gamma\}$ is closed under combination. Further, if $\Delta_1(\omega) \leq \Delta_2(\omega)$ for a random map $\Delta_2 \in \Gamma$ and a $\omega \in \Omega$, define the map $\Delta'_1$ by $\Delta'_1(\omega) = \Delta_1(\omega)$ and $\Delta'_1(\omega') = \Delta_2(\omega')$ for $\omega' \not= \omega$. Then $\Delta'_1 \leq \Delta_2$ and therefore $\Delta'_1 \in \Gamma$, hence $\Delta_1(\omega)$ belongs to the set $\{\Delta(\omega):\Delta \in \Gamma\}$, which is thus downwards closed, hence an ideal in $\Phi$. It follows that the elements $\Gamma$ of $I_{\mathcal{R}_{\Phi}}$ are associated with random maps $\Gamma : \Omega \rightarrow I_\Phi$.

As $\Phi$ is embedded in $I_\Phi$, so is $\mathcal{R}_\Phi$ in $I_{\mathcal{R}_\Phi}$. As usual, we may consider $\Phi$ a subalgebra of $I_\Phi$ and $\mathcal{R}_\Phi$ a subalgebra $I_{\mathcal{R}_\Phi}$. Any element $\Gamma \in I_{\mathcal{R}_{\Phi}}$ may in this view be represented as the supremum of all random maps dominated by $\Gamma$,
\begin{eqnarray}
\Gamma = \bigvee\{\Delta : \Delta \in \mathcal{R}_{\Phi}, \Delta \leq \Gamma\},
\nonumber
\end{eqnarray}
Obviously, we also have  in the ideal completion $I_{\Phi}$ of $\Phi$,
\begin{eqnarray}
\Gamma(\omega) = \bigvee\{\Delta(\omega) : \Delta \in \mathcal{R}_{\Phi}, \Delta \leq \Gamma\} = \bigvee\{\phi : \phi \in \Phi, \phi \leq \Gamma(\omega)\}
\nonumber
\end{eqnarray}
for all $\omega \in \Omega$. This shows that $\mathcal{R}_{I_\Phi}$ is essentially identical to the ideal completion $I_{\mathcal{R}_\Phi}$ of the algebra $\mathcal{R}_\Phi$.

As in the case of simple random variables we may define the allocation of support $s_{\Gamma}$ of a random mapping by
\begin{eqnarray} \label{eq:DefAllocOfSupp}
s_{\Gamma}(\psi) = \{\omega \in \Omega: \psi \leq \Gamma(\omega)\}. 
\end{eqnarray}
We do not any more distinguish here between the semantic categories of support and quasi-support as before for simple random variables and speak simply of support, even though (\ref{eq:DefAllocOfSupp}) is strictly speaking a quasi-support.

This support, as defined in (\ref{eq:DefAllocOfSupp}), has the same properties as the support of simple random variables, in particular, as in Theorem \ref{th:AllolcSupp}, $s_{\Gamma}(1) = \Omega$ and $s_{\Gamma}(\phi \cdot \psi) = s_{\Gamma}(\phi) \cap s_{\Gamma}(\psi)$. Again, as before, with simple random variables, we may try to define the degree of support induced by a random mapping $\Gamma$ of a piece of information $\psi$ by 
\begin{eqnarray} \label{eq:SpFct}
sp_{\Gamma}(\psi) = P(s_{\Gamma}(\psi)). 
\end{eqnarray}
This probability is however only defined if $s_{\Gamma}(\psi) \in \mathcal{A}$. There is no guarantee that this holds in general. The only element which we know for sure to be measurable is $s_{\Gamma}(1) = \Omega$. A simple way out of this problem would be to restrict random mappings to mappings $\Gamma$ for which $s_{\Gamma}(\psi) \in \mathcal{A}$ for all $\psi \in \Psi$ or even for all elements of the ideal completion $I_{\Psi}$. However, there is a priori no reason why we should restrict ourselves exactly to those mappings. Therefore we prefer other, more rational approaches to overcome the difficulty of an only partial definition of degrees of support. Here we propose a first solution. Later we present some alternatives.

\cite{shafer79} advocates the use of \textit{probability algebras} instead of probability spaces as a natural framework  for studying belief functions. Since degrees of support are similar to belief functions, we can adapt this idea here. First, we introduce the probability algebra associated with a probability space \cite{kappos69}. Let $\mathcal{J}$ be the $\sigma$-ideal of $P$-null sets in the $\sigma$-algebra $\mathcal{A}$ of the probability space. Two sets $A',A'' \in \mathcal{A}$ are equivalent modulo $\mathcal{J}$, if $A' - A'' \in \mathcal{J}$ and  $A''- A' \in \mathcal{J}$. This means that the two sets have the same probability measure $P(A') = P(A'')$. This equivalence is a congruence in the Boolean algebra $\mathcal{A}$. Hence the quotient algebra $\mathcal{B} = \mathcal{A}/\mathcal{J}$ is a Boolean $\sigma$-algebra too. If $[A]$ denotes the equivalence class of $A$, then, for any \textit{countable} family of sets $A_{i}$, $i \in I$,
\begin{eqnarray} \label{eq:ProjHomomorph}
[A]^{c} &= &[A^{c}],
\nonumber \\
\bigvee_{i \in I} [A_{i}] &= &\left[ \bigcup_{i \in I} A_{i} \right],
\nonumber \\
\bigwedge_{i \in I} [A_{i}] &= &\left[ \bigcap_{i \in I} A_{i} \right].
\end{eqnarray}
So $[A]$ defines a Boolean homomorphism from $\mathcal{A}$ onto $\mathcal{B}$, called \textit{projection}. We denote $[\Omega]$ by $\top$ and $[\emptyset]$ by $\bot$. These are of course the top and bottom elements of $\mathcal{B}$. Now, as is well known, $\mathcal{B}$ has some further important properties (see \cite{halmos63}): It satisfies the \textit{countable chain condition}, which means that any  family of disjoint elements of $\mathcal{B}$ is countable. Further, any Boolean algebra $\mathcal{B}$ satisfying the countable chain condition is \textit{complete}. That is, any subset $E \subseteq \mathcal{B}$ has a supremum $\bigvee E$ and an infimum $\bigwedge E$ in $\mathcal{B}$. Furthermore, the countable chain condition implies also that there is always a \textit{countable} subset $D$ of $E$ with the same supremum and infimum, i.e. $\bigvee D = \bigvee E$ and $\bigwedge D = \bigwedge E$. We refer to \cite{halmos63} for these results. Finally, by $\mu([A]) = P(A)$ a \textit{normalised, positive measure} $\mu$ is defined on $\mathcal{B}$. Positive means here that $\mu(b) = 0$ implies $b = \bot$. A pair $(\mathcal{B},\mu)$ of a Boolean $\sigma$-algebra $\mathcal{B}$, satisfying the countable chain condition, and a normalised, positive measure $\mu$ on it, is called a \textit{probability algebra}.

We use now this construction of a probability algebra from a probability space to extend
the definition of the degrees of support $s_{\Gamma}$ beyond elements $\psi$ for which
 $s_{\Gamma}(\psi)$ are measurable. Even if $s_{\Gamma}(\psi)$ is not measurable, any $A \in
\mathcal{A}$ such that $A \subseteq s_{\Gamma}(\psi)$ represents an argument for $\psi$, that is a set of assumptions which supports $\psi$. To exploit this remark, define for every set $H \in \mathcal{P}(\Omega)$
\begin{eqnarray} \label{rho0}
\rho_{0}(H) &=& \bigvee\{[A]: A \subseteq H, A \in
\mathcal{A}\}.
\end{eqnarray}
This mapping  has interesting properties as the following theorem shows.

\begin{theorem} \label{th:ExtOfProj}
The application $\rho_{0}: \mathcal{P}(\Omega) \rightarrow \mathcal{A}/\mathcal{J}$ as
defined in (\ref{rho0}) has the  following properties:
\begin{eqnarray}
\rho_{0}(\Omega) &= &\top,
\nonumber \\
\rho_{0}(\emptyset) &= &\bot,
\nonumber \\
\rho_{0}\left(\bigcap_{i \in I} H_{i}\right) 
&= &\bigwedge_{i \in I}\rho_{0}(H_{i}).
\end{eqnarray}
if $\{H_{i}, i \in I\}$ is a countable family of subsets of $\Omega$.
\end{theorem}

\begin{proof}
Clearly, $\rho_{0}(\Omega) = [\Omega] = \top \in \mathcal{A}/\mathcal{J}$. Similarly,
$\rho_{0}(\emptyset) = [\emptyset] = \bot \in \mathcal{A}/\mathcal{J}$.

In order to prove the remaining identity, let $H_{i}, i \in I$ be a countable family of
subsets of $\Omega$. For every index $i$, there is a countable family of sets $H'_{j} \in
\mathcal{A}$ such that $H'_{j} \subseteq H_{i}$ and $\rho_{0}(H_{i}) = \bigvee [H'_{j}] = [\bigcup H'_{j}]$ since $\mathcal{A}/\mathcal{J}$ satisfies the countable chain condition. Take $A_{i} = \bigcup H'_{j}$. Then $A_{i} \subseteq H_{i}$, $A_{i} \in \mathcal{A}$ and $P(A_{i}) =
\mu(\rho_{0}(H_{i}))$. Define $A = \bigcap_{i \in I} A_{i} \in \mathcal{A}$. It follows that $A
\subseteq \bigcap_{i \in I} H_{i}$ and, because the projection is a $\sigma$-homomorphism, we
obtain $[A] = \bigwedge_{i \in I} [A_{i}] = \bigwedge_{i \in I}
\rho_{0}(H_{i})$.

We are going to show now that $[A] = \rho_{0}(\bigcap_{i \in I} H_{i})$ which proves then
the theorem. For this, it is sufficient to show that $P(A) = \mu(\rho_{0}(\bigcap_{i \in I}
H_{i}))$ because $P(A) = \mu([A])$ and $A \subseteq \bigcap H_{i}$, hence $[A] \leq
\rho_{0}(\bigcap H_{i})$. Therefore, if $\mu([A]) = \mu(\rho_{0}(\bigcap H_{i}))$ we must well
have $[A] = \rho_{0}(\bigcap H_{i})$, since $\mu$ is positive.

Now, clearly $P(A) \leq \mu(\rho_{0}(\bigcap H_{i}))$. As above, we conclude that there is an
$A' \in \mathcal{A}, A' \subseteq \bigcap H_{i}$ such that $P(A') = \mu(\rho_{0}(\bigcap H_{i}))$. Further, 
$A' \cup (A-A') \subseteq \bigcap H_{i}$ implies that $P(A' \cup (A-A')) = P(A')$, hence
$P(A-A') = 0$. Define $A'_{i} = A_{i} \cup (A-A') \subseteq  H_{i}$. Then $A_{i} - A'_{i} =
\emptyset$ and therefore,
\begin{eqnarray}
\mu(\rho_{0}(H_{i})) &=&P(A_{i}) \leq P(A'_{i}) = 
P(A_{i}) + P(A'_{i}-A_{i}) 
\nonumber \\
&\leq &\mu(\rho_{0}(H_{i})).
\end{eqnarray}
This implies that $P(A'_{i}-A_{i}) = 0$, therefore we have $[A_{i}] = [A'_{i}]$. Further 
\begin{eqnarray}
\bigcap A'_{i} &=& \bigcap(A_{i} \cup (A'-A)) = (A'-A) \cup (\bigcap A_{i})
\nonumber \\
&=& (A'-A) \cup A = A \cup A' = A' \cup (A-A').
\nonumber
\end{eqnarray}
But $\bigcap A'_{i}$ and $\bigcap A_{i}$ are equivalent,
since $[\bigcap A'_{i}] = \bigwedge [A'_{i}] = \bigwedge [A_{i}] = [\bigcap A_{i}]$. This implies
finally that $P(A) = P(\bigcap A_{i}) = P(\bigcap A'_{i}) = P(A') + P(A-A') = P(A') =
\mu(\rho_{0}(\bigcap H_{i}))$. This is what was to be proved.
\end{proof}

Take now  $\mathcal{B} = \mathcal{A}/\mathcal{J}$ and consider the probability algebra
$(\mathcal{B},\mu)$. Then we compose the allocation of support $s$ from $\Phi$ into the
power set $\mathcal{P}(\Omega)$ with the mapping $\rho_{0}$ from
$\mathcal{P}(\Omega)$ into $\mathcal{B}$ to a mapping $\rho = \rho_{0} \circ s: \Phi
\rightarrow \mathcal{B}$. Now we see that
\begin{eqnarray} \label{eq:AllocOfProb1}
\rho(1) &= &\rho_{0}(s(1)) = \rho_{0}(\Omega) = \top,
\nonumber \\[.5em]
\rho(\phi \cdot \psi) &= &\rho_{0}(s(\phi \cdot \psi)) 
= \rho_{0}(s(\phi) \cap s(\psi)) 
\nonumber \\[.5em]
&= &\rho_{0}(s(\phi)) \wedge \rho_{0}(s(\psi))
=\rho(\phi) \wedge \rho(\psi).
\end{eqnarray}
A mapping $\rho$ satisfying these two properties is called an \textit{allocation of probability (a.o.p)} on the information algebra $\Phi$. In fact, it allocates an element of the probability algebra $\mathcal{B}$ to any element of the algebra $\Psi$.  In this way, a random mapping $\Gamma$ leads always to an allocation of probability $\rho_{\Gamma} = \rho_{0} \circ s_{\Gamma}$, once a probability measure on the assumptions is introduced. 

In particular, we may now define the degree of support for any $\psi \in \Phi$ by
\begin{eqnarray} \label{eq:DefOfSupExt}
sp_{\Gamma}(\psi) &=& \mu(\rho_{\Gamma}(\psi)).
\end{eqnarray}
This extends the support function (\ref{eq:SpFct}) to all elements $\psi$ of $\Phi$. 

In this way, the degree of support $sp_{\Gamma}(\psi)$ is, according to (\ref{rho0}), equal to the probability of the supremum of all $[A]$, where $A$ is measurable and supports $\psi$. This can also be expressed in another way. In order to see this, we note an important property of probability algebras: Clearly $\mu(\bigwedge b_{i}) \leq \inf_{i} \mu(b_{i})$ and $\mu(\bigvee b_{i}) \geq \sup_{i} \mu(b_{i})$ holds for any family of elements $\{b_{i}\}$. But there are important cases where equality hold \cite{halmos63}. A subset $D$ of $\mathcal{B}$ is called {\em
downward (upward) directed}, \index{downward directed set}\index{upward directed
set}\index{directed set!upward}\index{directed set!downward} if for every pair $b',b'' \in
D$ there is an element $b \in D$ such that $b \leq b' \wedge b'' (b \geq b' \vee b'')$. 

\begin{lemma} \label{downward}
If $D$ is a downward (upward) directed subset of $\mathcal{B}$, then
\begin{eqnarray}
\mu(\bigwedge_{i \in D} b_{i}) = \inf_{i \in D} \mu(b_{i}), \quad
\left(\mu(\bigvee_{i \in D} b_{i}) = \sup_{i \in D} \mu(b_{i}) \right)
\end{eqnarray}
\end{lemma}

\begin{proof}
There is a countable subfamily of elements $c_{i} \in D$, $i=1,2,\ldots$, which have the same meet as $D$. Define $c'_{1} = c_{1}$ and select elements $c'_{i}$
in the downward directed set $D$ such that $c'_{2} \leq c'_{1} \wedge c_{2}$, $c'_{3} \leq
c'_{2} \wedge c_{3}, \ldots$. Then $c'_{1} \geq c'_{2} \geq c'_{3} \geq \ldots$ and this
sequence has still the same infimum. However, by the continuity of probability we have
\begin{eqnarray}
\mu(\bigwedge b_{i}) = \mu(\bigwedge c'_{i}) 
= \lim_{i \rightarrow \infty} \mu(c'_{i}) \geq \inf_{i} \mu(b_{i}).
\end{eqnarray}
But as $\mu(b_i) \geq \mu(\bigwedge b_{i})$, this implies $\mu(\bigwedge b_{i}) = \inf_{i} \mu(b_{i})$. The case of upwards directed
sets is proved in the same way.
\end{proof}

Note now that $\{[A]: A \subseteq H, A \in \mathcal{A}\}$ is an upward directed
family in $\mathcal{B}$. Therefore, according to Lemma~\ref{downward} we have
\begin{eqnarray} \label{eq:InnerProbExt}
sp_{\Gamma}(\psi) &= &\mu(\rho_{\Gamma}(\psi)) 
= \mu(\bigvee \{[A]: A \in \mathcal{A}, A \subseteq s_{\Gamma}(\psi) \})
\nonumber \\
&= &\sup \{\mu([A]): A \in \mathcal{A}, A \subseteq s_{\Gamma}(\psi) \}
\nonumber \\
&= &\sup\{P(A): A \in \mathcal{A}, A \subseteq s_{\Gamma}(\psi)\}
\nonumber \\
&= &P_{*}(s_{\Gamma}(\psi)),
\end{eqnarray}
where $P_{*}$ is the inner probability measure \index{inner probability measure} associated with $P$. This
shows, that the degree of support of a piece of information $\psi$ as defined by (\ref{eq:DefOfSupExt}) is the inner probability of the support $s_{\Gamma}(\psi)$. Note that definitions (\ref{eq:DefOfSupExt}) and (\ref{eq:SpFct}) coincide, if $s_{\Gamma}(\psi) \in \mathcal{A}$. Support functions and inner probability measures are thus closely related. This result is very appealing: any measurable set~$A$, which is contained in $s_{\Gamma}(\psi)$ supports $\psi$. So we expect $P(A) \leq sp_{\Gamma}(\psi)$. In the absence of further information, it is reasonable to take $sp_{\Gamma}(\psi)$ to be the least upper bound of the probabilities of $A$ supporting $\psi$. 

A similar consideration can be made with respect to the possibility sets associated with elements of $\Phi$ with respect to a random mapping $\Gamma$. As before we define the possibility set of $\psi$ as
\begin{eqnarray}
p_{\Gamma}(\psi) = \{\omega \in \Omega: \Gamma(\omega) \cdot \psi \not= 0\}.
\nonumber
\end{eqnarray}
This set contains all assumptions $\omega$ which do not lead to a contradiction with $\psi$ under the mapping $\Gamma$. Thus, the probability of this set, if it is defined, measures the \textit{degree of possibility} or the \textit{degree of plausibility} of $\psi$,
\begin{eqnarray} \label{eq:DegrOfPlaus}
pl_{\Gamma}(\psi) = P(p_{\Gamma}(\psi)).
\end{eqnarray}
As in the case of the degree of support, there is no guarantee that $p_{\Gamma}(\psi)$ is $\mathcal{A}$-measurable. But we can solve this problem in a way similar to the case of the degree of support. A measurable set $A \subseteq p^{c}_{\Gamma}(\psi)$ can be seen as an argument \textit{against} the hypothesis $\psi$, in particular, if $\Gamma$ is normalised. But $A \subseteq p^{c}_{\Gamma}(\psi)$ is equivalent to $A^{c} \supseteq p_{\Gamma}(\psi)$. So a measurable set $A \supseteq p_{\Gamma}(\psi)$ can be considered as an argument that hypothesis $\psi$ cannot be excluded. Therefore we define for every set $H \in \mathcal{P}$
\begin{eqnarray} \label{eq:xi0}
\xi_{0}(H) = \bigwedge\{[A]: A \supseteq H,A \in \mathcal{A}\}.
\end{eqnarray}
Note that $A \supseteq H$ if and only if $A^{c} \subseteq H^{c}$. This implies that $\xi_{0}(H) = (\rho_{0}(H^{c}))^{c}$. From this in turn we conclude that the following corollary to Theorem \ref{th:ExtOfProj} holds:

\begin{corollary} \label{cor;ExtOfProj2}
The application $\xi_{0}: \mathcal{P}(\Omega) \rightarrow \mathcal{A}/\mathcal{J}$ as
defined in (\ref{eq:xi0}) has the  following properties:
\begin{eqnarray}
\xi_{0}(\Omega) &= &\top,
\nonumber \\
\xi_{0}(\emptyset) &= &\bot,
\nonumber \\
\xi_{0}\left(\bigcup_{i \in I} H_{i}\right) 
&= &\bigvee_{i \in I}\xi_{0}(H_{i}).
\end{eqnarray}
if $\{H_{i}, i \in I\}$ is a countable family of subsets of $\Omega$.
\end{corollary}

As before we can now compose $p_{\Gamma}$ with $\xi_{0}$ to obtain a mapping $\xi_{\Gamma} = \xi_{0} \circ p_{\Gamma}:\Phi \rightarrow \mathcal{B} =  \mathcal{A}/ \mathcal{J}$. We may then define for any $\psi \in \Psi$ a degree of plausibility by
\begin{eqnarray} \label{eq:ExtOpPlaus}
pl_{\Gamma}(\psi) = \mu(\xi_{\Gamma}(\psi)).
\end{eqnarray}
Using Lemma \ref{downward} we obtain also
\begin{eqnarray}
pl_{\Gamma}(\psi) = inf\{P(A):A \in \mathcal{A},A \supseteq p_{\Gamma}(\psi)\} 
= P^{*}(p_{\Gamma}(\psi)).
\nonumber
\end{eqnarray}
Here $P^{*}$ is the outer probability measure \index{outer probability measure} of the set $p_{\Gamma}(\psi)$. Thus, if $p_{\Gamma}(\psi)$ is measurable, then $P^{*}(p_{\Gamma}(\psi)) = P(p_{\Gamma}(\psi))$, which shows that (\ref{eq:ExtOpPlaus}) defines in fact an extension of the plausibility defined by (\ref{eq:DegrOfPlaus}). 

In the general case considered here, no properties comparable to those of support (for instance Theorem \ref{th:AllolcSupp}) exist for possibility sets and degrees of possibility. This notion gets its full power only in the case of Boolean information algebra, where it becomes a dual concept to support (see Section \ref{subsec:BooleanCase}).


\section{Random variables} \label{subsec:RandVar}

We propose now a number of  alternative approaches to define certain special random maps in an information algebra. We start with an information algebra $\mathcal{R}_s$ of \textit{simple random variables} with values in a domain-free information algebra $(\Phi,\cdot,0,1;E)$ with $E = \{\epsilon_x:x \in Q\}$ and defined on a sample space $(\Omega,\mathcal{A},P)$. Consider the \textit{ideal completion} $I_{\mathcal{R}_s}$ of this algebra. This is a compact information algebra with simple random variables $\mathcal{R}_{s}$ as \textit{finite} elements, see Section \ref{subsec:CompInfAlg}. We call the elements of $I_{\mathcal{R}_s}$ \textit{random variables}. 

A random variable is thus an ideal of simple random variables. As usual, we identify henceforth $\mathcal{R}_{s}$ with its image in $I_{\mathcal{R}_s}$, that is, we identify the simple random variables $\Delta \in \mathcal{R}_{s}$ with their principal ideals $\downarrow\!\Delta$ in $I_{\mathcal{R}_s}$. We also write $\Delta \leq \Gamma$ for $\Delta \in \Gamma$, referring to the order in $I_{\mathcal{R}_s}$. So, for any $\Gamma \in I_{\mathcal{R}_s}$ we may within the algebra $I_{\mathcal{R}_s}$ write $\Gamma = \bigvee \{\Delta \in \mathcal{R}_{s}:\Delta \leq \Gamma\}$. Using the associativity of join in the complete lattice $I_{\mathcal{R}_s}$, we obtain
\begin{eqnarray} \label{eq:CombGenRV}
\lefteqn{\Gamma_{1} \vee \Gamma_{2} = \Gamma_{1} \cdot \Gamma_{2}}
\nonumber \\
&&= \left( \bigvee \{\Delta_{1} \in \mathcal{R}_{s}:\Delta_{1} \leq \Gamma_{1}\} \right) \vee
\left( \bigvee \{\Delta_{2} \in \mathcal{R}_{s}:\Delta_{2} \leq \Gamma_{2}\} \right) 
\nonumber \\
&&=\bigvee \{\Delta_{1} \cdot \Delta_{2}: \Delta_{1},\Delta_{2} \in \mathcal{R}_{s},\Delta_{1} \leq \Gamma_{1},\Delta_{2} \leq \Gamma_{2}\}.
\end{eqnarray}
Note that this corresponds also to the combination of two ideals, see Section \ref{subsec:IdealExt}. In a similar way, by Theorem \ref{th:ContOfExtr}, we find that
\begin{eqnarray} \label{eq:ExtrGenRV}
\epsilon_x(\Gamma) = \epsilon_x(\bigvee \{\Delta \in \mathcal{R}_{s}:\Delta \leq \Gamma\}
= \bigvee \{\epsilon_x(\Delta):\Delta \in \mathcal{R}_{s}:\Delta \leq \Gamma\}.
\end{eqnarray}
Again, this corresponds to the definition of extraction in the ideal completion, Section \ref{subsec:IdealExt}.

To any random variable $\Gamma \in \mathcal{R}$ we may associate a random mapping $\Gamma: \Omega \rightarrow I_{\Phi}$ from the underlying sample space into the ideal completion of $\Phi$ by defining
\begin{eqnarray} \label{eq:DefRandMap}
\Gamma(\omega) = \bigvee \{\Delta(\omega):\Delta \in \mathcal{R}_{s}, \Delta \leq \Gamma\}.
\end{eqnarray}
This random mapping is defined by a sort of \textit{point-wise limit} within $I_{\Phi}$. We denote the random mapping $\Gamma$ deliberately with the same symbol as the generalised random variable $\Gamma$. The reason is that the two concept can essentially by identified as the following lemmata show. In the following lemma, combination and extraction in $I_{\mathcal{R}_s}$ are defined as in (\ref{eq:CombGenRV}) and (\ref{eq:ExtrGenRV}). Note that we denote combination (join) and information extraction for $x \in D$ with the same symbol in $I_{\mathcal{R}_s}$ and in $I_{\Phi}$.

\begin{lemma} \label{le:GebRVasMap}
\begin{enumerate}
\item If $\Gamma_{1},\Gamma_{2} \in I_{\mathcal{R}_s}$, then
\begin{eqnarray}
(\Gamma_{1} \cdot \Gamma_{2})(\omega) = \Gamma_{1}(\omega) \cdot \Gamma_{2}(\omega) \textrm{ for all}\ \omega \in \Omega .
\nonumber
\end{eqnarray}
\item If $\Gamma \in I_{\mathcal{R}_s}$, then $\forall x \in D$
\begin{eqnarray}
(\epsilon_x(\Gamma))(\omega) = \epsilon_x(\Gamma(\omega)) \textrm{ for all}\ \omega \in \Omega.
\nonumber
\end{eqnarray}
\end{enumerate}
\end{lemma}

\begin{proof}
(1) By definition of the random mapping (\ref{eq:DefRandMap}) associated with $\Gamma_{1} \cdot \Gamma_{2}$ we have 
\begin{eqnarray}
(\Gamma_{1} \cdot \Gamma_{2})(\omega) = \bigvee \{\Delta(\omega): \Delta \leq \Gamma_{1} \cdot \Gamma_{2}\},
\nonumber
\end{eqnarray}
where $\Delta$ denote as always simple random variables. Consider now an element $\psi \in (\Gamma_{1} \cdot \Gamma_{2})(\omega)$. In the compact information algebra $I_{\Phi}$ this means that $\psi \leq \bigsqcup \{\Delta(\omega): \Delta \in \Gamma_{1} \cdot \Gamma_{2}\}$. The supremum on the right hand side is over a directed set in $I_{\Phi}$. By compactness, there is therefore a $\Delta \leq \Gamma_{1} \cdot \Gamma_{2}$ such that $\psi \leq \Delta(\omega)$. Now, $\Delta \leq \Gamma_{1} \cdot \Gamma_{2}$ means by the definition of combination in the ideal completion $I_{\mathcal{R}_s}$ that there is a $\Delta_{1} \leq \Gamma_{1}$, $\Delta_{1} \in \mathcal{R}_{s}$, and a $\Delta_{2} \leq \Gamma_{2}$, $\Delta_{2} \in \mathcal{R}_{s}$ such that $\Delta \leq \Delta_{1} \cdot \Delta_{2}$. This implies that $\psi \leq (\Delta_{1} \cdot \Delta_{2})(\omega) = \Delta_{1}(\omega) \cdot \Delta_{2}(\omega)$, where $\Delta_{1}(\omega) \in \Gamma_{1}(\omega)$ and $\Delta_{2}(\omega) \in \Gamma_{2}(\omega)$. But this shows that $\psi \in \Gamma_{1}(\omega)  \cdot \Gamma_{2}(\omega)$. 

Conversely, consider an element $\psi \in \Gamma_{1}(\omega) \cdot \Gamma_{2}(\omega)$. By the definition of the join in $I_{\Phi}$ this means that there are elements $\psi_{1},\psi_{2} \in \Phi$ such that $\psi \leq \psi_{1} \cdot \psi_{2}$, where $\psi_{1} \leq \Gamma_{1}(\omega)$ and $\psi_{2} \leq \Gamma_{2}(\omega)$. Now, $\psi_{1} \leq \Gamma_{1}(\omega)$ means that $\psi_{1} \leq \bigvee \{\Delta(\omega):\Delta \leq \Gamma_{1}\}$. As above, by compactness, there is a $\Delta_{1} \leq \Gamma_{1}$ such that $\psi_{1} \leq \Delta_{1}(\omega)$. Similarly, there is a $\Delta_{2} \leq \Gamma_{2}$ such that $\psi_{2} \leq \Delta_{2}(\omega)$. Thus, $\psi \leq \Delta_{1}(\omega) \cdot \Delta_{2}(\omega) = (\Delta_{1} \cdot \Delta_{2})(\omega)$. Further $\Delta_{1} \cdot \Delta_{2} \leq \Gamma_{1} \cdot \Gamma_{2}$. This implies $\psi \in (\Gamma_{1} \cdot \Gamma_{2})(\omega)$, hence finally $(\Gamma_{1} \cdot \Gamma_{2})(\omega) = \Gamma_{1}(\omega) \cdot \Gamma_{2}(\omega)$. 

(2) Assume next that $\psi \in (\epsilon_x(\Gamma))(\omega)$. As above, using the definition of the random mapping associated with $\epsilon_x(\Gamma)$, this implies that there is a $\Delta \leq \epsilon_x(\Gamma)$ such that $\psi \leq \Delta(\omega)$. By the definition of $\epsilon_x(\Gamma)$ and compactness there is a $\Delta' \leq \Gamma$ such that $\Delta \leq \epsilon_x(\Delta')$. This implies $\psi \leq (\epsilon_x(\Delta'))(\omega) = \epsilon_x(\Delta'(\omega))$, which, together with $\Delta'(\omega) \leq \Gamma(\omega)$ shows that $\psi \in \epsilon_x(\Gamma(\omega))$. 

Conversely, assume $\psi \in \epsilon_x(\Gamma(\omega))$. Then $\psi \leq \epsilon_x(\phi)$ for some $\phi \in \Gamma(\omega)$. Again, as above, there is a $\Delta \leq \Gamma$ such that $\phi \leq \Delta(\omega)$. Therefore, we conclude that $\psi \leq \epsilon_x(\Delta(\omega)) = (\epsilon_x(\Delta))(\omega)$ and $\epsilon_x(\Delta) \leq \epsilon_x(\Gamma)$. This implies that $\psi \in (\epsilon_x(\Gamma))(\omega)$, hence $(\epsilon_x(\Gamma))(\omega) = \epsilon_x(\Gamma(\omega))$.
\end{proof}

According to this lemma we have a homomorphism between the algebras of random variables and of random mappings. In fact, it is an embedding, since $\Gamma_{1}(\omega) = \Gamma_{2}(\omega)$ for all $\omega \in \Omega$ implies $\Gamma_{1} = \Gamma_{2}$. 

The next lemma strengthens Lemma \ref{le:GebRVasMap}.
\begin{lemma} \label{le:SupDirSet}
If $D \subseteq I_{\mathcal{R}_s}$ is a directed set, then
\begin{eqnarray}
(\bigsqcup_{\Gamma \in D} \Gamma)(\omega) = \bigsqcup_{\Gamma \in D} \Gamma(\omega).
\nonumber
\end{eqnarray}
\end{lemma}

\begin{proof}
If $\Gamma' \in D$, then $\Gamma' \leq \bigsqcup_{\Gamma \in D} \Gamma$, hence $\Gamma'(\omega) \leq (\bigsqcup_{\Gamma \in D} \Gamma)(\omega)$ and therefore 
\begin{eqnarray}
\bigsqcup_{\Gamma \in D} \Gamma(\omega) \leq (\bigsqcup_{\Gamma \in D} \Gamma)(\omega).
\nonumber
\end{eqnarray}

Conversely, consider $\psi \in \Phi$ such that $\psi \leq (\bigsqcup_{\Gamma \in D} \Gamma)(\omega)$. Since, according to (\ref{eq:DefRandMap}),
\begin{eqnarray}
(\bigsqcup_{\Gamma \in D} \Gamma)(\omega) = \bigsqcup \{\Delta(\omega):\Delta \leq \bigsqcup_{\Gamma \in D} \Gamma\}
\nonumber
\end{eqnarray}
we have by compactness $\psi \leq \Delta(\omega)$ for some simple random variable $\Delta \leq \bigsqcup_{\Gamma \in D}\Gamma$. Now, since $D$ is a directed set, by compactness, there is a $\Gamma\in D$ such that $\Delta \leq \Gamma$, hence $\Delta(\omega) \leq \Gamma(\omega)$. It follows then that $\psi \leq \bigsqcup_{\Gamma \in D} \Gamma(\omega)$, which in turn implies
\begin{eqnarray}
(\bigsqcup_{\Gamma \in D} \Gamma)(\omega) \leq \bigsqcup_{\Gamma \in D} \Gamma(\omega).
\nonumber
\end{eqnarray}
This concludes the proof of the lemma.
\end{proof}

This lemma shows that the mapping associating a random variable to its random mapping is continuous.

The theory of random variables developed above may be presented particularly in a natural way in the framework of \textit{compact information algebras}. Let $\Phi$ be a compact information algebra with finite elements $\Phi_{f}$. We assume that $\Phi_{f}$ is a subalgebra of $\Phi$. Define then \textit{simple random variables} $\Delta$ with finite elements from $\Psi_{f}$ as values. They form still an information algebra $\mathcal{R}_{s}$ with combination and extraction defined point-wise.  Since the ideal completion $I_{\Phi_{f}}$ of the information algebra $\Phi_{f}$ is isomorphic to the compact algebra $\Phi$ (see Section \ref{subsec:CompInfAlg}), the theory above applies to the present case. Random variables in a compact information algebra can thus be considered as random mappings with values in $\Phi$, defined as point-wise limits of simple random variables with finite elements as values. 

As before with random mappings, there is no guarantee that the support $s_{\Gamma}(\psi)$ of a random variable $\Gamma$ is measurable for every $\psi \in \Phi$. But of course we can extend the support function to all of $\Phi$ by the allocation of probability as proposed above. However, we shall show later that the degrees of support $sp_{\Gamma}(\psi)$ of a random variable $\Gamma$ is in fact determined by the degrees of support of its approximating simple random random variables, see Section \ref{subsec:SuppFcts}.

Information algebras are closed under \textit{finite} combinations. But there are information algebras which are  also closed under \textit{countable} combinations. In this section we consider such algebras and uncertain information relative to such algebras. Here follows the definition which will be used in the sequel:

\begin{definition} \textbf{$\sigma$-Information Algebra.}
A domain-free information algebra $(\Phi,\cdot,0,1;E)$ with $E = \{\epsilon_x:x \in Q\}$ is called a $\sigma$-information algebra, if
\begin{enumerate}
\item \textit{Countable Combination:} $\Phi$ is closed under \textit{countable} combinations (joins).
\item \textit{Continuity of Extraction:} For every montone sequence $\phi_{1} \leq \phi_{2} \leq \ldots \in \Phi$, and for any $x \in Q$, it holds that
\begin{eqnarray}
\epsilon_x(\bigvee_{i=1}^{\infty} \phi_{i}) = \bigvee_{i=1}^{\infty} \epsilon_x(\phi_{i}).
\nonumber
\end{eqnarray} 
\end{enumerate}
\end{definition}
The second condition is a weaker version of the continuity of extraction

There are many examples of $\sigma$-information algebras. First of all, any \textit{continuous} or \textit{compact} information algebra $\Phi$ is a $\sigma$-information algebra: Since in these cases $\Phi$ is a complete lattice it is surely closed under countable join. The continuity of extraction follows from Theorems \ref{th:ContOfExtr} and \ref{th:ComExtrJoin2}, since a monotone sequence is a directed set. 
 
Further important examples of $\sigma$-information algebras are minimal extensions of information algebras $\Phi$ which are closed under countable combination. Such extensions can be obtained using ideal completion. In order to do this, we need to introduce a new concept. Let $\Phi$ be an information algebra and $I_{\Phi}$ its ideal completion. A subset $S$ of $I_{\Phi}$ is called $\sigma$-closed, if it is closed under \textit{countable} combinations or joins. The intersection of any family of $\sigma$-closed sets is also $\sigma$-closed. Further the set $I_{\Phi}$ itself is $\sigma$-closed. Therefore, for any subset $X \subseteq I_{\Phi}$ we may define the $\sigma$-closure $\sigma(X)$ as the intersection of all $\sigma$-closed sets containing $X$. 

We are particularly interested in $\sigma(\Phi)$, the $\sigma$-closure of $\Phi$ in $I_{\Phi}$. Note that here, as in the sequel, we identify as usual $\Phi$ with its embedding in $I_\Phi$ under the mapping $\phi \mapsto \downarrow\!\phi$ for simplicity of notation. Also we shall write $\phi$, even if we operate within $I_\Phi$. The $\sigma$-closure of $\Phi$ can be characterized as follows:

\begin{theorem} \label{th:CharSigmaSet}
If $\Phi$ is an information algebra, then
\begin{eqnarray} \label{eq:SigmaSet}
\sigma(\Phi) = \{I \in I_\Phi: I = \bigvee_{i=1}^{\infty} \phi_{i}, \phi_{i} \in \Phi\}.
\end{eqnarray}
\end{theorem}

\begin{proof}
Clearly, the set on the right hand side of equation (\ref{eq:SigmaSet}) contains $\Phi$ and is contained in $\sigma(\Phi)$. We claim that this set is itself $\sigma$-closed. In fact, consider a countable set $I_{j}$ of elements of this set, such that
\begin{eqnarray}
I_{j} = \bigvee_{i=1}^{\infty} \psi_{j,i}
\nonumber
\end{eqnarray}
with $\psi_{j,i} \in \Phi$. Define the set $J = \{(j,i): j=1,2\ldots;i=1,2\ldots\}$ and the sets $J_{j} = \{(j,i);i=1,2,\ldots\}$ for $j=1,2,\ldots$, and $K_{i} = \{(h,j):1 \leq h,j \leq i\}$ for $i=1,2,\ldots$. Then we have
\begin{eqnarray}
J = \bigcup_{j=1}^{\infty} J_{j} =  \bigcup_{i=1}^{\infty} K_{i}.
\nonumber
\end{eqnarray}
By the laws of associativity in the complete lattice $I_{\Phi}$ we obtain then
\begin{eqnarray}
\bigvee_{j=1}^{\infty} I_{j} &= &\bigvee_{j=1}^{\infty} (\bigvee_{(j,i) \in J_{j}} \psi_{j,i}) 
\nonumber \\
&= &\bigvee_{(j,i) \in J} \psi_{j,i} 
= \bigvee_{i=1}^{\infty} (\vee_{(h,j) \in K_{i}} \psi_{h,j}).
\nonumber 
\end{eqnarray}
But $\vee_{(h,j) \in K_{i}} \psi_{h,j} \in \Phi$ for $i=1,2,\ldots$.  Hence $\bigvee_{j=1}^{\infty} I_{j}$ belongs itself to the set on the right hand side of (\ref{eq:SigmaSet}). This means that this set is indeed $\sigma$-closed. Since the set contains $\Phi$, it contains also $\sigma(\Phi)$, hence it equals $\sigma(\Phi)$. 
\end{proof}

Consider now a \textit{monotone} sequence $\psi_{1} \leq \psi_{2} \leq \ldots$ of elements of $\Phi$. Its supremum exists in $I_{\Phi}$ and belongs in fact to $\sigma(\Phi)$. The sequence is furthermore a directed set. Therefore, by Theorem \ref{th:ContOfExtr} join commutes with information extraction, this is expressed in the following theorem. It shows that \textit{continuity of extraction} holds:

\begin{theorem} \label{th:ComJoinFoc}
For a monotone sequence $\psi_{1} \leq \psi_{2} \leq \ldots$ of elements of $\Phi$, and for any $x \in Q$, we have in $\sigma(\Phi)$ that
\begin{eqnarray} \label{eq:Continuity}
\epsilon_x(\bigsqcup_{i=1}^{\infty} \psi_{i}) = \bigsqcup_{i=1}^{\infty} \epsilon_x(\psi_{i}).
\end{eqnarray}
\end{theorem}

Theorem \ref{th:ComJoinFoc} shows in particular that $\sigma(\Phi)$ is closed under extraction. In fact, if $\phi_{i}$ is any sequence of elements of $\Phi$, and $I = \bigvee_{i=1}^\infty \phi_i$, then we may define $\psi_{i} = \vee_{k=1}^{i} \phi_{k} \in \Phi$, such that $\psi_{k}$ for $k=1,2,\ldots$ is a monotone sequence and $I = \bigvee_{i=1}^{\infty} \phi_{i} = \bigsqcup_{i=1}^{\infty} \psi_{i}$. So, for $I \in \sigma(\Psi)$ and any $x \in Q$ by Theorem \ref{th:ComJoinFoc}
\begin{eqnarray} \label{eq:Continuity2}
\epsilon_x(I) = \bigvee_{i=1}^{\infty} \epsilon_x(\psi_{i}),
\end{eqnarray}
where $\epsilon_x(\psi_{i}) \in \Phi$ and hence $\epsilon_x(I) \in \sigma(\Phi)$ by Theorem \ref{th:CharSigmaSet}. As a $\sigma$-closed set, $\sigma(\Phi)$ is closed under combination and contains the null and unit element. Therefore $\sigma(\Phi)$ is itself an information algebra, a subalgebra of $\mathcal{R}_{\Phi}$. Since it is closed under combination (i.e. join) of countable sets, contains $0$ and $1$, and satisfies condition (\ref{eq:Continuity}) it is a $\sigma$-\textit{information algebra}, the $\sigma$-algebra induced by $\Phi$. 

A particular and import case of such a construction is $\sigma(\Phi_{f})$ in a \textit{compact} information algebra. Due to Theorem \ref{th:IdCompFiniteEl}, this can be reduced to the situation of ideal completion, described above.


It should be noted however that $\Phi$ is embedded into the ideal completion $I_{\Phi}$ only by a homomorphism $\phi \mapsto \downarrow\!\phi$, perserving \textit{finite} combination only. Thus, if $\phi_{1},\phi_{2},\ldots$ is a countable set of elements of $\Phi$ and $I = \bigvee_{i=1}^{\infty} \phi_{i}$, then $I$ is not in $\Phi$. 

\begin{example} \textbf{Algebra of Borel Sets.}
The Borel sets $\mathcal{B}$ in $\mathbb{R}^{n}$ form a Boolean $\sigma$-algebra and the cylindrification $\sigma_s(B)$ relative to subsets $s$ of the index set $I = \{1,\ldots,n\}$ of any Borel $B$ set is a Borel set. We take intersection as combination, hence join, under the information order. Then $(\mathcal{B},\cap,\emptyset,\mathbb{R}^n;\Sigma)$, with $\Sigma = \{\sigma_s:s \subseteq I\}$, is an information algebra, a subalgebra of the algebra of all subsets of $\mathbb{R}^{n}$. Further, the countable combination condition of a $\sigma$-information algebra is satisfied. It remains to verify the continuity of extraction. Consider a sequence $B_{1} \supseteq B_{2} \supseteq \ldots$. Assume $\cap_{i} B_{i} \not= \emptyset$. In extension of Lemma \ref{le:propSatOp}, we show that $\sigma_s(\bigcap_i B_i) = \bigcap_i \sigma_s(B_i)$. First, $\bigcap_{i} B_{i} \subseteq B_{i}$ implies $\sigma_s(\bigcap_{i} B_{i}) \subseteq \bigcap_i \sigma_s(B_{i})$. Define $x \equiv_{s} y$ for $x,y \in \mathbb{R}^n$ if the projections $x[s]$ and $y[s]$ coincide (compare Section \ref{sec:SetAlg}). Select an element $x \in \bigcap_{i} \sigma_s(B_{i})$ (assuming this intersection nonempty). Then $x[s] = y_i$ for some tuple $y_i \in B_i$ for every $i$. But since we assume $\bigcap_i B_i \not= \emptyset$, there is a $y \in \bigcap_i B_i$ and $y_i = y[s]$ for all $i$, hence $x \equiv_s y$ and so $x \in \sigma_s(\bigcap_i B_i)$. Therefore, $\sigma_s(\bigcap_{i} B_{i}) = \bigcap_{i}\sigma_s (B_{i})$ and this is the continuity of extraction. 

Other, similar examples of a $\sigma$-information algebra are provided by closed or convex sets in $\mathbb{R}^{n}$.
\end{example}

Consider simple random variables as defined as in Section \ref{subsec:SimpleRanMaps}. We may define a random mapping $\Gamma : \Omega \rightarrow I_{\Phi}$ from a countable family of simple random variables $\Delta_{i}$ by
\begin{eqnarray}
\Gamma(\omega) = \bigvee_{i=1}^{\infty} \Delta_{i}(\omega).
\nonumber
\end{eqnarray} 
We call such a random mapping $\Gamma$ a \textit{proper random variable} in the information algebra $\Phi$. Note that its values are ideals of $\Phi$. In the case of a compact information algebra $\Phi$, the values of the simple random variables are considered to be finite, that is to be in $\Phi_{f}$ and then $\Gamma(\omega)$ may be comnsidered as an element of $\Phi$, since $(\Phi,\leq)$ is a complete lattice.

Let now $\mathcal{R}_{\sigma}$ be the family of proper random variables in the algebra $\Phi$. 

\begin{lemma} \label{le:MonApproxRV}
A proper random variable $\Gamma$ is always the supremum of a monotone increasing sequence $\Delta_{1} \leq \Delta_{2} \leq \ldots$ of simple random variables, such that for all $\omega \in \Omega$,
\begin{eqnarray}
\Gamma(\omega) = \bigvee_{i=1}^{\infty} \Delta_{i}(\omega).
\nonumber
\end{eqnarray}
\end{lemma}

\begin{proof}
If $\Gamma$ is a random variable, then $\Gamma(\omega) = \bigvee_{i=1}^{\infty} \Delta'_{i}(\omega)$ for some sequence $\Delta'_{i}$ of simple random variables. Define
\begin{eqnarray}
\Delta_{i} = \vee_{j=1}^{i} \Delta'_{j}.
\nonumber
\end{eqnarray}
Then each $\Delta_{i}$ is a simple random variable, $i = 1,2,\ldots$ and $\Delta_{1} \leq \Delta_{2} \leq \ldots$. From $\Delta'_{i} \leq \Delta_{i}$, we conclude that $\Gamma(\omega) = \bigvee_{i=1}^{\infty} \Delta'_{i}(\omega) \leq \bigvee_{i=1}^{\infty} \Delta_{i}(\omega)$. On the other hand, $\Delta_{i}(\omega) \leq \Gamma(\omega)$, hence $\bigvee_{i=1}^{\infty} \Delta_{i}(\omega) \leq \Gamma(\omega)$, such that finally $\Gamma(\omega)  = \bigvee_{i=1}^{\infty} \Delta_{i}(\omega)$.
\end{proof}

Proper random variables are random mappings and as such can be combined and extracted point-wise in the ideal completion $I_{\Phi}$:
\begin{enumerate}
\item \textit{Combination:} 
$(\Gamma_{1} \cdot \Gamma_{2})(\omega) = \Gamma_{1}(\omega) \cdot \Gamma_{2}(\omega)$,
\item \textit{Extraction:}
$\epsilon_x(\Gamma)(\omega) = \epsilon_x(\Gamma(\omega))$.
\end{enumerate}
Note that the random maps $0(\omega) = 0$ and $1(\omega) = 1$ are the null and unit element of combination. We have to verify that the resulting random mappings still belong to $\mathcal{R}_{\sigma}$, that is are \textit{proper random variables}. So, let
\begin{eqnarray}
\Gamma_{1} = \bigvee_{i=1}^{\infty} \Delta_{1,i}, \quad \Gamma_{2} = \bigvee_{i=1}^{\infty} \Delta_{2,i}.
\nonumber
\end{eqnarray} 
Then we obtain, using associativity of the supremum
\begin{eqnarray*}
\lefteqn{(\Gamma_1 \cdot \Gamma_2)(\omega) = (\Gamma_{1}\vee \Gamma_{2})(\omega)} \\
&&=\Gamma_{1}(\omega) \vee \Gamma_{2}(\omega)
= (\bigvee_{i=1}^{\infty} \Delta_{1,i}(\omega)) \vee (\bigvee_{i=1}^{\infty} \Delta_{2,i}(\omega))
\\
&&=\bigvee_{i=1}^{\infty} (\Delta_{1,i}(\omega) \vee \Delta_{2,i}(\omega))
= \bigvee_{i=1}^{\infty} (\Delta_{1,i} \vee \Delta_{2,i})(\omega).
\end{eqnarray*} 
Since $\Delta_{1,i} \vee \Delta_{2,i} \in \mathcal{R}_{s}$, this proves that $\Gamma_{1}\vee \Gamma_{2} \in \mathcal{R}_{\sigma}$. Note then that, as usual, $\Gamma_{1} \leq \Gamma_{2}$ if and only if $\Gamma_{1}(\omega) \leq \Gamma_{2}(\omega)$ for all $\omega \in \Omega$, since random variables are random mappings.

Further, let
\begin{eqnarray}
\Gamma(\omega) = \bigvee_{i=1}^{\infty} \Delta_{i}(\omega),
\nonumber
\end{eqnarray}
where $\Delta_{i}$ is an increasing sequence of simple random variables (see Lemma \ref{le:MonApproxRV}). Then, by the continuity of extraction in a compact information algebra (Theorem \ref{th:ContOfExtr})
\begin{eqnarray}
\epsilon_x(\Gamma)(\omega)
= \epsilon_x(\Gamma(\omega)) = \epsilon_x(\bigsqcup_{i=1}^{\infty} \Delta_{i}(\omega))
= \bigsqcup_{i=1}^{\infty} \epsilon_x(\Delta_{i}(\omega))
= \bigsqcup_{i=1}^{\infty} \epsilon_x(\Delta_{i})(\omega).
\nonumber
\end{eqnarray} 
Again, if  $\Delta_{i}$ are simple random variables, then so are the $\epsilon_x(\Delta_{i})$, therefore $\epsilon_x(\Gamma)$ is indeed a proper random variable.

We expect $(\mathcal{R}_{\sigma},\cdot,0,1;E)$, with $E = \{\epsilon_x:x \in Q\}$ where $\epsilon_x$ are extraction operators in the ideal completion $I_\Phi$, to form an information algebra, even a $\sigma$-algebra. This is indeed true. We use the following lemma to prove this statement:

\begin{lemma} \label{le:ContOfRV}
Assume $\Gamma_{i} \in \mathcal{R}_{\sigma}$ for $i = 1,2,\ldots$ to be proper random variables. Then $\bigvee_{i=1}^{n} \Gamma_{i}$ exists in the information algebra $\mathcal{R}_{I_{\Phi}}$ of random mappings into $I_{\Phi}$, and for all $\omega \in \Omega$,
\begin{eqnarray}
\left( \bigvee_{i=1}^{\infty} \Gamma_{i} \right)(\omega) = \bigvee_{i=1}^{\infty} \Gamma_{i}(\omega)
\nonumber
\end{eqnarray} 
\end{lemma}

\begin{proof}
Consider the random mapping $\eta$ defined by $\eta(\omega) = \bigvee_{i=1}^{\infty} \Gamma_{i}(\omega)$. Since $\Gamma_{i}(\omega) \leq  \bigvee_{i=1}^{\infty} \Gamma_{i}(\omega)$, it follows that $\Gamma_{i} \leq \eta$, hence $\eta$ is an upper bound of the random mappings $\Gamma_{i}$. If $\chi$ is another upper bound, then $\Gamma_{i}(\omega) \leq \chi(\omega)$, hence $\eta(\omega) = \bigvee_{i=1}^{\infty} \Gamma_{i}(\omega) \leq \chi(\omega)$, therefore $\eta \leq \chi$. Thus, $\eta$ is the supremum of the random mappings $\Gamma_{i}$.
\end{proof}

\begin{theorem} \label{th:SigmaInfAlgOfRV}
The system $(\mathcal{R}_{\sigma},\cdot,0,1;E)$ of proper random variables in the information algebra $\Phi$, with combination and extraction defined point-wise as above forms a  $\sigma$-information algebra. 
\end{theorem}

\begin{proof}
As we have seen above, $\mathcal{R}_{\sigma}$ is closed under combination (join) and extraction. The bottom element, the mapping $1(\omega) = 1$ as well as the top element $0(\omega) = 0$ belong also to $\mathcal{R}_{\sigma}$. So $(\mathcal{R}_{\sigma},D;\leq,\bot,\cdot,\epsilon)$ is a subalgebra of the algebra of random mappings $\mathcal{R}_{I_{\Phi}}$, hence an information algebra.

We show that $\mathcal{R}_{\sigma}$ is $\sigma$-closed, that is, if $\Gamma_{i} \in \mathcal{R}_{\sigma}$ for $i = 1,2,\ldots$, then $\bigvee_{i=1}^{\infty} \Gamma_{i} \in \mathcal{R}_{\sigma}$. Let
\begin{eqnarray}
\Gamma_{j}(\omega) = \bigvee_{i=1}^{\infty} \Delta_{j,i}(\omega), \textrm{ for}\ j=1,2,\ldots,
\nonumber
\end{eqnarray}
where $\Delta_{j,i}$ are simple random variables, and define the random mapping $\Gamma$, using Lemma \ref{le:ContOfRV}, by
\begin{eqnarray}
\Gamma(\omega) = \left( \bigvee_{i=1}^{n} \Gamma_{i} \right)(\omega) = \bigvee_{j=1}^{\infty} \Gamma_{j}(\omega) 
= \bigvee_{j=1}^{\infty} \left( \bigvee_{i=1}^{\infty} \Delta_{j,i}(\omega) \right).
\nonumber
\end{eqnarray}
As in the proof of Theorem \ref{th:CharSigmaSet} define the sets $K_{i} = \{(h,j):1 \leq h,j \leq i\}$. Then, as there, we obtain
\begin{eqnarray}
\Gamma(\omega) = \bigvee_{i=1}^{\infty} \left( \vee_{(h,j) \in K_{i}} \Delta_{h,j}(\omega) \right).
\nonumber
\end{eqnarray}
Since $\vee_{(h,j) \in K_{i}} \Delta_{h,j}(\omega)$ defines simple random variables, the random mapping $\Gamma$ is indeed a peroper random variable and $\mathcal{R}_{\sigma}$ is closed under countable combination. 

It remains to verify the continuity of extraction. Assume $\Gamma_{1} \leq \Gamma_{2} \leq \ldots$ be a monotone sequence of proper random variables in $\mathcal{R}_{\sigma}$ and $x \in Q$. Then, the continuity of extraction in $\mathcal{R}_{\sigma}$ follows from this property in $\sigma(\Phi)$, using Lemma \ref{le:ContOfRV} and the continuity of extraction in $\sigma(\Phi)$, as follows:
\begin{eqnarray*}
\lefteqn{\epsilon_x(\bigvee_{i=1}^{\infty} \Gamma_{i})(\omega)} \\
&&= \epsilon_x((\bigvee_{i=1}^{\infty} \Gamma_{i})(\omega))
= \epsilon_x(\bigvee_{i=1}^{\infty} \Gamma_{i}(\omega))
= \bigvee_{i=1}^{\infty} \epsilon_x(\Gamma_{i}(\omega))
\\
&&= \bigvee_{i=1}^{\infty} \epsilon_x(\Gamma_{i})(\omega)
= (\bigvee_{i=1}^{\infty} \epsilon_x(\Gamma_{i}))(\omega).
\end{eqnarray*}
So, we see that $\epsilon_x(\bigvee_{i=1}^{\infty} \Gamma_{i}) = \bigvee_{i=1}^{\infty} \epsilon_x(\Gamma_{i})$. This concludes the proof.
\end{proof}

Certainly, $\mathcal{R}_{s}$ is a subalgebra of $\mathcal{R}_{\sigma}$. Within the algebra $\mathcal{R}_{\sigma}$, each element of $\mathcal{R}_{\sigma}$ is the supremum of the simple random variables it dominates as the following lemma shows. 

\begin{lemma} \label{le:RVasGenRV}
Let $\Gamma \in \mathcal{R}_{\sigma}$, defined by
\begin{eqnarray}
\Gamma(\omega) = \bigvee_{i=1}^{\infty} \Delta_{i}(\omega).
\nonumber
\end{eqnarray}
Then, in the information algebra $\mathcal{R}_{\sigma}$
\begin{eqnarray} \label{eq:RanVarSup}
\Gamma = \bigvee_{i=1}^{\infty} \Delta_{i} = \bigvee \{\Delta: \Delta \in \mathcal{R}_{s}, \Delta \leq \Gamma\}.
\end{eqnarray}
\end{lemma}

\begin{proof}
The first equality in (\ref{eq:RanVarSup}) follows directly from the definition of $\Gamma$. Trivially, $\Gamma$ is an upper bound of the set $\{\Delta:\Delta \leq \Gamma\}$. If $\Gamma'$ is another upper bound of this set, then it is also an upper bound of the $\Delta_{i}$, hence $\Gamma \leq \Gamma'$. Therefore, $\Gamma$ is the least upper bound of the set $\{\Delta:\Delta \leq \Gamma\}$.
\end{proof}
This lemma shows that a proper random variable is also random variable.

We now take the $\sigma$-closure of $\mathcal{R}_{s}$ in the algebraic information algebra $I_{\mathcal{R}_\Phi}$ of  random variables. According to Theorem \ref{th:CharSigmaSet}, elements of $\sigma({R}_{s})$ are defined as
\begin{eqnarray} 
\Gamma = \bigvee_{i=1}^{\infty} \Delta_{i}, \textrm{ with}\ \Delta_{i} \in \mathcal{R}_{s}, \forall i = 1,2,\ldots.
\nonumber
\end{eqnarray}
Then $\sigma(\mathcal{R}_{s})$ is a $\sigma$-information algebra, containing $\mathcal{R}_{s}$, i.e. the simple random variables. To $\Gamma$ we associate a random mapping, just as with random variables, defined by
\begin{eqnarray} 
\Gamma(\omega) = \bigvee_{i=1}^{\infty} \Delta_{i}(\omega), \textrm{ with}\ \Delta_{i} \in \mathcal{R}_{s}, \forall i = 1,2,\ldots.
\nonumber
\end{eqnarray}
Note that $\Gamma(\omega) \in \sigma(\Phi)$ by Theorem \ref{th:CharSigmaSet}. Therefore, the elements of $\sigma(\mathcal{R}_{s})$ are \textit{random variables} with values in the information algebra $(\sigma(\Phi),D;\leq,\bot,\cdot,\epsilon)$. This shows the equivalence of taking the $\sigma$-closure of $\mathcal{R}_{s}$ and the definition of proper random variables as suprema of sequences of simple random variables.


\section{Allocations of probability} \label{subsec:AllocProb}

In Section \ref{subsec:RanMaps} we have introduced the concept of an \textit{allocation of probability} (a.o.p) as a means to extend the degrees of support of a random mapping beyond the measurable elements $\phi$, that is, the elements for which $s_{\Gamma}(\phi) \in \mathcal{A}$. These allocations of probability play an important role in the theory of uncertain information. Therefore, we start here with a study of this concept, first independently of its relation to random mappings and random variables. Subsequently we examine the relation between random mappings and their associated allocations of probability. 

Random mappings, and in particular random variables and proper random variables, provide means to model explicitly the mechanisms which generate uncertain information. We refer to \cite{kohlasmonney95,haennikohlas00,kohlas03,kohlasmonney07,poulykohlas11} for more specific applications of this idea. Alternatively, allocations of probability may serve to directly assign beliefs to pieces of information. This is more in
the spirit of a subjective, epistemological description of belief, advocated especially by G. Shafer
\cite{shafer73,shafer76,shafer79}. In this view, allocations of probability are taken as the primitive
elements, rather than random variables or hints. This is the point of view developed in this
section (see also \cite{kohlas97b,kohlas03b}).

We introduce first the concept of an allocation of probability:

\begin{definition} \textbf{Allocation of Probability.}
If $(\Phi;\leq)$ is a bounded join-semilattice and $(\mu,\mathcal{B})$ a probability algebra, then an \textit{allocation of probability (a.o.p)} \index{allocation!of probability} is a
mapping $\rho: \Phi \rightarrow \mathcal{B}$ such that
\begin{itemize}
\item[(A1)] $\rho( 1) = \top$,
\item[(A2)] $\rho(\phi \vee \psi) = \rho(\phi) \wedge \rho(\psi)$.
\end{itemize}
If furthermore $\rho(0) = \bot$ holds, then the allocation is called
normalised \index{allocation!of probability!normalised}\index{normalised!allocation
of probability}.
\end{definition}

We shall apply this definition to domain-free information algebras $(\Phi,\cdot,0,1;E)$, where in the semilattice $(\Phi,\leq)$ join corresponds to combination. (A1) says then that the full belief is allocated to the trivial vacuous information. More important is (A2). It says that the belief allocated to a combined information $\phi \cdot \psi$ equals the \textit{common} part of belief $\rho(\phi) \wedge \rho(\psi)$ allocated to both of the two pieces of information $\phi$ and $\psi$ individually. We remind that the a.o.p derived from a random mapping satisfies these two properties (see (\ref{eq:AllocOfProb1})). Note, that if $\phi \leq \psi$, that is, $\phi \vee \psi = \psi$, then $\rho(\phi \vee \psi) = \rho(\phi) \wedge \rho(\psi) = \rho(\psi)$, hence $\rho(\psi) \leq \rho(\phi)$. A particular a.o.p is defined by $\nu(\phi) = \bot$, unless $\phi = 1$, in which case $\nu( 1) = \top$. This is called the \textit{vacuous allocation}; no belief is allocated to a non-trivial piece of information. It is associated with the vacuous information represented by the random mapping $\Gamma(\omega) =  1$ for all $\omega \in \Omega$. By $\zeta(\phi) = \top$ for all $\phi \in \Phi$ another a.o.p is defined, which obviously dominates any other a.o.p. It represents the \textit{contradictory allocation}.

We may think of an allocation of probability as the description of a body of belief relative to pieces of information in an information algebra $(\Phi,\cdot,0,1;E)$ obtained from a source of information. Two (or more) distinct sources of information will lead to the definition of two (or more) corresponding allocations of probability. Thus, in a general setting, let $A_{\Phi}$ be the set of all allocations of probability on $\Phi$ in $(\mathcal{B},\mu)$. Select two allocations $\rho_{i}, i = 1,2,$ from $A_{\Phi}$. How can
they be combined in order to synthesise the two bodies of information they represent into a
single, aggregated body? 

The basic idea is as follows: Consider a piece of information $\phi$ in $\Phi$. If now
$\phi_{1}$ and $\phi_{2}$ are two other pieces of information in $\Phi$, such that $\phi
\leq \phi_{1}\cdot \phi_{2}$, then the common belief $\rho_{1}(\phi_{1}) \wedge \rho_{2}(\phi_{2})$ allocated to $\phi_{1}$ and to $\phi_{2}$ by the two allocations $\rho_{1}$ and $\rho_{2}$ respectively, is a belief allocated to $\phi$ by the two allocations simultaneously. That is, the total belief $\rho(\phi)$ to be
allocated to $\phi$ by the two allocations $\rho_{1}$ and $\rho_{2}$ together must equal at
least the \textit{common} belief allocated to $\phi_{1}$ and $\phi_{2}$ individually by each of the two
allocations respectively, that is, if $\phi \leq \phi_1 \cdot \phi_2$,
\begin{eqnarray}
\rho(\phi) &\geq& \rho_{1}(\phi_{1}) \wedge \rho_{2}(\phi_{2}).
\end{eqnarray}
In the absence of other information, it seems then reasonable to define the combined belief
in $\phi$, as obtained from the two sources of information, as the least upper bound of all
these implied beliefs, 
\begin{eqnarray} \label {eq:CombOfAoP}
\rho(\phi) &=& \bigvee \{\rho_{1}(\phi_{1}) \wedge \rho_{2}(\phi_{2}):
\phi \leq \phi_{1} \cdot \phi_{2}\}.
\end{eqnarray}
This defines indeed a new allocation of probability:

\begin{theorem} \label{comballoc}
Let $\rho_{1},\rho_{2} \in A_{\Phi}$ be two allocations of probability. The map $\rho : \Phi \rightarrow \mathcal{B}$ as defined by (\ref{eq:CombOfAoP}) is then an allocation of probability.
\end{theorem}

\begin{proof}
First, we have
\begin{eqnarray}
\rho(1) &= 
&\bigvee \{\rho_{1}(\phi_{1}) \wedge \rho_{2}(\phi_{2}): 1 \leq \phi_{1} \cdot \phi_{2}\}
\nonumber \\
&= &\rho_{1}(1) \wedge \rho_{2}(1) = \top.
\nonumber
\end{eqnarray}
So (A1) is satisfied.

Next, let $\psi_{1},\psi_{2} \in \Phi$. By definition we have
\begin{eqnarray} 
\rho(\psi_{1} \vee \psi_{2}) 
&=& \bigvee \{\rho_{1}(\phi_{1}) \wedge \rho_{2}(\phi_{2}): \psi_{1} \cdot \psi_{2} \leq \phi_{1} \cdot \phi_{2}\}.
\nonumber
\end{eqnarray}
Now, $\psi_{1} \leq \psi_{1} \vee \psi_{2}$ implies that
\begin{eqnarray}
\lefteqn{\bigvee \{\rho_{1}(\phi_{1}) \wedge \rho_{2}(\phi_{2}):
\psi_{1} \cdot \psi_{2} \leq \phi_{1} \cdot \phi_{2}\}}
\nonumber\\
&\leq& \bigvee \{\rho_{1}(\phi_{1}) \wedge \rho_{2}(\phi_{2}):
\psi_{1} \leq \phi_{1} \cdot \phi_{2}\} =\rho_1(\psi_1) \wedge \rho_2(1) =  \rho_1(\psi_1)
\nonumber
\end{eqnarray}
and similarly for $\psi_{2}$. Thus, we have
$\rho(\psi_{1} \vee \psi_{2}) \leq
\rho(\psi_{1}),\rho(\psi_{2})$, that is
$\rho(\psi_{1}  \vee \psi_{2}) \leq \rho(\psi_{1}) \wedge
\rho(\psi_{2})$.

On the other hand,
\begin{eqnarray}
\lefteqn{\{(\phi_{1},\phi_{2}): \psi_{1}  \cdot \psi_{2} \leq
\phi_{1}  \cdot \phi_{2}\}}
\nonumber \\
&\supseteq& \{(\phi_{1},\phi_{2}): \phi_{1} =
\phi'_{1}  \cdot \phi''_{1}, \phi_{2} =
\phi'_{2}  \cdot \phi''_{2}, \psi_{1} \leq \phi'_{1}  \cdot
\phi'_{2},
\psi_{2} \leq \phi''_{1}  \cdot \phi''_{2}\}.
\nonumber
\end{eqnarray}
By the distributive law for complete Boolean algebras we obtain then
\begin{eqnarray}
\lefteqn{\rho(\psi_{1}  \cdot \psi_{2}) }
\nonumber \\
&\geq &\bigvee \{\rho_{1}(\phi'_{1}  \cdot \phi''_{1}) \wedge
\rho_{2}(\phi'_{2}  \cdot \phi''_{2}): \psi_{1} \leq
\phi'_{1}  \cdot \phi'_{2},
\psi_{2} \leq \phi''_{1}  \cdot \phi''_{2}\}
\nonumber \\
&= &\bigvee \{(\rho_{1}(\phi'_{1}) \wedge
\rho_{1}(\phi''_{1}))
\wedge
(\rho_{2}(\phi'_{2}) \wedge \rho_{2}(\phi''_{2})): \psi_{1} \leq
\phi'_{1}  \cdot \phi'_{2},
\psi_{2} \leq \phi''_{1}  \cdot \phi''_{2}\}
\nonumber \\
&= &\left( \bigvee\{\rho_{1}(\phi'_{1}) \wedge
\rho_{2}(\phi'_{2}): \psi_{1} \leq
\phi'_{1}  \cdot \phi'_{2}\} \right) \wedge
\nonumber \\
&&\left( \bigvee\{\rho_{1}(\phi''_{1}) \wedge
\rho_{2}(\phi''_{2}): \psi_{2} \leq
\phi''_{1}  \cdot \phi''_{2}\} \right)
\nonumber \\
&= &\rho(\psi_{1}) \wedge \rho(\psi_{2}).
\end{eqnarray}
This implies finally that $\rho(\psi_{1} \vee \psi_{2}) = \rho(\psi_{1}) \wedge
\rho(\psi_{2})$. Thus (A2) holds too and~$\rho$ is indeed an allocation of probability.
\end{proof}

In this way, in the set of allocations of probability $A_{\Phi}$ a binary combination
operation is defined\index{combination!of allocations of probability}. We denote this
operation by $\cdot$. Thus, $\rho$ as defined by (\ref{eq:CombOfAoP}) is written as $\rho = \rho_{1} \cdot \rho_{2}$. The following theorem gives us the elementary properties of this
operation.

\begin{theorem} \label{th:AoPsemilattice}
The combination operation, as defined by (\ref{eq:CombOfAoP}), is commutative, associative,
idempotent and the vacuous allocation is the {\em neutral} element and the contradictory allocation the {\em null} element of this operation.
\end{theorem}

\begin{proof}
The commutativity of (\ref{eq:CombOfAoP}) is evident. For the associativity note that for a $\psi
\in \Phi$ we have, due to the associativity and distributivity of meet and join in complete Boolean algebras,
\begin{eqnarray}
\lefteqn{((\rho_{1} \cdot \rho_{2}) \cdot \rho_{3})(\psi)}
\nonumber \\
&= &\bigvee\{(\rho_{1} \cdot \rho_{2})(\phi_{1,2})\wedge \rho_{3}(\phi_{3}):
\psi \leq \phi_{1,2} \cdot \phi_{3}\}
\nonumber \\
&= &\bigvee\{ \bigvee\{\rho_{1}(\phi_{1}) \wedge \rho_{2}(\phi_{2}):
\phi_{1,2} \leq \phi_{1} \cdot \phi_{2} \} 
\wedge\rho_{3}(\phi_{3}):
\psi \leq \phi_{1,2} \cdot \phi_{3}\}
\nonumber \\
&= &\bigvee \{\rho_{1}(\phi_{1}) \wedge \rho_{2}(\phi_{2})
\wedge\rho_{3}(\phi_{3}): \psi \leq \phi_{1} \cdot
\phi_{2} \cdot \phi_{3} \}.
\nonumber
\end{eqnarray}
For $(\rho_{1} \cdot (\rho_{2} \cdot \rho_{3}))(\psi)$ we obtain exactly the same
result in the same way. This proves associativity.

To show idempotency consider
\begin{eqnarray*}
\lefteqn{(\rho \cdot \rho)(\psi) 
= \bigvee \{\rho(\phi_{1}) \wedge
\rho(\phi_{2}): \psi \leq \phi_{1} \cdot \phi_{2}\}} 
\\
&&=\bigvee\{\rho(\phi_{1} \cdot \phi_{2}): \psi \leq
\phi_{1} \cdot \phi_{2}\} 
= \rho(\psi)
\end{eqnarray*}
since the last supremum is attained for $\phi_{1} = \phi_{2} = \psi$.

Finally let $\nu$ and $\zeta$ denote the vacuous and contradictory allocations. Then, for any allocation $\rho$ and any
$\psi \in \Psi$ we have, noting that $\nu(\phi) = \bot$, unless $\phi = 1$, in which case
$\nu(1) = \top$,
\begin{eqnarray}
(\rho \cdot \nu)(\psi) 
= \bigvee \{\rho(\phi_{1}) \wedge \nu(\phi_{2}): \psi \leq
\phi_{1} \cdot \phi_{2}\}
= \rho(\psi).
\nonumber
\end{eqnarray}
This shows that $\nu$ is the neutral element for combination. Similarly,
\begin{eqnarray*}
(\rho \cdot \zeta)(\psi) = \bigvee \{\rho(\phi_1) \wedge \zeta(\phi_2):\psi \leq \phi_1 \cdot \phi_2\}
= \bigvee \{\rho(\phi_1) \wedge \top:\psi \leq \phi_1 \cdot \phi_2\} = \top.
\end{eqnarray*}
So, we have $\rho \cdot \zeta = \zeta$ and $\zeta$ is the null element of combination.
\end{proof}

This theorem shows that $A_{\Phi}$ is a {\em semilattice}. Indeed, a partial order between allocations can be introduced as usual
by defining $\rho_{1} \leq \rho_{2}$ if $\rho_{1} \cdot \rho_{2} = \rho_{2}$. This means
that for all $\psi \in \Phi$,
\begin{eqnarray}
\rho_{1} \cdot \rho_{2}(\psi) 
= \bigvee\{\rho_{1}(\psi_{1}) \wedge \rho_{2}(\psi_{2}): \psi \leq
\psi_{1} \cdot \psi_{2}\} = \rho_{2}(\psi).
\nonumber
\end{eqnarray}
We have therefore always $\rho_{1}(\psi_{1}) \wedge \rho_{2}(\psi_{2}) \leq \rho_{2}(\psi)$
if $\psi \leq \psi_{1} \cdot \psi_{2}$. Take now $\psi_{1} = \psi$ and $\psi_{2}
= 1$, such that $\psi \leq \psi \cdot 1 = \psi$, to obtain  $\rho_{1}(\psi) \wedge
\rho_{2}(1) = \rho_{1}(\psi) \leq \rho_{2}(\psi)$. Thus we have $\rho_{1} \leq \rho_{2}$ if
and only if $\rho_{1}(\psi) \leq \rho_{2}(\psi)$ for all $\psi \in \Psi$. Clearly, the combination $\rho_{1} \cdot \rho_{2}$ is the supremum of the two a.o.p in this order. Therefore we shall henceforth write $\rho_{1} \vee \rho_{2}$ for this combination if we want to emphasise the order-theoretic aspects. The vacuous a.o.p is the least element of this semilattice or the unit element for combination, $\rho \vee \nu = \rho$. And the contradictory allocatiob $\zeta$ for all information elements is the greatest element to the semilattice $A_{\Phi}$. So the semilattice of a.o.ps $A_{\Phi}$ is a bounded semilattice.

Next we turn to the operation of  extracting a part of an allocation of probability in an information algebra relativ to a question $x$. More precisely, let $\rho$ be an allocation of probability on an information algebra $(\Phi,\cdot,0,1;E)$ with $E = \{\epsilon_x:x \in Q\}$. Just as it is possible to extract a part of a piece of information $\psi$ from $\Phi$ with the aid of the operator $\epsilon_x$, it should also be possible to focus the belief represented by the a.o.p~$\rho$ to the information
supported by the domain $x$. This means to extract the information related to $x$ from $\rho$. Thus, for a $\psi \in \Phi$ consider the beliefs allocated to pieces of information $\phi$ which are supported by $x$ and which entail $\psi$, i.e.\ $\psi \leq \phi = \epsilon_x(\phi)$. The part of the belief allocated to $\psi$ and relating to the domain $x$, $\epsilon_x(\rho)(\psi)$ must then be at least $\rho(\phi)$,
\begin{eqnarray}
\epsilon_x(\rho)(\psi) \geq \rho(\phi)  \textrm{ for any}\
\phi = \epsilon_x(\phi) \geq \psi.
\end{eqnarray}
In the absence of other information, it seems again, as above, reasonable  to define
$\epsilon_x(\rho)(\psi)$ to be the least upper bound of all these implied supports,
\begin{eqnarray} \label{eq:ExtractAoP}
\epsilon_x(\rho)(\psi) = \bigvee\{\rho(\phi):\psi \leq \phi = \epsilon_x(\phi)\}.
\end{eqnarray}
This defines indeed an allocation of probability.

\begin{theorem}
Let $\rho \in A_{\Phi}$ be an allocation of probability. The map $\epsilon_x(\rho) : \Phi \rightarrow \mathcal{B}$ as defined by (\ref{eq:ExtractAoP}) is an allocation of probability.
\end{theorem}

\begin{proof}
We have by definition
\begin{eqnarray}
\epsilon_x(\rho)(1) = \bigvee\{\rho(\phi):1 \leq \phi = \epsilon_x(\phi)\} = \rho(1) &=& \top.
\nonumber
\end{eqnarray}
Thus (A1) is verified.

Again by definition, 
\begin{eqnarray} 
\epsilon_x(\rho)(\phi_{1} \cdot \phi_{2}) = \bigvee\{\rho(\phi):\phi_{1} \cdot \phi_{2} \leq
\phi = \epsilon_x(\phi)\}.
\nonumber
\end{eqnarray}
From $\phi_{1},\phi_{2} \leq \phi_{1} \cdot \phi_{2}$ it follows that $\epsilon_x(\rho)(\phi_{1} \vee \phi_{2}) \leq \epsilon_x(\rho)(\phi_{1}),\epsilon_x(\rho)(\phi_{2})$ and thus $\epsilon_x(\rho)(\phi_{1} \cdot \phi_{2}) \leq
\epsilon_x(\rho)(\phi_{1}) \wedge \epsilon_x(\rho)(\phi_{2})$.

On the other hand, we have
\begin{eqnarray}
\lefteqn{\{\psi:\phi_{1} \cdot \phi_{2} \leq \psi = \epsilon_x(\psi)\}}
\nonumber\\
&\supseteq
&\{\psi = \psi_{1} \cdot \psi_{2}:\phi_{1} \leq \psi_{1} = \epsilon_x(\psi_{1}),\phi_{2} \leq \psi_{2} = \epsilon_x(\psi_{2})\}.
\nonumber
\end{eqnarray}
From this we obtain, using the distributive law for complete Boolean algebras,
\begin{eqnarray} 
\lefteqn{\epsilon_x(\rho)(\phi_{1} \cdot \phi_{2}) }
\nonumber \\
&&\geq \bigvee\{\rho(\psi_{1} \cdot \psi_{2}):
\phi_{1} \leq \psi_{1} = \epsilon_x(\psi_{1}),\phi_{2} \leq \psi_{2} = \epsilon_x(\psi_{2})\}
\nonumber \\
&&=\bigvee\{\rho(\psi_{1}) \wedge \rho(\psi_{2}):
\phi_{1} \leq \psi_{1} = \epsilon_x(\psi_{1}),\phi_{2} \leq \psi_{2} = \epsilon_x(\psi_{2})\}
\nonumber \\
&&=\left( \bigvee\{\rho(\psi_{1}):\phi_{1} \leq \psi_{1} = \epsilon_x(\psi_{1})\} \right) \wedge
\left( \bigvee\{\rho(\psi_{2}):\phi_{2} \leq \psi_{2} = \epsilon_x(\psi_{2}) \}  \right)
\nonumber \\
&&=\rho(\phi_{1}) \wedge \rho(\phi_{2}).
\nonumber
\end{eqnarray}
This proves property (A2) for an allocation of support.
\end{proof}

We are now going to show that the a.o.p in $A_\Phi$ in fact define a domain-free information algebra $(A_\Phi,\cdot,\nu,\zeta;E)$, where $E = \{\epsilon_x:x \in Q\}$ with operator $\epsilon_x$ defined by \ref{eq:ExtractAoP}, without the Support Axiom (unless $(D,\leq)$ has a largest element). 

The Semigroup Axiom is proved in Theorem \ref{th:AoPsemilattice}. Concerning the unit and null elements we have already noted above that the vacuous allocation $\nu$ is the unit element of combination and the the a.o.p $\zeta$ is the null element of combination. It remains to verify that the operators $\epsilon_x$ are existential quantifiers relative to $(\Phi,\leq)$.

\begin{theorem}
The extraction operator $\epsilon_x$ on $A_\Phi$ is an existential quantifier for all $x \in Q$ .
\end{theorem}

\begin{proof}
First, we have $\epsilon_x(\zeta)(\phi) = \bigvee \{\zeta(\psi):\phi \leq \psi = \epsilon_x(\psi\} = \top$ for all $\phi \in \Phi$, since $\zeta(\psi) = \top$. So, $\epsilon_x(\zeta) = \zeta$. Secondly, for any $\phi \in \Phi$, $\bigvee \{\rho(\psi):\phi \leq \psi = \epsilon_x(\psi)\} \leq \rho(\phi)$ since $\phi \leq \psi$ implies $\rho(\psi) \leq \rho(\phi)$ and so $\epsilon_x(\rho) \leq \rho$ or $\epsilon_x(\rho) \cdot \rho = \rho$.

It remains to prove that $\epsilon_x(\epsilon_x(\rho_1) \cdot \rho_2) = \epsilon_x(\rho_1) \cdot \epsilon_x(\rho_2)$. Fix any $\phi \in \Phi$. Then, by definition of combination and extraction, using the associate and distributive laws in the Boolean algebra $\mathcal{B}$, we have
\begin{eqnarray*}
&&(\epsilon_x(\rho_1) \cdot \epsilon_x(\rho_2))(\phi)  \\
&=& \bigvee \{\epsilon_x(\rho_1)(\phi_1) \wedge \epsilon_x(\rho_2)(\phi_2):\phi \leq \phi_1 \cdot \phi_2\} \\
&=& \bigvee \{\bigvee \{\rho_1(\psi_1):\phi_1 \leq \psi_1 = \epsilon_x(\psi_1)\} \\
&& \wedge \{\bigvee \{\rho_2(\psi_2):\phi_2 \leq \psi_2 = \epsilon_x(\psi_2)\}:\phi \leq \phi_1 \cdot \phi_2\} \\
&=& \bigvee \{\rho_1(\psi_1) \wedge \rho_2(\psi_2):\phi_1 \leq \psi_1 = \epsilon_x(\psi_1),\phi_2 \leq \psi_2 = \epsilon_x(\psi_2),\phi \leq \phi_1 \cdot \phi_2\} \\
&=& \bigvee \{\rho_1(\psi_1) \wedge \rho_2(\psi_2):\psi_1 = \epsilon_x(\psi_1),\psi_2 = \epsilon_x(\psi_2),\phi \leq \psi_1 \cdot \psi_2\}.
\end{eqnarray*}
Also by definition of combination we have
\begin{eqnarray*}
(\epsilon_x(\rho_1) \cdot \rho_2)(\phi) = \bigvee \{\epsilon_x(\rho_1)(\phi_1) \cdot \rho_2(\phi_2):\phi \leq \phi_1 \cdot \phi_2\}.
\end{eqnarray*}
Therefore, we obtain, again using associativity and distributivity
\begin{eqnarray*}
&&(\epsilon_x(\epsilon_x(\rho_1) \cdot \rho_2)(\phi) \\
&=& \bigvee \{ \bigvee \{\epsilon_x(\rho_1)(\phi_1) \wedge \rho_2(\phi_2): \psi \leq \phi_1 \cdot \phi_2\}:\phi \leq \psi = \epsilon_x(\psi)\} \\
&=& \bigvee \{\epsilon_x(\rho_1)(\phi_1) \wedge \rho_2(\phi_2):\phi \leq \psi = \epsilon_x(\psi) \leq \phi_1 \cdot \phi_2\} \\
&=& \bigvee \{ (\bigvee \{ (\rho_1(\psi_1):\phi_1 \leq \psi_1 = \epsilon_x(\psi_1)\} ) \wedge \rho_2(\phi_2):\phi \leq \psi = \epsilon_x(\psi) \leq \phi_1 \cdot \phi_2\} \\
&=& \bigvee \{\rho_1(\psi_1) \wedge \rho_2(\phi_2):\phi_1 \leq \psi_1 = \epsilon_x(\psi_1),\phi \leq \psi = \epsilon_x(\psi) \leq \phi_1 \cdot \phi_2\} \\
&=& \bigvee \{\rho_1(\psi_1) \wedge \rho_2(\psi_2):\psi_1 = \epsilon_x(\psi_1),\phi \leq \psi = \epsilon_x(\psi) \leq \psi_1 \cdot \psi_2\}.
\end{eqnarray*}
Now, consider a pair of elements $\psi_1$ and $\psi_2$ such that $\psi_1 = \epsilon_x(\psi_1),\psi_2 = \epsilon_x(\psi_2),\phi \leq \psi_1 \cdot \psi_2$. Define $\psi = \psi_1 \cdot \psi_2$. Then $\psi = \epsilon_x(\psi)$ and $\psi_1 = \epsilon_x(\psi_1),\phi \leq \psi = \epsilon_x(\psi) \leq \psi_1 \cdot \psi_2$. This implies that 
\begin{eqnarray*}
\epsilon_x(\rho_1) \cdot \epsilon_x(\rho_2) \leq \epsilon_x(\epsilon_x(\rho_1) \cdot \rho_2).
\end{eqnarray*}
On the other hand, if $\psi_1 = \epsilon_x(\psi_1),\phi \leq \psi = \epsilon_x(\psi) \leq \psi_1 \cdot \psi_2$, then $\phi \leq \psi = \epsilon_x(\psi) \leq \epsilon_x(\psi_1 \cdot \psi_2) = \epsilon_x(\psi_1) \cdot \epsilon_x(\psi_2) = \psi_1 \cdot \epsilon_x(\psi_2)$. Further, since $\psi_2 \geq \epsilon_x(\psi_2)$ it follows that $\rho(\psi_2) \leq \rho(\epsilon_x(\psi_2))$. Therefore
\begin{eqnarray*}
&& (\epsilon_x(\epsilon_x(\rho_1) \cdot \rho_2))(\phi) \\
&\leq& \bigvee \{\rho_1(\psi_1) \wedge \rho_2(\epsilon_x(\psi_2)):\psi_1 = \epsilon_x(\psi_1),\phi \leq \psi = \epsilon_x(\psi) \leq \psi_1 \cdot \epsilon_x(\psi_2)\}.
\end{eqnarray*}
Then recall that $\epsilon_x(\epsilon_x(\psi_2)) = \epsilon_x(\psi_2)$. Therefore, in the inequality above, renaming $\epsilon_x(\psi_2)$ by $\psi_2$, we obtain
\begin{eqnarray*}
&&(\epsilon_x(\epsilon_x(\rho_1) \cdot \rho_2))(\phi)  \\
&\leq& \bigvee \{\rho_1(\psi_1) \wedge \rho_2(\psi_2):\psi_1 = \epsilon_x(\psi_1),\psi_2 = \epsilon_x(\psi_2),\phi \leq \psi = \epsilon_x(\psi) \leq \psi_1 \cdot \psi_2\} \\
&=& \bigvee \{\rho_1(\psi_1) \wedge \rho_2(\psi_2):\psi_1 = \epsilon_x(\psi_1),\psi_2 = \epsilon_x(\psi_2),\phi \leq \psi_1 \cdot \psi_2\}
\end{eqnarray*}
This shows that 
\begin{eqnarray*}
\epsilon_x(\rho_1) \cdot \epsilon_x(\rho_2) \geq \epsilon_x(\epsilon_x(\rho_1) \cdot \rho_2),
\end{eqnarray*}
hence the quality between the two terms. This concludes the proof.
\end{proof}

These results show that $(A_\Phi,\cdot,\zeta,\nu;E)$ with $E = \{\epsilon_x:x \in Q\}$, where $\epsilon_x$ are extraction operators on $A_\Phi$, is a domain-free information algebra, without the support axiom.

We show now that the algebra $A_\Phi$ is in fact an extension of the information algebra $\Phi$. Consider for any $\phi \in \Psi$ the the following map of $\Phi$ into $\mathcal{B}$:
\begin{eqnarray}
\rho_{\phi}(\psi) 
&=& \left\{ \begin{array}{ll}
\top &\textrm{ if } \psi \leq \phi, \\
\bot &\textrm{ otherwise,}
\end{array} \right.
\end{eqnarray}
It allocates total belief to all elements of information implied by $\phi$, that is to all elements of the principal ideal $\downarrow\!\phi$, and no belief to all other elements. This map is clearly an \textit{allocation of probability}; it is called a \textit{deterministic} allocation. It is a \textit{degenerate} allocation in so far as there is no uncertainty in the information it expresses. It states simply that the piece of information $\phi$ is sure to hold. Obviously the least a.o.p $\nu = \rho_{1}$ is a deterministic allocations, and so is the greatest a.o.p $\zeta = \rho_{0}$. Now, for $\phi_{1},\phi_{2} \in \Phi$ we have 
\begin{eqnarray}
\rho_{\phi_{1}} \cdot \rho_{\phi_{2}}(\psi) 
&=& \bigvee \{\rho_{\phi_{1}}(\psi_{1}) \wedge \rho_{\phi_{2}}(\psi_{2}):
\psi \leq \psi_{1} \cdot \psi_{2}\}
\nonumber \\[.5em]
&=& \left\{ \begin{array}{ll}
\top &\textrm{ if } \psi \leq \phi_{1} \cdot \phi_{2}, \\
\bot &\textrm{ otherwise}
\end{array} \right\}
= \rho_{\phi_{1} \cdot \phi_{2}}(\psi).
\end{eqnarray}
So, the combination of deterministic allocations of $\phi_{1}$ and $\phi_{2}$ produces the deterministic a.o.p of $\phi_{1} \cdot \phi_{2}$. 

Further, for any $\psi \in \Phi$,
\begin{eqnarray}
\epsilon_x(\rho_{\phi})(\psi)
&=& \bigvee \{\rho_{\phi}(\psi'):\psi \leq \psi' = \epsilon_x(\psi') \}.
\nonumber
\end{eqnarray}
This equals $\top$, if there is a $\psi' = \epsilon_x(\psi') \geq \psi$ such that $\psi' \leq  \phi$, and $\bot$ otherwise. But, we have $\psi' = \epsilon_x(\psi') \leq \phi$ if and only if $\psi' = \epsilon_x(\psi') \leq \epsilon_x(\phi)$. This shows that $\epsilon_x(\rho_{\phi})(\psi) = \rho_{\epsilon_x(\phi)}(\psi)$, hence $\epsilon_x(\rho_{\phi}) = \rho_{\epsilon_x(\phi)}$. The extraction of a deterministic a.o.p associated with $\phi$ by $x$ yields the deterministic a.o.p associated with $\epsilon_x(\phi)$.

The mapping $\phi \mapsto \rho_{\phi}$ is thus an embedding of $\Phi$ in $A_{\Phi}$. In this sense, $A_{\Phi},$ extends the information algebra $\Phi$. By the way, we remark that if $(\Phi,\cdot,0,1;E)$ is a commutative information algebra  then the corresponding algebra of a.o.p is obviously also a commutative information algebra.


\section{Allocations and random variables} \label{subsec:AoPAndRV}

We pursue the subject by examining the question how random mappings and allocations of probability, and especially their respective information algebras, are related. In Section \ref{subsec:RanMaps} it has been shown that a random mapping generates an allocation of probability, which specifies how much belief, according to the information represented by the random mapping, is to be assigned to an element of $\Phi$. In this section the relations between random mappings and allocations of probability will be examined in more detail. In particular, we address the question, whether the operations between random mappings, combination and extraction, are reflected in the corresponding operations of the associated a.o.p, in other words, whether the mapping $\Gamma \mapsto \rho_{\Gamma}$ is a homomorphism between random mappings and associated allocations of probability.

We start with \textit{simple random variables}. Fix an information algebra $(\Phi,\cdot,0,1;E)$ with $E = \{\epsilon_x:x \in Q\}$ and a probability space $(\Omega,\mathcal{A},P)$. For any simple random variable $\Delta \in \mathcal{R}_{s}$ defined on this probability space, we have seen that all elements of $\Phi$ and even of $I_{\Phi}$ have measurable allocations of support $s_{\Delta}(\psi) \in \mathcal{A}$ and their degree of support is well defined. If we pass in this case from the probability space $(\Omega,\mathcal{A},P)$ to its associated probability algebra $(\mathcal{B},\mu)$ (see Section \ref{subsec:RanMaps}), then we can define the \textit{allocation of probability} (a.o.p) associated with the random variable $\Delta$,
\begin{eqnarray}
\rho_{\Delta}(\psi) = [s_{\Delta}(\psi)]
\nonumber
\end{eqnarray}
for all elements $\psi \in \Phi$ and even for all elements in $I_{\Phi}$. Thus, we obtain for the degree of support induced by the random variable $\Delta$,
\begin{eqnarray}
sp_{\Delta}(\psi) = P(s_{\Delta}(\psi)) = \mu(\rho_{\Delta}(\psi)).
\nonumber
\end{eqnarray}
Again this holds for all elements of $\Phi$ and even of its ideal completion $I_{\Phi}$. The mapping $\rho_{\Delta}:\Phi \rightarrow \mathcal{B}$ clearly satisfies the defining properties of an allocation of probability introduced above in this Section (see Theorem \ref{th:AllolcSupp} and (\ref{eq:ProjHomomorph})). 

A simple random variable $\Delta$ is defined by a partition $\{B_{1},\ldots,B_{n}\}$ of $\Omega$ consisting of measurable blocks $B_{i}$ and a mapping defined by $\Delta(\omega) = \psi_{i}$ for all $\omega \in B_{i}$ and $i = 1,\ldots,n$. We write $\Delta(\omega) = \Delta(B_{i})$, if $\omega \in B_{i}$. To the partition $\{B_{1},\ldots,B_{n}\}$ of $\Omega$ corresponds a partition $\{[B_{1}],\ldots,[B_{n}]\}$ of the probability algebra $\mathcal{B}$. That is, we have $[B_{i}] \wedge [B_{j}] = \bot$ if $ i \not= j$, and $\vee_{i=1}^{n} [B_{i}] = \top$. The simple random variable $\Delta$ can also be defined by a mapping $\Delta([B_{i}]) = \psi_{i}$ from the partition of $\mathcal{B}$ into $\Psi$. Its allocation of probability can then also be determined as
\begin{eqnarray} \label{eq:AoSofSimpleRV}
\rho_{\Delta}(\psi) = \vee \{[B_{i}]:\psi \leq \Delta([B_{i}])\}.
\end{eqnarray}
We note that $\rho_{\Delta} = \rho_{\Delta^{\rightarrow}}$. So, as far as allocation of probability (and support) is concerned we might as well restrict ourselves to considering \textit{canonical} simple random variables and their information algebra $\mathcal{R}_{s,c}$ (see Section \ref{subsec:SimpleRanMaps}). 

We now consider the mapping $\rho: \Delta \mapsto \rho_{\Delta}$ which maps simple random variables into a.o.p.s. This mapping is a \textit{homomorphism}:

\begin{theorem} \label{thOpAllocProbSimRV}
Let $\Delta_{1},\Delta_{2},\Delta \in \mathcal{R}_{s}$ be simple random variables, defined on partitions in a probability algebra $(\mathcal{B},\mu)$ with values in an information algebra
$(\Phi,\cdot,0,1;E)$. Then, for all $\psi \in \Phi$ and $x \in D$,
\begin{eqnarray} \label{eqCombHomo}
\rho_{\Delta_{1} \cdot \Delta_{2}}(\psi)
&=& (\rho_{\Delta_{1}} \cdot \rho_{\Delta_{2}})(\psi)
\\[.6em] \label{eqFocHomo}
\rho_{\epsilon_x(\Delta)}(\psi)
&=& \epsilon_x(\rho_{\Delta})(\psi).
\end{eqnarray}
\end{theorem}

It is understood that in this theorem the combination on the left is the one in the algebra of simple random variables, whereas on the right it is the one in the algebra of a.o.p s. Similarly, the extraction operator $\epsilon_x$ on the left is the one in the information algebra $\mathcal{R}_{s}$ of simple random variables, the one on the right is the one in the information algebra $A_{\Phi}$ of a.o.p s.

\begin{proof}
(1) Assume that $\Delta_{1}$ is defined on the partition $\{B_{1,1},\ldots,B_{1,n}\}$ and $\Delta_{2}$ on
the partition $\{B_{2,1},\ldots,B_{2,m}\}$ of $\mathcal{B}$.  From the definition of an allocation of
probability, of combination of a.o.p s and the distributive and associative laws for Boolean algebras, we
obtain
\begin{eqnarray}
\lefteqn{(\rho_{\Delta_{1}} \cdot \rho_{\Delta_{2}})(\psi)}
\nonumber \\
&=&\vee \{\rho_{\Delta_{1}}(\psi_{1}) \wedge \rho_{\Delta_{2}}(\psi_{2}): \psi \leq \psi_{1}
\cdot \psi_{2}\}
\nonumber \\
&= &\vee \{ \left( \vee \{B_{1,i}: \psi_{1} \leq \Delta_{1}(B_{1,i}\}) \right) 
\nonumber \\
&&\wedge
\left( \vee \{B_{2,j}: \psi_{2} \leq \Delta_{2}(B_{2,j}\}) \right):\psi \leq \psi_{1} \cdot \psi_{2} \}
\nonumber \\
&= & \vee \{ \vee \{B_{1,i} \wedge B_{2,j} \not= \bot: 
\psi_{1} \leq \Delta_{1}(B_{1,i}), \psi_{2} \leq \Delta_{2}(B_{2,j})\}: \psi \leq \psi_{1} \cdot
\psi_{2} \}
\nonumber \\
&= & \vee \{B_{1,i} \wedge B_{2,j} \not= \bot: 
\psi_{1} \leq \Delta_{1}(B_{1,i}), \psi_{2} \leq \Delta_{2}(B_{2,j}), \psi \leq \psi_{1} \cdot
\psi_{2} \}.
\nonumber
\end{eqnarray}
But $\psi \leq \psi_{1} \cdot \psi_{2}$, $\psi_{1} \leq \Delta_{1}(B_{1,i})$ and $\psi_{2} \leq
\Delta_{2}(B_{2,j})$ if and only if $\psi \leq \Delta_{1}(B_{1,i}) \cdot \Delta_{2}(B_{2,j})$. So we
conclude that
\begin{eqnarray}
\lefteqn{(\rho_{\Delta_{1}} \cdot \rho_{\Delta_{2}})(\psi)}
\nonumber \\
&= &\vee \{B_{1,i} \wedge B_{2,j} \not= \bot: 
\psi \leq \Delta_{1}(B_{1,i}) \cdot \Delta_{2}(B_{2,j})\}
\nonumber \\
&= & \vee \{B_{1,i} \wedge B_{2,j} \not= \bot:
\psi \leq (\Delta_{1} \cdot \Delta_{2})(B_{1,i} \wedge B_{2,j})\}
\\&= & \rho_{\Delta_{1} \cdot \Delta_{2}}(\psi).
\nonumber
\end{eqnarray}

(2) Assume that $\Delta$ is defined on the partition $B_{1},\ldots,B_{n}$ of $\mathcal{B}$. Then $\epsilon_x(\Delta)$ is also defined on $B_{1},\ldots,B_{n}$. The associative law of complete Boolean algebra gives us then,
\begin{eqnarray}
\lefteqn{\epsilon_x(\rho_{\Delta})(\psi)}
\nonumber \\
&=& \vee \{\rho_{\Delta}(\phi): \psi \leq \phi = \epsilon_x(\phi)\}
\nonumber \\
&= & \vee \left\{\vee \{B_{i}: \phi \leq \Delta(B_{i}) \}: \psi \leq \phi = \epsilon_x(\phi) \right\}
\nonumber \\
&= & \vee \{B_{i}: \psi \leq \phi = \epsilon_x(\phi) \leq \Delta(B_{i}) \}.
\nonumber
\end{eqnarray}
But, $\psi \leq \phi = \epsilon_x(\phi) \leq \Delta(B_{i})$ holds if and only if $\psi \leq
\epsilon_x(\Delta(B_{i})) = \epsilon_x(\Delta)(B_{i})$. Hence we see that
\begin{eqnarray}
\epsilon_x(\rho_{\Delta})(\psi)
= \vee \{B_{i}: \psi \leq \epsilon_x(\Delta)(B_{i})\}
= \rho_{\epsilon_x(\Delta)}(\psi).
\nonumber
\end{eqnarray}
This completes the proof.
\end{proof}

As far as allocations of probability induced by simple random
variables are concerned, this theorem shows that the combination and focusing of
allocations reflects correctly the corresponding operations of the underlying random variables. Let $A_{s}$ be the image of $\mathcal{R}_{s,}$ under the mapping $\rho$. That is
$A_{s}$ is the set of all allocations of probability which are induced by simple
random variables in $(\mathcal{B},\mu)$. The mapping satisfies
\begin{eqnarray} \label{eq:RhoIsHomom}
\rho_{\Delta_{1} \cdot \Delta_{2}} &=& \rho_{\Delta_{1}} \cdot \rho_{\Delta_{2}},
\nonumber \\
\rho_{\epsilon_x(\Delta)} &=& \epsilon_x(\rho_{\Delta}).
\end{eqnarray}
Also the vacuous random variable $1$ maps to the vacuous allocation $\nu$ and the null random variable $0$ to $\zeta$. Thus we conclude that the map $\Delta \mapsto \rho_{\Delta}$ is a homomorphism between $\mathcal{R}_{s}$ and $A_{\Phi}$ and that $A_{s}$ is a subalgebra of the information algebra $A_{\Phi}$. We remark that if we restrict the mapping $\rho$ to \emph{canonical} random variables, then the mapping $\Delta^{\rightarrow} \mapsto \rho_{\Delta}$ becomes an \textit{embedding}. 

Now we turn to \textit{random variables} $\Gamma$. Remind that they can be identified with certain random mappings into the ideal completion $I_{\Phi}$ of the information algebra $\Phi$ (see Section \ref{subsec:RandVar}) and as such their allocation of probability is defined by $\rho_{\Gamma}(\psi) = \rho_{0}(s_{\Gamma}(\psi))$ or $\rho_{\Gamma} = \rho_{0} \circ s_{\Gamma}$ (see Section \ref{subsec:RanMaps}). We remind that this covers also the important case of \textit{compact} information algebras $\Phi$, where the simple random variables have finite values in $\Phi_{f}$, if $\Phi_{f}$ is a subalgebra of $\Phi$. Now we show that the a.o.p of a random variable can also be obtained as the limit of the a.o.p of the simple random variables it dominates.

\begin{theorem} \label{th:IdOfAoPs}
For all random variables $\Gamma$, 
\begin{eqnarray} \label{eq:IdComplSimplAoP}
\rho_{\Gamma} = \bigvee \{\rho_{\Delta}: \Delta \leq \Gamma\}.
\end{eqnarray}
\end{theorem}

\begin{proof}
Fix an element $\psi \in \Phi$ and consider a measurable subset $A \subseteq s_{\Gamma}(\psi)$. We define a simple random variable
\begin{eqnarray}
\Delta(\omega) = \left\{
\begin{array}{ll}
\psi & \textrm{if}\ \omega \in A, \\
1& \textrm{otherwise}.
\end{array}
\right.
\nonumber
\end{eqnarray}
Then certainly $\Delta(\omega) \leq \Gamma(\omega)$ for all $\omega \in \Omega$, hence $\Delta \leq \Gamma$. Furthermore we have $\rho_{\Delta}(\psi) = [A]$. This implies that
\begin{eqnarray}
\bigvee \{\rho_{\Delta}(\psi):\Delta \leq \Gamma\} \geq \bigvee\{[A]:A \subseteq s_{\Gamma}(\psi),A \in \mathcal{A}\} = \rho_{0}(s_{\Gamma}(\psi)).
\nonumber
\end{eqnarray}
Conversely, for all $\Delta \leq \Gamma$ it holds that $s_{\Delta}(\psi) \subseteq s_{\Gamma}(\psi)$ and that $s_{\Delta}(\psi) \in \mathcal{A}$. Therefore, we conclude that
\begin{eqnarray}
\bigvee \{\rho_{\Delta}(\psi):\Delta \leq \Gamma\} \leq \bigvee\{[A]:A \subseteq s_{\Gamma}(\psi),A \in \mathcal{A}\} = \rho_{0}(s_{\Gamma}(\psi)).
\nonumber
\end{eqnarray}
This proves that $\rho_{\Gamma}(\psi) =  \bigvee \{\rho_{\Delta}(\psi): \Delta \leq \Gamma\}$ for all $\psi \in \Psi$, hence (\ref{eq:IdComplSimplAoP}) holds.
\end{proof}

Theorem \ref{th:IdOfAoPs} shows that the a.o.p of a random variable is in the ideal completion of the information algebra $A_{s}$ of simple a.o.p. This ideal completion contains \textit{allocations of probability} $\rho_{\Gamma} : \mathcal{B} \rightarrow I_\Phi$ of the random mappings associated with random variables. The ideal completion of $A_{s}$ is a compact information algebra and $(\ref{eq:IdComplSimplAoP})$ shows that the mapping $\Gamma \mapsto \rho_{\Gamma}$ is \textit{continuous}.
It is in fact a homomorphism between the algebra of generalised random variables and their a.o.p as the following theorem shows:

\begin{theorem} \label{th:HomoOfGenRanVar}
Let $\Gamma,\Gamma_{1},\Gamma_{2}$ be random variables on an information algebra $(\Phi,\cdot,0,1;E)$ with $E = \{\epsilon_x:x \in Q\}$ and $x \in Q$. Then
\begin{eqnarray}
\rho_{\Gamma_{1} \cdot \Gamma_{2}} &=& \rho_{\Gamma_{1}} \cdot \rho_{\Gamma_{2}},
\nonumber \\
\rho_{\epsilon_x(\Gamma)} &=& \epsilon_x(\rho_{\Gamma}).
\nonumber
\end{eqnarray}
\end{theorem}

The operations on the left hand side of these identities belong to the algebra of random variables, whereas those on the right hand side to the algebra of a.o.p.

\begin{proof}
We have to show that 
\begin{eqnarray}
\rho_{\Gamma_{1} \cdot  \Gamma_{2}}(\psi) &=& (\rho_{\Gamma_{1}} \cdot  \rho_{\Gamma_{2}})(\psi),
\nonumber \\
\rho_{\epsilon_x(\Gamma)}(\psi) &=& \epsilon_x(\rho_{\Gamma})(\psi),
\nonumber
\end{eqnarray}
for all $\psi \in \Phi$. 

(1) We noted above that the mapping $\Gamma \mapsto \rho_{\Gamma}$ is continuous. Therefore, using (\ref{eq:CombGenRV}) and continuity, $\Delta$ denoting always simple random variables, we have
\begin{eqnarray}
\rho_{\Gamma_{1} \cdot  \Gamma_{2}}
&= &\rho_{\bigvee \{\Delta_{1} \cdot  \Delta_{2}: \Delta_{1} \leq \Gamma_{1},\Delta_{2} \leq \Gamma_{2}\}}
= \bigvee \{\rho_{\Delta_{1} \cdot  \Delta_{2}}: \Delta_{1} \leq \Gamma_{1},\Delta_{2} \leq \Gamma_{2}\}.
\nonumber
\end{eqnarray}
On the other hand, for every $\psi \in \Phi$, we obtain, using Theorem \ref{th:IdOfAoPs} and Theorem~\ref{thOpAllocProbSimRV}, and the associative and
distributive laws of Boolean algebras,
\begin{eqnarray}
\lefteqn{(\rho_{\Gamma_{1}} \cdot  \rho_{\Gamma_{2}})(\psi)}
\nonumber \\
&= &\bigvee \{\rho_{\Gamma_{1}}(\psi_{1}) \wedge \rho_{\Gamma_{2}}(\psi_{1}): \psi \leq \psi_{1} \cdot 
\psi_{2}\}
\nonumber \\
&= &\bigvee \{(\bigvee \{\rho_{\Delta_{1}}(\psi_{1}): \Delta_{1} \leq \Gamma_{1}\})
\nonumber \\
&&\wedge (\bigvee \{\rho_{\Delta_{2}}(\psi_{2}): \Delta_{2} \leq \Gamma_{2}\}): \psi \leq \psi_{1} \cdot 
\psi_{2}\}
\nonumber \\
&= &\bigvee \{\rho_{\Delta_{1}}(\psi_{1}) \wedge \rho_{\Delta_{2}}(\psi_{2}): \Delta_{1} \leq \Gamma_{1},
\Delta_{2} \leq \Gamma_{2}, \psi \leq \psi_{1} \cdot  \psi_{2}\}
\nonumber \\
&= &\bigvee \{ \bigvee \{\rho_{\Delta_{1}}(\psi_{1}) \wedge \rho_{\Delta_{2}}(\psi_{2}):
\psi \leq \psi_{1} \cdot  \psi_{2}\}: \Delta_{1} \leq \Gamma_{1}, \Delta_{2} \leq \Gamma_{2}\}
\nonumber \\
&= &\bigvee \{(\rho_{\Delta_{1}} \cdot  \rho_{\Delta_{2}})(\psi): \Delta_{1} \leq \Gamma_{1}, \Delta_{2} \leq \Gamma_{2}\}
\nonumber \\
&= &\bigvee \{\rho_{\Delta_{1} \cdot  \Delta_{2}}(\psi): \Delta_{1} \leq \Gamma_{1}, \Delta_{2} \leq \Gamma_{2}\}.
\nonumber
\end{eqnarray}
This proves that $\rho_{\Gamma_{1} \cdot  \Gamma_{2}} = \rho_{\Gamma_{1}} \cdot  \rho_{\Gamma_{2}}$.

(2) Again by continuity, we obtain from (\ref{eq:ExtrGenRV}) 
\begin{eqnarray}
\rho_{\epsilon_x(\Gamma)}
= \rho_{\bigvee \{\epsilon_x(\Delta): \Delta \leq \Gamma\}}
= \bigvee \{\rho_{\epsilon_x(\Delta)}: \Delta \leq \Gamma\}.
\nonumber
\end{eqnarray}
But, we have also, by Theorem~\ref{thOpAllocProbSimRV}, (\ref{eqFocHomo}) and Theorem \ref{th:IdOfAoPs} ,
\begin{eqnarray}
\epsilon_x(\rho_{\Gamma})(\phi)
&= &\bigvee \{\rho_{\Gamma}(\psi): \phi \leq \psi = \epsilon_x(\psi) \}
\nonumber \\
&= &\bigvee \{ \bigvee \{\rho_{\Delta}(\psi): \Delta \leq \Gamma\}: \phi \leq \psi = \epsilon_x(\psi) \}
\nonumber \\
&= &\bigvee \{ \bigvee \{\rho_{\Delta}(\psi): \phi \leq \psi = \epsilon_x(\psi) \}: \Delta \leq \Gamma\}
\nonumber \\
&= & \bigvee \{\rho_{\epsilon_x(\Delta)}(\phi): \Delta \leq \Gamma\}.
\nonumber 
\end{eqnarray}
This proves that $\rho_{x(\Gamma)} = \epsilon_x(\rho_{\Gamma})$.
\end{proof}

The following is a remarkable property of generalised random variables, which we formulate in the framework of compact information algebras. The interest of this theorem will become clear later especially in relation to support functions, see Chapter \ref{subsec:SuppFcts}.

\begin{theorem} \label{thCondOfGenRV}
Let $(\Phi,\Phi_f,\cdot,0,1:E)$ be a compact information algebra with finite elements $\Phi_{f}$ such that $\Phi_{f}$ is a subalgebra of $\Phi$. Let $\Gamma$ be a random variable in $\Phi,$. Then, for any directed set $D \subseteq \Phi$,
\begin{eqnarray} \label{eq:AoPCond}
\rho_{\Gamma}(\bigsqcup D) = \bigwedge_{\psi \in D} \rho_{\Gamma}(\psi).
\end{eqnarray}
\end{theorem}

\begin{proof}
We prove first the identity 
\begin{eqnarray} \label{eq:CondOfGenRV1}
\rho_{\Delta}(\phi) = \bigwedge \{\rho_{\Delta}(\psi):\psi \in \Psi_{f},\psi \leq \phi\}.
\end{eqnarray}
for simple random variables $\Delta$. Using the convention introduced above, we write $\Delta([B_{i}]) = \psi_{i} \in \Phi_f$, where the $[B_{i}]$ form a partition of $\mathcal{B}$ for $i=1,\ldots,n$. Then its a.o.p is given by $\rho_{\Delta}(\psi) = \vee\{[B_{i}]:\psi \leq \psi_{i}\}$ (see (\ref{eq:AoSofSimpleRV}). Using this,we obtain
\begin{eqnarray}
 \bigwedge \{\rho_{\Delta}(\psi):\psi \in \Phi_{f},\psi \leq \phi\}
 =  \bigwedge \{\vee_{\psi \leq \psi_{i}} [B_{i}]:\psi \in \Phi_{f},\psi \leq \phi\}
\nonumber
\end{eqnarray}
Since the partition $[B_{i}]$ of $\mathcal{B}$ is finite, the join on the right hand side extends for every $\psi$ only over a finite number of elements $[B_{i}]$. Further, as $\psi$ increases, the number of these elements can only decrease. 
But in $\rho_{\Delta}(\phi) = \vee\{[B_{i}]:\phi \leq \psi_{i}\}$ also only a finite number of elements $[B_{i}]$ appear and this number must be less or equal to the number for any $\psi \leq \phi$. So, as $\psi$ increases towards $\phi$, a minimal number of elements must be attained for some $\psi_{0} \leq \phi$. Say this number is $m$ and assume that the elements are numbered as $[B_{1}],\ldots,[B_{m}]$. Then we conclude that the infimum $\bigwedge \{\rho_{\Delta}(\psi):\psi \in \Phi_{f},\psi \leq \phi\}$ equals $\vee_{i=1}^{m} [B_{i}] $. Now, for all $\psi \in \Phi_f$ such that $\psi_{0} \leq \psi \leq \phi$ we have $\psi \leq \psi_{1},\ldots,\psi_{m}$. Since $\phi = \bigvee_{\psi_{0} \leq \psi \leq \phi} \psi$, we conclude that $\phi  \leq \psi_{1},\ldots,\psi_{m}$. But this means that $\rho_{\Delta}(\phi) = \vee_{i=1}^{m} [B_{i}]$ and this proves (\ref{eq:CondOfGenRV1}).

Next, we extend (\ref{eq:CondOfGenRV1}) to any  random variable $\Gamma = \bigvee\{\Delta:\Delta \in \mathcal{R}_{s},\Delta \leq \Gamma\}$. For this purpose we use the distributive law in the complete Boolean algebra $\mathcal{B}$:
\begin{eqnarray} \label{eq:CondOfGenRV2}
\lefteqn{\rho_{\Gamma}(\phi)}
\nonumber \\
&=& \bigvee\{\rho_{\Delta}(\phi):\Delta \in \mathcal{R}_{s},\Delta \leq \Gamma\}
\nonumber \\
&=&\bigvee\{\bigwedge \{\rho_{\Delta}(\psi):\psi \in \Phi_{f},\psi \leq \phi\}:\Delta \in \mathcal{R}_{s},\Delta \leq \Gamma\}
\nonumber \\
&=&\bigwedge\{\bigvee \{\rho_{\Delta}(\psi):\Delta \in \mathcal{R}_{s},\Delta \leq \Gamma\}:\psi \in \Phi_{f},\psi \leq \phi\}
\nonumber \\
&=& \bigwedge \{\rho_{\Gamma}(\psi):\psi \in \Phi_{f},\psi \leq \phi\}
\end{eqnarray}

To conclude the proof, let $D \subseteq \Phi$ be directed. Consider $\psi \in D$. Then $\psi \leq \bigvee D$, hence $\rho_{\Gamma}(\psi) \geq \rho_{\Gamma}(\bigvee D)$, and it follows that $\bigwedge_{\psi \in D} \rho_{\Gamma}(\psi) \geq \rho_{\Gamma}(\bigvee D)$. On the other hand, if $\eta$ is a finite element and $\eta \leq \bigvee D$, then there is a $\psi \in D$ such that $\eta \leq \psi$. This implies that $\rho_{\Gamma}(\eta) \geq \rho_{\Gamma}(\psi)$. From this we conclude, using (\ref{eq:CondOfGenRV2})
\begin{eqnarray}
\lefteqn{\rho_{\Gamma}(\bigvee D)}
\nonumber \\
&=& \bigwedge\{\rho_{\Gamma}(\eta):\eta \in \Phi_{f},\eta \leq \bigvee D\}
\nonumber \\
&\geq& \bigwedge_{\psi \in D} \rho_{\Gamma}(\psi).
\nonumber
\end{eqnarray}
This proves (\ref{eq:AoPCond}).
\end{proof}

Following \cite{shafer79} we call an allocation of probsbilitxy satisfying (\ref{eq:AoPCond}) \textit{condensable}. Thus, the a.o.p s associated with random variables are condensable.

Next we examine the case of \textit{proper random variables} and their allocations of probability. According to Section \ref{subsec:RandVar}, proper random variables $\Gamma$ are ideals in $I_{\mathcal{R}_{s}}$ and as random mappings $\Gamma(\omega) = \bigvee_{i=1}^{\infty} \Delta_{i}(\omega)$, where $\Delta_{i}$ are simple random variables, they map into $I_{\Phi}$, or more precisely into $\sigma(\Phi) \subseteq I_{\Phi}$. This is equivalent to looking at an \textit{compact} information algebra $\Phi$ and considering proper random variables on the finite elements $\Phi_{f}$. By the Representation Theorem \ref{th:IdCompFiniteEl} the information algebra $\Phi, $ is isomorphic to the ideal completion $I_{\Phi_{f}}$ of the subalgebra of the finite elements $\Phi_{f}$. In the sequel, we consider this case.

A proper random variable $\Gamma$ is then the join (or the limit) of a monotone nondecreasing sequence of simple random variables $\Delta_{i}$ with $\Delta_{1} \leq \Delta_{2} \leq \ldots$, $\Gamma = \bigvee_{i=1}^{\infty} \Delta_{i}$. The simple random variables take values in $\Phi_{f}$, and the proper random variable $\Gamma$ in $\Phi$. By Lemma \ref{le:RVasGenRV} a proper random variable $\Gamma$ is a also a random variable. Therefore Theorem \ref{th:HomoOfGenRanVar} applies also to random variables. So, the mapping $\Gamma \mapsto \rho_{\Gamma}$ is a \textit{homomorphism} of the information algebra $\mathcal{R}_{\sigma}$ of proper random variables into the information algebra $A_{\Phi}$ of a.o.ps.

We are going to show more, namely that the map $\Gamma \mapsto \rho_{\Gamma}$ is a $\sigma$-homomorphism from the $\sigma$-information algebra $\mathcal{R}_{\sigma}$ into the information algebra $A_{\Phi}$. 

\begin{theorem} \label{th:SigmaHomom}
Let $(\Phi,\cdot,0,1;E)$, with $E = \{\epsilon_x:x \in Q\}$, to be an information algebra, and $\Gamma_{i} \in \mathcal{R}_{\sigma}$ for $i=1,2,\ldots$. Then
\begin{eqnarray} \label{eq:SigmaHomom}
\rho_{\bigvee_{i=1}^{\infty} \Gamma_{i}} = \bigvee_{i=1}^{\infty} \rho_{\Gamma_{i}}.  
\end{eqnarray}
\end{theorem}

\begin{proof}
Since the mapping $\Gamma \mapsto \rho_{\Gamma}$ is a homomorphism, it preserves order. As a proper random variable, $\Gamma$ equals $\bigvee_{i=1}^{\infty} \Delta_{i}$, where the $\Delta_{i}$ form a monotone sequence of simple random variables. Since $\Gamma$ is also a random variable, we have by (\ref{eq:IdComplSimplAoP}) $\rho_{\Gamma} = \bigvee\{\rho_{\Delta}:\rho_{\Delta} \in \mathcal{R}_{s},\Delta \leq \Gamma\}$. The monotone sequence $\Delta_{i}$ is directed in $\mathcal{R}$. By compactness there is for every $\Delta \leq \Gamma$ an index $j$ so that $\Delta \leq \Delta_{j}$. This implies $\rho_{\Delta} \leq \rho_{\Delta_{j}}$ from which it follows that $\rho_{\Gamma} \leq \bigvee_{i=1}^{\infty} \rho_{\Delta_{i}}$. The converse inequality is evident. So we conclude that 
\begin{eqnarray} \label{eq:GenAoP}
\rho_{\Gamma} = \bigvee_{i=1}^{\infty} \rho_{\Delta_{i}}
\end{eqnarray}
if $\Gamma = \bigvee_{i=1}^{\infty} \Delta_{i}$.

Consider now the proper random variables $\Gamma_{i}$ for $i=1,2,\ldots$ and define $\Gamma = \bigvee_{i=1}^{\infty} \Gamma_{i}$. Let $\Gamma_{i} = \bigvee_{j=1}^{\infty} \Delta_{i,j}$, where for every $i=1,2,\ldots$ the sequence $\Delta_{i,1},\Delta_{i,2},\ldots$ is a monotone sequence of simple random variables. Then
\begin{eqnarray}
\Gamma = \bigvee_{i=1}^{\infty} \bigvee_{j=1}^{\infty} \Delta_{i,j}.
\nonumber
\end{eqnarray}
In the standard way, we define $\Delta_{i} = \vee_{h=1}^{i} \vee_{j=1}^{h} \Delta_{h,j}$. The $\Delta_{i}$ form a monotone sequence of simple random variables and $\Gamma = \bigvee_{i=1}^{\infty} \Delta_{i}$. By (\ref{eq:GenAoP}), the associative law for joins and the homomorphism between simple random variables and their a.o.ps we obtain
\begin{eqnarray*}
\lefteqn{\rho_{\Gamma}
= \bigvee_{i=1}^{\infty} \rho_{\Delta_{i}} 
=\bigvee_{i=1}^{\infty} \left( \vee_{h=1}^{i} \vee_{j=1}^{h} \rho_{\Delta_{i,j}}  \right)} \\
&&= \bigvee_{i=1}^{\infty} \left( \bigvee_{j=1}^{\infty} \rho_{\Delta_{i,j}}  \right)
= \bigvee_{i=1}^{\infty} \rho_{\Gamma_{i}}.
\nonumber
\end{eqnarray*}
This proves (\ref{eq:SigmaHomom}).
\end{proof}

As a preparation for an interpretation of this result, we remark that for a $\sigma$-information algebra the following general result holds:

\begin{lemma} \label{le:SigmaSupport}
Suppose $\Phi$ to be a $\sigma$-information algebra and $\Gamma$ a random mapping. Then
\begin{eqnarray} \label{eq:SigmaSupport}
s_{\Gamma}(\bigvee_{i=1}^{\infty} \psi_{i}) = \bigcap_{i=1}^{\infty} s_{\Gamma}(\psi_{i}).
\end{eqnarray}
\end{lemma}

\begin{proof}
We have 
\begin{eqnarray}
s_{\Gamma}(\bigvee_{i=1}^{\infty} \psi_{i}) 
= \{\omega \in \Omega:\bigvee_{i=1}^{\infty} \psi_{i} \leq \Gamma(\omega)\}.
\nonumber
\end{eqnarray}
Let $\psi = \bigvee_{i=1}^{\infty} \psi_{i}$. Since $\psi_{i} \leq \psi$ we conclude that $s_{\Gamma}(\psi) \subseteq s_{\Gamma}(\psi_{i})$, hence $s_{\Gamma}(\psi) \subseteq \bigcap_{i=1}^{\infty} s_{\Gamma}(\psi_{i})$. On the other hand, consider $\omega \in \bigcap_{i=1}^{\infty} s_{\Gamma}(\psi_{i})$, that is $\psi_{i} \leq \Gamma(\omega)$ for all $i=1,2,\ldots$. Then we have $\bigvee_{i=1}^{\infty} \psi_{i} = \psi \leq \Gamma(\omega)$, hence $\omega \in s_{\Gamma}(\psi)$. This shows that $s_{\Gamma}(\psi) \supseteq \bigcap_{i=1}^{\infty} s_{\Gamma}(\psi_{i})$ and this proves (\ref{eq:SigmaSupport}).
\end{proof}

Since for any proper random variable $\Gamma$ and every $\psi \in \Psi$, we have $\rho_{\Gamma}(\psi) = \rho_{0}(s_{\Gamma}(\psi))$ and the mapping $\rho_{0}$ is a $\sigma$-homomorphism from the power set of $\Omega$ onto $\mathcal{B}$ (see Theorem \ref{th:ExtOfProj}) it follows also from (\ref{eq:SigmaSupport})
\begin{eqnarray}
\rho_{\Gamma}(\bigvee_{i=1}^{\infty} \psi_{i}) = \bigwedge_{i=1}^{\infty} \rho_{\Gamma}(\psi_{i}).
\nonumber
\end{eqnarray}
An allocation of probability, which satisfies this identity is called a $\sigma$-\textit{allocation of probabiilty}. Thus, a proper random variable induces a $\sigma$-a.o.p.  Let $A_{\sigma}$ denote the image of $\mathcal{R}_{\sigma}$ under the mapping $\Gamma \mapsto \rho_{\Gamma}$ in $A_{\Phi}$. 

Next we show that continuity of extraction is also satisfied in the algebra $(A_{\sigma},D;\leq,\bot,\cdot,\epsilon)$:

\begin{theorem} \label{th:ContOfExtrAoP}
Let $(\Phi,\Phi_f,\cdot,0,1;E)$ with $E = \{\epsilon_x:x \in Q\}$ be a compact information algebra, and $\Gamma_{i} \in \mathcal{R}_{\sigma}$ for $i=1,2,\ldots$ a monotone sequence of proper random variables, $\Gamma_{1} \leq \Gamma_{2} \leq \ldots$. Then for very $x \in Q$,
\begin{eqnarray} \label{eq:ContOfExtrAoP}
\epsilon_x(\bigsqcup_{i=1}^{\infty} \rho_{\Gamma_{i}}) = \bigsqcup_{i=1}^{\infty} \epsilon_x(\rho_{\Gamma_{i}}).
\end{eqnarray}
\end{theorem}

\begin{proof}
The proof is based on the continuity of extraction in the $\sigma$-information algebra $(\mathcal{R}_{\sigma},D;\leq,\bot,\cdot,\epsilon)$ of  proper random variables, see Theorem \ref{th:SigmaInfAlgOfRV},
\begin{eqnarray}
\epsilon_x(\bigsqcup_{i=1}^{\infty} \Gamma_{i}) = \bigsqcup_{i=1}^{\infty} \epsilon_x(\Gamma_{i}).
\nonumber
\end{eqnarray}

Take the a.o.p of both sides.  Using the fact that the mapping is a homomorphism of random variables, Theorem \ref{th:HomoOfGenRanVar}, and Theorem \ref{th:SigmaHomom}, this leads on the left hand to
\begin{eqnarray}
\rho_{\epsilon_x(\bigsqcup_{i=1}^{\infty} \Gamma_{i})} = \epsilon_x(\rho_{\bigsqcup_{i=1}^{\infty} \Gamma_{i}})
= \epsilon_x(\bigvee_{i=1}^{\infty} \rho_{\Gamma_{i}}).
\nonumber
\end{eqnarray}
On the right hand side we obtain by the same argument
\begin{eqnarray}
\rho_{\bigsqcup_{i=1}^{\infty} \epsilon_x(\Gamma_{i})} = \bigsqcup_{i=1}^{\infty} \rho_{\epsilon_x(\Gamma_{i})}
= \bigsqcup_{i=1}^{\infty} \epsilon_x(\rho_{\Gamma_{i}})
\nonumber
\end{eqnarray}
This proves the identity (\ref{eq:ContOfExtrAoP}).
\end{proof}

What can be said about the mapping $\Gamma \mapsto \rho_{\Gamma}$ for random mappings $\Gamma$ in general? Let $(\Phi,\cdot,0,1;E)$ be an information algebra, $(\Omega,\mathcal{A},P)$ a probability space and $\Gamma: \Omega \rightarrow \Phi$ a random mapping. The mapping $\Gamma \mapsto \rho_{\Gamma}$ is obviously \textit{order-preserving}: $\Gamma_{1} \leq \Gamma_{2}$ means that $\Gamma_{1}(\omega) \leq \Gamma_{2}(\omega)$ for all $\omega \in \Omega$. This implies that $s_{\Gamma_{1}}(\psi) \subseteq s_{\Gamma_{2}}(\psi)$ for all $\psi \in \Psi$, and from this it follows that $\rho_{\Gamma_{1}}(\psi) = \rho_{0}(s_{\Gamma_{1}}(\psi)) \leq \rho_{0}(s_{\Gamma_{2}}(\psi)) = \rho_{\Gamma_{1}}(\psi)$ for all $\psi \in \Psi$, hence $\rho_{\Gamma_{1}} \leq \rho_{\Gamma_{2}}$. 

But the mapping is no more a homomorphism. In fact, let $\Gamma_{1}$ and $\Gamma_{2}$ be two random mappings. Then the support of the combination of these random mappings is
\begin{eqnarray}
s_{\Gamma_{1} \cdot \Gamma_{2}}(\psi) &=& \{\omega \in \Omega:\psi \leq \Gamma_{1}(\omega) \cdot \Gamma_{2}(\omega)\}
\nonumber \\
&=& \bigcup \left\{ \omega: \psi_{1} \leq \Gamma_{1}(\omega),\psi_{2} \leq \Gamma_{2}(\omega), \psi \leq \psi_{1} \cdot \psi_{2}  \right\}
\nonumber \\
&=& \bigcup \{s_{\Gamma_{1}}(\psi_{1}) \cap s_{\Gamma_{2}}(\psi_{2}):\psi \leq \psi_{1} \cdot \psi_{2}\}.
\nonumber
\end{eqnarray}
Note that for any index set $I$ and $H_i \subseteq \Omega$, $H_i \subseteq \bigcup_{i \in I} H_i$, hence $\rho_0(H_i) \leq \rho_0(\bigcup_{i \in I} H_i)$ and therefore $\bigvee_{i \in I}\rho_0(H_i) \leq \rho_0(\bigcup_{i \in I} H_i)$. This implies then for all $\psi \in \Psi$
\begin{eqnarray}
\rho_{\Gamma_{1} \cdot \Gamma_{2}}(\psi)
&=& \rho_0(s_{\Gamma_1 \cdot \Gamma_2}(\psi)) \\
&=& \rho_{0}(\bigcup \{s_{\Gamma_{1}}(\psi_{1}) \cap s_{\Gamma_{2}}(\psi_{2}):\psi \leq \psi_{1} \cdot \psi_{2}\})
\nonumber \\
&\geq& \bigvee \{\rho_{0}(s_{\Gamma_{1}}(\psi_{1}) \cap s_{\Gamma_{2}}(\psi_{2})):\psi \leq \psi_{1} \cdot \psi_{2}\}
\nonumber \\
&=& \bigvee \{\rho_{0}(s_{\Gamma_{1}}(\psi_{1})) \wedge \rho_{0}(s_{\Gamma_{2}}(\psi_{2})):\psi \leq \psi_{1} \cdot \psi_{2}\}
\nonumber \\
&=& \bigvee \{\rho_{\Gamma_{1}}(\psi_{1})) \wedge \rho_{\Gamma_{2}}(\psi_{2})):\psi \leq \psi_{1} \cdot \psi_{2}\}
\nonumber \\
&=& (\rho_{\Gamma_{1}} \cdot \rho_{\Gamma_{2}})(\psi).
\end{eqnarray}
So, we have $\rho_{\Gamma_{1} \cdot \Gamma_{2}} \geq \rho_{\Gamma_{1}} \cdot \rho_{\Gamma_{2}}$.
Equality holds only in particular cases, like for instance for random variables. Since $\rho_{\Gamma_{1} \cdot \Gamma_{2}}$ allocates more probability to a hypothesis $\psi \in \Psi$ than $\rho_{\Gamma_{1}} \cdot \rho_{\Gamma_{2}}$ does, it seems that by the map to the allocation of probability some information is lost in general. 

Consider also extraction, that is a random mapping $\Gamma$ and $x \in Q$. Then, since $(\epsilon_x(\Gamma))(\omega) = \epsilon_x(\Gamma(\omega))$,
\begin{eqnarray}
\lefteqn{s_{\epsilon_x(\Gamma)}(\psi) = \{\omega \in \Omega:\psi \leq \epsilon_x(\Gamma(\omega))\}}
\nonumber \\
&&=\bigcup \{s_{\Gamma}(\phi):\phi = \epsilon_x(\phi),\psi \leq \phi\}.
\nonumber \\
\end{eqnarray}
Thus, we obtain for the a.o.p of $\epsilon_x(\Gamma)$,
\begin{eqnarray}
\lefteqn{\rho_{\epsilon_x(\Gamma)}(\psi)
= \rho_{0}(\bigcup \{s_{\Gamma}(\phi):\psi \leq \phi = \epsilon_x(\phi)\})}
\nonumber \\
&&\geq \bigvee \{\rho_{0}(s_{\Gamma}(\phi)):\psi \leq \phi = \epsilon_x(\phi)\}
\nonumber \\
&&= \bigvee \{\rho_{\Gamma}(\phi):\psi \leq \phi = \epsilon_x(\psi)\}
\nonumber \\
&&=(\epsilon_x(\rho_{\Gamma}))(\psi).
\nonumber
\end{eqnarray}
So, here we find that $\rho_{\epsilon_x(\Gamma)} \geq \epsilon_x(\rho_{\Gamma})$ and again equality holds only in particular cases. This is a second indication that the random mapping $\Gamma$ contains more information than its a.o.p $\rho_{\Gamma}$. It follows that random maps and a.o.p.s are \textit{not} equivalent models of uncertainty, except in special cases.

%


\section{Characterization of support functions} \label{subsec:SuppFcts}

As we have noted in Section \ref{subsec:RanMaps}, we may consider a random mapping $\Gamma$ as information, that is, $\Gamma(\omega)$ is a``piece of information'', which can be asserted, provided $\omega$ is the sample element chosen by a chance process, or the ``correct'' assumption in a set of possible assumptions $\Omega$. Here, information $\Gamma(\omega)$ may either be an element of the set $\Phi$ of an information algebra $(\Phi,\cdot,\epsilon)$ or else an \textit{ideal} of $\Phi$, hence an element of the ideal completion $I_\Phi$ of $\Phi$. We have defined the \textit{allocation of support} $s_{\Gamma}(\psi)$ of a random mapping as the set of elements $\omega \in \Omega$, which imply $\psi$, i.e. such that $\psi$ belongs to the ideal $\Gamma(\omega)$, $\psi \in \Gamma(\omega)$ or $\psi \leq \Gamma(\omega)$, see Sections \ref{subsec:SimpleRanMaps}  and \ref{subsec:RanMaps}.  Any  $\omega \in s_{\Gamma}(\psi)$ is an assumption, i.e. an argument, which permits to infer the piece of information $\psi$ in the light of the random mapping $\Gamma$. So, the larger the set $s_{\Gamma}(\psi)$, the more arguments are available to support $\psi$. Or, more to the point, the more probable, the more likely it is that the correct, but unknown assumption $\omega$ belongs to $s_{\Gamma}(\psi)$, the stronger the hypothesis $\psi$ is supported. This probability was denoted by $sp_{\Gamma}(\psi)$ and called the \textit{degree of support} of a hypothesis allocated by a random mapping $\Gamma$. We refer to Section \ref{subsec:RanMaps} for this point of view. The degrees of support can be seen as a numerical map or function $sp_{\Gamma} : \Psi \rightarrow [0,1]$ of $\Psi$ into the unit interval. The goal of this section is to study this function. 

We do not exclude in this section that $\Gamma(\omega) = 0$ for some $\omega$. This represents improper information, which can be interpreted as contradictory information. Under semantic aspects such improper information could and should be excluded. We refer to Section \ref{subsec:SimpleRanMaps} for a discussion of this issue in the context of simple random functions. But for the present discussion this is not essential. If $\Gamma(\omega) \not= 0$ for all $\omega$, the random mapping is called \textit{normalised}.

Consider then a random mapping $\Gamma : \Omega \rightarrow \Phi$ from a probability space $(\Omega,\mathcal{A},P)$ into an idempotent generalised information algebra $\Phi$. The corresponding support is defined for any $\psi \in \Phi$ as
\begin{eqnarray}
s_{\Gamma}(\psi) = \{\omega \in \Omega:\psi \leq \Gamma(\omega)\}.
\nonumber
\end{eqnarray}
The set $s_{\Gamma}$ thus contains all assumptions $\omega$ for which $\Gamma(\omega)$ implies $\psi$. The following theorem collects a few elementary properties of the mapping $s_{\Gamma} : \Phi \rightarrow \mathcal{P}(\Omega)$ (see also Theorem \ref{th:AllolcSupp}):

\begin{theorem} \label{th:ElPropAllSp}
If $\Gamma : \Omega \rightarrow \Phi$, then
\begin{enumerate}
\item $s_{\Gamma}(1) = \Omega$,
\item If $\phi \leq \psi$, then $s_{\Gamma}(\psi) \subseteq s_{\Gamma}(\psi)$,
\item $s_{\Gamma}(\phi \cdot \psi) = s_{\Gamma}(\phi) \cap s_{\Gamma}(\psi)$ for all $\phi,\psi \in \Phi$,
\item if $\Gamma$ is normalised, then $s_{\Gamma}(0) = \emptyset$.
\end{enumerate}
\end{theorem}

\begin{proof}
(1) follows since $1$ is the least element in $\Phi$, hence $1 \leq \Gamma(\omega)$ for all $\omega \in \Omega$. (2) is obvious. (3) follows, since $\phi,\psi \leq \Gamma(\omega)$ if and only if $\phi \cdot \psi \leq \Gamma(\omega)$ and (4) follows from the definition of a normalised random mapping.
\end{proof}

Sometimes $(\Phi,\leq)$ may be a $\sigma$-semilattice or even a complete lattice under information order, for instance, if $(\Phi,\cdot,0,1;E)$ is a compact or continuous information algebra. Then something more can be said about the support of a random mapping.

\begin{theorem} \label{th:ElPropAllSp2}
Let $\Gamma : \Omega \rightarrow \Phi$ be a random mapping.
\begin{enumerate}
\item If $(\Phi,\leq)$ is a $\sigma$-semilattice, $\psi_{1},\psi_{2},\ldots \in \Phi$, then
\begin{eqnarray} \label{eq:AllOfSupSigma}
s_{\Gamma}(\bigvee_{i=1}^{\infty} \psi_{i}) = \bigcap_{i=1}^{\infty} s_{\Gamma}(\psi_{i}).
\end{eqnarray}
\item If $(\Phi,\leq)$ is a complete lattice, $X \subseteq \Phi$, then
\begin{eqnarray} \label{eq:AllOfSupCompl}
s_{\Gamma}(\bigvee X) = \bigcap_{\psi \in X} s_{\Gamma}(\psi).
\end{eqnarray}
\end{enumerate}
\end{theorem}

\begin{proof}
(1) We have $\psi_{1},\psi_{2},\ldots \leq \Gamma(\omega)$ if and only if $\bigvee_{i=1}^{\infty} \psi_{i} \leq \Gamma(\omega)$. This implies (\ref{eq:AllOfSupSigma}).

(2) Similarly, we have $\psi \leq \Gamma(\omega)$ for all $\psi \in X$ if and only if $\bigvee X \leq \Gamma(\omega)$ and this implies (\ref{eq:AllOfSupCompl}).
\end{proof}

We want to make use of the probability space $(\Omega,\mathcal{A},P)$ to judge the likelihood that a random mapping $\Gamma$ supports a hypothesis $\psi \in \Psi$. The degree of support $sp_{\Gamma}(\psi)$ of an element $\psi \in \Psi$ is measured by the probability of its support $s_{\Gamma}(\psi)$, provided this probability is defined. This is the case only if $s_{\Gamma}(\psi) \in \mathcal{A}$. Therefore, we define:

\begin{definition}
If $\Gamma : \Omega \rightarrow \Phi$ is a random mapping from a probability space $(\Omega,\mathcal{A},P)$ into an information algebra $(\Phi,\cdot,0,1;E)$, then $\psi \in \Phi$ is called $\Gamma$-measurable, if $s_{\Gamma}(\psi) \in \mathcal{A}$.
\end{definition}

The set of all $\Gamma$-measurable elements $\psi \in \Phi$ will be denoted by $\mathcal{E}_{\Gamma}$. 

\begin{theorem} \label{th:MeasInfEl}
For any random mapping $\Gamma$, $(\mathcal{E}_{\Gamma},\leq)$ is a subsemilattice of the join-semilattice $(\Phi,\leq)$, containing $1$; if $\Gamma$ is normalised, then $0$ belongs to $\mathcal{E}_{\Gamma}$ too. Further, if $\Phi$ is a $\sigma$-semilattice, then $\mathcal{E}_{\Gamma}$ is a $\sigma$-semilattice.
\end{theorem}

\begin{proof}
The first part of the theorem follows from the definition of $\mathcal{E}_{\Gamma}$ and Theorem \ref{th:ElPropAllSp}. The second part follows from Theorem \ref{th:ElPropAllSp2} since $\mathcal{A}$ is a $\sigma$-field.
\end{proof}

On the semilattice $\mathcal{E}_{\Gamma}$ we define $sp_{\Gamma}(\psi) = P(s_{\Gamma}(\psi))$. Thus, $sp_{\Gamma}$ is a function with values in $[0,1]$, defined on $\mathcal{E}_{\Gamma}$. This function is called the \textit{support function} of the random mapping $\Gamma$. The next theorem collects the basic  properties of this function.

\begin{theorem} \label{th:SpFctProp}
Let $\Gamma$ be a random mapping from the probability space $(\Omega,\mathcal{A},P)$ into the information algebra $(\Phi,\cdot,0,1;E)$, and $sp_{\Gamma}$ the associated support function, defined on $\mathcal{E}_{\Gamma}$. Then $sp_{\Gamma}$ has the following properties:
\begin{enumerate}
\item $sp_{\Gamma}(1) = 1$.
\item If $\psi_{1},\ldots,\psi_{m} \geq \psi$, $\psi_{1},\ldots,\psi_{m},\psi \in \mathcal{E}_{\Gamma}$, $m = 1,2,\ldots$
\begin{eqnarray} \label{eq:MonOrderInf0}
sp_{\Gamma}(\psi) \geq \sum_{\emptyset \not= I \subseteq \{1,\ldots,m\}} (-1)^{\vert I \vert + 1} sp_{\Gamma}(\vee_{i \in I} \psi_{i}).
\end{eqnarray}
\item If $\mathcal{E}_{\Gamma}$ is a $\sigma$-semilattice, and if $\psi_{1} \leq \psi_{2} \leq \ldots \in \mathcal{E}_{\Gamma}$, then
\begin{eqnarray} \label{eq:ContSupFct1}
sp_{\Gamma}(\bigvee_{i=1}^{\infty} \psi_{i}) = \lim_{i \rightarrow \infty} sp_{\Gamma}(\psi_{i}).
\end{eqnarray}
\item If $\Gamma$ is normalised, then $sp_{\Gamma}(0) = 0$.
\end{enumerate}
\end{theorem}

\begin{proof}
(1) and (4) follow from Theorem \ref{th:ElPropAllSp} items 1 and 4..

(2) Note that by Theorem  \ref{th:ElPropAllSp} item 3 we have $sp_{\Gamma}(\vee_{i \in I} \psi_{i}) = P(s_{\Gamma}(\vee_{i \in I} \psi_{i})) = P(\cap_{i \in I} s_{\Gamma}(\psi_{i}))$ for a finite index set $I$. On the right hand side of (\ref{eq:MonOrderInf0}) we have then by the inclusion-exclusion formula of probability theory,
\begin{eqnarray} \label{eq:ContOfSp}
\sum_{\emptyset \not= I \subseteq \{1,\ldots,m\}} (-1)^{\vert I \vert + 1} P(\cap_{i \in I} s_{\Gamma}(\psi_{i}))
= P(\cup_{i=1}^{m} s_{\Gamma}(\psi_{i})).
\nonumber
\end{eqnarray}
But $\psi \leq \psi_{1},\ldots,\psi_{m}$ implies $s_{\Gamma}(\psi) \supseteq s_{\Gamma}(\psi_{i})$, hence
\begin{eqnarray}
s_{\Gamma}(\psi) \supseteq \cup_{i=1}^{m} s_{\Gamma}(\psi_{i})
\nonumber
\end{eqnarray}
This implies (\ref{eq:MonOrderInf0})

(3) In this case $\bigvee_{i=1}^{\infty} \psi_{i} \in \mathcal{E}_{\Gamma}$. Further, by Theorem \ref{th:ElPropAllSp2}, $sp_{\Gamma}(\bigvee_{i=1}^{\infty} \psi_{i}) = P(s_{\Gamma}(\bigvee_{i=1}^{\infty} \psi_{i})) = P(\bigcap_{i=1}^{\infty} s_{\Gamma}(\psi_{i}))$. Now, $\psi_1 \leq \psi_2 \leq \ldots$ implies $s_{\Gamma}(\psi_{1}) \supseteq s_{\Gamma}(\psi_{2}) \supseteq \ldots$ (Theorem \ref{th:ElPropAllSp} (2)). By the continuity of probability it follows that $P(\bigcap_{i=1}^{\infty} s_{\Gamma}(\psi_{i})) = \lim_{i \rightarrow \infty} P(s_{\Gamma}(\psi_{i}))$. This proves (\ref{eq:ContSupFct1}).
\end{proof}

As a consequence we deduce from (2) of the theorem above that for $\phi \leq \psi$ we have $sp_{\Gamma}(\psi) \leq sp_{\Gamma}(\phi)$. Thus the function $sp_{\Gamma}$ is (inversely) monotone. In fact a function satisfying property (2) of the theorem above is called \textit{monotone of order} $\infty$ \cite{choquet53,choquet69}. 

In Section \ref{subsec:RanMaps} we proposed to extend the support function of a random mapping $\Gamma$ beyond the measurable elements by $sp_{\Gamma}(\psi) = \mu(\rho_{\Gamma}(\psi))$, where $\rho_{\Gamma}(\psi) = \rho_{0}(s_{\Gamma}(\psi))$ is the allocation of probability associated with the random mapping $\Gamma$ and $(\mu,\mathcal{B})$ is the probability algebra associated with the probability space $(\Omega,\mathcal{A},P)$. Now, any allocation of probability $\rho : \mathcal{B} \rightarrow \Psi$ generates a function $sp = \mu \circ \rho$ which satisfies properties (1) and (2) of Theorem \ref{th:SpFctProp} as stated in Theorem \ref{th:SpFctPropAoP} below. Therefore, in particular the function $sp_{\Gamma} = \mu \circ \rho_{\Gamma}$, which is defined on $\Phi$, and even $I_{\Phi}$ has the properties  stated in Theorem \ref{th:SpFctProp}. 

\begin{theorem} \label{th:SpFctPropAoP}
Let $(\mu,\mathcal{B})$ be a probability algebra, $\rho :  \Phi \rightarrow \mathcal{B}$ an allocation of probability, and $sp = \mu \circ \rho$. 
\begin{enumerate}
\item $sp$ satisfies properties (1) and (2) of Theorem \ref{th:SpFctProp}
\item If $\Phi$ is a $\sigma$-semilattice and if for all $\psi_1,\psi_2,\ldots$
\begin{eqnarray*}
\rho(\bigvee_{i=1}^{\infty} \psi_i) = \bigwedge_{i=1}^{\infty} \rho(\psi_i),
\end{eqnarray*}
then (3) of Theorem \ref{th:SpFctProp} holds.
\item If $\Phi$ is a complete lattice and if for any directed set $X \subseteq \Phi$
\begin{eqnarray}
\rho(\bigsqcup X) = \bigwedge_{\psi \in X} \rho(\psi),
\nonumber
\end{eqnarray}
then
\begin{eqnarray} \label{eq:CondSupFct1}
sp(\bigsqcup X) = \inf_{\psi \in X} sp(\psi).
\end{eqnarray}
\end{enumerate}
\end{theorem}

\begin{proof}
(1) and (2) are proved as in the proof of Theorem \ref{th:SpFctProp}. 

(3) The set $\{\rho(\psi):\psi \in X\}$ is downwards directed if $X$ is directed. Therefore, by Lemma \ref{downward}
\begin{eqnarray}
\mu(\rho(\bigsqcup X)) = \mu(\bigwedge_{\psi \in X}\rho(\psi)) = \inf_{\psi \in X} \mu(\rho(\psi)).
\nonumber
\end{eqnarray}
This proves (\ref{eq:CondSupFct1}).
\end {proof}

Next, we consider \textit{compact} information algebras $(\Phi,\cdot,0,1;E)$, with finite elements $\Phi_{f}$. By Theorem \ref{th:IdCompFiniteEl} the algebra $\Phi$ is isomorphic to the information algebra of the ideal completion $I_{\Phi_f}$ of its finite elements $\Phi_f$. In other words, the results to be derived below apply also to the ideal completion $I_{\Phi}$ of any information algebra $\Phi$. In this context we remind that a \textit{random variable} $\Gamma$ is the supremum of the simple random variables it dominates, $\Gamma = \bigvee\{\Delta:\Delta \in \mathcal{R}_{s},\Delta \leq \Gamma\}$, see Section \ref{subsec:RandVar}. Simple random variables are here and in the sequel always assumed to take finite elements as values, that is $\Delta(\omega) \in \Psi_{f}$ for all $\omega$. In such a case, the support function of a generalised random variable can be approximated by its values for finite elements.

\begin{theorem} \label{th:SpFctOfGenRV}
Let $(\Phi,\cdot,0,1;E)$ be a compact information algebra, with $\Phi_{f}$ as finite elements and $\Gamma$ a random variable with values in $\Phi$. Further let $sp_{\Gamma} = \mu \circ \rho_{\Gamma}$, where $\rho_{\Gamma} = \rho_{0} \circ s_{\Gamma}$ (see (\ref{eq:DefOfSupExt})). Then for all $\psi \in \Phi$,
\begin{eqnarray} \label{eq:CondSpFct0}
sp_{\Gamma}(\psi) = \inf\{sp_{\Gamma}(\phi):\phi \in \Phi_{f},\phi \leq \psi\}.
\end{eqnarray}
Furthermore, if $X \subseteq \Phi$ is directed, then
\begin{eqnarray} \label{eq:CondSpFct1}
sp_{\Gamma}(\bigvee X) = \inf_{\psi \in X}sp_{\Gamma}(\psi).
\end{eqnarray}
\end{theorem}

\begin{proof}
Note that (\ref{eq:CondSpFct0}) is a particular case of (\ref{eq:CondSpFct1}). By Theorem \ref{thCondOfGenRV} we have $\rho_{\Gamma}(\bigvee X) = \bigwedge_{\psi \in X} \rho_{\Gamma}(\psi)$. Then (\ref{eq:CondSpFct1}) follows from Theorem \ref{th:SpFctPropAoP} (\ref{eq:CondSupFct1}).
\end{proof}

In the same framework, if $\Gamma = \bigvee_{i=1}^{\infty} \Delta_{i}$ is a \textit{proper random variable} defined by a sequence of simple random variables $\Delta_{1},\Delta_{2},\ldots$, then the degree of support of any element in $\sigma(\Phi_{f})$ may be obtained as a limit of the degrees of support of finite elements. In fact, if $\psi \in \sigma(\Psi_{f})$, then $\psi = \bigvee_{i=1}^{\infty} \psi_{i}$, where $\psi_{i} \in \Phi_{f}$ (Theorem \ref{th:CharSigmaSet}). We may always assume that the sequence $\psi_{i}$ is monotone, $\psi_{1} \leq \psi_{2} \leq \ldots$. Then this sequence is a directed set in $\Phi$ and Theorem \ref{th:SpFctOfGenRV} applies. But due to the monotonicity of the sequence, we have $\inf\{sp_{\Gamma}(\psi_{i}):i=1,2,\ldots\} = \lim_{i \rightarrow \infty} sp_{\Gamma}(\psi_{i})$. So, if $\psi = \bigvee_{i=1}^{\infty} \psi_{i}$ and $\psi_{1} \leq \psi_{2} \leq \ldots \in \Psi_{f}$, then
\begin{eqnarray}
sp_{\Gamma}(\psi) = \lim_{i \rightarrow \infty} sp_{\Gamma}(\psi_{i}).
\end{eqnarray}

The degree of support of a proper random variable can in some cases also be approximated by the degrees of support of the simple random variables which approximate the random variable.

\begin{theorem} \label{th:AllElMeas}
Let $(\Phi,\cdot,0,1;E)$ be an information algebra and $\sigma(\Phi)$ its $\sigma$-extension in $I_{\Phi}$. If  $\Gamma = \bigvee_{i=1}^{\infty} \Delta_{i}$, where $\Delta_{i}$ are simple random variables with values in $\Phi$, is a proper random variable, defined on the probability space $(\Omega,\mathcal{A},P)$ with values in $\sigma(\Phi)$, then all elements $\psi \in \Phi$ are $\Gamma$-measurable, $\mathcal{E}_{\Gamma} = \Phi$. Furthermore, if the $\Delta_{i}$ form a monotone increasing sequence of simple random variables, then for all $\psi \in \Phi$,
\begin{eqnarray} \label{eq:ApproxBySimpleRV}
sp_{\Gamma}(\psi) = \lim_{i \rightarrow \infty} sp_{\Delta_{i}}(\psi).
\end{eqnarray}
\end{theorem}

\begin{proof}
If $\Gamma$ is a proper random variable defined by $\Gamma = \bigvee_{i=1}^{\infty} \Delta_{i}$, we may always assume that the $\Delta_{i}$ form a monotone sequence of simple random variables. Consider any $\psi \in \Phi$ and its support $s_{\Gamma}(\psi)$ relative to the random variable $\Gamma$. Then $\Delta_{i} \leq \Gamma$ implies $s_{\Delta_{i}}(\psi) \subseteq s_{\Gamma}(\psi)$, hence $\bigcup_{i=1}^{\infty} s_{\Delta_{i}}(\psi) \subseteq s_{\Gamma}(\psi)$. On the other hand we have
\begin{eqnarray}
s_{\Gamma}(\psi) = \{\omega \in \Omega:\psi \leq \bigvee_{i=1}^{\infty} \Delta_{i}(\omega)\}.
\nonumber
\end{eqnarray}
Consider an $\omega \in s_{\Gamma}(\psi)$. As a monotone sequence, the $\Delta_{i}(\omega)$ form a directed set. Its supremum $\Gamma(\omega)$ belongs to the compact information algebra $I_\Phi$, whose finite elements are given by $\Phi$. Therefore, by compactness, there must be an index $i$ such that $\psi \leq \Delta_{i}(\omega)$, hence $\omega \in s_{\Delta_{i}}(\psi)$. But this shows that $s_{\Gamma}(\psi) \subseteq \bigcup_{i=1}^{\infty} s_{\Delta_{i}}(\psi)$, hence
\begin{eqnarray} \label{eq:SupEqUnion}
s_{\Gamma}(\psi) = \bigcup_{i=1}^{\infty} s_{\Delta_{i}}(\psi).
\end{eqnarray}
Now, $ s_{\Delta_{i}}(\psi)$ is measurable for all $i$, hence $s_{\Gamma}(\psi)$ is so too. This proves the first part of the theorem.

If the sequence of the $\Delta_{i}$ is monotone increasing, then so is $s_{\Delta_{i}}(\psi)$ for any $\psi \in \Phi$. Then (\ref{eq:ApproxBySimpleRV}) follows from (\ref{eq:SupEqUnion}) and the continuity of probability.
\end{proof}

Another approximation of degrees of support by the degrees of support of simple random variables can be stated for random variables.

\begin{corollary} \label{cor:ApproxDegOfSupSimRV}
Let $(\Phi,\cdot,0,1;E)$ be an information algebra and $\Gamma$ a random variable in $\Phi$. Then, for all $\psi \in \Phi$,
\begin{eqnarray} \label{eq:ApproxBySimpRV}
sp_{\Gamma}(\psi) = \sup \{sp_{\Delta}(\psi):\Delta \in \mathcal{R}_{s},\Delta \leq \Gamma\}.
\end{eqnarray}
\end{corollary}

\begin{proof}
We have by Theorem \ref{th:IdOfAoPs}  that
\begin{eqnarray}
\rho_{\Gamma}(\psi) = \bigvee \{\rho_{\Delta}(\psi):\Delta \leq \Gamma\}.
\nonumber
\end{eqnarray}
Here, as in the sequel, $\Delta$ always denote simple random variables. Let $(\mu,\mathcal{B})$ be the probability algebra associated with the probability space on which $\Gamma$ is defined. Then $sp_{\Gamma} = \mu \circ \rho_{\Gamma}$. The set $ \{\rho_{\Delta}(\psi):\Delta \leq \Gamma\}$ is downwards directed in $\mathcal{B}$. Therefore, by Lemma \ref{downward}, we conclude that $sp_{\Gamma}(\psi) = \mu(\rho_{\Gamma}(\psi)) = \sup \{\mu(\rho_{\Delta}(\psi)):\Delta \leq \Gamma\} = \sup \{sp_{\Delta}(\psi):\Delta \leq \Gamma\}$.
\end{proof}

We are in this chapter going to study functions monotone of order $\infty$, satisfying properties (1) and (2) from Theorem \ref{th:SpFctProp} above. As we have seen, such functions do arise from random mappings in different ways and also from allocations of probability. Therefore, we define a corresponding class of functions.

\begin{definition} \label{def:SpFct}
Let $\mathcal{E}$ be a join-semilattice with a least element $1$. Then a function $sp:\mathcal{E} \rightarrow $[0,1] satisfying (1) and (2) below is called a support function on $\mathcal{E}$:
\begin{enumerate}
\item $sp(1) = 1$.
\item If $\psi_{1},\ldots,\psi_{m} \geq \psi$, $\psi_{1},\ldots,\psi_{m},\psi \in \mathcal{E}$, 
\begin{eqnarray} \label{eq:MonOrderInf1}
sp(\psi) \geq \sum_{\emptyset \not= I \subseteq \{1,\ldots,m\}} (-1)^{\vert I \vert + 1} sp(\vee_{i \in I} \psi_{i}).
\end{eqnarray}

\item If in addition $\mathcal{E}$ is closed under countable joins, and  for any montone sequence $\psi_{1} \leq \psi_{2} \leq \cdots$ the condition
\begin{eqnarray} \label{eq:ContOfSpFct}
sp(\bigvee_{i=1}^{\infty} \psi_{i}) = \lim_{i \rightarrow \infty} sp(\psi_{i})
\end{eqnarray}
holds, then $sp$ is called a continuous support function of $\mathcal{E}$.
\item If further $\mathcal{E}$ is a complete semilattice and for any directed set $X \subseteq \mathcal{E}$,
\begin{eqnarray} \label{eqCondOfSpFct}
sp(\bigvee X) = \inf_{\psi \in X} sp(\psi)
\end{eqnarray}
holds, then $sp$ is called a condensable support function on $\mathcal{E}$.
\end{enumerate}
\end{definition}

So, for any random mapping $\Gamma$, the function $sp_{\Gamma}$ is a \textit{support function} on $\mathcal{E}_{\Gamma}$ and even on $\Phi$ (see Theorem \ref{th:SpFctPropAoP}). Proper random variables $\Gamma$ have \textit{continuous} support functions $sp_{\Gamma}$ and the support functions $sp_{\Gamma} = \mu \circ \rho_{\Gamma}$ of random variables $\Gamma$ are \textit{condensable} on $\Phi$, if $(\Phi,\cdot,0,1;E)$ is a compact information algebra (Theorems \ref{th:SpFctOfGenRV} and \ref{th:AllElMeas}). We are going to study such support functions. The first question we are going to examine, is whether any support function can be obtained as the support function of a random mapping. This question will be addressed in the next section. Further, if a support function is defined on some sub-semilattice $\mathcal{E}$ of an information algebra $(\Phi,\cdot,0,1:E)$, how can this function be extended to all of $\Psi$? This question will be studied in Sections  \ref{subsec:CanRandMap} and \ref{subsec:MinExt}.

\section{Generating support functions} \label{sec:GenSpFct}

Any random mapping $\Gamma$ from some probability space $(\Omega,\mathcal{A},P)$ into an information algebra $(\Phi,\cdot,0,1;E)$ generates a support function $sp_{\Gamma}$ on the join-semilattice $\mathcal{E}_{\Gamma} \subseteq \Psi$ of its $\Gamma$-measurable elements. We remind that $\mathcal{E}_{\Gamma}$ contains at least the element $1$ of $\Psi$. Now, suppose that $\mathcal{E}$ is a join-semilattice containing a least element $1$ and that $sp : \mathcal{E} \rightarrow \mathbb{R}$ is a support function according to Definition \ref{def:SpFct} in the previous section. In fact, we shall always consider $\mathcal{E}$ as a sub-semilattice of some information algebra $(\Phi,\cdot,0,1;E)$. Is there a random mapping $\Gamma$ into $\Phi$ such that its support function $sp_{\Gamma}$ coincides with $sp$ on $\mathcal{E}$? We show in this section that the answer is affirmative, with the small amendment, that the mapping is into the ideal completion $I_{\Phi}$ of $\Phi$ rather than into $\Phi$ itself. It is an extension and generalization of \cite{kohlas93}.

This result is based on the \textit{Theorem of Krein-Milman} which states that in a locally convex topological space which is Hausdorff, any compact convex set $S$ is the closure of the convex hull of its extreme points \cite{phelps01}. The set $S$ consists in our case of the support functions as elements in the space of real-valued functions on $\mathcal{E}$. We shall use a result of Choquet on the extreme points of monotone functions of order $\infty$ \cite{choquet53}. In fact, the theory presented here can be seen as part of Choquet's theory of capacities, and illustrates in particular the connection of capacities to probability.

Let $\mathcal{E}$ be a join-semilattice, containing the least element $1$. Consider the vector space $V$ of functions $f:\mathcal{E} \rightarrow \mathbb{R}$ with pointwise addition and scalar multiplication. It becomes a topological space with pointwise convergence. Since $\mathbb{R}$ is Hausdorff, so is $V$ \cite{kelley55}. Define $p_{\psi}(f) = \vert f(\psi) \vert$ for $f \in V$ and $\psi \in \mathcal{E}$. Then $p_{\psi}$ is a \textit{semi-norm}, that is 
\begin{enumerate}
\item it is \textit{positive semidefinite:} $p_{\psi}(f) \geq 0$ for all $f \in V$,
\item it is \textit{positive homogeneous:} $p_{\psi}(\lambda \cdot f) = \lambda \cdot p_{\psi}(f)$, for all $\lambda \geq 0$,
\item and it satisfies the \textit{triangle inequality:} $p_{\psi}(f + g) \leq p_{\psi}(f) + p_{\psi}(g)$. 
\end{enumerate}
A vector space with a family of seminorms is called \textit{locally convex}. Therefore $V$ is a \textit{locally convex topological Hausdorff }. 

Now, let $S$ denote the set of all support functions on $\mathcal{E}$, which is a subset of $V$. The set $S$ is obviously \textit{convex} and \textit{closed} in $V$. Furthermore, $S$ is contained in the product space $\mathbb{R}^{\mathcal{E}} = \prod \{\mathbb{R}:\psi \in \mathcal{E}\}$. Define $S[\psi] = \{f(\psi):f \in S\}$. These sets are \textit{bounded} for all $\psi \in \mathcal{E}$ and their closures $\bar{S}[\psi]$ are therefore \textit{compact}. By Tychonov's theorem \cite{kelley55} the product $\prod \{\bar{S}[\psi]:\psi \in \mathcal{E}\}$ is compact and since $S \subseteq \prod \{\bar{S}[\psi]:\psi \in \mathcal{E}\}$, $S$ is \textit{compact} too.

Next we are going to apply the Krein-Milman theorem to the convex, compact set $S$. Here is the theorem:

\begin{theorem} \label{th:KreinMilm} \textbf{Theorem of Krein-Milman:} 
A non-empty convex, compact subset $S$ of a locally convex Hausdorff space is the closed convex hull of its extreme points.
\end{theorem}

Before we are going to apply this theorem to our problem of finding a random mapping inducing a given support function, we transform the theorem into an integral representation, following \cite{phelps01}. As a preparation we need a further notion. Let $P$ be a probability measure on a subset $C$ of $V$, that is, a nonnegative regular measure on the $\sigma$-algebra of Borel sets in $S$, such that $P(C) = 1$. A point $f \in V$ is said to be represented by $P$, if for every linear functional $h : V \rightarrow \mathbb{R}$, 
\begin{eqnarray}
h(f) = \int_{C} h(v) dP(v).
\nonumber
\end{eqnarray}
We cite the following lemma from \cite{phelps01}:

\begin{lemma} \label{le:Phelps}
Let $C$ be a compact subset of a locally convex topological space $V$. A point $f \in V$ belongs to the closed convex hull $H$ of $C$, if and only if there is a probability measure $P$ on $C$ which represents $f$.
\end{lemma}

Now, with the aid of this lemma, we reformulate the Krein-Milman Theorem \ref{th:KreinMilm}.

\begin{theorem} \label{th:KreinMilm2}
Every point $f$ of a convex, compact subset $S$ of a locally convex Hausdorff space $V$ is represented by a probability measure on $S$, which is supported by the closure of the extreme points $ext(S)$ of $S$, i.e. $P(\overline{ext}(S)) = 1$.
\end{theorem}

\begin{proof}
By the Krein-Milman Theorem  \ref{th:KreinMilm}, $f \in S$ means, that $f$ belongs to the closure of the convex hull of the extreme points $ext(S)$ of $S$. Clearly, the set of extreme point of $S$ is bounded, its closure is therefore compact. Hence, by Lemma \ref{le:Phelps}, $f$ is represented by a probability on the closure of the extreme points of $S$.
\end{proof}

What are the extreme points of the set $S$ of support functions? This question is answered by Theorem 43.4 in \cite{choquet53}. In this theorem Choquet considers functions \textit{alternating of order} $\infty$. This means that in (\ref{eq:MonOrderInf1}) of Definition \ref{def:SpFct} the inverse inequality holds. Now, if $f$ is \textit{monotone} of order $\infty$, then $g(\psi) = f(1) - f(\psi)$ is \textit{alternating} of order $\infty$. So there is a close relation between the two notions. Choquet further considers alternating functions on an ordered commutative semigroup with a zero-element with all elements greater than zero. This applies to our join-semigroup $\mathcal{E}$, which, in addition, is an \textit{idempotent} semigroup. If $\mathcal{C}$ is a convex cone in $V$ and $\mathcal{H}$ is an affine subspace of $V$, not containing the zero function, and which meets every ray of $\mathcal{C}$, then $\mathcal{C} \cap \mathcal{H}$ is a convex set and $f \in \mathcal{C} \cap \mathcal{H}$ is an extreme point of this convex set, if and only if $f$ is an extremal point of the convex cone $\mathcal{C}$. As a consequence of Theorem 43.4, Choquet states in Section 46 of \cite{choquet53} that the extremal points of the convex cone $\mathcal{M}$ of functions monotone to the order $\infty$ are the exponentials on $\mathcal{E}$, that is functions $e :  \mathcal{E} \rightarrow \mathbb{R}$ such that $0 \leq e(\psi) \leq 1$, for all $\psi \in \mathcal(E)$ and
\begin{eqnarray}
e(\phi \cdot \psi) = e(\psi) \times e(\psi).
\nonumber
\end{eqnarray}
for all $\phi,\psi \in \mathcal{E}$ (here $\cdot$ on the left denotes the semigroup operation, $\times$ on the right arithmetic multiplication).

Note now that item 1 of Definition \ref{def:SpFct} requires for a support function that $f(1) = 1$. This defines an affine hyperplane $\mathcal{H}$ in $V$ and $\mathcal{M} \cap \mathcal{H}$ is exactly the set of support functions on $\mathcal{E}$. So its extreme points are the exponentials $e$ on $\mathcal{E}$ with $e(1) = 1$. Since $\mathcal{E}$ is idempotent, we have for any exponential $e(\psi) = e(\psi \cdot \psi) = e(\psi) \times e(\psi)$. Hence $e(\psi)$ takes only the values $0$ or $1$. Let $e_{i}$ for $i = 1,2,\ldots$ be a convergent sequence of exponentials on $\mathcal{E}$, such that
\begin{eqnarray}
e(\psi) = \lim_{i \rightarrow \infty} e_{i}(\psi).
\nonumber
\end{eqnarray}
Then $e$ is a support function, since $S$ is closed, and it is also an exponential on $\mathcal{E}$. So the set of exponentials is both bounded and closed, hence compact. Define for an exponential $e$
\begin{eqnarray}
I_{e} = \{\psi \in \mathcal{E}:e(\psi) = 1\}.
\nonumber
\end{eqnarray}
This is obviously an \textit{ideal} in $\mathcal{E}$ and any ideal $I$ in $\mathcal{E}$ defines an exponential by $e(\psi) = 1$ if $\psi \in I$ and $e(\psi) = 0$ otherwise. So, there is a one-to-one relation between exponentials on $\mathcal{E}$ and ideals of $\mathcal{E}$. We may identify the set of exponentials on $\mathcal{E}$ by the set $I_{\mathcal{E}}$ of ideals in $\mathcal{E}$. 

Fix an element $\psi \in \mathcal{E}$. Define, for $f \in V$, $h_{\psi}(f) = f(\psi)$. This defines a continuous linear function $h_{\psi} : V \rightarrow \mathbb{R}$. Consider now any support function $sp \in  S$. By the reformulated version of the Krein-Milman Theorem, \ref{th:KreinMilm2}, $sp$ is represented by a probability measure on the closed set of its extreme points, that is, the set of exponentials on $\mathcal{E}$. Hence, we have
\begin{eqnarray}
sp(\psi) = h_{\psi}(sp) = \int_{ext(S)} h_{\psi}(e) dP(e) = \int_{ext(S)} e(\psi) dP(e),
\nonumber
\end{eqnarray}
for some probability measure $P$ supported by $ext(S)$ and for all $\psi \in \mathcal{E}$. But, because $e$ is a $0$-$1$-function, this gives
\begin{eqnarray}
sp(\psi) = P\{e:e(\psi) = 1\}.
\nonumber
\end{eqnarray}

Now, we are nearly done. We consider the probability space $(ext(S),\mathcal{B},P)$, where $\mathcal{B}$ denotes the Borel $\sigma$-algebra of subsets of $ext(S)$ and $P$ the probability introduced above. We now construct a mapping from $ext(S)$ into $(I_\Phi,\cdot,0,1;E($, the ideal extension of the information algebra $(\Phi,,\cdot,0,1;E)$. Since $\mathcal{E}$ is supposed to be a sub-semilattice of $\Phi$, the ideal $I_{e}$ associated with the exponential $e$ can be extended to an ideal in $\Phi$, generally in many ways, for example by
\begin{eqnarray}
J_{e} = \{\psi \in \Phi:\psi \leq \phi \textrm{ for some}\ \phi \in I_{e}\}.
\nonumber
\end{eqnarray}
Then we define the random mapping $\Gamma(e) = J_{e}$ from the probability space $(ext(S),\mathcal{B},P)$ into the information algebra $I_\Phi$. As usual, we consider $\Psi$ as a subset of $I_{\Phi}$ by the embedding $\psi \mapsto \downarrow\! \psi$. Let $\psi \in \mathcal{E}$. Then for the support of $\psi$ by $\Gamma$ we obtain
\begin{eqnarray}
s_{\Gamma}(\psi) &=& \{e \in ext(S):\psi \in J_{e}\} = \{e \in ext(S):\psi \in I_{e}\}
\nonumber \\
 &=& \{e \in ext(S):e(\psi) = 1\}.
\nonumber
\end{eqnarray}
As we have seen above, the last set is measurable, that is belongs to $\mathcal{B}$. Hence we see that all elements of $\mathcal{E}$ are $\Gamma$-measurable, $\mathcal{E} \subseteq \mathcal{E}_{\Gamma}$. Further, 
\begin{eqnarray}
sp_{\Gamma}(\psi) = P(s_{\Gamma}(\psi)) = P\{e \in ext(S):e(\psi) = 1\} = sp(\psi).
\nonumber
\end{eqnarray}
So $sp_{\Gamma}$ and $sp$ coincide on $\mathcal{E}$. In this sense $sp$ is induced by the random mapping $\Gamma$, hence $\Gamma$ generates $sp$. We should stress that the $\Gamma$ defined above is not the unique random mapping generating $sp$. This issue will be addressed in Section \ref{subsec:CanRandMap}.

Next we turn to \textit{continuous} support functions. This time let $\mathcal{E}$ be a $\sigma$-join-semilattice, a semilattice closed under \textit{countable} joins. Again, we assume $\mathcal{E}$ to be a sub-semilattice of some $\sigma$-information algebra $(\Phi,\cdot,0,1;E)$. Let $S_{c}$ denote the set of continuous support functions on $\mathcal{E}$. As above, we argue that  $S_{c}$ is still a convex, compact subset of the function space $V$. Therefore, the revised Theorem of Krein-Milman \ref{th:KreinMilm2} still applies. Because the elements of $S_{c}$ are still monotone of order $\infty$, Choquet's Theorem 43.4 \cite{choquet53} is also still applicable. The extreme elements of $S_{c}$ are therefore again \textit{exponentials} on $\mathcal{E}$. But since they belong to $S_{c}$, they must be \textit{continuous} exponentials. That is, if $\psi_{1} \leq \psi_{2} \leq \ldots$ is a monotone sequence in $\mathcal{E}$, then 
\begin{eqnarray}
e(\bigvee_{i=1}^{\infty} \psi_{i}) = \lim_{i \rightarrow \infty} e(\psi_{i}).
\nonumber
\end{eqnarray}
Since $e$ is a monotone $0$-$1$ function it follows that
\begin{eqnarray}
e(\bigvee_{i=1}^{\infty} \psi_{i}) = \prod_{i=1}^{\infty} e(\psi_{i}).
\nonumber
\end{eqnarray}
The set of extreme points $ext(S_{c})$ is again bounded and closed, hence compact. As above, define $I_{e} = \{\psi \in \mathcal{E}:e(\psi) = 1\}$. This time $I_{e}$ becomes a $\sigma$-ideal in $ \mathcal{E}$.

Consider a continuous support function $sp \in S_{c}$. Define, as above, $h_{\psi}(f) = f(\psi)$, a linear function from $V$ into $\mathbb{R}$. By Theorem \ref{th:KreinMilm2} there exists a probability measure $P$ on $ext(S_{c})$ such that  
\begin{eqnarray}
sp(\psi) = h_{\psi}(sp) = \int_{ext(S_{c})} h_{\psi}(e) dP(e) = \int_{ext(S_{c})} e(\psi) dP(e).
\nonumber
\end{eqnarray}
As above this gives
\begin{eqnarray}
sp(\psi) = P\{e \in ext(S_{c}):e(\psi) = 1\},
\nonumber
\end{eqnarray}
So, again as above, we may define a random mapping from the probability space $(ext(S_{c}),\mathcal{B}_{c},P)$ into the ideal completion $I_\Phi$ of the information algebra $\Phi$, by $\Gamma(e) = J_{e}$. Here $\mathcal{B}_{c}$ is the $\sigma$-field of Borel sets in $ext(S_{c})$. Note that in this case $J_{e}$ is a $\sigma$-ideal in $\Psi$. As above we verify that
\begin{eqnarray}
sp_{\Gamma}(\psi) = P(s_{\Gamma}(\psi)) = P\{e \in ext(S_{c}):e(\psi) = 1\} = sp(\psi)
\nonumber
\end{eqnarray}
for all $\psi \in \mathcal{E}$. So, $\Gamma$ is a random mapping generating the continuous support function $sp$ on $\mathcal{E}$. 

To conclude this part, we formulate the main result of this section in the following theorem.

\begin{theorem} \label{th:RanMapGenSpFct}
Let $(\Phi,\cdot,0,1;E)$ be an information algebra and $\mathcal{E} \subseteq \Phi$ a join-sub-semilattice of $(\Phi;\leq)$ under information order containing $1$. If $sp$ is a support function on $\mathcal{E}$, then there exists a probability space $(\Omega,\mathcal{A},P)$ and a random mapping $\Gamma$ from this space into the ideal completion of $I_\Phi$ of $\Phi,$, such that $\mathcal{E} \subseteq \mathcal{E}_\Gamma$ and its support function coincides on $\mathcal{E}$, with $sp$, that is $sp_{\Gamma}(\psi) = sp(\psi)$ for all $\psi \in \mathcal{E}$. 

If $(\Phi,\cdot,0,1;E)$ is a $\sigma$-information algebra, $\mathcal{E} \subseteq \Phi$ a $\sigma$-semilattice and $sp$ continuous, then there is a random mapping $\Gamma$ generating $sp$, as in the first part of the theorem, which maps to $\sigma$-ideals of $\Phi$.
\end{theorem}

We remark for completeness sake that for continuous support functions there is an alternative approach to generate them from a random mapping, due to \cite{norberg89}.

\section{Canonical support functions} \label{subsec:CanRandMap}

According to the previous Section \ref{sec:GenSpFct} any support function can be generated by some random mapping. In this section we are going to examine the random mappings generating a given support function in more detail. In particular, we shall compare these random mappings and single out a particular one, which we shall call the \textit{canonical} mapping.

Let $(\Phi,\cdot,0,1;E)$ be an information algebra and $\mathcal{E} \subseteq \Phi$ a join-sub-semilattice of $(\Phi,\leq)$, under information order containing $1$. Consider a support function $sp$ on $\mathcal{E}$. According to the discussion in Section  \ref{sec:GenSpFct} there is a probability space $(ext(S),\mathcal{A},P)$ on the set of exponentials $ext(S)$ on $\mathcal{E}$ and a random mapping into the \textit{ideal completion} of $(\Phi,\cdot,,0,1;E)$  defined by
\begin{eqnarray}
\nu(e) = J_{e} =  \{\psi \in \Phi: \psi \leq \phi \textrm{ for some}\ \phi \in I_{e}\}
\nonumber
\end{eqnarray}
where $I_{e}$ is the ideal $\{\psi \in \mathcal{E}:e(\psi) = 1\}$ in $\mathcal{E}$ associated with the exponential $e$. Then we obtain for $\psi \in\mathcal{E}$
\begin{eqnarray}
sp_\nu(\psi) = P\{e \in ext(S):e(\psi) = 1\},
\nonumber
\end{eqnarray}
which shows that the random mapping $\nu$ from $ext(S)$ into the ideal completion $I_{\Phi}$ of $\Phi$ indeed generates the support function on $\mathcal{E}$.

We noted in Section \ref{sec:GenSpFct} that there is a one-to-one relation between exponentials $e \in ext(S)$ on $\mathcal{E}$ and the ideals $I_{\mathcal{E}}$ in $\mathcal{E}$. To each exponential $e$ corresponds the ideal $I_{e}$ in $\mathcal{E}$ and conversely, any ideal $I$ of $\mathcal{E}$ defines an exponential $e_{I}$ by $e_{I}(\psi) = 1$, if $\psi \in I$ and $e_{I}(\psi) = 0$ otherwise. We may therefore replace the probability space $(ext(S),\mathcal{A},P)$ on $ext(S)$ by an equivalent probability space $(I_{\mathcal{E}},\mathcal{A},P)$ on $I_{\mathcal{E}}$. By abuse of notation we denote here the $\sigma$ fields and the probability measures in both spaces by the same symbol. The random mapping $\nu$ is then changed in the obvious way to
\begin{eqnarray}
\nu(I) = \{\psi \in \Phi: \psi \leq \phi \textrm{ for some}\ \phi \in I\}.
\nonumber
\end{eqnarray}
for any $I \in I_\Phi$.

We remarked in Section \ref{sec:GenSpFct} that this random mapping $\nu$ is not the only one inducing the support function on $\mathcal{E}$. Let's examine this in more detail. The restriction of an ideal $I$ of $\Phi$ to $\mathcal{E}$ is clearly an ideal of $\mathcal{E}$. We define the mapping $p : I_{\Phi} \rightarrow I_{\mathcal{E}}$ by $p(I) = I \vert \mathcal{E} = I \cap \mathcal{E}$; to each ideal in $\Phi$, we associate its restriction to $\mathcal{E}$. Then the inverse mapping $p^{-1}(I) = \{J \in I_{\Phi}:p(J) = I\}$ induces a partition of $I_{\Phi}$. Consider any ideal $J \in p^{-1}(I)$. Obviously we have $\nu(I) \subseteq J$ if $p(J) = I$. Thus, $\nu(I)$ is the \textit{least} ideal in $p^{-1}(I)$. 

Consider any random mapping $\Gamma$ from $I_{\mathcal{E}}$ into the ideal completion $I_{\Phi}$ of $\Phi$, such that $\Gamma(I) \in p^{-1}(I)$. Its allocation of support is, for $\psi \in \mathcal{E}$,
\begin{eqnarray}
s_{\Gamma}(\psi) = \{I \in I_{\mathcal{E}}:\psi \in \Gamma(I)\} =  \{I \in I_{\mathcal{E}}:\psi \in I\}.
\nonumber
\end{eqnarray}
It follows that the random mapping $\Gamma$ induces also the support function $sp$ on $ \mathcal{E}$,
\begin{eqnarray}
sp_{\Gamma}(\psi) = P(s_{\Gamma}(\psi)) = P\{I \in I_{\mathcal{E}}:\psi \in I\} = sp(\psi).
\nonumber
\end{eqnarray}
Hence, $\nu$ is the \textit{minimal} random mapping on $I_{\mathcal{E}}$ generating $sp$. 

Let's pursue this observation. Consider the probability algebra $(\mathcal{B},\mu)$ associated with the probability space $(I_{\mathcal{E}},\mathcal{A},P)$ (see Section \ref{subsec:RanMaps}). We remind that the mapping $\rho_{\nu} = \rho_{0} \circ s_{\nu}$ from $\Phi$ into $\mathcal{B}$ is an allocation of probability (a.o.p) (see Section \ref{subsec:RanMaps}). This a.o.p, as every a.o.p on $\Phi$, induces a support function $sp_{\nu} = \mu \circ \rho_{0} \circ s_{\nu}$ on $\Phi$ (see Theorem  \ref{th:SpFctPropAoP}), and its restriction to $\mathcal{E}$ equals $sp$. So, $sp_{\nu}$ is an extension of $sp$ to $\Phi$. Now, for any random mapping $\Gamma$ from $I_{\mathcal{E}}$ into the ideal completion $I_{\Phi}$ of $\Phi$, such that $\Gamma(I) \in p^{-1}(I)$, we have $\nu(I) \subseteq \Gamma(I)$. This implies for the allocations of support that $s_{\nu}(\psi) \subseteq s_{\Gamma}(\psi)$, hence $\rho_{\nu}(\psi) = \rho_{0} (s_{\nu}(\psi)) \leq \rho_{0}(s_{\Gamma}(\psi)) = \rho_{\Gamma}(\psi)$ and for $\psi \in \mathcal{E}$, we have $\rho_{\nu}(\psi) = \rho_{\Gamma}(\psi)$. It follows that
\begin{eqnarray}
sp_{\nu}(\psi) = \mu(\rho_{0}(s_{\nu}(\psi))) \leq \mu(\rho_{0}(s_{\Gamma}(\psi))) = sp_{\Gamma}(\psi)
\nonumber 
\end{eqnarray}
We shall see later (Section \ref{subsec:MinExt}) that the random mapping $\nu$ generates indeed the \textit{least} extension of the support function $sp$ on $\mathcal{E}$ to $\Phi$ among all extensions. But before we turn to this question, we return to the random mappings generating $sp$ on $\mathcal{E}$. 

Consider the family of sets $\{I \in I_{\mathcal{E}}:\psi \in I\}$ for $\psi \in \mathcal{E}$. All these sets belong to the $\sigma$-field $\mathcal{A}$ in the probability space $(I_{\mathcal{E}},\mathcal{A},P)$ used to define the random mapping $\nu$ to generate the support function $sp$ on $\mathcal{E}$ and $sp(\psi)=P(I \in I_{\mathcal{E}}:\psi \in I)$. Let $\mathcal{A}_{\mathcal{E}} \subseteq \mathcal{A}$ be the $\sigma$-field of subsets generated by the family of these subsets. Note that this set depends \textit{only} on the semi lattice $\mathcal{E}$, but not on $sp$ itself. Denote the restriction of the probability measure $P$ to $\mathcal{A}_{\mathcal{E}}$ by $P_{sp}$. This probability depends on the support function $sp$, and thereby indirectly of course also on $\mathcal{E}$. Consider the probability space $(I_{\mathcal{E}},\mathcal{A}_{\mathcal{E}},P_{sp})$. We remark that the random mapping $\nu$, as well as the related mappings $\Gamma$ considered above, still generate $sp$ on $\mathcal{E}$. 

In order to facilitate comparisons between random mappings generating the support function $sp$ on $\mathcal{E}$, we transport probability from the set of ideals $I_{\mathcal{E}}$ in $\mathcal{E}$ to the set $I_{\Phi}$ of ideals in $\Phi$. The family of sets $p^{-1}(A)$ for $A \in \mathcal{A}_{\mathcal{E}}$ forms a $\sigma$-field of subsets of $I_{\Phi}$ and by $P(p^{-1}(A)) = P_{sp}(A)$ a probability measure is defined on this $\sigma$-field. By abuse of notation, we denote the new probability space by $(I_{\Phi}, \mathcal{A}_{\mathcal{E}},P_{sp})$ and call it the \textit{canonical probability space} associated with $sp$. The random mapping $\nu$ from $I_{\mathcal{E}}$ into the ideal completion of $\Phi$ is redefined as $\nu(p(I))$ for $I \in I_{\Phi}$. Again, we call this new mapping $\nu$, that is,
\begin{eqnarray} \label{eq:DefCanRandMap}
\nu(I) = \{\psi \in \Phi: \psi \leq \phi \textrm{ for some}\ \phi \in p(I)\}.
\end{eqnarray}
We call this random mapping $\nu$, together with the associated probability space $(I_{\Psi},\mathcal{A}_{\mathcal{E}},P_{sp})$, the \textit{canonical random mapping} generating the support function $sp$ on the semilattice $\mathcal{E}$. Any other random mapping $\Gamma$ defined above on $I_{\mathcal{E}}$ may similarly be redefined as $\Gamma(p(I))$. 

We can now compare different extensions of support functions from $\mathcal{E}$. Consider semilattices $\mathcal{E}_{1}$ and $\mathcal{E}_{2}$ such that $\mathcal{E}_{1} \subseteq \mathcal{E}_{2} \subseteq \Phi$ and support functions $sp_{1}$ and $sp_{2}$ on $\mathcal{E}_{1}$ and $\mathcal{E}_{2}$ respectively, such that $sp_{2}$ is an extension of $sp_{1}$. Then, these support functions have their canonical random mappings $\nu_{1}$ and $\nu_{2}$ defined on the probability spaces $(I_{\Phi},\mathcal{A}_{\mathcal{E}_{1}},P_{sp_{1}})$ and $(I_{\Phi},\mathcal{A}_{\mathcal{E}_{2}},P_{sp_{2}})$ respectively. The next theorem shows how these canonical random mappings are related.

\begin{theorem} \label{th:CompCanRandMap}
Let $(\Phi,\cdot,0,1;E)$ be an information algebra and let $\nu_{1}$ and $\nu_{2}$, defined on the probability spaces $(I_{\Phi},\mathcal{A}_{\mathcal{E}_{1}},P_{sp_{1}})$ and $(I_{\Phi},\mathcal{A}_{\mathcal{E}_{2}},P_{sp_{2}})$, be the canonical random mappings associated with the support functions $sp_{1}$ and $sp_{2}$ on the semilattices $\mathcal{E}_{1} \subseteq \mathcal{E}_{2} \subseteq \Phi$. If $sp_{2}$ is an extension of $sp_{1}$, that is $sp_{1} = sp_{2} \vert \mathcal{E}_{2}$, then
\begin{enumerate}
\item $\nu_{1} \leq \nu_{2}$, in the order of the information algebra of random mappings into $I_\Phi$,
\item $\mathcal{A}_{\mathcal{E}_{1}} \subseteq \mathcal{A}_{\mathcal{E}_{2}}$,
\item $P_{sp_{1}} = P_{sp_{2}} \vert \mathcal{A}_{\mathcal{E}_{1}}$, on $\mathcal{A}_{\mathcal{E}_{1}}$ the two probability measures are equal.
\item $sp_{\nu_{1}}(\psi) \leq sp_{\nu_{2}}(\psi)$ for all $\psi \in \Phi$.
\end{enumerate}
\end{theorem}

\begin{proof}
(1) By definition we have $p_{1}(I) = I \vert \mathcal{E}_{1}$ and  $p_{2}(I) = I \vert \mathcal{E}_{2}$, hence 
$p_{1}(I) \subseteq p_{2}(I)$. Therefore, from (\ref{eq:DefCanRandMap}), we conclude that $\nu_{1}(I) \subseteq \nu_{2}(I)$ for all $I \in I_{\Psi}$, hence $\nu_{1} \leq \nu_{2}$. 

(2) Consider an element $\psi \in  \mathcal{E}_{1} \subseteq  \mathcal{E}_{2}$. Then, the allocations of support relative to $\nu_{1}$ and $\nu_{2}$, respectively, are
\begin{eqnarray}
s_{\nu_{1}}(\psi) &=&\{I \in I_{\Phi}:\psi \in \nu_{1}(I)\} = \{I \in I_{\Phi}:\psi \in I \vert  \mathcal{E}_{1}\},
\nonumber \\
s_{\nu_{2}}(\psi) &=&\{I \in I_{\Phi}:\psi \in \nu_{2}(I)\} = \{I \in I_{\Phi}:\psi \in I \vert  \mathcal{E}_{2}\}.
\nonumber
\end{eqnarray}
But $\psi \in I \vert  \mathcal{E}_{1}$ implies $\psi \in I \vert  \mathcal{E}_{2}$. On the other hand, $\psi \in \mathcal{E}_{1}$ and $\psi \in I \vert  \mathcal{E}_{2}$ implies $\psi \in I \vert  \mathcal{E}_{2} \cap \mathcal{E}_{1} = I \vert  \mathcal{E}_{1}$. So, we conclude that $s_{\nu_{1}}(\psi) = s_{\nu_{2}}(\psi)$ for every $\psi \in \mathcal{E}_{1}$. Since $\mathcal{A}_{\mathcal{E}_{1}}$ is the $\sigma$-field generated by the allocations $s_{\nu_{1}}(\psi)$ for $\psi \in \mathcal{E}_{1}$, and $\mathcal{A}_{\mathcal{E}_{2}}$ the one generated by $s_{\nu_{2}}(\psi)$ for $\psi \in \mathcal{E}_{2} \supseteq \mathcal{E}_{1}$, this shows that $\mathcal{A}_{\mathcal{E}_{1}} \subseteq \mathcal{A}_{\mathcal{E}_{2}}$.

(3) To prove this claim, we use Dynkin's Theorem \cite{billingsley95}. Dynkin calls a family of sets, closed under finite intersections, a $\pi$-system. The family $P$ of sets $s_{\nu_{1}}(\psi)$ for $\psi \in \mathcal{E}_{1}$ is a $\pi$-system (see Theorem \ref{th:ElPropAllSp}). The family $L$ of sets $A \in \mathcal{A}_{\mathcal{E}_{1}}$ for which
\begin{eqnarray}
P_{sp_{1}}(A) = P_{sp_{2}}(A)
\nonumber
\end{eqnarray}
is closed under complementation, and contains $\bigcup_{i} A_{i}$, if $A_{i}$ is a countable family of disjoint sets in $L$. This is called a $\lambda$-system by Dynkin. From the considerations above, we conclude that $P \subseteq L$. The theorem of Dynkin states that if $P$ is a $\pi$-system and $L$ a $\lambda$-system, then $P \subseteq L$ implies that the $\sigma$-closure of $P$ is contained in $L$, that is $\sigma(P) \subseteq L$. In our case the $\sigma$-closure of $P$ is $\mathcal{A}_{\mathcal{E}_{1}}$, hence we have $\mathcal{A}_{\mathcal{E}_{1}} \subseteq L$, where $L$ contains all sets of $\mathcal{A}_{\mathcal{E}_{1}}$ on which the two probabilities coincide. So, indeed for all $A \in \mathcal{A}_{\mathcal{E}_{1}}$ we have $P_{sp_{1}}(A) = P_{sp_{2}}(A)$.

(4) We have for any $\psi \in \Phi$ (see (\ref{eq:InnerProbExt})) $sp_{\nu_{1}}(\psi) = \mu(\rho(s_{\nu_1}(\psi))) \leq \mu(\rho(s_{\nu_2}(\psi))) = sp_{\nu_{2}}(\psi)$, because $s_{\nu_{1}}(\psi) \subseteq  s_{\nu_{2}}(\psi)$. Therefore, $sp_{\nu_{1}}(\psi) \leq sp_{\nu_{2}}(\psi)$.
\end{proof}

This theorem shows in particular, that the canonical random mapping associated with a support function $sp$ on a semilattice $\mathcal{E} \subseteq \Phi$ is \textit{unique}. It permits also to conclude that $sp_{\nu}$ is the \textit{least} extension of the support function $sp$ from $\mathcal{E}$ to $\Phi$. Indeed, suppose that $sp'$ is any extension of $sp$ to $\Phi$. Then, $sp'$ is generated by a canonical random mapping $\nu'$. According to Theorem \ref{th:CompCanRandMap} (4) we have then
\begin{eqnarray}
sp_{\nu}(\psi) \leq sp_{\nu'}(\psi) = sp'(\psi).
\nonumber
\end{eqnarray}
The last equity holds because $sp'$ is defined on $\Phi$. So, we have

\begin{corollary} \label{cor:LeastExt}
If $sp$ is a support function defined on a join-semilattice $\mathcal{E} \subseteq \Phi$, then $sp_{\nu}$ is the least extension of $sp$ to $\Phi$, that is, $sp_{\nu} \leq sp'$ for any support function $sp'$ on $\Psi$ such that $sp = sp' \vert \mathcal{E}$.
\end{corollary}

We remark that a similar analysis can be made for $\sigma$-semilattices or complete lattices $\mathcal{E}$ and continuous or condensable support functions $sp$. However, more interesting is the case of \textit{compact} information algebras $(\Phi,\Phi_f,\cdot,0,1;E)$. We consider a support function $sp$ defined on $\Phi_{f}$, the finite elements of $\Phi$, hence $\mathcal{E} = \Phi_{f}$. Since its ideal completion $I_{\Phi_{f}}$ is isomorphic to $\Phi$ (see Theorem \ref{th:IdCompFiniteEl}) we identify ideals $I$ of $\Phi_{f}$ with their suprema $\bigvee I \in \Phi$. For the support function $sp$, we consider its canonical probability space $(I_{\Phi_{f}},\mathcal{A}_{\Phi_{f}},P_{sp})$. .  

Beside the canonical random mapping, 
\begin{eqnarray}
\nu(I) = \{\psi \in \Phi:\psi \leq \phi \textrm{ for some}\ \phi \in I\}
\nonumber
\end{eqnarray}
we consider also the random mappings
\begin{eqnarray} \label{eq:SigmaRandMap}
\sigma(I) &=& \{\psi \in \Phi:\psi \leq \bigvee_{i=1}^{\infty} \psi_{i}, \psi_{i} \in I\},
\\ \label{eq:GammaRandMap}
\gamma(I) &=& \downarrow\!\bigvee I.
\end{eqnarray}
Both map $I_{\Phi_{f}}$ into $I_{\Phi}$. However, given the isomorphism between $I_{\Phi_f}$ and $\Phi$, we may also consider $\gamma$ as a map into $\Phi$, $\gamma(I) = \bigvee I$. Note also that $\nu \leq \sigma \leq \gamma$. We are going to examine the support functions on $\Phi$ induced by these random mappings.

We start with the random mapping $\sigma$. Here are its basic properties:

\begin{lemma} \label{le:SigmaRandMap}
Let $(\Phi,\Phi_f,\cdot,0,1;E)$ be a compact information algebra with finite elements $\Phi_{f}$ and $\sigma$ the random map defined by (\ref{eq:SigmaRandMap}). Then for an ideal $I \in I_{\Phi_f}$,
\begin{enumerate}
\item the ideal $\sigma(I)$ is a $\sigma$-ideal in $\Phi$,
\item its restriction to $\Phi_{f}$ equals $I$, $\sigma(I) \cap \Phi_{f} = I$,
\item the $\sigma$-ideal $\sigma(I)$ is minimal among all $\sigma$-ideals in $\Phi$ extending $I$.
\end{enumerate}
\end{lemma}

\begin{proof}
(1) Consider the elements $\psi_{1},\psi_{2},\ldots \in \sigma(I)$,. Then we have $\psi_{i} \leq \bigvee_{j=1}^{\infty} \psi_{i,j}$ with $\psi_{i,j} \in I$ for all $i = 1,2\ldots$ and $j = 1,2, \ldots$. But then we obtain
\begin{eqnarray}
\bigvee_{i=1}^{\infty} \psi_{i} \leq \bigvee_{i=1}^{\infty} \bigvee_{j=1}^{\infty} \psi_{i,j} = \bigvee_{h=1}^{\infty} \psi'_{h},
\nonumber
\end{eqnarray}
where $\psi'_{h} = \vee_{i=1}^{h} \vee_{j=1}^{i} \psi_{i,j} \in I$. This shows that $\bigvee_{i=1}^{\infty}\psi_{i}  \in \sigma(I)$, hence $\sigma(i)$ is indeed a $\sigma$-ideal in $\Phi$.

(2) Assume that $\psi \in \sigma(I)$ and $\psi \in \Phi_{f}$. Then $\psi \leq \bigvee_{i=1}^{\infty} \psi_{i}$, with $\psi_{i} \in I$ for $i=1,2,\ldots$. By the usual transformation, we may always assume that $\psi_{1} \leq \psi_{2} \leq \ldots$. This monotone sequence is a directed set in $\Psi$. By compactness there exists a $\psi_{i}$ such that $\psi \leq \psi_{i}$. This shows that $\psi \in I$. But $I \subseteq \sigma(I)$, therefore we see that indeed the restriction of the ideal $\sigma(I)$ to $\Phi_{f}$ equals $I$.

(3) Consider a $\sigma$-ideal $J$ whose restriction to $\Phi_{f}$ equals $I$. Assume $\psi \in \sigma(I)$. Then $\psi \leq \bigvee_{i=1}^{\infty} \psi_{i}$, with $\psi_{i}$ in $I$, hence in $J$. But then $\bigvee_{i=1}^{\infty} \psi_{i} \in J$ since $J$ is a $\sigma$-ideal, therefore $\psi \in J$. This shows that $\sigma(I) \subseteq J$. Hence $\sigma(I)$ is indeed minimal among the $\sigma$-ideals extending $I$.
\end{proof}

The random map $\sigma$ generates a support function $sp_{\sigma} = \mu \circ \rho_{\sigma}$ on $\Phi$, where $(\mu,\mathcal{B})$ is the probability algebra associated with the probability space $(I_{\Phi_{f}},\mathcal{A}_{\Phi_{f}},P_{sp})$, and $\rho_{\sigma} = \rho_{0} \circ s_{\sigma}$. We are going to show that $sp_{\sigma}$ is a \textit{continuous} extension of $sp$. The key is the following lemma:

\begin{lemma} \label{le:SigmaContSpFct}
Let $(\Phi,\Phi_f,\cdot,0,1;E)$ be a compact information algebra  with finite elements $\Phi_{f}$, $\sigma$ the random map defined by (\ref{eq:SigmaRandMap}), and $s_{\sigma}$ the allocation of support for the random map $\sigma$. Then, if $\psi_{i} \in \Phi$ for $i=1,2,\ldots$,
\begin{eqnarray}
s_{\sigma}(\bigvee_{i=1}^{\infty} \psi_{i}) = \bigcap_{i=1}^{\infty} s_{\sigma}(\psi_{i}).
\nonumber
\end{eqnarray}
\end{lemma}

\begin{proof}
Since $\Phi$ is a complete lattice, $\bigvee_{i=1}^{\infty} \psi_{i} \in \Phi$, and
\begin{eqnarray}
s_{\sigma}(\bigvee_{i=1}^{\infty} \psi_{i}) = \{I \in I_{\Phi_{f}}:\bigvee_{i=1}^{\infty} \psi_{i} \leq \bigvee_{i=1}^{\infty} \phi_{i},\phi_{i} \in I\}.
\nonumber
\end{eqnarray}
If $I \in s_{\sigma}(\bigvee_{i=1}^{\infty} \psi_{i})$, then clearly $I \in s_{\sigma}(\psi_{i})$ for all $i=1,2,\ldots$. Conversely, assume $I \in s_{\sigma}(\psi_{i})$ for all $i=1,2,\ldots$. Then we have $\psi_{i} \leq \bigvee_{j=1}^{\infty} \psi_{i,j}$ with $\psi_{i,j} \in I$. This implies in the same way as in the proof of Lemma \ref{le:SigmaRandMap} that $\bigvee_{i=1}^{\infty} \psi_{i} \in \sigma(I)$, hence $I \in s_{\sigma}(\bigvee_{i=1}^{\infty} \psi_{i})$ and this proves the lemma.
\end{proof}

As a consequence of this lemma, we find that
\begin{eqnarray}
\rho_{\sigma}(\bigvee_{i=1}^{\infty} \psi_{i}) &=&\rho_{0}(s_{\sigma}(\bigvee_{i=1}^{\infty} \psi_{i}))
= \rho_{0}(\bigcap_{i=1}^{\infty} s_{\sigma}(\psi_{i}))
\nonumber \\
&=&\bigwedge_{i=1}^{\infty} \rho_{0}(s_{\sigma}(\psi_{i})) = \bigwedge_{i=1}^{\infty} \rho_{\sigma}(\psi_{i}).
\end{eqnarray}
The allocation of probability $\rho_{\sigma}$ is a $\sigma$-a.o.p. By Theorem \ref{th:SpFctPropAoP} $sp_{\sigma}$ is a \textit{continuous} support function extending $sp$ on $\Phi_{f}$ to $\Phi$. Since $\sigma(I)$ is the least $\sigma$-ideal among all $\sigma$-ideals extending the ideal $I$ of $\Phi_{f}$ to $\Phi$, we conclude that $sp_{\sigma}$ is also the minimal continuous support function among all continuous support functions $sp$ extending $sp$ from $\Phi_{f}$ to $\Psi$,
\begin{eqnarray}
sp_{\sigma} \leq \tilde{sp}(\psi), \textrm{ if}\ \tilde{sp} \textrm{ continuous},\ \tilde{sp} \vert \Psi_{f} = sp
\nonumber
\end{eqnarray}
for all $\psi \in \Phi$. 

Let's fix this result in the following theorem:

\begin{theorem} \label{th:MinContExt}
Let $(\Phi,\Phi_f,\cdot,0,1;E)$ be a compact information algebra, with finite elements $\Phi_{f}$, $sp$ a support function defined on $\Phi_{f}$ and $\sigma$ the random map defined by (\ref{eq:SigmaRandMap}). Then, if $(\mu,\mathcal{B})$ is the probability algebra associated with the canonical probability space $(I_{\Phi_{f}},\mathcal{A}_{\Phi_{f}},P_{sp})$ and $\rho_{\sigma} = \rho_{0} \circ s_{\sigma}$, then $sp_{\sigma} = \mu \circ \rho_{\sigma}$ is the minimal continuous extension of $sp$ to $\Phi$ among all continuous extensions.
\end{theorem}
 
We turn to the random mapping $\gamma$, defined in (\ref{eq:GammaRandMap}). This mapping is characterised as follows:

\begin{lemma} \label{le:GammaRandMap}
Let $(\Phi,\Phi_f,\cdot,0,1;E)$ be a compact information algebra, with finite elements $\Phi_{f}$ and $\gamma$ the random mapping defined by (\ref{eq:GammaRandMap}). Then the ideal $\gamma(I)$ is the minimal complete ideal in $\Phi$ whose restriction to $\Phi_{f}$ equals $I$, $\gamma(I) \cap \Phi_{f} = I$.
\end{lemma}

\begin{proof}
We have $I \subseteq \downarrow\!\bigvee I \cap \Phi_{f}$. Consider then an element $\psi \in \downarrow\!\bigvee I \cap \Phi_{f}$. From $\psi \leq \bigvee I$ it follows, since $I$ is a directed set, by compactness that there is a $\chi \in I$ such that $\psi \leq \chi$. But then $\psi \in \Phi_{f}$ implies $\psi \in I$. This proves that $\gamma(I) \cap \Phi_{f} = I$.

As a principal ideal in a complete lattice, $\gamma(I)$ is a complete ideal. Consider any other complete ideal $J$, whose restriction to $\Phi_{f}$ equals $I$. But then $\bigvee I \leq \bigvee J$ and $J =\ \downarrow\!\bigvee J$, hence $\gamma(I) \subseteq J$. This proves the minimality of $\gamma(I)$. 
\end{proof}

Consider now \textit{simple} random variables $\Delta$ on the canonical probability space $(I_{\Phi_{f}},\mathcal{A}_{\Phi_{f}},P_{sp})$. Any such random variable is defined by a measurable partition $B_{i} \in \mathcal{A}_{\Phi_{f}}$, $i=1,\ldots,m$ of $I_{\Phi_{f}}$ and $\Delta(I) = \psi_{i} \in \Phi_{f}$ if $I \in B_{i}$. Note that $\Delta \leq \gamma$ if and only if $\psi_{i} \leq \vee I$ for $I \in B_{i}$ and $i=1,\ldots,m$. This leads to the following result in which we consider $\gamma$ to be a map into $\Phi$.

\begin{lemma} \label{le:GammaAsGenRandVar}
The random mapping $\gamma$ defined by (\ref{eq:GammaRandMap}) is a random variable,
\begin{eqnarray}
\gamma = \bigvee\{\Delta:\Delta \textrm{ simple random variable}, \Delta \leq \gamma\}.
\nonumber
\end{eqnarray}
\end{lemma}

\begin{proof}
We claim that for all $I \in I_{\Phi_{f}}$ we have $\gamma(I) = \bigvee \{\Delta(I):\Delta \leq \gamma\}$ where it is understood that $\Delta$ denotes a simple random variable. Clearly $\gamma(I) \geq \bigvee \{\Delta(I):\Delta \leq \gamma\}$. To prove the converse inequality, consider $I \in I_{\Phi_{f}}$. Then we have by density $\gamma(I) =\ \downarrow\!\bigvee\{\psi \in \Phi_{f}:\psi \leq \bigvee I\}$. By Lemma \ref{le:GammaRandMap} $\psi \in I$ if and only if $\psi \in \gamma(I)$ and $\psi \in \Phi_f$. Define, for a $\psi \in I$,
\begin{eqnarray}
\Delta_{\psi}(I) =  \left\{ \begin{array}{ll}
\psi & \textrm{if}\ \psi \in I, \\
1 & \textrm{otherwise}. \end{array} \right.
\nonumber
\end{eqnarray}
The set $\{I:\psi \in I\}$ is measurable (belongs to $\mathcal{A}_{\Phi_{f}}$), hence $\Delta_{\psi}$ is a simple random variable and $\Delta_{\psi}(I) \leq \bigvee I$, hence $\Delta_{\psi} \leq \gamma$. Thus, we obtain 
\begin{eqnarray}
\gamma(I) = \bigvee\{\Delta_{\psi}(I):\psi \in I\} \leq \bigvee\{\Delta(I):\Delta \leq \gamma\} \leq \gamma(I).
\nonumber
\end{eqnarray}
This proves the identity $\gamma(I) = \bigvee \{\Delta(I):\Delta \leq \gamma\}$, hence the lemma.
\end{proof}

From this lemma it follows according to Theorem \ref{thCondOfGenRV} that for a directed subset $D$ of $\Phi$
\begin{eqnarray}
\rho_{\gamma}(\bigsqcup D) = \bigwedge_{\psi \in D} \rho_{\gamma}(\psi).
\nonumber
\end{eqnarray}
Further, from Theorem \ref{th:SpFctPropAoP} it follows that 
\begin{eqnarray}
sp_{\gamma}(\bigsqcup D) = \inf_{\psi \in D} sp_\gamma(\psi).
\nonumber
\end{eqnarray}
This implies also that for any $\psi \in \Phi$,
\begin{eqnarray} \label{eq:GammaCondExt}
sp_{\gamma}(\psi) = \inf\{ sp(\phi):\phi \in \Phi_{f},\phi \leq \psi\}.
\end{eqnarray}
This means that $sp_{\gamma}$ is the unique \textit{condensable} extension of $sp$ from $\Phi_f$ to $ \Phi$. We note also that according to Theorem \ref{th:CompCanRandMap}, since $\nu \leq \sigma \leq \gamma$, we have $sp_{\nu}(\psi) \leq sp_{\sigma}(\psi) \leq sp_{\Gamma}(\psi)$. These results (Theorem \ref{th:MinContExt} and (\ref{eq:GammaCondExt})) partly answer an open question posed in \cite{shafer79}. In this work it was shown that continuous and condensable extensions always exist if $\mathcal{E}$ is a subset lattice. Here it is shown that they always exist if $\mathcal{E}$ corresponds to the finite elements of a compact information algebra, independently whether $\Psi_{f}$ is a lattice or not.

We summarise these results in the following theorem.

\begin{theorem}
Let $(\Phi,\Phi_f,\cdot,0,1;E)$ be a compact information algebra, with finite elements $\Phi_f$, $sp$ a support function defined on $\Phi_f$ and $\gamma$ the random map defined by (\ref{eq:GammaRandMap}). If $(\mu,\mathcal{B})$ is the probability algebra associated with the canonical probability space $(I_{\Phi_{f}},\mathcal{A}_{\Phi_{f}},P_{sp})$ and if $\rho_\gamma = \rho_0 \circ s_\gamma$, then $sp_\gamma = \mu \circ \rho_\gamma$ is the unique condensable extension of $sp$ to $\Phi$.
\end{theorem}

We conclude by proving the converse of Theorem \ref{th:SpFctPropAoP} and thus characterizing continuous and condensable support functions by their associated allocations of support. 

\begin{theorem} \label{th:SpFctPropAoP1}
\begin{enumerate}
\item If $(\Phi,\leq)$ is a $\sigma$-semilattice under information orderr, then $sp = \mu \circ \rho$ is continuous on $\Phi$ if and only if $\rho$ is a $\sigma$-allocation of probability, that is for $\psi_{i} \in \Phi$, $i= 1,2,\ldots$
\begin{eqnarray} \label{eq:CharSigmaAoP}
\rho(\bigvee_{i=1}^{\infty} \psi_{i}) = \bigwedge_{i=1}^{\infty} \rho(\psi_{i}).
\end{eqnarray}
\item If $(\Phi,\leq)$ is a complete lattice under information order, then $sp = \mu \circ \rho$ is condensable on $\Phi$ if and only if for any directed set $D \subseteq \Phi$,
\end{enumerate}
\begin{eqnarray} \label{eq:CharCondAoP}
\rho(\bigsqcup D) = \bigwedge_{\psi \in D} \rho(\psi).
\end{eqnarray}
\end{theorem}

\begin{proof}
The if-part of both parts is already proved in Theorem \ref{th:SpFctPropAoP}, it remains thus only to prove the only-if-part

(1) Consider a countable set of elements $\psi_{1},\psi_{2},\ldots \in \Phi$. We may always replace this sequence by a monotone sequence $\psi'_{1} \leq \psi'_{2} \leq \ldots$ having the same supremum, $\bigvee_{i=1}^{\infty} \psi_{i} = \bigvee_{i=1}^{\infty} \psi'_{i}$, by defining $\psi'_{i} = \bigvee_{j=1}^{i} \psi_{i}$. Then $\rho(\psi'_{1}) \geq \rho(\psi'_{2}) \geq \ldots$ is downwards directed. Therefore, by the continuity of $sp$ and Lemma \ref{downward},
\begin{eqnarray}
sp(\bigvee_{i=1}^{\infty} \psi_{i}) &=&sp(\bigvee_{i=1}^{\infty} \psi'_{i}) =\lim_{i \rightarrow \infty} sp(\psi'_{i})
\nonumber \\
&=&\lim_{i \rightarrow \infty} \mu(\rho(\psi'_{i})) = \mu(\bigwedge_{i=1}^{\infty} \rho(\psi'_{i}))
= \mu(\bigwedge_{i=1}^{\infty} \rho(\psi_{i})).
\nonumber
\end{eqnarray}
From $sp(\bigvee_{i=1}^{\infty} \psi_{i}) = \mu(\rho(\bigvee_{i=1}^{\infty} \psi_{i}))$ it follows that $\mu(\rho(\bigvee_{i=1}^{\infty} \psi_{i})) = \mu(\bigwedge_{i=1}^{\infty} \rho(\psi_{i}))$. Since $\bigwedge_{i=1}^{\infty} \rho(\psi_{i}) \geq \rho(\bigvee_{i=1}^{\infty} \psi_{i})$ and $\mu$ is a positive measure, it follows that   $\bigwedge_{i=1}^{\infty} \rho(\psi_{i}) = \rho(\bigvee_{i=1}^{\infty} \psi_{i})$.

(2) Let $D \subseteq \Phi$ be directed. By the condensability of $sp$ we obtain
\begin{eqnarray}
\mu(\rho(\bigsqcup D)) = sp(\bigsqcup D) = \inf_{\psi \in D} sp(\psi) =  \inf_{\psi \in D} \mu(\rho(\psi)).
\nonumber
\end{eqnarray}
Since the set $\{\rho(\psi):\psi \in D\}$ is downwards directed, we get from Lemma  \ref{downward} that $\inf_{\psi \in D} \mu(\rho(\psi)) = \mu(\bigwedge_{\psi \in D} \rho(\psi))$, hence $\mu(\rho(\bigsqcup D)) = \mu(\bigwedge_{\psi \in D} \rho(\psi))$. Since $\bigwedge_{\psi \in D} \rho(\psi) \geq \rho(\bigsqcup D)$, we conclude that $\bigwedge_{\psi \in X} \rho(\psi) = \rho(\bigsqcup D)$.
\end{proof}

If $(\Phi,\Phi_f,\cdot,0,1;E)$ is a compact information algebra and $sp = \mu \circ \rho$ condensable on $\Phi$, then (\ref{eq:CharCondAoP}) implies also that for all $\psi \in \Phi$
\begin{eqnarray}
\rho(\psi) = \bigwedge\{\rho(\psi):\psi \in \Phi_{f},\phi \leq \psi\}.
\nonumber
\end{eqnarray}

We are going to study these different extensions of a support functions from a part of $\Phi$ to the whole of $\Phi$ in the next section from a different angle.

To conclude this section, consider an a.o.p $\rho$ defined on $\Phi$. It is generated by some random mapping $\Gamma$ into the ideal completion $I_{\Phi}$ of $\Phi$. However, this map is not unique as we have seen. This confirms a former remark, that a random map $\Gamma$ contains more information than its associated a.o.p $\rho_{\Gamma}$. This explains why the map $\Gamma \mapsto \rho_{\Gamma}$ is, in general, not a homomorphism (see the end of Section \ref{subsec:AoPAndRV}).


\section{Minimal extensions	} \label{subsec:MinExt}

In the previous section, we have found an extension $sp_{\nu}$ for any support function $sp$ on some join-sub-semilattice $\mathcal{E}$ of an information algebra $(\Phi,\cdot,0,1;E)$ to the whole of the algebra. This extension is defined in terms of the canonical random mapping associated with $sp$. In this section, we shall show how the extension $sp_{\nu}$ and other extensions can be defined explicitly in terms of the support function $sp$ on $\mathcal{E}$. The following theorem is an extension to information algebras of a result due to \cite{shafer73} for set algebras. 

\begin{theorem} \label{th:ExplLeastExt}
If $sp$ is a support function defined on a join-semilattice $\mathcal{E} \subseteq \Phi$, where $(\Phi,\cdot,0,1;E)$ is an information algebra, then
\begin{eqnarray} \label{eq:ExplLeastExt}
sp_{\nu}(\phi) = 
\sup \left\{ \sum_{\emptyset \not= I \subseteq \{1,\ldots,n\}} (-1)^{\vert I \vert + 1} 
sp(\vee_{i \in I} \psi_{i})  \right\}
\end{eqnarray}
where the supremum is to be taken over all finite sets $I$ of elements $\psi_{1},\ldots,\psi_{n} \geq \phi$, $n=1,2,\ldots $ with $\psi_{1},\ldots,\psi_{n} \in \mathcal{E}$.
\end{theorem}

\begin{proof}
Let $f$ denote the function on the right hand side of (\ref{eq:ExplLeastExt}). We remark that $f$ is equal to $sp$ on $\mathcal{E}$ (compare Theorem \ref{th:SpFctProp}). Note also that $f$ is less or at most equal to $sp_{\nu}$, since the latter, as a support function on $\Phi$, is monotone of order $\infty$. Therefore, it is sufficient to show that $f$ is a support function on $\Phi$, because then, according to Corollary \ref{cor:LeastExt} it must be greater or equal to $sp_{\nu}$, so that $sp_{\nu} = f$ as claimed.

In order to prove $f$ to be a support function, we use, following \cite{shafer73} allocations of probability. Let $\rho$ be the allocation of support associated with the canonical random mapping generating $sp$, such that for $\psi \in \mathcal{E}$,
\begin{eqnarray}
sp(\psi) = \mu(\rho(\psi)),
\nonumber
\end{eqnarray}
where $\mu$ is the probability of the probability algebra $(\mathcal{B},\mu)$ associated with the probability space $(I_{\Phi},\mathcal{A}_{\mathcal{E}},P_{sp})$ of the canonical probability space associated with the support function $sp$ on $\mathcal{E}$ (Section \ref{subsec:CanRandMap}). Further, $\rho = \rho_{0} \circ s_{\nu}$ (see Section \ref{subsec:RanMaps}). Define for $\phi \in \Phi$,
\begin{eqnarray}
\bar{\rho}(\phi) = \bigvee \{\rho(\psi):\psi \in \mathcal{E},\phi \leq \psi\}.
\end{eqnarray}
We are going to show that $\bar{\rho}$ is an a.o.p on $\Phi$. Obviously, for $\psi \in \mathcal{E}$, we have $\bar{\rho}(\psi) = \rho(\psi)$, hence in particular $\bar{\rho}(1) = \top$. Consider $\phi_{1},\phi_{2} \in \Phi$. Then $\phi_{1},\phi_{2} \leq \phi_{1} \cdot \phi_{2}$, hence $\bar{\rho}(\phi_{1}),\bar{\rho}(\phi_{2}) \geq \bar{\rho}(\phi_{1} \cdot \phi_{2})$ or $\bar{\rho}(\phi_{1}) \wedge \bar{\rho}(\phi_{2}) \geq \bar{\rho}(\phi_{1} \cdot \phi_{2})$. On the other hand, let $\phi_{1} \leq \psi_{1} \in \mathcal{E}$ and $\phi_{2} \leq \psi_{2} \in \mathcal{E}$. Then, $\psi_{1} \cdot \psi_{2} \in \mathcal{E}$ and $\phi_{1} \cdot \phi_{2} \leq \psi_{1} \cdot \psi_{2}$ such that $\rho(\psi_{1}) \wedge \rho(\psi_{2}) = \rho(\psi_{1} \cdot \psi_{2}) \leq \bar{\rho}(\psi_{1} \cdot \psi_{2})$. It follows that
\begin{eqnarray}
\bar{\rho}(\phi_{1} \cdot \phi_{2}) &\geq&\bigvee \{\rho(\psi_{1}) \wedge \rho(\psi_{2}): \phi_{1} \leq \psi_{1},\phi_{2} \leq \psi_{2},\psi_{1},\psi_{2} \in \mathcal{E} \}
\nonumber \\
&=&\left( \bigvee \{\rho(\psi_{1}):\phi_{1} \leq \psi_{1} \in \mathcal{E}\}  \right) \wedge \left( \bigvee \{\rho(\phi_{2}):\psi_{2} \leq \psi_{2} \in \mathcal{E}\}  \right)
\nonumber \\
&=&\bar{\rho}(\phi_{1}) \wedge \bar{\rho}(\phi_{2}).
\nonumber 
\end{eqnarray}
So, we conclude that $\bar{\rho}(\phi_{1} \cdot \phi_{2}) = \bar{\rho}\phi_{1}) \wedge \bar{\rho}(\phi_{2})$ and that, therefore, $\bar{\rho}$ is an a.o.p.

In the formula (\ref{eq:ExplLeastExt}) for $f$, we may replace $sp$ by $\mu \circ \rho$,
\begin{eqnarray} \label{eq:SupRepr}
f(\phi) &=& \sup \left\{ \sum_{\emptyset \not= I \subseteq \{1,\ldots,n\}} (-1)^{\vert I \vert + 1} 
\mu(\rho(\vee_{i \in I} \psi_{i}))  \right\}
\nonumber \\
&=&\sup \left\{ \sum_{\emptyset \not= I \subseteq \{1,\ldots,n\}} (-1)^{\vert I \vert + 1} 
\mu(\wedge_{i \in I} \rho(\psi_{i}))  \right\}
\nonumber \\
&=&\sup \left\{ \mu(\vee_{i=1}^{n} \rho(\psi_{i})) \right\}
\end{eqnarray}
by the inclusion-exclusion-formula of probability theory. The supremum ranges over the same range as in (\ref{eq:ExplLeastExt}). The family of elements $\vee_{i=1}^{n} \rho(\psi_{i})$ in this supremum forms an upwards directed set in $\mathcal{B}$. By Lemma \ref{downward} we obtain therefore
\begin{eqnarray}
f(\psi)&=& \mu( \bigvee \{\vee_{i=1}^{n} \rho(\psi_{i}): \psi_{i} \in \mathcal{E},\psi_{i} \geq \phi,i=1,\ldots,n;n=1,2,\ldots\})
\nonumber \\
&=&\mu( \bigvee \{ \rho(\psi): \psi \in \mathcal{E},\psi \geq \phi\})
\nonumber \\
&=&\mu(\bar{\rho}(\psi))
\nonumber
\end{eqnarray}
Here, the associate law for joins in a complete lattice is used. Since $\bar{\rho}$ is an a.o.p, $f = \mu \circ \bar{\rho}$ is a support function on $\Phi$ (see Theorem \ref{th:SpFctPropAoP}). This concludes the proof.
\end{proof}

In the proof above we used the a.o.p $\rho$ associated with the support function $sp$ on $\mathcal{E}$. We remind that $sp_{\nu} = \mu \circ \rho = \mu \circ \rho_{0} \circ s_{\nu}$. On the other hand the a.o.p $\bar{\rho}$ generates $f$, that is $f = \mu \circ \bar{\rho}$. From $sp_{\nu} = f$, as stated in the theorem, we deduce as a complement that $\rho = \rho_{0} \circ s_{\nu} = \bar{\rho}$.  In fact, we have seen that for $\psi \in \mathcal{E}$ we have $\rho(\psi) = \bar{\rho}(\psi)$ and for any $\phi \in \Phi$, $\phi \leq \psi \in \mathcal{E}$ implies $\rho(\psi) \leq \rho(\phi)$, hence $\bar{\rho}(\phi) \leq \rho(\phi)$. Then we have $\rho(\phi) = \bar{\rho}(\phi) \vee (\rho(\phi) - \bar{\rho}(\phi))$. It follows that
\begin{eqnarray}
sp_{\nu}(\phi) = \mu(\rho(\phi)) = \mu(\bar{\rho}(\phi)) + \mu((\rho(\phi) - \bar{\rho}(\phi))
\nonumber
\end{eqnarray}
But from $sp_{\nu}(\phi) = f(\phi) = \mu(\bar{\rho}(\phi))$ we deduce that $\mu(\rho(\phi) - \bar{\rho}(\phi)) = 0$, hence $\rho(\phi) - \bar{\rho}(\phi) = \bot$. Since $\bar{\rho}(\phi) \leq \rho(\phi)$ this means that indeed $\bar{\rho}(\phi) = \rho(\phi)$. We may rephrase this result in the following Corollary.

\begin{corollary} \label{cor:CanAoP}
If $\rho = \rho_{0} \circ s_{\nu}$ is the allocation of probability associated with the support function $sp_{\nu} = \mu \circ \rho$, which is the least extension of the support function $sp$ on $\mathcal{E}$, then
\begin{eqnarray}
\rho(\phi) = \vee \{\rho(\psi):\psi \in \mathcal{E},\phi \leq \psi\}.
\nonumber
\end{eqnarray}
\end{corollary}

If the support function $sp$ is defined on a \textit{lattice} $\mathcal{E}$, then Theorem \ref{th:ExplLeastExt} may be sharpened \cite{shafer73}.

\begin{theorem} \label{th:SpFctOnLat}
If $sp$ is a support function defined on a lattice $\mathcal{E} \subseteq \Phi$, then
\begin{eqnarray} \label{eq:ExplLeastExt2}
sp_{\nu}(\phi) = \sup \{sp(\psi):\psi \in \mathcal{E},\phi \leq \psi\}.
\end{eqnarray}
\end{theorem}

\begin{proof}
Since $sp_{\nu}$ is monotone, the right hand side of (\ref{eq:ExplLeastExt2}) is less or equal to $sp_{\nu}$. It remains to show the converse inequality. Again, let $\rho = \rho_{0} \circ s_{\nu}$ be the a.o.p associated with the support function $sp$ and $\mu$ the probability in the corresponding probability algebra $(\mathcal{B},\mu)$. Consider $\psi_{1},\ldots,\psi_{n} \in \mathcal{E}$. Since $\mathcal{E}$ is a lattice, $\wedge_{i=1}^{n} \psi_{i}$ belongs to $\mathcal{E}$ too. Note that
\begin{eqnarray*}
\lefteqn{s_{\nu}(\wedge_{i=1}^{n} \psi_{i}) = \{I \in I_{\Phi}:\wedge_{i=1}^{n} \psi_{i} \in \nu(I)\} } \\
&&\supseteq \cup_{i=1}^{n} \{I \in I_{\Phi}:\psi_{i} \in \nu(I)\} = \cup_{i=1}^{n} s_{\nu}(\psi_{i}).
\end{eqnarray*}
Therefore,
\begin{eqnarray}
\rho(\wedge_{i=1}^{n} \psi_{i}) &=& [s_{\nu}(\wedge_{i=1}^{n} \psi_{i})] \geq [ \cup_{i=1}^{n} s_{\nu}(\psi_{i})]
\nonumber \\
&=& \vee_{i=1}^{n} [s_{\nu}(\psi_{i})] = \vee_{i=1}^{n} \rho(\psi_{i}).
\nonumber
\end{eqnarray}
Here $[A]$ denotes, as usual, the projection of $A \in \mathcal{A}_{\mathcal{E}}$ to the associated Boolean algebra $\mathcal{B}$ in the probability algebra $(\mathcal{B},\mu)$, see Section \ref{subsec:RanMaps}. Using (\ref{eq:SupRepr}) in the proof of Theorem \ref{th:ExplLeastExt} and $sp_{\nu}(\phi) = f(\phi)$, we obtain now  
\begin{eqnarray}
sp_{\nu}(\phi) = \sup \{ \mu(\vee_{i=1}^{n} \rho(\psi_{i}) \} \leq \sup \{\mu(\rho(\wedge_{i=1}^{n}\psi_{i})\},
\nonumber \\
\end{eqnarray}
where the supremum ranges over $\psi_{i} \in \mathcal{E}$, $\phi \leq \psi_{i}$, $i=1,\ldots,n$ and $n =1,2,\ldots$. Recall that $\wedge_{i=1}^n \in \mathcal{E}$, if $\psi_i \in \mathcal{E}$. Therefore it follows that 
\begin{eqnarray}
sp_{\nu}(\phi) \leq \sup \{ \mu(\rho(\psi)):\psi \in \mathcal{E},\phi \leq \psi\} = \sup \{sp(\psi):\psi \in \mathcal{E},\phi \leq \psi\} .
\nonumber
\end{eqnarray}
This concludes the proof.
\end{proof}

There are in particular several examples of compact information algebras where the finite elements form a lattice, hence where Theorem \ref{th:SpFctOnLat} applies if $\mathcal{E} = \Phi_{f}$.

We have seen in Section \ref{subsec:CanRandMap}, that support functions $sp$, defined on the finite elements $\Phi_{f}$ of a compact information algebra $\Phi$ may be extended either to a continuous support function $sp_{\sigma}$ or to a condensable one $sp_{\gamma}$. Further, by definition of condensability, $sp_{\gamma}$ is determined by the values of $sp$ on $\Phi_{f}$. This is like $sp_{\nu}$, which according to Theorem \ref{th:ExplLeastExt} is also determined by the values of $sp$ on $\Phi_{f}$, if $\mathcal{E} = \Psi_{f}$. Does a similar result also hold for the continuous extension $sp_{\sigma}$? Yes, but as far as we know, only for a very special case, namely if $\mathcal{E}$ is a \textit{distributive lattice}, see \cite{shafer79}, Theorem 4. The following theorem is a particular case of Shafer's result, a case of special interest for us, where we assume that the finite elements form a distributive lattice, like the cofinite elements in a subset algebra.

\begin{theorem} \label{th:ContExtFromLat}
Let $(\Phi,\Phi_f,\cdot,0,1;E)$ be a compact information algebra, with finite elements $\Phi_{f}$  and $(\Phi_{f},\leq)$ a distributive lattice. If $sp$ is a support function defined on $\Phi_{f}$, then for all $\phi \in \Phi$,
\begin{eqnarray} \label{eq:ContExtFromLat}
sp_{\sigma}(\phi) = \sup\{\lim_{i \rightarrow \infty} sp(\psi_{i}):\psi_{1} \leq \psi_{2} \leq \ldots \in \Phi_{f},
\bigvee_{i=1}^{\infty} \psi_{i} \geq \phi\}.
\end{eqnarray}
\end{theorem}

\begin{proof}
We denote the right hand side of (\ref{eq:ContExtFromLat}) by $f$. Note that $\lim_{i \rightarrow \infty} sp(\psi_{i}) = sp_{\sigma}(\bigvee_{i=1}^{\infty} \psi_{i}) \leq sp_{\sigma}(\phi)$ if $\bigvee_{i=1}^{\infty} \psi_{i} \geq \phi$, see Theorem \ref{th:SpFctProp}. This shows that $sp_{\sigma} \geq f$. We are going to show that $f$ is a continuous support function extending $sp$. Since $sp_{\sigma}$ is the minimal continuous support function extending $sp$ (Theorem \ref{th:MinContExt}), this proves then that $sp_{\sigma} = f$.

Let $(\mu,\mathcal{B})$ be the probability algebra associated with the canonical probability space (Section \ref{subsec:CanRandMap}) of the support function $sp$ and $\rho$ the corresponding allocation of probability, so that $sp = \mu \circ \rho$. For each $\phi \in \Psi$ define $\mathcal{D}(\phi) \subseteq \mathcal{B}$ by
\begin{eqnarray}
\mathcal{D}(\phi)= \{\bigwedge_{i=1}^{\infty} \rho(\psi_{i}):\psi_{i}\in \Phi_{f},i=1,2,\ldots,
\bigvee_{i=1}^{\infty} \psi_{i} \geq \phi\}
\nonumber
\end{eqnarray}
(here we follow the proof of Theorem 4 in \cite{shafer79}). The sets $\mathcal{D}(\phi)$ are upwards directed: In fact, consider two countable sets $\psi_{1,i},\psi_{2,j} \in \Psi_{f}$ such that $\bigvee_{i=1}^{\infty} \psi_{1,i},\bigvee_{j=1}^{\infty} \psi_{2,j}  \geq \phi$. Then, since $\Psi_{f}$ is a lattice, the set $\psi_{1,i} \wedge \psi_{2,j}$ is still a countable subset of $\Psi_{f}$. And, since the lattice $\Psi_{f}$ is assumed distributive,
\begin{eqnarray}
\bigvee_{i,j=1}^{\infty} (\psi_{1,i} \wedge \psi_{2,j}) = (\bigvee_{i=1}^{\infty} \psi_{1,i}) \wedge (\bigvee_{j=1}^{\infty} \psi_{2,j}) \geq \phi.
\nonumber
\end{eqnarray}
Finally, $\psi_{1,i} \wedge \psi_{2,j} \leq \psi_{1,i},\psi_{2,j}$ implies $\rho(\psi_{1,i} \wedge \psi_{2,j}) \geq \rho(\psi_{1,i}),\rho(\psi_{2,j})$, hence $\bigwedge_{i,j=1}^\infty \rho(\psi_{1,i} \wedge \psi_{2,j}) \geq \bigwedge_{i=1}^\infty \rho(\psi_{1,i}),  \bigwedge_{j=1}^\infty \rho(\psi_{2,j})$. So indeed, $\mathcal{D}(\phi)$ is upwards directed.

Define now $\tilde{\rho}(\phi) = \bigvee \mathcal{D}(\phi)$. We claim that $\tilde{\rho}$ is a $\sigma$-a.o.p and that $f = \mu \circ \tilde{\rho}$. This shows then that $f$ is a continuous support function. Since obviously $f \vert \Psi_{f} = sp$ this proves the theorem.

It is evident that $\tilde{\rho}(1) = \top$. So, it only remains to show that $\tilde{\rho}(\bigvee_{i=1}^{\infty} \psi_{i}) = \bigwedge_{i=1}^{\infty} \tilde{\rho}(\psi_{i})$ or $\bigvee \mathcal{D}(\bigvee_{i=1}^{\infty} \psi_{i}) = \bigwedge_{i=1}^{\infty} \bigvee \mathcal{D}(\psi_{i})$. Fix a sequence $\psi_{1},\psi_{2},\ldots \in \Phi$. To simplify notation let $\mathcal{D} = \mathcal{D}(\bigvee_{i=1}^{\infty} \psi_{i})$, $\mathcal{D}_{i} = \mathcal{D}(\psi_{i})$ and $M = \bigwedge_{i} \bigvee \mathcal{D}_{i}$. The task is then to show that 
\begin{eqnarray}
\bigvee \mathcal{D} = M.
\nonumber
\end{eqnarray}
Now, $\mathcal{D} \subseteq \mathcal{D}_{i}$ for all $i$, hence $\bigvee \mathcal{D} \leq \bigwedge_{i} \bigvee \mathcal{D}_{i}$. Further, since $\mathcal{D}_{i}$ are upwards directed sets, by Lemma \ref{downward}
 we have
\begin{eqnarray}
\mu(\bigvee \mathcal{D}_{i}) = \sup_{\psi \in \mathcal{D}_{i}} \mu(\psi).
\nonumber
\end{eqnarray}
Choose an $\epsilon > 0$. Then for all $i$, there is an $M_{i} \in \mathcal{D}_{i}$ such that
\begin{eqnarray}
\mu(\bigvee \mathcal{D}_{i} - M_{i}) \leq \frac{\epsilon}{2i}.
\nonumber
\end{eqnarray}
Since $M \leq \bigvee \mathcal{D}_{i}$, we obtain also
\begin{eqnarray}
\mu(M - M_{i}) = \mu(M \wedge M_{i}^{c}) \leq \mu(\bigvee\mathcal{D}_{i} \wedge M_{i}^{c}) =
\mu(\bigvee \mathcal{D}_{i} - M_{i}) \leq \frac{\epsilon}{2i}.
\nonumber
\end{eqnarray}
Let $B_{i}$ denote a set of elements $\phi_{1},\phi_{2},\ldots \in \Psi_{f}$ such that $\bigvee B_{i} \geq \psi_{i}$ and $M_{i} = \bigwedge_{\phi \in B_{i}} \rho(\phi)$. Let $B_{\epsilon} = \bigcup_{i=1}^{\infty} B_{i}$ and $M_{\epsilon} = \bigwedge_{i=1}^{\infty} M_{i}$. Then $\bigvee B_{\epsilon} \geq \bigvee_{i=1}^{\infty} \psi_{i}$ and
\begin{eqnarray}
M_{\epsilon} = \bigwedge_{i=1}^{\infty} \bigwedge_{\phi \in B_{i}} \rho(\phi) = \bigwedge_{\phi \in B_{\epsilon}} \rho(\phi).
\nonumber
\end{eqnarray}
Thus $M_{\epsilon}$ belongs to $\mathcal{D}$, hence $M_{\epsilon} \leq \bigvee \mathcal{D}$. We have
\begin{eqnarray}
M - M_{\epsilon} &=& M \wedge M_{\epsilon}^{c} = M \wedge (\bigwedge_{i=1}^{\infty} M_{i})^{c}
= M \wedge (\bigvee_{i=1}^{\infty} M_{i}^{c})
\nonumber \\
&=& \bigvee_{i=1}^{\infty} (M \wedge M_{i}^{c}) 
= \bigvee_{i=1}^{\infty} (M - M_{i}).
\nonumber
\end{eqnarray}
Thus we obtain
\begin{eqnarray}
\mu(M - M_{\epsilon}) = \mu(\bigvee_{i=1}^{\infty} (M - M_{i})) \leq  \sum_{i=1}^\infty \frac{\epsilon}{2^i} = \epsilon.
\nonumber
\end{eqnarray}
Now, $M_{\epsilon} \leq \bigvee \mathcal{D}$ implies $M_{\epsilon}^{c} \geq (\bigvee \mathcal{D})^{c}$ and therefore $M - \bigvee \mathcal{D} = M \wedge (\bigvee\mathcal{D})^{c} \leq M \wedge M_{\epsilon}^{c} = M - M_{\epsilon}$. This shows that
\begin{eqnarray}
\mu(M -  \mathcal{D}) \leq \mu(M - M_{\epsilon}) \leq \epsilon
\nonumber
\end{eqnarray}
Since $\epsilon$ is arbitrarily small, we conclude that $\mu(M -  \mathcal{D}) = 0$ and from this it follows that $\bigvee \mathcal{D} = M$, because $\vee \mathcal{D} \leq M$. This proves that $\tilde{\rho}$ is a $\sigma$-a.o.p.

Next, we are going to show that $f = \mu \circ \tilde{\rho}$, hence that $f$ is a continuous support function. Note that $sp_{\sigma}= \mu \circ \rho$. Then, since $sp_{\sigma} \vert \Psi_{f} = sp$ and since $\rho$ is a $\sigma$-a.o.p, because $sp_\sigma$ is continuous, we have
\begin{eqnarray}
f(\phi) &=& \sup\{\mu(\rho(\bigvee_{i=1}^{\infty} \psi_{i})):\psi_{i} \in \Phi_{f},\bigvee_{i=1}^{\infty} \psi_{i} \geq \phi\}
\nonumber \\
&=&\sup\{\mu(\bigwedge_{i=1}^\infty \rho(\psi_i)):\psi_i \in \Phi_f, \bigvee_{i=1}^\infty \psi_i \geq \psi\}
\nonumber \\
&=&\sup\{\mu(\rho(\chi)):\chi \in \mathcal{D}(\phi)\}.
\nonumber 
\end{eqnarray}
Since $\mathcal{D}(\phi)$ is upwards directed, we obtain (Lemma \ref{downward})
\begin{eqnarray}
f(\phi) = \mu(\bigvee \mathcal{D}(\phi)) = \mu(\tilde{\rho}(\phi))
\nonumber
\end{eqnarray}
This concludes the proof.
\end{proof}

Under the assumptions of Theorem \ref{th:ContExtFromLat} we may, according to the considerations in the proof, also write
\begin{eqnarray}
sp_{\sigma}(\phi) = \sup\{sp_{\sigma}(\bigvee_{i=1}^{\infty} \psi_{i}):\psi_{i} \in \psi_{f},\bigvee_{i=1}^{\infty} \psi_{i} \geq \phi\},
\nonumber
\end{eqnarray}
or, equivalently,
\begin{eqnarray}
sp_{\sigma}(\phi) = \sup\{sp_{\sigma}(\chi_{i}):\chi_{i} \in \sigma(\psi_{f}),\chi \geq \phi\},
\nonumber
\end{eqnarray}
If $\Phi_{f}$ is in addition \textit{countable}, then $\sigma(\Phi_{f}) = \Phi$ and 
\begin{eqnarray}
sp_{\sigma}(\psi) = \lim_{i \rightarrow \infty} sp(\psi_i)
\nonumber
\end{eqnarray}
if $\psi_{1} \leq \psi_{2} \leq \ldots \in \Phi_{f}$ and $\vee_{i=1}^{\infty} \psi_{i} = \phi$. We remark that this result holds in general if the set of finite elements is countable, without the additional assumption that $(\Phi_f,\leq)$ is a distributive lattice. This follows from the alternative approach to generate continuous support function, based on results of \cite{norberg89} mentioned at the end of Section \ref{sec:GenSpFct}.

Just as Corollary \ref{cor:CanAoP}, we may also derive the following result:

\begin{corollary}
Under the conditions of Theorem \ref{th:ContExtFromLat}, if $\rho_{\sigma}$ is the a.o.p associated with the support function $sp_{\sigma}$, then $\rho_{\sigma} = \tilde{\rho}$, where the latter a.o.p is defined in the proof of Theorem \ref{th:ContExtFromLat}.
\end{corollary}

We have shown that a support function defined on some join-semilattice $\mathcal{E} \subseteq \Psi$ of an information algebra $\Psi$ can have different kinds of exentsion, defined in terms of its values in $\mathcal{E}$. Similar and more results of this kind can be found in \cite{shafer79} in a more restricted context.


\section{The Boolean case} \label{subsec:BooleanCase}

In this section, the information algebra $(\Phi,\cdot,0,1;E)$ is assumed to be \textit{Boolean}, that is, the semilattice $(\Phi,\leq)$ is a Boolean lattice under information order. Everything said so far about random mappings, allocations of probability and support functions remains valid. However the Boolean nature of $\Phi$ allows to present an equivalent dual view to allocations of probability and support functions. This dual view comes from considering possibility sets and associated degrees of plausibility as introduced in Section \ref{subsec:RanMaps}. In a general information algebra these concepts are of no particular interest, they are far less interesting and important than allocations of support and support functions. In the Boolean case however their status changes to one of equal importance and interest. 

Consider a random mapping $\Gamma$ from a probability space $(\Omega,\mathcal{A},P)$ to a Boolean information algebra $\phi$. Generalising the discussion in Section  \ref{subsec:SimpleRanMaps} with respect to simple random variables, we define the set of assumptions $\omega$ under which a hypothesis $\psi \in \Psi$ its \textit{possible}, that is not excluded, by
\begin{eqnarray}
p_{\Gamma}(\psi) = \{\omega \in \Omega:\psi \cdot \Gamma(\omega)\not= 0\}.
\nonumber
\end{eqnarray}
Given, that the top element $0$ of the Boolean algebra $(\Psi,\leq)$ is considered to represent the contradiction, an assumption $\omega$ such that $\psi \cdot \Gamma(\omega) = 0$ must be considered as impossible, as excluded by the information contained in the random mapping $\Gamma$. Therefore, $p_{\Gamma}(\psi)$ is called the \textit{possibility set} of $\psi$, relative to the random mapping $\Gamma$. 

In a Boolean algebra we have $\psi \cdot \Gamma(\omega) = \psi \vee \Gamma(\omega) = 0$ if and only if $\psi^{c} \leq \Gamma(\omega)$, where $\psi^c$ denotes the complement of $\psi$ in $\Psi$. Therefore, we see that
\begin{eqnarray} \label{eq:PossSupDuality}
p_{\Gamma}(\psi) = \{\omega \in \Omega:\psi^{c} \leq \Gamma(\omega)\}^{c} = (s_{\Gamma}(\psi^{c}))^{c},
\end{eqnarray}
where $s_{\Gamma}$ is the allocation of support associated with the random mapping $\Gamma$ (see (\ref{eq:DefAllocOfSupp})). This is the first of the basic duality relations between support and plausibility or possibility considered in this section. It allows to translate results on allocations of support immediately to possibility sets.

\begin{theorem} 
If $\Gamma: \Omega \rightarrow \Phi$, where $(\Phi,\leq)$ is a Boolean lattice, then
\begin{enumerate}
\item $p_{\Gamma}(0) = \emptyset$,
\item If $\phi \leq \psi$, then $p_{\Gamma}(\psi) \subseteq p_{\Gamma}(\phi)$.
\item $p_{\Gamma}(\phi \wedge \psi) = p_{\Gamma}(\phi) \cup p_{\Gamma}(\psi)$.
\item If $\Gamma$ is normalised, then $p_{\Gamma}(1) = \Omega$.
\item If $\Phi$ is a Boolean $\sigma$-algebra, then 
\begin{eqnarray}
p_{\Gamma}(\bigwedge_{i=1}^{\infty} \psi_{i}) = \bigcup_{i=1}^{\infty} p_{\Gamma}(\psi_{i}).
\nonumber
\end{eqnarray}
\item If $\Phi$ is a complete Boolean algebra, then for any subset $X$ of $\Phi$,
\end{enumerate}
\begin{eqnarray}
p_{\Gamma}(\bigwedge X) = \bigcup_{\psi \in X} p_{\Gamma}(\psi).
\nonumber
\end{eqnarray}
\end{theorem}

\begin{proof}
Items (1) to (4) follow immediately from Theorem \ref{th:ElPropAllSp} and the duality relation  (\ref{eq:PossSupDuality}). Items (5) and (6) follow similarly from Theorem \ref{th:ElPropAllSp2}, (\ref{eq:PossSupDuality}) and de Morgan laws.
\end{proof}

If $p_{\Gamma}(\psi)$ is measurable, the probability, that $\psi$ is not excluded by $\Gamma$, $pl_{\Gamma}(\psi) = P(p_{\Gamma}(\psi))$ is defined. This is called the \textit{degree of possibility} or \textit{plausibility} of $\psi$ under the random mapping $\Gamma$. Let $\mathcal{Z}_{\Gamma} = \{\psi \in \Psi:p_{\Gamma}(\psi) \in \mathcal{A}\}$ be the set of $\psi$ for which $p_{\Gamma}(\psi)$ is measurable. Recall that $\mathcal{E}_{\Gamma}$ is the set of elements of $\Phi$ for which $s_{\Gamma}(\psi)$ is measurable. Clearly, $\psi \in \mathcal{Z}_{\Gamma}$ implies $\psi^{c} \in \mathcal{E}_{\Gamma}$. According to Theorem \ref{th:MeasInfEl}, $\mathcal{E}_{\Gamma}$ is a join-semilattice, containing $1$. Thus $\mathcal{Z}_{\Gamma}$ is a meet-semilattice, containing $0$. Let's fix this result in the following theorem.

\begin{theorem}
If $(\Phi,\leq)$ is a Boolean lattice, $\Gamma$ a random mapping into $\Phi$, then $\mathcal{Z}_{\Gamma}$ is a meet-subsemilattice of $\Phi$ containing $0$. If $\Gamma$ is normalised, then $1$ belongs to $\mathcal{Z}_{\Gamma}$ too. If $(\Phi,\leq)$ is a Boolean $\sigma$-algebra, then $\mathcal{Z}_{\Gamma}$ is a $\sigma$-semilattice.
\end{theorem}

Note that
\begin{eqnarray} \label{eq:PlDual}
pl_{\Gamma}(\psi) = P(p_{\Gamma}(\psi)) = P((s_{\Gamma}(\psi^{c})^{c}) = 1 - sp_{\Gamma}(\psi^{c}).
\end{eqnarray}
This is a second duality relation between support and plausibility in a Boolean algebra.

The function $pl_{\Gamma} : \mathcal{Z}_{\Gamma} \rightarrow [0,1]$ is called the \textit{plausibility function} associated with the random mapping $\Gamma$. Just as the support function $sp_{\Gamma}$ can be extended from $\mathcal{E}_{\Gamma}$ to $\Psi$ by defining $sp_{\Gamma} = \mu \circ \rho_{\Gamma}$, where $(\mu,\mathcal{B})$ is the probability algebra associated to the probability space $(\Omega,\mathcal{A},P)$ and $\rho_{\Gamma} = \rho_{0} \circ s_{\Gamma}$ the allocation of probability associated with $\Gamma$, we may extend $pl_{\Gamma}$ in a similar way, see Section \ref{subsec:RanMaps}. This is done with the help of $\xi_{0}$ as defined by (see (\ref{eq:xi0}),
\begin{eqnarray}
\xi_{0}(H) = (\rho_{0}(H^{c}))^{c} = \bigwedge \{[A]: A \supseteq H,A \in \mathcal{A}\},
\nonumber
\end{eqnarray}
and $\xi_{\Gamma} = \xi_{0} \circ p_{\Gamma}$ and $pl_{\Gamma} = \mu \circ \xi_{\Gamma}$ (see Section \ref{subsec:RanMaps}). Then we obtain
\begin{eqnarray}
\lefteqn{pl_{\Gamma}(\psi)= \mu(\xi_{\Gamma}(\psi)) = \mu(\xi_{0}(p_{\Gamma}(\psi))) = \mu((\rho_{0}((p_{\Gamma}(\psi))^{c}))^{c})}
\nonumber\\
&& = \mu((\rho_{0}(s_{\Gamma}(\psi^{c})))^{c})
= \mu((\rho_{\Gamma}(\psi^{c}))^{c}) = 1 - sp_{\Gamma}(\psi^{c}).
\nonumber
\end{eqnarray}
So, the extension $pl_{\Gamma} = \mu \circ \xi_{\Gamma}$ of the plausibility function to $\Phi$ preserves the duality relation to the support function. Further, we have seen in Section \ref{subsec:RanMaps} that $p_{\Gamma}(\psi) = P^{*}(p_{\Gamma}(\psi))$, where $P^{*}$ is the outer probability measure of $P$.

In the present case of a Boolean algebra $\Psi$, we note that
\begin{eqnarray}
\xi_{\Gamma}(\psi) = \xi_{0}(p_{\Gamma}(\psi)) = \xi_{0}((s_{\Gamma}(\psi^{c}))^{c}) = (\rho_{0}(s_{\Gamma}(\psi^{c})))^{c}
= (\rho_{\Gamma}(\psi^{c}))^{c}.
\nonumber
\end{eqnarray}
Here we have a third duality relation, which implies immediately, that $\xi(0) = \bot$ and $\xi(\psi \wedge \psi) = \xi(\psi) \vee \xi(\psi)$. A function from $\Phi$ to $\mathcal{B}$ with these two properties is called an \textit{allowment of probability} \cite{shafer79}.

\begin{definition} \textbf{Allowment of probability.}
If $(\Phi;\leq)$ is a Boolean algebra, $(\mu,\mathcal{B})$ a probability algebra, then an allowment of probability is a mapping $\xi : \Phi \rightarrow \mathcal{B}$ such that
\begin{enumerate}
\item $\xi(0) = \bot$,
\item $\xi(\psi \wedge \psi) = \xi(\psi) \vee \xi(\psi)$.
\end{enumerate}
If furthermore, $\xi(1) = \top$ holds, then the allowment is called normalised.
\end{definition}

To any allocation of probability $\rho : \Phi \rightarrow \mathcal{B}$ we associate an allowment of probability $\xi : \Phi \rightarrow \mathcal{B}$ defined by 
\begin{eqnarray} \label{eq:AllowDual}
\xi(\psi) = (\rho(\psi^{c}))^{c}
\end{eqnarray}
and vice versa to any allowment of probability $\xi$, an allocation of probability $\rho$, defined by $\rho(\psi) = (\xi(\psi^{c}))^{c}$ is associated.

In order to exploit this duality we consider the \textit{dual Boolean algebra} $(\Phi^{op};\leq_{op})$ of $(\Phi,\leq)$, with inverse order $\leq_{op}$ and the corresponding dual meet $\wedge_{op}$ and join $\vee_{op}$, so that
\begin{eqnarray}
\psi \leq_{op} \phi &\textrm{ if and only if}& \phi \leq \psi,
\nonumber \\
\phi \vee_{op} \psi &=& \phi \wedge \psi = (\phi^{c} \vee \psi^{c})^{c},
\nonumber \\
\phi \wedge_{op} \psi &=& \phi \vee \psi = (\phi^{c} \wedge \psi^{c})^{c},
\nonumber \\
0_{op} &=& 1, 
\nonumber \\
1_{op} &=& 0.
\nonumber
\end{eqnarray}
To any extraction operator $\epsilon_x$ for $x \in D$ we associate a mapping $\epsilon_x^{op} : \Phi^{op} \rightarrow \Phi^{op}$ defined by
\begin{eqnarray} \label{eq:DualExtr}
\epsilon_x^{op}(\psi) = (\epsilon_x(\psi^{c}))^{c}.
\nonumber
\end{eqnarray}

If we interpret dual join $\vee_{op}$ as (dual) combination $\cdot_{op}$ and the maps $\epsilon_x^{op}$ as (dual) extraction, then it turns out, that $(\Phi^{op},\cdot_{op},0_{op},1_{op};E^{op})$ with $E^{op} = \{\epsilon_x^{op}:x \in Q\}$ is in fact still a Boolean information algebra. 

For later reference let's also consider the dual of an \textit{compact} Boolean information algebra with finite elements $\Phi_{f}$. Then $(\Phi,\leq)$ is a complete lattice and, therefore, $(\Phi,\leq_{op})$ is a complete lattice too. Define the set
\begin{eqnarray}
\Phi_{cf} = \{\psi:\psi^{c} \in \Phi_{f}\}
\end{eqnarray}
whose elements are called \textit{cofinite}. Density in $\Phi$ leads by de Morgan laws to
\begin{eqnarray}
\phi = \bigvee_{op} \{\psi \in \Phi_{cf}:\psi \leq_{op} \phi\}.
\nonumber
\end{eqnarray}
Similary, strong density implies
\begin{eqnarray}
\epsilon_x^{op}(\phi) = \bigvee_{op} \{\psi \in \Phi_{cf}:\psi = \epsilon_x^{op}(\psi) \leq_{op} \phi\}.
\nonumber
\end{eqnarray}
Thus, the dual information algebra $\Phi^{op}$ is also \textit{compact} and the cofinite elements of $(\Phi,\leq)$ are its finite elements..

As an example consider multivariate algebras.

\begin{example} \textbf{Dual Set Algebras}
Let $\Phi$ be a multivariate set algebra (see Section \ref{sec:SetAlg}) in a set $\Omega_{I}$, where $I$ is an index set and 
\begin{eqnarray}
\Omega_{I} = \prod_{i \in I} \Omega_{i}
\nonumber
\end{eqnarray}
and $\Omega_{i}$ are sets of possible values for variables $X_{i}$, $i \in I$. Elements of $\Phi$ are subsets of $\Omega_{I}$. This is a Boolean information algebra where join is intersection, meet is union. The (finite) subsets $s$ of $I$ form the lattice $Q$ and extraction relative to $s \in Q$ is defined as $s$-saturation, that is as saturation relative to the partition of $\Omega_I$ induced by the subset $s$ of the index set $I$. In the dual information algebra $\Phi$ join is union, meet intersection. Dual extraction is defined according to (\ref{eq:DualExtr}) by $\sigma_s^{op}(S) = (\sigma_s(S^{c}))^{c}$, for any subset $S$ of $\Omega_I$.

The algebra $\Phi,$ is \textit{compact}, its finite elements are the \textit{cofinite} sets of $\Omega_I$, that is the complements of finite subsets of $\Omega_I$. The cofinite elements of $\Psi$, that is the finite elements of $\Phi^{op}$, are the \textit{finite} subsets of $\Omega_I$. \end{example}

Now we have the means to exploit duality between allocations and allowments of probability (\ref{eq:AllowDual}) and between degrees of support and plausibility (\ref{eq:PlDual}). Let $\rho : \Phi \rightarrow \mathcal{B}$ be an allocation of probability to a Boolean information algebra $(\Phi,\cdot,0,1;E)$ relative to a probability algebra $(\mu,\mathcal{B})$. The corresponding allowment of probability $\xi$, defined by (\ref{eq:AllowDual}) can be seen as a mapping $\xi : \Phi^{op} \rightarrow \mathcal{B}^{op}$ between the dual Boolean algebras of $\Phi$ and $\mathcal{B}$. Then, in this view, $\xi$ is an \textit{allocation of probability} in $\Phi^{op}$, that is
\begin{enumerate}
\item $\xi(0_{op}) = \top_{op}$,
\item $\xi(\phi \vee_{op} \psi) = \xi(\phi) \wedge_{op} \xi(\psi)$.
\end{enumerate}
As a consequence, as allocations of probability, the $\xi$ form an information algebra $A_{\Phi^{op}}$ (see Section \ref{subsec:AllocProb}). Let's denote combination by $\vee_{op}$, such that according to (\ref{eq:CombOfAoP})
\begin{eqnarray}
\lefteqn{(\xi_{1} \vee_{op} \xi_{2})(\psi) 
= \bigvee_{op} \{\xi_{1}(\psi_{1} \wedge_{op} \xi_{2}(\psi_{2}):\psi \leq_{op} \psi_{1} \vee_{op} \psi_{2}\}}
\nonumber \\
&=& \bigwedge \{\xi_{1}(\psi_{1} \vee \xi_{2}(\psi_{2}):\psi \geq \psi_{1} \wedge \psi_{2}\}.
\nonumber
\end{eqnarray}
Similarily, for extraction, we obtain, using (\ref{eq:ExtractAoP}),
\begin{eqnarray}
\lefteqn{\epsilon_x^{op}(\xi)(\phi) = \bigvee_{op}\{\xi(\psi):\psi = \epsilon_x^{op}(\psi) \geq_{op}\phi\}}
\nonumber \\
&=&\bigwedge\{\xi(\psi):\psi = x(\psi) \leq \phi\}.
\nonumber
\end{eqnarray}
Clearly, by the map $\rho \mapsto \xi$, defined by $\xi(\psi) = \rho(\psi^{c})^{c}$, is an isomorphism between information algebras. 

We write $\xi_{1} \leq_{op} \xi_{2}$ if $\xi_{1} \vee_{op} \xi_{2} = \xi_{2}$. Then, $\xi_{1} \leq_{op} \xi_{2}$ if and only if $\xi_{1}(\psi) \leq_{op} \xi_{2}(\psi)$ for all $\psi \in \Phi^{op}$. If we look at this relative to the original algebra $(\Phi,D;\leq,\bot,\cdot,\epsilon)$, then it is convenient to write $\xi_{1} \wedge \xi_{2} = \xi_{1} \vee_{op} \xi_{2}$ and hence $\xi_{1} \geq \xi_{2}$ if $\xi_{1} \leq_{op} \xi_{2}$. Finally, we write simply $\epsilon_x(\xi)$ instead of $\epsilon_x^{op}(\xi)$ for $x \in D$. In the following we shall use this convention.

Next, we use the duality relation (\ref{eq:AllowDual}) to translate results relating random mappings to allocations of probability obtained in Section \ref{subsec:AllocProb} to allowments of probability. Here is a list of such results, which can be easily obtained by (\ref{eq:AllowDual}) and de Morgan laws:
\begin{enumerate}
\item If $\Delta_{1}$, $\Delta_{2}$ and $\Delta$ are simple random variables, then by (\ref{eq:RhoIsHomom})
\begin{eqnarray}
\xi_{\Delta_{1} \cdot \Delta_{2}} &=& \xi_{\Delta_{1}} \wedge \xi_{\Delta_{2}},
\nonumber \\
\xi_{\epsilon_x(\Delta)} &=& \epsilon_x(\xi_{\Delta}).
\nonumber
\end{eqnarray}
\item If $\Gamma$ is a  random variable, then (Theorem \ref{th:IdOfAoPs})
\begin{eqnarray}
\xi_{\Gamma} = \bigwedge \{\xi_{\Delta}:\Delta \leq \Gamma\}.
\nonumber
\end{eqnarray}
Here, $\Delta$ denote as usual simple random variables.
\item if $\Gamma_{1}$, $\Gamma_{2}$ and $\Gamma$ are random variables, then (Theorem \ref{th:HomoOfGenRanVar})
\begin{eqnarray}
\xi_{\Gamma_{1} \cdot \Gamma_{2}} &=& \xi_{\Gamma_{1}} \wedge \xi_{\Gamma_{2}},  \\
\nonumber
\xi_{\epsilon_x(\Gamma)} &=& \epsilon_x(\xi_\Gamma) .
\nonumber
\end{eqnarray}
\item If $\Gamma$ is a random variable, $\Phi$ a compact Boolean information algebra, $X \subseteq \Phi$ a downwards directed set, then (Theorem \ref{thCondOfGenRV})
\begin{eqnarray}
\xi_{\Gamma}(\bigwedge X) = \bigvee_{\psi \in X} \xi_{\Gamma}(\psi).
\nonumber
\end{eqnarray}
\item Suppose $\Phi$ is a compact information algebra and $\Gamma_i \in \mathcal{R}_\sigma$ for $i = 1,2,\ldots$, then (Theorem \ref{th:SigmaHomom})
\begin{eqnarray}
\xi_{\bigvee_{i=1}^{\infty} \Gamma_{i}} = \bigwedge_{i=1}^{\infty} \xi_{\Gamma_{i}}.
\nonumber
\end{eqnarray}
\item If $\Gamma_{i}$ form a montone sequence random variables $\Gamma_{1} \leq \Gamma_{2} \leq \ldots$, $(\Psi,D;\leq,\bot,\cdot,\epsilon)$ in an algebraic Boolean information algebra, then (Theorem \ref{th:ContOfExtrAoP})
\begin{eqnarray}
\epsilon_x(\bigwedge_{i=1}^{\infty} \xi_{\Gamma_{i}})= \bigwedge_{i=1}^{\infty} \epsilon_x(\xi_{\Gamma_{i}}).
\nonumber
\end{eqnarray}
\end{enumerate}

Now we turn to plausibility and exploit duality relation (\ref{eq:PlDual}) to derive results on degrees of plausibility from support functions.  If $\Gamma$ is a random map, mapping a probability space into an information algebra $(\Phi,\cdot,0,1;E)$ (or its ideal completion), then recall that its support function is defined by $sp_{\Gamma} = \mu \circ \rho_{\Gamma}$, where $\rho_{\Gamma} = \rho_{0} \circ s_{\Gamma}$ and $(\mathcal{B},\mu)$ is the probability algebra associated with the probability space (see Section \ref{subsec:RanMaps}). Similarly, the associated degrees of plausibility $pl_{\Gamma}$, related to $sp_{\Gamma}$ by the duality relation (\ref{eq:PlDual}), is given by $pl_{\Gamma}
= \mu \circ \xi_{\Gamma}$. where $\xi_{\Gamma} = \xi_{0} \circ p_{\Gamma}$. And $\rho_{\Gamma}$ and $\xi_{\Gamma}$ are related by the duality relation (\ref{eq:AllowDual}).

Here follows a list of results on plausibility, derived from corresponding results on support function via the duality relation (\ref{eq:PlDual}):
\begin{enumerate}
\item Let $\Gamma$ be a random mapping, then (Theorem  \ref{th:SpFctProp})
\begin{enumerate}
\item $pl_{\Gamma}(0) = 0$.
\item If $\psi_{1},\ldots,\psi_{m} \leq \psi$, $\psi_{1},\ldots,\psi_{m},\psi \in \mathcal{Z}_{\Gamma}$, 
\begin{eqnarray} \label{eq:MonOrderInf}
pl_{\Gamma}(\psi) \leq \sum_{\emptyset \not= I \subseteq \{1,\ldots,m\}} (-1)^{\vert I \vert + 1} pl_{\Gamma}(\wedge_{i \in I} \psi_{i}).
\end{eqnarray}
\item If $\mathcal{Z}_{\Gamma}$ is a $\sigma$-meet semilattice, and if $\psi_{1} \geq \psi_{2} \geq \ldots \in \mathcal{Z}_{\Gamma}$, then
\begin{eqnarray} \label{eq:ContSupFct}
pl_{\Gamma}(\bigwedge_{i=1}^{\infty} \psi_{i}) = \lim_{i \rightarrow \infty} pl_{\Gamma}(\psi_{i}).
\end{eqnarray}
\item If $\Gamma$ is normalised, then $pl_{\Gamma}(1) = 1$.
\end{enumerate}
\item If $(\mathcal{B},\mu)$ is a probability algebra and $\xi:\Phi \rightarrow \mathcal{B}$ is an allowment of probability and $pl = \mu \circ \xi$, then (Theorem \ref{th:SpFctPropAoP}) 
\begin{enumerate}
\item $pl$ satisfies properties (a) and (b) of item 1 above.
\item If $\Phi$ is a $\sigma$-meet-semilattice and if for all $\psi_1,\psi_2,\ldots$, we have $\xi(\bigwedge_{i=1}^\infty \psi_1) = \bigvee_{i=1}^\infty \xi(\psi_i)$,
then (c) of item 1 above holds.
\item If $\Phi$ is a complete lattice and if for any downwards directed set $X \subseteq \Psi$
\begin{eqnarray}
\xi(\bigwedge X) = \bigvee_{\psi \in X} \xi(\psi)
\nonumber
\end{eqnarray}
holds, then
\begin{eqnarray} \label{eq:CondSupFct0}
pl(\bigwedge X) = \sup_{\psi \in X} pl(\psi).
\end{eqnarray}
\end{enumerate}
\item If $\Gamma$ is a random variable, $\Phi$ a compact Boolean information algebra, $pl_{\Gamma} = \mu \circ \xi_{\Gamma}$, $\xi_{\Gamma} = \xi_{0} \circ p_{\Gamma}$, then (Theorem \ref{th:SpFctOfGenRV})
\begin{eqnarray} 
pl_{\Gamma}(\psi) = \sup\{pl_{\Gamma}(\phi):\psi \in \Psi_{cf},\phi \geq \psi\}.
\nonumber
\end{eqnarray}
Furthermore, if $X \subseteq \Phi$ is downwards directed, then
\begin{eqnarray} 
pl_{\Gamma}(\bigwedge X) = \sup_{\psi \in X}pl_{\Gamma}(\psi).
\nonumber
\end{eqnarray}
\item Let $\sigma(\Phi)$ be the $\sigma$-extension of the Boolean information algebra $\Phi$, $\Gamma$ a random variable, that is, $\Gamma = \bigvee_{i=1}^{\infty} \Delta_{i}$, where $\Delta_{i}$ is a monotone increasing sequences of simple random variables with values in $\Phi$, then for all $\psi \in \Phi$ (Theorem \ref{th:AllElMeas})
\begin{eqnarray}
pl_{\Gamma}(\psi) = \lim_{i \rightarrow \infty} pl_{\Delta_{i}}(\psi).
\nonumber
\end{eqnarray}
\item If $\Gamma$ is a random variable, then for all $\psi \in \Phi$ (Corollary \ref{cor:ApproxDegOfSupSimRV})
\begin{eqnarray}
pl_{\Gamma}(\psi) = \inf\{pl_{\Delta}(\psi):\Delta \leq \Gamma\},
\nonumber
\end{eqnarray}
where $\Delta$ as usual are simple random variables.
\end{enumerate}

These results allow to give a dual version of Definition \ref{def:SpFct}, now regarding plausibility functions:

\begin{definition} \label{def:PlFct}
Let $\mathcal{Z}$ be a meet-semilattice with a top element $0$. Then a function $pl:\mathcal{Z} \rightarrow $[0,1] satisfying (1) and (2) below is called a plausibility function on $\mathcal{Z}$:
\begin{enumerate}
\item $pl(0) = 0$.
\item If $\psi_{1},\ldots,\psi_{m} \leq \psi$, $\psi_{1},\ldots,\psi_{m},\psi \in \mathcal{Z}$ for $m = 1,2,\ldots$
\begin{eqnarray} \label{eq:AltOrderInf1}
pl(\psi) \leq \sum_{\emptyset \not= I \subseteq \{1,\ldots,m\}} (-1)^{\vert I \vert + 1} pl(\wedge_{i \in I} \psi_{i}).
\end{eqnarray}
\item If in addition $\mathcal{Z}$ is closed under countable meets, and  for any montone sequence $\psi_{1} \geq \psi_{2} \geq \cdots$ the condition
\begin{eqnarray} \label{eq:ContOfPlFct}
pl(\bigwedge_{i=1}^{\infty} \psi_{i}) = \lim_{i \rightarrow \infty} pl(\psi_{i})
\end{eqnarray}
holds, then $pl$ is called a continuous plausibility function of $\mathcal{Z}$.
\item If further $\mathcal{Z}$ is a complete meet-semilattice and for any downwards directed set $X \subseteq \mathcal{Z}$,
\begin{eqnarray} \label{eqCondOfPlFct}
pl(\bigwedge X) = \sup_{\psi \in X} pl(\psi)
\end{eqnarray}
holds, then $pl$ is called a condensable plausibility function on $\mathcal{Z}$.
\end{enumerate}
\end{definition}
A function satisfying (2) above is also called \textit{alternating of order} $\infty$ \cite{choquet53}. Thus, the degrees of plausibility of any random mapping $\Gamma$ form a plausibility function. If $\Gamma$ is a random variable in an algebraic Booolean information algebra, then $pl_{\Gamma}$ is condensable, and if $\Gamma$ is a proper random variable, then $pl_{\Gamma}$ is continuous.

Given a plausibility function $pl$ on a meet-semilattice $\mathcal{Z} \subseteq \Phi$, where $\Psi$ is a Boolean information algebra, the function $sp(\psi) = 1 - pl(\psi^{c})$ is a support function on a join-semilattice $\mathcal{E} \subseteq \Phi$. Based on this remark we conclude that there is a random mapping generating $sp$, hence $pl$. In fact, the canonical random mapping $\nu$ (see Section \ref{subsec:CanRandMap}) generates the plausibility function $pl_{\nu}(\psi) = 1 - sp_{\nu}(\psi^{c})$ on $\Phi$, which is the \textit{maximal} extension of $pl$ from $\mathcal{Z}$ to $\Phi$. If the Boolean information algebra $\Phi$ is compact, the random mapping $\sigma$ (\ref{eq:SigmaRandMap}) generates the maximal \textit{continuous} extension $pl_{\sigma}(\psi) = 1 - sp_{\sigma}(\psi^{c})$ (see Theorem \ref{th:MinContExt}). And the random mapping $\gamma$ (\ref{eq:GammaRandMap}) generates according to (\ref{eq:GammaCondExt}) a \textit{condensable} plausibility function (Theorem \ref{th:SpFctPropAoP1}). This concludes the duality discussion between support and plausibility in Boolean information algebras.


%% file: chapter10.tex

\chapter{Probabilistic Information} \label{sec:ProbabInf}

\section{Gambles} \label{subsec:Gambles}

A particular form of uncertain information is probabilistic information, defined by a probability measure over some set of possibilities. The most popular form of this kind of information is given by a Bayesian network, where a multivariate discrete probability distribution is factorized into a product of prior and conditional distributions \cite{henoyshafer90}. It is well-known that associated  with this concept are valuation algebras, a kind of non-idempotent information algebras \cite{shenoyshafer90,kohlas03}. A more general form of probabilistic information has been proposed in \cite{walley91}. This theory of imprecise rpobability is based on the concept of desirable gambles and the derived notion of lower and upper previsions. It has been shown, that there are again information algebras associated with this model of probabilistic information \cite{CasaKohECSQARU,CasaKohIIJAR,CasaKohl22}. This is the subject of the section.

Consider a set $\Theta$ of possible worlds. A gamble over this set is a bounded function
\begin{eqnarray*}
f : \Theta \rightarrow \mathbb{R}.
\end{eqnarray*}
Let $\mathcal{L}(\Theta)$ be the set of all gambles over $\Theta$ and $\mathcal{L}^+(\Theta)$ the subset of non-vanishing, non-negative functions $f(\theta) \geq 0$, $f \not= 0$. A coherent set of (desirable) gambles over $\Theta$ is a subset $D$ of $\mathcal{L}(\Theta)$ siuch that
\begin{enumerate}
\item $\mathcal{L^+}(\Theta) \subseteq D$,
\item $0 \not\in D$,
\item $f,g \in D$ implies $f + g \in D$,
\item $f \in D$, and $\lambda > 0$ implies $\lambda \cdot f \in D$.
\end{enumerate}
So, $D$ is a convex cone. The idea is that gambles in $\mathcal{L}^+$ which guarantee a sure gain are desirable and positive multiples of a desirable gambles as well as the sum of two (or more) desirable gambles are also desirable. And the null gamble is not desirable. This may be questionable and in fact there are a number of different concepts of coherence, see below and \cite{walley91}. 

If $D'$ is any subset of $\mathcal{L}(\Theta)$, then 
\begin{eqnarray*}
\mathcal{E}(D') = posi(\mathcal{L}^+(\Theta) \cup D'),
\end{eqnarray*}
is called the natural extension of a set of gambles, where $posi(D)$ denotes all finite positive linear combinations $\lambda_1 f_1 + \ldots + \lambda_n f_n$,  $\lambda_i > 0$ of elements $f_1,\ldots,f_n$ of $D$. The natural extension of a set of gambles $\mathcal{E}(D)$ is coherent if and only if it $0 \not\in \mathcal{E}(D)$. Coherent sets are closed under intersection, that is they form a topless $\cap$-structure, \cite{daveypriestley97}. By standard order theory, coherent sets of gambles are ordered by inclusion, intersection is meet in this order and coherent sets of gambles  $D_i$ have a supremum or join if they have an upper bound among coherent sets,
\begin{eqnarray*}
\bigvee_{i \in I} D_i = \bigcap \{D \textrm{ coherent}: D \subseteq \bigcup_{i \in I} D_i\}.
\end{eqnarray*}
Also, $\mathcal{E}(D')$ is the smallest coherent set containing $D'$, if $\mathcal{E}(D')$ is coherent,
\begin{eqnarray*}
\mathcal{E}(D') = \bigcap \{D \textrm{ coherent}: D' \subseteq D\},
\end{eqnarray*}
so that 
\begin{eqnarray*}
\bigvee_{i \in I} D_i = \mathcal{E}(\bigcup_{i \in I} D_i)
\end{eqnarray*}
if $\mathcal{E}(\bigcup_{i \in I} D_i)$ is coherent. Let $\mathcal{C}(\Theta)$ be the family of coherent sets of gambles on $\Theta$. 

In view of the following development, it is convenient to add $\mathcal{L}(\Theta)$ to $\mathcal{C}(\Theta)$ and let $\Phi = \mathcal{C}(\Theta) \cup \{\mathcal{L}(\Theta)\}$. The family of sets in $\Phi$ is still a $\cap$-structure, but now a topped one. So, again by standard results of order theory, $\Phi$ is a complete lattice under inclusion, meet is intersection and join is defined for any family of sets $D_i \in \Phi$ as
\begin{eqnarray*}
\bigvee_{i \in I} D_i = \bigcap \{D \in \Phi: \bigcup_{i \in I} D_i \subseteq D\}.
\end{eqnarray*}
Note that, if the family of coherent sets $D_i$ has no upper bound in $\mathcal{C}$, then its join is simply $\mathcal{L}(\Theta)$. In this topped $\cap$-structure, 
\begin{eqnarray*}
\mathcal{C}(D') = \bigcap \{D \in \Phi: D' \subseteq D\}
\end{eqnarray*}
is a closure (or consequence) operator on the subsets of gambles, that is, $\mathcal{C}$ satisfies the following properties:
\begin{enumerate}
\item $D \subseteq \mathcal{C}(D)$
\item $D \subseteq D'$ implies $\mathcal{C}(D) \subseteq \mathcal{C}(D')$
\item $\mathcal{C}(\mathcal{C}(D)) = \mathcal{C}(D)$.
\end{enumerate}
For further reference, we prove the following well-know result for closure operators.

\begin{lemma} \label{le:UnionClosSets}
For any set of gambles,
\begin{eqnarray*}
\mathcal{C}(\mathcal{C}(D_1) \cup D_2) = \mathcal{C}(D_1 \cup D_2).
\end{eqnarray*}
\end{lemma}

\begin{proof}
Since $D_1 \cup D_2 \subseteq \mathcal{C}(D_1) \cup D_2$ we have $\mathcal{C}(\mathcal{C}(D_1) \cup D_2) \supseteq \mathcal{C}(D_1 \cup D_2)$. On the other hand $D_1,D_2 \subseteq D_1 \cup D_2$ so that $\mathcal{C}(D_1) \cup D_2 \subseteq \mathcal{C}(D_1 \cup D_2)$, thus $\mathcal{C}(\mathcal{C}(D_1) \cup D_2) \subseteq \mathcal{C}(D_1 \cup D_2)$. This proves equality.
\end{proof}

Note that $\mathcal{C}(D) = \mathcal{E}(D)$ if $0 \not\in \mathcal{E}(D)$, that is if $\mathcal{E}(D)$ is coherent. Otherwise we may have $\mathcal{E}(D) \not= \mathcal{L}(\Theta)$.  These results prepare the way below to an information algebra of coherent sets of gambles.  

A further important class of coherent sets of gambles are \textit{strictly desirable gambles} $D^+$. In addition to the conditions 1.) to 4.) above for coherence the following condition is added:
\begin{enumerate}
\item[5] $f \in D^+$ implies either $f \geq 0, f \not = 0$ or $f - \delta \in D^+$ for some $\delta > 0$.
\end{enumerate}
So, strictly desirable gambles are coherent, they form a subfamily $\Phi^+$ of coherent sets of gambles.

Another concept is given by  \textit{almost desiriable gambles}, satisfying the following conditions \cite{walley91}\begin{enumerate}
\item $f \in \bar{D}$ implies $\sup f \geq 0$,
\item $\inf f > 0$ implies $f \in \bar{D}$,
\item $f,g \in \bar{D}$ implies $f + g \in \bar{D}$,
\item $f \in \bar{D}$ and $\lambda > 0$ imply $\lambda \cdot f \in \bar{D}$,
\item $f + \delta \in \bar{D}$ for all $\delta > 0$ implies $f \in \bar{D}$.
\end{enumerate}
Such a set is no more coherent since it contains $f = 0$. But we remark that almost desirable sets of gambles again form a $\cap$-system, still topped by $\mathcal{L}(\Theta)$. Therefore, they form a complete lattice under inclusion too. So, we may define the natural extension of a set $D'$ to an almost desirable set of gambles as before as the smallest such set, containing $D'$, provided $D'$ is contained in an almost desirable set of gambles
\begin{eqnarray*}
\bar{\mathcal{C}}(D') = \bigcap \{\bar{D}:D' \subseteq \bar{D}\}.
\end{eqnarray*}
This  is still a closure operator on subsets of gambles.

So far we have considered sets of gambles in $\mathcal{L}(\Theta)$ relative to a fixed set of possibilities $\Theta$. As in set algebras, Section \ref{sec:SetAlg}, we consider a set of question $56t454565454545675treeo$k, each question $x \in Q$ represented by an equivalence relation $\theta \equiv_x \theta'$ on the set of possibilities $\Theta$. Recall that such an equivalence relation induces a partition $P_x$ of equivalence classes, and these partitions are ordered by $P_x \leq P_y$ if any block (equivalence class) of $P_y$ is contained in a block of $P_x$. A gamble $f$ which is constant on every block of a partition $P_x$, that is $f(\theta) =  f(\theta')$ if $\theta \equiv_x \theta'$, is called $x$-measurable. The subset of $x$-measurable gambles in $\mathcal{L}(\Theta)$ is denoted by $\mathcal{L}_x$.
 
We define now the operations of combination, capturing aggregation of pieces of belief, and extraction, describing filtering the part of information relative to a question $x$, among the augmented sets of coherent gambles $\Phi = \mathcal{C}(\Theta) \cup \{\mathcal{L}(\Theta)\}$ on $\Theta$ and for $x \in Q$. Combination is essentially union of  the sets of desirable gambles defining the two pieces of information, followed by closure. Extraction filters out the part of desirable gambles which are $x$-measurable by intersection with $\mathcal{L}_x$, again followed by closure.

\begin{enumerate}
\item Combination: $D_1 \cdot D_2 = \mathcal{C}(D_1 \cup D_2)$,
\item Extraction: $\epsilon_x(D) = \mathcal{C}(D \cap \mathcal{L}_x)$.
\end{enumerate}
Define $\mathcal{C}_x(D) = \mathcal{C}(D) \cap \mathcal{L}_x$ so that $\epsilon_x(D) = \mathcal{C}(\mathcal{C}_x(D))$ if $D$ is coherent. Note that $\mathcal{L}(\Theta)$ is the null element of combination since $\mathcal{C}(D_1 \cup D_2) = \mathcal{L}(\Theta)$ if $D_1 \cup D_2$ is not coherent, and $\mathcal{L}(\Theta)^+$ is the unit element of combination.  As usual, the null element signals contradiction, it destroys any other piece of information. The unit or neutral element represents vacuous information. It changes no other piece of information. To simplify notation we denote the null and unit element in the sequel by $0$ and $1$. Then $(\Phi,\cdot)$ is a commutative, idempotent semigroup with null and unit elements. The information order is defined by $D_1 \leq D_2$ if $D_1 \cdot D_2 = D_2$. Then $D_1 \leq D_2$ if and only if $D_1 \subseteq D_2$. In this order, the combination $D_1 \cdot D_2$ is the supremum or join of $D_1$ and $D_2$, since $\Phi$ is a lattice,
\begin{eqnarray*}
D_1 \cdot D_2 = D_1 \vee D_2.
\end{eqnarray*}
Note also that $\epsilon_x(D) \leq D$ and also $D_1 \leq D_2$ implies $\epsilon_x(D_1) \leq \epsilon_x(D_2)$.

We state and prove now the fundamental theorems about the extraction operator.

\begin{theorem} \label{th:ExQuant}
For all $D,D_1,D_2 \in \Phi$  and $x \in Q$ we have
\begin{enumerate}
\item $\epsilon_x(0) = 0$,
\item $\epsilon_x(D) \leq D$,
\item $\epsilon_x(\epsilon_x(D_1) \vee D_2) = \epsilon_x(D_1) \vee \epsilon_x(D_2)$.
\end{enumerate}
\end{theorem}

\begin{proof}
The first two items are obvious.

For item 3 define, using Lemma \ref{le:UnionClosSets},
\begin{eqnarray*}
A &=& \mathcal{C}_x(\mathcal{C}_x(D_1) \cup D_2) \cap \mathcal{L}_x =  \mathcal{C}((D_1 \cap \mathcal{L}_x) \cup D_2) \cap \mathcal{L}_x, \\
B &=&  \mathcal{C}( \mathcal{C}_x(D_1)) \cup  \mathcal{C}_x(D_2)) =  \mathcal{C}((D_1 \cap \mathcal{L}_x) \cup (D_2 \cap  \mathcal{L}_x)).
\end{eqnarray*}
Then $ \mathcal{C}(A) = \epsilon_x(\epsilon_x(D_1) \vee D_2)$ and $B =  \epsilon_x(D_1) \vee \epsilon_x(D_2)$. Obviously we have $B \subseteq  \mathcal{C}(A)$. We claim first that $\epsilon_x(D_1) \vee D_2 = 0$ if and only if $\epsilon_x(D_1) \vee \epsilon_x(D_2) = 0$. Indeed, if the latter equals $0$, so does the former. 

Conversely, $\epsilon_x(D_1) \cdot D_2 = 0$ means that $\mathcal{C}(\mathcal{C}(D_1 \cap \mathcal{L}_x) \cup D_2) = \mathcal{C}((D_1 \cap \mathcal{L}_x) \cup D_2) = \mathcal{L}$. If $D_1 = \mathcal{L}$ or $D_2 = \mathcal{L}$, then trivially $\epsilon_x(D_1) \cdot \epsilon_x(D_2) = \mathcal{L}$. Therefore assume that both $D_1$ and $D_2$ are coherent. Then $\mathcal{C}((D_1 \cap \mathcal{L}_x) \cup D_2) = \mathcal{L}$ implies $0 \in \mathcal{E}((D_1 \cap \mathcal{L}_x) \cup D_2)$ by definition of $\mathcal{C}$. So there are gambles $f \in D_1 \cap \mathcal{L}_x$ and $g \in D_2$ so that $0 = f + g$. Therefore $g = -f$ is $x$-measurable, since $f$ is so, hence $g \in D_2 \cap \mathcal{L}_x$. From this it follows that $0 \in \epsilon_x(D_1) \cdot \epsilon_x(D_2)$, hence $\epsilon_x(D_1) \cdot \epsilon_x(D_2) = \mathcal{L}$.

Assume now that $D_1 \vee D_2$ is coherent and consider a gamble $f \in A$. Then $f \in  \mathcal{L}_x$ and
\begin{eqnarray*}
f \geq \lambda g + \mu h, \quad g \in D_1 \cap \mathcal{L}_x, h \in D_2, \quad \lambda,\mu \geq 0, f \not = 0.
\end{eqnarray*}
So, we have $f = \lambda g + \mu h + h'$, where $h' \geq 0$. Since both $f$ and $g$ are $x$-measurable, $\mu h + h'$ must be $x$-measurable either. This means that $\mu h + h' \in D_2 \cap \mathcal{L}_x$, and therefore $f \in B$, hence $\mathcal{C}(A) = B$. This concludes the proof.
\end{proof}

Thus $\epsilon_x$ is an \textit{existential quantifier}. Item 3 can also be written as
\begin{eqnarray*}
\epsilon_x(\epsilon_x(D_1) \cdot D_2) = \epsilon_x(D_1) \cdot \epsilon_x(D_2)
\end{eqnarray*}
This shows that $(\Phi,\cdot,0,1;E)$ with $E = \{\epsilon_x:x \in Q\}$ is a domain-free information algebra. 

In this algebra, extraction commutes with intersection.

\begin{theorem} \label{th:CommExtrInt}
Let $D_j$ for $j \in J$ be any family of sets of gambles fro $\Phi$ and $x \in Q$. Then
\begin{eqnarray} \label{eq:ExtrCommMeet}
\epsilon_x(\bigcap_{j \in J} D_j) = \bigcap_{j \in J} \epsilon_x(D_j).
\end{eqnarray}
\end{theorem}

\begin{proof}
If all $D_j = \mathcal{L}(\Theta)$, then (\ref{eq:ExtrCommMeet}) holds trivially. Otherwise, eliminate all $D_j = \mathcal{L}(\Theta)$ from the family, so that we may assume that all elements $D_j$ are coherent sets of gambles. We have
\begin{eqnarray*}
\epsilon_x(\bigcap_{j \in J} D_j) &=& \mathcal{C}((\bigcap_{j \in J} D_j) \cap \mathcal{L}_x), \\
\bigcap_{j \in J} \epsilon_x(D_j) &=& \bigcap_{j \in J} \mathcal{C}(D_j \cap \mathcal{L}_x).
\end{eqnarray*}
Consider first a gamble $f$ in $\epsilon_x(\bigcap_{j \in J} D_j)$, so that $f = \lambda g + \mu h$, where $\lambda,\mu$ are nonnegative and not both equal to zero, and $g \in (\bigcap_{j \in J} D_j) \cap \mathcal{L}_x = \bigcap_j (D_j \cap \mathcal{L}_x) \subseteq \bigcap_j \mathcal{C}(D_j \cap \mathcal{L}_x)$ and $h \in \mathcal{L}^+(\Theta)$. Therefore we have $f \in \bigcap_{j \in J}(\epsilon_x(D_j))$. 

Conversely, consider a gamble $f \in \bigcap_{j \in J}(\epsilon_x(D_j)) $. If $f \in \mathcal{L}^+(\Theta)$, then $f \in \epsilon_x(\bigcap_{j \in J} D_j)$. Otherwise we have $f \geq g_j$ for some $g_j \in D_j \cap \mathcal{L}_x$ and this for all $j \in J$. Define
\begin{eqnarray*}
g(\theta) = \sup_{j \in J} g_j(\theta).
\end{eqnarray*}
Then $f \geq g$ and $g \in D_j$ for all $j$ and $g$ is $x$-measurable. Therefore we have $g \in (\cap_j D_j )\cap \mathcal{L}^+(\Theta)$, hence $f \in \epsilon_x(\bigcap_{j \in J} D_j)$.
\end{proof}

An information algebra like $\Phi$, where $(\Phi,\leq)$ is a lattice under information order and satisfies the condition of this theorem is called a \textit{lattice information algebra}.

What is the role of strictly desirable gambles in the information algebra of coherent sets of gambles? Here is the answer:

\begin{proposition} \label{prop:SubalgStr.des}
The family of strictly desirable gambles $\Phi^+$ is a subalgebra of the information algebra $\Phi$
\end{proposition}

\begin{proof}
Obviously, $\mathcal{L}$ and $\mathcal{L}^+$ belong to $\Phi^+$.

Consider then two sets of strictly desirable gambles $D_1^+$ and $D_2^+$ from $\Phi^+$. If $D_1^+ \cdot D_2^+ = \mathcal{L}$, then the combination belongs trivially to $\Phi^+$. Therefore assume $D_1^+ \cdot D_2^+$ to be coherent. Then, if $f \in D_1^+ \cdot D_2^+$, we have $f \geq g_1 + g_2$ with $g_1 \in D_1^+$ and $g_2 \in D_2^+$. If neither $g_1 \in \mathcal{L}^+$ nor $g_2 \in \mathcal{L}^+$, there are $\delta_1 > 0$ and $\delta_2 > 0$ such that $g_1 - \delta_1 \in D_1^+$ and $g_2 - \delta_2 \in D_2^+$. It follows that $f - \delta = (g_1 - \delta_1) + (g_2 - \delta_2) \in D_1^+ \cdot D_2^+$, where $\delta = \delta_1 + \delta_2 > 0$. If either $g_1 \in \mathcal{L}^+$ or $g_1 \in \mathcal{L}^+$, then $f \geq g_2$ or $f \geq g_1$ and then $f - \delta_2$ or $f - \delta_1$ belong still to $D_1^+ \cdot D_2^+$. Finally if both $g_1$ and $g_2$ belong to $\mathcal{L}^+$ then so does $f$. This shows that $D_1^+ \cdot D_2^+$ is strictly desirable, and $\Phi^+$ is closed under combination.

Similarly, if $D^+ \not= \mathcal{L}$, $\epsilon_x(D^+) = posi((D + \mathcal{L}_x) \cup \mathcal{L}^+)$. So, if $f \in \epsilon_x(D^+)$ and $f \not\in \mathcal{L}^+$, then $f \geq g \in D^+ \cup \mathcal{L}_x$ and $g \notin \mathcal{L}^+$ and if $D^+$ is strictly desirable, then there is a $\delta > 0$ such that $g - \delta \in D^+ \cap \mathcal{L}_x$, hence $f - \delta \in D \cap \mathcal{L}^+$. This shows that $\epsilon_x(D^+)$ is strictly desirable, if $D^+$ is so, hence $\Phi^+$ is also closed under extraction for all $x \in Q$, therefor indeed a subalgebra of $\Phi$.
\end{proof}

By this proposition, $\Phi^+$ is itself an information algebra.

Associated with a set of desirable gambles is another concept, namely the one of lower (and upper) previsions. This will be discussed in the next section and we shall show that it gives rise to another information algebra.


\section{Lower Previsions} \label{subsec:LowPrev}

Associated with a set of gambles $D$ on $\mathcal{L}(\Theta)$ is the lower prevision
\begin{eqnarray} \label{eq:LowPrevision}
\underline{P}(f) = \sup\{\mu \in \mathbb{R}:f - \mu \in D\}.
\end{eqnarray}

We remark that $\underline{P}(f)$ is only defined if the set $\{\mu \in \mathbb{R}:f - \mu \in D\}$ is not empty and bounded from above. For coherent sets $D$, the lower prevision is defined on the whole set of gambles as the following lemma shows. We write $\sigma(D)$ for the lower prevision associated with $D$ by (\ref{eq:LowPrevision}) and $dom(\underline{P})$ for the set of gambles for which $\underline{P}$ is defined.

\begin{lemma} \label{DomOfLowePrev}
For a set $D$ of gamble $D \subseteq \mathcal{L}(\Theta)$ we have
\begin{enumerate}
\item if $0 \not\in \mathcal{E}(D)$, then $D \subseteq dom(\sigma(D))$,
\item if $D \in \mathcal{C}(\Theta)$, then $dom(\sigma(D)) = \mathcal{L}(\Theta)$.
\end{enumerate}
\end{lemma}

\begin{proof}
1.) Consider $f \in D$. Then the set $\{\mu:f - \mu \in D\}$ is not empty, since it contains at least $0$. Further, assume $f - \mu \in D$. Then $\mu \geq \sup f$ is not possible, since otherwise $f - \mu < 0$ and this would imply $0 \in \mathcal{E}(D)$. So the set $\{\mu:f - \mu \in D\}$ is bounded from above, hence $D \in dom(\sigma(D))$.

2.) If $D$ is a coherent set of gambles, then $0 \not\in D$ and $D = \mathcal{E}(D)$. So by item 1 we have $D \subseteq dom(\sigma(D))$. Consider then a gamble $f \in \mathcal{L}(\Theta) - D$. Then $\inf f \leq 0$ and if $\mu < \inf f$, then $f - \mu \geq 0$, hence $f - \mu \in D$. So the set $\{\mu:f - \mu \in D\}$ is not empty. And we must have $\mu < 0$ in the set $\{\mu:f - \mu \in D\}$, since $\mu \geq 0$ would imply $f - \mu \leq f$, hence $f \in D$ contrary to the assumption. So the set $\{\mu:f - \mu \in D\}$ has $0$ as an upper bound and $f \in dom(\sigma(D))$.
\end{proof}

In the case that non-empty the set $\{\mu:f - \mu \in D\}$ is not bounded from above, we set $\underline{P}(f) = \infty$. If $D$ is a coherent set of gambles, then the functional $\underline{P}(f)$ on $\mathcal{L}(\Theta)$ is called a coherent lower prevision. It is characterized by the following properties \cite{walley91}: For every $f,g \in \mathcal{L}(\Theta)$,
\begin{enumerate}
\item $\underline{P}(f) \geq \inf_{\theta \in \Theta} f(\theta)$,
\item $\underline{P}(\lambda f) = \lambda \underline{P}(f)$,
\item $\underline{P}(f + g) \geq \underline{P}(f) + \underline{P}(g)$.
\end{enumerate}
There is also the upper prevision, defined by
\begin{eqnarray*}
\bar{P}(f) = \inf\{\mu \in \mathbb{R}:\mu - f \in D\} = - \underline{P}(-f).
\end{eqnarray*}
It is called coherent, if the associated lower prevision is. 

Let as before $\Phi = \mathcal{C}(\Theta) \cup \{\mathcal{L\}}$ denote the elements of the domain-free information algebra of coherent sets of gambles (see Section \ref{subsec:Gambles}). Similarly, let $\underline{\Psi} =  \underline{\mathcal{P}}(\Theta) \cup \{\infty\}$ denote the family of coherent lower previsions, augmented by the infinite prevision $\underline{P}(f) = \infty$ for all $f \in \mathcal{L}$. There is a map $\sigma$ from any set $D$ of gambles to lower previsions defined by (\ref{eq:LowPrevision}), which assigns to any set of gambles the corresponding lower prevision. We shall be especially interested in this map restricted to the domain of coherent sets of gambles in $\mathcal{C}(\Theta)$. Then the images are coherent lower previsions. This map is not one-to-one as different coherent sets of gambles may induce the same lower prevision. 

Now, among lower previsions in $\underline{\mathcal{P}}(\Theta)$ we define $\underline{P} \leq \underline{Q}$ if $\underline{P}(f) \leq \underline{Q}(f)$ for all $f$ in $\mathcal{L}(\Theta)$. This is a partial order. Note that $\sigma$ applied to coherent sets of gambles preserves order. We recall that the map $\sigma$ restricted to almost desirable sets of gambles is one-to-one \cite{walley91}, and
\begin{eqnarray} \label{eq:AlmDesSetAndPrev}
\underline{P}(f) = \max\{\mu:f - \mu \in \bar{D}\}, \quad \bar{D} = \{f:\underline{P}(f) \geq 0\}.
\end{eqnarray}
The map $\sigma$ restricted to almost desirable sets of gambles maintains also order: $\bar{D}' \leq \bar{D}$ if and only if $\sigma(\bar{D}') \leq \sigma(\bar{D})$. There is also a one-to-one relation between coherent lower previsions $\underline{P}$ and strictly desirable sets of gamble $D^+$, so that, \cite{walley91} 
\begin{eqnarray*}
\underline{P}(f) = \sup\{\mu:f - \mu \in D^+\}, \quad D^+ = \{f: \underline{P}(f) > 0\} \cup \mathcal{L}^+(\Theta).
\end{eqnarray*}

Define  the maps $\tau$ and $\bar{\tau}$ from coherent lower previsions to strictly desirable sets of gambles and almost desirable sets of gambles accordingly by
\begin{eqnarray*}
\tau(\underline{P}) = \{f: \underline{P}(f) > 0\} \cup \mathcal{L}^+(\Theta), \quad \bar{\tau}(\underline{P}) =  \{f:\underline{P}(f) \geq 0\}.
\end{eqnarray*}
Then $\tau$ and $\bar{\tau}$ are the inverses of the map $\sigma$ restricted to strictly desirable and almost desirable sets of gambles respectively. The following lemma shows how coherent, strictly desirable and almost desirable sets are linked relative to the coherent lower previsions they induce \footnote{This result follows also from the fact that, in the sup-norm topology of the linear space $\mathcal{L}(\Theta)$, the strictly desirable gambles $D^+$ are the relative interior of $D$ plus the non-negative, non-zero gambles and $\bar{D}$ is the relative closure of $D$, \cite{walley91}.}

\begin{lemma} \label{le:StrictAlmGamb}
Let $D$ be a coherent set of gambles. Then
\begin{eqnarray*}
D^+ = \tau(\sigma(D)) \subseteq D \subseteq \bar{\tau}(\sigma(D))
\end{eqnarray*}
and $\sigma(D^+) = \sigma(D) = \sigma(\bar{D})$.
\end{lemma}

\begin{proof}
Let $\underline{P} = \sigma(D)$. Then $f \in D^+$ means that $0 < \underline{P}(f) = sup\{\mu:f - \mu \in D\}$ or $f \in \mathcal{L}^+$. In the second case $f \in D$. Otherwise there is a $\delta$ so that $0 < \delta < \underline{P}(f)$ and $f-\delta \in D$. Therefore $f \in D$ and $D^+ \subseteq D$. Further, consider $f \in D$. Then we must have $\underline{P}(f) = \sup\{\mu:f - \mu \in D\} \geq 0$, hence $f \in \bar{D}$. The second part follows since $\tau$ and $\bar{\tau}$ are the inverse maps of $\sigma$ on strictly desirable and almost desirable sets of gambles.
\end{proof}

Next, we claim that the map $\sigma$ restricted to coherent sets of gambles preserve infima. Here we define $\inf\{\underline{P}_j:j \in J\}$ by $\inf\{\underline{P}_j:j \in J\}(f) = \inf\{\underline{P}_j(f):j \in J\}$ for all $f \in \mathcal{L}(\Theta)$.

\begin{lemma} \label{le:MaintainInf}
Let $D_j$, $j \in J$ be any family of coherent sets in $\mathcal{C}(\Theta)$. Then we have
\begin{eqnarray*}
\sigma(\bigcap_{j \in J} D_j) = \inf\{\sigma(D_j)\}
\end{eqnarray*}
\end{lemma}

\begin{proof}
Recall that the intersection of the coherent sets $D_j$ is a coherent set $D$ and $D \subseteq \bar{D}$. Then let
\begin{eqnarray*}
\sigma(D) = \sigma(\bigcap_{j \in J} D_j) = \underline{P}.
\end{eqnarray*}
The coherent lower prevision $\underline{P}$ is a lower bound of the $\sigma(D_j)$. Consider any other coherent lower prevision $\underline{Q}$, which a lower bound of the coherent lower previsions $\sigma(D_j)$ so that $\tau(\underline{Q}) \subseteq \tau(\sigma(D_j)) = D_j^+ \subseteq D_j$. Then we have $\tau(\underline{Q}) \subseteq \bigcap_j D_j = D$ and this implies $\underline{Q} = \sigma(\tau(\underline{Q})) \leq \sigma(D) = \underline{P}$, hence $\underline{P}$ is the infima of the $\sigma(D_j)$.
\end{proof}

If $\underline{P}'$ is a lower prevision which is dominated by a coherent lower prevision, then its natural extension is defined as the infimum of the coherent lower prevision which dominate it, \cite{walley91},
\begin{eqnarray} \label{eq:NatExtLowPrev}
E(\underline{P}') = \inf\{\underline{P} \textrm{ coherent}\ : \underline{P}' \leq \underline{P}\}.
\end{eqnarray}
So, $E(\underline{P})$ is the minimal coherent lower prevision which dominates $\underline{P}'$. Now, we prove the key result, that the map $\sigma$ commutes with natural extension.

\begin{theorem} \label{th:CommSigmaExt}
Let $D'$ be a set of gambles which satisfies the following two consitions:
\begin{enumerate}
\item $0 \not\in \mathcal{E}(D')$,
\item for all $f \in D' - \mathcal{L}^+(\Theta)$ there exists a $\delta > 0$ such that $f - \delta \in D'$.
\end{enumerate}
Then we have
\begin{eqnarray*}
\sigma(\mathcal{C}(D')) = E(\sigma(D')).
\end{eqnarray*}
\end{theorem}

\begin{proof}
If $D' = \mathcal{L}^+(\Theta)$, then $D' = \mathcal{C}(D')$ and $\sigma(\mathcal{C}(D')) = E(\sigma(D')) = \sigma(D')$ since $\sigma(D')$ is already coherent. So, assume $D' \not= \mathcal{L}^+(\Theta)$. Then by the first assumption, $\mathcal{E}(D') = \mathcal{C}(D')$ so that (Lemma \ref{le:MaintainInf}),
\begin{eqnarray*}
\sigma(\mathcal{C}(D')) = \sigma\{\bigcap\{D \textrm{ coherent}:D' \subseteq D\} = \inf\{\sigma(D):D \textrm{ coherent}:D' \subseteq D\}.
\end{eqnarray*}
It follows that $\sigma(\mathcal{C}(D')) \geq E(\sigma(D'))$. Consider now any coherent lower prevision $\underline{P}$ so that $\underline{P}' = \sigma(D') \leq \underline{P}$. We claim that $D' \subseteq \tau(\underline{P})$. Indeed, if $f \in D'$ then $\underline{P}'(f) \geq 0$. If  $f \in \mathcal{L}^+(\Theta)$, then $f \in \tau(\underline{P})$. Otherwise, if $f \in D' - \mathcal{L}^+(\Theta)$, then there is by assumption a $\delta > 0$ so that $f - \delta \in D'$, hence we have $0 < \underline{P}'(f) \leq \underline{P}(f)$. But this means that $f \in \tau(\underline{P})$. Since a strictly desirable set of gambles is coherent, it follows, using Lemma \ref{le:MaintainInf} and $\sigma(\tau(\underline{P})) = \underline{P}$, that
\begin{eqnarray*}
\sigma(\mathcal{C}(D'))  \leq \sigma(\bigcap \{\tau(\underline{P}):D' \subseteq \tau(\underline{P})\}) = \inf \{\underline{P}:\underline{P}' \leq \underline{P}\} = E(\underline{P}')
\end{eqnarray*}
so that $\sigma(\mathcal{C}(D')) = E(\sigma(D'))$.
\end{proof}

We can now introduce into $\underline{\Psi}$ like in $\Phi$ operations of combination and extraction. As before consider the family of questions $Q$ together with associated equivalence relations $\equiv_x$ on $\Theta$ and partitions $P_x$ for all $x \in Q$. Consider then for two coherent lower previsions $\underline{P}_1$ and $\underline{P}_2$
\begin{eqnarray*}
\underline{P}'(f) = \max\{\underline{P}_1(f),\underline{P}_2(f)\}
\end{eqnarray*}
or $\underline{P}' = \max\{\underline{P}_1,\underline{P}_2\}$. We may take the natural extension of $E(\underline{P}')$ to define combination of two coherent lower previsions $\underline{P}_1$ and $\underline{P}_2$. For extraction, we may take the natural extension of the marginal $\underline{P}_x$ of $\underline{P}$, defined as the restriction of $\underline{P}$ to $\mathcal{L}_x$. Thus, in summary, we define $\underline{P}_1 \cdot \underline{P}_2$ and $\b{e}_x(\underline{P})$ by
\begin{enumerate}
\item Combination: $\underline{P}_1 \cdot \underline{P}_2(f) = E(\max\{\underline{P}_1,\underline{P}_2\})(f)$, if $\max\{\underline{P}_1,\underline{P}_2\}$ is dominated by a coherent lower prevision, $\underline{P}_1 \cdot \underline{P}_2(f) = \infty$ otherwise.
\item Extraction: $\b{e}_x(\underline{P})(f) = E(\underline{P}_x)(f)$.
\end{enumerate}

Using Theorem \ref{th:CommSigmaExt} linking natural extensions in the two formalisms of coherent sets of gambles and coherent lower previsions, the following theorem permits to conclude that the set $\underline{\Psi}$ of coherent lower previsions $\underline{\mathcal{P}}(\Theta)$ augmented by $\underline{P}(f) = \infty$ forms a domain-free information algebra under these operations.

\begin{theorem} \label{th:HomomLowPrev}
Consider the the map $\sigma$ restricted to the algebra of strictly desirable gambles $\Phi^+$. Then, for any $D_1^+,D_2^+$ and $D^+$ in $\Phi^+$ and any $x \in Q$,
\begin{enumerate}
\item $\sigma(D^+_1 \cdot D^+_2) = \sigma(D^+_1) \cdot \sigma(D^+_2)$,
\item $\sigma(\mathcal{L}(\Theta))(f) = \infty$, $\sigma(\mathcal{L}^+(\Theta))(f) = \inf f$ for all $f \in \mathcal{L}(\Theta)$.
\item $\sigma(\epsilon_x(D^+)) = \b{e}_x(\sigma(D^+))$.
\end{enumerate}
\end{theorem}

\begin{proof}
Assume first that $D^+_1 \cdot D^+_2 = \mathcal{L}(\Theta)$ and let $\underline{P}_1 = \sigma(D^+_1), \underline{P}_2 = \sigma(D^+_2)$. Then there can be no coherent prevision $\underline{P}$ dominating both $\underline{P}_1$ and $\underline{P}_2$. Because otherwise we would have $D^+_1 = \tau(\underline{P}_1)$ and  $D^+_2 = \tau(\underline{P}_2)$ both contained in the coherent set $\tau(\underline{P})$, But this contradicts $D^+_1 \cdot D^+_2 = \mathcal{L}(\Theta)$. Therefore, $\sigma(D^+_1 \cdot D^+_2)(f) = \infty$ for all gambles $f$ in $\mathcal{L}(\Theta)$. 

Assume then $D^+_1 \cdot D^+_2 \not= \mathcal{L}(\Theta)$. Then $D^+_1 \cdot D^+_2$ as well as $D^+_1 \cup D^+_2$ satisfy the condition of Theorem \ref{th:CommSigmaExt}. Therefore we have
\begin{eqnarray*}
\lefteqn{\sigma(D^+_1 \cdot D^+_2) = \sigma(\mathcal{C}(D^+_1 \cup D^+_2)) = E(\sigma(D^+_1 \cup D^+_")) } \\
&&= E(\max\{\sigma(D^+_+),\sigma(D^+_2)\}) = \sigma(D^+_1) \cdot \sigma(D^+_2).
\end{eqnarray*}
This proves item 1. 

Item 2 is obvoous.

For 3. remark that $D^+ \cap \mathcal{L}_x$ satisfy the condiktions of Theorem \ref{th:CommSigmaExt}. Thus we obtain
\begin{eqnarray*}
\sigma(\epsilon_x(D^+)) = \sigma(\mathcal{C}(D^+ \cap \mathcal{L}_x)) = E(\sigma(D^+ \cap \mathcal{L}_x)).
\end{eqnarray*}
Now,
\begin{eqnarray*}
\sigma(D^+ \cap \mathcal{L}_x) = \sup\{\mu:f - \mu \in D^+ \cap \mathcal{L}_x\}.
\end{eqnarray*}
But $f - \mu \in D^+ \cap \mathcal{L}_x$ implies that $f$ is $x$-measurable and $f - \mu \in D^+$. Therefore, we conclude that $\sigma(D^+ \cup \mathcal{L}_x) = \sigma(D^+)_x$. But we have $E(\sigma(D^+)_x) = e_x(\sigma(D^+))$. This concludes the proof.
\end{proof}

Note that the map $\sigma$ restricted to $\Phi^+$ is bijective. This theorem shows then that $\Psi = \underline{\mathcal{P}}(\Theta) \cup \{\sigma(\mathcal{L}(\Theta)\}$ is, under the operations of combination and extraction defined above, a domain-free information algebra, isomorphic to the information algebra $\Phi^+$, the algebra of strictly desirable sets of gambles under the maps $D^+ \mapsto \sigma(D^+)$ and $\epsilon_x \mapsto \underline{e}_x$. Inversely, under the inverse maps, $\underline{\Psi}$ is embedded in the information algebra $\Phi$ of coherent sets of gambles. There is obviously the connected (isomorphic) information algebra of upper previsions. We shall see below that there are other homomorphisms and isomorphisms between lower previsions and sets of gambles.

It follows further from Theorem \ref{th:HomomLowPrev}  and Lemma \ref{le:MaintainInf} that for any family of strictly desirable sets 
\begin{eqnarray*}
\sigma(\epsilon_x(\bigcap_j D^+_j)) =  \underline{e}_x(\sigma(\bigcap_j D^+_j)) =  \underline{e}_xs(\inf\{\sigma(D^+_j))\}
\end{eqnarray*}
and
\begin{eqnarray*}
\sigma(\bigcap_j \epsilon_x(D^+_j)) = \inf\{\sigma(\epsilon_x(D^+_j))\} = \inf\{\underline{e}_x(\sigma(D^+_j))\}.
\end{eqnarray*}
so that for any family of coherent lower previsions $\underline{P}_j$ we have also by Theorem \ref{th:CommExtrInt},
\begin{eqnarray*}
\underline{e}_x(\inf\{\underline{P}_j\}) = \inf\{\underline{e}_x(\underline{P}_j)\}.
\end{eqnarray*}
In the information algebra of lower prevision extraction distributes over meet (infimum) as in the algebra of coherent sets of gambles. 

We come back to the relations between the information algebra of coherent lower previsions and different algebras of sets of gambles in the next  Section \ref{sec:DiffAlgs}.


\section{The Algebras of Coherent and Almost Desirable Sets of Gambles} \label{sec:DiffAlgs}

We show  in this section that there is also an information algebra of almost desirable gambles, isomorphic to the algebra of lower previsions. Then we examine the question how  the algebra of coherent sets of gambles $\Phi$ is related to the information algebra of coherent lower previsions.

Consider first sets of almost desirable gambles on $\mathcal{L}(\Theta)$, see Section \ref{subsec:Gambles}. It is no surprise that the sets of almost desirable gambles form also an information algebra. We use the algebra of lower previsions together with the bijective map $\bar{\tau}$ to define combination and extraction among almost desirable sets of gambles. Afterwards, we show how these operations may also be defined in terms of almost desirable gambles themselves. Denote by $\bar{\Phi}$ the family of almost desirable sets of gambles, including $\mathcal{L}(\Theta)$. We denote generic almost desirable sets by $\bar{D}$. Define combination and extraction as follows:
\begin{enumerate}
\item \textit{Combination:} $\bar{D}_1 \cdot \bar{D}_2 = \bar{\tau}(\sigma(\bar{D}_1) \cdot \sigma(\bar{D}_2))$,
\item \textit{Extraction:} $\bar{\epsilon}_x(\bar{D}) = \bar{\tau}(e_x(\sigma(\bar{D}))$.
\end{enumerate}
Note that we denote combination by dot, in $\bar{\Phi}$ as in $\Phi$ or $\Phi^+$ or also in $\underline{\Psi}$. It will always be clear from the context, which operation is meant. For instance in the definition above, on the left $\cdot$ denotes combination in $\bar{\Phi}$, whereas on the right it denotes combination in $\underline{\Psi}$. By this definition, it is immediately clear that by the map $\bar{\tau}$ the axioms of an information algebra are induced into $\bar{\Phi}$ from $\underline{\Psi}$ and thereby $\bar{\tau}$ becomes a homomorphism, even an isomorphism, since $\bar{\tau}$ is bijective, between the information algebra of lower previsions and the one of almost desirable sets of gambles. Furthermore, the map $D^+ \mapsto \bar{D}$ defined by $\bar{D} = \bar{\tau}(\sigma(D^+))$ is an isomorphism between the information algebras $\Phi^+$ of strictly desirable gambles and the algebra $\bar{\Phi}$ of almost desirable gambles. Note that unit element in $\bar{\Phi}$ is $\mathcal{L}^+$ completed with the null function. The null element is again $\mathcal{L}$. 

As we have seen in Lemma \ref{le:StrictAlmGamb} we have $D^+ = \tau(\sigma(D)) \subseteq D \subseteq \bar{\tau}(\sigma(D)) = \bar{D}$. We mentioned that $\bar{D} = \bar{\tau}(\sigma(D))$ is the topological closure of the coherent set $D$ in the sup-norm topology on $\mathcal{L}(\Theta)$, \cite{walley91}. Consider then two coherent sets $D_1$ and $D_2$ and let $D = D_1 \cdot D_2$, $D^+ = D^+_1 \cdot D^+_2$. Then, using Theorem \ref{th:HomomLowPrev},
\begin{eqnarray*}
\bar{D} = \bar{D}_1 \cdot \bar{D}_2 = \bar{\tau}(\sigma(D^+_1) \cdot \sigma(D^+_2)) = \bar{\tau}(\sigma(D_1) \cdot \sigma(D_2)) = \bar{\tau}(\sigma(D^+_1 \cdot D^+_2)).
\end{eqnarray*}
We denote the topological closure operator in the sup-norm in $\mathcal{L}(\Theta)$ by $cl$. Let $cl(D) = cl(D_1 \cdot D_2) = cl(posi(D_1 \cup D_2)) = \bar{D}$, hence $\bar{D} \subseteq cl(posi(\bar{D}_1 \cup \bar{D}_2))$. But we also have $\bar{D}_1 \cup \bar{D}_2 \subseteq \bar{D}$ and since $\bar{D}$ is a closed convex cone, we must have $\bar{D} = cl(posi(\bar{D}_1 \cup \bar{D}_2))$. So, we conclude that
\begin{eqnarray*}
\bar{D}_1 \cdot \bar{D}_2 = cl(posi(\bar{D}_1 \cup \bar{D}_2)).
\end{eqnarray*}
Remark that this holds even if $\bar{D}_1 \cdot \bar{D}_2 = \mathcal{L}(\Theta)$. 

Similarly, for an almost desirable set $\bar{D}$  we have $D^+ \subseteq D \subseteq \bar{D}$ and then
\begin{eqnarray*}
\bar{D}' = \bar{\epsilon}_x(\bar{D}) = \bar{\tau}(e_x(\sigma(\bar{D}))) = \bar{\tau}(e_x(\sigma(D))) \supseteq \tau(e_x(\sigma(D)) = \epsilon_x(D^+) = D'^+.
\end{eqnarray*}
Now, $\bar{D}' = cl(D'^+) = cl(posi((D^+ \cap \mathcal{L}_x) \cup \mathcal{L}(\Theta)))$ so that $\bar{D}' \subseteq cl(posi((\bar{D} \cap \mathcal{L}_x) \cup \mathcal{L}(\Theta)$. On the other hand, $(\bar{D} \cap \mathcal{L}_x) \cup \mathcal{L}(\Theta) \subseteq \bar{D}'$. Since $\bar{D}'$ is a closed convex cone we must therefore have
\begin{eqnarray*}
\bar{\epsilon}_x(\bar{D}) = cl(posi((\bar{D} \cap \mathcal{L}_x) \cup \mathcal{L}(\Theta)))
\end{eqnarray*}
Again, this obviously holds also if $\bar{D} = \mathcal{L}(\Theta)$, the null element of the information algebra $\bar{\Phi}$.

We turn next to the information algebra $\Phi$ of coherent sets gambles. We shall prove that the information algebra $\Phi$ is in a weak form homomorphic to its subalgebra $\Phi^+$. As a preparation, we need the following lemma.

\begin{lemma} \label{le:StrictDes}
If $D$ is a coherent set of gambles and $D^+ = \tau(\sigma(D))$, then $f \not\in \mathcal{L}^+(\Theta)$ implies that $f \in D^+$ if and only if there is a $\delta > 0$ so that $f - \delta \in D$.
\end{lemma}

\begin{proof}
Since $D^+$ is a strictly desirable set of gambles contained in $D$, by the definition of strictlly desirable set we have that $f \in D^+$ and $f \not\in  \mathcal{L}^+(\Theta)$. implies that there is a $\delta > 0$ so that $f - \delta \in D^+ \subseteq D$. Conversely, consider a gamble $f$ with $\delta > 0$ such that $f - \delta \in D$ and note that $D^+ = \{f:\sigma(D)(f) > 0\} \cup  \mathcal{L}^+(\Theta)$ where $\sigma(D)(f) = \sup \{\mu:f - \mu \in D\}$. From $f - \delta \in D$ it follows that $\sigma(D)(f) > 0$, hence $f \in D^+$.
\end{proof}

Consider the map $D \mapsto D^+$ defined by $D^+ = \tau(\sigma(D))$. The next theorem establishes that this map preserves extraction and combination, if the combination is coherent.. 

\begin{theorem} \label{th:HomomCohDesSets}
Let $D_1,D_2$ and $D$ be coherent sets and $x \in Q$. Then,
\begin{itemize}
\item if $D_1 \cdot D_2 \not= 0$, then $D_1 \cdot D_2 \mapsto (D_1 \cdot D_2)^+ = D^+_1 \cdot D^+_2$,
\item $\epsilon_x(D) \mapsto (\epsilon_x(D))^+ = \epsilon_x(D^+)$.
\end{itemize}
\end{theorem}

\begin{proof}
For 1.) note first that $D^+_1 \subseteq D_1$ and $D^+_" \subseteq D_2$ so that 
\begin{eqnarray*}
D^+_1 \cdot D^+_2 = \tau(\sigma(D^+_1+ \cdot D^+_2)) \subseteq \tau(\sigma(D_1 \cdot D_2)) = (D_1 \cdot D_2)^+.
\end{eqnarray*}
Further,
\begin{eqnarray*}
(D_1 \cdot D_2)^+ = \{f:\sigma(D_1 \cdot D_2)(f) > 0 \cup \mathcal{L}^+(\Theta).
\end{eqnarray*}
So, if $f \in (D_1 \cdot D_2)^+$, then either $f \in \mathcal{L}^+(\Theta)$ or 
\begin{eqnarray*}
\sigma(D_1 \cdot D_2)(f) = \sup\{\mu:f - \mu \in \mathcal{C}(D_1 \cup D_2)\} > 0.
\end{eqnarray*}
In the first case clearly $f \in D^+_1 \cdot D^+_2$. In the second case there is a $\delta > 0$ so that $f - \delta \in \mathcal{C}(D_1 \cup D_2)$. This means that $f - \delta = h + \lambda_1f_1 + \lambda_2f_2$, where $h \in \mathcal{L}^+(\Theta)$, $f_1 \in D_1$, $f_2 \in D_2$ and $\lambda_1,\lambda_2 \geq 0$ and not both equal $0$. If both $\lambda_1$ and $\lambda_2$ are different from $0$, it follows
\begin{eqnarray*}
f = h + (\lambda_1f_1 + \delta/2) + (\lambda_2f_2 + \delta/2).
\end{eqnarray*}
Then $f'_1 = \lambda_1f_1 + \delta/2 \in D_1$ and $f'_2 = \lambda_2f_2 + \delta/2 \in D_2$. We have then $\lambda_1f_1 = f'_1 - \delta/2 \in D_1$ and $\lambda_2f_2 = f'_2 - \delta/2 \in D_2$ so that according to Lemma \ref{le:StrictDes} $f'_1 \in D^+_1$ and $f'_2 \in D^+_2$ which implies $f \in \mathcal{C}(D^+_+ \cup D^+_2) = D^+_1 \cdot D^+_2$. If one of the two coefficients $\lambda_1$ or $\lambda_2$ are null, a similar argument shows also that $f \in D^+_+ \cdot D^+_2$. This proves that $(D_1 \cdot D_2)^+ = D^+_1 \cdot D^+_2$.

To prove 2.) note that $D^+ \subseteq D$, hence
\begin{eqnarray*}
(\epsilon_x(D))^+ = \tau(\sigma(e_x(D))) \supseteq \tau(\sigma(e_x(D^+))) = \epsilon_x(D^+).
\end{eqnarray*}
Now, we have
\begin{eqnarray*}
(\epsilon_x(D))^+ = \{f:\sigma(\epsilon_x(D))(f) > 0\} \cup \mathcal{L}^+(\Theta)
\end{eqnarray*}
where
\begin{eqnarray*}
\sigma(\epsilon_x(D))(f) = \sup\{f:f -\mu \in \mathcal{C}(D \cap \mathcal{L}_x\}.
\end{eqnarray*}
So, if $f \in (\epsilon_x(D))^+$, then either $f \in \mathcal{L}^+(\Theta)$, in which case $f \in \epsilon_x(D^+)$, or there is a $\delta > 0$ so that $f - \delta \in \mathcal{C}(D \cap \mathcal{L}_x) = posi(\mathcal{L}^+(\Theta) \cup (D \cap \mathcal{L}_x))$. In this case we have $f - \delta = h + g$, where $h \in \mathcal{L}^+(\Theta)$ and $g \in D \cap \mathcal{L}_x$. It follows that $f = h + (g + \delta)$ and $g' = g + \delta$ is still $x$-measurable and $g' \in D$. Using Lemma \ref{le:StrictDes} we deduce from $g = g' - \delta \in D \cap \mathcal{L}_x$,  that $g' \in D^+ \cap \mathcal{L}_x$ hence $f \in \epsilon_x(D^+)$. So we have $\epsilon_x(D) \mapsto (\epsilon_x(D))^+ = \epsilon_x(D^+)$.
\end{proof}

We call such a map a weak homomorphism. As a corollary of this theorem, we deduce that the map $\sigma : \Phi \rightarrow \underline{\Psi}$ is also a weak homomorphism. 

\begin{corollary} \label{cor:WeakHomom}
Let $D_1,D_2$ and $D$ be coherent sets and $x \in Q$. Then
\begin{itemize}
\item if $D_1 \cdot D_2 \not= 0$, then $D_1 \cdot D_2 \mapsto \sigma(D_1 \cdot D_2) = \sigma(D_1) \cdot \sigma(D_2$,
\item $\epsilon_x(D) \mapsto \sigma(\epsilon_x(D)) = \epsilon_x(\sigma((D))$.
\end{itemize}
\end{corollary}

\begin{proof}
The result follows since $\sigma = (\sigma^+ \circ\ \tau) \circ \sigma$, where $\tau \circ \sigma$ is the weak homomorphism between $\Phi$ and $\Phi^+$ and $\sigma^+$ is the isomorphism between $\Phi^+$ and $\underline{\Psi}$.
\end{proof}

Note that in general $D^+$ is a proper subset of $D$, so it is understandable, that $D_1 \cdot D_2$ may be contradictory, whereas $D_1^+ \cdot D_2^+$ is not. For an example for this, we refer to \cite{CasaKohIIJAR}. This shows that the homomorphism can be only weak.

In summary, we have the following relations between the different information algebras $\Phi$, $\Phi^+$, $\bar{\Phi}$ and $\underline{\Psi}$, if $\sigma^+$ and $\bar{\sigma}$ denote the restrictions of $\sigma$ to $\Phi^+$ and $\bar{\Phi}$.
\begin{itemize}
\item $\tau,\sigma^+$ inverse isomorphisms between $\Phi^+$ and $\underline{\Psi}$,
\item $\bar{\tau},\bar{\sigma}$ inverse isomorphisms between $\bar{\Phi}$ and $\underline{\Psi}$,
\item $\sigma$ weak homomorphism between $\Phi$ and $\underline{\Psi}$,
\item $\tau \circ \sigma$ weak homomorphism between $\Phi$ and $\Phi^+$,
\item $\bar{\tau} \circ \sigma$ weak homomorphism between $\Phi$ and $\bar{\Phi}$,
\item $\sigma^+ \circ \bar{\tau},\bar{\sigma} \circ \tau$ inverse isomorphisms between $\Phi^+$ and $\bar{\Phi}$,
\item $id : \Phi^+ \rightarrow \Phi$ embedding of $\Phi^+$ as a subalgebra in $\Phi$.
\end{itemize}
For all these maps, we have to add the associations  between the extraction operators in the different algebras to be complete.


\section{Set algebras of possibilities and of atoms} \label{subsec:ImpProbRepres}

In this section, we show first that set algebras of possibilities can be embedded into algebras of imprecise probabilities. We argue that therefore, in a certain sense classical propositional and predicate calculi are part of imprecise probability. So consider the set of possibilities $\Omega$ as in Section \ref{subsec:Gambles} together with the family $P_Q$ of partitions $P_x$ for $x \in Q$ which determine the set $E$ of extraction operators $\epsilon_x$ of the information algebra of coherent sets of gambles (and as well those of the algebras of almost and strictly desirable sets of gambles and of lower previsions). Let further $\Sigma_Q$ be the corresponding set of saturation operators associated with the partitions $P_x$. We assume that $P_Q$ is upwards directed under our order between partitions (see Section \ref{sec:SetAlg}). Consider the associated set algebra $(\mathcal{S}_Q,\cap,\emptyset,\Omega;\Sigma_Q)$, where $\mathcal{S}_Q$ is the set of all subsets of $\Omega$, which are saturated with respect to some $x \in Q$.

We now associate with any set $S \in \mathcal{S}_Q$ a strictly desirable set of gambles $D^+_S$ defined as 
\begin{eqnarray*}
D^+_S = \{f \in \mathcal{L}(\Omega):\inf_{\omega \in S} f(\omega) > 0\} \cup \mathcal{L}^+(\Omega).
\end{eqnarray*}
If $S$ is not empty this is clearly a strictly desirable set of gambles., otherwise, if $S$ is the empty set, then define $D^+_\emptyset = \mathcal{L}(\Theta)$. The next theorem shows that the map $f : \mathcal{S}_Q \mapsto D^+_S$ is a homomorphism between the the set algebra $\mathcal{S}_Q$ and $\Phi^+$, the algebra od strictly desirable  set of gambles.

\begin{theorem} \label{th:EmbedSXetPsoos}
Let $S,T \in \mathcal{S}_Q$ and $x \in Q$. Then
\begin{enumerate}
\item $D^+_S \cdot D^+_T = D^+_{S \cap T}$,
\item $D^+_\emptyset = \mathcal{L}(\Theta)$, $D^+_\Theta = \mathcal{L}^+(\Theta)$,
\item $\epsilon_x(D^+_S) = D^+_{\sigma_x(S)}$.
\end{enumerate}
\end{theorem}

\begin{proof}
Item 2 is obvious. Now, if $S$ or $T$ equal $\Theta$ then $D_S^+ = \mathcal{L}^+$ or $D_T^+ = \mathcal{L}^+$. In this case 1.) trivially holds. Similarly, if $S = \emptyset$ or $T = \emptyset$, then $D^+_S = \mathcal{L}$ or $D^+_T = \mathcal{L}$ and again 1.) holds trivially. 

So, to prove 1.) in the other cases, suppose that both $S$ and $T$ are neither empty nor equal to $\Theta$. Assume first that $S \cap T = \emptyset$. Then $D^+_{S \cap T} = \mathcal{L}$. Consider gambles $f \in D^+_S$ and $g \in D^+_T$. and define
\begin{eqnarray*}
\tilde{f}(\theta) = \left\{ \begin{array}{ll} f(\theta) & \textrm{for}\ \theta \in S, \\ -g(\theta) & \textrm{for}\ \theta \in T, \\ 0 & \textrm{for}\ \theta \in (S \cup T)^c, \end{array} \right.
\tilde{g}(\theta) = \left\{ \begin{array}{ll} -f(\theta) & \textrm{for}\ \theta \in S, \\ g(\theta) & \textrm{for}\ \theta \in T, \\ 0 & \textrm{for}\ \theta \in (S \cup T)^c. \end{array} \right.
\end{eqnarray*}
Since $S$ and $T$ are disjoint, we have $\tilde{f} \in D^+_S$ and $\tilde{g} \in D^+_T$. However we also have $\tilde{f} + \tilde{g} = 0 \in \mathcal{E}(D^+_S \cup D^+_T)$, hence $D^+_S \cdot D^+_T = \mathcal{L}$ and 1.) is verified in this case too.

So, assume finally that $S \cap T \not= \emptyset$. We have $D^+_S \cup D^+_T \subseteq D^+_{S \cap T}$, so that $\mathcal{E}(D^+_S \cup D^+_T)$ is also coherent and hence $D^+_S \cdot D^+_T = \mathcal{E}(D^+_S \cup D^+_T) \subseteq D^+_{S \cap T}$. Consider a gamble $f$ in $D^+_{S \cap T}$. Select a number $\delta > 0$ and define two gambles 
\begin{eqnarray}
f_1(\theta) = \left\{ \begin{array}{ll} 1/2f(\theta) & \textrm{for}\ \theta \in S \cap T, \\ \delta & \textrm{for}\ \theta \in S \setminus T, \\
f(\theta) - \delta & \textrm{for}\ \theta \in T \setminus S, \\ 1/2f(\theta) & \textrm{for}\ \theta \in (S \cup T)^c,  \end{array} \right.
f_2(\theta) = \left\{ \begin{array}{ll} 1/2f(\theta) & \textrm{for}\ \theta \in S \cap T, \\ f(\theta) - \delta & \textrm{for}\ \theta \in S \setminus T, \\
\delta & \textrm{for}\ \theta \in T \setminus S, \\ 1/2f(\theta) & \textrm{for}\ \theta \in (S \cup T)^c.  \end{array} \right.
\end{eqnarray}
Then we have $f = f_1 + f_2$ and $f_1 \in D^+_S$, $f_2 \in D^+_T$. Therefore $f \in \mathcal{E}(D^+_S \cup D^+_T) = D^+_S \cdot D^+_T$, hence we conclude that $D^+_S \cdot D^+_T = D^+_{S \cap T}$.

For 3.) if $S$ is empty, then $\sigma_x(\emptyset) = \emptyset$ and $\epsilon_x(D^+_\emptyset) = \mathcal{L}(\Theta)$, so that item 3.) is valid in this case. Assume then that $S \not= \emptyset$. Then $D^+_S$ is coherent and therefore
\begin{eqnarray*}
\epsilon_x(D^+_S) = \mathcal{C}(D^+_S \cup \mathcal{L}_x) = posi(\mathcal{L}^+ \cup (D^+_S \cap \mathcal{L}_x)).
\end{eqnarray*}
Consider a gamble $f \in D^+_S \cap \mathcal{L}_x$. If $f \in \mathcal{L}^+(\Theta)$, then $f \in D^+_{\sigma_x(S)}$. Otherwise, $\inf_S f(\theta) > 0$ and $f$ is $x$-measurable. If $\theta \equiv_x \theta'$ for some $\theta' \in S$ and $\theta \in \Theta$, then $f(\theta) = f(\theta')$. Therefore $\inf_{\sigma_x(S)} f(\theta) = \inf_S f(\theta) > 0$, hence $f \in D^+_{\sigma_x(S)}$. So we have $D^+_S \cap \mathcal{L}_x \subseteq \mathcal{C}(D^+_S \cap \mathcal{L}_x) \subseteq D_{\sigma_x(S)}$.

Conversely, consider a gamble $f \in D^+_{\sigma_x(S)}$, which is a strictly desirable set of gambles. If $f \in \mathcal{L}^+(\Theta)$, then $f \in \epsilon_x(D^+_S)$. Otherwise, there is a $\delta > 0$ such that $f - \delta \in D^+_{\sigma_x(S)}$.  Define for every $\theta \in \Theta$
\begin{eqnarray*}
g(\theta) = \inf_{\theta' \equiv_x \theta} f(\theta') - \delta.
\end{eqnarray*}
If $\theta \in S$, then $g(\theta) > 0 $ since $\inf_{\sigma_x(S)} f(\theta) - \delta > 0$. So, we have $\inf_S g(\theta) \geq 0$ and $g$ is $x$-measurable. However, then $\inf_S (g(\theta) + \delta) = \inf_S g(\theta) + \delta > 0$, hence $(g + \delta) \in D^+_S \cap \mathcal{L}_x$ and $f \geq g + \delta$. Therefore we conclude that $f \in \mathcal{C}(D^+_S \cap \mathcal{L}_x) = \epsilon_x(D^+_S)$ and this shows that $\epsilon_x(D^+_S) = D^+_{\sigma_x(S)}$.
\end{proof}

The map $S \mapsto D^+_S$ is clearly injective, hence an embedding of the set algebra of subsets of $\Theta$ in the information algebra of striclly desirable set sof gambles $\Phi^+$, hence also into $\Phi$ and by extension in $\underline{\Phi}$. In this sense, imprecise probability is an extension of propositional logic, see for instance \cite{kohlas03} for information and set algebras associated. with propositional logic.

Next, we discuss another relation of imprecise probabilities to set algebras. These will be related to atoms in the algebras $\Phi$, $\Phi^+$ and $\underline{\Phi}$. It turns out that these algebras are all atomistic closed (see Section \ref{subsec:AtomicAlg}), hence  embedded into the set algebras of their respective atoms (Section \ref{subsec:SetAlgAtoms}).

A coherent set of gambles is $M$ called \textit{maximal}, if it is no proper subset of a coherent set of gambles. Such sets exist and they play an important role because of the following facts proved in \cite{CooQua12}:
\begin{enumerate}
\item Any coherent set of gambles is a subset of a maximal one,
\item Any coherent set of gambles is the intersection of all maximal coherent sets it is contained in.
\end{enumerate}
In addition, maximal coherent sets of gambles are characterized by the following condition, \cite{CooQua12}
\begin{eqnarray*}
\forall f \in \mathcal{L} - \{0\}:f \not\in M \Rightarrow -f \in M.
\end{eqnarray*}
Such subsets of a linear space are called semispaces, see \cite{hammer55,klee56}. Obviously, maximal coherent sets are atoms in the information algebra $\Phi$ and this algebra is atomistic closed according to the two properties of maximal sets cited above, since meet in the lattice of coherent sets is set intersection. As usual, we denote by $At(\Phi)$ the set of all atoms or maximal sets, and by $At(D)$ the set of maximal sets or atoms $M$, such that $D \leq M$. According to Corollary \ref{cor:AtomisticEmbed} in Section \ref{subsec:SetAlgAtoms}, the map $D \mapsto At(D)$ (and $\mathcal{L}(\Theta) \mapsto \emptyset$) is an embedding of the information algebra of coherent sets of gambles into the set algebra of maximal sets $At(\Phi)$.

Let's turn to lower and upper previsions. If $\underline{P}(f) = - \underline{P}(-f)$ for all $f$ in $\mathcal{L}(\Theta)$, that is, if lower and upper prevision coincide, $\underline{P}$ is called a linear prevision. Then its usual to write $\underline{P} = \bar{P} = P$. Linear previsions have an important role in the theory of imprecise probabilities, and also in information algebras. First of all a linear prevision is a lower (and upper) prevision. So, if $\mathcal{P}(\Theta)$ denote the set of linear previsions on $\mathcal{L}(\Theta)$, we have $\mathcal{P}(\Theta) \subseteq \underline{\mathcal{P}}(\Theta)$. Note that from the third coherence property of lower previsions it follows that $P(f + g) = P(f) + P(g)$. 

First, we show that linear previsions are atoms in the information algebra of lower (and upper) previsions.

\begin{lemma} \label{le:LinPrevAsAtoms}
Let $\underline{P} \in \underline{\Psi}$ and $P$ a linear prevision. Then $P \leq \underline{P}$ implies either $\underline{P} = P$ or $\underline{P}(f) = \infty$ for all $f \in \mathcal{L}(\Theta)$.
\end{lemma}

\begin{proof}
If $\underline{P}$ is coherent then $P \leq \underline{P}$ implies $\bar{P}(f) \leq P(f) \leq \underline{P}(f)$, which in turn implies $\bar{P}(f) = P(f) = \underline{P}(f)$.
\end{proof}

So linear previsions are atoms of the information algebras of lower and upper previsions. It follows that if $M$ is an atom in $\Phi$, then $\sigma(M)$ is a linear prevision, that is an atom in $\underline{\Psi}$. In fact, since $M$ is an atom of $\Phi$, either $f - \mu \in M$ or else $-f + \mu = -f - \mu'  \in M$ with $\mu = -\mu'$. It follows that $\underline{P}(f) = \sup\{\mu:f - \mu \in M\} = - \sup\{\mu':-f - \mu' \in M\} = - \underline{P}(f)$, and so $\sigma(M) = \underline{P}$ is a linear prevision.

The next thing to note is that the strictly desirable set of gambles associated with a linear prevision is given by
\begin{eqnarray*}
\tau(P) = \{f:P(f) > 0\} \cup \mathcal{L}^+(\Theta) = \{f:-P(-f) > 0\} \cup \mathcal{L}^+(\Theta).
\end{eqnarray*}

Now, if $f \not\in \mathcal{L}^+(\Theta)$, then $\underline{P}(f) < 0$ implies $P(-f) > 0$ and therefore either $f$ or $-f$ belongs to $\tau(P)$. These sets of strictly desirable gambles $\tau(P)$ associated with linear previsions are the atoms of the information algebra of strictly desirable sets of gambles. If $M$ is an atom of $\Phi$, then $M^+ = \tau(\sigma(M))$ and $\sigma(M^+) = \sigma(M) = P$ is a linear prevision. It follows that $M^+ = \tau(P)$ is an atom in $\Phi^+$. Any atom of $\Phi^+$ is of the form $M^+$ for some atom $M$ of $\Phi$. This is so, since if $M'$ is an atom of $\Phi^+$, then there is an atom $M$ of $\Phi$ such that $M' \leq M$ (since $\Phi$ is atomic), hence $\sigma(M') \leq \sigma(M) = P$, which implies $M' \leq \tau(\sigma(M)) = M^+$ und therefore $M' = M^+$.

Note in passing that subalgebras generally have different atoms, if any, than the embedding algebra. Now, consider two linear previsions $P_1 $ and $P_2 $. Then, if $P_1 \not= P_2$, we have, by general properties of atoms, $P_1 \cdot P_2 = 0$ and $P \cdot P = P$. As a consequence we have also $P_1 \leq P_2$ if and only if $P_1 = P_2$.

Since $\Phi$ is atomistic, we may conjecture that this holds also for the homomorphic algebra of lower previsions. This is confirmed by the next theorem. As with coherent sets of gambles, we denote by $At(\underline{P})$ the set of atoms, that is, linear previsions so that $\underline{P} \leq P$ and $At(\underline{\Psi})$ is the set of all linear previsions on $\Theta$.

\begin{theorem} \label{th:AtomLowPrev}
In the information algebra of lower previsions $\underline{\Psi}$ the following holds:
\begin{enumerate}
\item $\underline{\Psi}$ is atomic.
\item If $\underline{P}$ is a coherent lower prevision, then 
\begin{eqnarray*}
\underline{P} = \inf At(\underline{P}),
\end{eqnarray*}
\item if $A$ is any non-empty subset of linear previsions in $At(\underline{\Psi})$, then
\begin{eqnarray*}
\underline{P} = \inf A
\end{eqnarray*}
is a coherent lower prevision in $\underline{\Psi}$.
\end{enumerate}
\end{theorem}

\begin{proof}
If $\underline{P}$ is a coherent lower prevision, then $\tau(\underline{P})$ is a (strictly) coherent set of gambles. Since $\Phi$ is atomic, there is an atom $M$ such that $\tau(\underline{P}) \leq M$, hence $\underline{P} = \sigma(\tau(\underline{P})) \leq \sigma(M)$ and $\sigma(M)$ is atom in $\underline {P}$. So $\underline{\Psi}$ is atomistic.

We have further by the atomisticity of $\Phi$
\begin{eqnarray*}
\tau(\underline{P}) = \bigcap At(\tau(\underline{P}).
\end{eqnarray*}
By Lemma \ref{le:MaintainInf} we obtain
\begin{eqnarray*}
\underline{P} = \sigma(\bigcap At(\underline{P})) = \inf \sigma(At(\tau(\underline{P})).
\end{eqnarray*}
But $\sigma(At(\tau(\underline{P}))$ equals $At(\underline{P})$ since $\tau(\underline{P}) \subseteq M$ if and only if $\underline{P} \leq \sigma(M) = P$.

Finally let $D = \bigcap \tau(A)$ where $\tau(A) = \{\tau(P):P \in A\}$. This is a coherent set of gambles, since $\Phi$ is a complete lattice under inclusion. Thus $\sigma(D)$ is a coherent lower prevision $\underline{P}$ and (Lemma \ref{le:MaintainInf})
\begin{eqnarray*}
\sigma(D) = \underline{P} = \sigma(\bigcap \tau(A)) = \inf \sigma(\tau(A)) = \inf A
\end{eqnarray*}
and this concludes the proof.
\end{proof}

Note that these are well-known results for lower previsions \cite{walley91}. Since if $\underline{P} = \inf A$ implies that $A \subseteq At(\underline{P})$ this theorem says simply that the coherent lower prevision $\underline{P}$ is the lower envelope of $A$, and in particular of $At(\underline{P})$, that is of the linear previsions which dominate it. According to this theorem, if $A$ is any non-empty set of linear previsions on $\mathcal{L}(\Theta)$, then $\inf A$ exists and is a coherent lower prevision $\underline{P}$. Then we have $A \subseteq At(\underline{P})$ and 
\begin{eqnarray*}
\underline{P} = \inf A = \inf At(\underline{P}).
\end{eqnarray*}

As any atomistic information algebra, the algebra of lower previsions is embedded in the set algebra $At(\underline{\Psi})$ by the maps $\underline{P} \mapsto At(\underline{P})$, see Section \ref{subsec:SetAlgAtoms}. This rises the question how to characterize the images of $\underline{\Psi}$ in $At(\underline{\Psi})$. The answer is given by the weak* compactness theorem \cite{walley91}: The sets $At(\underline{P})$ for any coherent lower prevision are exactly the weak* compact convex subsets of $At(\underline{\Psi})$ in the weak* topology on $At(\underline{\Psi})$. Such sets are called credal sets. So, associated to the algebra of lower previsions $\underline{\Psi}$ there is an isomorphic information algebra of credal sets $At(\underline{P})$. There are many other sets $A$ of linear previsions with $\inf A = \underline{P}$. If $\underline{P} = \inf A$ and $A \subseteq B \subseteq At(\underline{P})$, then $\inf B = \underline{P}$. In fact, there is a minimal set $E \subseteq At(\underline{P})$ so that $\inf E = \underline{P}$ and this is the set of the extremal points of the convex set $At(\underline{P})$. This follows from the extreme point theorem \cite{walley91}. Finally, since the set algebra of $\Theta$ is embedded into the algebra $\Phi^+$, by the isomorphism to $\mathcal{P}$, it is also embedded into the latter algebra by the map of a subset $S$ of $\Theta$ to $\underline{P}(f) = \inf_{\theta \in S} f(\theta)$. 

By isomorphism, $\bar{\tau}(P)$ is an atom in $\bar{\Phi}$, the information algebra of almost desirable sets of gambles. For a linear prevision $\bar{\tau}(P) = \{f:P(f) \geq 0\} = \{f:-P(-f) \leq 0\}$. So, from $P(f) > 0$ we obtain $P(-f) < 0$. Therefore, together with the null function either $f$ or else $-f$ belong to $\bar{\tau}(P)$. This characterizes atoms in $\bar{\Phi}$. As before we conclude that $\bar{M}$ is an atom in $\bar{\Phi}$ if and only if $M$ is an atom in $\Phi$.

In conclusion, we have an embedding of $\Phi$, $\Phi^+$, $\bar{\Phi}$ and $\underline{\Psi}$ into the different set algebras of atoms $At(\Phi)$, $At(\Phi^+)$, $At(\bar{\Phi})$ and $At(\underline{\Psi})$. according to Corollary \ref{cor:AtomisticEmbed} by the maps $f : D,D^+,\bar{D},\underline{P} \mapsto At(D),At(D^+),At(\bar{D}),At(\underline{\Psi})$, where At denotes the corresponding sets of atoms.

Let's examine the embedding of $\Phi^+$ in $At(\Phi^+)$ a bit more in detail. In this case the extraction operators $\epsilon_x$ in $\Phi^+$, which are restrictions of the operator $\epsilon_x$ in $\Phi$ to $\Phi^+$, are associated with the saturation operators $\sigma_x$ corresponding to partitions $At_x$ defined by the relation $M^+_1 \equiv_x M^+_2$ iff $\epsilon_x(M^+_1) = \epsilon_x(M^+_2)$. Recall that the strictly desirables sets $\epsilon_x(M^+)$ are local atoms relative to $x$ in the information algebra $\Phi^+$. The following proposition shows how such local atoms are related to blocks $P_x$ in the set algebra of possibilities $\mathcal{S}_Q$.

\begin{proposition} \label{prop:LocAtom}
Let $B_x$ be any block of partition $P_x$ in $\Omega$, then there is an atom $M^+$ in $At(\Phi^+)$ such that
\begin{eqnarray*}
D^+_{B_x} = \epsilon_x(M^+). 
\end{eqnarray*}
\end{proposition}

\begin{proof}
First, note that $D^+_{B_x} \in \epsilon_x(\Phi^+)$. Indeed, we have $\epsilon_x(D^+_{B_x}) = D^+_{\sigma_x(B_x)} = D^+_{B_x}$ by Theorem \ref{th:EmbedSXetPsoos}. To show that $D^+_{B_x}$ is a local atom relative to $x$ in $\Phi^+$, we must prove that for every $D^+ \in \Phi^+$ such that $\epsilon_x(D^+) \geq D^+_{B_x}$ we have either $\epsilon_x(D^+) = D^+_{B_x}$ or $\epsilon_x(D^+) = \mathcal{L}(\Theta)$.

Assume on the contrary that there is $D^+ \in \Phi^+$ such that $\epsilon_x(D^+) >  D^+_{B_x}$ and $\epsilon_x(D^+) \not= \mathcal{L}(\Theta)$. There exists then a gamble $f \in \epsilon_x(D^+)$ such that $f \not\in D^+_{B_x}$, that is $\inf_{B_x} f \leq 0$. From the definition of $\epsilon_x(D^+) = posi((D^+ \cap \mathcal{L}_x) \cup \mathcal{L}^+)$ we conclude that either $f \in D^+ \cap \mathcal{L}_x$ or $f = g + h$ for some gamble $g \in D^+ \cap \mathcal{L}_x$ and $h \in \mathcal{L}^+(\Theta)$. In both cases we conclude that there is a gamble $g \in (D^+ \cap \mathcal{L}_x) \setminus D^+_{B_x}$ such that $f \geq g$. Since $g \in D^+$ there is a $\delta > 0$ such that $g - \delta \in D^+$ and since $g$ is $x$-measurable, hence constant on a block $B_x$, $g -\delta$ is so too, hence $g - \delta \in D^+ \cap \mathcal{L}_x$. From $g \leq f$ it follows that $g(\theta) - \delta < 0$ for all $\theta \in B_x$ and this implies that $-(g - \delta) \in D^+_{B_x} \subset \epsilon_x(D^+)$. But this is a contradiction since it implies $(g - \delta) - (g - \delta) =0 \in \epsilon_x(D^+)$. So, $D^+_{B_x}$ must indeed be a local atom in $x$, that is, there is an atom $M^+ \in At(\Phi^+)$ such that $D^+_{B_x} = \epsilon_x(M^+)$ and this concludes the proof.
\end{proof}

As a corollary we conclude that $\sigma(D^+_{B_x})$ is also a local atom in the algebra $\underline{\Psi}$ of lower previsions. 

As a complement we show in the next proposition, that the order between questions in $Q$ corresponds exactly to the order between partition $P_x$ induced by the equivalence relation $\equiv_x$, for $x \in Q$ in the set of possibilties.

\begin{proposition} \label{prop:PartOrder}
The identities $\epsilon_x = \epsilon_x\epsilon_y = \epsilon_y\epsilon_x$ hold if and only if $P_x \leq P_y$.
\end{proposition}

\begin{proof}
Assume first that $P_x \leq P_y$. Then $\mathcal{L}_x \subseteq \mathcal{L}_y$ and $\epsilon_x(D) = C(D \cap \mathcal{L}_x) \subseteq C(D \cap \mathcal{L}_y) = \epsilon_y(D)$ for any coherent set of gambles $D$, and so $\epsilon_x(D) \leq \epsilon_y(D)$ in information order. Recall that extraction operators $\epsilon_x$ preserve order, see Proposition \ref{prop:MonOfExtr}, $\epsilon_x(D)$ has support $D$ and $\epsilon_x(D),\epsilon_y(D) \leq D$. Therefore $\epsilon_x(D) = \epsilon_x(\epsilon_x(D)) \leq \epsilon_x(\epsilon_y(D)) \leq \epsilon_x(D)$ and so $\epsilon_x = \epsilon_x\epsilon_y$. Further, $\epsilon_x(D) \geq \epsilon_y(\epsilon_x(D)) \geq \epsilon_x(\epsilon_x(D)) = \epsilon_x(D)$, hence we conclude that $\epsilon_x(D) = \epsilon_y(\epsilon_x(D))$, hence $\epsilon_x = \epsilon_y\epsilon_x$.

Conversely, assume $\epsilon_x = \epsilon_x\epsilon_y = \epsilon_y\epsilon_x$ in $\Phi$. This identity holds also for the restrictions of $\epsilon_,,\epsilon_y$ to the image of the set algebra $\mathcal{S}_Q$ in $\Phi^+$ by the embedding. But then $\epsilon_x$ and $\epsilon_y$ correspond one-to.one to the saturation operators $\sigma_x$ and $\sigma_y$ of partitions $P_x$ and $P_y$. By inverting of the embedding, from $\epsilon_x = \epsilon_x\epsilon_y = \epsilon_y\epsilon_x$ we obtain therefore $\sigma_x = \sigma_x\sigma_y = \sigma_y\sigma_x$. But this implies $P_x \leq P_y$.  
\end{proof}

This shows that the order $x \leq y$ induced by $\Phi$ in $Q$ corresponds precisely to our information order between partitions of possibilities. The same holds also relative to the algebra $\underline{\Psi}$ of lower previsions. 

Furthermore, if $(P_Q,\leq)$ is a join-semilattice, then so is the order in $Q$ induced by $\epsilon_x = \epsilon_x\epsilon_y = \epsilon_y\epsilon_x$ and vice versa. This discussion can be extended also to the order between partitions $At_x$ of $At(\Phi)$ (or $At(\Phi^+)$, $At(\bar{\Phi})$, defined by $\epsilon_x(M) = \epsilon_x(M')$ between atoms $M$ and $M'$ of $\Phi$ (or $\Phi^+$, $\bar{\Phi}$)). We renounce to develop this subject here.


\section{Finite gambles}

In this section we show that all the information algebras related to imprecise probabilities are \textit{compact}, see Section \ref {subsec:CompInfAlg}. We start with an domain-free information algebra of coherent sets $(\Phi,\cdot,0,1;E)$ on a set of possibilities $\Theta$ and where $E = \{\epsilon_x:x \in Q\}$. We show first that \textit{finitely generated} generated coherent sets $\mathcal{C}(F)$, where $F$ is a finite subset of $\mathcal{L}(\Theta)$ are the finite elements in the algebra $\Phi$\footnote{Our notion of finitely generated coherent sets is not exactly the same as the one of finitely generated models of \cite{walley91}}. This fact is based on the constatation that the consequence operator $\mathcal{C}$ is \textit{algebraic}, \cite{daveypriestley97}. This means that for any subset $D$ of $\mathcal{L}(\Theta)$ we have that
\begin{eqnarray*}
\mathcal{C}(D) = \bigcup \{\mathcal{C}(F):F \subseteq D,\textrm{ finite}\}.
\end{eqnarray*}
if $D$ is coherent.

\begin{proposition}
The consequence operator $\mathcal{C}$ related to desirable sets of gambles on a set of possibilities $\Theta$ is algebraic.
\end{proposition}

\begin{proof}
Obviously we have $\mathcal{C}(D) \supseteq \bigcup \{\mathcal{C}(F):F \subseteq D,\textrm{ finite}\}$. Assume first that $D$ is coherent. Any gamble $f$ in $\mathcal{C}(F)$ is then either in $\mathcal{L}^+(\Theta)$ or greater than a finite linear combination $f \geq \lambda_1f_1 + \ldots + \lambda_nf_n$, $f_i \in F$, $\lambda_i \geq 0$ and not all $\lambda_i = 0$. In both cases $f$ belongs to $\mathcal{C}(\{f_1,\ldots,f_n\}$ so that indeed $\mathcal{C}(D) = \bigcup \{\mathcal{C}(F):F \subseteq D,\textrm{ finite}\}$. If $\mathcal{C}(D) = \mathcal{L}(\Theta)$, then there must be a combination  $\lambda_1f_1 + \ldots + \lambda_nf_n= 0$, of elements of $D$. But then $\mathcal{C}(\{f_1,\ldots,f_n\} = \mathcal{L}(\Theta)$ and again $\mathcal{C}(D) = \bigcup \{\mathcal{C}(F):F \subseteq D,\textrm{ finite}\}$. This concludes the proof.
\end{proof}

Now, in \cite{kohlas03} it has been shown that the information algebra induced by an algebraic consequence operator $\mathcal{C}$ is compact with $\mathcal{C}(F)$, $F$ finite, as finite elements. Although in $\cite{kohlas03}$ only the multivariate case is considered this result carries over to the present more general case and in particular to the information algebra of coherent sets of gambles.

\begin{theorem}
The information algebra $(\Phi,\cdot,0,1;E)$ of coherent sets of gambles is compact with finite elements $\Phi_f = \{\mathcal{C}(F):F \subseteq \mathcal{L}, \textrm{ finite}\}$.
\end{theorem}

\begin{proof}
We verify the defining conditions of a compact information algebra according to Section \ref{subsec:CompInfAlg}. Obviously the combination of two finitely generated coherent sets $\mathcal{C}(F_1) \cdot \mathcal{C}(F_2) = \mathcal{C}(F_1 \cup F_2)$ is still finitely generated. Note that the unit and null element $\mathcal{C}(\emptyset)$ and $\mathcal{L}(\Theta)$ are finitely generated too. So the \textit{Combination} property holds.

Let next $X$ be a directed set of finitely generated coherent sets of gambles in $\Phi_f$. We claim that the supremum $\bigsqcup X$ of this directed set equals $\bigcup X$. To prove this, we must show that $\mathcal{C}(\bigcup X) \subseteq \bigcup X$, because this implies that $\bigcup X$ is closed. So consider a gamble $f$ in $\mathcal{C}(\bigcup X)$. Since the the consequence operator $\mathcal{C}$ is algebraic, there is a finite set $F \subseteq \bigcup X$ such that $f \in \mathcal{C}(F)$. Note then  that every element of $F$ is in some of the closed sets of $X$. Since $X$ is directed, there must be a set $E \in X$ such that $F \subseteq E$. But then we conclude that $f \in \mathcal{C}(F) \subseteq \mathcal{C}(E) = E \subseteq \bigcup D$. This proves the inclusion $\mathcal{C}(\bigcup X) \subseteq \bigcup X$ and therefore $\bigsqcup X = \bigcup X$. This is the \textit{Convergence} property.

Consider an extraction $\epsilon_x(D)$ of a coherent set of gambles. Since $\mathcal{C}$ is algebric, we have
\begin{eqnarray*}
\epsilon_x(D) = \mathcal{C}(D \cap \mathcal{L}_x) = \bigsqcup \{\mathcal{C}(F):F \subseteq D \cap \mathcal{L}_x, F \textrm{ finite}\}.
\end{eqnarray*}
We claim that if $F$ is a set $x$-measurable gambles, then $\epsilon_x(\mathcal{C}(F)) = \mathcal{C}(\mathcal{C}(F) \cap \mathcal{L}_x) = \mathcal{C}(F)$. Indeed, $F \subseteq \mathcal{C}(F) \cap \mathcal{L}_x$, hence $\mathcal{C}(F) \subseteq \mathcal{C}(\mathcal{C(F)} \cap \mathcal{L}_x)$. On the other hand, $\mathcal{C}(F) \cap \mathcal{L}_x \subseteq \mathcal{C}(F)$, so $\mathcal{C}(\mathcal{C}(F) \cap \mathcal{L}_x) \subseteq \mathcal{C}(F)$ which establishes the identity. Using this result we obtain
\begin{eqnarray*}
\epsilon_x(D) \leq \bigsqcup \{\mathcal{C}(F):\mathcal{C}(F) \subseteq D,\epsilon_x(\mathcal{C}(F)) =\mathcal{C}(F)\} \leq \epsilon_x(\mathcal{C}(D)) =\epsilon_x(D).
\end{eqnarray*}
This shows that \textit{Local Density} holds.

Finally, if $X$ is a directed subset of coherent ses $\mathcal{C}(F)$, $F$ finite, and $D = \mathcal{C}(E) \subseteq \bigsqcup X$, $E$ finite, then $E \leq \mathcal{C}(E)$ and for all $f_i \in E$ there must be a finite set $F_i$ such that $\mathcal{C}(F_i) \in X$, hence, since $X$ is directed, there is a finite set $F$ such that $\mathcal{C}(F) \in X$ and $\mathcal{C}(F_i) \subseteq X$. But then $\mathcal{C}(E) \subseteq \cup_i\ \mathcal{C}(F_i) \subseteq \mathcal{C}(F)$ and this is compactness.
\end{proof}

Note that this is a general result concerning algebraic consequence operators and information algebras derived from them and not limited to the present case of an algebra of coherent sets of gambles \cite{kohlas03}. By standard methods from order theory \cite{daveypriestley97} we may derive some additional results. First, $\Phi$ is closed under the union of any directed sets of elements of $\Phi$. By the Theorem above, it is closed under union of directed sets of $\Phi_f$. Let $X$ be any directed set of coherent sets of gambles. If $F$ is a finite set and $F \subseteq \bigcup X$, then as in the proof above we infer that $F \subseteq D$ for some $D \in X$. It follows that
\begin{eqnarray*}
\lefteqn{\mathcal{C}(\bigcup X) = \bigcup \{\mathcal{C}(F):F \subseteq \bigcup X,F \textrm{ finite}\} } \\
&&= \bigcup \{\mathcal{C}(F):F \subseteq D \textrm{ for some}\ D \in X,F \textrm{ finite}\} \\
&&\subseteq \bigcup_{D \in X} \mathcal{C}(D).
\end{eqnarray*}
The reverse inclusion is always valid, so $\mathcal{C}(\bigcup X) = \bigcup X$. This implies that $\Phi$ is an \textit{algebraic} $\cap$-system, see Section \ref{subsec:Gambles} and \cite{daveypriestley97}. We recall also that the finitely generated coherent sets of gambles, the finite elements of $\Phi$, determine the algebra $\Phi$ of coherent sets fully, since the algebra is isomorphic to the algebra of ideals of finite sets, by general results about compact information algebras, see Theorem \ref{th:IdCompFiniteEl}.

It may be expected that finite elements $\mathcal{C}(F)$ in the information algebra of coherent sets of gambles map to finite elements $\sigma(\mathcal{C}(F))$ in the algebra of lower previsions. Further $\tau(\sigma(\mathcal{C}(F)))$ and $\bar{\tau}(\sigma(\mathcal{C}(F)))$ may be expected to be finite elements in the algebras of strictly and almost desirable sets of gambles. These ideas will be examined in the next section, using credal sets.


\section{Credal sets}

In this section we look at credal sets more closely. In Section \ref{subsec:ImpProbRepres} we referred to sets of atoms $At(\underline{P})$ of a coherent lower prevision on $\mathcal{L}(\Theta)$ as \textit{credal sets}. We recall that these are closed convex sets of linear previsions on $\mathcal{L}(\Theta)$ and $At(\underline{P}) = \{P \in At(\underline{P}):\underline{P} \leq P\}$. And the information algebra of these sets is isomorphic to the algebra of $\underline{\Psi}$ of lower previsions such that
\begin{enumerate}
\item $At(\underline{P}_1 \cdot \underline{P}_2) = At(\underline{P}_1) \cap At(\underline{P}_2)$,
\item $At(\epsilon_x(\underline{P})) = \sigma_x(At(\underline{P}))$.
\end{enumerate}
Here $\sigma_x$ denotes the saturation operator relative to the partition induced by the equivalence relation $P_1 \equiv_x P_2$ on $At(\underline{\Psi})$ if and only if $\epsilon_x(P_1) = \epsilon_x(P_2)$ for two atoms $P_1$ $P_2$. 

There are several characterizations of linear previsions, see \cite{walley81}. For our purpose the following one is most important.

\begin{theorem} \label{th:CharLinPrev}
The functional $P : \mathcal{L} \rightarrow \mathbb{R}$ is a linear prevision if and only if it satisfies
\begin{enumerate}
\item \textit{Linearity:} $P(f_1 + f_2) = P(f_1) + P(f_2)$,
\item \textit{Homogeneity:} $P(\lambda f) = \lambda P(f)$,
\item \textit{Positivity:} $f \geq 0$ implies $P(f) \geq 0$,
\item \textit{Unit norm:} $P(1) = 1$.
\end{enumerate}
\end{theorem}

\begin{proof}
Let $P$ be a linear prevision. Linearity and Homogeenity follow from the the properties $\underline{P}(f_1 + f_2) \geq \underline{P}(f_1) + \underline{P}(f_2)$ and $\underline{P}(\lambda f) = \lambda \underline{P}(f)$ of lower previsions (see Section \ref{subsec:LowPrev}) and the definition of a linear prevision $\underline{P}(f) = - \underline{P}(-f) = P(f)$ (Section \ref{subsec:ImpProbRepres}). Positivity follows from the property $P(f) \geq \underline{P}(f) \geq \inf_{\theta \in \Theta} f(\theta)$. Finally Unit norm is a consequence of 
\begin{eqnarray*}
P(1) = \underline{P}(1) =  \sup \{\mu:1 - \mu \in D\} = 1
\end{eqnarray*}
if $D$ is a coherent set of gambles.

Conversely, suppose that $P$ is a functional satisfying the properties of the theorem. Then we claim that $P(\mu) = \mu$ for any $\mu \in \mathbb{R}$. Assume first that $\mu > 0$. Then by Homogeneity and Unit norm we have $(1/\mu) P(\mu) = P((1/\mu) \cdot \mu) = P(1) =1$. If $\mu \leq 0$, then we have by Linearity $P(-\mu) + P(\mu) = P(0) = 0$, hence $P(\mu) = P(-\mu) = \mu$. Further, for $f \in \mathcal{L}$ let $\mu = \inf_{\theta \in \Theta} f(\theta)$. Consider then $P(f - \mu)$ where $f - \mu \geq 0$. Then it follows by positivity $P(f) = P(f - \mu) + P(\mu) \geq P(\mu) = \mu$, hence $P(f) \geq \inf_{\theta \in \Theta} f(\theta)$. If we add Linearity and Homogeneity, then $P$ satisfies all defining properties of lower prevision and is thus a lower prevision and we have also $P(f) + P(-f) = P(0) = 0$, hence $P(f) = - P(-f)$. This shows that $P$ is a linear prevision.
\end{proof}

Now a coherent lower prevision $\underline{P}$ is induced by some coherent set of gambles $D \in \mathcal{C}(\Theta)$, $\underline{P} = \sigma(D)$. If the linear prevision $P$ belongs to the credal set $At(\underline{P})$, that is $\underline{P}(f) \leq P(f)$ for all $f \in \mathcal{L}(\Theta)$, then in particular $P(f) \geq \underline{P}(f) \geq 0$ for all $f \in D$. Define the set
\begin{eqnarray*}
\mathcal{P}_D = \{P \in At(\underline{\Psi}):P(f) \geq 0 \textrm{ for all}\ f \in D\}.
\end{eqnarray*}
Then $At(\underline{P})$ equals the closed convex set $\mathcal{P}_D$, see Section \ref{subsec:ImpProbRepres} and \cite{walley91}.
If $D = \mathcal{C}(X)$ is a coherent set of gambles, then we have also that $\mathcal{P}_{\mathcal{C}(X)} = \mathcal{P}_X = \{P \in At(\underline{P}):P(f) \geq 0 \textrm{ for all}\ f \in X\}$. This follows since the gambles $f$ in $\mathcal{C}(X)$ dominate finite positive linear combinations of gambles from $X$ and $P(f)$ is a linear functional. 

Recall that the information algebra $\Phi$ is (weakly) homomorphic to the algebra $\underline{\Psi}$, which in turn is isomorphic to the algebra of closed convex sets in $\mathcal{L}(\Theta)$, that is a subset algebra of $At(\underline{\Psi})$, see Sections \ref{subsec:LowPrev} and \ref{subsec:SetAlgAtoms}. This implies that the map $D \mapsto At(\sigma{D}$ is also a (weak) homomorphism, so that
\begin{enumerate}
\item $At(\sigma(D_1 \cdot D_2)) = At(\sigma(D_1)) \cap At(\sigma(D_2)$ if $D_1 \cdot D_2 \not= 0$,
\item $At(\sigma(\epsilon_x(D))) = \sigma_x(At(\sigma(D))$.
\end{enumerate}

After these preparations, we are going to look for finite elements in the subset algebra of closed convex sets in $At(\underline{\Psi})$, that is the image $Im(\underline{\Psi})$ of $\underline{\Psi}$ under the map $\underline{P} \mapsto At(\underline{P})$. In view of the the homomorphism $D \mapsto At(\sigma{D})$ and the fact that $At(\sigma(D))$ is generated by $\mathcal{P}_D = \mathcal{P}_X$, if $D = \mathcal{C}(X)$, it seems plausible to define finite elements in  $At(\underline{\Psi})$ as those which are the closed convex hull of sets $\mathcal{P}_F$, where $F$ is a finite set of gambles. So, define
\begin{eqnarray*}
At_f(\underline{\Psi}) = \{\mathcal{P}_F \in At(\underline{\Psi}):F \subseteq \mathcal{L}, \textrm{ finite set}\}.
\end{eqnarray*}
Then we have the following theorem.

\begin{theorem} \label{th:FiniteAtoms}
The subset $At_f(\underline{\Psi})$ of $At(\underline{\Psi})$ is the set of finite elements of the information algebra $At(\underline{\Psi})$ and this algebra is compact.
\end{theorem}

\begin{proof}
We verify thar $At_f(\underline{\Psi})$ satisfies the four defining properties Combination, Convergence, Local Density and Compactness of finite elements, see Section \ref{subsec:CompInfAlg}.

a) Combination. Consider two elements $\mathcal{P}_{F_1}$ and $\mathcal{P}_{F_2}$ in $At_f(\underline{\Psi})$. Assume first that $F_1$ and $F_2$ are contradictory, that is $C(F_1) \cdot C(F_2) = C(F_1 \cup F_2) = \mathcal{L}$. Then we have
\begin{eqnarray*}
\mathcal{P}_{F_1 \cup F_2} = \mathcal{P}_{C(F_1 \cup F_2)} = \emptyset, \quad \mathcal{P}_{F_1 \cup F_2} = \mathcal{P}_{F_1} \cap \mathcal{P}_{F_2} = \emptyset. 
\end{eqnarray*}
If, on the other hand $C(F_1 \cup F_2) = C(F_1) \cdot C(F_2)$ is coherent, then by weak homomorphism
\begin{eqnarray*}
\mathcal{P}_{F_1 \cup F_2} = \mathcal{P}_{C(F_1) \cup F_2)}  = \mathcal{P}_{C(F_1)} \cap \mathcal{P}_{C(F_2)} 
= \mathcal{P}_{F_1} \cap \mathcal{P}_{F_2}
\end{eqnarray*}
Since the empty set is the null element of the set algebra $At(\underline{\Psi})$ which belongs to $At_f(\underline{\Psi})$ and the set $F_1 \cup F_2$ is finite, hence $\mathcal{P}_{F_1 \cup F_2} \in At_f(\underline{\Psi})$, this proves the \textit{Combination} property.

b) Convergence. Let $\mathcal{D}$ be a directed set of elements in $At_f(\underline{\Psi})$. Define the set
\begin{eqnarray*}
G = \bigcup \{F:\mathcal{P}_F \in \mathcal{D}\}.
\end{eqnarray*}
Consider now
\begin{eqnarray*}
\mathcal{P}_G = \{P \in \mathcal{P}:P(f) \geq 0 \textrm{ for all} f \in G\}.
\end{eqnarray*}
If $P \in \mathcal{P}_G$, then we have in particular $P(f) \geq 0$ for all $f \in F$ for any $F$ such that $\mathcal{P}_F \in \mathcal{D}$. This means that $\mathcal{P}_G$ is an upper bound of $\mathcal{D}$. Recall that in a set algebra, information order is the inverse of inclusion so that $\mathcal{P}_G \subseteq \mathcal{P}_F$ for all $\mathcal{P}_F \in \mathcal{D}$. Now, if $\mathcal{P}_G = \emptyset$, then there must be contradictory elements in $\mathcal{D}$ and therefore $\bigsqcup \mathcal{D} = \bigcap \mathcal{D} = \emptyset$. Otherwise consider any upper bound $\mathcal{P}_D \not= \emptyset$ of $\mathcal{D}$, where $D$ is a coherent set of gambles. Then $P \in \mathcal{P}_D$ implies $P(f) \geq 0$ for all $f \in F$ for any $F$ such that $\mathcal{P}_F \in \mathcal{D}$. But this implies $P \in \mathcal{P}_G$, hence $\mathcal{P}_D \geq \mathcal{P}_G$. This shows that $\mathcal{P}_G$ is the supremum of $\mathcal{D}$, and this proves that the \textit{Convergence} property holds.

c) Local density. Consider any credal set $\mathcal{P}_D$, where $D$ is a coherent set of gambles, and such that $\sigma_x(\mathcal{P}_D) = \mathcal{P}_D$. In addition consider the set
\begin{eqnarray*}
A = \{\mathcal{P}_F \in At_f\underline{\Psi}):\sigma_x(\mathcal{P}_F) = \mathcal{P}_F \leq \mathcal{P}_D\}.
\end{eqnarray*}
Clearly we have $\mathcal{P}_D \geq \bigsqcup A$. Recall that the map $\Phi^+ \rightarrow At(\underline{\Psi})$ defined by $D^+ \mapsto \mathcal{P}_{D^+}$ is an isomorphism and so is also the inverse map $\mathcal{P}_{D^+} = \mathcal{P}_D \mapsto D^{+}$. Consider now the image $D^+$ of $\mathcal{P}_D$ under this map. Since $D^+ \in \Phi$ and $\Phi$ is a compact information algebra, we have by local density in $\Phi$
\begin{eqnarray*}
D^+ = \bigsqcup B' \textrm{ with}\ B' = \{C(F) \in \Phi_f:\epsilon_x(C(F)) = C(F) \leq D^+\}.
\end{eqnarray*}
Note that by isomorphism from $\mathcal{P}_{D^+} = \mathcal{P}_D = \sigma_x(\mathcal{P}_D) = \sigma_x(\mathcal{P}_{D^+})$ it follows that $D^+ = \epsilon_x(D^+)$. Let's map the set $B'$ to $At(\underline{\Psi})$, which gives the set
\begin{eqnarray*}
B = \{\mathcal{P}_F \in At_f(\underline{\Psi}): \sigma_x(\mathcal{P}_F) = \mathcal{P}_F \leq \mathcal{P}_D\}
\end{eqnarray*}
since this map is a weak homomorphism. We claim that that $\bigsqcup B'$ maps to $\bigsqcup B$. Obviously, $\mathcal{P}_D = \mathcal{P}_{D^+}$ is an upper bound of $B$. Consider any upper bound $\mathcal{P}_{D'}$ of  $\mathcal{P}_D$. We have again $\mathcal{P}_{D'} = \mathcal{P}_{D'^+}$. Then by isomorphism, $D' \geq D'^+ \geq D^+$ and therefore $\mathcal{P}_{D'} \geq \mathcal{P}_D^+$. This shows that $\mathcal{P}_D = \mathcal{P}_{D^+}$ is indeed the supremum of $B$. Now $B \subseteq A$, hence $\mathcal{P}_D \leq \bigsqcup A$. Since the inverse inequality is valid too, we have finally $\mathcal{P}_D = \bigsqcup A$ and this in \textit{Local Density} for $At_f(\underline{\Psi})$ in $At(\underline{\Psi})$.

d) Compactness. Consider an element $\mathcal{P}_F$ of $At_f(\underline{\Psi})$ such that $\mathcal{P}_F \leq \bigsqcup \mathcal{D}$, where $\mathcal{D}$ is a directed set in $At_f(\underline{\Psi})$. By Convergence $\bigsqcup \mathcal{D}$ exists and $\mathcal{P}_G = \bigsqcup \mathcal{D}$, where $G$ is as in a) above the union of all finite sets $F'$ such that $\mathcal{P}_F' \in \mathcal{D}$. Now, since $F = \{f_1,\ldots,f_m\}$ for some integer $m$, there is a subset $F'_i$ which contains $f_i$ for $i = 1,\ldots,m$ and since $\mathcal{D}$ is directed, there is a set $F' \in \mathcal{D}$ which contains all $F_i$. But then $\mathcal{P}_F \leq \mathcal{P}_{F'}$. This shows that the \textit{Compactness} property holds for $At_f(\underline{\Psi})$.

In summary a) to d) show that $At_f(\underline{\Psi})$ represents indeed the finite elements in $At(\underline{\Psi})$ and this information algebra is therefore compact.
\end{proof}

This theorem allows us to determine the finite elements of the isomorphic algebras $\underline{\Psi}$ of lower previsions, of strictly desirable gambles $\Phi^+$ and of almost desirable gambles $\bar{\Phi}$. The  finite elements in these information algebras are simply the image of $At_f(\underline{\Psi})$ by the corresponding isomorphisms, see proposition \ref{prop:IsomCompactAlg}.

Consider first lower previsions. The inverse map to the isomorphism $\underline{P} \mapsto At(\underline{P})$ is given by $\underline{P} = \inf At(\underline{P})$. In particular, the finite elements $\underline{\Psi}_f$ of the compact  information algebra of lower previsions are determined by
\begin{eqnarray*}
\underline{P}(f) = \inf\{P(f):P(f) \geq 0 \textrm{ for all}\ f \in F\}
\end{eqnarray*}
where $F$ is a finite set of gambles. It is well-known that in the case of a finite set of possibilities $\Theta = \{\theta_1,\ldots,\theta_n\}$ this reduces to a problem of linear programming, \cite{walley91}. Consider $F = \{f_1,\ldots,f_m\}$ a finite set of gambles on $\Theta$. In this case $\mathcal{L}$, the linear space of gambles is simply a vector space $\mathbb{R}^n$, where a gamble $f$ is represented by the $n$-vector $(f(\theta_1),\ldots,f(\theta_n))$. The dual space of linear functionals containing linear previsions is equally a vector space $\mathbb{R}^n$ and a linear prevision $P$ is given by the vector $(p(\theta_1),\ldots,p(\theta_n)$ and $P(f)$ is simply the scalar product 
\begin{eqnarray*}
P(f) = \sum_{j = 1}^n p(\theta_j)f(\theta_j).
\end{eqnarray*}
So, if $f_i(\theta_j) = f_{i,j}$ and $p(\theta_j) = p_j$, then $P(f) \geq 0$ gives the following system of linear inequalities
\begin{eqnarray*}
\sum_{j=1}^n f_{i,j}p_j, \quad i = 1,\ldots,m.
\end{eqnarray*}
In addition we have the Unit norm $P(1) = 1$ and Positivity $P(f) \geq 0$ if $f \geq 0$. The former condition is
\begin{eqnarray*}
\sum_{j=1}^n p_j = 1,
\end{eqnarray*}
and the later condition translates for the gambles $g_{i,j} = \delta_{i,j}$ into
\begin{eqnarray*}
p_j \geq 0 \textrm{ for all}\ j = 1,\ldots,n.
\end{eqnarray*}
These two last conditions define a simplex in $\mathbb{R}^n$ and the whole system of linear inequalities a convex polyhedron, contained in the simplex. So these polyhedron represent the finite elements in the algebra of credal sets relative to finite sets of possibilities $\Theta$. A lower prevision $\underline{P}$ defines a credal sets $At(\underline{P})$ which is a closed convex subset of the simplex. The finite credal sets approximating this convex set are the polyhedron in the simplex containing $At(\underline{P})$. Finally $P(f)$ can be obtained as 
\begin{eqnarray*}
min \sum_{j=1}^n f_jp_j
\end{eqnarray*}
under the system of linear inequalities defined above. This is a classical problem of linear programming.

Let us now consider the compact information algebras of strictly desirable gambles $\Phi^+$ and of almost desirable gambles $\bar{\Phi}$. What are the finite elements in these two algebras? Note that the inverse map of the isomorphism $D^+ \mapsto At(\sigma(D^+))$ is defined by $\tau(\sigma(D^+))$ since the credal set $At(\sigma(D^+))$ maps to the associated lower prevision $\sigma(D^+)$. In the case of strictly desirable gambles, the inverse of the map $D^+ \mapsto At(\sigma(D^+))$ is also given by
\begin{eqnarray*}
\tau(\sigma(D^+)) = \{f:P(f) > 0 \textrm{ for all}\ P \in At(\sigma(D^+))\} \cup \mathcal{L}^+
\end{eqnarray*}
if $P$ is a linear prevision, as noted in Section \ref{sec:DiffAlgs}.

Consider now a finite element $\mathcal{P}_F$ where $F$ is a set of gambles in the algebra of credal sets. This set maps then to
\begin{eqnarray*}
\{f:P(f) > 0 \textrm{ for all}\ P \in \mathcal{P}_F\} \cup \mathcal{L}^+
\end{eqnarray*}
These sets are finite elements in the information algebra $\Phi^+$. Assume that $\mathcal{P}_F = \mathcal{P}_D$ for a coherent set of gambles. Then we have $At(\sigma(D)) = At(\sigma(D^+))$, where $D^+ = \tau(\sigma(D))$ is the set of strictly desirable gambles associate with $D$. So we have that the finite elements $D^+ = \{f:P(f) > 0 \textrm{ for all}\ P \in \mathcal{P}_F\} \cup \mathcal{L}^+$ in $\Phi_f^+$ are the strictly desirable  gambles associated with finite elements $D \in \Phi_f$. .

A similar result holds for almost desirable gambles. For any credal set $\mathcal{P}$, the set 
\begin{eqnarray*}
\{f:P(f) \geq 0 \textrm{ for all}\ P \in \mathcal{P}_F\}
\end{eqnarray*}
is an almost desirable set and belongs to the finite elements $\bar{\Phi}_f$ of the compact algebra of almost desirable gambles as the image of a finite credal set.  As before, if  $\mathcal{P}_F = \mathcal{P}_D$ for a coherent set of gambles, then $\bar{D} = \{f:P(f) \geq 0 \textrm{ for all}\ P \in \mathcal{P}_F\}$ is the almost desirable set $\bar{\tau}(\sigma(D))$ corresponding to the coherent set $D$. So the almost desirable sets $\bar{D} \in \bar{\Phi}_f$ are the almost desirable sets corresponding to finite sets $D \in \Phi_f$. This completes the picture or the different information algebras related to imprecise probability.

%% file: chapter11.tex

\chapter{Non-idempotent information algebras} \label{sec:nonidempotent}


\section{Valuation algebras}

In many cases it does make sense to drop the idempotency requirement $\epsilon_x(\phi) \cdot \phi = \phi$ and thus in particular also $\phi \cdot \phi = \phi$. Whereas it seems generally reasonable to assume that repeating the same piece of information gives nothing new, one may take a more liberal view on information: Assume that the information transmitted by a sensor is signal of alarm $\phi$, then assuming that such a sensor may also fail and give false alarms, obtaining the same alarm signal $\phi$ from a second, independent sensor, then combining the two signals, $\phi \cdot \phi$ may be different from $\phi$, in fact, more informative than a single signal of alarm. 

As before let $\Phi$ denote a set of elements, which can be (in some sense) considered as pieces of information and $Q$ set of elements representing different questions. Again as before we assume two operations in $\Phi$, combination and extraction,
\begin{enumerate}
\item \textit{Combination:} $\cdot : \Phi \times \Phi \rightarrow \Phi$, $(\phi,\psi) \mapsto \phi \cdot \psi$,
\item \textit{Extraction:} $\epsilon : \Phi \times Q \rightarrow \Phi$, $(\phi,x) \mapsto \epsilon_x(\phi)$.
\end{enumerate}
On these elements, we impose the following requirements,
\begin{enumerate}
\item \textit{Semigroup:} $(\Phi,\cdot)$ is a commutative semigroup with a unit $1$,
\item \textit{Extraction:} for all $\phi,\psi \in \Phi$ and $x \in Q$, we have
\begin{eqnarray*}
\epsilon_x(\epsilon_x(\phi) \cdot \psi) = \epsilon_x(\phi) \cdot \epsilon_x(\psi),
\end{eqnarray*}
\item \textit{Unit:} for all $x \in Q$, $\epsilon_x(1) = 1$,
\item \textit{Support:} For all $\phi \in \Phi$ there is a $x \in Q$ such that $\epsilon_x(\phi) = \phi$.
\end{enumerate}
This is a reduct of the axioms imposed on information algebras in the previous part, see Section \ref{sec:Basics}. A system $(\Phi,\cdot,1;E)$ where $E = \{\epsilon_x : x \in Q\}$ is called a  (domain-free) \textit{valuation algebra}, since it depends on conditions similar to those studied in \cite{kohlas03} and in particular to the axiomatic system proposed by \cite{shenoyshafer90}. Note that we do not necessarily require a null in the semigroup $(\Phi,\cdot)$. The existence of a unit together with the extraction axiom implies that any $x \in Q$ is a support of $\epsilon_x$ and also that the combination of two elements with support $x$ still have support $x$ (Lemma \ref{le:Supp1} in Section \ref{subsec:Basics})

Since idempotency is no more required, we can not define an information order as in Section \ref{sec:InfOrder}. Neverheless we shall see that an order can be defined even in this case, see Section \ref{subsec:NOnIdOrder}. On the other hand, order between questions and conditional independence between questions can still be defined as before, Section \ref{sec:StructOfQuest}, that is, $x \leq y$ if and only if $\epsilon_x = \epsilon_x\epsilon_y = \epsilon_y\epsilon_x$. Correspondingly, Lemma \ref{le:Supp2} is still valid. Similarly, the conditional independence relation $x \bot y \vert z$ between questions can still be defined by $\epsilon_{y \vee z}\epsilon_{x \vee z} = \epsilon_z$ and $\epsilon_{x \vee z}\epsilon_{y \vee z} = \epsilon_z$ assuming that the order in $Q$ defines a join-semilattice. This relation is a q-separoid and Theorem \ref{th:CombExtrProp} is still valid. All these items do not depend on idempotency.  Also a valuation algebra is called \textit{commutative} if $\epsilon_x\epsilon_y = \epsilon_y\epsilon_x$ for all pairs of questions $x$ and $y$. 

If the order $(Q,\leq)$ defined a join-semilattice, we may also derive the labeled version of a valuation algebra, as in Section \ref{subsec:LabInfAlg}. That is, we consider pairs $(\phi,x)$ with $\phi \in \Phi$ and $x \in Q$ such that $x$ is a support of $\phi$. Combination and transport are defined as in Section \ref{subsec:LabInfAlg}. This leads to exactly the same axioms as in Section \ref{subsec:LabInfAlg}, but without the idempotency axiom and not necessarily with a null element. So we have \textit{labeled} valuation algebras $(\Psi,\cdot,T)$ with $T = \{t_x:x\in Q\}$, the family of transport operators, satisfying the following axioms:
\begin{enumerate}
\item \textit{Semigroup:} $(\Psi,\cdot)$ is a commutative semigroup.
\item \textit{Transport:} 
\begin{enumerate}
\item For all pairs $x,y \in Q$ exists a $z = x \vee y \in Q$ such that $t_x = t_xt_z$ and $t_y = t_yt_z$,
\item for all $u \in Q$, $t_x = t_xt_u$ and $t_y = t_yt_u$ imply $t_z = t_zt_u$,
\item for all pairs $x,y \in Q$, $t_x = t_xt_y$ and $t_y = t_yt_x$ jointly imply $x = y$.
\end{enumerate}
\item \textit{Labeling:} $d(\phi \cdot \psi) = d(\phi) \vee d(\psi)$, $d(t_x(\psi)) = x$.
\item \textit{Unit:} For all $x \in Q$ the semigroups $(\Psi_x,\cdot)$ have a unit element $1_x$ and $t_y(1_x) = 1_y$ for all $x,y \in Q$.
\item \textit{Combination:} For all $\phi,\psi \in \Psi$ and $x \in Q$, if $d(\phi) = x$, then $t_x(\phi \cdot \psi) = \phi \cdot t_x(\psi)$.
\item \textit{Identity:} For all $x \in Q$ if $d(\psi) = x$, then $t_x(\psi) = \psi$.
\end{enumerate}

Note that we need to add condition $t_y(1_x) = 1_y$ as an axiom, in the idempotent case, this follows from idempotency. Then Lemma \ref{le:ElPropLabInfAlg} is still valid and we have also $1_{x \vee y} = 1_x \vee 1_y$ (Lemma \ref{le:UnitNull}). Sometimes a null element is present, then it will satisfy the same conditions as in the idempotent case, see also \cite{kohlas03}. In addition, conditional independence among questions in $Q$ can be defined based on transport operations as in the idempotent case, Section \ref{subsec:LabInfAlg}, and Theorem \ref{the:CombExtrProplab} holds again. Finally, we may characterize non-idempoten labeled valuation algebras in a second way as idempotent ones, see Section \ref{subsec:LabInfAlg}. From such a labeld valuation algebra we may reconstruct a domain-free one as in the idempotent case, see Section \ref{subsec:Dual}. It follows that local computation still works, see Section \ref{sec:LocComp}, except for the method described in Section \ref{subsec:CompHypTree}.

There is also a commutative version of a labeled valuation algebra. The transport operstions $t_x$ can be replaced by projection operators $\pi_x$ defined for $x \leq d(\phi)$ only. Its axioms are like in the idempotent case, without idempotency. If as usual $\Psi_x$ is the set of all element of $\Psi$ with label $x$,
\begin{enumerate}
\item \textit{Semigroup:} $( \Psi,\cdot)$ is a commutative semigroup.
\item \textit{Lattice:} $(Q,\leq)$ is a lattice.
\item \textit{Labeling:} $d(\phi \cdot \psi) = d(\phi) \vee d(\psi)$, $d(\pi_y(\psi)) = y$ if $y \leq d(\psi)$.
\item \textit{Unit:} For all $x \in Q$, the semigroups $(\Psi_x,\cdot)$ have a unit element $1_x$ and for all $y \leq x \in Q$, $\pi_y(1_x) = 1_y$ and $1_x \cdot 1_y = 1_{x \vee y}$.
\item \textit{Projection:} If $x \leq y \leq z = d(\psi)$, then $\pi_x(\psi) = \pi_x(\pi_y(\psi))$.
\item \textit{Combination:} If $d(\phi) = x$ and $d(\psi) = y$, then $\pi_x(\phi \cdot \psi) = \phi \cdot \pi_{x \wedge y}(\psi)$.
\item \textit{Identity:} If $x = d(\psi)$, then $\pi_x(\psi) = \psi$.
\end{enumerate}
The condition $\pi_y(1_x) = 1_y$ if $y \leq x$ is called \textit{stabilioty}. There are important instances where stability does not hold, see the example below. Then, however, the labeled valuation algebra has no associated dual domain-free valuation algebra. This is then essentially the axiomatic system proposed in \cite{shenoyshafer90}. The prototype of such a valuation algebra is presented in the following example, it is an abstraction of Bayesian networks, \cite{lauritzenspiegelhalter88}.

\begin{example} \textbf{Probability potentials:} \label{expl:ProbPot}
Consider a multivariate model with a (finite) set $X_j$ of variables, $i \in J$. If $U_s = \prod_{j \in s} U_i$ is the domain of the set $X_i$, $i \in s \subseteq J$ and the $U_j$ are finite sets, then a non-negative, non-null function $p : U_s \rightarrow \mathbb{R}$ is called a probability potential on domain $U_s$. We label it with $s$, $d(p) = s$. Combination and projection are defined as
\begin{enumerate}
\item \textit{Combination:} If $d(p_1) ) = s$ and $d(p_2) = t$, then for a tuple $x \in U_{s \cup t}$,
\begin{eqnarray*}
p_1 \cdot p_2(x) = p_1(x \vert s)p_2(x \vert t),
\end{eqnarray*}
where $x \vert s$ and $x \vert t$ are the restriction of tuple $x$ to subsets $s$ and $t$ of components.
\item \textit{Projection:} If $d(p) = s$ and $t \subseteq s$, then if $x$ and $y$ are tuples in $U_t$ and $U_s$ repectively
\begin{eqnarray*}
\pi_t(x) = \sum_{y:y \vert t = x} p(y).
\end{eqnarray*}
\end{enumerate}
Probability potentials are called so, since they may be normalized to probability distribution on the domains $U_s$. Then, projection is seen to be essentially marginalization. We refer to \cite{shafer96,kohlas03} for a discussion how this system relates to probabilistic reasoning and also for a proof that probability potentials with these operations form a valuation algebra. Obviously it is not idempotent. The unit of combination on domain $U_s$ is the function $p(x) = 1$ for all $x \in U$. Stability clearly does not hold, so there is no associated domain-free version. 
\end{example}

In conclusion, so far, the theory of non-idempotent valuation algebras differs not much from idempotent information algebras. The big difference comes with the definition and exploitation of information order, Section \ref{sec:InfOrder}, which depends on idempotency. This concerns especially extensions, Section \ref{sec:ExtInfAlg}, and atoms, Section \ref{sec:AtomicAlg} and also the whole question of finiteness, Section \ref{sec:FiniteInf}. In Section \ref{subsec:NOnIdOrder}, we shall see how we can introduce an order also in (some cases of) non-idempotent valuation algebras. This order needs some additional structures which are presented in Sections \ref{subsec:RegularAlg} and \ref{subsec:SeparativeAlg}. These refined structures allow then for an interesting concept, \textit{continuation}, a concept which is uninteresting in the idempotent case, Section \ref{subsec:Continuation}:


\section{Regular algebras} \label{subsec:RegularAlg}

Order in semigroup theory has been studied in several papers, we cite only two of them, \cite{nambooripad80,mitsch86}. These papers study natural order, that is an order, which can be defined in terms of the operations of the semigroup. This is surely desirable. Of particular interest for these theories are \textit{regular} semigroups. In the context of valuation algebras, such regular semigroups or rather the generalisation of them to valuation algebras, turns out to be of interest in two respects: They allow to introduce partial division into the algebra, which allows to adapt local computation architectures known for Bayesian networks to valuation algebras \cite{lauritzenjensen97,kohlas03}. Secondly, this division permits also to generalise conditioning, as known in probability, to valuation algebras \cite{kohlas03}, see Section \ref{subsec:Continuation}.. Further, as we shall see in Section \ref{subsec:NOnIdOrder}, this is relevant for information order too.

We summarise here the theory of regular semigroups and adapt it to \textit{valuation algebras}, generalizing the theory of regular valuation algebras in \cite{kohlas03}. We start with the definition of regularity in valuation algebras. We do this in the domain-free case, although it could also be done in the labeled one.

\begin{definition} \textit{Regular Valuation Algebras:} \label{def:RegInfAlg}
Let $(\Phi,\cdot,1,E)$ be a domain-free valuation algebra. An element $\phi \in \Phi$ is called \textit{regular}, if for all $x \in D$ there is an element $\chi \in \Phi$ with support $x$ such that
\begin{eqnarray} 
\phi = \epsilon_{x}(\phi) \cdot \chi \cdot \phi.
\end{eqnarray}
The information algebra $(\Phi,\cdot,1:E)$ is called regular, if all its elements are regular.
\end{definition}

Note that the unit element $1$ is regular. Of course, the element $\chi$ above in the definition of regularity depends both on $x$ and $\psi$, although we do not express this dependence explicitly. If $y$ is a support of $\psi$, then regularity implies also
\begin{eqnarray} 
\psi = \psi \cdot \chi \cdot \psi.
\end{eqnarray}
This is exactly the definition of regularity in a semigroup $(\Phi;\cdot)$ and establishes the link to semigroup theory, see for example \cite{cliffordpreston67} and the work cited above. Note that in these references semigroups are not assumed to be commutative, as is the case here.  

In this section we assume that $(\Phi,\cdot,1,E)$ is regular. Two elements $\phi$ and $\psi$ from $\Phi$ are called inverses, if
\begin{eqnarray}
\phi = \phi \cdot \psi \cdot \phi \textrm{ and}\ \psi = \psi \cdot \phi \cdot \psi
\end{eqnarray}
We keep with the notation in the literature, although in our commutative case we could also have written $\phi = \phi \cdot \phi \cdot \psi,\ldots$.  

The following results are well-known from semigroup theory (see for instance \cite{kohlas03}): If $\phi = \phi \cdot \psi \cdot \phi$, then $\phi$ and $\psi \cdot \phi \cdot \psi$ are inverses. Each element of a regular semigroup has thus an inverse, and this inverse is unique. If $\phi$ and $\psi$ are inverses, then $f = \phi \cdot \psi$ is an idempotent element, $f \cdot f = f$. Further we have $f \cdot \phi = \phi$ and $f \cdot \psi = \psi$. If $S$ is a subset of $\Phi$, define $\psi \cdot S$ to be the set $\{\psi \cdot \phi:\phi \in S\}$. There exists for any $\psi \in \Phi$ a unique idempotent $f_{\psi}$ such that  $\psi \cdot \Phi = f_{\psi} \cdot \Phi$, since if $\phi$ and $\psi$ are inverses, $\phi \cdot \psi = f_\psi$ implies $\psi = f_\psi \cdot \psi$. The \textit{Green relation} is defined as
\begin{eqnarray}
\phi \equiv_{\gamma} \psi \textrm{ if}\ \phi \cdot \Phi = \psi \cdot \Phi.
\end{eqnarray}
It is an equivalence relation in $\Phi$. Its equivalence classes $[\psi]_{\gamma}$ are obviously \textit{commutative groups} for all $\psi \in \Phi$ \cite{kohlas03}. So $\Phi$ is a union of disjoint groups. The unit element of the group $[\phi]_{\gamma}$ is the idempotent $f_{\phi}$ and for any $\psi \in [\phi]_\gamma$ its inverse in the semigroup is the inverse in $[\phi]_\gamma$.

Consider now the idempotents $F = \{f_{\psi}:\psi \in \Phi\}$. They form an idempotent sub-semigroup of $(\Psi;\cdot)$. According to Section \ref{sec:InfOrder} they are partially ordered by $f_{\phi} \leq f_{\psi} $ if $f_{\phi} \cdot f_{\psi} = f_{\psi}$, just as in information order. The unit $1$ and (the null element $0$ if present) are idempotents. So, the idempotents $F$ form a semilattice where $f_{\phi} \cdot f_{\psi} = f_{\phi} \vee f_{\psi}$. Further, we have also
\begin{eqnarray}
f_{\phi} \cdot f_{\psi} = f_{\phi \cdot \psi}.
\end{eqnarray}
Since the idempotents $f_{\phi}$ uniquely represent their class $[\phi]_{\gamma}$, we may also define a partial order among classes by $[\phi]_{\gamma} \leq [\psi]_{\gamma}$ if $f_{\phi} \leq f_{\psi}$. Then we obtain
\begin{eqnarray}
[\phi \cdot \psi]_{\gamma} = [\phi]_{\gamma} \vee [\phi]_{\gamma}.
\end{eqnarray}

Note that in an \textit{idempotent} semigroup, as for instance in information algebras, any element is its own inverse and the groups $[\phi]_\gamma$ degenerate to trivial single-element groups. So, the theory of regular semigroups is not of interest  for information algebras.

So far, this is essentially semigroup theory. We now consider extraction and extend thus this theory to valuation algebras. Here is a first important result:

\begin{theorem} \label{th:GreenCongr}
Let $(\Phi,\cdot,1;E)$ be a regular valuation algebra. The Green relation $\equiv_{\gamma}$ is a congruence relative to combination and extraction in the algebra $\Phi,$
\end{theorem}

\begin{proof}
The relation $\equiv_{\gamma}$ is an equivalence relation. If $\phi \equiv_{\gamma} \psi$, then $[\phi]_{\gamma} = [\psi]_{\gamma}$. Consider any element $\eta$ of $\Phi$. Then $[\phi]_{\gamma} \vee [\eta]_{\gamma} = [\psi]_{\gamma} \vee [\eta]_{\gamma}$, hence $[\phi \cdot \eta]_{\gamma} = [\psi \cdot \eta]_{\gamma}$ and thus $\phi \cdot \eta \equiv_{\gamma} \psi \cdot \eta$.

Assume again $\phi \equiv_{\gamma} \psi$ such that $\phi \cdot \Phi = \psi \cdot \Phi$,  and consider the operator $\epsilon_{x}$. From $\phi \in \phi \cdot \Psi$ we conclude that $\phi = \psi \cdot \chi$ for some $\chi \in \Phi$ and therefore $\epsilon_{x}(\phi) = \epsilon_{x}(\psi \cdot \chi)$. By regularity we have $\psi = \epsilon_{x}(\psi) \cdot \chi' \cdot \psi$ and thus
$\epsilon_{x}(\phi) = \epsilon_{x}(\epsilon_{x}(\psi) \cdot \chi \cdot \chi' \cdot \psi) = \epsilon_{x}(\psi) \cdot \epsilon_{x}(\chi \cdot \chi' \cdot \psi)$. This means that $\epsilon_x(\phi) \in \epsilon_x(\psi) \cdot \Phi$. By symmetry we have also $\epsilon_x(\psi) \in \epsilon_x(\phi) \cdot \Phi$, and therefore $\epsilon_{x}(\phi) \equiv_{\gamma} \epsilon_{x}(\psi)$. This proves that $\equiv_{\gamma}$ is a congruence.
\end{proof}

Based on Theorem \ref{th:GreenCongr}, we may consider the quotient algebra $(\Phi/\gamma,\cdot,[1];E)$, which by general results of universal algebra must still be a valaution algebra. In fact, we define the following operations between classes
\begin{enumerate}
\item \textit{Combination:} $[\phi]_{\gamma} \cdot [\psi]_{\gamma} = [\phi \cdot \psi]_{\gamma}$,
\item \textit{Extraction:} $\epsilon_{x}([\psi]_{\gamma}) = [\epsilon_{x}(\psi)]_{\gamma}$.
\end{enumerate}
We denote the operations of combination and extraction in $\Phi/\gamma$ by the same symbols as in $\Phi$; there is no risk of confusion. The projection pair of maps $(f,g)$, where $f(\psi) = [\psi]_{\gamma}$ and $g(\epsilon_{x}) = \epsilon_{x}$ (meaning at the right hand side, the operator in $\Psi/\gamma$) is clearly a homomorphism. In addition, it turns out that the information algebra $(\Phi/\gamma,\cdot,1;E)$ is idempotent. 

\begin{theorem} \label{th:QuotAlkgera}
Let $(\Phi,\cdot,1;E)$ be a regular valuation algebra and $\equiv_{\gamma}$ the Green relation. Then the quotient algebra $(\Phi/\gamma,\cdot,[1]_\gamma;E)$ is an idempotent information algebra, homomorphic to $(\Phi,\cdot,1;E)$.
\end{theorem}

\begin{proof}
That $(\Phi/\gamma,\cdot,[1]_\gamma;E)$ is a valuation follows since the pair of maps defined above form a homomorphism. We claim that $\epsilon_x(\psi) \cdot \psi \equiv_\gamma \psi$. This implies then idempotency in $\Phi/\gamma$. In fact, if $\eta \in \epsilon_x(\psi) \cdot \psi \cdot \Phi$, then $\eta \in \psi \cdot \Phi$. Conversely, by regularity $\psi = \epsilon_x(\psi) \cdot \chi \cdot \psi$ for some element $\chi$, therefore, if $\eta \in \psi \cdot \Phi$, then $\eta \in \epsilon_x(\psi) \cdot \psi \cdot \Phi$.
\end{proof}

Instead of the quotient algebra $\Phi/\gamma$ we can also consider the idempotents in the equivalence classes, because there is a one-to-one association between idempotents and their classes. In the signature $(F,\cdot,f_1;\bar{E})$, where $F = \{f_{\psi}:\psi \in \Phi)$, $\bar{E} = \{\bar{\epsilon}_x:x \in Q\}$, again the two operations of combination and extraction are defined:
\begin{enumerate}
\item \textit{Combination:} $f_{\phi} \cdot f_{\psi} = f_{\phi \cdot \psi}$,
\item \textit{Extraction:} $\bar{\epsilon}_{x}(f_{\psi}) = f_{\epsilon_{x}(\psi)}$.
\end{enumerate}
This algebra is still an (idempotent) information algebra, homomorphic to $\Phi$. Because of the idempotency, it can be considered as the \textit{deterministic part} of $\Phi$ (although it is not a subalgebra of $\Phi$ since $\bar{\epsilon}$ and $\epsilon$ are different). By the pair of maps $[\psi]_{\gamma} \mapsto f_{\psi}$ and $\epsilon \mapsto \bar{\epsilon}$, the algebras $\Psi/\gamma$ and $F$ are isomorphic. We refer to the example of probability potentials below for an illustration.

We remark that parallel to the domain-free case the theory of regular in a labeled valuation algebras may be developed, even in the case of commutative algebras and even if stability does not hold. In fact in this last case, regularity of a labeled element $\psi$ is defined as follows:
\begin{enumerate}
\item An element $\psi$ of a commutative labeled valuation algebra $\Psi$ is called regular, if there exists for all $x \leq d(\psi)$ an element $\chi \in \Psi$ with $d(\chi) = x$ such that
\begin{eqnarray*}
\psi = \pi_x(\psi) \cdot \chi \cdot \psi.
\end{eqnarray*}
\item The valuation algebra is called regular, if all its elements are regular.
\end{enumerate}
This permits to derive a theory fully parallel to the domain-free case. Rather than to develop this, we prefer to illustrate it with the example of probability potentials. For the full labeled theory in the commutative case we refer to \cite{kohlas03}.

\begin{example} \textbf{Probability Potentials:}  \label{expl:ProbPotInv} Probability potentials were introduced as mappings $p : \Theta \mapsto \mathbb{R}^+ \cup \{0\})$ from the domains of a multivariate model to nonnegative real numbers. This labeled valuation algebra is regular, in the sense that for any probability potential $p$ with label $s$ and $t \leq s$ there is a potential $q$ with label $t$ such that $p = p \cdot \pi_t(p) \cdot q$. In fact, the potential $q$ is determined as follows, for a tuple $x \in U_t$,
\begin{eqnarray*}
q(x) = \left\{ \begin{array}{ll} \frac{1}{\pi_t(p)(x)} & \textrm{if }\ \pi_t(p)(x) \not= 0, \\ 0 & \textrm{otherwise.} \end{array} \right.
\end{eqnarray*}
The idempotents of the group $[p]_\gamma$ of a potential $p$ is the potential $f_p(x) = 1$ for all $x \in U_s$ for which $p(x) > 0$ and $f_p(x) = 0$ if $p(x) = 0$. So, the idempotents are the indicator functions of the \textit{support sets} $\{x: p(x) > 0\}$ of the probability potentials. Note that the projection of an idempotent is not itself an idempotent. The idempotent labeled valuation algebra $F$, defined similarly as in the domain-free case,  corresponds to the labeled set algebra of subsets of the frames $U_s$, but is not exactly a subset algebra. 
\end{example}


\section{Separative algebras} \label{subsec:SeparativeAlg}

Here we go one step beyond regular algebras. Consider again a domain-free valuation algebra $(\Phi,\cdot,1;E)$, $E = \{\epsilon_x:x \in Q)$. Instead of assuming it to be regular, and then use the Green relation, we start with a congruence, similar to the Green relation and base the theory on this relation. Thus, assume that there is a congruence $\equiv_{\gamma}$ relative to combination and extraction in $\Phi$ such that
\begin{eqnarray} \label{eq:SepCongr}
\epsilon_{x}(\psi) \cdot \psi \equiv_{\gamma} \psi
\end{eqnarray}
for all $\psi \in \Phi$ and $x \in Q$. Since any element $\psi$ has a support, we have also
\begin{eqnarray}
\psi \cdot \psi \equiv_{\gamma} \psi
\nonumber
\end{eqnarray}
The equivalence classes $[\psi]_{\gamma}$ are semigroups. Indeed, if $\phi,\chi \in [\psi]_{\gamma}$, then $\phi \equiv_\gamma \chi$ and $\chi \equiv_\gamma \psi$, hence $\phi \cdot \chi \equiv_{\gamma} \psi \cdot \psi$ since $\equiv_{\gamma}$ is a congruence. But $\psi \cdot \psi \equiv_{\gamma} \psi$, thus $\phi \cdot \chi \equiv_{\gamma} \psi$ hence $\phi \cdot \chi \in [\psi]_{\gamma}$. 

As in the previous section the quotient algebra $\Phi/\gamma$ is an idempotent information algebra, homomorphic to $\Phi$, if the operations are defined as
\begin{enumerate}
\item \textit{Combination:} $[\phi]_{\gamma} \cdot [\psi]_{\gamma} = [\phi \cdot \psi]_{\gamma}$.
\item \textit{Extraction:} $\epsilon_{x}([\psi]_{\gamma}) = [\epsilon_{x}(\psi)]_{\gamma}$.
\end{enumerate}
Idempotency of $\Phi/\gamma$ follows from condition (\ref{eq:SepCongr}).

Again, since the classes form an idempotent algebra, they are partially ordered by $[\phi]_{\gamma} \leq [\psi]_{\gamma}$ if $[\phi]_{\gamma} \cdot [\psi]_{\gamma} = [\phi]_{\gamma}$. Under this order we have
\begin{eqnarray}
[\phi]_{\gamma} \cdot [\psi]_{\gamma} = [\phi]_{\gamma} \vee [\psi]_{\gamma}.
\nonumber
\end{eqnarray}

Contrary to regular algebras, this is not sufficient for the classes $[\psi]_\gamma$ to be groups. In semigroup theory embeddings of semigroups into a disjoint union of groups is studied, see \cite{cliffordpreston67}. A sufficient condition for this to be possible is \textit{cancellativity}, that is
\begin{eqnarray}
\phi \cdot \psi = \phi \cdot \psi'
\end{eqnarray}
implies $\psi = \psi'$. We assume therefore that all semigroups $[\phi]_{\gamma}$ are cancellative. This leads to the following definition.

\begin{definition} \textit{Separative Information Algebras:} \label{def:SepInfAlg}
Let $(\Phi,\cdot,1;E)$ be a domain-free ivaluation algebra. It is called separative, if there exists a congruence $\equiv_{\gamma}$ relative to combination and extraction in $\Psi$ such that
\begin{enumerate}
\item $\epsilon_{x}(\psi) \cdot \psi \equiv_{\gamma} \psi$ for all $\psi \in \Psi$ and for all $x \in D$.
\item The semigroups $[\psi]_{\gamma}$ are cancellative for all $\psi \in \Psi$.
\end{enumerate}
\end{definition}

We remark that separative valuation algebras have been studied in \cite{kohlas03} with respect to local computation with division and to generalisation of conditionals from probability to valuations or information, see also Section \ref{subsec:Continuation}. Obviously, regular valuation algebras are separative.  For further examples of separative valuation algebras, we refer to \cite{kohlas03,poulykohlas11}. We mention also, that as far as local computation with division and conditioning is concerned, it is sufficient that $\equiv_\gamma$ is a congruence with respect to combination only. But for our present concern, congruence with respect to extraction is also desirable and many separative instances satisfy this condition.

A cancellative semigroup such as $[\psi]_{\gamma}$ can be embedded into a group. The classical procedure is like for extending integers to rational numbers as follows: Consider ordered pairs $(\phi,\psi)$ for $\phi,\psi \in [\eta]_{\gamma}$ and define a relation among pairs by
\begin{eqnarray}
(\phi,\psi) \equiv (\phi',\psi') \textrm{ iff}\ \phi \cdot \psi' = \phi' \cdot \psi.
\nonumber
\end{eqnarray}
This is an equivalence relation thanks to cancellativity. Let $[\phi,\psi]$ denote the equivalence classes of this equivalence and let $\gamma(\eta)$ denote the set of these equivalence classes $[\eta]_{\gamma}$. Then we define the operation
\begin{eqnarray}
[\phi,\psi] \cdot [\phi',\psi'] = [\phi \cdot \phi',\psi \cdot \psi']
\nonumber
\end{eqnarray}
in $\gamma(\eta)$. This is well defined, since the equivalence is a congruence relative to the operation $(\phi,\psi) \cdot (\phi',\psi') = (\phi \cdot \phi',\psi \cdot \psi')$ between pairs. With this operation every $\gamma(\eta)$ becomes a group. Its unit is $[\psi,\psi]$ and the inverse of $[\phi,\psi]$ is $[\psi,\phi]$. The class $[\psi]_{\gamma}$ is embedded into $\gamma(\psi)$ as a a semigroup by the map
\begin{eqnarray}
\psi \mapsto [\psi \cdot \psi,\psi].
\nonumber
\end{eqnarray}
Define
\begin{eqnarray}
\Phi^* = \bigcup_{\psi \in \Phi} \gamma(\psi).
\nonumber
\end{eqnarray}
In order to distinguish elements of $\Phi^*$ from those of $\Phi$, we denote elements of $\Phi^*$ by lower case letters like $a,b,\ldots$. The union of groups $\Phi^*$ becomes a semigroup, if we define for $a = [\phi_{a},\psi_{b}]$ and $b = [\phi_{b},\psi_{b}]$,
\begin{eqnarray}
a \cdot b = [\phi_{a} \cdot \phi_{b},\psi_{a} \cdot \psi_{b}].
\nonumber
\end{eqnarray}
This operation is well-defined, associative and commutative. Thus $(\Phi^*;\cdot)$ is a commutative semigroup and $(\Phi;\cdot)$ is embedded into it as a semigroup by the map $\psi \mapsto [\psi \cdot \psi,\psi]$ as can easily be verified. In the sequel, in order to simplify notation, we denote the elements $[\psi \cdot \psi,\psi]$ of the image of $(\Phi;\cdot)$ under this map simply by $\psi$.

If $(\Phi,\cdot,1;E)$ is a separative valuation algebra, then the quotient algebra $\Phi/\gamma$, is an idempotent information algebra, homomorphic to $\Phi$ as noted above. Any group $\gamma(\psi)$ has a unique unit and idempotent element, denoted by $f_{\psi}$. The idempotent information algebra $F$ of idempotents or the units of the groups $\gamma(\psi)$, with the operations defined as follows
\begin{enumerate}

\item \textit{Combination:} $f_{\phi} \cdot f_{\psi} = f_{\phi \cdot \psi}$,
\item \textit{Extraction:} $\bar{\epsilon}_{x}(f_{\psi}) = f_{\epsilon_{x}(\psi)}$,
\end{enumerate}
is isomorphic to the quotient algebra $\Phi/\gamma$. Note however, that the elements of $F$ do not, in general, belong to $\Phi$ as in the regular case. Nevertheless, we may still consider the elements of $F$ as the deterministic parts of $\Phi^*$. 

To conclude this section, we introduce as an illustration the valuation algebra of probability densities, which turns out to be separative \cite{kohlas03}.

\begin{example} \textbf{Probability densities:} \label{expl:ProbDens}
As in the case of probability potentials, we consider a labeled valuation algebra of probability densities. The domains $\Omega_s$ for a group of variables $s \subseteq I = \{1,\ldots,n\}$ is the set $\mathbb{R}^s$ of real-valued $s$-tuples. These tuples will be denoted by boldface letters like $\mathbf{x},\mathbf{y},\ldots$. A density $f$ on domain $\mathbb{R}^s$ is then a non-negative, continuous function $f : \mathbb{R}^2 \rightarrow \mathbb{R}$, $f(\mathbf{x} \geq 0$, whose integral
\begin{eqnarray*}
\int_{\infty}^\infty f(\mathbf{x}) d\mathbf{x}
\end{eqnarray*}
exists and is finite. If $f$ is a non-null density on $\mathbb{R}^s$, then we may consider the corresponding normalized (proper) density function
\begin{eqnarray*}
f^\rightarrow(\mathbf{x}) = \frac{f(\mathbf{x})}{\int f(\mathbf{x}) d\mathbf{x}},
\end{eqnarray*}
Then we have $\int f^\rightarrow(\mathbf{x}) d\mathbf{x} = 1$. So it is a probability density.

Let $\Psi_s$ denote the set of all densities on domain $\mathbb{R}^s$ (including the null density) and
\begin{eqnarray*}
\Psi = \bigcup_{s \subseteq I} \Phi_s.
\end{eqnarray*}
Among these densities, we define the operations of labeling, combination and projection. Let $\mathbf{x}[t]$ denote the restriction of the $s$-tuple $\mathbf{x}$ to the subset $t \subseteq s$ of indexes.
\begin{enumerate}
\item \textit{Labeling:} $d(f) = s$, if $f$ is a density on $\mathbb{R}^s$.
\item \textit{Combination:} $f \cdot g(\mathbf{x}) = f(\mathbf{x}[s])g(\mathbf{x}[t])$ where $\mathbf{x}$ is a $(s \cup t)$-tuple if $d(f) = s$, $d(g) =t$.
\item \textit{Projection:} $\pi_t(f)(\mathbf{x}) = \int_{- \infty}^\infty f(\mathbf{x},\mathbf{y})$, where $\mathbf{x}$ and $\mathbf{y}$ are $t$ and $(s - t)$-tuples respectively if $d(f) = s$.
\end{enumerate}
Of course, projection corresponds to marginalization for proper density functions. It can be shown that this system is indeed a (labeled) valuation algebra, where the null-density on $\mathbb{R}^s$ is the null element $0_s$ and $f\mathbf{x}) = 1$ the unit $1_x$. 

The relation
\begin{eqnarray*}
f \equiv_\gamma g \textrm{ if}\ d(f) = d(g) \textrm{ and}\ f(\mathbf{x}) = 0 \Leftrightarrow g(\mathbf{x}) = 0
\end{eqnarray*}
is clearly an equivalence. If $supp(f)$ denotes the subset of tuples where $f(\mathbf{x}) > 0$, then two densities $f$ and $g$ are equivalent if they have the same support sets, $supp(f) = supp(g)$. It follows $f \cdot \pi_t(f) \equiv f$, if $t \subseteq d(f)$ since $\pi_t(f)(\mathbf{x}[t]) = 0$ implies $f(\mathbf{x}) = 0$. The semigroup of densities with the same support sets is obviously cancellative. It is thus embedded into the group of quotients of densities with the same support sets. Therefore the valuation algebra of densities is separative. It is embedded into the semigroup $\Psi^*$ which is the union of the groups of densities with the same support. The functions $e_f(\mathbf{x}) = 1$ for $\mathbf{x} \in supp(f)$ and $= 0$ otherwise are the units of these groups. Note that these unit are not necessarily densities since their integral may be infinite. In particular the function $f(\mathbf{x}) = 1$ for all $\mathbf{x} \in \mathbb{R}^s$ is a unit, but not a density. The inverse of a density $f$ is $f^{-1}(\mathbf{x}) = 1/f(\mathbf{x})$ if $\mathbf{x} \in supp(f)$ and $= 0$ otherwise.  

We remark that we could also have considered measurable functions and Lebesgue integrals, see \cite{kohlas03}
\end{example}


\section{Information order} \label{subsec:NOnIdOrder}

We now use the theories developed in the previous two sections for studying information order in a valuation algebra. Information as represented by valuations may be, in informal terms, more or less precise, more or less informative. This should be reflected by some order between pieces of information. This has been modelled for information algebras in Section \ref{sec:InfOrder} by defining $\phi \leq \psi$ if $\phi \cdot \psi = \psi$. As already stated, this information order in information algebras depends essentially on idempotency and can not be carried over to valuation algebras. A different approach is needed.

Let then $(\Phi,\cdot,1;E)$ be a domain-free valuation algebra. The basic idea is that a piece of information is more informative than an other one, if one needs to add a further piece of information to the second one to get the first one. So, we define, for $\phi,\psi \in \Phi$,
\begin{eqnarray} \label{eq:Preorder}
\phi \leq \psi, \textrm{ iff there exists}\ \chi \in \Phi \textrm{ such that}\ \psi = \phi \cdot \chi.
\end{eqnarray}
This relation satisfies
\begin{enumerate}
\item \textit{Reflexivity:} $\psi \leq \psi$, since $\psi = \psi \cdot 1$,
\item \textit{Transitivity:} $\phi \leq \psi$ and $\psi \leq \eta$ imply $\phi \leq \eta$, since $\psi = \phi \cdot \chi_{1}$, $\eta = \psi \cdot \chi_{2}$ imply $\eta = \phi \cdot \chi_{1} \cdot \chi_{2}$.
\end{enumerate}
Antiysymmetry however does not hold in general. Therefore, the relation $\leq$ defined in (\ref{eq:Preorder}) is a \textit{preorder} in $\Phi$. 

This order is the subject of the present section. Information order can be studied both in labeled or domain-free valuation algebras. We propose to base our discussion on domain-free algebras.

If the valuation algebra $\Phi$ is \textit{idempotent}, that is, if it is an information algebra, then $\psi = \phi \cdot \chi$, gives by idempotency, if both sides are combined by $\psi$, $\psi = (\phi \cdot \chi) \cdot \psi = \phi \cdot (\phi \cdot \chi) \cdot \psi = \phi \cdot \psi \cdot \psi = \phi \cdot \psi$. So the information order in information algebra is the same order as the one proposed here. In idempotent information algebras, the relation $\leq$ is a \textit{partial order}, since $\phi \leq \psi$ and $\psi \leq \phi$ imply $\phi = \psi \cdot \phi =  \psi$. Here $\phi \leq \psi$ means that nothing is gained if the piece of information $\phi$ is added to $\psi$, the information in $\phi$ is already covered by $\psi$. Recall that in this idempotent case 
\begin{enumerate}
\item $1 \leq \psi \leq 0$ for all $\psi \in \Psi$,
\item $\phi,\psi \leq \phi \cdot \psi$,
\item $\phi \leq \psi$ implies $\phi \cdot \eta \leq \psi \cdot \eta$ for all $\eta \in \Psi$,
\item $\epsilon_{x}(\psi) \leq \psi$ for all $x \in D$ and $\psi \in \Psi$,
\item $\phi \leq \psi$ implies $\epsilon_{x}(\phi) \leq \epsilon_{x}(\psi)$ for all $x \in D$,
\item $x \leq y$ implies $\epsilon_{x}(\psi) \leq \epsilon_{y}(\psi)$ for all $\psi \in \Psi$.
\end{enumerate}
These are clearly properties one would expect from an information order in general: Vacuous information is least informative, contradiction  (which properly speaking is not an information) is the greatest element in the information order; combined information is more informative than each of its parts, the order is compatible with combination and extraction of information does not increase information.

Note that the preorder defined in (\ref{eq:Preorder}), satisfies the first three of these requirements (if possibly null elment is present). The remaining ones are not guaranteed in general and need special consideration. In particular we show in this section, that in regular and separative valuation algebras, the information order indeed satisfies also the remaining three properties.
This will also illuminate the relation of the preorder to the partial order of idempotent information and exhibits the limits of the preorder.

We summarize now some results about the preorder in $\Phi$ and partial order among idempotents in $F$ and among the classes $[\phi]_{\gamma}$ in regular valuation algebras.

\begin{lemma} \label{le:PreorderClasses}
Let $(\Phi,\cdot,1;E)$ be a regular valuation algebra. Then
\begin{enumerate}
\item $\phi \leq \psi$ iff $[\phi]_{\gamma} \leq [\psi]_{\gamma}$,
\item $\phi \leq \psi$ iff $\psi \cdot \Phi = \phi \cdot \psi \cdot \Phi$,
\item $\phi \leq \psi$ iff $\psi \cdot \Phi \subseteq \phi \cdot \Phi$,
\item $\phi \leq \psi$ and $\psi \leq \phi$ iff $\phi \equiv_{\gamma} \psi$,
\end{enumerate}
\end{lemma}

\begin{proof}
1.) Assume $\phi \leq \psi$, that is $\phi \cdot \chi = \psi$. Then $[\phi \cdot \chi]_{\gamma} = [\phi]_{\gamma} \vee [\chi]_{\gamma} = [\psi]_{\gamma}$. This shows that $[\phi]_{\gamma} \leq [\psi]_{\gamma}$. 

Conversely, assume $[\phi]_{\gamma} \leq [\psi]_{\gamma}$ such that $[\phi \cdot \psi]_{\gamma} = [\phi]_{\gamma} \vee [\psi]_{\gamma} = [\psi]_{\gamma}$. This means that $\psi \cdot \Psi = \phi \cdot \psi \cdot \Psi$, hence $\psi \in \phi \cdot \psi \cdot \Psi$, therefore $\psi = \phi \cdot \psi \cdot \chi$ for some $\chi$. But this means that $\phi \leq \psi$.

2.) We have just proved that $\psi \cdot \Psi = \phi \cdot \psi \cdot \Phi$ implies $\phi \leq \psi$. Assume then that $\phi \leq \psi$. By item 1 we have also $f_{\phi} \leq f_{\psi}$ or $f_{\phi} \cdot f_{\psi} = f_{\phi \cdot \psi} = f_{\psi}$. But then $\psi \cdot \Phi = f_{\psi} \cdot \Phi = f_{\phi \cdot \psi} \cdot \Phi = \phi \cdot \psi \cdot \Phi$.

3.) If $\phi \leq \psi$, then $\psi = \phi \cdot \chi$. Consider $\eta \in \psi \cdot \Phi$, then $\eta = \psi \cdot \chi' = \phi \cdot \chi \cdot \chi'$. So $\eta \in \phi \cdot \Phi$. Conversely, if $\psi \cdot \Phi \subseteq \phi \cdot \Phi$, then $\psi \in \phi \cdot \Phi$, hence there is a $\chi$ such that $\psi = \phi \cdot \chi$, and thus $\phi \leq \psi$. 

4.) We have by item 2 $\phi \leq \psi$ iff $\psi \cdot \Phi = \phi \cdot \psi \cdot \Phi$ and $\psi \leq \phi$ iff $\phi \cdot \Phi = \phi \cdot \psi \cdot \Phi$. Therefore, $\phi \cdot \Phi = \psi \cdot \Phi$, hence $\phi \equiv_{\gamma} \psi$.
\end{proof}

Here follow a few results on order and extraction, which show the validity in a regular valuation algebra of the expected properties 4.) to 6.) of an information order formulated above.

\begin{theorem} \label{th:OrderAndExtr}
Let $(\Phi,\cdot,1,E)$ be a regular valuation algebra. Then
\begin{enumerate}
\item $\epsilon_{x}(\psi) \leq \psi$ for all $x \in D$ and $\psi \in \Phi$.
\item $\phi \leq \psi$ implies $\epsilon_{x}(\phi) \leq \epsilon_{x}(\psi)$ for all $x \in D$.
\item $x \leq y$ implies $\epsilon_{x}(\psi) \leq \epsilon_{y}(\psi)$ for all $\psi \in \Phi$.
\end{enumerate}
\end{theorem}

\begin{proof}
1.) By regularity $\psi = \psi \cdot \chi \cdot \epsilon_{x}(\psi)$ where $\epsilon_{x}(\chi) = \chi$. Applying the extraction operator on both sides gives $\epsilon_{x}(\psi) = \epsilon_{x}(\psi) \cdot \epsilon_{x}(\psi) \cdot \chi$, hence $\epsilon_{x}(\psi) \geq \chi$ and therefore $[\epsilon_{x}(\psi)]_{\gamma} \geq [\chi]_{\gamma}$ (Lemma \ref{le:PreorderClasses}). From the regularity formula we obtain also $[\psi]_{\gamma} = [\psi]_{\gamma} \vee [\chi]_{\gamma} \vee[\epsilon_{x}(\psi)]_{\gamma} = [\psi]_{\gamma} \vee [\epsilon_{x}(\psi)]_{\gamma}$, hence $[\epsilon_{x}(\psi)]_{\gamma} \leq [\psi]_{\gamma}$. This implies $\epsilon_{x}(\psi) \leq \psi$ (Lemma \ref{le:PreorderClasses}).

2.) If $\phi \leq \psi$, then $\psi \cdot \Phi = \phi \cdot \psi \cdot \Phi$ (Lemma \ref{le:PreorderClasses}). This implies $\psi = \psi \cdot \phi \cdot \chi$ for some $\chi \in \Phi$. By regularity we have $\phi = \phi \cdot \epsilon_{x}(\phi) \cdot \mu$ and $\psi = \psi \cdot \epsilon_{x}(\psi) \cdot \mu'$, where $x$ is a support of both $\mu$ and $\mu'$. From this we deduce
\begin{eqnarray}
\epsilon_{x}(\psi) &=& \epsilon_{x}(\psi \cdot \phi \cdot \chi)
\nonumber \\
&=&\epsilon_{x}(\epsilon_{x}(\psi) \cdot \epsilon_{x}(\phi) \cdot \mu \cdot \mu' \cdot \psi \cdot \phi \cdot \chi) 
\nonumber \\
&=&\epsilon_{x}(\psi) \cdot \epsilon_{x}(\phi) \cdot \epsilon_{x}(\cdot \mu \cdot \mu' \cdot \psi \cdot \phi \cdot \chi)
\end{eqnarray}
This proves that $\epsilon_{x}(\phi) \leq \epsilon_{x}(\psi)$. 

3.) By definition $x \leq y$ means that $\epsilon_x(\psi) = \epsilon_x(\epsilon_y(\psi))$. Then item 1 above shows that $\epsilon_{x}(\psi) \leq \epsilon_{y}(\psi)$.
\end{proof}

Further, we remark that the relation $\phi \leq_2 \psi$ if there is an idempotent $f$ such that $\psi = f \cdot \phi$ is a \textit{partial order}. Of course $\phi \leq_2 \psi$ implies $\phi \leq \psi$. This is the partial order studied in semigroup theory \cite{nambooripad80,mitsch86}, the goal there being to study the structure of semigroups. The condition $\psi = f \cdot \phi$ means in our context that $\psi$ is obtained by combination of $\phi$ with a deterministic information $f$. So $\psi$ results from a kind of conditioning of $\phi$ on $f$. We refer to \cite{kohlas03} for an illustration in the context of probability potentials. So, $\psi$ is, according to this order, more informative than $\phi$, if it is obtained by conditioning of $\phi$. Although this makes sense, this order does not seem very interesting from the point of view of valuation algebras. For example it does not follow that $\epsilon_{x}(\psi) \leq \psi$ or $\phi,\psi \leq \phi \cdot \psi$.  

Next let's turn to separative algebras $(\Phi,\cdot,1:E)$. Note first that we may carry over the order between the equivalence classes $[\psi]_{\gamma}$ to the groups $\gamma(\psi)$, since there is a one-to-one relation between classes and groups. Hence $\gamma(\phi) \leq \gamma(\psi)$ iff $[\phi]_{\gamma} \leq [\psi]_{\gamma}$. Then we deduce that
\begin{eqnarray}
\gamma(\phi \cdot \psi) = \gamma(\phi) \vee \gamma(\psi).
\nonumber
\end{eqnarray}
We define next the natural order (\ref{eq:Preorder}) in the semigroup $(\Phi^*;\cdot)$,
\begin{eqnarray}
a \leq b, \textrm{ iff there exists a}\ c \in \Phi^* \textrm{ such that}\ b = a \cdot c.
\end{eqnarray}
Note then that for elements of $\Phi$, this preorder $\phi \leq \psi$ admits that in $\psi = \phi \cdot c$, the factor which completes $\phi$ to $\psi$ does no more need to be an element of $\Phi$, but only of $\Phi^*$.

\begin{lemma} \label{le:OrderGroup}
In $\Phi^*$ we have $a \leq b$ iff $\gamma(a) \leq \gamma(b)$.
\end{lemma}

\begin{proof}
Assume first $a \leq b$, hence $a \cdot c = b$ for some $c \in \Phi^*$. Then $\gamma(b) = \gamma(a \cdot c) = \gamma(a) \vee \gamma(c)$, hence $\gamma(a) \leq \gamma(c)$. Conversely, assume $\gamma(a) \leq \gamma(b)$. Then $\gamma(b) = \gamma(a) \vee \gamma(b) = \gamma(a \cdot b)$. Therefore we see that $a \cdot b$ and $b$ belong both to the group $\gamma(b)$ and therefore $b = a \cdot b \cdot (a \cdot b)^{-1} \cdot b$, thus $a \leq b$.
\end{proof}

We remark that for any element $a$ of $\Phi^*$ we have $a = a \cdot a^{-1} \cdot a$. This means that the semigroup $(\Phi^*,\cdot)$ is \textit{regular}. And further $a \equiv_{\gamma} b$ implies $a \cdot \Phi^* = b \cdot \Phi^*$. In fact, if $d \in a \cdot \Phi^*$, then $d = a \cdot c$ for some $c \in\Phi^*$. It follows then $d = b \cdot b^{-1} \cdot a \cdot c$, hence $d \in b \cdot \Phi^*$. In the same way it follows that $d \in b \cdot \Phi^*$ implies $d \in a \cdot \Phi^*$, hence $a \cdot \Phi^* = b \cdot \Phi^*$. Conversely, if $a \cdot \Phi^* = b \cdot \Phi^*$, then $a = b \cdot c$ and $b = a \cdot c'$ for some $c,c' \in \Phi^*$. This means that $a \leq b$ and $b \leq a$, hence $\gamma(a) = \gamma(b)$, or $a \equiv_{\gamma} b$. This shows that the congruence $\equiv_{\gamma}$ is the Green relation in the regular semigroup $(\Phi^*,\cdot)$.

As a consequence of this remark and of Lemma \ref{le:OrderGroup} we have, as in the previous section (Lemma \ref{le:PreorderClasses}), the following result:

\begin{lemma} \label{le:PreorderClasses2}
Let $(\Phi,\cdot,1;E)$ be a separative valuation algebra embedded int $\Phi^*$. Then, for $a,b \in \Phi^*$,
\begin{enumerate}
\item $a \leq b$ iff $\gamma(a) \leq \gamma(b)$
\item $a \leq b$ and $b \leq a$ iff $\gamma(a) = \gamma(b)$.
\end{enumerate}
\end{lemma}

As in the case of regular valuation algebras, we have for separative information algebras the same results regarding order and extraction (see Theorem \ref{th:OrderAndExtr}).

\begin{theorem} \label{th:SepOrderExt}
Let $(\Phi,\cdot,1;E)$ be a separative valuation algebra. Then
\begin{enumerate}
\item $\epsilon_{x}(\psi) \leq \psi$ for all $x \in Q$ and $\psi \in \Phi$.
\item $\phi \leq \psi$ implies $\epsilon_{x}(\phi) \leq \epsilon_{x}(\psi)$ for all $x \in Q$.
\item $x \leq y$ implies $\epsilon_{x}(\psi) \leq \epsilon_{y}(\psi)$ for all $\psi \in \Phi$.
\end{enumerate}
\end{theorem}

\begin{proof}
1.) From (\ref{eq:SepCongr}) we obtain $\gamma(\epsilon_{x}(\psi) \cdot \psi) = \gamma(\epsilon_{x}(\psi)) \vee \gamma(\psi) = \gamma(\psi)$. This shows that $\gamma(\epsilon_{x}(\psi)) \leq \gamma(\psi)$, which implies $\epsilon_{x}(\psi) \leq \psi$ (Lemma \ref{le:PreorderClasses2}).

2.) From $\phi \leq \psi$ we obtain $\gamma(\phi) \leq \gamma(\psi)$ and from item 1 just proved $\gamma(\epsilon_{x}(\phi)) \leq \gamma(\phi)$. Thus we have $\gamma(\epsilon_{x}(\phi) \cdot \psi) = \gamma(\epsilon_{x}(\phi)) \vee \gamma(\psi) = \gamma(\psi)$. Further, we have $\epsilon_{x}(\epsilon_{x}(\phi) \cdot \psi) = \epsilon_{x}(\phi) \cdot \epsilon_{x}(\psi)$. Therefore, from the congruence of $\equiv_{\gamma}$, we conclude that $\gamma(\epsilon_{x}(\phi) \cdot \epsilon_{x}(\psi)) = \gamma(\epsilon_{x}(\psi))$, and this shows that $\epsilon_{x}(\phi) \leq \epsilon_{x}(\psi)$.

3.) This is proved exactly as item 3 of Theorem \ref{th:OrderAndExtr}.
\end{proof}

As in the regular case, we may define an order $\phi \leq_2 \psi$ if there is an idempotent $f$ such that $\psi = f \cdot \phi$ and again $\phi \leq_2 \psi$ implies $\phi \leq \psi$. This is as before a partial order, since $\phi \leq \psi$ and $\psi \leq \phi$ imply $\gamma(\phi) = \gamma(\psi)$ and $\psi = f_{\psi} \cdot \phi$. But $f_{\psi} = f_{\phi}$, hence $\psi = f_{\phi} \cdot \phi = \phi$. The expression $f_{\psi} \cdot \phi$ is again a kind of conditioning, namely the combination of a deterministic element $f_{\psi}$ with an information element $\phi$. We refer to \cite{kohlas03} for a discussion of the separative valuation algebra of probability densities, which illustrates these statements. Again, it makes sense that an information $\psi$ obtained from another one by condition $\psi = f_{\psi} \cdot \phi$, where $f_{\phi} \leq f_{\psi}$ is considered to be more informative. At least in probability theory this seems evident.


\section{Regular conditionals} \label{subsec:Continuation}

In this section, we introduce a concept, \textit{conditionals}, which is motivated by the concept of (discrete) conditional distributions in probability theory. It turns out that this concept, both in regular and separative algebras, share many properties with conditional probability distributions. So, this sheds some light on this concept from an information theoretic point of view. In probability theory, if $p(x,y)$ is a (discrete) probability distribution, then 
\begin{eqnarray*}
p(x \vert y) = \frac{p(x,y)}{\sum_x p(x,y)}
\end{eqnarray*}
is called the conditional probability distribution of $x$ given $y$. This involves, from an algebraic point of view, the division of the probability distribution $p$ with a marginal distribution of it, or the multiplication of $p$ with the inverse of its marginal. 

Now, multiplication corresponds to combination in the valuation algebra of probability potentials, and marginalization to extraction, see Example \ref{expl:ProbPot}. This consideration motivates the following definition.

\begin{definition} \textbf{Conditional in a regular valuation algebra.} \label{def:CondIndepRelVal}
Let $(\Phi,\cdot,1;E)$ with $E = \{\epsilon_x:x \in Q\}$ be a regular valuation algebra, $\phi \in \Phi$, $x,y \in Q$. Then 
\begin{eqnarray*}
\phi_{x \vert y} = \epsilon_{x \vee y}(\phi) \cdot (\epsilon_y(\phi))^{-1}
\end{eqnarray*}
is called the conditional of $\phi$ for $x$ given $y$..
\end{definition}
The conditional $\phi_{x \vert y}$ is well defined and $(\epsilon_y(\phi))^{-1}$ is the inverse of $\epsilon_y(\phi)$ in the the group of the equivalence class $[\epsilon_y(\phi)]_\gamma$ of the Green relation. 

For the study of this concept, we need some preparatory results.

\begin{lemma} \label{le:OrderOfClasses} \
\begin{enumerate}
\item $[\phi]_\gamma \leq [\psi]_\gamma$ implies $[\epsilon_x(\phi)]_\gamma \leq [\epsilon_x(\psi)]_\gamma$,
\item $[\epsilon_x(\phi)]_\gamma \leq [\phi]_\gamma$.
\end{enumerate}
\end{lemma}

\begin{proof}
1.) By Lemma \ref{le:PreorderClasses} we have $\phi \leq \psi$ if and only if $[\phi]_\gamma \leq [\psi]_\gamma$ and $\phi \leq \psi$ implies $\epsilon_x(\phi) \leq \epsilon_x(\psi)$, Theorem \ref{th:OrderAndExtr}. This implies item 1 as well as item 2. 
\end{proof}

Note that these order results among equivalence classes of the Green relation induce the same order results for the units of the groups. A further result is needed.

\begin{lemma} \label{le:SuppOfClasses} \
\begin{enumerate}
\item If $x$ is a support of $\phi$, then $x$ is a support for all elements $\psi \in [\phi]_\gamma$,
\item $[\psi]_\gamma \leq [\phi]_\gamma$ implies $\phi \cdot f_\psi = \phi$.
\end{enumerate}
\end{lemma}

\begin{proof}
1.) We have, if $x$ is a support of $\phi$, $\epsilon_x(\phi) = \phi \equiv_\gamma \psi \in [\phi]_\gamma$. The Green relation is a congruence also relative to extraction, so $\phi \equiv_\gamma \psi$ implies $\epsilon_x(\phi) \equiv_\gamma \epsilon_x(\psi)$, hence, by transitivity $\epsilon_x(\psi) \equiv_\gamma \psi$. Thus, since $\epsilon_x(\psi) \leq \psi$, we conclude that $\epsilon_x(\psi) = \psi$, that is, $x$ is a support of $\psi$.

2.) The assumption that $[\psi]_\gamma \leq [\phi]_\gamma$ implies $f_\psi \leq f_\phi$ and therefore $\phi \cdot f_\psi = \phi \cdot f_\phi \cdot f_\psi = \phi \cdot f_\phi = \phi$.
\end{proof}

We remark now that the element $\epsilon_{x \vee y}(\phi)$ can be reconstructed if the conditional $\phi_{x \vert y}$ and the extraction $\epsilon_y(\phi)$ is known. This can be deduced, if both sides of the defining equation of a conditional is combined with $\epsilon_y(\phi)$, using Lemma \ref{le:OrderOfClasses} and \ref{le:SuppOfClasses} and noting that $[\epsilon_{x \vee y}(\phi)]_\gamma \geq [\epsilon_y(\phi)]_\gamma$,
\begin{eqnarray*}
\epsilon_{x \vee y}(\phi) \cdot f_{\epsilon_y(\phi)} = \epsilon_{x \vee y}(\phi) = \phi_{x \vert y} \cdot \epsilon_y(\phi).
\end{eqnarray*}
An element $\chi \in \Phi$ such that $\epsilon_{x \vee y}(\phi) = \chi \cdot \epsilon_y(\phi)$ is called in \cite{shafer96} a \textit{continuation} for $\phi$ from $y$ to $x \vee y$. So, the conditional $\phi_{x \vert y}$ is such a continuation. A continuation is in general not unique. However, consider elements $\phi \in [1]_\gamma$, such that $\phi \cdot \phi^{-1} = 1$. Such elements are called \textit{positive}. Then, from $\epsilon_{x \vee y} = \chi \cdot \epsilon_y(\phi)$ we obtain $\phi_{x \vert y} = \epsilon_{x \vee y} \cdot (\epsilon_y(\phi))^{-1} = \chi \cdot f_{\epsilon_y(\phi)} = \chi$ since $\epsilon_y(\phi) \leq \phi$ implies $[\epsilon_y(\phi)]_\gamma \leq [\phi]_\gamma$ and so $f_{\epsilon_y(\phi)} \leq f_\phi = 1$, see Lemma \ref{le:SuppOfClasses} 

We need one further result.

\begin{lemma} \label{le:CondOrder}
For all $\phi \in \Phi$ and $x,y \in Q$, 
\begin{eqnarray*}
[\phi_{x \vert y}]_\gamma \geq [\epsilon_y(\phi)]_\gamma.
\end{eqnarray*}
\end{lemma}

\begin{proof}
By definition $[\phi_{x \vert y}]_\gamma = [\epsilon_{x \vee y}(\phi)  \cdot (\epsilon_y(\phi))^{-1}]_\gamma = [\epsilon_{x \vee y}(\phi)]_\gamma \vee   [(\epsilon_y(\phi))^{-1}]_\gamma$ and so $[\phi_{x \vert y}]_\gamma \geq  [(\epsilon_y(\phi))^{-1}]_\gamma = [\epsilon_y(\phi)]_\gamma$.
\end{proof}

So far we have not exploited the important concept of conditional independence among questions (Section \ref{sec:StructOfQuest}). Now, we extend this concept, motivated by stochastic conditional independence of random variables in probability theory, to a similar concept related to information.

\begin{definition} \label{def.CondIndepInf} \textbf{Conditional independence relative to a valuation.}
We call $x,y \in Q$ conditionally independent given $z \in Q$ relative to $\phi \in \Phi$, if
\begin{enumerate}
\item $x \bot y \vert z$.
\item $\epsilon_{x \vee y \vee z}(\phi) = \psi_1 \cdot \psi_2$, where $\psi_1$ and $\psi_2$ have supports $x \vee z$ and $y \vee z$ respectively.
\end{enumerate}
We then write $x \bot_\phi y \vert z$.
\end{definition}

We shall see below that this corresponds in the example of probability potentials to stochastic conditional independence, see also \cite{kohlas03}. As in this case the concept is closely related to factorizations of information over conditionally independent domains or questions. This is fundamental for local computation procedures, not only for idempotent information algebra, as discussed in Section \ref{sec:LocComp}, but also for valuation algebras, \cite{kohlas03}. The next proposition clarifies this.

\begin{proposition} \label{prop:CondIndepInf}
Assume $x \bot_\phi y \vert z$. Then, if $\epsilon_{x \vee y \vee z}(\phi) = \psi_1 \cdot \psi_2$, where $x \vee z$ and $y \vee z$ are supports of $\psi_1$ and $\psi_2 $ respectively, 
\begin{enumerate}
\item $\epsilon_{x \vee z}(\phi) = \psi_1 \cdot \epsilon_z(\psi_2)$ and $\epsilon_{y \vee z}(\phi) = \psi_2 \cdot \epsilon_z(\psi_1)$.
\item $\epsilon_z(\phi) = \epsilon_z(\psi_1) \cdot \epsilon_z(\psi_2)$.
\end{enumerate}
\end{proposition} 

\begin{proof}
1.) From $x \bot_\phi y \vert z$ we have
\begin{eqnarray*}
\epsilon_{x \vee z}(\phi) = \psi_1 \cdot \epsilon_{x \vee z}(\psi_2).  
\end{eqnarray*}
And from $x \vee z \bot y \vee z \vert z$ and that $y \vee z$ is a support of $\psi_2$ we obtain $\epsilon_{x \vee z}(\psi_2) = \epsilon_{x \vee z}(\epsilon_z(\psi_2))$ and since $z \leq x \vee z$ we have $\epsilon_{x \vee z}(\epsilon_z(\psi_2)) = \epsilon_z(\psi_2)$. This proves the first identity in 1.), the second follows similarly.

2.) From 1.) we have $\epsilon_z(\phi) = \epsilon_z(\epsilon_{x \vee z}(\phi)) = \epsilon_z(\psi_1 \cdot \epsilon_z(\psi_2)) = \epsilon_z(\psi_1) \cdot \epsilon_z(\psi_2)$.
\end{proof}

In this section, we shall discuss conditional independence in relation to conditionals and show that results as in stochastic conditional independence and conditional probability distributions hold. More on conditional independence can be found in Section \ref{sec:CondIndep}. 

Here follow a few preliminary results on conditionals and conditional independence.

\begin{proposition} \label{prop:PreliminarRes} \
\begin{enumerate}
\item $\epsilon_y(\phi_{x \vert y}) = f_{\epsilon_y(\phi)}$,
\item $\phi_{x \vee y \vert z} = \phi_{x \vert y \vee z} \cdot \phi_{y \vert z}$,
\item if $z \leq x$, then $\epsilon_{y \vee z}(\phi_{x \vert y}) = \phi_{z \vert y}$,
\item $\epsilon_{y \vee z}(\phi_{z \vert x \vee y} \cdot \phi_{x \vert y}) = \phi_{z \vert y}$,
\item if $y$ is a support of $\psi$, then $(\epsilon_{x \vee y}(\phi) \cdot \psi)_{x \vert y} = \phi_{x \vert y} \cdot f_\psi$.
\end{enumerate}
\end{proposition}

\begin{proof}
1.) By definition we have $\epsilon_y(\phi_{x \vert y}) = \epsilon_y(\epsilon_{ x \vee y}(\phi) \cdot (\epsilon_y(\phi))^{-1}) = \epsilon_y(\epsilon_{x \vee y}(\phi)) \cdot (\epsilon_y(\phi))^{-1} = \epsilon_y(\phi) \cdot (\epsilon_y(\phi))^{-1} = f_{\epsilon_y(\phi)}$.

2.) Again, by definition, $\phi_{x \vee y \vert z} = \epsilon_{x \vee y \vee z}(\phi) \cdot (\epsilon_z(\phi))^{-1}$ and $\phi_{x \vert y \vee z} \cdot \phi_{z \vert y} = \epsilon_{x \vee y \vee z}(\phi) \cdot (\epsilon_{y \vee z}(\phi))^{-1} \cdot  \epsilon_{y \vee z}(\phi) \cdot (\epsilon_z(\phi))^{-1} =  \epsilon_{x \vee y \vee z}(\phi) \cdot  (\epsilon_z(\phi))^{-1} \cdot f_{\epsilon_{y \vee z}(\phi))} =  \epsilon_{x \vee y \vee z}(\phi) \cdot  (\epsilon_z(\phi))^{-1}$. This proves the identity claimed.

3.) We have $\epsilon_{y \vee z}(\phi_{ x \vert y}) = \epsilon_{y \vee z}(\epsilon_{x \vee y}(\phi) \cdot (\epsilon_y(\phi))^{-1}$. We introduce now the following lemma, which will also be used later.

\begin{lemma} \label{le:ExtCombRs}
$y \leq z \leq x$ implies $\epsilon_z(\epsilon_x(\phi) \cdot \epsilon_y(\psi)) = \epsilon_z(\phi) \cdot \epsilon_y(\psi)$. 
\end{lemma}

\begin{proof}
If $y \leq z$ then $x \bot z \vert z$ implies $x \bot y \vert z$. Then (Theorem \ref{th:CombExtrProp}) it follows that $\epsilon_z(\epsilon_x(\phi) \cdot \epsilon_y(\psi)) = \epsilon_z(\epsilon_x(\phi)) \cdot \epsilon_z(\epsilon_y(\psi))$. Now, since $y \leq z$, $z$ is a support of $\epsilon_y(\psi)$, so that $\epsilon_z(\epsilon_y(\psi)) = \epsilon_y(\psi)$ and since $z \leq x$ we have also $\epsilon_z(\epsilon_x(\phi)) = \epsilon_z(\phi)$. This proves the identity. 
\end{proof}

If we apply the identity of the lemma, we get $\epsilon_{y \vee z}(\phi_{ x \vert y}) = \epsilon_{y \vee z}(\phi) \cdot (\epsilon_y(\phi))^{-1} = \phi_{z \vert y}$. Here we use the fact that the inverse of $\epsilon_y(\phi)$ has also support $y$.

4.) By item 2 above, $\phi_{z \vert x \vee y} \cdot \phi_{x \vert y} = \phi_{x \vee z \vert y}$. Then, by item 3, $\epsilon_{y \vee z}(\phi_{z \vert x \vee y} \cdot \phi_{x \vert y}) = \epsilon_{y \vee z}(\phi_{x \vee z \vert y}) = \phi_{z \vert y}$ since $z \leq x \vee z$.

5.) We have $\epsilon_{x \vee y}(\phi) \cdot \psi = \phi_{x \vert y} \cdot \epsilon_y(\phi) \cdot \psi$. On the other hand we have also $\epsilon_{x \vee y}(\phi) \cdot \psi = (\epsilon_{x \vee y}(\phi) \cdot \psi)_{x \vert y} \cdot \epsilon_y(\epsilon_{x \vee y}(\phi) \cdot \psi) = (\epsilon_{x \vee y}(\phi) \cdot \psi)_{x \vert y}  \cdot \epsilon_y(\phi) \cdot \psi$. From this we conclude that $\phi_{x \vert y} \cdot f_{\epsilon_y(\phi) \cdot \psi} = (\epsilon_{x \vee y}(\phi) \cdot \psi)_{x \vert y} \cdot f_{\epsilon_y(\phi) \cdot \psi}$. But $[(\epsilon_{x \vee y}(\phi) \cdot \psi)_{x \vert y}]_\gamma = [\epsilon_{x \vee y}(\phi) \cdot \psi]_\gamma \vee [\epsilon_y \cdot \psi]_\gamma \geq [\epsilon_y \cdot \psi]_\gamma$. Thus $(\epsilon_{x \vee y}(\phi) \cdot \psi)_{x \vert y} = \phi_{x \vert y} \cdot f_{\epsilon_y(\phi)} \cdot f_\psi = \phi_{x \vert y} \cdot f_\psi$, by Lemma \ref{le:CondOrder}.
\end{proof}

Here follows the main theorem about regular conditionals, establishing a parallelism to stochastic conditional independence.

\begin{theorem} \label{th:EquivCondIndepRel}
Assume $x \bot y \vert z$. The the following statements are all equivalent.
\begin{enumerate}
\item $x \bot_\phi y \vert z$, that is $x$ and $y$ are conditionally independent given $z$ relativee to $\phi \in \Phi$.
\item $\epsilon_{x \vee y \vee z}(\phi) = \phi_{x \vert z} \cdot \phi_{y \vert z} \cdot \epsilon_z(\phi)$.
\item $\phi_{x \vee y \vert z} = \phi_{x \vert z} \cdot \phi_{y \vert z}$.
\item $\phi_{x \vee y \vert z} = \chi_1 \cdot \chi_2$, where $\chi_1$ and $\chi_2$ have supports $x \vee z$ and $y \vee z$ respectively.
\item $\epsilon_{x \vee y \vee z}(\phi) \cdot \epsilon_z(\phi) = \epsilon_{x \vee z}(\phi) \cdot \epsilon_{y \vee z}(\phi)$.
\item $\epsilon_{x \vee y \vee z}(\phi) = \phi_{x \vert z} \cdot \epsilon_{y \vee z}(\phi)$.
\item $\phi_{x \vert y \vee z} = \phi_{x \vert z} \cdot f_{\epsilon_{y \vee z}(\phi)}$.
\item $\phi_{x \vert y \vee z} = \chi \cdot f_{\epsilon_{y \vee z}(\phi)}$, where $\chi$ has support $x \vee z$.
\end{enumerate}
\end{theorem}

\begin{proof}
(1) $\Rightarrow$ (2): By (1) and Proposition \ref{prop:CondIndepInf}, we have $\epsilon_z(\phi) = \epsilon_z(\psi_1) \cdot \epsilon_z(\psi_2)$. Further, $\psi_1 = \psi_{1,x \vert z} \cdot \epsilon_z(\psi_1)$ and $\psi_2 = \psi_{2,y \vert z} \cdot \epsilon_z(\psi_2)$. It follows that $\epsilon_{x \vee y \vee z} = \psi_1 \cdot \psi_2 = \psi_{1,x \vert z} \cdot \psi_{2,y \vert z} \cdot \epsilon_z(\psi_1) \cdot \epsilon_z(\psi_2) = \psi_{1,x \vert z} \cdot \psi_{2,y \vert z} \cdot \epsilon_z(\phi)$. Again by Proposition \ref{prop:CondIndepInf}, $\epsilon_{x \vee z}(\phi) = \psi_1 \cdot \epsilon_z(\psi_2) = \psi_{1,x \vert z} \cdot \epsilon_z(\psi_1) \cdot \epsilon_z(\psi_2) = \psi_{1,x \vert z} \cdot \epsilon_z(\phi)$ and $\epsilon_{y \vee z}(\phi) = \psi_2 \cdot \epsilon_z(\psi_1) = \psi_{2,y \vert z} \cdot \epsilon_z(\psi_1) \cdot \epsilon_z(\psi_2) = \psi_{2,y \vert z} \cdot \epsilon_z(\phi)$. This leads to the equations $\epsilon_{x \vee z}(\phi) = \phi_{x \vert z} \cdot \epsilon_z(\phi) = \psi_{1,x \vert z} \cdot \epsilon_z(\phi)$ and $\epsilon_{y \vee z}(\phi) = \phi_{y \vert z} \cdot \epsilon_z(\phi) = \psi_{2,y \vert z} \cdot \epsilon_z(\phi)$, thus $\phi_{x \vert z} = \psi_{1,x \vert z} \cdot f_{\epsilon_z(\phi)}$ and $\phi_{y \vert z} = \psi_{2,y \vert z} \cdot f_{\epsilon_z(\phi)}$, and then finally $\epsilon_{x \vee y \vee z} = (\psi_{1,x \vert z} \cdot f_{\epsilon_z(\phi)}) \cdot (\psi_{2,y \vert z} \cdot f_{\epsilon_z(\phi)})\cdot \epsilon_z(\phi) =  \phi_{x \vert z} \cdot \phi_{y \vert z} \cdot \epsilon_z(\phi)$.

(2) $\Rightarrow$ (3): We have $\epsilon_{x \vee y \vee z}(\phi) = \phi_{x \vee y \vert z} \cdot \epsilon_z(\phi) = \phi_{x \vert z} \cdot \phi_{y \vert z} \cdot \epsilon_z(\phi)$. This implies $ \phi_{x \vee y \vert z} = \phi_{x \vert z} \cdot \phi_{y \vert z}$, since $f_{\epsilon_z} \leq f_{\phi_{x \vert z}}, f_{\phi_{x \vee y \vert z}}$, Lemma \ref{le:CondOrder}.

(3) $\Rightarrow$ (4): Take $\chi_1 = \phi_{x \vert z}$ and $\chi_2 = \phi_{y \vert z}$.

(4) $\Rightarrow$ (5): From (4), $\epsilon_{x \vee y \vee z}(\phi) \cdot \epsilon_z(\phi) = \phi_{x \vee y \vert z} \cdot \epsilon_z(\phi) = (\chi_1 \cdot \epsilon_z(\phi)) \cdot (\chi_2 \cdot \epsilon_z(\phi))$. Further, using $x \vee y \bot y \vee z \vert z$,
\begin{eqnarray*}
\epsilon_{x \vee z}(\phi) &=& \epsilon_{x \vee z}(\epsilon_{x \vee y \vee z}(\phi)) = \epsilon_{x \vee z}(\phi_{x \vee y \vert z} \cdot \epsilon_z(\phi))  \\
&=& \epsilon_{x \vee z}(\chi_1 \cdot \chi_2 \cdot \epsilon_z(\phi)) = \chi_1 \cdot \epsilon_{x \vee z}(\chi_2 \cdot \epsilon_z(\phi)) \\
&=& \chi_1 \cdot \epsilon_{x \vee z}(\epsilon_z(\chi_2 \cdot \epsilon_z(\phi)) = \chi_1 \cdot \epsilon_{x \vee z}(\epsilon_z(\chi_2) \cdot \epsilon_z(\phi)) \\
&=& \chi_1 \cdot \epsilon_z(\chi_2) \cdot \epsilon_z(\phi),
\end{eqnarray*}
since $z$, hence $x \vee z$ is a support of $\epsilon_z(\chi_2) \cdot \epsilon_z(\phi)$. In the same way we obtain $\epsilon_{y \vee z} = \chi_2 \cdot \epsilon_z(\chi_1) \cdot \epsilon_z(\phi)$. By Propositions \ref{prop:CondIndepInf} and \ref{prop:PreliminarRes}, $\epsilon_z(\phi_{x \vee y \vert z}) = \epsilon_z(\chi_1 \cdot \chi_2) = \epsilon_z(\chi_1) \cdot \epsilon_z(\chi_2) = f_{\epsilon_z(\phi)}$. This gives us finally
\begin{eqnarray*}
\epsilon_{x \vee z}(\phi) \cdot \epsilon_{y \vee z}(\phi) &=&  (\chi_1 \cdot \epsilon_z(\phi)) \cdot (\chi_2 \cdot \epsilon_z(\phi)) \cdot \epsilon_z(\chi_1) \cdot \epsilon_z(\chi_2) \\
&=& (\chi_1 \cdot \epsilon_z(\phi)) \cdot (\chi_2 \cdot \epsilon_z(\phi)) \cdot f_{\epsilon_z(\phi)} \\
&=& (\chi_1 \cdot \epsilon_z(\phi)) \cdot (\chi_2 \cdot \epsilon_z(\phi)) = \epsilon_{x \vee y \vee z}(\phi) \cdot \epsilon_z(\phi).
\end{eqnarray*}

(5) $\Rightarrow$ (6): By (5) $\epsilon_{x \vee y \vee z}(\phi) \cdot \epsilon_z(\phi) = \epsilon_{x \vee z}(\phi) \cdot \epsilon_{y \vee z}(\phi) = \phi_{x \vert z} \cdot \epsilon_z(\phi) \cdot \epsilon_{y \vee z}(\phi)$. Combining both sides with the inverse of $\epsilon_z(\phi)$ we obtain $\epsilon_{x \vee y \vee z}(\phi) = \phi_{x \vert z} \cdot \epsilon_{y \vee z}(\phi)$ since $f_{\epsilon_z(\phi)}$ is absorbed on both sides.

(6) $\Rightarrow$ (7): On the one hand we have $\epsilon_{x \vee y \vee z}(\phi) = \phi_{x \vert y \vee z} \cdot \epsilon_{y \vee z}(\phi)$ and on the other hand, by (6), $\epsilon_{x \vee y \vee z}(\phi) =  \phi_{x \vert y} \cdot \epsilon_{y \vee z}$. From this we obtain $\phi_{x \vert y \vee z} = \phi_{x \vert y} \cdot f_{\epsilon_{y \vee z}(\phi)}$. 

(7) $\Rightarrow$ (8): Take $\chi = \phi_{x \vert y}$. 

(8) $\Rightarrow$ (1): Here we have $\epsilon_{x \vee y \vee z}(\phi) = \phi_{x \vert y \vee z} \cdot \epsilon_{y \vee z}(\phi) = \chi \cdot f_{\epsilon_{y \vee z}(\phi)} \cdot \epsilon_{y \vee z}(\phi)$. Take then $\psi_1 = \chi$ and $\psi_2 = \epsilon_{y \vee z}(\phi)$.

This concludes the proof.
\end{proof}

In the trivial case of an idempotent information algebra, algebra most items of this theorem collapse to the unique statement that $x \bot_\phi y \vert z$ is equivalent to $\epsilon_{x \vee y \vee z}(\phi) = \epsilon_{x \vee z}(\phi) \cdot \epsilon_{y \vee z}(\phi)$. In fact, items 2,3,5,6 and 7 reduce to this formula, since a conditional $\phi_{x \vert y}$ equals simply $\epsilon_{x \vee y}(\phi)$.

Note that item 4 of this theorem states that $x \bot_\phi y \vert z$ if and only if $x \bot_{\phi_{x \vee y \vert z}}  y \vert z$. In the following example we compare this result with the classical case of stochastic conditional independence in the valuation algebra of probability potentials.

\begin{example} \textbf{Conditional independence among probability potentials.}  \label{expl:ProbCondIndep}
We refer to the example of probability potentials. They form a labeled valuation algebra on multivariate system. The definition of conditionals carries in an obvious way over to labeled algebras. So, let $s$ and $t$ be disjoint sets of variables and and $x \in U_s$ and $y \in U_t$. For a probability potential $(x,y)$ on domain $U_{s \cup t}$ the conditional $p_{s \vert t}$ is defined as follows
\begin{eqnarray*}
p_{s \vert t}(x,y) = \left\{ \begin{array}{ll} \frac{p(x,y)}{\pi_t(p)(y)} & \textrm{if}\ \pi_t(y) > 0, \\ 0 & \textrm{otherwise}. \end{array} \right.
\end{eqnarray*}
If the potential $p$ is a probability distribution, then clearly this is the usual definition of a conditional probability distribution. Of course, it is an arbitrary definition to put $p_{s \vert t}(x,y) =$ if $\pi_t(p)(y) = 0$. In the conditional distribution $p_{s \vert t}$ is simply not defined in this case. This illustrates the fact, that there are many continuations, if $p$ is not strictly positive on every tuple $(x,y)$.

We write $p_{s \vert t}(x,y) = p(x \vert y)$ in favour of a notation which is more usual in probability theory. Let now $s,t,u$ be three disjoint families of variables such that $s \cup u \bot t \cup u \vert u$. Then, if $p$ is a probability potential $p(x,y,z)$, on the domain of the set of variables $s \cup t \cup u$, we have that $s$ and $t$ are conditionally independent given $u$, if there are probability potentials $q_1$ and $q_2$ on the sets $s \cup u$ and $t \cup u$ of variables such that $p(x,y,z) = q_1(x,z)q_2(y,z)$. We write then $s \bot_p t \vert u$. Theorem \ref{th:EquivCondIndepRel} given then the following equivalent conditions,
\begin{enumerate}
\item $s \bot_p t \vert u$,
\item $p(x,y,z) = p(x \vert z)p(y \vert z)\pi_u(p)(z)$,
\item $p(x,y \vert z) = p(x \vert z)p(y \vert z)$,
\item $p(x,y \vert z) = p_1(x,z)p_2(y,z)$,
\item $p(x,y,z) \cdot \pi_u(p)(z) = \pi_{s \cup u}(p)(x,z)\pi_{t \cup u}(p)(y,z)$,
\item $p(x,y,z) = p(x \vert y)\pi_{t \cup u}(p)(y,z)$,
\item $p(x \vert y,z) = p(x \vert z)f_{\pi_{t \cup u}(p)}(y,z)$,
\item $p(x \vert y,z) = q(x,z)f_{\pi_{t \cup u}(p)}(y,z)$.
\end{enumerate}
Here $f_{\pi_{t \cup u}(p)}$ is the indicator function of the support of the marginal $\pi_{t \cup u}(p)$, that is $f_{\pi_{t \cup u}(p)}(y,z) = 1$ if $\pi_{t \cup u}(p)(y,z) > 0$ and $f_{\pi_{t \cup u}(p)}(y,z) = 0$ otherwise. These are all well-known properties of conditional probability distributions over discrete domains. And this illustrates how conditionals in regular valuation algebras generalize this concept. 
\end{example}

We refer to another interesting view of regular conditionals in the domain of dynamic programming, see \cite{kohlas03}.


\section{Separative conditionals} \label{subsec:SepContinuation}

In a separative valuation algebra, we have still a notion of inverse or division which allows the definition of conditionals similar to regular valuation algebras. But how far do these conditionals share the same properties as those in regular algebras? This is the question addressed in this section. So, let $(\Phi,\cdot,1;E)$ with $E = \{\epsilon_x:x \in Q\}$ be a separative valuation algebra. We recall from Section \ref{subsec:SeparativeAlg} that $\Phi$ is embedded as a semigroup into a semigroup $\Phi^*$, which is a union of disjoint commutative groups $\gamma(\phi)$, where $\gamma(\phi)$ are equivalence classes $[\phi,\psi]$ of pairs $(\phi,\psi)$ of elements of $\Phi$. The semigroup $(\Phi,\cdot)$ is embedded into $\Phi^*$ by the map $\phi \mapsto [\phi \cdot \phi,\phi]$. As in   Section \ref{subsec:SeparativeAlg} we identify $\Phi$ with its image in $\Phi^*$, and consider $\Phi$ a subset of $\Phi^*$. That is we write $\phi$ for $[\phi \cdot \phi,\phi]$ and $\phi^{-1}$ for the inverse element $[\phi,\phi \cdot \phi]$. The unit element in the group $\gamma(\phi)$ is denoted by $f_\phi$.

The concept of a conditional in a separative valuation algebra $\Phi$ can be defined exactly as in the case of a regular one.

\begin{definition} \textbf{Conditional in a separative valuation algebra.}
Let $(\Phi,\cdot,1;E)$ with $E = \{\epsilon_x:x \in Q\}$ be a separative valuation algebra, $\phi \in \Phi$, $x,y \in Q$. Then 
\begin{eqnarray*}
\phi_{x \vert y} = \epsilon_{x \vee y}(\phi) \cdot (\epsilon_y(\phi))^{-1}
\end{eqnarray*}
is called the conditional of $\phi$ for $x$ given $y$..
\end{definition}
In contrast to the case of regular valuation algebras, in a separative algebra, a conditional is not necessarily an element of $\Phi$, but only of $\Phi^*$, for an illustration we refer to probability densities (see example in Section \ref{subsec:SeparativeAlg}), another example is given by set potentials in \cite{kohlas03}. As a consequence, in a separative valuation algebra, extraction is in principle no more defined. We shall however see below that we may still introduce this operation at least partially.

On the other hand, a conditional is still a continuation. In fact, from the definition of a conditional, we have as in the regular case, $\phi_{x \vert y} \cdot \epsilon_y(\phi) = \epsilon_{x \vee y \vee z}(\phi) \cdot f_{\epsilon_z(\phi)}$ But Lemma \ref{le:OrderOfClasses}, \ref{le:SuppOfClasses} and \ref{le:CondOrder} hold obviously also in a separative algebra. Therefore $\epsilon_{x \vee y \vee z}(\phi) \cdot f_{\epsilon_z(\phi)} = \epsilon_{x \vee y \vee z}(\phi)$. And again, in general the conditional is not the only possible continuation. As in the case of regular algebra, we call a valuation $\phi \in \Phi$ \textit{positive}, if $\phi \in [1]_\gamma$. And as in regular algebras, we verify that for a positive element $\phi$ the conditional $\phi_{x \vert y}$ is the unique continuation of $\phi$ from $y$ to $x \vee y$.

This can be illustrated by the example of probability densities, see the example below.

The question is, whether Theorem \ref{th:EquivCondIndepRel} carries over to separative conditionals. It turns out that this is not the case in general. This is a consequence of the fact that conditionals do not belong to $\Phi$. There is however a weaker form of conditional independence relative to a valuation $\phi$.

\begin{definition} \textbf{Weak independence relative to a valuation.} \label{def:WeakCondIndep}
We call $x,y \in Q$ weakly conditionally independent given $z \in Q$ relative to $\phi \in \Phi$, if
\begin{enumerate}
\item $x \bot y \vert z$,
\item $\epsilon_{x \vee y \vee z}(\phi) = \phi_{x \vert z} \cdot \phi_{y \vert z} \cdot \epsilon_z(\phi)$.
\end{enumerate}
We then write $x \amalg_\phi y \vert z$.
\end{definition}

Of course, the relation $x \bot_\phi y \vert z$ is still defined as before and conditional independence implies weak conditional independence.

\begin{proposition} \label{prop:StrongWeak}
If $\Phi$ is a separative valuation algebra, then $x \bot_\phi y \vert z$ implies $x \amalg_\phi y \vert z$.
\end{proposition}

\begin{proof}
This is proved just as (1) $\Rightarrow$ (2) in the proof of Theorem \ref{th:EquivCondIndepRel}.
\end{proof}

That the converse does not hold in general is shown in the example of set potentials, see \cite{kohlas03}. As stated in Theorem \ref{th:EquivCondIndepRel} it is valid in regular algebras, but also for instance in the example of densities. Below we give a sufficient condition for the equivalence of these two concepts.

But first we state the equivalent to Theorem \ref{th:EquivCondIndepRel} for separative valuation algebras

\begin{theorem} \label{th:WeakEquivCindIndep}
Assume $x \bot y \vert z$. The following statements are all equivalent.
\begin{enumerate}
\item $x \amalg_\phi y \vert z$.
\item $\phi_{x \vee y \vert z} = \phi_{x \vert z} \cdot \phi_{y \vert z}$.
\item $\epsilon_{x \vee y \vee z}(\phi) \cdot \epsilon_z(\phi) = \epsilon_{x \vee z}(\phi) \cdot \epsilon_{y \vee z}(\phi)$.
\item $\epsilon_{x \vee y \vee z}(\phi) = \phi_{x \vert z} \cdot \epsilon_{y \vee z}(\phi)$.
\item $\phi_{x \vert y \vee z} = \phi_{x \vert z} \cdot f_{\epsilon_{y \vee z}(\phi)}$.
\end{enumerate}
\end{theorem}

\begin{proof}
(1) $\Rightarrow$ (2) is proved just as in Theorem \ref{th:EquivCondIndepRel}.

(2) $\Rightarrow$ (3) Since $\phi_{x \vee y \vert z}$ is a continuation, we have
\begin{eqnarray*}
\epsilon_{x \vee y \vee z}(\phi) = \phi_{x \vee y \vert z} \cdot \epsilon_z(\phi).
\end{eqnarray*}
Therefore, using (2),
\begin{eqnarray*}
\epsilon_{x \vee y \vee z}(\phi) \cdot \epsilon_z(\phi) &=& \phi_{x \vert z} \cdot \phi_{y \vert z} \cdot \epsilon_z(\phi) \cdot \epsilon_z(\phi) \\
&=& (\phi_{x \vert z} \cdot \epsilon_z(\phi)) \cdot (\phi_{y \vert z} \cdot \epsilon_z(\phi)) \\
&=& \epsilon_{x \vee z}(\phi) \cdot \epsilon_{y \vee z}(\phi)
\end{eqnarray*}

(3) $\Rightarrow$ (4) is proved like (5) $\Rightarrow$ (6) in Theorem \ref{th:EquivCondIndepRel}

(4) $\Rightarrow$ (5) Again, since $\phi_{x \vert y \vee z}$ is a continuation, 
\begin{eqnarray}
\epsilon_{x \vee y \vee z}(\phi) = \phi_{x \vert y \vee z} \cdot \epsilon_{y \vee z}(\phi).
\end{eqnarray}
Therefore, by (4) we have the equation
\begin{eqnarray*}
\phi_{x \vert y \vee z} \cdot \epsilon_{y \vee z}(\phi) = \phi_{x \vert z} \cdot \epsilon_{y \vee z}(\phi).
\end{eqnarray*}
Multiplying both sides with the inverse of $\epsilon_{y \vee z}(\phi)$ we obtain (5).

(5) $\Rightarrow$ (1) Using (5), we have
\begin{eqnarray*}
\epsilon_{x \vee y \vee z}(\phi) &=& \phi_{x \vert y \vee z} \cdot \epsilon_{y \vee z}(\phi) = \phi_{x \vert z} \cdot f_{\epsilon_{y \vee z}(\phi)} \cdot \epsilon_{y \vee z}(\phi) \\
&=& \phi_{x \vert z} \cdot \epsilon_{y \vee z}(\phi) = \phi_{x \vert z} \cdot \phi_{y \vert z} \cdot \epsilon_z(\phi)
\end{eqnarray*}
and this means $x \amalg_\phi y \vert z$.

This concludes the proof.
\end{proof}

\begin{example} \textbf{Exztaczion among conditional probability desities.} \label{expl:SepCondIndep}
Consider a density $f$ on a domain $s$, see the example in Section \ref{subsec:SeparativeAlg}. Then, if $t \subseteq s$, we have the conditional $f_{s \vert t} = f \cdot (\pi_t(f))^{-1}$ or more explicitly, using the inverse as defined in Example \ref{expl:ProbDens}
\begin{eqnarray*}
f_{s \vert t}(\mathbf{x} \vert \mathbf{y}) = \frac{f(\mathbf{x},\mathbf{y})}{\pi_t(f)(\mathbf{y})} 
= \frac{f(\mathbf{x},\mathbf{y})}{\int f(\mathbf{x},\mathbf{y}) d\mathbf{x}}
\end{eqnarray*}
if $\mathbf{x}$  and $\mathbf{y}$ are $s \setminus t$ and $t$-tuples respectively and $\pi_t(f)(\mathbf{y}) > 0$. Otherwise the conditional is zero. Such a conditional density $f_{s \vert t}$ is no more a density, since it is no more integrable. However, for any fixed tuple $\mathbf{y}$, the function $f(\mathbf{x},\mathbf{y})$ as a function $\mathbf{x}$ is a density on $\mathbb{R}^{s \setminus t}$. As such it may be marginalized. This can be used to extend projection to conditionals,
\begin{eqnarray*}
\pi_r(f_{s \vert t})(\mathbf{z},\mathbf{u} \vert \mathbf{y}) = \frac{\int_{- \infty}^\infty f(\mathbf{z},\mathbf{u},\mathbf{y}) d\mathbf{z}}{\int_{- \infty}^\infty f(\mathbf{x},\mathbf{y}) d\mathbf{x}},
\end{eqnarray*}
where $\mathbf{x} = (\mathbf{z},\mathbf{u})$ and $\mathbf{z}$ and $\mathbf{u}$ are $s \setminus r$- and $r \setminus t$-tuples for $t \subseteq r \subseteq s$. We shall see below that in this way, extraction can also partially be extended to conditionals in a separative valuation algebra.
\end{example}

As remarked above, conditionals are in general no more elements of $\Phi$. As a consequence extraction does not extend to conditionals. Nevertheless, it is possible to define an extraction operation for conditionals as a partial operation. In fact, if $y \leq z \leq x \vee y$, define
\begin{eqnarray*}
\epsilon_z(\phi_{x \vert z}) = \epsilon_z(\phi) \cdot (\epsilon_y(\phi))^{-1}.
\end{eqnarray*}
This is obviously again a conditional. It turns out that Proposition \ref{prop:PreliminarRes} with this definition still holds in the case of separative valuation algebras, however with weak conditional independence.

\begin{proposition} \label{prop:SpepValPrepRes} 
If $\phi$ is a separative valuation algebra, then
\begin{enumerate}
\item $\epsilon_y(\phi_{x \vert y}) = f_{\epsilon_y(\phi)}$,
\item if $x \amalg_\phi y \vert z$, then $\phi_{x \vee y \vert z} = \phi_{x \vert y \vee z} \cdot \phi_{y \vert z}$,
\item if $z \leq x$, then $\epsilon_{y \vee z}(\phi_{x \vert y}) = \phi_{z \vert y}$,
\item if $x \amalg_\phi y \vert z$, then $\epsilon_{y \vee z}(\phi_{z \vert x \vee y} \cdot \phi_{x \vert y}) = \phi_{z \vert y}$,
\item if $y$ is a support of $\psi$, then $(\epsilon_{x \vee y}(\phi) \cdot \psi)_{x \vert y} = \phi_{x \vert y} \cdot f_\psi$.
\end{enumerate}
\end{proposition}

\begin{proof}
The proof is exactly as in Proposition \ref{prop:PreliminarRes}.
\end{proof}

Conditional independence implies weak conditional independence, Proposition \ref{prop:StrongWeak}. But the two concepts are not equivalent in general. The following two conditions are sufficient for the equivalence of the two concepts:
\begin{enumerate}
\item $\epsilon_z(\phi) = \chi_1 \cdot \chi_2$ where $\chi_1,\chi_2 \in \Phi$ both with supports $z$,
\item $\psi_1 = \phi_{x \vert z} \cdot \chi_1$ and $\psi_2 = \phi_{y \vert z} \cdot \chi_2$ belong both to $\Phi$.
\end{enumerate}
In fact under these conditions we have, if $x \amalg_\phi y \vert z$, 
\begin{eqnarray*}
\epsilon_{x \vee y \vee z}(\phi) = \phi_{x \vert z} \cdot \phi_{y \vert z} \cdot \epsilon_z(\phi) = ( \phi_{x \vert z} \cdot \chi_1) \cdot ( \phi_{y \vert z} \cdot \chi_2) = \psi_1 \cdot \psi_2,
\end{eqnarray*}
and $\psi_1$ and $\psi_2$ have supports $x \vee z$ and $y \vee z$ respectively. So in this case $x \amalg_\phi y \vert z$ implies $x \bot_\phi y \vert z$.

%% file: chapter12.tex

\chapter{Conditional independence} \label{sec:CondIndep}

\section{Related separoids}

In this section, we examine the relations of conditional independence of domains or questions relative to a piece of information or a valuation as introduced in Section \ref{subsec:Continuation} and Section \ref{subsec:SepContinuation}, but not only for regular or separative valuation algebras, but for valuation and information algebras in general. We ask whether these relations form a q-separoid or even a separoid, and we address the so-called marginal problem.

Let $(\phi,\cdot,1;E)$ with $E = \{\epsilon_x:x \in Q\}$, be a valuation algebra or an iinformation algebra. Definition \ref{def.CondIndepInf} of conditional independence $x \bot_\phi y \vert z$ relative to a valuation or a piece of information $\phi$ is general and does not depend on regularity or separativity. In a first step, we study whether or under what conditions separoid properties are valid for the relation $x \bot_\phi y \vert z$, see Section \ref{sec:StructOfQuest}. Obviously, Symmetry, C2 is valid, for any $\phi \in \Phi$ and $x,y,z \in Q$,
\begin{description}
\item[C2] $x \bot_\phi y \vert z$ implies $y \bot_\phi x \vert z$
\end{description}
Further, since, trivially, $\epsilon_{x \vee y}(\phi) = \epsilon_{x \vee y}(\phi) \cdot 1$ and the unit element has support $y$ we have
\begin{description}
\item[C1] $x \bot_\phi y \vert y$.
\end{description}
Further, if $x \bot_\phi y \vert z$ we have by the definition of this relation also $x \bot_\phi y \vee z \vert z$, hence
\begin{description}
\item[C4] $x \bot_\phi y \vert y$ implies $x \bot_\phi y \vee z \vert z$..
\end{description}
All this is trivial. The separoid condition $\mathbf{C3}$ however is less trivial: Assume $x \bot_\phi y \vert z$ and $u \leq y$. Does this imply $x \bot_\phi u \vert z$? By Lemma \ref{le:ExtCombRs} from $\epsilon_{x \vee y \vee z}(\phi) = \psi_1 \cdot \psi_2$, where $\psi_1$ and $\psi_2$ have supports $x \vee z$ and $y \vee z$ respectively, we obtain $\epsilon_{x \vee u \vee z}(\phi) = \psi_1 \cdot \epsilon_{x \vee u \vee z}(\psi_2)$, since $x \vee z \leq x \vee u \vee z \leq x \vee y \vee z$. But does the second factor have support $u \vee z$? This does not seem the case in general. Since $\psi_2$ has support $y \vee z$ we would have
\begin{eqnarray*}
\epsilon_{x \vee u \vee z}(\phi) = \epsilon_{x \vee u \vee z}(\epsilon_{u \vee z}(\phi)) = \epsilon_{u \vee z}(\phi)
\end{eqnarray*}
if $x \vee u \vee z \bot y \vee z \vert u \vee z$. But this is not the case in general. However, this holds if $(Q,\leq)$ is a distributive lattice. In this case, the relation $x \bot y \vert z$ is commutative, $x \bot y \vert z = x \bot_L y \vert z$ (Proposition \ref{prop:LCond_1}) and thus we have $x \bot y \vert x \wedge y$. Now, in a distributive lattice $(x \vee u \vee z) \wedge (y \vee z) = u \vee z$ and therefore in this case $x \vee u \vee z \bot y \vee z \vert u \vee z$. This proves the following theorem.

\begin{theorem} \label{th:CondIndepQSep}
If $(\Phi,\cdot,1;E)$ with $E = \{\epsilon_x:x \in Q\}$ is a valuation algebra, where $(Q,\leq)$ is a distributive lattice, then for all $\phi \in \Phi$, the relation $x \bot_\phi y \vert z$ forms a q-separoid.
\end{theorem}

If the lattice $(Q,\leq)$ is \textit{modular}, then the relation $x \bot y \vert z$ is a separoid (proposition \ref{prop:ModLatt}). In particular we have property C5, namely $x \bot y \vert z$ and $u \leq y$ imply $x \bot y \vert z \vee u$. Assume now $x \bot_\phi y \vert z$ so that $\epsilon_{x \vee y \vee z}(\phi) = \psi_1 \cdot \psi$ with supports $x \vee z$ and $y \vee z$ for $\psi_1$ and $\psi_2$ respectively. But then $\psi_1$ has also support $x \vee z \vee u$ since this domain is greater than $x \vee z$ and, if $u \leq y$, then $y \vee z = y \vee z \vee u$. So, \textbf{C5} holds also for the relation $x \bot_\phi y \vert z$.

\begin{proposition} \label{propC5CondIndepRel}
If $(\Phi,\cdot,1;E)$ with $E = \{\epsilon_x:x \in Q\}$ is a valuation algebra, where $(Q,\leq)$ is a modular lattice, then
\begin{description}
\item[C5] $x \bot_\phi y \vert z$ and $u \leq y$ imply $x \bot_\phi y \vert z \vee u$.
\end{description}
\end{proposition}

The condition C6 of a separoid is another question. It holds in a regular valuation algebra. In fact, from $x \bot_\phi u\vert y \vee z$ we have $\epsilon_{x \vee y \vee z \vee u}(\phi) = \phi_{x \vert y \vee z} \cdot \phi_{u \vert y \vee z} \cdot \epsilon_{y \vee z}(\phi)$ (Theorem \ref{th:EquivCondIndepRel}, item 2).  Further, $\phi_{u \vert y \vee z}  \cdot \epsilon_{y \vee z}(\phi) = \epsilon_{y \vee z \vee u}(\phi) = \phi_{y \vee u \vert z} \cdot \epsilon_z(\phi)$. By the same theorem (item 7) we have also that $x \bot_\phi y \vert z$ implies $\phi_{x \vert y \vee z} = \phi_{x \vert z} \cdot f_{\epsilon_{y \vee z}(\phi)}$ It follows 
\begin{eqnarray*}
\epsilon_{x \vee y \vee z \vee u}(\phi) &=& \phi_{x \vert z} \cdot f_{\epsilon_{y \vee z}(\phi)} \cdot \phi_{u \vert y \vee z} \cdot \epsilon_{y \vee z}(\phi) \\
&=&  \phi_{x \vert z} \cdot \phi_{u \vert y \vee z} \cdot \epsilon_{y \vee z}(\phi).
\end{eqnarray*}
But this means that $x \bot_\phi y \vee u \vert z$. This is \textbf{C6} for the relation $x \bot_\phi y \vert z$. Thus we have proved the following result.

\begin{theorem} \label{th:RegSep}
If $(\Phi,\cdot,1;E)$ with $E = \{\epsilon_x:x \in Q\}$ is a regular valuation algebra, where $(Q,\leq)$ is a modular lattice, then the relation $x \bot_\phi y \vert z$ is a separoid. That is in addition to C1 to C4 we have further
\begin{description}
\item[C5] $x \bot_\phi y \vert z$ and $u \leq y$ imply $x \bot_\phi y \vert z \vee u$.
\item[C6] $x \bot_\phi y \vert z$ and $x \bot_\phi u \vert y \vee z$ imply $x \bot_\phi y \vee u \vert z$
\end{description}
\end{theorem}

As an illustration, we consider the multivariate case. Let $I$ be the index set of variables. Then the set of questions can be identified with the subsets of $I$, see Section \ref{sec:SetAlg}. This is a distributive lattice and the conditional independence relation for subsets $x$, $y$ and $z$ of $I$ is defined by $x \bot y \vert z = x \bot_L y \vert z = x \bot_d y \vert z$, that is $(x \cup z) \cap (y \cup z) = z$ or $x \cap y \leq z$, see Section \ref{sec:StructOfQuest}. Since here $(Q,\leq)$ is even a Boolean lattice, there is still another definition of the conditional independence relation. Suppose $x \bot y \vert z$ and define $r = x \setminus z$, $s = y \setminus z$ and $t = z$. Then $r$,$s$ and $t$ are disjoint subsets of $I$. Then, obviously $r \bot s \vert t$ since $(r \cup t) \cap (s \cup t) = t$. So for disjoint subsets $r$, $s$ and $t$ of $I$ we always (trivially) have $r \bot s \vert t$. This allows us to reformulate the separoid properties in a multivariate case in an alternative form, familiar from conditional independence between random variables in probability theory.

Now, if $r$, $s$ and $t$ are disjoint subsets of $I$, then for an element $\phi \in \Phi$ we have $r \bot_\phi s \vert t$ if \begin{eqnarray*}
\epsilon_{r \cup s \cup t}(\phi) = \psi_1 \cdot \psi_2,
\end{eqnarray*}
where $\psi_1$ and $\psi_2$ have supports $r \cup t$ and $s \cup t$ respectively, according to Definition \ref{def.CondIndepInf}. This relation has for a valuation algebra the following properties:

\begin{theorem} \label{th:Graphoid}
Let $(\phi,\cdot,1;E)$ be a valuation algebra, with $E = \{\epsilon_x:x \in Q\}$, where $Q$ is the Boolean subset lattice of an index set $I$ and $s,t,u,v \in Q$ disjoint sets. Then
\begin{description}
\item[G1] Symmetry: $s \bot_\phi t \vert u$ implies $t \bot_\phi s \vert u$,
\item[G2] Decomposition: $s \bot_\phi t \cup v \vert u$ implies $s \bot_\phi t \vert u$,
\item[G3] Weak Union: $s \bot_\phi t \cup v \vert u$ implies $s \bot_\phi t \vert u \cup v$,
\end{description}
\end{theorem}

In the proof of this theorem we need the following simple result for commutative algebras.

\begin{lemma} \label{le:ComCombAx}
Let $(\phi,\cdot,1;E)$ be a commutative valuation algebra. Then if $x \leq z \leq x \vee y$ in $Q$, $\psi_1$ has support $x$ and $\psi_2$ support $y$,
\begin{eqnarray*}
\epsilon_z(\psi_1 \cdot \psi_2) = \psi_1 \cdot \epsilon_{y \wedge z}(\psi_2).
\end{eqnarray*}
\end{lemma}

\begin{proof}
Note that $\psi_1$ has also support $z$, since $x \leq z$. So $\epsilon_z(\psi_1 \cdot \psi_2) = \psi_1 \cdot \epsilon_z(\psi_2)$. But in a commutative algebra $y \bot z \vert y \wedge z$, so that $\epsilon_z(\psi_2) = \epsilon_z(\epsilon_{y \wedge z}(\psi_2))$ and since the extraction operators commute, $\epsilon_z(\psi_2) = \epsilon_{y \wedge z}(\psi_2)$.
\end{proof}

\begin{proof}
Now, we turn to the proof of the theorem. G1, symmetry is obvious from the definition. For G2, $s \bot_\phi t \cup v \vert u$ means that $\epsilon_{s \cup t \cup u \cup v}(\phi) = \psi_1 \cdot \psi_2$, where $\psi_1$ has support $s \cup u$ and $\psi_2$ support $t \cup u \cup v$. Applying Lemma \ref{le:ComCombAx}, we obtain $\epsilon_{s \cup t \cup u}(\phi) = \psi_1 \cdot \epsilon_{t \cup u}(\psi_2)$ since $(s \cup t \cup u) \cap (t \cup u \cup v) = t \cup u$ by distributivity. But this shows that $s \bot_\phi t \vert u$.

For G3 we have from $s \bot_\phi t \cup v \vert u$ that $\epsilon_{s \cup t \cup u \cup v}(\phi) = \psi_1 \cdot \psi_2$, where $s \cup u$ is a support for $\psi_1$ and $t \cup u \cup v$ a support for $\psi_2$. But then $s \cup u \cup v$ is also a support for $\psi_1$ and we have indeed $s \bot_\phi t \vert u \cup v$.
\end{proof}

If the valuation algebra is regular, then in addition the following holds.

\begin{theorem} \label{th:Graphoid1}
Let $(\Phi,\cdot,1;E)$ be a regular valuation algebra, with $E = \{\epsilon_x:x \in Q\}$, where $Q$ is the Boolean subset lattice of an index set $I$ and $s,t,u,v \in Q$ disjoint sets. Then
\begin{description}
\item[G4] Contraction: $s \bot_\phi t \vert u$ and $s \bot_\phi v \vert t \cup u$ imply $s \bot_\phi t \cup v \vert u$,
\end{description}
\end{theorem}

\begin{proof}
The assumption $s \bot_\phi t \vert u$ means that $\epsilon_{s \cup t \cup u}(\phi) = \psi_1 \cdot \psi_2$ where $\psi_1$ and $\psi_2$ have supports $s \cup u$ and $t \cup u$ respectively. Further $s \bot_\phi v \vert t \cup u$ on the other hand means that
$\epsilon_{s \cup t \cup u \cup v}(\phi) = \eta_1 \cdot \eta_2$ where $\eta_1$ and $\eta_2$ have supports $s \cup t \cup u$ and $t \cup u \cup v$ respectively. Using Lemma \ref{le:ComCombAx}  we obtain from this $\epsilon_{s \cup t \cup u}(\phi) = \eta_1 \cdot \epsilon_{t \cup u}(\eta_2)$ since $(s \cup t \cup u) \cap (t \cup u \cup v) = t \cup u$. Then we conclude that
\begin{eqnarray*}
\epsilon_{s \cup t \cup u \cup v}(\phi) = \eta_1 \cdot \eta_{2 v \vert t \cup u} \cdot \epsilon_{t \cup u}(\eta_2) = \psi_1 \cdot (\psi_2 \cdot \eta_{2 v \vert t \cup u}).
\end{eqnarray*}
Here the first factor has support $s \cup u$ whereas the second has support $t \cup u \cup v$ and this means that $s \bot_\phi t \cup v \vert u$. 
\end{proof}

Properties G1 to G4 define a structure termed a \textit{semi-graphoid} in \cite{PearlPaz1989}. 

Still, for regular valuation algebras and positive valuations, we have yet another result.

\begin{theorem} \label{th:Graphoid1}
Let $(\Phi,\cdot,1;E)$ be a regular valuation algebra, with $E = \{\epsilon_x:x \in Q\}$, where $Q$ is the Boolean subset lattice of an index set $I$ and $s,t,u,v \in Q$ disjoint sets. If $\phi \in \Phi$ is positive, then
\begin{description}
\item[G5] Intersection: $s \bot t \vert u \cup v$ and $s \bot v \vert t \cup u$ imply $s \bot t \cup v \vert u$.
\end{description}
\end{theorem}

\begin{proof}
Recall that $\phi \in [1]_\gamma$ since $\phi$ is assumed positive. Then, since $\equiv_\gamma$ is a congruence relative to extraction, $\phi \equiv_\gamma 1$ implies for any $x \in Q$ that $\epsilon_x(\phi) \equiv_\gamma \epsilon_x(1) = 1$. That is, if $\phi$ is positive, so is $\epsilon_x(\phi)$. Now according to the assumptions $s \bot t \vert u \cup v$ and $s \bot v \vert t \cup u$ and by Theorem \ref{th:EquivCondIndepRel}
\begin{eqnarray*}
\phi_{s \vert t \cup u \cup v} &=& \phi_{s \vert u \cup v}, \\
\phi_{s \vert t \cup u \cup v} &=& \phi_{s \vert t \cup u}
\end{eqnarray*}
since the element $\phi$ is positive, that is $f_{\epsilon_{t \cup u \cup v}}(\phi) = 1$. So we have $\phi_{s \vert u \cup v} = \phi_{s \vert t \cup u}$. Combine now both sides of this identity with $\epsilon_{t \cup u}(\phi) \cdot \epsilon_{u \cup v}(\phi)$ to obtain 
\begin{eqnarray*}
\epsilon_{s \cup u \cup v}(\phi) \cdot \epsilon_{t \cup u}(\phi) = \epsilon_{s \cup u \cup v}(\phi) \cdot \epsilon_{u \cup v}(\phi).
\end{eqnarray*}
Next, apply the operator $\epsilon_{s \cup t \cup u}$ to both sides and use Lemma \ref{le:ComCombAx}. The we get 
\begin{eqnarray*}
\epsilon_{s \cup u}(\phi) \cdot \epsilon_{t \cup u}(\phi) = \epsilon_{s \cup t \cup u}(\phi) \cdot \epsilon_u(\phi).
\end{eqnarray*}
By Theorem \ref{th:EquivCondIndepRel} this means that $s \bot_\phi t \vert u$. But then, by the same theorem and positivity of $\phi$, we have $\phi_{s \vert t \cup u} = \phi_{s \vert u}$ and from this it follows that $\phi_{s \vert t \cup u \cup v} = \phi_{s \vert u}$ and this means that $s \bot t \cup v \vert u$.
\end{proof}

All these results have already been stated and proved for the labeled version of valuation algebras in \cite{kohlas03}. For separative valuation algebras and weak conditional independence, similar results hold.

\begin{theorem} \label{th:SpeValGraph}
Let $(\Phi,\cdot,1;E)$ be a separative valuation algebra, with $E = \{\epsilon_x:x \in Q\}$, where $(Q,\leq)$ is a distributive lattice. Then, the weak conditional independence relation $x \amalg_\phi y \vert z$ forms a q-separoid.
\end{theorem}

\begin{proof}
C1 follows since $\epsilon_{x \vee y}(\phi) \cdot \epsilon_y(\phi) = \epsilon_{x \vee y \vee y}(\phi) \cdot \epsilon_{y \vee y}(\phi)$. C2 is obvious from the definition of the relation. For C3 note that $x \amalg_\phi y \vert z$ implies 
\begin{eqnarray*}
\epsilon_{x \vee y \vee z}(\phi) \cdot \epsilon_z(\phi) = \epsilon_{x \vee z}(\phi) \cdot \epsilon_{y \vee z}(\phi).
\end{eqnarray*}
If we extract both sides for $x \vee z \vee u$, we obtain, using Lemma \ref{le:ComCombAx},
\begin{eqnarray*}
\epsilon_{x \vee z \vee u}(\phi) \cdot  \epsilon_z(\phi) = \epsilon_{x \vee z}(\phi) \cdot \epsilon_{z \vee u}(\phi)
\end{eqnarray*}
since $x \vee z \leq x \vee z \vee u \leq x \vee y \vee z$ and $(x \vee z \vee u) \cap (y \vee z) = u \vee z$. This is then $x \amalg_\phi u \vert z$. Finally from $x \amalg_\phi y \vert z$ we have $\epsilon_{x \vee y \vee z}(\phi) \cdot \epsilon_z(\phi)  = \epsilon_{x \vee z}(\phi) \cdot \epsilon_{y \vee z}(\phi)$. Since $y \vee z = (y \vee z) \vee z$ we have also $x \amalg_\phi y \vee z \vert z$, hence C4.
\end{proof}

Theorem \ref{th:RegSep} carries over to separative algebras too.

\begin{theorem} \label{th:SepaSep}
If $(\Phi,\cdot,1;E)$ with $E = \{\epsilon_x:x \in Q\}$ is a separative valuation algebra, where $(Q,\leq)$ is a modular lattice, then the relation $x \amalg_\phi y \vert z$ is a separoid. That is in addition to C1 to C4 we have further
\begin{description}
\item[C5] $x \amalg_\phi y \vert z$ and $u \leq y$ imply $x \amalg_\phi y \vert z \vee u$.
\item[C6] $x \amalg_\phi y \vert z$ and $x \amalg_\phi u \vert y \vee z$ imply $x \amalg_\phi y \vee u \vert z$
\end{description}
\end{theorem}

\begin{proof}
For C5 according to $x \amalg_\phi y \vert z$ we have, see Theorem \ref{th:WeakEquivCindIndep}, 
\begin{eqnarray*}
\epsilon_{x \vee y \vee z}(\phi) \cdot \epsilon_z(\phi) = \epsilon_{x \vee z}(\phi) \cdot \epsilon_{y \vee z}(\phi).
\end{eqnarray*}
Combine both sides with $\phi_{z \vert u}$. This gives
\begin{eqnarray*}
\epsilon_{x \vee y \vee z}(\phi) \cdot \epsilon_{z \vee u}(\phi) = \epsilon_{x \vee z \vee u}(\phi) \cdot \epsilon_{y \vee z}(\phi)
\end{eqnarray*}
and this means $x \amalg_\phi y \vert z \vee u$ by Theorem \ref{th:WeakEquivCindIndep} since $u \leq y$.

To show C6 we use the assumption $x \amalg_\phi u \vert y \vee z$ which tells us that 
\begin{eqnarray*}
\epsilon_{x \vee y \vee z \vee u}(\phi) \cdot \epsilon_{y \vee z}(\phi) = \epsilon_{x \vee y \vee z}(\phi) \cdot \epsilon_{y \vee z \vee u}(\phi).
\end{eqnarray*}
Then we further have $\epsilon_{ y \vee z}(\phi) = \phi_{y \vert z} \cdot \epsilon_z(\phi)$ and $\epsilon_{x \vee y \vee z}(\phi) = \phi_{y \vert x \vee z} \cdot \epsilon_{x \vee z}(\phi)$. From $x \amalg_\phi y \vert z$ we deduce that $\phi_{y \vert x \vee z} = \phi_{y \vert z} \cdot f_{\epsilon_{x \vee z}(\phi)}$ (Theorem \ref{th:WeakEquivCindIndep}) so that $\epsilon_{x \vee y \vee z}(\phi) = \phi_{y \vert z} \cdot \epsilon_{x \vee z}(\phi)$. Introducing this above, we obtain
\begin{eqnarray*}
\epsilon_{x \vee y \vee z \vee u}(\phi) \cdot \phi_{y \vert z} \cdot \epsilon_z(\phi)  = \phi_{y \vert z} \cdot \epsilon_{x \vee z}(\phi) \cdot \epsilon_{y \vee z \vee u}(\phi).
\end{eqnarray*}
Elimination $\phi_{y \vert z}$ on both sides and noting that $f_{\epsilon_{y \vee z}(\phi)}$ is absorbed on both sides, it follows
\begin{eqnarray*}
\epsilon_{x \vee y \vee z \vee u}(\phi) \cdot \epsilon_z(\phi)  = \epsilon_{x \vee z}(\phi) \cdot \epsilon_{y \vee z \vee u}(\phi).
\end{eqnarray*}
This means that $x \amalg_\phi y \vee u \vert z$.
\end{proof}
 So, weak conditional independence exhibts the same structure as conditional independence.

\section{The marginal problem} \label{subsec:MarpProblem}

The marginal problem consists in finding relative to a set of domains $x_1,\ldots,x_n$ a set of valuations $\phi_1,\ldots,\phi_n$ with domains $x_1,\ldots,x_n$ respectively, which are compatible among themselves in the sense that there is a valuation $\phi$ such that $\epsilon _{x_i} = \phi_i$ for $i = 1,\ldots,n$. This is called the marginal problem. The solution of this general marginal problem, that is to decide whether given valuations $\phi_1,\ldots,\phi_n$ are compatible in this sense and to find $\phi$ is difficult. But there are important, more specific instances of the problem where simple necessary and even sufficient conditions for compatibility can be found. 

The first case we examine is the one of a (domain-free) information algebra, that is of an idempotent valuation algebra. This case has been treated in \cite{CasaKohIIJAR,CasaKohl22} in the context of imprecise probabilities, especially coherent sets of gambles. It was already noted there, that the results do not depend on the specific example, but are general for information algebras. So, here we discuss the general case of an idempotent information algebra. First, we formulate the problem formally.

\begin{definition} \textbf{Compatibility}
A set $\phi_1,\ldots,\phi_n$ of elements of a valuation or information algebra $\Phi$ with supports $x_1,\ldots,x_n$ respectively is called compatible if there is an element $\phi \in \Phi$ such that
\begin{eqnarray*}
\epsilon_{x_i}(\phi) = \phi_i \textrm{ for}\ i = 1,\ldots,n.
\end{eqnarray*}
\end{definition}

As noted above, the elements are compatible, if the they are the marginals or extractions of a common element. In the case of an information algebra, there is a very simple necessary and sufficient condition for compatibility.

\begin{proposition} \label{prop:IdempotComp}
Let $(\Phi,\cdot,0,1;E)$ with $E = \{\epsilon_x:x \in Q\}$ be an information algebra. Then $\phi_1,\ldots,\phi_n \in \Phi$ is compatible if an only if
\begin{eqnarray} \label{eq:IdempotComp}
\phi_i = \epsilon_{x_i}(\phi_1 \cdot \ldots \cdot \phi_n).
\end{eqnarray}
\end{proposition}

\begin{proof}
If (\ref{eq:IdempotComp}) holds, the elements $\phi_1,\ldots,\phi_n$ are compatible with $\phi = \phi_1 \cdot \ldots \cdot \phi_n$. On the other, hand, if $\phi_1,\ldots,\phi_n$ are compatible, then there is an element $\phi \in \Phi$ such that $\phi_i = \epsilon_{x_i}(\phi)$. Now $\phi \geq \epsilon_{x_i}(\phi) = \phi_i$ so that $\phi \geq \phi_1 \cdot \ldots \cdot \phi_n$. It follows that
\begin{eqnarray} \label{eq:CompInequ}
\phi_i = \epsilon_{x_i}(\phi) \geq \epsilon_{x_i}(\phi_1 \cdot \ldots \cdot \phi_n) \geq \epsilon_{x_i}(\phi_i) = \phi_i,
\end{eqnarray}
since $\phi_i$ has support $x_i$. So we have indeed $\phi_i = \epsilon_{x_i}(\phi_1 \cdot \ldots \cdot \phi_n)$.
\end{proof}

Of particular interest is the case of the compatibility of two elements $\phi_i$ and $\phi_j$ with support $x_i $ and $x_j$. If such two elements are compatible, we call them pairwise compatible. It is obvious that compatibility of $\phi_1,\ldots,\phi_n$ implies pairwise compatibility of all pairs $\phi_i$ and $\phi_j$, since 
\begin{eqnarray*}
\phi_i = \epsilon_{x_i}(\phi) \geq \epsilon_{x_i}(\phi_1 \cdot \ldots \cdot \phi_n) \geq \epsilon_{x_i}(\phi_i \cdot \phi_j) \geq \epsilon_{x_i}(\phi_i) = \phi_i,
\end{eqnarray*}
so that $\phi_i = \epsilon_{x_i}(\phi_i \cdot \phi_j)$ and $\phi_j = \epsilon_{x_j}(\phi_i \cdot \phi_j)$. However, pairwise compatibility of all pairs in the set $\phi_1,\ldots,\phi_n$ does not imply in general compatibility of this set. We return to this question below.

Pairwise compatibility, as well as compatibility in general, are closely related to conditional independence. It provides a sufficient condition for pairwise compatibility.

\begin{proposition} \label{prop:COMPCondIndep}
Let $\Phi$ be an information algebra. If $x \bot y \vert z$ and $\phi_1,\phi_2 \in \Phi$ are two elements with support $x \vee z$ and $y \vee z$, such that $\epsilon_z(\phi_1) = \epsilon_z(\phi_2)$, then $\phi_1$ and $\phi_2$ are pairwise compatible. 
\end{proposition}

\begin{proof}
Using $x \bot y \vert z$ we have
\begin{eqnarray*}
\epsilon_{x \vee z}(\phi_1 \cdot \phi_2) &=& \phi_1 \cdot \epsilon_{x \vee z}(\phi_2) = \phi_1 \cdot \epsilon_{x \vee z}(\epsilon_z(\phi_2))
= \phi_1 \cdot \epsilon_{x \vee z}(\epsilon_z(\phi_1)) ) \\
&=& \phi_1 \cdot \epsilon_z(\phi_1) = \phi_1.
\end{eqnarray*}
For $\phi_2$ we obtain in the same way $\epsilon_{y \vee z}(\phi_1 \cdot \phi_2) = \phi_2$. So, $\phi_1$ and $\phi_2$ are indeed pairwise compatible.
\end{proof}

This sufficiency result extends to a family $\phi_1,\ldots \phi_n$ with $n \geq 2$ if the domains $x_1,\ldots,x_n$ form a hypertree, see Section \ref{subsec:CondIndepStruct}.

\begin{theorem} \label{th:InfAlgHyptree}
Let $\Phi$ be an information algebra. Consider a set of elements $\phi_1,\ldots,\phi_n \in \Phi$ with supports $x_1,\ldots,x_n$. If the set $S = \{x_1,\ldots,x_n\}$ forms a hypertree and the elements of $\phi_1,\ldots,\phi_n $ are pairwise compatible, then they are compatible
\end{theorem}

This theorem has been proved in \cite{CasaKohIIJAR,CasaKohl22} in the context of imprecise probability. It has been noted there, that the proof does not depend on the particularities of the example of imprecise probability. Therefore, we do not repeat the proof here. Also, below we extend this theorem to regular valuation algebras. The proof of this theorem covers then also Theorem \ref{th:InfAlgHyptree} since information algebras are regular valuation algebras.

We turn now to the case of  regular valuation algebras. Note that in this case (\ref{eq:CompInequ}) does not imply $\phi_i = \epsilon_{x_i}(\phi_1 \cdot \ldots \cdot \phi_n)$, since the information order is only a preorder in regular algebras, so that Proposition \ref{prop:IdempotComp} is no more valid. But we have a sufficient condition similar to Proposition \ref{prop:COMPCondIndep}

\begin{proposition} \label{prop:RegpairWiseComp}
Let $(\Phi,\cdot,1;E)$ with $E = \{\epsilon_x:x \in Q\}$ be a regular valuation algebra. If $x \bot y \vert z$ and $\phi_1,\phi_2 \in \Phi$ are two elements with support $x \vee z$ and $y \vee z$, such that $\epsilon_z(\phi_1) = \epsilon_z(\phi_2)$, then $\phi_1$ and $\phi_2$ are pairwise compatible. 
\end{proposition}

\begin{proof}
Define $\eta = \epsilon_z(\phi_1) = \epsilon_z(\phi_2)$ and
\begin{eqnarray*}
\phi = \phi_1 \cdot \phi_2 \cdot \eta^{-1}.
\end{eqnarray*}
Then, using $x \vee z \bot y \vee z\vert z$, we have
\begin{eqnarray*}
\epsilon_{x \vee z}(\phi) &=& \phi_1 \cdot \epsilon_{x \vee z}(\phi_2 \cdot \eta^{-1}) = \phi_1 \cdot \epsilon_{x \vee z}(\epsilon_z(\phi_2 \cdot \eta^{-1})) \\
&=& \phi_1 \cdot \epsilon_{x \vee z}(\epsilon_z(\phi_2) \cdot \eta^{-1})  \\
&=& \phi_1 \cdot \epsilon_{x \vee z}(f_\eta)) = \epsilon_{x \vee z}(\phi_1 \cdot f_\eta) = \epsilon_{x \vee z}(\phi_1) = \phi_1.
\end{eqnarray*}
In the same way, we obtain $\epsilon_{y \vee z}(\phi) = \phi_2$.
\end{proof}

This proposition is a generalization of Proposition \ref{prop:COMPCondIndep}. Of course the same results holds also, if $\phi_1$ and $\phi_2$ have support $x$ and $y$ since then they have also support $x \vee z$ and $y \vee z$. This is so, because $\epsilon_x(\phi) = \epsilon_x(\epsilon_{x \vee z}(\phi)) = \epsilon_x(\phi_1) = \phi_1$.

Next, we want to extend Theorem \ref{th:InfAlgHyptree}. Consider a hypertree $\{x_1,\ldots,x_n\}$ with the numbering selected such  that $x_i \bot \vee_{j=i+1}^n x_j \vert x_{b(i)}$ for $i =1,\ldots,n-1$, see Section \ref{subsec:CondIndepStruct} and $\phi_i$ for $i = 1,\ldots,n$ with supports $x_i$, and so that its pairs $\phi_i$ and $\phi_{b (i)}$ are pairwise compatible in the sense of Proposition \ref{prop:RegpairWiseComp}. That is there a domains $z_i \in Q$ such that $x_i \bot x_{b(i)} \vert z_i$ and $\epsilon_{z_i}(\phi_i) = \epsilon_{z_i}(\phi_{b(i})$ for $i =1,\ldots,n-1$. The elements $z_i$ are called separators in the hypertree. Then, we have the following extension of Theorem \ref{th:InfAlgHyptree}.

\begin{theorem} \label{th:ValAlgHyptree}
Let $\Phi$ be a regular valuation algebra. Consider a set of elements $\phi_1,\ldots,\phi_n \in \Phi$ with supports $x_1,\ldots,x_n$. If the set $S = \{x_1,\ldots,x_n\}$ forms a hypertree and the elements $\phi_i$ and $\phi_{b(i)}$ of $\phi_1,\ldots,\phi_n $ are pairwise compatible in the sense that $x_i \bot x_{b(i)} \vert z_i$, where $z_i \leq x_i$, and $\epsilon_{z_i}(\phi_i) = \epsilon_{z_i}(\phi_{b(i)}) = \eta_i$ for $i =1,\ldots,n-1$, then the elements $\phi_1,\ldots,\phi_n$ are compatible and 
\begin{eqnarray} \label{eq:ProdComp}
\phi_i = \epsilon_{x_i}(\phi_1 \cdot \ldots \cdot \phi_n \cdot \eta_1^{-1} \cdot \ldots \cdot \eta_{n-1}^{-1}),\ i = 1,\ldots,n.
\end{eqnarray}
\end{theorem}

\begin{proof}
Define
\begin{eqnarray*}
\phi = \phi_1 \cdot \ldots \cdot \phi_n \cdot \eta_1^{-1} \cdot \ldots \cdot \eta_{n-1}^{-1}
\end{eqnarray*}
and further $y_i = x_{i+1} \vee \ldots \vee x_n \vee z_{i+1} \vee \ldots \vee z_n $ for $i =1,\ldots,n-1$. In a first step, we eliminate $x_1$ from the hypertree by extracting $\phi$ to the domain $y_{1}$. We obtain 
\begin{eqnarray*}
\epsilon_{y_1}(\phi) &= &\epsilon_{y_1}(\phi_1 \cdot \ldots \cdot \phi_n \cdot \eta_1^{-1} \cdot \ldots \cdot \eta_{n-1}^{-1}) \\
&=& \epsilon_{y_1}(\phi_1 \cdot \eta_1^{-1}) \cdot \phi_2 \cdot \ldots \cdot \phi_n \cdot \eta_2^{-1} \cdot \ldots \cdot \eta_{n-1}^{-1},
\end{eqnarray*}
since the part $\phi_2 \cdot \ldots \cdot \phi_n \cdot \eta_2^{-1} \cdot \ldots \cdot \eta_{n-1}^{-1}$ of $\phi$ has support $y_1$. Now, we use the hypertree condition $x_1 \bot y_1 \vert x_{b(1)}$  which gives us
\begin{eqnarray*}
\epsilon_{y_1}(\phi) &=& \epsilon_{y_1}(\epsilon_{x_{b(1)}}(\phi_1 \cdot \eta_1^{-1})) \cdot \phi_2 \cdot \ldots \cdot \phi_n \cdot \eta_2^{-1} \cdot \ldots \cdot \eta_{n-1}^{-1} \\
&=&  \epsilon_{x_{b(1)}}(\phi_1 \cdot \eta_1^{-1}) \cdot \phi_2 \cdot \ldots \cdot \phi_n \cdot \eta_2^{-1} \cdot \ldots \cdot \eta_{n-1}^{-1} \end{eqnarray*}
since  $x_{b(1)} \leq y_1$. Now, we use $x_1 \bot x_{b(1)} \vert z_1$, recalling that $\phi_{x_{b(1)}}$ is a factor in the second part of the combination above. Then, we have
\begin{eqnarray*}
\epsilon_{y_1}(\phi) 
= \epsilon_{x_{b(1)}}(\phi_1  \cdot \phi_{x_{b(1)}} \cdot \eta_1^{-1}) \cdot \phi_2 \cdot \ldots \cdot \phi_n \cdot \eta_2^{-1} \cdot \ldots \cdot \eta_{n-1}^{-1} 
\end{eqnarray*}
Note that in this combination it is understood that $\phi_{x_{b(1)}}$ is no more contained as a factor in the combination $\phi_2 \cdot \ldots \cdot \phi_n \ldots$. Now pairwise compatibility (Proposition \ref{prop:RegpairWiseComp}) implies 
\begin{eqnarray*}
\epsilon_{x_{b(1}}(\phi_1  \cdot \phi_{x_{b(1)}} \cdot \eta_1^{-1}) = \phi_{x_{b(1)}}
\end{eqnarray*}
It follows then that 
\begin{eqnarray*}
\epsilon_{y_1}(\phi) =  \epsilon_{x_{b(1)}}(\phi_1 \cdot \phi_{b(1)} \cdot \eta_1^{-1}) \cdot \phi_2 \cdot \ldots \cdot \phi_n \cdot \eta_2^{-1} \cdot \ldots \cdot \eta_n^{-1}
= \phi_2 \cdot \ldots \cdot \phi_n \cdot \eta_2^{-1} \cdot \ldots \cdot \eta_n^{-1}.
\end{eqnarray*}

Now by induction over $n$, we obtain in exactly the same way for $i = n-1,\ldots,1$.
\begin{eqnarray*}
\epsilon_{y_i}(\phi) &=&\epsilon_{y_i}(\phi_i \cdot \phi_{i+1} \cdot \ldots \cdot \phi_n \cdot \eta_i^{-1} \cdot \eta_{i+1}^{-1} \cdot \ldots \cdot \eta_n^{-1})
= \phi_{i+1} \cdot \ldots \cdot \phi_n \cdot \eta_{i+1}^{-1} \cdot \ldots \cdot \eta_n^{-1}
\end{eqnarray*}
for $i = 1,\ldots,n-1$. In particular for $i = n-1$ we obtain 
\begin{eqnarray*}
\epsilon_{x_n}(\phi_1 \cdot \ldots \cdot \phi_n \cdot \eta_1^{-1} \cdot \ldots \cdot \eta_{n-1}^{-1}) = \phi_n.
\end{eqnarray*}
since $y_n = x_n$.

Now, we claim that $\epsilon_{x_i}(\phi) = \epsilon_{x_i}(\epsilon_{x_{b(i)}}(\phi)) \cdot \phi_i \cdot \eta_i^{-1}$. In fact,
\begin{eqnarray*}
\epsilon_{x_i }(\epsilon_{x_{b(i)}}(\phi)) \cdot \phi_i \cdot \eta_i^{-1} &=& \epsilon_{x_i}(\epsilon_{x_{b(i)}}(\epsilon_{y_i}(\phi)))\cdot \phi_i \cdot \eta_i^{-1} \\
&=& \epsilon_{x_i}(\epsilon_{y_i}(\phi)) \cdot \phi_i \cdot \eta_i^{-1} \\
&=& \epsilon_{x_i}(\phi_{i+1} \cdot \ldots \cdot \phi_n \cdot \eta_{i+1}^{-1} \cdot \ldots \cdot \eta_n^{-1}) \cdot \phi_i \cdot \eta_i^{-1} \\
&=& \epsilon_{x_i}(\phi_i \cdot \phi_{i+1} \cdot \ldots \cdot \phi_n \cdot \eta
_i^{-1} \cdot \eta_{i+1}^{-1} \cdot \ldots \cdot \eta_n^{-1}).
\end{eqnarray*}
This follows since $x_{b(i)} \leq y_i$, $z_i \leq x_i$ and $y_i \bot x_i \vert x_{b(i)}$. This verifies the claim for $i = 1$. For $i \geq 2$ we have, given that $x_i \leq y_{i-1}$, 
\begin{eqnarray*}
\epsilon_{x_i}(\phi_i \cdot \phi_{i+1} \cdot \ldots \cdot \phi_n \cdot \eta_i^{-1} \cdot \eta_{i+1}^{-1} \cdot \ldots \cdot \eta_n^{-1}) \
= \epsilon_{x_i}(\epsilon_{y_i}(\phi)) = \epsilon_{x_i}(\phi).
\end{eqnarray*}
Now we make the induction assumption that $\epsilon_{x_i}(\phi)= \phi_j$ for $j \geq i+1$ which is based on the case $i = n$. Then it follows using pairwise compatibility
\begin{eqnarray*}
\epsilon_{x_i}(\phi) = \epsilon_{x_i}(\epsilon_{x_{b(i)}}(\phi)) \cdot \phi \cdot \eta_i^{-1} = \epsilon_{x_i}(\phi_{x_{b(i)}} \cdot \phi \cdot \eta_i^{-1})  = \phi_i.
\end{eqnarray*}
This concludes the proof.

\end{proof}

As remarked above, this Theorem and its proof covers also the case of an idempotent information algebra. We recall that in this case $\eta_i ^{-1} = \epsilon_{z_i}(\phi_i) = \epsilon_{z_i}(\phi_{b(i)})$ and these terms are absorbed in (\ref{eq:ProdComp}), so that if $z_i \leq x_i$,
\begin{eqnarray*}
\phi_i = \epsilon_{x_i}(\phi_1 \cdot \ldots \cdot \phi_n),\ i = 1,\ldots,n.
\end{eqnarray*}

A particular case are commutative valuation algebras. Then $(Q,\leq)$ is a lattice and $x \bot y \vert y \wedge y$ fort all $x,y \in Q$. According to Proposition \ref{prop:RegpairWiseComp}, $\phi_1$ and $\phi_2$ are then pairwise compatible if $\phi_1$ has support $x$ and $\phi_2$ support $y$ and $\epsilon_{x \wedge y}(\phi_1) = \epsilon_{x \wedge y}(\phi_2)$. Theorem \ref{th:ValAlgHyptree} applies to this case, with $z_i = x_i \wedge x_{b(i)}$ and hypertrees are join trees (see Section \ref{subsec:CondIndepStruct}) satisfying the running intersection property, see Section \ref{subsec:CondIndepStruct}.

\section{Facorization and conditional independence structures} \label{subsec:FactCondIndep}

In Section \ref{subsec:Continuation} we have defined conditional independence relative to a valuation $\phi$, see Definition \ref{def.CondIndepInf}. Besides conditional independence of domains or questions, this definition exhibits the factorization of $\phi$ as a defining element. In probability theory, fatcorizations of a distribution into marginals or prior and conditional distributions are basic concepts to stochastic  conditional independence. The different equivalent forms this concept of conditional independence can take in the more general structure of a regular valuation algebra is shown in Theorem \ref{th:EquivCondIndepRel}. In this section the concept of conditional independence relative to a valuation $\phi$ will be generalized to factorizations with more than two factors. 

First we extend Definition \ref{def.CondIndepInf} to a set of questions. Let $(\Phi,\cdot,1;E)$ with $E = \{\epsilon_x:x \in Q\}$ be a valuation algebra.

\begin{definition} \label{def:CondIndepSetVal} \textbf{Conditional independence of a set of questions relative to a valuation:}
We call a set of questions $\{x_1,\ldots,x_n\}$, $x_i \in Q$ conditional independent given $z \in Q$ relative to $\phi \in \Phi$, if

\begin{enumerate}
\item $\bot \{x_1,\ldots,x_n\} \vert  z$,
\item $\epsilon_{x_1 \vee \cdots \vee x_n \vee z}(\phi) = \psi_1 \cdot \ldots \cdot \psi_n$,
\end{enumerate}
where $\psi_i \in \Phi$ have support $x_i \vee z$ for $i = 1,\ldots,n$. We then write $\bot_\phi \{x_1,\ldots,x_n\} \vert  z$.
\end{definition}

Proposition \ref{prop:CondIndepInf} extends in the following way to this more general case.

\begin{proposition} \label{prop:GenCondInedepInf}
Assume $\bot_\phi \{x_1,\ldots,x_n\} \vert z$. Then, if $\epsilon_{x_1 \vee \cdots \vee x_n \vee z}(\phi) = \psi_1 \cdot \ldots \cdot \psi_n$, where $\psi_i \in \Phi$ have support $x_i \vee z$ for $i = 1,\ldots,n$,
\begin{enumerate}
\item $\epsilon_{x_i \vee z}(\phi) = \psi_i \cdot \epsilon_z(\psi_1) \cdot \ldots \cdot \epsilon_z(\psi_{i-1}) \cdot \epsilon_z(\psi_{i+1}) \cdot \ldots \cdot \epsilon_z(\psi_n)$, $i = 1,\ldots,n$
\item $\epsilon_{z}(\phi) = \epsilon_z(\psi_1) \cdot \ldots \cdot \epsilon_z(\psi_n)$.
\end{enumerate}
\end{proposition}

\begin{proof}
Let $y_i = x_1 \vee \ldots \vee y_{i-1} \vee y_{i+1} \vee \ldots \vee y_n \vee z$. Then we have $x_i \vee z \bot y_i \vert z$, see Proposition \ref{prop:ElResMultCondIndep}. This implies
\begin{eqnarray} \label{eq:Fact}
\epsilon_{x_i \vee z}(\phi) &=& \psi_i \cdot \epsilon_{x_i \vee z}(\prod_{j \in y_i} \psi_j) \nonumber \\
&=& \psi_i \cdot \epsilon_{x_i \vee z}(\epsilon_z(\prod_{j \in y_i} \psi_j)) \nonumber \\
&=& \psi_i \cdot \epsilon_z(\prod_{j \in y_i} \psi_j).
\end{eqnarray}
From this we derive
\begin{eqnarray*}
\epsilon_z(\phi) &=& \epsilon_z(\psi_i) \cdot \epsilon_z(\prod_{j \in y_i} \psi_j)
\end{eqnarray*}
Now, we have also $\bot \{x_1,\ldots,x_{i-1},x_{i+},\ldots,x_n\} \vert z$.
By induction over $n = 2,3,,\ldots$ we get from this
\begin{eqnarray*}
\epsilon_z(\prod_{j \in y_i} \psi_j) = \prod_{j \in y_i} \epsilon_z(\psi_j).
\end{eqnarray*}
But this implies
\begin{eqnarray*}
\epsilon_z(\phi) = \epsilon_z(\psi_1) \cdot \epsilon_z(\psi_2) \cdot \ldots \cdot \epsilon_z(\psi_n).
\end{eqnarray*}
This is item 2 of the proposition. It implies also by (\ref{eq:Fact})
\begin{eqnarray*}
\epsilon_{x_i \vee z}(\phi) = \psi_i \cdot \prod_{j \in y_i} \epsilon_z(\psi_j),
\end{eqnarray*}
that is item 1.
\end{proof}
 
 As a variant, we consider the factorization
 \begin{eqnarray*}
\phi = \psi_1 \cdot \ldots \cdot \psi_n \cdot \psi_{n+1}
\end{eqnarray*}
where $\psi_i$ has support $x_i$ for $i = 1$ to $n$ and $\psi_{n+1}$ has support $z$. Proposition \ref{prop:GenCondInedepInf} applies to this facorization, since the elements $\psi_i$ have also support $x_i \vee z \geq x_i,z$. From 
\begin{eqnarray*}
\phi = \psi_1 \cdot \ldots \cdot (\psi_n \cdot \psi_{n+1})
\end{eqnarray*}
we obtain, 
\begin{eqnarray*}
\epsilon_z(\phi) &=& \epsilon_z(\psi_1) \cdot \ldots \cdot \epsilon_z(\psi_{n-1}) \cdot \epsilon_z(\psi_n \cdot \psi_{n+1}) \\
&=& \epsilon_z(\psi_1) \cdot \ldots \cdot \epsilon_z(\psi_{n-1}) \cdot \epsilon_z(\psi_n) \cdot \psi_{n+1}.
\end{eqnarray*}
and similarly
\begin{eqnarray*}
\epsilon_{x_i \vee z}(\phi) &=& (\psi_i \cdot \psi_{n+1}) \cdot \prod_{j \in y_i} \epsilon_z(\psi_j) \\
&=& \psi_i \cdot \prod_{j \in y_i} \epsilon_z(\psi_j) \cdot \psi_{n+1}.
\end{eqnarray*}

We shall see that this last result is a special case of the following more general situation. Let $(T;\lambda)$ be a Markov tree with $T =(V,E)$, see Section \ref{subsec:CondIndepStruct}. Recall that if $v$ is any node of the tree $T$, then $T_{v,u}$ with node set $V_{v,u}$ for $u \in ne(v)$ are  the partial Markov trees obtained if node $v$ and the edges $\{v.u\}$ are removed from $T$ (see Section \ref{subsec:CondIndepStruct}). Consider now a Markov tree factorization
\begin{eqnarray*}
\phi = \prod_{v \in V} \psi_v
\end{eqnarray*}
where $\psi_v$ has support $\lambda(v)$. Then, we conclude that
\begin{eqnarray*}
\bot_\phi \{\lambda(V_{v,u}):u \in ne(v)\} \vert \lambda(v)
\end{eqnarray*}
for all $v \in V$. In fact, we have
\begin{eqnarray*}
\phi = \prod_{u \in ne(v)} \phi_{v,u} \cdot \phi_v,
\end{eqnarray*}
where 
\begin{eqnarray*}
\phi_{v,u} = \prod_{w \in V_{v,u}} \psi_w.
\end{eqnarray*}
This is a factorization of the kind considered above after Proposition \ref{prop:GenCondInedepInf}. Accordingly, we see that
\begin{eqnarray*}
\epsilon_{\lambda(v)}(\phi) = \psi_{\lambda(v)} \cdot \prod_{u \in ne(v)} \epsilon_{\lambda(v)}(\phi_{v,u}).
\end{eqnarray*}
This leads then to the recursive procedure in Markov trees to compute $\epsilon_{\lambda(v)}(\phi)$ as in the case of an information algebra, see Section \ref{subsec:MarkovPropag}, and especially the proof of Theorem \ref{th:MarkovTreeProj}. This procedure applies therefore also to valuation algebras, and if the algebra allows for division, then this allows to improve the process.

In order to show this, we describe the Markov recursion in terms of a message passing scheme. This scheme has been proposed in  \cite{shenoyshafer90} for multivariate valuation algebras, it has also been described in \cite{kohlas03}. Since we have also $\lambda(v) \bot \lambda(V_{v,u} \bot \lambda(u)$ for all neighbours $u$ of node $v$ (see Theorem \ref{th:NeighCondIndep}), we have $\epsilon_{\lambda(v)}(\phi_{v,u}) = \epsilon_{\lambda(v)}(\epsilon_{\lambda(u)}(\phi_{v,u}))$ and therefore (see Theorem \ref{th:MarkovTreeProj})
\begin{eqnarray} \label{eq:EctrComp}
\epsilon_{\lambda(v)}(\phi) = \psi_{\lambda(v)} \cdot \prod_{u \in ne(v)} \epsilon_{\lambda(v)}(\epsilon_{\lambda(u)}(\phi_{v,u})).
\end{eqnarray}
Define then
\begin{eqnarray*}
\mu_{u \rightarrow v} = \epsilon_{\lambda(v)}(\epsilon_{\lambda(u)}(\phi_{v,u}))
\end{eqnarray*}
This can be considered as a message from node $u$ to node $v$. In order to describe how with the passing of such messages an extraction  $\epsilon_{\lambda(v)}(\phi)$ can be computed in a Markov tree, number the nodes in $V$ so that $j > i$ if node $v_j$ is on the (unique) path form node $v_i$ to node $v_n$, if $\vert V \vert = n$, see Section \ref{subsec:CondIndepStruct}. Further, direct all edges $\{v_i,v_j\}$ towards the root node $v_n$, such that $(v_i,v_j)$ is a directed arc associated with the edge $\{v_i,v_j\}$ such that $i < j$. We now denote the nodes simply be their number to simplify notation. For any node $i$ let $ch(i)$ denote the (unique) neighbourg on the outgoing arc $(i,ch(i))$, the child of $i$. All nodes except node $n$ have a child. On the other hand let $pa(i)$ denote the neighbours of node $i$ on the incoming arcs of node $i$, the parents of $i$. The set $pa(i)$ may be empty, then node $i$ is called a leaf. Note that node $1$ must be a leaf.

According to Section \ref{subsec:CondIndepStruct} this makes the node set of a Markov tree to a hypertree. We may now compute the messages in the tree in the sequence of the numbering. In fact node $1$ is a leaf and we have for any leaf $\phi_{ch(i),i} = \psi_i$ and so we may compute the message $\mu_{1 \rightarrow ch(1)}$ to its child. Then node $2$ is either a leaf or $pa(2) = \{1\}$. Then we may compute $\epsilon_{\lambda(2)}(\phi_{2,1}) = \psi_2 \cdot \mu_{1 \rightarrow 2}$. In general, if we proceed for $i = 3,4,...$ and arrive at the node $i$, then it is either a leaf or the messages $\mu_{j \rightarrow i}$ from all its parents $j \in pa(i)$ have been computed. So, again we may compute
\begin{eqnarray*}
\epsilon_{\lambda(i)}(\phi_{ch(i),i}) = \psi_i \cdot \prod_{j \in pa(i)} \mu_{j \rightarrow i}.
\end{eqnarray*}
This in turn allows to compute the message to its child $\mu_{i \rightarrow ch(i)} = \epsilon_{\lambda(ch(i))}(\epsilon_{\lambda(i)}(\phi_{ch(i),i})$. In this way we arrive finally at the root node $n$ and can then compute the extraction $\epsilon_{\lambda(n)}(\phi)$. This way to compute is called \textit{collect} algorithm. Compare this with the algorithm for information algebras described in Section \ref{subsec:CompHypTree}.

If the messages $\mu_{i \rightarrow ch(i)}$ computed in collect algorithm, are stored, they may be used to compute the extractions $\epsilon_{\lambda(i)}(\phi)$ for all nodes of the Markov tree by going back in the numbering. In fact, the root node $n$ may send messages to all its parents
\begin{eqnarray*}
\mu_{n \rightarrow j} = \psi_n \cdot \prod_{k \in ne(n),k \not= j} \mu_{k \rightarrow n},\quad  j \in pa(n).
\end{eqnarray*}{
Then all these parents can compute $\epsilon_{\lambda(j)}(\phi)$ by formulat \ref{eq:EctrComp}. And then these nodes may send their messages to their parents, etc. until all nodes are reached. This second procedure is called \textit{distribute} algorithm. The whole system is known as the \textit{Senoy-Shafer} architecture.

In this form there are a number of inefficiencies hidden. For example, many subcombinations of messages are recomputed. To avoid this, we may use division, that is we assume a regular or separative valuation algebra. Assume that we store at node $i$ at the beginning $\eta_i = \psi_i$. In the collect phase, any time a message $\mu_{j \rightarrow i}$ arrives at node $i$ we update $\eta_i := \eta_i \cdot \mu_{j \rightarrow i}$. Once the node $i$ sends its message $\mu_{i \rightarrow ch(i)}$ to its child $ch(i)$ we divide this message out of $\eta_i$, that is $\eta_i := \eta_i \cdot \mu^{-1}_{i \rightarrow ch(i)}$. In the distribute phase, starting with node $n$, the messages of a node $j$ to a parent node $i$ are as in the collect phase, namely $\epsilon_{\lambda(i)}(\epsilon_{\lambda(j)}(\eta_j)) = \epsilon_{\lambda(i)}(\epsilon_{\lambda(j)}(\phi))$. The receiving node combines the incoming message as in the collect phase with its store content.  This computational scheme is associated with the name of \textit{Lauritzen-Spiegelhalter (LS) architecture}. We claim that at the end each node $i$ contains its extraction $\epsilon_{\lambda(i)}(\phi)$.

\begin{theorem} \label{th:LSCorrectness}
Assume $(\Phi,\cdot,1;E)$ to be a regular or separative valuation algebra. Then, at the end of the computations according to the LS architecture, each node $i$ stores the extraction $\epsilon_{\lambda(i)}(\phi)$.
\end{theorem}

\begin{proof}
At end of the collect phase, the claim holds for the node $n$, $\eta_n = \epsilon_{\lambda(i)}(\phi)$. We proceed by induction. Assume that the claim holds for all nodes $j > i$, for some index $i = n-1,\ldots,1$. Then it holds for the child $ch(i)$ of node $i$, since $ch(i) \geq i$,that is
\begin{eqnarray*}
\epsilon_{\lambda(ch(i))}(\phi) = \psi_{ch(i)} \cdot \prod_{j \in ne(ch(i))} \mu_{j \rightarrow ch(i)}.
\end{eqnarray*}
The message of $ch(i)$ sent to node $i$ in the distribute phase is then
\begin{eqnarray*}
&&\epsilon_{\lambda(i)}(\epsilon_{\lambda(ch(i))}(\phi)) \\
&=& \epsilon_{\lambda(i)}(\psi_{ch(i)} \cdot \prod_{j \in ne(ch(i))} \mu_{j \rightarrow ch(i)}) \\
&=& \epsilon_{\lambda(i)}(\psi_{ch(i)} \cdot \prod_{j \in ne(ch(i)),j \not= i} \mu_{j \rightarrow ch(i)}) \cdot \mu_{i \rightarrow ch(i)}
\end{eqnarray*}
since $\mu_{i \rightarrow ch(i)}$ has support $\lambda(ch(i))$. But then by the definition of messages, it follows
\begin{eqnarray*}
\epsilon_{\lambda(i)}(\epsilon_{\lambda(ch(i))}(\phi)) = \mu_{ch(i) \rightarrow i} \cdot \mu_{i \rightarrow ch(i)}.
\end{eqnarray*}
If this meassage is combined with the value $\eta_i$ stored in node $i$, this gives
\begin{eqnarray*}
&&\psi_i \cdot \prod_{j \in pa(i)} \mu_{j \rightarrow i} \cdot \mu^{-1}_{i \rightarrow ch(i)} \cdot \mu_{ch(i) \rightarrow i} \cdot \mu_{i \rightarrow ch(i)} \\
&=& \epsilon_{\lambda(i)}(\phi) \cdot f_{\mu_{i \rightarrow ch(i)}} = \epsilon_{\lambda(i)}(\phi).
\end{eqnarray*}
The last equation follows from the definition of $\mu_{i \rightarrow ch(i)}$ as $\epsilon_{\lambda(ch(i))}(\epsilon_{\lambda(i)}(\phi_{ch(i),i})$ and $\phi_{ch(i),i} \leq \phi$ and Lemma \ref{le:PreorderClasses} and \ref{le:OrderGroup}.
\end{proof}

Note then in the LS architecture at the beginning we have $\phi = \prod_{i=1}^n \eta_i$ with $\eta_i = \psi_i$. In the collect phase at step $i$, node $i$ sends the message $\mu_{i \rightarrow ch(i)}$ to its child and this message is combined with $\eta_{ch(i)}$. On the other hand $\eta_i$ is combined with the inverse $\mu^{-1}_{i \rightarrow ch(i)}$ of this message. Therefore the contents of the nodes continue to combine to $\phi$. So, at the end of the collect phase we have
\begin{eqnarray*}
\phi = \epsilon_{\lambda(n)}(\phi) \cdot \prod_{i_1}^{n-1} \eta_i.
\end{eqnarray*}
In the distribute phase, any store $\eta_i$ for $i = n-1,\ldots,1$ is in turn updated with $\epsilon_{\lambda(i)}(\epsilon_{\lambda(j)}(\phi))$ if node $j$ is a child of node $i$. At the end of the distribute phase we have $\eta_i = \epsilon_{\lambda(i)}(\phi)$ for all $ i = 1,\ldots,n$. So we must have the identity
\begin{eqnarray*}
\phi \cdot \prod_{i = n-1}^1 \epsilon_{\lambda(i)}(\epsilon_{\lambda(ch(i))}(\phi))    = \prod_{i=1}^n \epsilon_{\lambda(i)}(\phi).
\end{eqnarray*}

We claim that $\epsilon_{\lambda(i)}(\epsilon_{\lambda(ch(i))}(\phi)) = \epsilon_{\lambda(ch(i))}(\epsilon_{\lambda(i)}(\phi))$. In fact, if we change the root node $n$ to one of its neigbhours $j \in pa(n)$, then the arc $(j,n)$ changes direction, but all the other arcs in the directed tree remain the same. So, we have
\begin{eqnarray*}
\phi \cdot \prod_{i = n-1,i \not= j}^1 \epsilon_{\lambda(i)}(\epsilon_{\lambda(ch(i))}(\phi)) \cdot \epsilon_{\lambda(j)}(\epsilon_{\lambda(n)}(\phi)) &=& \prod_{i=1}^n \epsilon_{\lambda(i)}(\phi), \\
\phi \cdot \prod_{i = n-1,i \not= j}^1 \epsilon_{\lambda(i)}(\epsilon_{\lambda(ch(i))}(\phi)) \cdot \epsilon_{\lambda(n)}(\epsilon_{\lambda(j)}(\phi)) &=& \prod_{i=1}^n \epsilon_{\lambda(i)}(\phi).
\end{eqnarray*}
This implies $\epsilon_{\lambda(i)}(\epsilon_{\lambda(ch(i))}(\phi)) = \epsilon_{\lambda(ch(i))}(\epsilon_{\lambda(i)}(\phi))$, since we obtain from the equations above
\begin{eqnarray*}
\epsilon_{\lambda(i)}(\epsilon_{\lambda(ch(i))}(\phi)) &=& \epsilon_{\lambda(ch(i))}(\epsilon_{\lambda(i)}(\phi)) \cdot f_A, \\
\epsilon_{\lambda(ch(i))}(\epsilon_{\lambda(i)}(\phi)) &=& \epsilon_{\lambda(i)}(\epsilon_{\lambda(ch(i)i)}(\phi)) \cdot f_A
\end{eqnarray*}
where $f_A$ is the unit in the group of of the expression $\phi \cdot \prod_{i = n-1,i \not= j}^1 \epsilon_{\lambda(i)}(\epsilon_{\lambda(ch(i))}(\phi))$. From this have
\begin{eqnarray*}
\epsilon_{\lambda(i)}(\epsilon_{\lambda(ch(i))}(\phi)) &=& \epsilon_{\lambda(i)}(\epsilon_{\lambda(ch(i))}(\phi)) \cdot f_A
\end{eqnarray*}
so that the idempotent $f_A$ is absorbed by $\epsilon_{\lambda(i)}(\epsilon_{\lambda(ch(i))}(\phi))$ and this implies the identity. Therefore we may finally state that
\begin{eqnarray*}
\phi \cdot \prod_{\{u,v\} \in E} \epsilon_{\lambda(u)}(\epsilon_{\lambda(v)}(\phi)) = \prod_{v \in V} \epsilon_{\lambda(v)}(\phi)
\end{eqnarray*}
since we may take an node $v$ as root.

If the regular or separative valuation algebra is commutative, then $\epsilon_{\lambda(i)}(\epsilon_{\lambda(ch(i))}(\phi)) = \epsilon_{\lambda(u) \wedge \lambda(v)}(\phi)$, hence
\begin{eqnarray*}
\phi \cdot \prod_{\{u,v\} \in E} \epsilon_{\lambda(u) \wedge \lambda(v)}(\phi) = \prod_{v \in V} \epsilon_{\lambda(v)}(\phi)
\end{eqnarray*}
or also
\begin{eqnarray*}
\phi = \prod_{v \in V} \epsilon_{\lambda(v)}(\phi) \cdot \prod_{\{u,v\} \in E} \epsilon^{-1}_{\lambda(u) \wedge \lambda(v)}(\phi)
\end{eqnarray*}
This is a well-known result in a multivariate regular valuation algebra, see \cite{kohlas03}.